\documentclass[a4paper,twoside,11pt,openright]{book}

\usepackage[utf8]{inputenc} 
\usepackage[ngerman,english]{babel}		
\usepackage{csquotes} 

\newcommand{\noop}[1]{} 
\usepackage[backend=biber, sorting=nyt, style=numeric-comp, useprefix=false, sortcites=true, doi=false, date=year, isbn=false, maxbibnames=9, url=true, firstinits=true, bibwarn=true, hyperref=true]{biblatex}
\AtEveryBibitem{
  \clearfield{issn}
  \clearfield{month}
  \clearfield{primaryclass}
  \clearfield{eprintclass}
  \ifentrytype{software}{}{\clearfield{url}} 
  \ifentrytype{book}{}{\clearfield{isbn}} 
  \ifentrytype{book}{\clearfield{pagetotal} \clearfield{location} }{} 
}

\DeclareBibliographyDriver{unpublished}{
  \usebibmacro{bibindex}%
  \usebibmacro{begentry}%
  \usebibmacro{author}%
  \setunit{\labelnamepunct}\newblock
  \usebibmacro{title}%
  \newunit
  \printlist{language}%
  \newunit\newblock
  \usebibmacro{byauthor}%
  \newunit\newblock
  \printfield{howpublished}%
  \newunit\newblock
  \printfield{note}%
  \newunit\newblock
  \usebibmacro{location+date}%
  \newunit\newblock
  \usebibmacro{doi+eprint+url}
  \newunit\newblock
  \usebibmacro{addendum+pubstate}%
  \setunit{\bibpagerefpunct}\newblock
  \usebibmacro{pageref}%
  \usebibmacro{finentry}}
\renewbibmacro*{volume+number+eid}{%
  \printfield{volume}%
  \setunit*{\addnbspace}
  \printfield{number}%
  \setunit{\addcomma\space}%
  \printfield{eid}}
\DeclareFieldFormat[article]{number}{no. #1}
  
\DeclareFieldFormat
[article,inbook,incollection,inproceedings,patent,thesis,unpublished]
{title}{\textit{#1}}

\DeclareFieldFormat{journaltitle}{#1\isdot}

\DeclareFieldFormat{journaltitle}{#1\isdot}

\DeclareFieldFormat[article,periodical]{volume}{\mkbibbold{#1}}

\renewbibmacro{in:}{
  \ifentrytype{article}{%
  }{%
    \printtext{\bibstring{in}\intitlepunct}%
  }%
}

\makeatletter 
\DeclareFieldFormat{eprint:arxiv}{%
  {\small\texttt{arXiv}\addcolon}\space%
  \href{https://arxiv.org/\abx@arxivpath/#1}{%
    \iffieldundef{eprintclass}
      {}
      {\nolinkurl{\thefield{eprintclass}/}}%
    \nolinkurl{#1}}}
\makeatother

\usepackage[inner=18.2mm, top=27mm, outer=36.4mm, bottom=54mm, marginparsep=0mm, bindingoffset=10mm, heightrounded]{geometry}

\usepackage{layout} 
\makeatletter
\renewcommand*{\lay@value}[2]{%
  \strip@pt\dimexpr0.351459\dimexpr\csname#2\endcsname\relax\relax mm%
}
\makeatother

\usepackage{fancyhdr}
\renewcommand{\sectionmark}[1]{\markright{\thesection\ #1}} 
\fancypagestyle{fancyStandard}{%
    \fancyhf{}
    \fancyhead[LE]{\small\scshape\leftmark}
    \fancyhead[RO]{\small\scshape\rightmark}
    \fancyfoot[LE,RO]{\thepage}
    
}
\fancypagestyle{withoutSections}{%
    \fancyhf{}
    \fancyhead[LE,RO]{\small\scshape\leftmark}
    \fancyfoot[LE,RO]{\thepage}
    
}
\fancypagestyle{plain}{%
    \fancyhf{}
    \fancyfoot[LE,RO]{\thepage}
    
}

\let\origdoublepage\cleardoublepage
\newcommand{\clearemptydoublepage}{%
  \clearpage
  {\pagestyle{plain}\origdoublepage}
}
\let\cleardoublepage\clearemptydoublepage

\usepackage[explicit, nobottomtitles*]{titlesec}  
\usepackage{xcolor}
\definecolor{light-gray}{gray}{0.7}
\usepackage{xstring} 

\titleformat{name=\chapter}{\normalfont\huge\bfseries}
{\parbox{0.1\textwidth}{\color{light-gray}\scalebox{8}{\textsl{\IfInteger{\thechapter}{\oldstylenums{\thechapter}}{\thechapter}}}}}
{10pt}{\rlap{\parbox{\textwidth}{\vspace{50pt} #1}}}

\titlespacing*{\section}{0pt}{5.25ex plus 1.5ex minus .2ex}{3.45ex plus .2ex}
\titlespacing*{\subsection} {0pt}{4.8ex plus 1.5ex minus .2ex}{2.25ex plus .2ex}
\titlespacing*{\subsubsection}{0pt}{4.8ex plus 1.5ex minus .2ex}{2.25ex plus .2ex}

    {\large\begin{center}%
    \bfseries{Abstract} \end{center}}

\usepackage{tocloft} 

\newcommand*{\insertlofspace}{
  \addtocontents{lof}{\protect\addvspace{10pt}}%
}

\usepackage{anyfontsize} 
\usepackage{wrapfig}
\usepackage{hhline} 
\usepackage{pdflscape}
\usepackage{placeins} 
\usepackage{fontawesome5} 
\usepackage{mdframed}
\usepackage{minibox}
\usepackage[super]{nth} 
\usepackage[labelfont=bf,textfont=small]{caption} 
\usepackage{subcaption} 
\makeatletter 
\newcommand\footnoteref[1]{\protected@xdef\@thefnmark{\ref{#1}}\@footnotemark}
\makeatother

\usepackage{amssymb}
\usepackage{amsmath}
\usepackage{amsfonts}   
\usepackage{mathtools}
\usepackage{mathrsfs}				
\usepackage{amsbsy}
\usepackage{faktor}                 
\usepackage{relsize}                
\usepackage{xfrac}
\usepackage{bbm}
\usepackage{dsfont}				

\usepackage{fancyvrb}               

\usepackage{amsthm}				

\definecolor{links}{rgb}{0,0.3,0} 
\definecolor{mylinks}{rgb}{0.8,0.2,0} 

\usepackage[ pdfencoding=auto, colorlinks=true, citecolor=blue, urlcolor=links, linkcolor=mylinks, pdfpagelabels, bookmarks, hyperindex, hyperfigures, pdftitle={Hypergeometric Feynman integrals}, pdfauthor={R. P. Klausen}]{hyperref}             
\usepackage[nameinlink, noabbrev]{cleveref}
\crefformat{equation}{{\color{mylinks}(#2#1#3)}} 

\usepackage{thmtools}
\newtheoremstyle{owntheorem}{}
  {}
  {\itshape}
  {}
  {}
  {:}
  {1em}
  {\textbf{\thmname{#1}\thmnumber{ #2}}\thmnote{ [#3]}}
\newtheoremstyle{owndefinition}%
  {}{}%
  {}{}%
  {}{:}%
  {1em}%
  {\textbf{\thmname{#1}\thmnumber{ #2}}\thmnote{ [#3]}}

\theoremstyle{owntheorem}
\newtheorem{theorem}{Theorem}[section]
\newtheorem{lemma}[theorem]{Lemma}
\newtheorem{cor}[theorem]{Corollary}

\theoremstyle{owndefinition}

\makeatletter  
\def\@endtheorem{\hfill{\lower0.3ex\hbox{\ensuremath{\triangle}}}\endtrivlist\@endpefalse } 
\makeatother

\newtheoremstyle{ownexample}{}{}{}{}{}{:}{1em}{\textbf{\thmname{#1}\thmnumber{ #2}}\thmnote{ [#3]}}
\theoremstyle{ownexample} 
\newtheorem{example}[theorem]{Example}

\usepackage{url} 				
\usepackage{wasysym}				
\usepackage{enumerate}
\usepackage{rotating}				
\usepackage{multirow}
\usepackage{array}
\usepackage{makecell} 
\usepackage{longtable}
\usepackage{lscape}

\usepackage{listings}				
\usepackage{tikz}
\usetikzlibrary{calc}
\usetikzlibrary{arrows}
\usetikzlibrary{shapes}
\usetikzlibrary{positioning}
\usepackage{tablefootnote}
\usepackage{blkarray}
\usepackage{authblk}
\usepackage{spverbatim}
\usepackage{pgothic}
\usepackage{egothic}
\usepackage{yfonts}
\DeclareMathAlphabet{\mathpgoth}{OT1}{pgoth}{m}{n} 
\DeclareMathAlphabet{\mathesstixfrak}{U}{esstixfrak}{m}{n} 
\DeclareMathAlphabet{\mathboondoxfrak}{U}{BOONDOX-frak}{m}{n} 
\DeclareMathAlphabet{\mathgoth}{U}{ygoth}{m}{n} 

\numberwithin{equation}{section} 

\newcommand{\dif}{\mathop{}\!\mathrm{d}}

\makeatletter
\newcommand{\spx}[1]{%
  \if\relax\detokenize{#1}\relax
    \expandafter\@gobble
  \else
    \expandafter\@firstofone
  \fi
  {^{#1}}%
}
\makeatother

\newcommand\pd[3][]{\frac{\partial\spx{#1}#2}{\partial#3\spx{#1}}}

\newcommand{\od}[3][]{\frac{\dif\spx{#1}#2}{\dif#3\spx{#1}}}

\newcommand{\genericdel}[4]{%
  \ifcase#3\relax
  \ifx#1.\else#1\fi#4\ifx#2.\else#2\fi\or
  \bigl#1#4\bigr#2\or
  \Bigl#1#4\Bigr#2\or
  \biggl#1#4\biggr#2\or
  \Biggl#1#4\Biggr#2\else
  \left#1#4\right#2\fi
}

\newcommand{\sVert}[1][0]{%
  \ifcase#1\relax
  \rvert\or\bigr|\or\Bigr|\or\biggr|\or\Biggr
  \fi
}

\newcommand{\fd}[2][]{\,\Delta^{#1}\!#2} 
\newcommand{\fpd}[3][]{\,\Delta_{#2}^{#1}#3} 

\newcommand{\diag}{\operatorname{diag}}
\newcommand{\Ann}{\operatorname{Ann}}

\newcommand{\Disc}{\operatorname{Disc}}

\newcommand{\ch}{\operatorname{char}}

\newcommand{\Arg}{\operatorname{Arg}}
\renewcommand\Re{\operatorname{Re}} 		
\renewcommand\Im{\operatorname{Im}} 		

\newcommand{\vol}{\operatorname{vol}}

\newcommand{\rank}{\operatorname{rank}}
\newcommand{\Sing}{\operatorname{Sing}}
\newcommand{\codim}{\operatorname{codim}}
\newcommand{\Li}{\operatorname{Li}}
\newcommand{\Conv}{\operatorname{Conv}}
\newcommand{\Cone}{\operatorname{Cone}}
\newcommand{\relint}{\operatorname{relint}}
\newcommand{\Adj}{\operatorname{Adj}}
\newcommand{\Newt}{\operatorname{Newt}}
\newcommand{\Aff}{\operatorname{Aff}}
\newcommand{\Vertices}{\operatorname{Vert}}
\newcommand{\Dep}{\operatorname{Dep}}
\newcommand{\Val}{\operatorname{Val}}
\newcommand{\gr}{\operatorname{gr}}
\newcommand{\initial}{\operatorname{in}}
\newcommand{\sign}{\operatorname{sign}}
\newcommand{\Gale}{\operatorname{Gale}}
\newcommand{\ord}{\operatorname{ord}}
\newcommand{\Sol}{\operatorname{Sol}}

\newcommand{\ddif}{\mathop{}\!\mathrm{d}^d}

\newcommand{\order}[1]{\mathcal{O} \!\left( #1 \right)}
\newcommand{\bigslant}[2]{{\raisebox{.2em}{\ensuremath{#1}}\!\!\left/\raisebox{-.2em}{\ensuremath{#2}}\right.}}
\newcommand{\normalslant}[2]{#1/#2}
\newcommand{\contraction}[2]{#1/#2}

\newcommand{\Aa}{\mathcal{A}}
\newcommand{\Bb}{\mathcal{B}}
\newcommand{\Tt}{\mathcal{T}}
\newcommand{\Aas}{\mathcal{A}_\sigma}
\newcommand{\Aabs}{\mathcal{A}_{\bar{\sigma}}}
\newcommand{\zsigma}{z_\sigma}
\newcommand{\zbarsigma}{z_{\bar\sigma}}
\newcommand{\nuu}{{\underline{\nu}}}
\newcommand{\Cc}[1][f]{\overline{\mathcal C}_{#1}^{\mathrm c}}
\newcommand{\AnK}{\mathbb{A}^n_{\mathbb{K}}}
\newcommand{\PnK}{\mathbb{P}^n_{\mathbb{K}}}
\newcommand{\Uu}{\mathcal{U}}
\newcommand{\Ff}{\mathcal{F}}
\newcommand{\Gg}{\mathcal{G}}
\newcommand{\Ss}{\mathcal{S}}
\newcommand{\Csn}{(\mathbb C^*)^n}
\newcommand{\hatT}{\widehat{\mathcal T}}
\newcommand{\btau}{{\bar\tau}}
\newcommand{\mink}[1]{\mathgoth{#1}} 
\newcommand{\dhalf}{\frac{d}{2}}
\newcommand{\Var}{\mathlarger{\mathbb{V}}} 

\newcommand{\HypF}[3]{ {_2}F_1 \!\left.\left(\genfrac{}{}{0pt}{}{#1}{#2}\right| #3 \right) }

\newcommand{\HypFpq}[5]{ {_{#1}}F_{#2} \!\left.\left(\genfrac{}{}{0pt}{}{#3}{#4}\right| #5 \right) }
\newcommand{\StirlingFirst}[2]{\begin{bmatrix} #1 \\ #2 \end{bmatrix}}
\newcommand{\StirlingFirstSmall}[2]{\begin{bsmallmatrix} #1 \\ #2 \end{bsmallmatrix}}
\newcommand{\StirlingSecond}[2]{\begin{Bmatrix} #1 \\ #2 \end{Bmatrix}}
\newcommand{\StirlingSecondSmall}[2]{\begin{Bsmallmatrix} #1 \\ #2 \end{Bsmallmatrix}}
\newcommand{\Ssum}[2]{S^{(#1)}_{#2} }
\newcommand{\Zsum}[2]{Z^{(#1)}_{#2} }
\newcommand{\Binomial}[2]{\begin{pmatrix} #1 \\ #2 \end{pmatrix}}

\newcommand{\comma}{\quad \textrm{ ,}} 
\newcommand{\point}{\quad \textrm{ .}} 
\newcommand{\monthword}[1]{\ifcase#1\or \or \or \or April\fi}

\newcommand*{\vcenteredhbox}[1]{\begingroup \setbox0=\hbox{#1}\parbox{\wd0}{\box0}\endgroup}
\newcommand{\softwareName}[1]{\textsl{#1}}

\newcommand{\CNV}{completely non-vanishing }

\newcommand{\HKP}{Horn-Kapranov-parameterization }
\newcommand{\HKp}{Horn-Kapranov-parameterization}

\usepackage[symbols,nogroupskip,stylemods]{glossaries-extra}

\glsdisablehyper 
\makeglossaries

\GlsXtrEnablePreLocationTag{\textit{-- see p.}~}{\textit{-- see pp.}~}

\setglossarystyle{long}
 {\begin{longtable}[l]{lp{.7\textwidth}}}
 {\end{longtable}}

\glsxtrnewsymbol[%
    description={$\mathbb K$ an arbitrary field (usually $\mathbb C$)}, sort={KK}]%
    {Kfield}%
    {\ensuremath{\mathbb K}}

\glsxtrnewsymbol[%
    description={the field without its zero element $\mathbb K^*=\mathbb K\setminus\{0\}$}, , sort={KKs}]%
    {Kfieldstar}%
    {\ensuremath{\mathbb K^*}}

\glsxtrnewsymbol[%
    description={vector configuration (usually acyclic) as a set $\Aa = \left\{a^{(1)},\ldots,a^{(N)}\right\}$ or as a matrix $\Aa = \begin{psmallmatrix} a^{(1)} & a^{(2)} & \cdots & a^{(N)} \end{psmallmatrix}$, homogenization of $A$}, sort={Aa}]%
    {Aa}%
    {\ensuremath{\Aa}}

\glsxtrnewsymbol[%
    description={point configuration as a set $A = \left\{a^{(1)},\ldots,a^{(N)}\right\}$ or as a matrix $A=\begin{psmallmatrix} a^{(1)} & a^{(2)} & \cdots & a^{(N)} \end{psmallmatrix}$, dehomogenization of $\Aa$}, sort={A}]%
    {A}%
    {\ensuremath{A}}

\glsxtrnewsymbol[%
    description={the restriction of $\Aa$ to elements/columns of an index set $\sigma$ (usually a simplex)}, sort={Aas1}]%
    {Aas}%
    {\ensuremath{\Aas}}
    
\glsxtrnewsymbol[%
    description={the restriction of $\Aa$ to elements/columns of an index set $\bar\sigma=\{1,\ldots,N\}\setminus\sigma$ (usually the complement of a simplex)}, sort={Aas2}]%
    {Aabs}%
    {\ensuremath{\Aabs}}
    
\glsxtrnewsymbol[%
    description={dimension of the basic affine Euclidean space, it is also the number of edges in a Feynman graph and the number of Schwinger parameters in a parametric Feynman integral}, sort={n}]%
    {n}%
    {\ensuremath{n}}

\glsxtrnewsymbol[%
    description={number of points in $A$ or vectors of $\Aa$, respectively / number of monomials}, sort={N}]%
    {N}%
    {\ensuremath{N}}

\glsxtrnewsymbol[%
    description={corank of a vector configuration $\Aa$, $r=N-n-1$}, sort={r}]%
    {r}%
    {\ensuremath{r}}

\glsxtrnewsymbol[%
    description={dual vector space of $V$}, sort={dual}]%
    {Vvee}%
    {\ensuremath{V^\vee}}

\glsxtrnewsymbol[%
    description={usually a linear functional}, sort={phi}]%
    {phi}%
    {\ensuremath{\phi}}
    
\glsxtrnewsymbol[%
    description={matrix transpose}, sort={transpose}]%
    {atop}%
    {\ensuremath{a^\top}}

\glsxtrnewsymbol[%
    description={Euclidean basis vectors of the considered vector space}, sort={e}]%
    {EuclideanBasisVectors}%
    {\ensuremath{e_1,\ldots}}
    
\glsxtrnewsymbol[%
    description={constant column vector of zeros with length $k$}, sort={0}]%
    {zeros}%
    {\ensuremath{\mathbf 0_k}}

\glsxtrnewsymbol[%
    description={constant column vector of ones with length $k$}, sort={1}]%
    {ones}%
    {\ensuremath{\mathbf 1_k}}

\glsxtrnewsymbol[%
    description={constant $m\times n$ matrix of ones}, sort={1}]%
    {onesMatrix}%
    {\ensuremath{\mathbf 1_{m\times n}}}

\glsxtrnewsymbol[%
    description={affine $n$-space over the field $\mathbb K$}, sort={AnK}]%
    {AnK}%
    {\ensuremath{\AnK}}

\glsxtrnewsymbol[%
    description={ring of polynomials in the variables $x_1,\ldots,x_n$ with coefficients in $\mathbb K$}, sort={Kx}]%
    {polyRing}%
    {\ensuremath{\mathbb K[x_1,\ldots,x_n]}}

\glsxtrnewsymbol[%
    description={the affine variety generated by a (finite or infinite) set of polynomials $S$, not necessarily irreducible}, sort={V}]%
    {variety}%
    {\ensuremath{\Var(S)}}

\glsxtrnewsymbol[%
    description={the algebraic torus $\Csn := (\mathbb C\setminus\{0\})^n$}, sort={Cstarn}]%
    {Csn}%
    {\ensuremath{\Csn}}

\glsxtrnewsymbol[%
    description={space of Laurent monomials having support $A$, i.e.\ the set of polynomials of the form $\sum_{a\in A} z_a x^a$}, sort={CA}]%
    {CA}%
    {\ensuremath{\mathbb C^A}}

\glsxtrnewsymbol[%
    description={Krull dimension of a commutative ring $R$}, sort={dimKr}]%
    {KrullDimension}%
    {\ensuremath{\dim_{\operatorname{Kr}}(R)}}

\glsxtrnewsymbol[%
    description={quotient ring (or factor ring) of a ring $R$ and an ideal $I$, i.e.\ the set of equivalence classes $\normalslant{R}{I} = \{r + I | r\in R\}$}, sort={quotientring}]%
    {quotientring}%
    {\ensuremath{R/I}}

\glsxtrnewsymbol[%
    description={affine hull or affine span generated by $a^{(1)},\ldots,a^{(k)}\in\AnK$ over $S\subseteq \mathbb K$}, sort={Aff}]%
    {affineHull}%
    {\ensuremath{\Aff_S\!\left(a^{(1)},\ldots,a^{(k)}\right)}}

\glsxtrnewsymbol[%
    description={projective $n$-space over $\mathbb K$}, sort={PnK}]%
    {PnK}%
    {\ensuremath{\PnK}}

\glsxtrnewsymbol[%
    description={homogeneous coordinates in projective space $\PnK$}, sort={s0sn}]%
    {homogeneousCoordinates}%
    {\ensuremath{[s_0:\ldots :s_n]}}

\glsxtrnewsymbol[%
    description={convex hull of a point configuration $A$, a convex polytope}, sort={Conv}]%
    {Conv}%
    {\ensuremath{\Conv(A)}}

\glsxtrnewsymbol[%
    description={usually a face of a polytope, by slight abuse of notation $\tau$ denotes the face as a polytope, the points generating the face and the corresponding index set, respectively}, sort={tau}]%
    {tauface}%
    {\ensuremath{\tau}}

\glsxtrnewsymbol[%
    description={the set of vertices of a polytope $P$}, sort={Vert}]%
    {Vert}%
    {\ensuremath{\Vertices(P)}}

\glsxtrnewsymbol[%
    description={relative interior of a polyhedra $P$}, sort={relint}]%
    {relint}%
    {\ensuremath{\relint(P)}}

\glsxtrnewsymbol[%
    description={normalized volume of a polytope $P$}, sort={vol}]%
    {vol}%
    {\ensuremath{\vol(P)}}

\glsxtrnewsymbol[%
    description={Newton polytope of a polynomial $f$}, sort={Newt}]%
    {Newt}%
    {\ensuremath{\Newt(f)}}

\glsxtrnewsymbol[%
    description={truncated polynomial w.r.t. a face $\tau$}, sort={ftau}]%
    {truncpoly}%
    {\ensuremath{f_\tau(x)}}

\glsxtrnewsymbol[%
    description={convex polyhedral cone generated by a vector configuration $\Aa$}, sort={Cone}]%
    {Cone}%
    {\ensuremath{\Cone(\Aa)}}

\glsxtrnewsymbol[%
    description={convex polyhedra generated by an intersection of halfspaces}, sort={PMb}]%
    {PMb}%
    {\ensuremath{P(M,b)}}

\glsxtrnewsymbol[%
    description={space of linear dependences}, sort={Dep}]%
    {Dep}%
    {\ensuremath{\Dep(\Aa)}}

\glsxtrnewsymbol[%
    description={space of linear evaluations}, sort={Val}]%
    {Val}%
    {\ensuremath{\Val(\Aa)}}

\glsxtrnewsymbol[%
    description={Gale dual (usually a Gale dual of $\Aa$)}, sort={Bb}]%
    {Bb}%
    {\ensuremath{\Bb}}

\glsxtrnewsymbol[%
    description={the set of all Gale duals of a vector configuration $\Aa$}, sort={Gale}]%
    {Gale}%
    {\ensuremath{\Gale(\Aa)}}

\glsxtrnewsymbol[%
    description={deletion operation of an oriented matroid}, sort={AaDCdeletion}]%
    {deletion}%
    {\ensuremath{\Aa\setminus j}}

\glsxtrnewsymbol[%
    description={contraction operation of an oriented matroid}, sort={AaDCcontraction}]%
    {contraction}%
    {\ensuremath{\contraction{\Aa}{j}}}

\glsxtrnewsymbol[%
    description={triangulation of a point/vector configuration}, sort={Tt}]%
    {Tt}%
    {\ensuremath{\Tt}}

\glsxtrnewsymbol[%
    description={maximal cells of a triangulation $\Tt$}, sort={Tthat}]%
    {hatT}%
    {\ensuremath{\hatT}}

\glsxtrnewsymbol[%
    description={refinement of subdivisions}, sort={refinement}]%
    {refinement}%
    {\ensuremath{\Ss \preceq \Ss^\prime}}

\glsxtrnewsymbol[%
    description={outer normal cone of $x$}, sort={NPx}]%
    {NPx}%
    {\ensuremath{N_P(x)}}

\glsxtrnewsymbol[%
    description={outer normal fan}, sort={NP}]%
    {NP}%
    {\ensuremath{\mathcal N_P}}

\glsxtrnewsymbol[%
    description={height function for regular subdivisions}, sort={omega}]%
    {omega}%
    {\ensuremath{\omega}}

\glsxtrnewsymbol[%
    description={regular subdivision of $\Aa$ w.r.t. $\omega$}, sort={SsAaomega}]%
    {regSubdivision}%
    {\ensuremath{\Ss(\Aa,\omega)}}

\glsxtrnewsymbol[%
    description={weight of a triangulation $\Tt$}, sort={phiT}]%
    {weight}%
    {\ensuremath{\varphi_\Tt}}

\glsxtrnewsymbol[%
    description={secondary polytope of $A$}, sort={SigmaA}]%
    {SecPoly}%
    {\ensuremath{\Sigma(A)}}

\glsxtrnewsymbol[%
    description={secondary cone of $\Tt$ in $\Aa$}, sort={SigmaCA}]%
    {SecCone}%
    {\ensuremath{\Sigma C(\Aa,\Tt)}}

\glsxtrnewsymbol[%
    description={secondary fan of $\Aa$}, sort={SigmaFA}]%
    {SecFan}%
    {\ensuremath{\Sigma F(\Aa)}}

\glsxtrnewsymbol[%
    description={Gale projection map}, sort={betaomega}]%
    {betaMap}%
    {\ensuremath{\beta(\omega)}}

\glsxtrnewsymbol[%
    description={mixed $(A_0,\ldots,A_n)$-resultant}, sort={RAAAAf}]%
    {mixedARes}%
    {\ensuremath{R_{A_0,\ldots,A_n}(f_0,\ldots,f_n)}}

\glsxtrnewsymbol[%
    description={$A$-resultant}, sort={RAf}]%
    {ARes}%
    {\ensuremath{R_A(f_0,\ldots,f_n)}}

\glsxtrnewsymbol[%
    description={$A$-discriminant}, sort={DeltaA}]%
    {ADisc}%
    {\ensuremath{\Delta_A(f)}}

\glsxtrnewsymbol[%
    description={reduced $A$-discriminant}, sort={DeltaB}]%
    {redADisc}%
    {\ensuremath{\Delta_\Bb(f)}}

\glsxtrnewsymbol[%
    description={effective variables}, sort={y}]%
    {effectiveVars}%
    {\ensuremath{y_1,\ldots,y_r}}

\glsxtrnewsymbol[%
    description={\HKP}, sort={psi}]%
    {HKP}%
    {\ensuremath{\psi[t_1:\ldots:t_r]}}

\glsxtrnewsymbol[%
    description={principal $A$-determinant}, sort={EA}]%
    {pAdet}%
    {\ensuremath{E_A(f)}}

\glsxtrnewsymbol[%
    description={simple principal $A$-determinant}, sort={EAhat}]%
    {spAdet}%
    {\ensuremath{\widehat E_A(f)}}

\glsxtrnewsymbol[%
    description={Weyl algebra}, sort={D}]%
    {Weyl}%
    {\ensuremath{D}}

\glsxtrnewsymbol[%
    description={solution space of a $D$-ideal $I$}, sort={Sol}]%
    {Sol}%
    {\ensuremath{\Sol(I)}}

\glsxtrnewsymbol[%
    description={order of elements in a Weyl algebra $D$}, sort={ord}]%
    {order}%
    {\ensuremath{\ord(P)}}

\glsxtrnewsymbol[%
    description={associated graded ring}, sort={gr}]%
    {assocGrRing}%
    {\ensuremath{\gr_{(u,v)}(D)}}

\glsxtrnewsymbol[%
    description={initial ideal/initial symbol w.r.t. weight $(u,v)$}, sort={in}]%
    {initialuv}%
    {\ensuremath{\initial_{(u,v)}(I)}}

\glsxtrnewsymbol[%
    description={characteristic variety of the ideal $I$}, sort={char}]%
    {char}%
    {\ensuremath{\ch(I)}}

\glsxtrnewsymbol[%
    description={singular locus of $I$}, sort={Sing}]%
    {Sing}%
    {\ensuremath{\Sing(I)}}

\glsxtrnewsymbol[%
    description={Pochhammer symbol}, sort={Pochhammer}]%
    {Pochhammer}%
    {\ensuremath{(a)_k}}

\glsxtrnewsymbol[%
    description={lattice of $\Aa$, integer kernel of $\Aa$}, sort={LL}]%
    {Ll}%
    {\ensuremath{\mathbb L}}

\glsxtrnewsymbol[%
    description={toric operators of $\Aa$-hypergeometric ideal}, sort={box}]%
    {toricOp}%
    {\ensuremath{\square_l}}

\glsxtrnewsymbol[%
    description={homogeneous operators of $\Aa$-hypergeometric ideal}, sort={Eibeta}]%
    {homOp}%
    {\ensuremath{E_i(\beta)}}

\glsxtrnewsymbol[%
    description={GKZ parameter}, sort={beta}]%
    {GKZpar}%
    {\ensuremath{\beta}}

\glsxtrnewsymbol[%
    description={$\Aa$-hypergeometric ideal}, sort={HA}]%
    {GKZIdeal}%
    {\ensuremath{H_\Aa(\beta)}}

\glsxtrnewsymbol[%
    description={toric ideal}, sort={IA}]%
    {toricId}%
    {\ensuremath{I_\Aa}}

\glsxtrnewsymbol[%
    description={$\Aa$-hypergeometric module}, sort={MAbeta}]%
    {GKZModule}%
    {\ensuremath{\mathcal M_\Aa(\beta)}}

\glsxtrnewsymbol[%
    description={$\Gamma$-series}, sort={phigamma}]%
    {gammaSer}%
    {\ensuremath{\varphi_\gamma(z)} / \ensuremath{\varphi_{\sigma,k}(z)}}

\glsxtrnewsymbol[%
    description={summation region of $\Gamma$-series}, sort={Lambdak}]%
    {Lambdak}%
    {\ensuremath{\Lambda_k}}

\glsxtrnewsymbol[%
    description={set of representatives defining a partition of $\mathbb N_0^r$}, sort={Ksigma}]%
    {Ksigma}%
    {\ensuremath{K_\sigma}}

\glsxtrnewsymbol[%
    description={normalized $\Gamma$-series}, sort={phigammaNor}]%
    {norGammaSer}%
    {\ensuremath{\phi_{\sigma,k}(z)} / \ensuremath{\phi_{\sigma}(z)}}

\glsxtrnewsymbol[%
    description={vectors in Minkowski space}, sort={pq}]%
    {mink}%
    {\ensuremath{\mink p, \mink q}}

\glsxtrnewsymbol[%
    description={scalar Feynman integral with Euclidean kinematics, depending on a Feynman graph $\Gamma$, the spacetime dimension $d$, the indices $\nu$, the external momenta $p$ and the internal masses $m$}, sort={IIGamma}]%
    {FI}%
    {\ensuremath{\mathcal I_\Gamma(d,\nu,p,m)}}

\glsxtrnewsymbol[%
    description={indices of a Feynman integral, i.e.\ powers of propagators}, sort={nu}]%
    {nu}%
    {\ensuremath{\nu_1,\ldots,\nu_n}}
    
\glsxtrnewsymbol[%
    description={external momenta}, sort={p}]%
    {externalMomenta}%
    {\ensuremath{p_1,\ldots,p_m}}

\glsxtrnewsymbol[%
    description={internal masses}, sort={m}]%
    {masses}%
    {\ensuremath{m_1,\ldots,m_n}}

\glsxtrnewsymbol[%
    description={incidence matrix of a graph $\Gamma$}, sort={I}]%
    {incidence}%
    {\ensuremath{I}}

\glsxtrnewsymbol[%
    description={number of components of a graph $\Gamma$ / \nth{0}-Betti number}, sort={b0}]%
    {b0}%
    {\ensuremath{b_0}}

\glsxtrnewsymbol[%
    description={the set of all spanning $k$-forests of a graph $\Gamma$}, sort={Tk}]%
    {spanT}%
    {\ensuremath{\mathscr T_k}}

\glsxtrnewsymbol[%
    description={the set of loops (circuits) of a graph $\Gamma$}, sort={C}]%
    {circuits}%
    {\ensuremath{\mathscr C_1,\ldots,\mathscr C_r}}

\glsxtrnewsymbol[%
    description={chord set, i.e.\ the complement of a tree $T$}, sort={Tstar}]%
    {chord}%
    {\ensuremath{T^*}}

\glsxtrnewsymbol[%
    description={fundamental loop matrix w.r.t. a chord set $T^*$}, sort={CTstar}]%
    {fundloop}%
    {\ensuremath{C_{T^*}}}

\glsxtrnewsymbol[%
    description={the number of loops of a graph $\Gamma$ / the first Betti number of $\Gamma$}, sort={L}]%
    {loop}%
    {\ensuremath{L}}

\glsxtrnewsymbol[%
    description={Schwinger parameters}, sort={x}]%
    {Schwingers}%
    {\ensuremath{x_1,\ldots,x_n}}

\glsxtrnewsymbol[%
    description={independent loop momenta}, sort={k}]%
    {loopmomenta}%
    {\ensuremath{k_1,\ldots,k_L}}

\glsxtrnewsymbol[%
    description={the internal momenta for all internal edges $e_1,\ldots,e_n$ of a Feynman graph $\Gamma$, i.e.\ a linear combination of loop momenta and external momenta }, sort={q}]%
    {internalmomenta}%
    {\ensuremath{q_1,\ldots,q_n}}

\glsxtrnewsymbol[%
    description={the part of internal momenta coming from external momenta}, sort={qhat}]%
    {internalmomentaHat}%
    {\ensuremath{\widehat{q_1},\ldots,\widehat{q_n}}}
    
\glsxtrnewsymbol[%
    description={first Symanzik polynomial}, sort={Uu}]%
    {Uu}%
    {\ensuremath{\Uu(x)}}

\glsxtrnewsymbol[%
    description={second Symanzik polynomial}, sort={Ff}]%
    {Ff}%
    {\ensuremath{\Ff(x)}}

\glsxtrnewsymbol[%
    description={massless part of the second Symanzik polynomial}, sort={Ff0}]%
    {Ff0}%
    {\ensuremath{\Ff_0(x)}}

\glsxtrnewsymbol[%
    description={a $L\times L$ matrix to define Symanzik polynomials}, sort={M}]%
    {M}%
    {\ensuremath{M}}

\glsxtrnewsymbol[%
    description={a $L$-vector to define Symanzik polynomials}, sort={Q}]%
    {Q}%
    {\ensuremath{Q}}

\glsxtrnewsymbol[%
    description={an expression to define Symanzik polynomials}, sort={J}]%
    {J}%
    {\ensuremath{J}}

\glsxtrnewsymbol[%
    description={inverse propagators}, sort={D1Dn}]%
    {propagators}%
    {\ensuremath{D_1,\ldots,D_n}}

\glsxtrnewsymbol[%
    description={weighted sum of all inverse propagators}, sort={Lambda}]%
    {Lambda}%
    {\ensuremath{\Lambda(k,p,x)}}

\glsxtrnewsymbol[%
    description={adjugate matrix of $M$, i.e.\ the transpose of the cofactor matrix}, sort={Adj}]%
    {adjugate}%
    {\ensuremath{\Adj(M)}}

\glsxtrnewsymbol[%
    description={superficial degree of divergence}, sort={omega}]%
    {sdd}%
    {\ensuremath{\omega}}

\glsxtrnewsymbol[%
    description={Dirac's $\delta$-distribution}, sort={DeltaDirac}]%
    {DiracDelta}%
    {\ensuremath{\delta(x)}}

\glsxtrnewsymbol[%
    description={Lee-Pomeransky polynomial, i.e.\ the sum of first and second Symanzik polynomial $\Gg=\Uu+\Ff$}, sort={Gg}]%
    {Gg}%
    {\ensuremath{\Gg(x)}}

\glsxtrnewsymbol[%
    description={vector of spacetime dimension and indices $\nuu = (\nu_0,\nu)$ with $\nu_0=\dhalf$, GKZ parameter}, sort={nuu}]%
    {nuu}%
    {\ensuremath{\nuu}}

\glsxtrnewsymbol[%
    description={generalized Feynman integral}, sort={IIA}]%
    {gFI}%
    {\ensuremath{\mathcal I_\Aa(\nuu,z)}}

\glsxtrnewsymbol[%
    description={coefficients in Symanzik polynomials / variables of the generalized Feynman integral}, sort={z}]%
    {z}%
    {\ensuremath{z_1,\ldots,z_N}}

\glsxtrnewsymbol[%
    description={generalized Feynman integral without prefactors}, sort={JJA}]%
    {gFJ}%
    {\ensuremath{\mathcal J_\Aa(\nuu,z)}}

\glsxtrnewsymbol[%
    description={extended Feynman integral for tensor reduction}, sort={IIGammaTilde}]%
    {tildeFI}%
    {\ensuremath{\widetilde{\mathcal I}_\Gamma(d,\nu,p,m,c)}}

\glsxtrnewsymbol[%
    description={holomorphic Feynman integral}, sort={PhiA}]%
    {PhiMero}%
    {\ensuremath{\Phi_\Aa(\nuu,z)}}

\glsxtrnewsymbol[%
    description={coefficients in the linear combination of basis solutions of $\Sol(H_\Aa(\nuu))$, meromorphic functions in $\nuu$}, sort={Csigmak}]%
    {Cfactors}%
    {\ensuremath{C_{\sigma,k}(\nuu)}}

\glsxtrnewsymbol[%
    description={(unsigned) Stirling numbers of the first kind}, sort={StirlingFirst}]%
    {StirlingFirst}%
    {\ensuremath{\StirlingFirst{n}{k}}}

\glsxtrnewsymbol[%
    description={$k$-th harmonic number}, sort={Hnk}]%
    {harmonicNumber}%
    {\ensuremath{H_n^{(k)}}}

\glsxtrnewsymbol[%
    description={floor function, i.e.\ the greatest integer less or equal $x$}, sort={floor}]%
    {floor}%
    {\ensuremath{\lfloor x\rfloor}}

\glsxtrnewsymbol[%
    description={$Z$-sum}, sort={Zsum}]%
    {Zsum}%
    {\ensuremath{\Zsum{n}{m_1, \ldots, m_k} (y_1, \ldots, y_k)}}

\glsxtrnewsymbol[%
    description={$Z$-sum at unity argument $\Zsum{n}{m_1, \ldots, m_k} := \Zsum{n}{m_1, \ldots, m_k} (1, \ldots, 1)$}, sort={Zsum1}]%
    {Zsum1}%
    {\ensuremath{\Zsum{n}{m_1, \ldots, m_k}}}

\glsxtrnewsymbol[%
    description={$k$-th forward finite difference}, sort={Deltaf}]%
    {fd}%
    {\ensuremath{\fd[k]{f(i)}}}

\glsxtrnewsymbol[%
    description={period integral}, sort={PGamma}]%
    {period}%
    {\ensuremath{\mathcal P_\Gamma(\nuu)}}

\glsxtrnewsymbol[%
    description={marginal Feynman integral}, sort={KGamma}]%
    {marginal}%
    {\ensuremath{\mathcal K_\Gamma(\nuu,z)}}

\glsxtrnewsymbol[%
    description={$S$-matrix, the operator describing the scattering}, sort={S}]%
    {Smatrix}%
    {\ensuremath{S}}

\glsxtrnewsymbol[%
    description={$T$-matrix, the operator defining the non-trivial part of scattering}, sort={T}]%
    {Tmatrix}%
    {\ensuremath{T}}
    
\glsxtrnewsymbol[%
    description={codimension $1$ part of the Landau variety $\mathcal L(\mathcal I_\Gamma)$}, sort={LlGamma}]%
    {Landau1}%
    {\ensuremath{\mathcal L_1(\mathcal I_\Gamma)}}

\glsxtrnewsymbol[%
    description={subgraph polynomial associated to $I$}, sort={FIx}]%
    {FIx}%
    {\ensuremath{\Ff_I(x)}}

\glsxtrnewsymbol[%
    description={simple principal $A$-determinant restricted to physically relevant contributions}, sort={EAhatph}]%
    {pAdetPh}%
    {\ensuremath{\widehat E_{\Aa_\Ff}^{ph}(\Ff)}}

\glsxtrnewsymbol[%
    description={codimension $1$ part of a Landau variety restricted to physically relevant contributions}, sort={LlGammaph}]%
    {Landau1Ph}%
    {\ensuremath{\mathcal L_1^{ph}(\mathcal I_\Gamma)}}

\glsxtrnewsymbol[%
    description={polynomial defining the mixed type singularities of proper faces}, sort={R}]%
    {R}%
    {\ensuremath{R}}

\glsxtrnewsymbol[%
    description={Feynman integral with Feynman's $i\varepsilon$ prescription}, sort={IIGammaEps}]%
    {FeynEps}%
    {\ensuremath{\mathcal I_\Gamma^\varepsilon (d,\nu,p,m) }}

\glsxtrnewsymbol[%
    description={discontinuity of Feynman integrals at the branch cut defined by $\varepsilon =0$}, sort={Disc}]%
    {Disc}%
    {\ensuremath{\Disc \mathcal I_\Gamma(d,\nu,p,m)}}

\glsxtrnewsymbol[%
    description={generalized Feynman integral with Feynman's $i\varepsilon$ prescription}, sort={IIAEps}]%
    {FeynGenEps}%
    {\ensuremath{\mathcal I_\Aa^\varepsilon(\nuu,z)}}

\glsxtrnewsymbol[%
    description={(multivariate) Euler-Mellin integral}, sort={Mf}]%
    {EM}%
    {\ensuremath{\mathscr M_f(s,t)}}

\glsxtrnewsymbol[%
    description={$\theta$-analogue Euler-Mellin integral}, sort={Mftheta}]%
    {EMtheta}%
    {\ensuremath{\mathscr M_f^\theta(s,t)}}

\glsxtrnewsymbol[%
    description={the component-wise argument map of a complex number $x\in\mathbb C^n$}, sort={Arg}]%
    {Arg}%
    {\ensuremath{\Arg(x)}}

\glsxtrnewsymbol[%
    description={zero locus of $f$}, sort={Zzf}]%
    {Zf}%
    {\ensuremath{\mathcal Z_f}}

\glsxtrnewsymbol[%
    description={coamoeba of $f$}, sort={Ccf}]%
    {Coamoeba}%
    {\ensuremath{\mathcal C_f}}

\glsxtrnewsymbol[%
    description={$n$-dimensional real torus}, sort={Tn}]%
    {RTorus}%
    {\ensuremath{\mathbb T^n}}

\glsxtrnewsymbol[%
    description={complement of the closure of the coamoeba $\mathbb T^n\setminus\overline{\mathcal C_f}$}, sort={Ccfcom}]%
    {Cc}%
    {\ensuremath{\Cc}}
    
\glsxtrnewsymbol[%
    description={shell of the coamoeba $\mathcal C_f$}, sort={Hf}]%
    {shell}%
    {\ensuremath{\mathcal H_f}}
    
\glsxtrnewsymbol[%
    description={lopsided coamoeba of $f$}, sort={LCcf}]%
    {lop}%
    {\ensuremath{\mathcal L \mathcal C_f}}
    
\glsxtrnewsymbol[%
    description={shell of the lopsided coamoeba $\mathcal L\mathcal C_f$}, sort={LHf}]%
    {lopshell}%
    {\ensuremath{\mathcal L \mathcal H_f}}

\glsxtrnewsymbol[%
    description={$S$-sum}, sort={Ssum}]%
    {Ssum}%
    {\ensuremath{\Ssum{n}{m_1, \ldots, m_k} (y_1, \ldots, y_k)}}

\glsxtrnewsymbol[%
    description={$S$-sum at unity argument $\Ssum{n}{m_1, \ldots, m_k} := \Ssum{n}{m_1, \ldots, m_k} (1, \ldots, 1)$}, sort={Ssum1}]%
    {Ssum1}%
    {\ensuremath{\Ssum{n}{m_1, \ldots, m_k}}}

\glsxtrnewsymbol[%
    description={numerically stable variant of Stirling numbers of the first kind}, sort={sigmank}]%
    {sigmank}%
    {\ensuremath{\sigma(n,k)}}

\glsxtrnewsymbol[%
    description={Stirling numbers of the second kind}, sort={StirlingSecond}]%
    {StirlingSecond}%
    {\ensuremath{\StirlingSecond{n}{k}}}    
    
\newcommand{\topologyDescr}[5]{\multirow{#1}{*}{\begin{tabular}{>{\centering\arraybackslash}p{.9cm} >{\centering\arraybackslash}p{.9cm} >{\centering\arraybackslash}p{.9cm}} \multicolumn{3}{c}{#2} \\ $ #3 $ & $ #4 $ & $ #5 $ \end{tabular}}}
\newcommand{\topologyDescrTikz}[7]{
    \multirow{#1}{*}{
        \makecell{
            \begin{minipage}{0.3cm}
                \hspace{-0.7cm}\rotatebox{90}{\textit{#2}}
            \end{minipage}%
            \begin{minipage}{2cm}
                \centering \begin{tikzpicture}[thick, dot/.style = {draw, shape=circle, fill=black, scale=.3}, every node/.style={scale=0.8}, scale=#7] #6 
                \end{tikzpicture}
            \end{minipage} \\[1cm] 
            \vspace*{\fill}
            \begin{tabular}{@{}>{\centering\arraybackslash}p{.9cm} >{\centering\arraybackslash}p{.9cm} >{\centering\arraybackslash}p{1cm}@{}}
                $ L=#3 $ & $ n=#4 $ & $ m=#5 $
            \end{tabular}
        }%
    }
}

\newcommand{\SymanzikSpec}[4]{\multirow{2}{*}{$ #1 $} & \multirow{2}{*}{$ #2 $} &\multirow{2}{*}{$ #3 $} &\multirow{2}{*}{$ #4 $} &} 
\newcommand{\mmultirow}[1]{\multirow{2}{*}{$ #1 $}} 
\newcommand{\tmultirow}[1]{\multirow{2}{*}{#1}} 
\newcommand{\triangsF}[6]{#1 & #2 & #3 & #4 & #5 \\} 
\newcommand{\triangsG}[5]{& & & & & & #1 & #2 & #3 & #4 & #5 \\}
\newcommand{\mNA}{\textsuperscript{$\dagger$}}
\newcommand{\mNR}{\textsuperscript{$\ddagger$}}

\title{Hypergeometric Feynman Integrals}
\author{René Pascal Klausen}
\date{\today}

\begin{document}


\newgeometry{bindingoffset=10mm, hmargin=25mm, vmargin=25mm} 
\frontmatter

\renewcommand{\thefootnote}{\fnsymbol{footnote}}
\begin{titlepage}
    
    {\centering
        \rule{\textwidth}{1.6pt}\vspace*{-\baselineskip}\vspace*{2pt}
        \rule{\textwidth}{0.4pt}\\[\baselineskip]
        {\Huge\MakeUppercase{Hypergeometric Feynman \\[0.2\baselineskip] integrals}}\\[0.4\baselineskip]
        \rule{\textwidth}{0.4pt}\vspace*{-\baselineskip}\vspace{3.2pt}
        \rule{\textwidth}{1.6pt}\\[1.5\baselineskip]
        {\Large René Pascal Klausen \!\raisebox{0.5em}{\scalebox{0.75}{\href{mailto:klausen@physik.hu-berlin.de}{\faEnvelope[regular]}}}\hspace{2pt}, \nth{25} February 2023 \par}
        \vspace{2\baselineskip}
    }
    
    {\noindent This is a minor updated version of my PhD thesis  (date of the original version \nth{3} August 2022), which I defend at the Institute of Physics at Johannes Gutenberg University Mainz on the \nth{24} of November 2022. The referees were Christian Bogner (Johannes Gutenberg University Mainz) and Dirk Kreimer (Humboldt University of Berlin).\par}
    \vspace{\baselineskip}
    
    \vfill
    {\centering\bfseries\Large Abstract}
    \vspace{\baselineskip}
    \currentpdfbookmark{Abstract}{bm:abstract}

    In this thesis we will study Feynman integrals from the perspective of $\Aa$-hypergeometric functions, a generalization of hypergeometric functions which goes back to Gelfand, Kapranov, Zelevinsky (GKZ) and their collaborators. This point of view was recently initiated by the works \cite{DeLaCruzFeynmanIntegralsAhypergeometric2019} and \cite{KlausenHypergeometricSeriesRepresentations2019}. Inter alia, we want to provide here a concise summary of the mathematical foundations of $\Aa$-hypergeometric theory in order to substantiate this viewpoint. This overview will concern aspects of polytopal geometry, multivariate discriminants as well as holonomic $D$-modules.

    As we will subsequently show, every scalar Feynman integral is an $\Aa$-hypergeometric function. Furthermore, all coefficients of the Laurent expansion as appearing in dimensional and analytical regularization can be expressed by $\Aa$-hypergeometric functions as well. By applying the results of GKZ we derive an explicit formula for series representations of Feynman integrals. Those series representations take the form of Horn hypergeometric functions and can be obtained for every regular triangulation of the Newton polytope $\Newt(\Uu+\Ff)$ of the sum of Symanzik polynomials. Those series can be of higher dimension, but converge fast for certain kinematical regions, which also allows an efficient numerical application. We will sketch an algorithmic approach which evaluates Feynman integrals numerically by means of these series representations. Further, we will examine possible issues which can arise in a practical usage of this approach and provide strategies to solve them. As an illustrative example we will present series representations for the fully massive sunset Feynman integral.

    Moreover, the $\Aa$-hypergeometric theory enables us to give a mathematically rigorous description of the analytic structure of Feynman integrals (also known as Landau variety) by means of principal $A$-determinants and $A$-discriminants. This description of the singular locus will also comprise the various second-type singularities. Furthermore, we will find contributions to the singular locus occurring in higher loop diagrams, which seem to have been overlooked in previous approaches. By means of the \HKP we also provide a very efficient way to determine parameterizations of Landau varieties. We will illustrate those methods by determining the Landau variety of the dunce's cap graph. We furthermore present a new approach to study the sheet structure of multivalued Feynman integrals by use of coamoebas.
\end{titlepage}
\renewcommand{\thefootnote}{\arabic{footnote}}

\pagestyle{empty}
\restoregeometry
\origdoublepage


\pagestyle{plain}

\cleardoublepage


\currentpdfbookmark{Contents}{bm:toc}

\tableofcontents

\clearpage

\currentpdfbookmark{List of figures}{bm:tof}
\listoffigures

\vfill

The figures in this thesis were generated with \softwareName{TikZ} \cite{TantauTikZPGFManual2021} and \softwareName{Mathematica} \cite{WolframResearchIncMathematicaVersion12}.


%
%
%
%
%
%
%
%
%
%
%
%
%

\pagestyle{withoutSections}

\mainmatter

\chapter{Introduction} \insertlofspace \label{ch:Introduction}

%

The modern physics' perspective to the fundamental interactions of nature is formulated in terms of quantum fields theories (QFTs). Tacitly, one often assumes perturbative quantum field theories when talking about QFT. This is because, with the exception of a few toy models, QFTs are only accessible perturbatively, i.e.\ we will consider interactions in the infinitesimal neighbourhood of free QFTs. In the framework of perturbative QFTs, Feynman integrals are indispensable building blocks for almost every prediction within these theories.

Hence, we are confronted with the problem of solving a huge number of Feynman integrals. But what does it actually mean to ``solve'' a Feynman integral? Of course, one is interested in a numerical result when comparing it with experimental data. However, on all the intermediate steps to this final result, analytical solutions of those integrals are preferred. This is less a question of elegance and more a need to understand the structure beyond Feynman integrals\footnote{Analytical expressions are also preferred in the renormalization procedure. However, one can also consider renormalized Feynman integrals so that an analytical intermediate solution can in principle be omitted.}. Hence, to get a deeper understanding of Feynman integrals and their role in perturbative QFTs, it is essential to investigate their functional relationships. Thus, ``solving Feynman integrals'' actually stands for rewriting Feynman integrals in terms of other functions. Clearly, the goal is to know more about these rewritten functions, and we should be able to efficiently evaluate these functions numerically.

A very successful example of a function class for rewriting Feynman integrals are the so-called multiple polylogarithms and their generalizations, e.g.\ elliptic polylogarithms \cite{RemiddiHarmonicPolylogarithms2000, GoncharovMultiplePolylogarithmsMixed2001, BrownMasslessHigherloopTwopoint2009, PanzerFeynmanIntegralsHyperlogarithms2015, BognerSymbolicIntegrationMultiple2012, BognerFeynmanIntegralsIterated2015, HiddingAllOrdersStructure2019}. Thus, multiple polylogarithms and related functions appear for many Feynman integrals as coefficients in a Laurent expansion in dimensional and analytical regularization. However, this applies not for all Feynman integrals. This is our main motivation for proposing another class of functions for rewriting Feynman integrals herein: $\Aa$-hypergeometric functions. These functions are in general less easy to handle than multiple polylogarithms, but we will show that every Feynman integral can be treated within this functional class. Thereby, the coefficients in a Laurent expansion of the Feynman integrals as well as the Feynman integral itself belong to the class of $\Aa$-hypergeometric functions. In doing so, we will take a closer look at two aspects in this thesis: On the one hand, we want to indicate areas where we can benefit from knowledge about this class of functions in the investigation of Feynman integrals and on the other hand we will discuss possibilities for an efficient numerical evaluation.\bigskip

Describing Feynman integrals by hypergeometric functions is by no means a new idea. Already in the early days of calculating Feynman amplitudes, it was proposed by Regge to consider Feynman integrals as a kind of generalized hypergeometric functions \cite{ReggeAlgebraicTopologyMethods1968}, where the singularities of those hypergeometric functions coincide with the Landau variety. Later on Kashiwara and Kawai \cite{KashiwaraHolonomicSystemsLinear1976} showed that Feynman integrals satisfy indeed holonomic differential equations, where the singularities of those holonomic differential equations are determined by the Landau variety.

Apart from characterizing the Feynman integral by ``hypergeometric'' partial differential equation systems, many applications determine the Feynman integral as a generalized hypergeometric series. Usually, the often used Mellin-Barnes approach \cite{BergereAsymptoticExpansionFeynman1974,SmirnovFeynmanIntegralCalculus2006} results in Pochhammer series ${}_pF_q$, Appell functions, Lauricella functions and related functions by applying the residue theorem \cite{BoosMethodCalculatingMassive1991}. Furthermore, for arbitrary one-loop Feynman integrals it is known that they can always be represented by a small set of hypergeometric series \cite{FleischerNewHypergeometricRepresentation2003,PhanScalar1loopFeynman2019}.
Thirdly, the Feynman integral may be expressed by ``hypergeometric'' integrals like the generalized Meijer $G$- or Fox $H$-functions \cite{BuschmanFunctionAssociatedCertain1990, Inayat-HussainNewPropertiesHypergeometric1987, Inayat-HussainNewPropertiesHypergeometric1987a}. The connections between Feynman integrals and hypergeometric functions was investigated over decades and a summary of these results can be found in \cite{KalmykovHypergeometricFunctionsTheir2008}.

Thus, there arise three different notions of the term ``hypergeometric'' in the Feynman integral calculus, where every notion generalizes different characterizations of the ordinary Gaussian hypergeometric function ${_2}F_1(a,b,c|z)$. In the late 1980s Gelfand, Kapranov, Zelevinsky (GKZ) and collaborators \cite{GelfandCollectedPapersVol1989, GelfandHolonomicSystemsEquations1988, GelfandDiscriminantsPolynomialsSeveral1991, GelfandGeneralizedEulerIntegrals1990, GelfandGeneralHypergeometricSystems1992, GelfandDiscriminantsResultantsMultidimensional1994, GelfandHypergeometricFunctionsToric1991, GelfandAdiscriminantsCayleyKoszulComplexes1990, GelfandHypergeometricFunctionsToral1989, GelfandDiscriminantsPolynomialsMany1990, GelfandEquationsHypergeometricType1988, GelfandNewtonPolytopesClassical1990} were starting to develop a comprehensive method to generalize the notion of ``hypergeometric'' functions in a consistent way\footnote{Indeed, Feynman integrals from QED were one of the motivations for Gelfand to develop generalized hypergeometric differential equations. However, this connection does not seem to have been pursued further by Gelfand. Moreover, it was indicated in \cite{GolubevaReggeGelfandProblem2014}, that also Regge's conjecture \cite{ReggeAlgebraicTopologyMethods1968} was influenced by Gelfand.}. Those functions are called $\Aa$-hypergeometric functions and are defined by a special holonomic system of partial differential equations. As Gelfand, Kapranov and Zelevinsky illustrated with Euler integrals, the GKZ approach not only generalizes the concept of hypergeometric functions but can also be used for analyzing and solving integrals \cite{GelfandGeneralizedEulerIntegrals1990}.\bigskip
    
For physicists the GKZ perspective is not entirely new. Already in the 1990s, string theorists applied the $\Aa$-hypergeometric approach in order to calculate certain period integrals and worked out the mirror symmetry \cite{HosonoMirrorSymmetryMirror1995, HosonoMirrorSymmetryMirror1995a}. Recently, the GKZ approach was also used to obtain differential equations for the Feynman integral from the maximal cut \cite{VanhoveFeynmanIntegralsToric2018}. Furthermore, $\Aa$-hypergeometric functions can also be used in several other branches of physics \cite{GolubevaReggeGelfandProblem2014}. Still, the approach of Gelfand, Kapranov and Zelevinsky is no common practice among physicists.\bigskip

However, Feynman integrals have recently begun to be considered from the perspective of GKZ. In 2016 Nasrollahpoursamami showed that the Feynman integral satisfies a differential equation system which is isomorphic to a GKZ system \cite{NasrollahpoursamamiPeriodsFeynmanDiagrams2016}. Independently of each other, the fact that Feynman integrals are $\Aa$-hypergeometric was also shown directly in 2019 by de la Cruz \cite{DeLaCruzFeynmanIntegralsAhypergeometric2019} and Klausen \cite{KlausenHypergeometricSeriesRepresentations2019} based on the Lee-Pomeransky representation \cite{LeeCriticalPointsNumber2013} of the Feynman integral. Furthermore, explicit series representations for all Feynman integrals admitting unimodular triangulations of $\Newt(\Uu+\Ff)$ were given \cite{KlausenHypergeometricSeriesRepresentations2019}. In the manner of these two works \cite{DeLaCruzFeynmanIntegralsAhypergeometric2019,KlausenHypergeometricSeriesRepresentations2019}, a number of examples were presented in detail by \cite{FengGKZhypergeometricSystemsFeynman2020}. Shortly afterwards, a series of articles considered Feynman integrals by use of GKZ methods developed in string theory \cite{KlemmLoopBananaAmplitude2020, BonischFeynmanIntegralsDimensional2021, BonischAnalyticStructureAll2021}, which were applied mainly to the banana graph family. Moreover, the Landau variety of Feynman integrals was considered from the $\Aa$-hypergeometric perspective for banana graphs in \cite{BonischAnalyticStructureAll2021} and for arbitrary Feynman graphs in \cite{KlausenKinematicSingularitiesFeynman2022}. \bigskip

This thesis builds on the two mentioned articles \cite{KlausenHypergeometricSeriesRepresentations2019} and \cite{KlausenKinematicSingularitiesFeynman2022}. However, we will include several major generalizations of these works, and we also want to provide a more substantial overview of the mathematical theory behind it.

In particular, we will show that every Feynman integral belongs to the class of $\Aa$-hypergeometric functions. Moreover, every coefficient in a Laurent expansion as appearing in dimensional or analytical regularization can be expressed by $\Aa$-hypergeometric functions. Furthermore, we will give an explicit formula for a multivariate Horn series representation of generalized Feynman integrals for every  regular (unimodular as well as non-unimodular) triangulation. Inter alia, this allows to evaluate Feynman integrals numerically very efficiently for convenient kinematic regions. 

Further, we will connect the Landau variety with principal $A$-determinants and $A$-discriminants. This allows us to describe Landau varieties, including second-type singularities, with mathematical rigour. In doing so, we will find certain contributions to the singular locus of Feynman integrals in higher loops which seem to have been overlooked so far. The simplest example where those additional contributions appear is the so-called dunce's cap graph. By its connection to the $\Aa$-hypergeometric theory we can also give a very efficient parameterization of Landau varieties. In addition, we will sketch an approach to describe parts of the monodromy structure of Feynman integrals by means of a simpler geometric object, which is known as coamoeba. \bigskip

We will start this thesis with a comprehensive summary of the mathematical fundament in \cref{ch:AHypergeometricWorld}, which in our opinion is essential for understanding the following chapters. We will try to keep this chapter as short as possible without compromising understanding. For hurried readers we recommend \cref{sec:AhypergeometricNutshell}, which summarizes the main aspects of \cref{ch:AHypergeometricWorld}. After a mathematical introduction, we will continue with a physical introduction in \cref{ch:FeynmanIntegrals}. Thereby, we will focus mainly on the various representations of Feynman integrals in parametric space and recall several aspects of them. For experienced readers, \cref{sec:FeynmanIntegralsAsAHyp} may be sufficient, which connects Feynman integrals and $\Aa$-hypergeometric functions. 

In \cref{ch:seriesRepresentations} we will introduce series representations of Feynman integrals based on $\Aa$-hypergeometric theory. Apart from general questions, we will also discuss some features as well as possible difficulties which can arise in the evaluation of Feynman integrals by means of those series representations. In particular, we also explain the steps that would be necessary in an algorithmic implementation. \Cref{ch:singularities} is devoted to the study of Landau varieties (or more generally to the singular locus) of Feynman integrals from the $\Aa$-hypergeometric perspective. In this chapter we will develop a mathematically rigorous description of the singular locus by means of the principal $A$-determinant and introduce the coamoeba.

\clearpage


\currentpdfbookmark{Acknowledgements}{bm:acknowledgements}

\vspace*{3\baselineskip}
\section*{\centering Acknowledgements}
\vspace{\baselineskip}
I would like to express my deepest gratitude to Christian Bogner for his great encouragement and support, which led to this thesis. I owe him the invaluable freedom to work on my own research topic, and I am grateful to him for always encouraging and supporting me in all my ideas. Furthermore, I also wish to thank him for his time to supervise my work, despite difficult circumstances. I would like to extend my sincere thanks to Dirk Kreimer for warmly welcoming me into his group at Humboldt university, for never wavering in his support and for his help in all organizational issues.

I very much appreciated the stimulating email discussions with Uli Walther about general questions on $\Aa$-hypergeometric theory and I would like to thank him for his invitation to the $\Aa$-hypergeometric conference. As well, I have also benefited from the discussions with Christoph Nega and Leonardo de la Cruz about their personal perspective on GKZ. Moreover, I also benefited from the discussions with Konrad Schultka about toric geometry, with David Prinz about renormalization, the discussions with Marko Berghoff and Max Mühlbauer as well as with Erik Panzer about Landau varieties and related questions. Further, I wish to thank Ruth Britto for the enlightening discussions and for inviting me to Dublin. I would also like to thank Stefan Theisen for his constructive comments on series representations of $\Aa$-hypergeometric functions. I also wish to express my sincere thanks to Erik Panzer and Michael Borinsky for initiating the very inspiring conference about ``Tropical and Convex Geometry and Feynman integrals'' at ETH Zürich. 

This research was supported by the International Max Planck Research School for Mathematical and Physical Aspects of Gravitation, Cosmology and Quantum Field Theory and by the cluster of excellence ``Precision Physics, Fundamental Interactions and Structure of Matter'' (PRISMA+) at Johannes-Gutenberg university of Mainz.

\newpage


\pagestyle{fancyStandard}

\chapter{The \texorpdfstring{$\Aa$}{A}-hypergeometric world} \label{ch:AHypergeometricWorld}
  
  
\section{Why \texorpdfstring{$\Aa$}{A}-hypergeometric systems?} \label{sec:AhypergeometricNutshell}


Remarkably often, functions appear in the calculation of Feynman integrals which are labelled as ``hypergeometric'' functions. For example the Gauss' hypergeometric function ${_2}F_1(a,b,c | z)$ shows up in the $1$-loop self-energy graph, the $2$-loop graph with three edges (also known as sunset graph) can be written by hypergeometric Lauricella functions \cite{BerendsClosedExpressionsSpecific1994}, the hypergeometric Appell $F_1$ function appears in the triangle graph and the hypergeometric Lauricella-Saran function is used in the $1$-loop box graph \cite{FleischerNewHypergeometricRepresentation2003}, to name a few. At first sight, all these hypergeometric functions appear to be rather a loose collection of unrelated functions. Thus, it is a justified question, what all these functions have in common and what is a general meaning of the term ``hypergeometric''. The subsequent question is then if all Feynman integrals are in that general sense ``hypergeometric''. 

Going back to a talk of Tullio Regge \cite{ReggeAlgebraicTopologyMethods1968} it is supposed for a long time that the second question has an affirmative answer and the appearance of hypergeometric functions in Feynman calculus is not just arbitrary. Regge also gave a suggestion as to which characteristics those general hypergeometric functions have to satisfy, which was later more specified by Sato \cite{SatoRecentDevolpmentHyperfunction1975}, Kashiwara and Kawai \cite{KashiwaraConjectureReggeSato1976} in the framework of microlocal analysis. All these approaches had in common that they based on holonomic $D$-modules, i.e.\ roughly speaking ``well-behaved'' systems of linear partial differential equations. \bigskip

A comprehensive investigation of the notion of generalized hypergeometric functions based on specific holonomic $D$-modules was then initiated by Gelfand, Graev, Kapranov and Zelevinsky in the late 1980s under the term $\Aa$-hypergeometric functions. Thereby, $\Aa$ is a finite configuration of vectors in $\mathbb Z^{n+1}$ which determines the type of the hypergeometric function. Additionally, an $\Aa$-hypergeometric function depends also on a parameter $\beta\in\mathbb C^{n+1}$ and variables $z\in\mathbb C^{|\Aa|}$.

It will turn out that the $\Aa$-hypergeometric functions fit perfectly into Regge's idea of hypergeometric functions. Further, with the $\Aa$-hypergeometric theory in mind we can answer both questions: firstly, what do all hypergeometric functions have in common, and secondly we can show that indeed any scalar Feynman integral without tadpoles is an $\Aa$-hypergeometric function. In that process, the vector configuration $\Aa$ will be determined by the Feynman graph, the parameter $\beta$ depends on the spacetime dimension and the indices of propagators and the variables $z$ turn out to be quotients of external momenta and masses. Therefore, the $\Aa$-hypergeometric functions cover the structure of Feynman integrals very naturally. Moreover, $\Aa$-hypergeometric theory has a very constructive nature, and we can make use of several features of this theory in the calculation and structural investigation of Feynman integrals.\bigskip

Thus, an $\Aa$-hypergeometric perspective on Feynman integrals could be very fruitful for physicists and the whole chapter is devoted to give an overview about the features of $\Aa$-hypergeometric functions. Before entering the $\Aa$-hypergeometric world in all its details, we want to give a first glimpse by highlighting the most important cornerstones.\bigskip

As the $\Aa$-hypergeometric system $H_\Aa(\beta)$ we will understand the holonomic $D$-ideal in the Weyl algebra $D=\mathbb C \langle z_1,\ldots,z_N,\partial_1,\ldots,\partial_N\rangle$ generated by
\begin{align}
	\{ \partial^u - \partial^v \,\rvert\, \Aa u = \Aa v, \, u,v\in\mathbb N^N \} \cup \langle \Aa \theta + \beta \rangle \label{eq:AHypIdeal1}
\end{align}
where $\Aa\in\mathbb Z^{(n+1)\times N}$ is understood as a matrix representation of a non-degenerated, acyclic vector configuration $\Aa\subset\mathbb Z^{n+1}$, $\beta\in\mathbb C^{n+1}$ is a complex parameter and $\theta$ is the Euler operator in $D$. Every holomorphic solution of those systems will be called an $\Aa$-hypergeometric function. It will turn out, that for generic $\beta$ the holonomic rank of $H_\Aa(\beta)$ will be given by the volume of the convex polytope $\Conv(A)$, where $A$ is the dehomogenized point configuration of $\Aa$. Hence, for a given triangulation $\Tt$ of $\Conv(A)$, we can construct a basis of the solution space $\Sol(H_\Aa(\beta))$ by assigning to each maximal cell of the triangulation as many independent solutions as the volume of the maximal cell is. This allows us to write $\Aa$-hypergeometric functions in terms of e.g.\ Horn series, Mellin-Barnes integrals or Euler integrals. Also, the singular locus of $H_\Aa(\beta)$, which is mainly the Landau variety in the application to Feynman integrals, gets a very natural expression in the $\Aa$-hypergeometric world in terms of principal $A$-determinants $E_A(f)$ and $A$-discriminants $\Delta_A(f)$
\begin{align}
    \Sing(H_\Aa(\beta)) = \Var (E_A(f)) = \bigcup_{\emptyset\neq \tau\subseteq\Conv(A)} \Var (\Delta_{A\cap\tau} (f_\tau)) \point
\end{align}

In the following sections we will go step by step more through all the details. Starting by the underlying geometric spaces, we will continue to convex polyhedra, which will cover the combinatorial structure of $\Aa$-hypergeometric functions. After a short introduction about $A$-discriminants and holonomic $D$-modules we are ready to state $\Aa$-hypergeometric systems and sketch their properties and their relations to $A$-discriminants.


\section{Affine and projective space} \label{sec:affineSpace}



The study of polynomials in several variables, as they appear for example in Feynman integrals, leads almost automatically to algebraic geometry. Therefore, we want to start by introducing the very basic notions of this subject adapted to $\Aa$-hypergeometric systems. This summary is orientated towards \cite{HartshorneAlgebraicGeometry1977, HolmeRoyalRoadAlgebraic2012}, which we also refer for further reading. \bigskip
    
Let $\mathbb K$ be a field. For convenience we will most often assume that $\mathbb K$ is algebraically closed, and we will use typically the complex numbers $\mathbb C$ or a convenient subfield. The \textit{affine $n$-space} $\AnK$ over $\mathbb K$ is simply the set of all $n$-tuples with elements from $\mathbb K$
\begin{align}
    \gls{AnK} := \{ (x_1,\ldots,x_n) \, \rvert \, x_i\in \mathbb K \quad  \textrm{for } i=1,\ldots,n \} \point
\end{align}
Thus, the affine $n$-space $\AnK$ is the vector space $\mathbb K^n$ as a set. However, $\AnK$ will have different morphisms as $\mathbb K^n$, which is the reason, why we want to use clearly separated terms. Especially, we do not want to have a distinguished point as the origin of $\mathbb K^n$. However, we can treat relative positions of points in $\AnK$ by referring on the corresponding points of $\mathbb K^n$. \bigskip

By $\gls{polyRing}$ we denote the coordinate ring of $\AnK$, i.e.\ the ring of polynomials in the variables $x_1,\ldots,x_n$. Thus, we consider elements of $\mathbb K[x_1,\ldots,x_n]$ as functions from $\AnK$ to $\mathbb K$. For a subset of polynomials $S\subseteq \mathbb K[x_1,\ldots,x_n]$ we define the zero locus
\begin{align}
    \gls{variety} := \left\{ x\in \AnK \, \rvert\, f(x) = 0 \quad\text{for all } f\in S \right\} \subseteq \AnK
\end{align}
as the \textit{affine variety} generated by $S$. By the structure of those subsets of $\AnK$ and Hilbert's basis theorem, we can restrict ourselves to considering only ideals $S$ or even merely their generators. It is not hard to see that the union and the intersection of affine varieties are again affine varieties \cite{HartshorneAlgebraicGeometry1977}
\begin{align}
    \Var(S_1) \cup \Var(S_2) &= \Var(S_1 \cdot S_2) \\
    \Var(S_1) \cap \Var(S_2) &= \Var(S_1 \cup S_2)
\end{align}
where $S_1\cdot S_2:=\{f\cdot g \,\rvert\, f\in S_1, g\in S_2\}$. Therefore, by setting these varieties as closed sets, we can define a topology on $\AnK$, the so-called \textit{Zariski topology}. Note that also the empty set and the full space are varieties generated by the polynomials $1$ and $0$, respectively. An affine variety is called \textit{irreducible}, if it can not be written as the union of two proper subvarieties\footnote{Two different definitions of affine varieties are common in literature. Several authors assume affine varieties always to be irreducible. They would call our notion of affine varieties an affine algebraic set.}. \bigskip

Laurent polynomials can be treated in an analogous way. The so-called \textit{algebraic torus} $\gls{Csn}:= (\mathbb C\setminus \{0\})^n$ is an affine variety in $\mathbb C^{2n}$ and \textit{Laurent monomials} in the variables $x_1,\ldots,x_n$ are nothing else than the characters of $\Csn$
\begin{align}
    (\mathbb C^*)^n \rightarrow \mathbb C^*, \qquad x \mapsto x^a := x_1^{a_1} \cdots x_n^{a_n} \label{eq:LaurentMonomial}
\end{align}
with the exponent $a=(a_1,\ldots,a_n)\in\mathbb Z^n$. Here and in the following we will make use of a multi-index notation as indicated in \cref{eq:LaurentMonomial}. A \textit{Laurent polynomial} is a finite linear combination of these monomials and can be uniquely written as  
\begin{align}
    f(x)=f(x_1,\ldots,x_n) = \sum_{a\in A} z_a x^a \in \mathbb C[x_1^{\pm 1},\ldots,x_n^{\pm 1}]
\end{align}
where $A\subset \mathbb Z^n$ is a finite set of non-repeating points and we consider the coefficients $z_a\in\mathbb C$ to be complex numbers not identically zero. We will call the subset $A$ the \textit{support} of $f$. The space of all Laurent polynomials with fixed monomials from $A$ but with indeterminate coefficients $z_a$ will be denoted as $\gls{CA}$. Thus, the set of coefficients $\{z_a\}_{a\in A}$ are coordinates of $\mathbb C^A$. Considering polynomials with fixed monomials but with variable coefficients is often referred as the ``$A$-philosophy'' \cite{GelfandDiscriminantsResultantsMultidimensional1994}. \bigskip 

The \textit{dimension} of an irreducible affine variety $V$ will be defined as the largest integer $d$ such that there exists a chain of irreducible closed subsets $V_1,\ldots,V_d$
\begin{align}
	V_1 \subsetneq \ldots \subsetneq V_d = V \point
\end{align}
The dimension of an affine variety is then the highest dimension of its irreducible components. This definition agrees with a definition via the Krull dimension of its coordinate ring. For the \textit{Krull dimension} of a commutative ring $R$ we consider chains of prime ideals $I_0 \subsetneq \ldots \subsetneq I_l$ in the ring $R$. The Krull dimension $\gls{KrullDimension}$ is then the suppremum of the lengths $l$ of all chains of prime ideals in $R$. We can extend the definition of Krull dimensions also to ideals $I\subseteq R$, by setting the dimension of an ideal, as the Krull dimension of its coordinate ring $\normalslant{R}{I}$. Thus, if the variety is generated by the ideal $I\subseteq R$, we have $\dim \Var(I) = \dim_{\operatorname{Kr}}(I) = \dim_{\operatorname{Kr}}(\normalslant{R}{I})$. This implies also that the dimension of the affine space (as a variety) agrees with the dimension of its corresponding vector space $\dim \AnK = \dim \mathbb K^n = n$. We refer to \cite{GreuelSingularIntroductionCommutative2002} for further details. \bigskip

For elements $a^{(1)},\ldots,a^{(k)}\in \AnK$ of an affine space we will call
\begin{align}
    \lambda_1 a^{(1)} + \ldots + \lambda_k a^{(k)} \qquad \text{with}\qquad \sum_{j=1}^k \lambda_j = 1, \quad \lambda_j\in\mathbb K
\end{align}
an \textit{affine combination}. Points in $\AnK$ are called \textit{affinely independent} if they can not be affinely combined to zero and \textit{affinely dependent} otherwise. Similarly, for a subset $S\subseteq \mathbb K$ we denote by
\begin{align}
    \gls{affineHull} := \left\{\lambda_1 a^{(1)} + \ldots + \lambda_k a^{(k)} \,\Big\rvert\, \lambda_j \in  S, \ \sum_{j=1}^k \lambda_j = 1\right\}
\end{align}
the \textit{affine span} or \textit{affine hull} generated by the elements $a^{(1)},\ldots,a^{(k)}\in\AnK$ over $S$.  A discrete subgroup of an affine space $\AnK$ is called an \textit{affine lattice} if the subgroup spans the full space $\AnK$. Further, a map $f: \AnK \rightarrow \mathbb A_{\mathbb K}^m$ between two affine spaces is called \textit{affine transformation}, if it preserves all affine combinations, i.e.\ $f(\sum_j \lambda_j a^{(j)}) = \sum_j \lambda_j f(a^{(j)})$ with $\sum_j \lambda_j = 1$. Every affine map $f(a) = l(a) + b$ can be splitted into a linear map $l(a)$ and a translation $b\in \mathbb A_{\mathbb K}^m$. \bigskip

Since every affine hull is a translated linear hull, we will similarly define the dimension of an affine hull to be the same as the dimension of its corresponding linear hull. However, there are at most $n+1$ affinely independent elements in $\AnK$. A set of those $n+1$ affinely independent elements $a^{(0)},\ldots,a^{(n)}$ is called a \textit{basis} or a \textit{barycentric frame} of an affine space $\AnK$, since it spans the whole affine space $\AnK$ over $\mathbb K$. Thus, for a given basis $a^{(0)},\ldots,a^{(n)}$ we can write every element $a\in \AnK$ as an affine combination of that basis and we will call the corresponding tuple $(\lambda_0,\ldots,\lambda_{n})$ the \textit{barycentric coordinates} of $a$. \bigskip

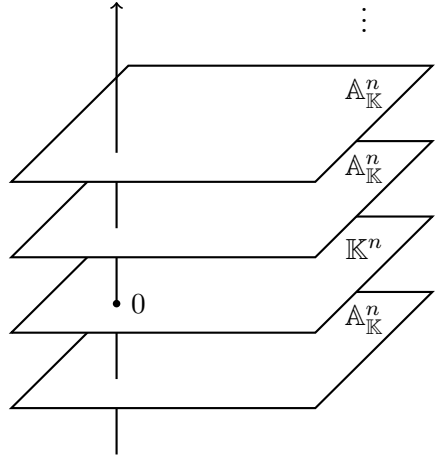
\begin{figure}
	\begin{center}
		\begin{tikzpicture}[scale=2]
            \draw[thick] (0,-1,0) -- (0,-.5,0); 
            \draw[thick,fill=white] (-.5,-.5,-1.5) -- ++(2,0,0) -- ++(0,0,2) -- ++(-2,0,0) -- cycle;
            \draw[thick] (0,-.5,0) -- (0,0,0); 
            \draw[thick,fill=white] (-.5,0,-1.5) -- ++(2,0,0) -- ++(0,0,2) -- ++(-2,0,0) -- cycle; 
            \draw[thick] (0,0,0) -- (0,0.5,0); 
            \draw[thick,fill=white] (-.5,.5,-1.5) -- ++(2,0,0) -- ++(0,0,2) -- ++(-2,0,0) -- cycle;
            \draw[thick] (0,.5,0) -- (0,1,0); 
            \draw[thick,fill=white] (-.5,1,-1.5) -- ++(2,0,0) -- ++(0,0,2) -- ++(-2,0,0) -- cycle;
            \draw[thick,->] (0,1,0) -- (0,2,0); 
            \coordinate[circle, fill, inner sep = 1pt, label=right:$0$] (O) at (0,0,0); 
            \coordinate[label=$\AnK$] (l1) at (2.4,2,2);
            \coordinate[label=$\AnK$] (l2) at (2.4,1.5,2);
            \coordinate[label=$\mathbb K^n$] (l3) at (2.4,1,2);
            \coordinate[label=$\AnK$] (l4) at (2.4,.5,2);
            \coordinate[label=$\vdots$] (l5) at (2.4,2.5,2);
		\end{tikzpicture}
	\end{center}
	\caption[Embedding of an affine space $\AnK$ in a vector space $\mathbb K^{n+1}$]{Embedding of affine spaces $\AnK$ in a vector space $\mathbb K^{n+1}$.} \label{fig:affineEmbedding}
\end{figure}

These coordinates indicate that we can naturally identify the affine space $\AnK$ as a hyperplane in the vector space $\mathbb K^{n+1}$. Thus, we consider the vector space $\mathbb K^{n+1} = \mathbb K^n \cup (\mathbb K^* \times \AnK)$ consisting in the slice of the vector space $\mathbb K^n$, containing the origin, and the remaining slices, each corresponding to an affine space $\AnK$, see also \cref{fig:affineEmbedding}. Since all slices, which do not contain the origin behave the same, we can identify without loss of generality $\AnK$ as the slice $1\times \AnK$ of $\mathbb K^{n+1}$. Therefore, we can accomplish the embedding, by adding an extra coordinate
\begin{align}
    a = (a_1,\ldots,a_n) \mapsto (1,a_1,\ldots,a_n) \label{eq:homogenization} \point
\end{align}
This will enable us to treat affine objects with the methods of linear algebra. Since points lying on a common hyperplane of $\mathbb K^{n+1}$, correspond to exponents of quasi-homogeneous polynomials, we will call the map of \cref{eq:homogenization} \textit{homogenization}. For a finite subset of points $A=\{a^{(1)},\ldots,a^{(N)}\}\subset \AnK$, we will write $\mathcal A = \{(1,a^{(1)}),\ldots,(1,a^{(N)})\} \subset \mathbb K^{n+1}$ as its homogenized version. Whenever it is convenient, we will denote by $\Aa$ also the $(n+1)\times N$ matrix collecting the elements of the subset $\Aa\subset\mathbb K^{n+1}$ as columns. \bigskip

Closely related to the affine space is the projective space. The \textit{projective space} $\gls{PnK}$ is the set of equivalence classes in $\mathbb K^{n+1}\setminus\{0\}$, where two elements $p,q\in \mathbb K^{n+1}\setminus\{0\}$ are equivalent if there exists a number $\lambda\in\mathbb K^*$ such that $p=\lambda q$. Thus, points of the projective space can be described by homogeneous coordinates. A point $p\in\PnK$ is associated to the homogeneous coordinates $\gls{homogeneousCoordinates}$ if an arbitrary element of the equivalence class of $p$ is described in the vector space $\mathbb K^{n+1}$ by the coordinates $(s_0,\ldots,s_n)$. Note that the homogeneous coordinates are not unique, as they can be multiplied by any element $\lambda\in\mathbb K^*$. \bigskip

Furthermore, we can decompose the projective space $\PnK= \AnK \cup \mathbb P^{n-1}_{\mathbb K}$ into an affine space and a projective space of lower dimension. In coordinates this decomposition means, that in case of $s_0\neq 0$, we can choose without loss of generality $s_0=1$, which defines the aforementioned affine hyperplane in $\mathbb K^{n+1}$. For $s_0=0$, sometimes referred as the ``points at infinity'' due to $\left[1 : \frac{s_1}{s_0} : \ldots : \frac{s_n}{s_0} \right]$, we obtain the projective space of lower dimension by the other remaining coordinates $[s_1 : \ldots : s_n]$. Analogously to the affine varieties we can define projective varieties as subsets of $\PnK$. Affine varieties can be completed to projective varieties by adding those ``points at infinity''.


\section{Convex polyhedra and triangulations} \label{sec:ConvexPolytopes}

Feynman integrals have a very rich combinatorial structure owing from the underlying Feynman graphs. Taking the parametric perspective of Feynman integrals this combinatorics will be caused by the Symanzik polynomials. Focussing on their extreme monomials, we will consider the Newton polytopes of Symanzik polynomials. Thus, we will uniquely attach a convex polytope to every Feynman integral and it will turn out, that the analytic structure of Feynman integrals will depend on the shape of that polytope. Thus, the study of those polytopes will lead us to the convergence region of Feynman integrals (\cref{thm:FIconvergence}), as well as the poles and the meromorphic continuation in the spacetime dimension $d$ and the indices of propagators $\nu$ (\cref{thm:meromorphicContinuation}), which are useful in dimensional and analytic regularization. Moreover, a triangulation of these polytopes will result in series representations of Feynman integrals in terms of Horn hypergeometric series (\cref{thm:FeynSeries}). And finally the set of all triangulations of these polytopes determines the extreme monomials of the defining polynomial of the Landau variety (\cref{thm:NewtSec}) and provides a nice factorization of it (\cref{thm:pAdet-factorization}). 

In the following subsection we will start to recall the most basic terms of convex polytopes which will generalize the intuitive perspective to higher dimensions. We will continue with a more technical perspective on convex polyhedra, which will reveal some structure of the underlying oriented matroids. In the last two parts of this section we will give an overview about triangulations and the structure of the set of all triangulations. The whole section is mostly based on \cite{ZieglerLecturesPolytopes1995,ThomasLecturesGeometricCombinatorics2006,DeLoeraTriangulations2010}, which we will refer for further reading. We will also recommend \cite{BrondstedIntroductionConvexPolytopes1983, HenkBasicPropertiesConvex2017, BrunsPolytopesRingsKtheory2009, GrunbaumConvexPolytopesSecond2003, SchrijverTheoryLinearInteger1986} for more detailed descriptions.


\subsection{Convex polytopes from point configurations} \label{ssec:PolytopesPointConfigs}
%

Let $\mathbb A_{\mathbb R}^n$ be the real affine $n$-space. By adding the standard inner product $ u^\top\cdot v := \sum_{i=1}^n u_i v_i$ to the associated vector space $\mathbb R^n$ of $\mathbb A_{\mathbb R}^n$, the real affine $n$-space $\mathbb A_{\mathbb R}^n$ becomes to the \textit{affine Euclidean space}. By slight abuse of notation we will denote such an affine Euclidean space also by $\mathbb R^n$. 
Every finite subset of labelled points in that space will be called a \textit{point configuration} $A \subset\mathbb R^n$. In general, points can be repeated in a point configuration, but labels are unique. We will usually use the natural numbers $\{1,\ldots,N\}$ to label those points. Arranging the elements of $A$ as columns will produce a matrix, which we will also denote by $A\in\mathbb R^{n\times N}$. Although, the point configuration is invariant under a reordering of columns, we will usually sort them for convenience ascending by their labels. \bigskip

For a point configuration  $A=\{a^{(1)},\ldots,a^{(N)}\}\subset\mathbb R^n$ of $N$ points we will call
\begin{align}
    \gls{Conv} := \left\{ \lambda_1 a^{(1)} + \ldots + \lambda_N a^{(N)} \,\Big\rvert\, \lambda_j \in\mathbb R, \lambda_j \geq 0, \sum_{j=1}^N \lambda_j = 1 \right\} \subset \mathbb R^n \label{eq:VPolytope}
\end{align}
the \textit{convex hull} of $A$ or a \textit{convex polytope} generated by $A$. Additionally, if all points $a^{(j)} \in \mathbb Z^n$ lying on an affine lattice, $\Conv(A)$ is called a \textit{convex lattice polytope}. As we will never consider any non-convex polytope, we will call them simply ``polytopes'' and ``lattice polytopes'', respectively.

Thus, convex hulls can be seen as a special case of affine hulls and we will consequentially set the \textit{dimension} of a polytope $P = \Conv (A)=\Aff_{\mathbb R_+}(A)$ to be the same as the dimension of the affine hull $\Aff_{\mathbb R}(A)$ defined in \cref{sec:affineSpace}. By using the homogenization of $A$ described in \cref{sec:affineSpace} (see also \cref{ssec:vectorConfigurations}), we can relate the dimension of a polytope to a matrix rank
\begin{align}
    \dim (\Conv (A)) = \rank (\Aa) - 1 \point
\end{align}
If a polytope $P\subset \mathbb R^n$ has dimension $n$ it is called to be \textit{full dimensional} and \textit{degenerated} otherwise. In most cases we want to assume full dimensional polytopes and adjust the dimension of the ambient space $\mathbb R^n$ if necessary.

The simplest possible polytope of dimension $n$ consists in $n+1$ vertices. We will call such a polytope an $n$-\textit{simplex}. By the \textit{standard $n$-simplex} we understand the full dimensional simplex generated by the standard unit vectors of $\mathbb R^{n}$ and the origin. \bigskip

A subset $\tau\subseteq P$ of a polytope for which there exists a linear map $\phi : \mathbb R^n \rightarrow \mathbb R$, such that \begin{align}
    \gls{tauface} = \big\{ p \in P \, \rvert \, \phi(p) \geq \phi(q) \quad\text{for all } q\in P \ \big\} \subseteq P \label{eq:facedef}
\end{align}
the map $\phi$ is maximized exactly for points on $\tau$ is called a \textit{face} of $P$.  Every face $\tau$ is itself a polytope generated by a subset of points of $A$. Whenever it is convenient we will identify with $\tau$ also this subset of $A$, as well as the subset of $\{1,\ldots,N\}$ labelling the elements of $A$ corresponding to this subset. Note that by construction also the full polytope $P$ and the empty set are faces of $P$ as well. Faces with dimension $0$, $1$ and $\dim(P)-1$ will be called \textit{vertices}, \textit{edges} and \textit{facets}, respectively. We will denote the set of all vertices by $\gls{Vert}$. Furthermore, the smallest face containing a point $p\in P$ is called the \textit{carrier} of $p$ and the polytope without its proper faces is called the \textit{relative interior} $\gls{relint}$.\bigskip

For full dimensional polytopes $P$ we want to introduce a \textit{volume} $\gls{vol}\in\mathbb R_{\geq 0}$, which is normalized such that the standard $n$-simplex has a volume equal to $1$. In other words, this volume is connected to the standard Euclidean volume $\vol_E (P)$ by a factorial of the dimension $\vol (P) = n! \vol_E (P)$. Especially, for simplices $P_\bigtriangleup$ generated by a point configuration $A=\{a^{(1)},\ldots,a^{(n+1)}\}\subset\mathbb R^n$ the volume is given by the determinant of its homogenized point configuration $\vol (P_\bigtriangleup) = |\det \Aa|$. The volume of a degenerated polytope will be set to zero. When restricting to lattice polytopes $P$, the volume $\vol(P)\in \mathbb Z_{\geq 0}$ is always a positive integer.\bigskip

The special interest in polytopes for this work results from the following construction which connects polytopes and polynomials. For a Laurent polynomial $f(x) = \sum_{a\in A} z_a x^a\in\mathbb C[x_1^{\pm 1},\ldots,x_n^{\pm 1}]$ we define its \textit{Newton polytope}
\begin{align}
    \gls{Newt} := \Conv \!\left(\left\{a\in A \,\rvert\, z_a\not\equiv 0\right\}\right) \subset \mathbb R^n
\end{align}
as the convex hull of its exponents. Furthermore, for a face $\tau\subseteq\Newt(f)$ of a Newton polytope, we define the \textit{truncated polynomial} with respect to $\tau$ as
\begin{align}
    \gls{truncpoly} := \sum_{a\in A\cap \tau} z_a x^a \label{eq:truncatedPolynomialDefinition}
\end{align}
consisting only in the monomials corresponding to that face $\tau$. \bigskip


\subsection{Vector configurations and convex polyhedra} \label{ssec:vectorConfigurations}
%

We will slightly generalize the point configurations from the previous section. Let $\mathbb R^{n+1}$ denote the Euclidean vector space and $(\mathbb R^{n+1})^\vee := \operatorname{Hom}(\mathbb R^{n+1},\mathbb R)$ its dual vector space. A finite collection of labelled elements from $\mathbb R^{n+1}$ will be called a \textit{vector configuration} and we will denote such a vector configuration by the symbol $\Aa\subset\mathbb R^{n+1}$. As described in \cref{sec:affineSpace} we can always embed $n$-dimensional affine spaces into $(n+1)$-dimensional vector spaces by homogenization. Thus, every homogenized point configuration is a vector configuration, even though we can also consider vector configurations not originating from a point configuration. As before we will denote by $\Aa\in\mathbb R^{(n+1)\times N}$ also the matrix constructed from the elements of $\Aa\subset\mathbb R^{n+1}$ considered as column vectors. \bigskip

The analogue of polytopes for vector configurations $\Aa = \{ a^{(1)},\ldots, a^{(N)} \} \subset\mathbb R^{n+1}$ are \textit{(convex) cones}
\begin{align}
    \gls{Cone} :=  \left\{ \lambda_1 a^{(1)} + \ldots + \lambda_N a^{(N)} \,\Big\rvert\, \lambda_j \in\mathbb R, \lambda_j \geq 0 \right\} \subset \mathbb R^{n+1} \point
\end{align}
Similar to polytopes, we can introduce faces of cones, i.e.\ subsets where a linear functional $\phi\in(\mathbb R^{n+1})^\vee$ is maximized or minimized. In contrast to the polytopal faces, maximal or minimal values will always be equal to zero. Furthermore, the empty set will not always be a face. Faces of dimension $1$ will be called \textit{rays}. \bigskip

A fundamental result of convex geometry is the following statement which is often also called the ``main theorem''.
\begin{theorem}[Weyl-Minkowski theorem for cones \cite{ZieglerLecturesPolytopes1995}] \label{thm:WeylMinkowski}
	A subset $C\subseteq\mathbb R^{n+1}$ is a convex cone $C=\Cone(\Aa)$ if and only if it is an intersection of finitely many closed linear halfspaces
	\begin{align}
	   C = P(M,0) := \left\{ \mu\in\mathbb R^{n+1} \,\rvert\, M \mu \leq 0 \right\} \subseteq \mathbb R^{n+1} \comma
    \end{align}
	where $M\in\mathbb R^{k\times (n+1)}$ is a real matrix and $M\mu \leq 0$ is understood componentwise $(M\mu)_i \leq 0$ for all $i=1,\ldots,k$.
\end{theorem}

Therefore, there are two equivalent representations of cones. \Cref{thm:WeylMinkowski} can be proven iteratively via Fourier-Motzkin elimination \cite{ZieglerLecturesPolytopes1995}. Note that the alternative representation of cones as intersection of halfspaces shows also that the intersection of cones is a cone as well. \bigskip

From \cref{thm:WeylMinkowski} one can also derive Farkas' lemma, which is known in many variants and which is very useful, when working with inequalities.
\begin{lemma}[Farkas' lemma, see e.g.\ \cite{ZieglerLecturesPolytopes1995}]  \label{lem:Farkas}
	Let $\Aa\in\mathbb R^{(n+1)\times N}$ be an arbitrary matrix and $b\in\mathbb R^{n+1}$. Then precisely one of the following assertions is true
	\begin{enumerate}[i)]
		\item there exists a vector $\lambda\in\mathbb R^N$ such that $\Aa\lambda = b$ and $\lambda\geq 0$
		\item there exists a vector $m\in\mathbb R^{n+1}$ such that $m^\top\Aa \leq 0$ and $m^\top b > 0$.
	\end{enumerate}
\end{lemma}
\begin{proof}
    The proof is roughly oriented towards \cite{SchrijverTheoryLinearInteger1986}. Note first that not both statements can be true, as 
    \begin{align}
    	0 < m^\top b = m^\top (\Aa \lambda) = (m^\top \Aa) \lambda \leq 0
    \end{align}
    gives a contradiction. 
    
    The first statement i) describes a cone $b\in C = \Cone(\Aa)$. From \cref{thm:WeylMinkowski} we know that there are vectors $m_1,\ldots,m_k\in\mathbb R^{n+1}$ such that $C = \{\mu\in\mathbb R^{n+1} \,\rvert\, m_1^\top\mu \leq 0, \ldots, m_k^\top\mu \leq 0\} $. As every column of $\Aa$ corresponds to a point in the cone, we will surely have $m_i^\top\Aa \leq 0$ for all $i=1,\ldots,k$. On the other hand, i) is false if and only if $b\notin C$, thus there has to be at least one $m_i$ such that $m_i^\top b > 0$.
\end{proof}

A vector configuration $\Aa$ is called \textit{acyclic} (occasionally also \textit{pointed}) if there is no non-negative dependence, i.e.\ there is no $\lambda\in\mathbb R^{N}_{\geq 0}\setminus\{0\}$ which satisfies $\Aa \lambda = 0$. By construction, every homogenized point configuration is acyclic, since every dependence vector $\lambda$ has to satisfy $\sum_{i=j}^N \lambda_j = 0$, which allows no solution for $\lambda\in\mathbb R^{N}_{\geq 0}\setminus\{0\}$. In contrast to acyclic configurations, we call $\Aa$ \textit{totally cyclic}, if the cone $\Cone(\Aa)$ equals the linear span of $\Aa$. By means of the following variant of Farkas' lemma we can give an alternative characterization of acyclic vector configurations. 

\begin{cor} \label{cor:farkas}
	Let $\Aa\in\mathbb R^{(n+1)\times N}$ be a real matrix. Then precisely one of the two assertions is true
    \begin{enumerate}[i)]
    	\item there exists a vector $\lambda \in \mathbb R^N$ with $\Aa\lambda=0$, $\lambda_j \geq 0$ and $\lambda \neq 0$
        \item there exists a linear functional $h\in(\mathbb R^{n+1})^\vee$ such that $h\Aa > 0$
    \end{enumerate}
\end{cor}
\begin{proof}
    Let $t>0$ be an arbitrary positive real number. Then the first statement i) can be formulated as
    \begin{align}
    	& \exists \lambda \in \mathbb R^N \text{ such that } \Aa\lambda \leq 0, \, - \Aa \lambda \leq 0 ,\, \lambda \geq 0  \text{ and } (t,\ldots,t) \cdot \lambda >0 \nonumber\\
    	\Leftrightarrow & \exists \lambda\in\mathbb R^N \text{ such that }  \lambda^\top \cdot (\Aa^\top,-\Aa^\top,-\mathbbm{1})\leq 0 \text{ and } t \ \lambda^\top \mathbf{1} > 0
    \end{align}
    where $\mathbbm{1}$ denotes the unit matrix and $\mathbf{1}=(1,\ldots,1)^\top$ denotes the column vector of ones. Now we can apply Farkas' \cref{lem:Farkas}. Therefore, the first statement is equivalent with
    \begin{align}
    	\Leftrightarrow & \nexists x,y\in\mathbb R^{n+1},z\in\mathbb R^N \text{ such that } \Aa^\top x - \Aa^\top y - \mathbbm{1} z = t \ \mathbf{1} \text{ and } x\geq 0, y\geq 0, z\geq 0 \nonumber\\
    	\Leftrightarrow & \nexists \widetilde h\in\mathbb R^{n+1}, z\in\mathbb R^N \text{ such that } \Aa^\top \widetilde h = t \ \mathbf{1} + \mathbbm{1} z \text{ and } z\geq 0 \point
    \end{align}
    This shows the assertion, as we can choose any $t > 0$.
\end{proof}

Therefore, acyclic configurations can also be described by the existence of a linear functional $h$ with $h\Aa > 0$. In other words, every acyclic vector configuration can be scaled in such a way, that it becomes a homogenized point configuration. We will visualize this fact in \cref{fig:acyclicCone}. Thus, we can dehomogenize acyclic vector configurations in the following way. Let $h\in(\mathbb R^{n+1})^\vee$ be any linear functional with $h\Aa>0$. By scaling the vectors of $\Aa = \{a^{(1)},\ldots,a^{(N)}\}$ with the map $a^{(j)} \mapsto \frac{a^{(j)}}{h\cdot a^{(j)}}$ we obtain a vector configuration describing points on the hyperplane $\{\mu\in\mathbb R^{n+1} \,\rvert\, h \mu = 1\} \cong \mathbb A^n_{\mathbb R}$. This can be considered as an affine point configuration in $\mathbb R^n$. \bigskip

\begin{figure}
	\begin{center}
        \begin{tikzpicture}[scale=2]
            \coordinate (O) at (0.5,0);
            \coordinate (A) at (.8,1.4);
            \coordinate (B) at (1.8,1.2);
            \coordinate (C) at (2.9,1.5);
            \coordinate (D) at (2,1.8);
            \coordinate (E) at (1.3,1.7);
            \draw[thick] (0,1) -- (3,1) -- (4,2) -- (1,2) -- (0,1); 
            \draw[thick] (A) -- (B) -- (C) -- (D) -- (E) -- (A); 
            \coordinate (AX) at ($(O)!1.8!(A)$);
            \coordinate (BX) at ($(O)!2!(B)$);
            \coordinate (CX) at ($(O)!1.5!(C)$);
            \coordinate (DX) at ($(O)!1.5!(D)$);
            \coordinate (EX) at ($(O)!1.5!(E)$);
            \draw[thick,->] (O) -- (AX);
            \draw[thick,->] (O) -- (BX);
            \draw[thick,->] (O) -- (CX);
            \draw[thick,->] (O) -- (DX);
            \draw[thick,->] (O) -- (EX);
            \draw[thick,->] (.3,1.1) -- ++(0,.3);
            \coordinate[label=right:$h$] (h) at (.3,1.25);
        \end{tikzpicture}
	\end{center}
	\caption[Acyclic vector configurations]{Example of an acyclic vector configuration with five vectors in $\mathbb R^3$. We can see the two characterizations of acyclic vector configurations. On the one hand, the only linear combination of these vectors resulting in the origin is the trivial combination. On the other hand, we can construct a hyperplane intersecting every vector in exactly one point. Thus, we can scale acyclic vector configurations in such a way that it becomes a homogenized point configuration constructing a polytope.}
	\label{fig:acyclicCone}
\end{figure}
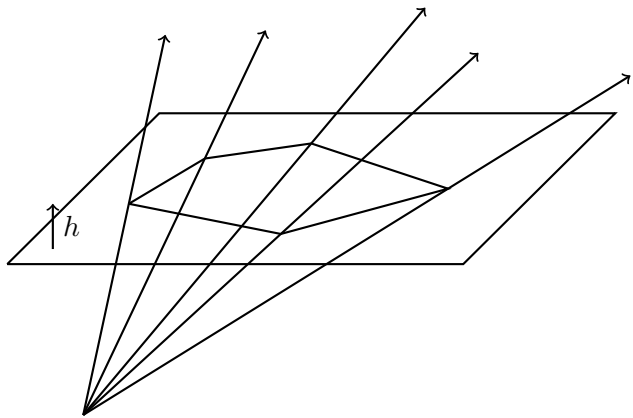

By the homogenization and dehomogenization procedure, we can transfer \cref{thm:WeylMinkowski} from cones also to polytopes. We call an intersection of finitely many, closed, linear halfspaces in $\mathbb R^n$
\begin{align}
    \gls{PMb} := \left\{ \mu\in\mathbb R^n \,\rvert\, M \mu \leq b \right\} \subseteq \mathbb R^n \label{eq:HPolytope}
\end{align}
with $M \in \mathbb R^{k\times n}$ and $b\in\mathbb R^k$ a \textit{(convex) polyhedron}. We will usually assume that \cref{eq:HPolytope} contains no redundant inequalities, i.e.\ $P(M,b)$ will change if we remove an inequality. It can be shown, that bounded polyhedra are equivalent to polytopes, which is also known as the Weyl-Minkowski theorem for polytopes \cite{ZieglerLecturesPolytopes1995}. Thus, we also have two different ways to represent polytopes: either by their vertices \cref{eq:VPolytope} or by their facets \cref{eq:HPolytope}. \bigskip

The conversion between the vertex representation and the facet representation is called facet enumeration problem and vertex enumeration problem, respectively. There are various implementations providing algorithms for the enumeration problem, e.g.\ the C-library \softwareName{lrslib} \cite{AvisLrslib} or the program \softwareName{polymake} \cite{GawrilowPolymakeFrameworkAnalyzing2000}. In the \cref{sec:SoftwareTools} we provide tips for the usage of these programs.

Although there is no analytic solution of the enumeration problem in general, for simplices it is nearly trivial. For the conversion in that case let $A = \{ a^{(1)},\ldots,a^{(n+1)}\} \subset \mathbb R^n$ be the generating set of a full dimensional simplex. Therefore, we have $\det\Aa\neq 0$ and we can invert the homogenized point configuration $\Aa$
\begin{align}
	P_\bigtriangleup = \Conv(A) = \left\{\mu\in\mathbb R^n \, \left| \, \begin{pmatrix} 1 \\ \mu \end{pmatrix} = \Aa k, k \geq 0 \right\}\right. =  \left\{ \mu\in\mathbb R^n \, \left| \, \Aa^{-1} \begin{pmatrix} 1 \\ \mu \end{pmatrix} \geq 0 \right\}\right.
\end{align}
which transforms the vertex representation into the facet representation. Furthermore, a vertex of a simplex is the intersection of $n$ facets. Thus, a vertex $v$ of $P_\bigtriangleup$ is the solution of $n$ rows of the $(n+1)\times(n+1)$ linear equation system $\Aa^{-1} x = 0$. Clearly, the $i$-th column of $\Aa$ solves the system $\Aa^{-1} x =0$ except for the $i$-th row. Hence, the intersection of $n$ facets is the $i$-th column of $\Aa$, where $i$ is the index which belongs to the facet which is not involved in the intersection. In a simplex that means that the $i$-th facet is opposite to the $i$-th vertex. Therefore, the $i$-th row of $\Aa^{-1} \begin{psmallmatrix} 1 \\ \mu \end{psmallmatrix} = 0$ describes the facet, which is opposite to the point defined by the $i$-th column of $\Aa$. \bigskip


\subsection{Gale duality} \label{ssec:GaleDuality}
%
%

In the previous section we have seen that the linear dependences determine whether a vector configuration is acyclic or not. However, linear dependences are much more powerful and reveal the structure of an oriented matroid. We will give a very short overview about the connection to oriented matroids and refer to \cite{ZieglerLecturesPolytopes1995} for a more detailed description.\bigskip

As before, let $\Aa = \{a^{(1)},\ldots,a^{(N)} \} \subset\mathbb R^{n+1}$ be a vector configuration spanning $\mathbb R^{n+1}$ as a vector space, for example a full dimensional homogenized point configuration. The set of all \textit{linear dependences} between those vectors generates a linear subspace of $\mathbb R^N$
\begin{align}
    \gls{Dep} := \left\{ \lambda\in\mathbb R^N \,\Big\vert\, \sum_{j=1}^N \lambda_j a^{(j)} = 0 \right \} \subseteq \mathbb R^N \label{eq:linearDep}
\end{align}
of dimension $r:=N-n-1$, which is nothing else than the kernel of $\Aa$. As seen above we are mainly interested whether these linear dependences are positive, zero or negative. Therefore, let $\sign : \mathbb R^N \rightarrow \{-,0,+\}^N$ be the componentwise sign function and for any vector $\lambda\in\mathbb R^N$, we call the non-zero components of $\lambda$ its \textit{support}. The elements of $\sign(\Dep(\Aa))$ are called the \textit{signed vectors of $\Aa$}. Furthermore, the elements of $\sign(\Dep(\Aa))$ having a minimal, non-empty support are called the \textit{signed circuits of $\Aa$}. \bigskip

As a further application of Farkas' lemma we can read off the face structure from the subspace of linear dependences. 

\begin{lemma}[{similar results can be found e.g.\ in \cite[ch. 5.4]{GrunbaumConvexPolytopesSecond2003} or \cite[ch. 5]{ThomasLecturesGeometricCombinatorics2006}}] \label{lem:face-kernel}
    Let $A=\{a^{(1)},\ldots,a^{(N)}\}\subset\mathbb R^n$ be a point configuration and $\Aa\subset\mathbb R^{n+1}$ its homogenization. Then $\tau$ is a (proper) face of $\Conv(A)$ if and only if there is no positive dependence, i.e.\ there is no $l\in\Dep(\Aa)$ with $l_j\geq 0$ for $j\notin\tau$ and $\{l_j\}_{j\notin\tau}\neq 0$. 
\end{lemma}
\begin{proof}
    Let $A_\tau = \{ a^{(j)} \}_{j\in\tau}$ and $A_\btau = \{ a^{(j)} \}_{j\notin\tau}$ be the subsets (matrices) collecting the elements (columns) corresponding to $\tau$ and its complement, respectively. We use the same convention for $l_\tau$ and $l_\btau$. The existence of a positive dependence means 
    \begin{align}
    	\exists l \in \mathbb R^N \text{ such that } A_\tau l_\tau + A_\btau l_\btau = 0,\ \sum_{j=1}^N l_j = 0,\, l_\btau \geq 0 \text{ and } l_\btau \neq 0 \comma
    \end{align}
    where $l_\btau\geq 0$ is understood componentwise. By writing $l_\tau = x - y$ with $x\geq 0$ and $y\geq 0$ we can reformulate this to
    \begin{align}
    	\Leftrightarrow \exists x,y\in\mathbb R^{|\tau|},\ l_\btau\in\mathbb R^{|\btau|},\ & r\in\mathbb R \text{ with } x\geq 0,\ y\geq 0,\ l_\btau\geq 0,\ r>0 \text{ such that} \nonumber\\
    	& \begin{pmatrix} 
            A_\tau & -A_\tau & A_\btau \\
            \mathbf{1}^\top & -\mathbf{1}^\top &\mathbf{1}^\top \\
            0 & 0 &\mathbf{1}^\top 
        \end{pmatrix} \begin{pmatrix} x \\ y \\ l_\btau \end{pmatrix} = \begin{pmatrix} 0 \\ 0 \\ r \end{pmatrix} \point
    \end{align}
    Hence, we can apply \cref{lem:Farkas}. Thus, the non-existence of a positive dependence is equivalent with
    \begin{align}
    	&\exists m\in\mathbb R^n, \exists s,t\in\mathbb R \text{ such that } m^\top A_\tau = s \mathbf{1}^\top,\, m^\top A_\btau  \leq (s-t) \mathbf{1}^\top \text{ and } t>0 \nonumber \\
    	\Leftrightarrow & \exists m\in\mathbb R^n, \exists s\in\mathbb R \text{ such that } m^\top A_\tau = s \mathbf{1}^\top \text{ and }  m^\top A_\btau  < s\mathbf{1}^\top \point
    \end{align}
    Since one can convince oneself easily that we can restrict the definition in equation \cref{eq:facedef} to the elements of $A$, this is nothing else than the definition of a face, where a linear functional $m^\top$ maximizes the points of $A_\tau$. 
\end{proof}

Analogously to the linear dependences, we call the subspace of linear forms
\begin{align}
    \gls{Val} := \left\{ \phi \Aa \,\rvert\, \phi \in \left(\mathbb R^{n+1}\right)^\vee \right \} \subseteq \left(\mathbb R^{N}\right)^\vee
\end{align}
the space of \textit{linear evaluations}, which has dimension $n+1$. The set $\sign(\Val(\Aa))$ will be called the \textit{signed covectors of $\Aa$} and its elements with a minimal, non-empty support are called the \textit{signed cocircuits of $\Aa$}.

These signed vectors and covectors together with the signed circuits and cocircuits can be associated with an oriented matroid of $\Aa$. Without going into detail we will focus on two aspects of oriented matroids only: the duality and the operations deletion and contraction. For more information about the connection to oriented matroids we refer to \cite{ZieglerLecturesPolytopes1995}. \bigskip

Note, that $\Dep(\Aa)$ and $\Val(\Aa)$ are orthogonal complements of each other, i.e.\ for every linear functional $\phi\in\Val(\Aa)$ every vector of $\Dep(\Aa)$ vanishes and vice versa. Thus, these subspaces are dual to each other, which vindicates the naming ``covectors'' and ``cocircuits'' from above. We will call this duality the \textit{Gale duality} and the transformation between these spaces the \textit{Gale transform}. Thus, a Gale dual of $\Aa$ is a vector configuration $\gls{Bb}$, such that the linear dependences of $\Aa$ are the linear evaluations of $\Bb$ and vice versa.

More specific, we will call a basis $\Bb\subset \mathbb R^N$ of the vector subspace $\Dep(\Aa)$ a Gale dual or a Gale transform of $\Aa$. Gale duals are not unique, as we can choose any basis of the space $\Dep(\Aa)$. However, all possible Gale duals are connected by regular linear transformations and we write $\gls{Gale}$ for the set of all Gale duals of $\Aa$. Note that if the vector configuration $\Aa$ is acyclic, every Gale dual is totally cyclic and vice versa. As before, we will associate $\Bb$ also with a matrix $\Bb\in\mathbb R^{N\times r}$. But contrary to the vector configurations $\Aa$, we suppose $\Bb$ as a collection of row vectors and we denote its elements by $\{b_1,\ldots,b_N\}$. The vector configuration of the rows of $\Bb$ is called a \textit{Gale diagram}.

\begin{example} \label{ex:GaleAppell}
    Consider the following point configuration of $6$ points in $\mathbb R^3$
    \begin{align}
        A = \begin{pmatrix} 
            0 & 1 & 0 & 0 & 1 & 0 \\
            0 & 0 & 1 & 0 & 0 & 1 \\
            0 & 0 & 0 & 1 & 1 & 1 \\
            \end{pmatrix} \label{eq:PtConfigAppell}
    \end{align}
    which turns out to generate the $\Aa$-hypergeometric system of the Appell $F_1$ function. The corresponding polytope $P=\Conv(A)$ is depicted in \cref{fig:PolytopeAppell}.

    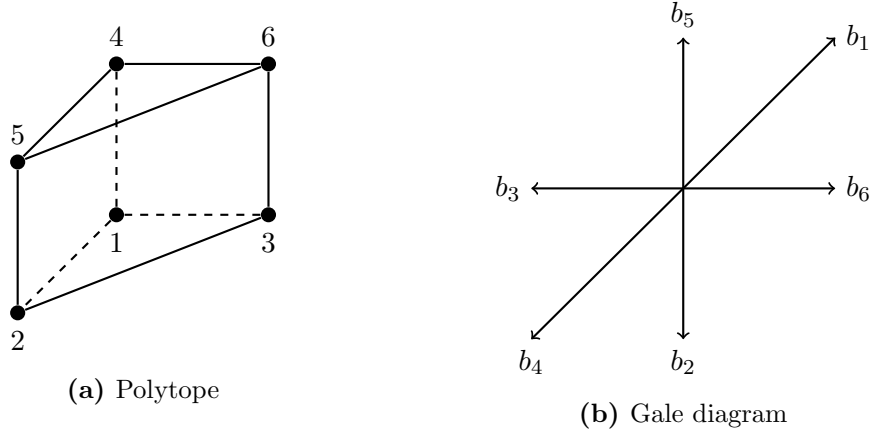
\begin{figure}
        \centering
        \begin{subfigure}{.48\textwidth}
            \centering
        	\begin{tikzpicture}[scale=1]
                \coordinate[circle, fill, inner sep = 2pt, label=below:$1$] (A) at (0,0);
                \coordinate[circle, fill, inner sep = 2pt, label=below:$2$] (B) at (-1.3,-1.3);
                \coordinate[circle, fill, inner sep = 2pt, label=below:$3$] (C) at (2,0);
                \coordinate[circle, fill, inner sep = 2pt, label=above:$4$] (D) at (0,2);
                \coordinate[circle, fill, inner sep = 2pt, label=above:$5$] (E) at (-1.3,.7);
                \coordinate[circle, fill, inner sep = 2pt, label=above:$6$] (F) at (2,2);
                \draw[thick] (E) -- (B) -- (C) -- (F);
                \draw[thick] (D) -- (E) -- (F) -- (D);
                \draw[thick,dashed] (A) -- (B);
                \draw[thick,dashed] (C) -- (A);
                \draw[thick,dashed] (A) -- (D);
            \end{tikzpicture}
            \caption{Polytope}  \label{fig:PolytopeAppell}
        \end{subfigure}
        \begin{subfigure}{.48\textwidth}
            \centering
            \begin{tikzpicture}[scale=2]
                \draw[thick,->] (0,0) -- (1,1) node[right] {$b_1$};
                \draw[thick,->] (0,0) -- (0,-1) node[below] {$b_2$};
                \draw[thick,->] (0,0) -- (-1,0) node[left] {$b_3$};
                \draw[thick,->] (0,0) -- (-1,-1) node[below] {$b_4$};
                \draw[thick,->] (0,0) -- (0,1) node[above] {$b_5$};
                \draw[thick,->] (0,0) -- (1,0) node[right] {$b_6$};
            \end{tikzpicture}
            \caption{Gale diagram}  \label{fig:GaleDualAppell}
        \end{subfigure}
    	\caption[Polytope and Gale diagram for Appell $F_1$ function]{Polytope and Gale diagram of the point configuration \cref{eq:PtConfigAppell}, which is related to the Appell $F_1$ function.}
    \end{figure}
    
    As pointed out above, the Gale transform is not unique. For example we can choose
    \begin{align}
        \Bb^\top = \begin{pmatrix}
            1 & 0 & -1 & -1 & 0 & 1 \\
            1 & -1 & 0 & -1 & 1 & 0
            \end{pmatrix} \point
    \end{align}
    By means of \cref{lem:face-kernel} we can read off the faces also from the Gale dual (\cref{fig:GaleDualAppell}). The cofaces of $\Conv(A)$, i.e.\ the complements of faces of $\Conv(A)$, are precisely those, where the cones spanned by corresponding elements of the Gale dual have the origin $0$ as their relative interior. E.g.\ $\Cone(b_2,b_5)$ contains the origin as its relative interior. Hence, the points $1,3,4,6$ generate a face of $\Conv(A)$. $\Cone(b_5,b_6)$ contains the origin only on his boundary. Therefore, $1,2,3,4$ is not related to a face of $\Conv(A)$.
\end{example}

For a vector configuration $\Aa\subset\mathbb R^{n+1}$ containing an element with label $j$, we mean by the \textit{deletion} \gls{deletion} simply the vector subconfiguration where we removed the element corresponding to the label $j$. Its dual operation, the \textit{contraction} \gls{contraction}, can be understood as a projection of $\Aa$ to a hyperplane not containing the element $a^{(j)}$. Thus, let $c\in\mathbb R^{n+1}$ with $c^\top a^{(j)} \neq 0$ be a vector defining a hyperplane. Then,
\begin{align}
    \contraction{\Aa}{j} := (\pi \Aa)\setminus j \,\text{ with }\, \pi := \mathbbm{1} - a^{(j)} \frac{1}{c^\top \cdot a^{(j)}} c^\top
\end{align}
is such a projection to a hyperplane. As a Gale dual $\Bb$ of $\Aa$ is nothing else than a matrix satisfying $\Aa \Bb = 0$, $\Bb$ will be also a Gale dual of $\pi \Aa$. Thus, it is not hard to see the duality of deletion and contraction \cite{DeLoeraTriangulations2010}
\begin{align}
	\Gale(\contraction{\Aa}{j}) = \Gale(\Aa)\setminus j \point
\end{align}


\subsection{Triangulations of polyhedra} \label{ssec:TriangulationsPolyhedra}
%

For the combinatorial structure of Feynman graphs CW-complexes are of special significance, see e.g.\ \cite{BerghoffGraphComplexesFeynman2021}. As is also for $\Aa$-hypergeometric functions, where we will consider subdivisions of polytopes. A \textit{subdivision} $\Ss$ of a point configuration $A\subset\mathbb R^n$ is a set of polytopes $\sigma$, which are called \textit{cells}, such that
\begin{enumerate}[(i)]
	\item if $\sigma\in\Ss$ then also every face of $\sigma$ is contained in $\Ss$,
	\item for all $\sigma,\tau\in\Ss$, the intersection $\sigma\cap\tau$ is either a face of both or empty,
	\item the union of all $\sigma\in \Ss$ is equal to $\Conv(A)$.
\end{enumerate}
Note that this definition does not demand to use all points of $A$. If additionally all cells are simplices we call the subdivision a \textit{triangulation} $\gls{Tt}$. By $\gls{hatT}$ we want to refer to the set of maximal cells of a triangulation $\Tt$, i.e.\ those cells which are not contained in other cells. If all $\sigma\in \hatT$ belong to simplices with volume $1$ (i.e.\ $|\det\Aas|=1$) the triangulation is called \textit{unimodular}. 

We say that a subdivision $\Ss$ refines $\Ss^\prime$, in symbols $\gls{refinement}$, if for every cell $c\in\Ss$ there exists a cell $c^\prime\in\Ss^\prime$ containing it $c\subseteq c^\prime$. Therefore, we can group the subdivisions by their refinements as a partially ordered set (poset). \bigskip

The underlying structure of subdivisions and triangulations are polyhedral complexes. We call a set $\mathcal K$ of polyhedra a \textit{polyhedral complex} if it satisfies
\begin{enumerate}[i)]
    \item if $P\in\mathcal K$ and $F\subseteq P$ is a face of $P$, it implies $F\in\mathcal K$
    \item $P\cap Q$ is a face of $P$ and a face of $Q$ for all $P,Q\in\mathcal K$.
\end{enumerate}
Thus, a triangulation is a simplicial polyhedral complex of polytopes covering $\Conv(A)$. But also the set of all faces of a polytope is a polyhedral complex. A subset $\mathcal K^\prime\subseteq\mathcal K$ of a polyhedral complex $\mathcal K$, generating a polyhedral complex by itself is called a \textit{polyhedral subcomplex} of $\mathcal K$. \bigskip

The definition of polyhedral complexes is not only restricted to polytopes, i.e.\ bounded polyhedra. We call a polyhedral complex consisting only in cones a \textit{(polyhedral) fan}. A fan is \textit{complete} if it covers $\mathbb R^n$. Especially, we can associate a fan to the faces of a polytope as follows. Let $P\subset\mathbb R^n$ be a polytope and $x\in P$ a point of it. The set of linear functionals
\begin{align}
	\gls{NPx} := \{ \phi\in(\mathbb R^n)^\vee \,\rvert\, \phi x \geq \phi y \quad \forall y\in P\}
\end{align}
is called the \textit{outer normal cone of $x$}. By considering the definition of a face \cref{eq:facedef} we see that for all relative interior points of a face $\tau\subseteq P$ the outer normal cone does not change. Consequentially, we will write $N_P(\tau)$. Note that $N_P(x)$ is full dimensional if and only if $x\in\Vertices(P)$. The collection of all outer normal cones of a polytope will be called the \textit{outer normal fan of $P$}
\begin{align}
	\gls{NP} := \{ N_P(\tau) \,\rvert\, \tau \,\text{ is a face of } \,P \} = \{ N_P(x) \,\rvert\, x\in P \} \subseteq (\mathbb R^n)^\vee \point
\end{align}
Analogously, we understand by the \textit{inner normal cone} the negatives of outer normal cones and the \textit{inner normal fan} will be the set of all inner normal cones. \bigskip

After the formal description of subdivisions by polyhedral complexes, we want to show several ways to construct them. Let $A = \{a^{(1)},\ldots,a^{(N)}\}\subset\mathbb R^n$ be a point configuration and $\gls{omega} : A \rightarrow \mathbb R$ a \textit{height} function, assigning a real number to every point of $A$. The so-called lifted point configuration $A^\omega := \{ (\omega_1, a^{(1)}), \ldots, (\omega_N,a^{(N)})\}$, where $\omega_j:=\omega(a^{(j)})$, describes a polytope $P^\omega = \Conv(A^\omega)\subset\mathbb R^{n+1}$ in dimension $n+1$. We call a face of $P^\omega$ \textit{visible from below} or a \textit{lower face} if the face is generated by a linear functional $\phi\in(\mathbb R^{n+1})^\vee$ having its first coordinate negative. By the projection $\pi: A^\omega\rightarrow A$, forgetting the first coordinate, we project all lower faces down to $\Conv(A)$ (see \cref{fig:regularTriangulation}). These projected lower faces satisfy all conditions of a subdivision and we will note such a subdivision by\glsadd{regSubdivision}\glsunset{regSubdivision} $\Ss(A,\omega)$. Moreover, if all projected lower faces are simplices, we have constructed a triangulation. All subdivisions and triangulations which can be generated by such a lift construction are called \textit{regular subdivisions} and \textit{regular triangulations}, respectively. Regular triangulations will become of great importance in the following and it can be shown, that every convex polytope admits always a regular triangulation \cite{DeLoeraTriangulations2010}. 

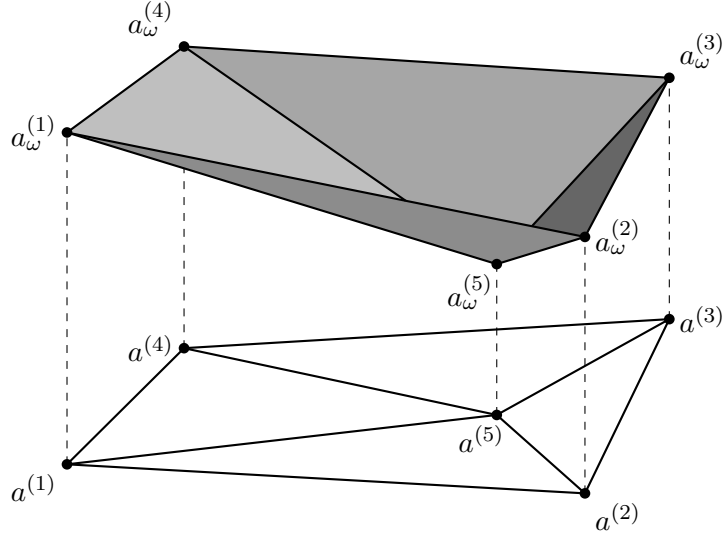
\begin{figure}[tb]
	\centering
	\begin{tikzpicture}[scale=2]              
        \coordinate[label=below left:$a^{(1)}$] (1) at (0,0,0);
        \coordinate[label=below right:$a^{(2)}$] (2) at (3.6,0,0.5);
        \coordinate[label=right:$a^{(3)}$] (3) at (3,0,-2.5);
        \coordinate[label=left:$a^{(4)}$] (4) at (0,0,-2);
        \coordinate[label=below:$\hspace{-.4cm}a^{(5)}$] (5) at (2.5,0,-0.85);

        \coordinate[label=left:$a_\omega^{(1)}$] (1h) at  ($(1) + (0,2.2,0)$);
        \coordinate[label=right:$a_\omega^{(2)}$] (2h) at  ($(2) + (0,1.7,0)$);
        \coordinate[label=above right:$a_\omega^{(3)}$] (3h) at  ($(3) + (0,1.6,0)$);
        \coordinate[label=above left:$a_\omega^{(4)}$] (4h) at  ($(4) + (0,2,0)$);
        \coordinate[label=below:$\hspace{-.7cm}a_\omega^{(5)}$] (5h) at  ($(5) + (0,1,0)$);
        \coordinate (I1) at (intersection of 1h--2h and 3h--5h);
        \coordinate (I2) at (intersection of 1h--2h and 4h--5h);

        \draw[thick] (1) -- (2) -- (3) -- (4) -- (1);
        \draw[dashed] (1) -- (1h); \draw[dashed] (2) -- (2h); \draw[dashed] (3) -- (3h); \draw[dashed] (4) -- (4h); \draw[dashed] (5) -- (5h);
        
        \fill[gray!50] (1h) -- (4h) -- (5h);
        \fill[gray!70] (3h) -- (4h) -- (5h);
        \fill[gray!120] (2h) -- (3h) -- (5h);
        \fill[gray!90] (1h) -- (2h) -- (5h);

        \draw[thick] (1h) -- (2h) -- (3h) -- (4h) -- (1h);
        \draw[thick] (1h) -- (5h); \draw[thick] (2h) -- (5h); \draw[thick] (3h) -- (I1); \draw[thick] (4h) -- (I2);
        
        \foreach \x in {1,2,3,4,5,1h,2h,3h,4h,5h} \fill (\x) circle (1pt);

        \draw[thick] (1) -- (5); \draw[thick] (2) -- (5); \draw[thick] (3) -- (5); \draw[thick] (4) -- (5);
	\end{tikzpicture}
	\caption[Constructing regular triangulations]{Regular triangulation of a point configuration $A=\{a^{(1)},\ldots,a^{(5)}\}\subset\mathbb R^2$. We denote $a_\omega^{(j)} = (\omega_j,a^{(j)})$ for the lifted points. The projection of the lower faces of the lifted point configuration to the polytope $\Conv(A)$ generates the regular triangulation.} \label{fig:regularTriangulation}
\end{figure}

Similarly, we can define a regular subdivision $\Ss(\Aa,\omega)$ of vector configurations $\Aa$. However, for a subdivision of a vector configuration we have to replace convex hulls by cones in the definitions above. Thus, a subdivision of a vector configuration is a polyhedral complex of cones covering $\Cone(\Aa)$. According to \cite{SturmfelsGrobnerBasesConvex1995} we can reformulate regular subdivisions as follows. For a convenient height $\omega\in\mathbb R^N$ a regular subdivision $\Ss(\Aa,\omega)$ of $\Aa=\{a^{(1)},\ldots,a^{(N)}\}$ consists of all subsets $\sigma\subseteq\{1,\ldots,N\}$ such that there exists a linear functional $c\in(\mathbb R^{n+1})^\vee$ with
\begin{align}
	c \cdot a^{(j)} &= \omega_j \qquad\text{for}\quad j \in \sigma \label{eq:regTriangVector1} \\
	c \cdot a^{(j)} &< \omega_j \qquad\text{for}\quad j \notin \sigma \point \label{eq:regTriangVector2}
\end{align}
For acyclic vector configurations this will produce a regular subdivision for any height $\omega\in\mathbb R^N$. Moreover, this subdivision will agree with the regular subdivision of the dehomogenized vector configuration $A\subset\mathbb R^n$. If $\omega$ is sufficiently generic $\Ss(\Aa,\omega)$ will be a regular triangulation. However, for non-acyclic vector configurations not all heights will produce a subdivision and we have to restrict us to non-negative heights $\omega\in\mathbb R_{\geq 0}^N$ in that case \cite{DeLoeraTriangulations2010}.\bigskip

Another, iterative construction of triangulations of point configurations is the so-called \textit{placing triangulation}. Consider a face $\tau$ of a convex polytope $P\subset \mathbb R^n$ and an arbitrary point in the relative interior of that face $x\in\relint (\tau)$. The face $\tau$ is \textit{visible} from another point $p\notin \tau$, if the line segment $[x,p]$ intersects $P$ only at $x$. With the concept of visibility, we can construct triangulations iteratively. Let $\Tt$ be a (regular) triangulation of the point configuration $A$. Then the set
\begin{align}
\Tt^\prime = \Tt \cup \{ \tau \cup \{p\} \, | \,  \tau \in \Tt \textrm{ and } \tau \textrm{ is visible from } p \} 
\end{align}
is a (regular) triangulation of the point configuration $A \cup p$ \cite{DeLoeraTriangulations2010}. Thus, by starting with a triangulation of an arbitrary point of $A$ we can place step by step the other points to the previous triangulation. The order of the added points will determine the triangulation. 

The placing triangulation is slightly more convenient in an algorithmical use. However, more efficient algorithms make use of the connection to oriented matroids. We refer to \cite{RambauTOPCOMTriangulationsPoint2002} for a consideration about efficiency of algorithms and its implementation in the software \softwareName{Topcom}. Also the comprehensive software \softwareName{polymake} \cite{GawrilowPolymakeFrameworkAnalyzing2000} includes algorithms to construct triangulations. We will discuss those possibilities in \cref{sec:SoftwareTools}. \bigskip

Let $\Tt$ be a point configuration of $A$ and $\Tt^\prime$ a triangulation of $A^\prime$. $\Tt^\prime$ is a \textit{subtriangulation} of $\Tt$, in symbols $\Tt^\prime \subseteq \Tt$, if every cell of $\Tt^\prime$ is contained in $\Tt$. Thus, it will be also $A^\prime\subseteq A $. In other words, $\Tt^\prime$ is a subcomplex of $\Tt$. The placing triangulation shows that for every point configurations $A$, $A^\prime$ with $A^\prime\subseteq A$ there will be (regular) triangulations $\Tt$, $\Tt^\prime$ with $\Tt^\prime\subseteq \Tt$. However, not all triangulations $\Tt$ of the point configuration $A$ will have a subtriangulation corresponding to $A^\prime$. But for regular triangulations, we will always find consistent triangulations in the following sense.

\begin{lemma}[triangulations of deleted point configurations \cite{DeLoeraTriangulations2010}] \label{lem:subtriangulations}
	Let $\Tt$ be a regular triangulation (subdivision) of a point configuration $A\subset \mathbb R^n$ and $a^{(j)}\in A$ a point with the label $j$. Then there is a regular triangulation (subdivision) $\Tt^\prime$ of $A\setminus j$, using all simplices (cells) of $\Tt$ which do not contain $a^{(j)}$.
\end{lemma}
Note that \cref{lem:subtriangulations} demand the triangulations to be regular. Except for several special cases as e.g.\ $n=2$, the lemma holds not necessarily for non-regular triangulations.

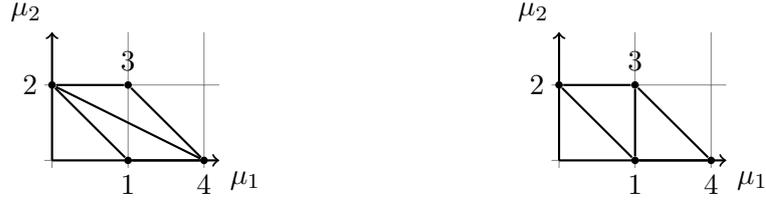
\begin{figure}
    \centering
    \begin{subfigure}{.45\textwidth}
        \centering
        \begin{tikzpicture}[scale=1]
            \draw[step=1cm,gray,very thin] (-0.1,-0.1) grid (2.2,1.7);  
            \draw[thick,->] (0,0) -- (2.2,0) node[anchor=north west] {$\mu_1$};
            \draw[thick,->] (0,0) -- (0,1.7) node[anchor=south east] {$\mu_2$};
            \coordinate[circle, fill, inner sep = 1pt, label=below:$1$] (A) at (1,0);
            \coordinate[circle, fill, inner sep = 1pt, label=left:$2$] (B) at (0,1);
            \coordinate[circle, fill, inner sep = 1pt, label=above:$3$] (C) at (1,1);
            \coordinate[circle, fill, inner sep = 1pt, label=below:$4$] (D) at (2,0);
            \draw[thick] (A) -- (B) -- (C) -- (D) -- (A);
            \draw[thick] (B) -- (D);
        \end{tikzpicture}
        \caption{regular triangulation $\Ss(A,\omega)$ of $A$ generated by $\omega = (0,0,1,0)$}
    \end{subfigure}
    \begin{subfigure}{.45\textwidth}
    	\centering
    	\begin{tikzpicture}[scale=1]
            \draw[step=1cm,gray,very thin] (-0.1,-0.1) grid (2.2,1.7);  
            \draw[thick,->] (0,0) -- (2.2,0) node[anchor=north west] {$\mu_1$};
            \draw[thick,->] (0,0) -- (0,1.7) node[anchor=south east] {$\mu_2$};
            \coordinate[circle, fill, inner sep = 1pt, label=below:$1$] (A) at (1,0);
            \coordinate[circle, fill, inner sep = 1pt, label=left:$2$] (B) at (0,1);
            \coordinate[circle, fill, inner sep = 1pt, label=above:$3$] (C) at (1,1);
            \coordinate[circle, fill, inner sep = 1pt, label=below:$4$] (D) at (2,0);
            \draw[thick] (A) -- (B) -- (C) -- (D) -- (A);
            \draw[thick] (A) -- (C);
        \end{tikzpicture}
        \caption{regular triangulation $\Ss(A,\omega)$ of $A$ generated by $\omega = (0,0,0,1)$}
    \end{subfigure}
    \caption[Example of regular triangulations]{Example of the two possible regular triangulations of the point configuration $A = \{ (1,0) ; (0,1) ; (1,1) ; (2,0) \}$ with heights $\omega$ generating those triangulations. This example turns out to describe the $1$-loop self-energy Feynman integral with one mass, see \cref{ex:1loopbubbleA}. The possible proper subtriangulations of these two triangulations are simply the single simplices itself.}
\end{figure}


\subsection{Secondary polytopes and secondary fans} \label{ssec:secondaryPolytope}
%

In the last part about convex polyhedra we will study the structure of the set of all subdivisions. For every triangulation $\Tt$ of a full dimensional point configuration $A\subset\mathbb R^n$ we will introduce the \textit{weight map} (occasionally also known as \textit{GKZ-vector}) $\gls{weight} :A\rightarrow \mathbb Z_{\geq 0}$
\begin{align}
    \varphi_\Tt (a) :=  \!\!\!\!\!\!\!\!\!\! \sum_{\substack{\sigma\in \Tt \,\,\text{s.t.}\\ a\in\Vertices(\Conv(\sigma))}} \!\!\!\!\!\!\!\!\!\! \vol \!\left(\Conv(\sigma)\right) = \varphi_{\hatT} (a)
\end{align}
which is the sum of all simplex volumes, having $a$ as its vertex. We define degenerated polytopes to have volume zero, which is the reason why we only have to consider the full dimensional simplices $\hatT$. We write $\varphi_\Tt(A) = \left(\varphi_\Tt(a^{(1)}),\ldots,\varphi_T(a^{(N)})\right)$ for the image of $A$. Note that two distinct regular triangulations also have a different weight $\varphi_\Tt(A)$, whereas two distinct non-regular triangulations may have the same weight \cite{DeLoeraTriangulations2010}.

The weights themselves define a further polytope of dimension $r:=N-n-1$, which is the so-called \textit{secondary polytope} $\Sigma(A)$. It is the convex hull of all weights
\begin{align}
    \gls{SecPoly} := \Conv\!\left(\left\{\varphi_\Tt(A) \, | \, \Tt \text{ is a triangulation of } A\right\}\right) \subset\mathbb R^N \point
\end{align}
The vertices of the secondary polytope $\Sigma(A)$ correspond precisely to the regular triangulations $\Tt$. Moreover, the refinement poset of regular subdivisions is isomorphic to the face lattice of $\Sigma(A)$ \cite{DeLoeraTriangulations2010}. We will demonstrate this with the example shown on the book cover of \cite{GelfandDiscriminantsResultantsMultidimensional1994}.

\begin{example} \label{ex:5pointsPlane}
	Let $A\subset\mathbb R^2$ be the following point configuration
	\begin{align}
		A = \begin{pmatrix}
		    	0 & 2 & 2 & 0 & 1 \\
		    	0 & 0 & 2 & 2 & 1
		    \end{pmatrix}
	\end{align}
	of five points forming a rectangle in the plane. There are three possible triangulations. Note that not all points of $A$ has to be used in a triangulation. These three triangulations have the weights $\varphi_{\Tt_1} = (8,4,8,4,0)$, $\varphi_{\Tt_2} = (4,8,4,8,0)$ and $\varphi_{\Tt_3} = (4,4,4,4,8)$. These weights generate the secondary polytope $\Sigma (A)\subset\mathbb R^5$, whose actual dimension is only $2$. Thus, $\Sigma(A)$ is a $2$-dimensional triangle in $\mathbb R^5$ as depicted in \cref{fig:SecondaryPolytope5pointsPlane}. As aforementioned the vertices of $\Sigma(A)$ correspond to the regular triangulations of $A$. The edges of $\Sigma(A)$ will correspond to regular subdivisions of $A$, which have the triangulations as their only strict refinements. Thus, in this example it will be the regular subdivisions containing $4$ points. The full secondary polytope will correspond to the trivial subdivision.
	\begin{figure}
		\centering
		\begin{tikzpicture}[scale=0.3]
		    \coordinate (T1) at (0,0); \coordinate (T2) at (15,0); \coordinate (T3) at (7.5,10);
		    \draw[thick] (T1) -- (T2) -- (T3) -- cycle;  
		
            \coordinate (A5) at ($(T1) + (-1.5,-1.5)$);
            \coordinate[circle, fill, inner sep = 1pt] (A1) at ($(A5) + (-1,-.8)$);
            \coordinate[circle, fill, inner sep = 1pt] (A2) at ($(A5) + (1,-.8)$);
            \coordinate[circle, fill, inner sep = 1pt] (A3) at ($(A5) + (1,.8)$);
            \coordinate[circle, fill, inner sep = 1pt] (A4) at ($(A5) + (-1,.8)$);
            \draw[thick] (A1) -- (A2) -- (A3) -- (A4) -- (A1); \draw[thick] (A1) -- (A3);
            
            \coordinate[circle, fill, inner sep = 1pt] (B5) at ($(T1)!0.5!(T2) + (0,-1.5)$);
            \coordinate[circle, fill, inner sep = 1pt] (B1) at ($(B5) + (-1,-.8)$);
            \coordinate[circle, fill, inner sep = 1pt] (B2) at ($(B5) + (1,-.8)$);
            \coordinate[circle, fill, inner sep = 1pt] (B3) at ($(B5) + (1,.8)$);
            \coordinate[circle, fill, inner sep = 1pt] (B4) at ($(B5) + (-1,.8)$);
            \draw[thick] (B1) -- (B2) -- (B3) -- (B4) -- (B1);  \draw[thick] (B1) -- (B3);
            
            \coordinate[circle, fill, inner sep = 1pt] (C5) at ($(T2) + (1.5,-1.5)$);
            \coordinate[circle, fill, inner sep = 1pt] (C1) at ($(C5) + (-1,-.8)$);
            \coordinate[circle, fill, inner sep = 1pt] (C2) at ($(C5) + (1,-.8)$);
            \coordinate[circle, fill, inner sep = 1pt] (C3) at ($(C5) + (1,.8)$);
            \coordinate[circle, fill, inner sep = 1pt] (C4) at ($(C5) + (-1,.8)$);
            \draw[thick] (C1) -- (C2) -- (C3) -- (C4) -- (C1);  
            \draw[thick] (C1) -- (C5); \draw[thick] (C2) -- (C5); \draw[thick] (C3) -- (C5); \draw[thick] (C4) -- (C5); 
            
            \coordinate[circle, fill, inner sep = 1pt] (D5) at ($(T2)!0.5!(T3) + (2.3,-.2)$);
            \coordinate[circle, fill, inner sep = 1pt] (D1) at ($(D5) + (-1,-.8)$);
            \coordinate[circle, fill, inner sep = 1pt] (D2) at ($(D5) + (1,-.8)$);
            \coordinate[circle, fill, inner sep = 1pt] (D3) at ($(D5) + (1,.8)$);
            \coordinate[circle, fill, inner sep = 1pt] (D4) at ($(D5) + (-1,.8)$);
            \draw[thick] (D1) -- (D2) -- (D3) -- (D4) -- (D1); \draw[thick] (D2) -- (D4); 
            
            \coordinate (E5) at ($(T3) + (0,1.5)$);
            \coordinate[circle, fill, inner sep = 1pt] (E1) at ($(E5) + (-1,-.8)$);
            \coordinate[circle, fill, inner sep = 1pt] (E2) at ($(E5) + (1,-.8)$);
            \coordinate[circle, fill, inner sep = 1pt] (E3) at ($(E5) + (1,.8)$);
            \coordinate[circle, fill, inner sep = 1pt] (E4) at ($(E5) + (-1,.8)$);
            \draw[thick] (E1) -- (E2) -- (E3) -- (E4) -- (E1); \draw[thick] (E2) -- (E4); 
            
            \coordinate (F5) at ($(T1)!0.5!(T3) + (-2.3,-.2)$);
            \coordinate[circle, fill, inner sep = 1pt] (F1) at ($(F5) + (-1,-.8)$);
            \coordinate[circle, fill, inner sep = 1pt] (F2) at ($(F5) + (1,-.8)$);
            \coordinate[circle, fill, inner sep = 1pt] (F3) at ($(F5) + (1,.8)$);
            \coordinate[circle, fill, inner sep = 1pt] (F4) at ($(F5) + (-1,.8)$);
            \draw[thick] (F1) -- (F2) -- (F3) -- (F4) -- (F1);
            
            \coordinate[circle, fill, inner sep = 1pt] (G5) at ($(T1)!0.5!(T2) + (0,4.8)$);
            \coordinate[circle, fill, inner sep = 1pt] (G1) at ($(G5) + (-1,-.8)$);
            \coordinate[circle, fill, inner sep = 1pt] (G2) at ($(G5) + (1,-.8)$);
            \coordinate[circle, fill, inner sep = 1pt] (G3) at ($(G5) + (1,.8)$);
            \coordinate[circle, fill, inner sep = 1pt] (G4) at ($(G5) + (-1,.8)$);
            \draw[thick] (G1) -- (G2) -- (G3) -- (G4) -- (G1);
		\end{tikzpicture}
		\caption[Example of a secondary polytope $\Sigma(A)$]{The secondary polytope $\Sigma(A)$ from \cref{ex:5pointsPlane} together with the regular subdivisions. Vertices of the secondary polytope correspond to regular triangulations, whereas edges of $\Sigma(A)$ corresponding to subdivisions into polytopes consisting in four points. The trivial subdivision is equal to the full secondary polytope.} \label{fig:SecondaryPolytope5pointsPlane}
	\end{figure}
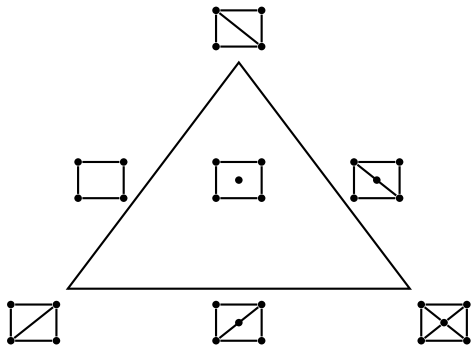
\end{example}

We want to consider this structure more in detail. Let $\Aa\subset\mathbb R^{n+1}$ be a vector configuration and $\Tt$ a regular subdivision of it. We call the set of heights $\omega$ generating $\Tt$ or a coarser subdivision
\begin{align}
	\gls{SecCone} := \left\{\omega\in\mathbb R^N \,\rvert\, \Tt \preceq \mathcal S(\Aa,\omega)\right\}
\end{align}
the \textit{secondary cone} of $\Tt$ in $\Aa$. The secondary cone $\Sigma C(\Aa,\Tt)$ is a polyhedral convex cone and it is full dimensional if and only if $\Tt$ is a regular triangulation \cite{DeLoeraTriangulations2010}. Furthermore, $\Sigma C(\Aa,\Tt^\prime)$ is a proper face of $\Sigma C(\Aa,\Tt)$ if and only if $\Tt \prec \Tt^\prime$. Therefore, the relative interior of the secondary cone describes the heights generating the regular subdivision $\Tt$, i.e.\ $\relint \!\left(\Sigma C(\Aa,\Tt)\right) = \{\omega\in\mathbb R^N \,\rvert\, \Tt = \mathcal S(\Aa,\omega)\}$. This set is known as \textit{relatively open secondary cone}. The set of all secondary cones is called the secondary fan
\begin{align}
	\gls{SecFan} = \big\{ \Sigma C(\Aa,\Tt) \,\rvert\, \Tt \text{ is a regular subdivision of } \Aa \big\} \point
\end{align}
It can be shown that $\Sigma F (\Aa)$ is a polyhedral fan, which is complete if and only if $\Aa$ is acyclic. Moreover, when $\Aa$ is the homogenized point configuration of $A$, $\Sigma F(\Aa)$ is the inner normal fan of the secondary polytope $\Sigma(A)$. Thus, by the secondary fan, we get the aforementioned relation between the secondary polytope and the refinement poset of subdivisions, which is encoded in the secondary cones. However, the secondary polytope is not full dimensional. Therefore, we have to find the right projection to connect the possible heights $\omega$ of subdivisions to the secondary polytope. \bigskip

It turns out that the Gale dual provides the right projection from heights to the secondary structure. Let $\Aa=\{a^{(1)},\ldots,a^{(N)}\}\subset\mathbb R^{n+1}$ be a full dimensional vector configuration and $\sigma\subseteq\{1,\ldots,N\}$ be a set of labels corresponding to a full dimensional simplex of $\Aa$. We write $\bar\sigma=\{1,\ldots,N\}\setminus \sigma$ for the complement of $\sigma$ and $\Aas = (a^{(j)})_{j\in\sigma}$ as well as $\Aabs = (a^{(j)})_{j\notin\sigma}$. Eliminating the linear forms $c$ from \cref{eq:regTriangVector1} and \cref{eq:regTriangVector2} we see that
\begin{align}
	- \omega_\sigma \Aas^{-1}\Aabs + \omega_{\bar\sigma} > 0 \label{eq:subdivisionCondGale}
\end{align}
describes a necessary and sufficient condition for $\sigma\in\Ss(\Aa,\omega)$, where we use the same nomenclature $\omega_\sigma := (\omega_j)_{j\in\sigma}$ and $\omega_{\bar\sigma} := (\omega_j)_{j\notin\sigma}$. Note, that $\Bb = \begin{psmallmatrix} - \Aas^{-1}\Aabs \\ \mathbbm 1 \end{psmallmatrix}$ is a possible Gale dual of $\Aa$, where the first rows correspond to $\sigma$ and the last rows correspond to $\bar\sigma$. Hence, we can equivalently write $\omega \Bb > 0$ instead of \cref{eq:subdivisionCondGale}. This relation can be extended to any Gale dual $\Bb\in\Gale(\Aa)$ as follows. Denote by\glsadd{betaMap}\glsunset{betaMap} $\beta : \mathbb R^N \rightarrow \mathbb R^r$ the projection $\omega \mapsto \omega \Bb$ and $\Bb_{\bar\sigma} = (b_i)_{i\in\bar\sigma}$ are the rows of $\Bb$ corresponding to $\bar\sigma$. Then we have \cite{DeLoeraTriangulations2010}
\begin{align}
    \sigma\in\Ss(\Aa,\omega) \, \Leftrightarrow \, \beta(\omega) \in \relint \!\left(\Cone (\Bb_{\bar\sigma})\right) \point \label{eq:subdivisionCondGaleBeta}
\end{align}
Therefore, the (maximal) cells of regular subdivisions $\Ss(\Aa,\omega)$ are directly related to the structure of the Gale diagram. Furthermore, regular subdivisions of $\Aa$ correspond to the intersection of cones spanned by subconfigurations of $\Bb$. Those intersections will also be called \textit{chambers}\footnote{To be precise, a \textit{relatively open chamber} is a minimal, non-empty intersection of cones corresponding to subconfigurations of the Gale dual $\Bb$. A \textit{closed chamber} is defined to be the closure of a relatively open chamber. The set of all closed chambers is called the \textit{chamber complex} or the \textit{chamber fan} \cite{DeLoeraTriangulations2010}.} of $\Bb$. These chambers are projections of secondary cones $\beta(\Sigma C(\Aa,\Tt))$ into the Gale diagram. Therefore, by projecting the secondary fan also by the map $\beta$ we obtain the so-called chamber complex which  can be constructed from the Gale diagram. Therefore, the Gale diagram contains all the essential part of the secondary fan. We will demonstrate this fact by the following example.

\begin{example} \label{ex:SecondaryAppell}
    We will continue the \cref{ex:GaleAppell} of six points in $\mathbb R^3$. This point configuration has $6$ triangulations, all of them are regular and unimodular \par    
    {\centering    
    \vspace{\baselineskip}
    \begin{tabular}{ccp{.8cm}cc}
        I:   & $\{1,2,3,4\}, \{2,3,4,5\}, \{3,4,5,6\}$ & & IV: & $\{1,2,3,5\}, \{1,4,5,6\}, \{1,3,5,6\}$ \\
        II:  & $\{1,2,3,4\}, \{2,4,5,6\}, \{2,3,4,6\}$ & & V:  & $\{2,4,5,6\}, \{1,2,4,6\}, \{1,2,3,6\}$ \\
        III: & $\{3,4,5,6\}, \{1,3,4,5\}, \{1,2,3,5\}$ & & VI: & $\{1,4,5,6\}, \{1,2,3,6\}, \{1,2,5,6\}$
    \end{tabular}
    \vspace{\baselineskip} \par
    }    
    \noindent where the numbers stand for the labels of the points and the curly braces denote the full dimensional simplices. The triangulations can be constructed either by hand, considering \cref{fig:PolytopeAppell} or by using a software program, e.g.\ \softwareName{Topcom} \cite{RambauTOPCOMTriangulationsPoint2002}. Using the same Gale dual as in \cref{ex:GaleAppell} the Gale diagram is depicted in \cref{fig:AppellGaleTriangs}  where we also denote the cones/chambers generating the triangulations. For example the intersection of the cones $\Cone(b_5,b_6) \cap\Cone(b_1,b_6)\cap\Cone(b_1,b_2)$ is the chamber generating the triangulation I. Thus, for any height vector $\omega$ with $\omega \Bb$ lying inside this intersection we will generate the triangulation I. Triangulations which are ``neighbours'' in the Gale diagram, i.e.\ they have a common facet, will be said to be related by a flip. Triangulations which are related  by a flip share all except two of their simplices \cite{DeLoeraTriangulations2010}.

    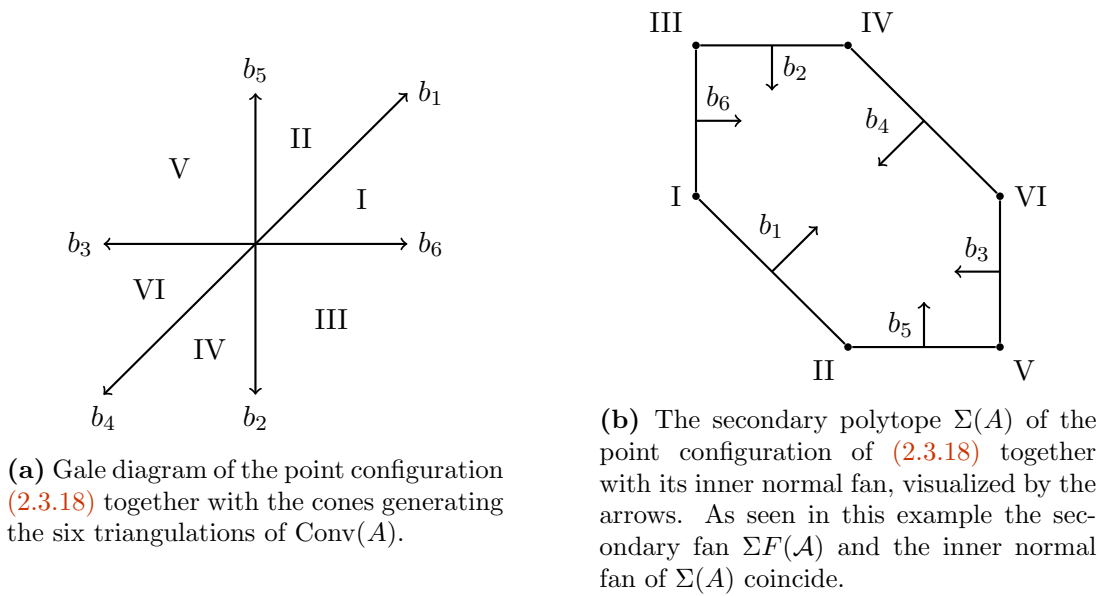
\begin{figure}[thb]
        \centering
    	\begin{subfigure}{.45\textwidth}
    		\centering
    		\begin{tikzpicture}[scale=2]
                \draw[thick,->] (0,0) -- (1,1) node[right] {$b_1$};
                \draw[thick,->] (0,0) -- (0,-1) node[below] {$b_2$};
                \draw[thick,->] (0,0) -- (-1,0) node[left] {$b_3$};
                \draw[thick,->] (0,0) -- (-1,-1) node[below] {$b_4$};
                \draw[thick,->] (0,0) -- (0,1) node[above] {$b_5$};
                \draw[thick,->] (0,0) -- (1,0) node[right] {$b_6$};

                \node at (.7,.3){I};
                \node at (.3,.7){II};
                \node at (.5,-.5){III};
                \node at (-.3,-.7){IV};
                \node at (-.5,.5){V};
                \node at (-.7,-.3){VI};
            \end{tikzpicture}
            \caption{Gale diagram of the point configuration \cref{eq:PtConfigAppell} together with the cones generating the six triangulations of $\Conv(A)$.} \label{fig:AppellGaleTriangs}
    	\end{subfigure}    	
    	\hspace{1cm}
    	\begin{subfigure}{.45\textwidth}
            \centering
            \begin{tikzpicture}[scale=2]
                \coordinate[circle,inner sep=1pt,fill,label=left:I] (I) at (0,0);
                \coordinate[circle,inner sep=1pt,fill,label=below left:II] (II) at (1,-1);
                \coordinate[circle,inner sep=1pt,fill,label=above left:III] (III) at (0,1);
                \coordinate[circle,inner sep=1pt,fill,label=above right:IV] (IV) at (1,1);
                \coordinate[circle,inner sep=1pt,fill,label=below right:V] (V) at (2,-1);
                \coordinate[circle,inner sep=1pt,fill,label=right:VI] (VI) at (2,0);
                \draw[thick] (I) -- (II) -- (V) -- (VI) -- (IV) -- (III) -- (I);
                \draw[thick,->] ($(I)!0.5!(II)$) -- +(.3,.3) node[midway,above left] {$b_1$};
                \draw[thick,->] ($(II)!0.5!(V)$) -- +(0,.3) node[midway,left] {$b_5$};
                \draw[thick,->] ($(V)!0.5!(VI)$) -- +(-.3,0) node[midway,above] {$b_3$};
                \draw[thick,->] ($(VI)!0.5!(IV)$) -- +(-.3,-.3) node[midway,above left] {$b_4$};
                \draw[thick,->] ($(IV)!0.5!(III)$) -- +(0,-.3) node[midway,right] {$b_2$};
                \draw[thick,->] ($(III)!0.5!(I)$) -- +(.3,0) node[midway,above] {$b_6$};
            \end{tikzpicture} 
            \caption{The secondary polytope $\Sigma(A)$ of the point configuration of \cref{eq:PtConfigAppell} together with its inner normal fan, visualized by the arrows. As seen in this example the secondary fan $\Sigma F(\Aa)$ and the inner normal fan of $\Sigma(A)$ coincide.}  \label{fig:AppellSecPolyFan}
    	\end{subfigure}
    	\caption[Gale diagram with regular triangulations and secondary polytope for the Appell $F_1$ function]{Gale diagram with regular triangulations and secondary polytope for the Appell $F_1$ function corresponding to \cref{eq:PtConfigAppell}.}
    \end{figure}

    Furthermore, we can calculate the weights of these triangulations\par  
    {\centering
    \vspace{\baselineskip}
    \begin{tabular}{llp{1.5cm}ll}
        I:   & $(1,2,3,3,2,1)$ & & IV: & $(3,1,2,1,3,2)$ \\
        II:  & $(1,3,2,3,1,2)$ & & V:  & $(2,3,1,2,1,3)$ \\
        III: & $(2,1,3,2,3,1)$ & & VI: & $(3,2,1,1,2,3)$ .
    \end{tabular}
    \vspace{\baselineskip} \par
    }
    \noindent The convex hull of these weights is the secondary polytope, which is a polytope of dimension $2$ in $\mathbb R^6$. On a convenient subspace one obtains the representation shown in \cref{fig:AppellSecPolyFan}. One can see the connection of the secondary fan and the Gale diagram. The secondary fan is the inner normal fan of the secondary polytope, visualized by the arrows. These agree with the Gale diagram. 
    
    However, we want to remark that the considered examples show the simplest, non-trivial case, where $r=N-n-1 = 2$. In these cases the Gale diagram is a diagram in the plane $\mathbb R^2$ and the intersection of cones of subconfigurations of $\Bb$ works out very simply. Thus, for $r=2$ the chambers are also spanned by the Gale diagram. In consequence, point configurations with $r=2$ can have at most $N$ regular triangulations. Furthermore, in that case there are no non-regular triangulations \cite{DeLoeraTriangulations2010}. For point configurations with $r=N-n-1>2$ the intersection of cones is much more involved. The procedure described above works also for those cases. However, the chamber complex can not be read off from the Gale diagram as easy as in the case $r=2$.
\end{example}


\sectionmark{$A$-discriminants, $A$-resultants \& principal $A$-determinants}
\section[\texorpdfstring{$A$}{A}-discriminants, \texorpdfstring{$A$}{A}-resultants and principal \texorpdfstring{$A$}{A}-determinants]{$A$-discriminants, $A$-resultants and\\ principal $A$-determinants} 
\sectionmark{$A$-discriminants, $A$-resultants \& principal $A$-determinants}

\label{sec:ADiscriminantsReultantsPrincipalADets}

We are often interested in the question whether a system of simultaneous polynomial equations 
\begin{align}
    f_0  (x_1,\ldots,x_n) = \ldots = f_n  (x_1,\ldots,x_n) = 0 \label{eq:polysys}
\end{align}
has a solution in a given (algebraically closed) field $\mathbb K$ or if it is inconsistent. That question could be answered in general by calculating the Gröbner basis of the ideal generated by $f_0,\ldots,f_n$. Thus, as a consequence of Hilbert's weak Nullstellensatz, a system of polynomial equations is inconsistent if and only if the corresponding reduced Gröbner basis is equal to $1$. Unfortunately, the calculation of Gröbner bases can be hopelessly complicated and computers fail even in simpler examples. This problem gets even harder if we want to vary the coefficients in the polynomials of \cref{eq:polysys}. Resultants, instead, can answer this question much more efficiently. In general a resultant is a polynomial in the coefficients of the polynomials $f_0,\ldots,f_n$, which vanishes whenever the system \cref{eq:polysys} has a common solution. \bigskip

However, the theory of multivariate resultants comes with several subtleties. We have to distinguish between \textit{classical multivariate resultants} (also known as \textit{dense resultants}) and \textit{(mixed) $A$-resultants} (or \textit{sparse resultants}). The classical multivariate resultant will be applied to $n$ homogeneous polynomials in $n$ variables, where every polynomial consists in all possible monomials of a given degree and detects common solutions in projective space $\mathbb P^{n-1}_{\mathbb K}$. In contrast, the $A$-resultant is usually used, if the polynomials do not consist in all monomials of a given degree. For a system of $n+1$ polynomials in $n$ variables, the $A$-resultant is a custom-made polynomial and reveals common solutions, which are ``mostly'' located in the affine space $(\mathbb C^*)^n$. However, what we accept as a ``solution'' in the latter case is slightly subtle, and we will give a precise definition below. Note, that the classical multivariate resultant is a special case of the $A$-resultant \cite{CoxUsingAlgebraicGeometry2005}. Furthermore, we want to distinguish between the case where all polynomials $f_0,\ldots,f_n$ have the same monomial structure, i.e.\ they all have the same support $A$ and the mixed case where the polynomials $f_0,\ldots,f_n$ have different monomial structure defined by several supports $A_0,\ldots,A_n$. 

Closely related to resultants are discriminants, which determine whether a polynomial $f$ has a multiple root. This is equivalent to ask if there is a solution such that the polynomial $f$ and its first derivatives vanish. Hence, discriminants play also an important role for identifying singular points of algebraic varieties.

In the following we will sketch various key features of the theory of $A$-resultants and $A$-discriminants, which were mainly introduced in a series of articles by Gelfand, Kapranov and Zelevinsky \cite{GelfandAdiscriminantsCayleyKoszulComplexes1990, GelfandDiscriminantsPolynomialsMany1990, GelfandNewtonPolytopesClassical1990, GelfandDiscriminantsPolynomialsSeveral1991} in the study of $A$-hypergeometric functions \cite{GelfandHypergeometricFunctionsToral1989, GelfandGeneralizedEulerIntegrals1990, GelfandGeneralHypergeometricSystems1992} and were collected in \cite{GelfandDiscriminantsResultantsMultidimensional1994}. For an introduction to $A$-resultants as well as the classical multivariate resultants we refer to \cite{CoxUsingAlgebraicGeometry2005} and \cite{SturmfelsSolvingSystemsPolynomial2002}.


\subsection{Mixed \texorpdfstring{$(A_0,\ldots,A_n)$}{(A\unichar{"2080},...,A\unichar{"2099})}-resultants and \texorpdfstring{$A$}{A}-resultants}

The key idea of resultants is to specify coefficients and variables in a system of polynomial equations and eliminate the variables from it. Hence, resultants are a main tool in elimination theory. We will summarize the basic definitions and several properties of the multivariate resultants, which can be found in \cite{PedersenProductFormulasResultants1993, SturmfelsNewtonPolytopeResultant1994, SturmfelsIntroductionResultants1997, SturmfelsSolvingSystemsPolynomial2002, CoxUsingAlgebraicGeometry2005, GelfandDiscriminantsResultantsMultidimensional1994}. \bigskip      

Let $A_0,\ldots, A_{n} \subset \mathbb Z^n$ be finite subsets of the affine lattice $\mathbb Z^n$ and for every set $A_i$ we will consider the corresponding Laurent polynomial 
\begin{align}
    f_i(x)=f_i(x_1,\ldots,x_n) = \sum_{a\in A_i} z^{(i)}_a x^a \in \mathbb C [x_1^{\pm 1}, \ldots, x_n^{\pm 1}]\point
\end{align}
For simplicity, we will assume that the supports $A_0,\ldots,A_n$ jointly generate the affine lattice $\mathbb Z^n$. Furthermore, by $P_i := \Conv(A_i)=\Newt(f_i)$ we denote the Newton polytope of $f_i$. According to \cite{SturmfelsNewtonPolytopeResultant1994, SturmfelsIntroductionResultants1997, SturmfelsSolvingSystemsPolynomial2002}, we will call a configuration $A_0,\ldots,A_n$ \textit{essential} if
\begin{align}
    \dim\!\left( \sum_{j=0}^n P_j \right) = n \qquad \text{and} \qquad  \dim\!\left(\sum_{j\in J} P_j\right) \geq  |J| \quad \text{for every } J \subsetneq\{0,\ldots,n\} \comma \label{eq:essential}
\end{align}
where the sum of polytopes denotes the Minkowski sum and $|J|$ is the cardinality of the proper subset $J$. If all polytopes $P_i$ are $n$-dimensional, the equations in \cref{eq:essential} are trivially satisfied.

In order to define the general resultants, we are interested in the set of coefficients $z_a^{(i)}$ for which there exists a solution of $f_0(x) = \ldots = f_n(x)=0$ in $x\in\Csn$. In other words we consider the following set in $\prod_{i=0}^n \mathbb C^{A_i}$
\begin{align}
    \mathscr Z = \left\{ (f_0,\ldots,f_{n}) \in \prod_i \mathbb C^{A_i} \,\rvert\, \Var(f_0,\ldots,f_{n})\neq\emptyset \text{ in } (\mathbb C^*)^{n} \right\} \subseteq \mathbb \prod_{i=0}^n \mathbb C^{A_i} \point \label{eq:defMixedAresultantsZ}
\end{align}
Furthermore, by $\overline{\mathscr Z}$ we will denote the Zariski closure of $\mathscr Z$. The \textit{mixed $(A_0,\ldots,A_n)$-resultant} $\gls{mixedARes}\in\mathbb Z[\{\{z^{(i)}_a\}_{a\in A_i}\}_{i = 0,\ldots,n}]$ is an irreducible polynomial in the coefficients of the polynomials $f_0,\ldots,f_n$. In case where $\overline{\mathscr Z}$ describes a hypersurface in $\prod_i \mathbb C^{A_i}$ we will define $R_{A_0,\ldots,A_n}(f_0,\ldots,f_n)$ to be the minimal defining polynomial of this hypersurface $\overline{\mathscr Z}$. Otherwise, so if $\codim \overline{\mathscr Z} \geq 2$, we will set $R_{A_0,\ldots,A_n}(f_0,\ldots,f_n)=1$. The mixed $(A_0,\ldots,A_n)$-resultant always exists and is uniquely defined up to a sign, which was shown in \cite{GelfandDiscriminantsResultantsMultidimensional1994}.

Further, we have $\codim \overline{\mathscr Z} = 1$ if and only if there exists a unique subset of $A_0,\ldots,A_n$ which is essential \cite{SturmfelsNewtonPolytopeResultant1994}. In that case the mixed $(A_0,\ldots,A_n)$-resultant coincides with the resultant of that essential subset.\bigskip

One has to remark as a warning, that the mixed $(A_0,\ldots,A_n)$-resultants not only detect common solutions of $f_0 = \ldots = f_n = 0$ inside $x\in\Csn$. Due to the Zariski closure in the definition of the resultants, the mixed $(A_0,\ldots,A_n)$-resultants may also describe solutions outside of $x\in\Csn$, e.g.\ ``roots at infinity''.\bigskip

If all polynomials $f_0,\ldots,f_n$ have the same monomial structure, i.e.\ $A_0 = \ldots = A_{n} =: A$, we will call $\gls{ARes}:=R_{A,A,\ldots,A}(f_0,\ldots,f_{n})$ simply the \textit{$A$-resultant}. The $A$-resultants satisfy a natural transformation law.

\begin{lemma}[Transformation law of $A$-resultants \cite{GelfandDiscriminantsResultantsMultidimensional1994}] \label{lem:trafoAres}
    Consider the polynomials $f_0,\ldots, f_n\in\mathbb C^A$ and let $D$ be an invertible $(n+1)\times (n+1)$ matrix. For the transformation $g_i = \sum_{j=0}^{n} D_{ij} f_j$ for $i=0,\ldots,n$ we have
    \begin{align}
    R_A(g_0,\ldots,g_n) = \det (D)^{\vol (P)} R_A (f_0,\ldots,f_n)
    \end{align}
    where $P = \Conv (A)$.
\end{lemma}

Especially, for linear functions $g_0,\ldots,g_n$ with $g_i:= \sum_{j=0}^{n} D_{ij}x_j$ in homogenization, this lemma implies $R_A(g_0,\ldots,g_n) = \det(D)$. This result extends also to all cases, where $A$ forms a simplex \cite{GelfandDiscriminantsResultantsMultidimensional1994}, which we want to demonstrate with an example.

\begin{example}
	Consider the system of polynomials
	\begin{align}
		f &= a_1 x_1^2 + a_2 x_1 x_2 + a_3 \nonumber \\
		g &= b_1 x_1^2 + b_2 x_1 x_2 + b_3 \label{eq:ExampleRes}\\
		h &= c_1 x_1^2 + c_2 x_1 x_2 + c_3 \point \nonumber
	\end{align}
	Since $A = \{(2,0),(1,1),(0,0)\}$ generates a full dimensional polytope the configuration is essential, and we will have $\codim \overline {\mathscr Z} = 1$. By eliminating the variables $x_1,x_2$ from the system $f=g=h=0$ we will obtain the $A$-resultant $R_A(f,g,h)$. Alternatively, we can make use of \cref{lem:trafoAres}. It follows that $R_A(f,g,h) = \det (D)^{\vol (P)} R_A(x_1^2,x_1 x_2,1)$, where $D$ is the coefficient matrix of \cref{eq:ExampleRes}. As the resultant of $x_1^2,x_1x_2,1$ is obviously equal to $1$, the $A$-resultant $R_A(f,g,h)$ is simply given by the determinant of $D$
	\begin{align}
		R_A(f,g,h) = - a_3 b_2 c_1 + a_2 b_3 c_1 + a_3 b_1 c_2 - a_1 b_3 c_2 - a_2 b_1 c_3 + a_1 b_2 c_3 \point
	\end{align}

\end{example}


\subsection{\texorpdfstring{$A$}{A}-discriminants} \label{ssec:Adiscriminants}

Closely related to the $A$-resultant is the so-called $A$-discriminant. For a given polynomial $f\in\mathbb C[x_1^{\pm 1},\ldots,x_n^{\pm 1}]$ it describes when the hypersurface $\{f=0\}$ is singular. Equivalently, the $A$-discriminant determines whether $f$ has multiple roots. Let $A\subset \mathbb Z^n$ be the support of the polynomial $f(x) = \sum_{a\in A} z_a x^a$ and consider
\begin{align}
    \nabla_0 = \left\{ f\in\mathbb C^A \,\rvert\, \Var\!\left(f,\frac{\partial f}{\partial x_1},\ldots,\frac{\partial f}{\partial x_n}\right) \neq\emptyset \text{ in } (\mathbb C^*)^n \right\} \subseteq \mathbb C^A \label{eq:Adisc}
\end{align}
the set of polynomials $f\in\mathbb C^A$ for which there exists a solution $x \in (\mathbb C^*)^n$ such that $f$ and its first derivatives vanish simultaneously. In analogy to the $A$-resultant, if the Zariski closure of $\nabla_0$ has codimension $1$, we will set the \textit{$A$-discriminant} $\gls{ADisc}\in\mathbb Z [\{z_a\}_{a\in A}]$ of $f$ as the minimal defining polynomial of the hypersurface $\overline \nabla_0$. For higher codimensions $\operatorname{codim}\!\left(\overline{\nabla}_0\right) > 1$ we will fix the $A$-discriminant to be $1$. Configurations $A$, which having $\Delta_A(f)=1$ are called \textit{defective}. Combinatorial criteria of defective configurations can be found in \cite{EsterovNewtonPolyhedraDiscriminants2010, CurranRestrictionADiscriminantsDual2006, DickensteinTropicalDiscriminants2007}. By definition, the $A$-discriminant is an irreducible polynomial in the coefficients $\{z_a\}_{a\in A}$, which is uniquely determined up to a sign \cite{GelfandDiscriminantsResultantsMultidimensional1994}. 
\bigskip

\begin{example} \label{ex:cubicDisc}
    Consider the cubic polynomial in one variable $f = z_0 + z_1 x + z_2 x^2 + z_3 x^3$ with its support $A = \{0,1,2,3\}$. Its $A$-discriminant is given (up to a sign) by
    \begin{align}
        \Delta_A(f) = z_1^2 z_2^2 - 4 z_0 z_2^3 - 4 z_1^3 z_3 + 18 z_0 z_1 z_2 z_3 - 27 z_0^2 z_3^2
    \end{align}
    which can be calculated either by eliminating $x$ from $f(x) = \pd{f(x)}{x} = 0$ or by the use of a convenient mathematical software program e.g.\ \softwareName{Macaulay2} \cite{GraysonMacaulay2SoftwareSystem} with additional libraries \cite{StaglianoPackageComputationsClassical2018, StaglianoPackageComputationsSparse2020}. Thus, the equations $f(x) = \pd{f(x)}{x} = 0$ have a common solution for $x\neq 0$ if and only if $\Delta_A(f) = 0$.
\end{example}

However, many polynomials share the same $A$-discriminant. Consider two finite subsets $A\subset\mathbb Z^n$ and $A'\subset\mathbb Z^m$, which are related by an injective, affine transformation $T:\mathbb Z^n \rightarrow \mathbb Z^m$ with $T(A)=A'$. Then, the corresponding transformation of $T$ connects also $\Delta_A$ with $\Delta_{A'}$, which was shown in \cite{GelfandDiscriminantsResultantsMultidimensional1994}. Thus, the $A$-discriminant only depends on the affine geometry of $A$. For example consider a configuration $\Aa\subset\mathbb Z^{n+1}$ which generates a homogeneous polynomial $\tilde f\in\mathbb C^\Aa$ and another configuration $A\subset \mathbb Z$, that arises when removing the first entries of $\Aa$. Therefore, we have the dehomogenization map
\begin{align}
    \mathbb C^{\Aa} \rightarrow \mathbb C^A, \qquad \tilde f(x_0,\ldots,x_n) \mapsto f(x_1,\ldots,x_n) = \tilde f(1,x_1,\ldots,x_n) \point
\end{align}
This map is affine and injective, and we can identify both discriminants $\Delta_{\Aa}(\tilde f)=\Delta_A (f)$. Similarly, we obtain for a finite subset $A\subset \mathbb Z^n$ and its homogenization $\Aa\subset\mathbb Z^{n+1}$ the same discriminants.\bigskip

By definition \cref{eq:Adisc} it can be seen, that $\Delta_A(f)$ has to be a homogeneous polynomial. Additionally, $\Delta_A(f)$ is even quasi-homogeneous for any weight defined by a row of $A$ \cite{GelfandDiscriminantsResultantsMultidimensional1994}.  Removing these homogenities leads us to the \textit{reduced $A$-discriminant} $\gls{redADisc}$. Let $\Aa\subset\mathbb Z^{n+1}$ be the homogenization of the support $A$ and $\Bb\in\Gale(\Aa)$ a Gale dual of $\Aa$. Then we can introduce ``effective'' variables\glsadd{effectiveVars}\glsunset{effectiveVars}
\begin{align}
    y_j = \prod_{i=1}^N z_i^{b_{ij}} \qquad\text{for}\quad j=1,\ldots, r \label{eq:effectiveVarsADisc}
\end{align}
where $b_{ij}$ denotes the elements of the Gale dual $\Bb$. The $A$-discriminant can always be rewritten as  $\Delta_A(f) = z^{\Lambda} \Delta_\Bb(f)$, where the reduced $A$-discriminant $\Delta_\Bb(f)$ is an inhomogeneous polynomial in the effective variables $y_1,\ldots,y_r$ and $\Lambda\in\mathbb Z^N$ defines a factor $z^\Lambda$. We will usually choose the smallest $\Lambda$ such that $\Delta_\Bb$ is a polynomial.\bigskip

\begin{example}[Continuation of \cref{ex:cubicDisc}] \label{ex:cubicGale}
    We will consider the cubic polynomial in one variable from \cref{ex:cubicDisc}. As a possible choice for the Gale dual $\Bb$ of the homogenized point configuration $\Aa\subset\mathbb Z^2$ we will use
    \begin{align}
        \mathcal B = \begin{pmatrix}
                    1 & 2 \\
                    -2 & -3\\
                    1 & 0\\
                    0 & 1
                \end{pmatrix}\text{ ,} \qquad y_1= \frac{z_0 z_2}{z_1^2}, \qquad y_2 = \frac{z_0^2 z_3}{z_1^3} \point \label{eq:exEffective}
    \end{align}
    Thus, we can rewrite the $A$-discriminant 
    \begin{align}
        \Delta_A(f) = \frac{z_1^6}{z_0^2} \left(27 y_2^2 + 4 y_1^3 + 4 y_2 - y_1^2 - 18 y_1y_2\right) =  \frac{z_1^6}{z_0^2} \Delta_\Bb(f)
    \end{align}
    as a reduced discriminant $\Delta_\Bb(f)\in\mathbb Z[y_1,y_2]$.
\end{example}

Except for special cases, where $A$ forms a simplex or a circuit \cite{GelfandDiscriminantsResultantsMultidimensional1994}, the determination of the $A$-discriminant can be a very intricate issue for bigger polynomials and calculations quickly get out of hand. Fortunately, there is an indirect description of $A$-discriminants which was invented by Kapranov \cite{KapranovCharacterizationAdiscriminantalHypersurfaces1991} with slight adjustments in \cite{CuetoResultsInhomogeneousDiscriminants2006}. This so called \textit{\HKP} states a very efficient way to study discriminants. Let $\mathcal S\subset(\mathbb C^*)^r$ be the hypersurface defined by the reduced $A$-discriminant $\{\Delta_\Bb(f)=0\}$. Then this hypersurface $\mathcal S$ can be parameterized by the map $\psi :\mathbb P^{r-1}_{\mathbb C} \rightarrow (\mathbb C^*)^r$, where $\psi$ is given by
\begin{align}
    \gls{HKP} = \left( \prod_{i=1}^N \left(\sum_{j=1}^r b_{ij} t_j \right)^{b_{i1}}, \ldots, \prod_{i=1}^N \left(\sum_{j=1}^r  b_{ij} t_j \right)^{b_{ir}}\right) \label{eq:HKPpsi}
\end{align}
and $b_{ij}$ are again the elements of a Gale dual $\mathcal B$ of $\Aa$. Hence, we can give an implicit representation of the $A$-discriminant very quickly only by knowing a Gale dual.

\begin{example}[Continuation of \cref{ex:cubicGale}] \label{ex:cubicParameterization}
    For the example of the cubic polynomial in one variable, we obtain with the Gale dual from \cref{eq:exEffective}
    \begin{align}
        \psi [t_1 : t_2] = \left( \frac{t_1 + 2 t_2}{(2t_1+3t_2)^2} t_1 , - \frac{(t_1 + 2 t_2)^2}{(2t_1+3t_2)^3} t_2\right) \point
    \end{align}
    Since $[t_1 : t_2]$ are homogeneous coordinates, only defined up to multiplication, we can set without loss of generality $t_2=1$. Hence, the statement of \HKP is, that for every $t_1\in\mathbb C$ we can identify
    \begin{align}
        y_1 = \frac{t_1 + 2 }{(2t_1+3)^2} t_1, \qquad y_2 = - \frac{(t_1 + 2)^2}{(2t_1+3)^3}
    \end{align}
    as the points characterizing the hypersurface $\Delta_\Bb(f)(y_1,y_2) = 0$ or equivalently the hypersurface $\Delta_A(f)(z_0,z_1,z_2,z_3)=0$ by means of the relations \cref{eq:exEffective}.
\end{example}

Moreover, Kapranov showed \cite{KapranovCharacterizationAdiscriminantalHypersurfaces1991}, that the map $\psi$ is the inverse of the (logarithmic) Gauss map, which is defined for an arbitrary hypersurface $\mathcal S_g = \{y\in(\mathbb C^*)^r \, |\, g(y)=0\}$ as $\gamma : (\mathbb C^*)^r \rightarrow \mathbb P^{r-1}_{\mathbb C}$, with $\gamma(y) = [ y_1 \partial_1 g(y) : \ldots : y_r \partial_r g(y)]$ for all regular points of $\mathcal S_g$. It is a remarkable fact, that all hypersurfaces $S_g$, which have a birational Gauss map are precisely those hypersurfaces defined by reduced $A$-discriminants \cite{KapranovCharacterizationAdiscriminantalHypersurfaces1991, CuetoResultsInhomogeneousDiscriminants2006}. \bigskip

To conclude this section, we want to mention the relation between $A$-discriminants and the mixed $(A_0,\ldots,A_n)$-resultants, which is also known as Cayley's trick.
\begin{lemma}[Cayley's trick \cite{GelfandDiscriminantsResultantsMultidimensional1994}] \label{lem:CayleysTrick}
    Let $A_0,\ldots,A_{n}\subset \mathbb Z^n$ be finite subsets jointly generating $\mathbb Z^n$ as an affine lattice. By $f_i\in\mathbb C^{A_i}$ we denote the corresponding polynomials of the sets $A_i$. Then we have
    \begin{align}
        R_{A_0,\ldots,A_{n}} (f_0,\ldots,f_{n}) = \Delta_A \!\left(f_{0}(x) + \sum_{i=1}^n y_i f_i(x)\right)
    \end{align}
    where $A\subset \mathbb Z^{2n}$ is defined as the support of the polynomial $f_{0}(x) + \sum_{i=1}^n y_i f_i(x)\in \mathbb C[x_1^{\pm 1},\ldots,x_{n}^{\pm 1},y_1,\ldots,y_n]$.        
\end{lemma}

Thus, every mixed $(A_0,\ldots,A_n)$-resultant can always be reduced to the case of an $A$-discriminant. Hence, the $A$-discriminant is the more general object than the resultant.


\subsection{Principal \texorpdfstring{$A$}{A}-determinants} \label{ssec:principalAdet}

The last object we want to introduce from the book of Gelfand, Kapranov and Zelevinsky \cite{GelfandDiscriminantsResultantsMultidimensional1994} is the principal $A$-determinant, which is a special $A$-resultant. Once again, we consider a finite subset $A\subset \mathbb Z^n$, and we will assume for the sake of simplicity that $\Aff_{\mathbb Z} (A) = \mathbb Z^n$. By $f=\sum_{a\in A} z_a x^a\in\mathbb C^A$ we denote the polynomial corresponding to $A$. The \textit{principal $A$-determinant} is then defined as the following $A$-resultant\footnote{A rather subtle point about this definition is the order of evaluation. A generic $A$-resultant $R_A(g_0, \ldots ,g_n)$ is determined first and the corresponding polynomials $g_0=f, g_1 = x_1 \pd{f}{x_1}, \ldots, g_n = x_n \pd{f}{x_n}$ are inserted second. I would like to thank Simon Telen for pointing this out to me.}
\begin{align}
    \gls{pAdet} := R_A\!\left(f,x_1 \pd{f}{x_1}, \ldots,x_n \pd{f}{x_n}\right) \in \mathbb Z [\{z_a\}_{a\in A}] \point \label{eq:principalAdetResultant}
\end{align}
Thus, the principal $A$-determinant is a polynomial with integer coefficients depending on $\{z_a\}_{a\in A}$, which is uniquely determined up to a sign \cite{GelfandDiscriminantsResultantsMultidimensional1994}. It indicates when the system of polynomial equations $f = x_1 \pd{f}{x_1} = \ldots = x_n \pd{f}{x_n}=0$ has a common solution.\bigskip

In the application of $\Aa$-hypergeometric functions and also for Feynman integrals, it will turn out, that the principal $A$-determinant will be the central object, when considering the analytic structure. We will consider the analytic structure, which is also known as Landau variety for the Feynman integral, more in detail in \cref{ch:singularities}. In this section we will focus on two main properties of the principal $A$-determinant, that will be crucial when applying them to the analytical structure.

The first property, we want to discuss here, is the decomposition of the principal $A$-determinant into a product of several $A$-discriminants.

\begin{theorem}[{Prime factorization of principal $A$-determinant \cite[ch. 10, thm. 1.2]{GelfandDiscriminantsResultantsMultidimensional1994}}] \label{thm:pAdet-factorization}
    The principal $A$-determinant can be written as a product of $A$-discriminants
    \begin{align}
        E_A(f) = \pm \prod_{\tau \subseteq \operatorname{Newt} (f)} \Delta_{A\cap\tau} (f_\tau)^{\mu(A,\tau)} \comma \label{eq:pAdet-factorization}
    \end{align}
    where the product is over all faces $\tau$ of the Newton polytope $\Newt (f)$ and $ \mu(A,\tau) \in\mathbb N_{> 0}$ are certain integers, called multiplicity of $A$ along $\tau$. The exact definition of the multiplicities is not crucial for the following, which is why we refer to \cite{GelfandDiscriminantsResultantsMultidimensional1994, ForsgardZariskiTheoremMonodromy2020} at this point.
\end{theorem}

Most often we are only interested in the roots of the principal $A$-determinant. Therefore, we want to define a \textit{simple principal $A$-determinant} according to \cite{ForsgardZariskiTheoremMonodromy2020} where all multiplicities $\mu(A,\tau)$ in \cref{eq:pAdet-factorization} are set to $1$
\begin{align}
    \gls{spAdet} = \pm \prod_{\tau \subseteq \operatorname{Newt} (f)} \Delta_{A\cap\tau} (f_\tau)
\end{align}
which generates the same variety as $E_A(f)$. \bigskip

\begin{example}[principal $A$-determinants of homogeneous polynomials] 
    To illustrate the principal $A$-determinant we will recall an example from \cite{GelfandDiscriminantsResultantsMultidimensional1994}. Let $\tilde f\in\mathbb C[x_0,\ldots,x_n]$ be a homogeneous polynomial consisting in all monomials of a given degree $d\geq 1$. The support of $\tilde f$ will be called $\Aa\subset\mathbb Z^{n+1}$ and its Newton polytope $\Newt(\tilde f)\subset \mathbb R^{n+1}$ is an $n$-dimensional simplex, having $2^{n+1}-1$ faces. Moreover, the faces $\tau\subseteq\Newt(\tilde f)$ are generated by all non-empty subsets of the set of vertices $\{0,\ldots,n\}$ and one can show that all multiplicities $\mu(\Aa,\tau)$ are equal to one \cite{GelfandDiscriminantsResultantsMultidimensional1994}. Thus, we get
    \begin{align}
        E_\Aa(\tilde f) = \pm \prod_{\emptyset\neq \tau \subseteq \{0,\ldots,n\}} \Delta_{\Aa\cap \tau} (\tilde f_{\tau}) \label{eq:pAdetHom} \point
    \end{align}
    Note that in this particular example, we can write $\tilde f_\tau$ as the polynomial $\tilde f$ with all variables $x_i=0$ set to zero that do not belong to $\tau$, i.e.\ $\tilde f_{\tau} (x) = \tilde f(x)|_{x_i=0, i\notin \tau}$. That decomposition is exactly the behaviour we would naively expect in the polynomial equation system
    \begin{align}
        \tilde f (x) = x_0 \frac{\partial \tilde f(x)}{\partial x_0} = \ldots =  x_n \frac{\partial \tilde f(x)}{\partial x_n} = 0 \point \label{eq:expAdetPolySystem}
    \end{align}
    That is, in \cref{eq:expAdetPolySystem} we can consider each combination of how variables $x_0,\ldots,x_n$ can vanish separately. However, it should be remarked that this behaviour is not true in general since multivariate division is not necessarily unique. Hence, when not considering simplices, we have rather to take the truncated polynomials into account as described in \cref{thm:pAdet-factorization}. \bigskip
\end{example}

Another noteworthy result of Gelfand, Kapranov and Zelevinsky is the connection between the principal $A$-determinant and the triangulations of $A$.
\begin{theorem}[Newton polytopes of principal $A$-determinants \cite{GelfandNewtonPolytopesClassical1990, SturmfelsNewtonPolytopeResultant1994, GelfandDiscriminantsResultantsMultidimensional1994}] \label{thm:NewtSec}
    The Newton polytope of the principal $A$-determinant and the secondary polytope coincide
    \begin{align}
        \Newt(E_A(f)) = \Sigma(A) \point
    \end{align}
    Further, if $\Tt$ is a regular triangulation of $A$, then the coefficient of the monomial $\prod_{a\in A} z_a^{\varphi_\Tt(a)}$ in $E_A(f)$ is equal to
    \begin{align}
        \pm \prod_{\sigma\in \hatT} \vol(\sigma)^{\vol(\sigma)} \point
    \end{align}
    And also the relative signs between the coefficients in $E_A(f)$ can be determined \cite[chapter 10.1G]{GelfandDiscriminantsResultantsMultidimensional1994}.
\end{theorem}

Thus, by the knowledge of all regular triangulations we can approximate the form of the principal $A$-determinant. As \cref{thm:NewtSec} gives us the extreme monomials we can make a suitable ansatz for the principal $A$-determinant. The unknown coefficients of the monomials corresponding to potential interior points of the Newton polytope $\Newt(E_A(f))$ can be determined then by \HKp.

Further, the theorem explains the typical appearing coefficients\footnote{Based on that theorem, Gelfand, Kapranov and Zelevinsky \cite{GelfandDiscriminantsResultantsMultidimensional1994} speculate on a connection between discriminants and probabilistic theory. Their starting point for this consideration are the ``entropy-like'' expressions as $\prod v_i^{v_i}=e^{\sum v_i \log v_i}$ for the coefficients of the principal $A$-determinant, as well as in the \HKp. In the latter case we can define an ``entropy'' $S := \ln\!\left (\prod_{l=1}^r \psi_l^{t_l}\right) = \sum_{i=1}^N \rho_i(t) \ln(\rho_i(t))$ where $\rho_i(t):= \sum_{j=1}^r b_{ij} t_j$. To the author's knowledge, more rigorous results about such a potential relation are missing. However, there are further connections known between tropical toric geometry and statistical thermodynamics as presented in \cite{KapranovThermodynamicsMomentMap2011, PassareAmoebasComplexHypersurfaces2013}. In any case, a fundamental understanding of such a relation could be very inspiring for a physical point of view.} in principal $A$-determinants and Landau varieties, like $1=1^1$, $4=2^2$, $27=3^3$, $4^4=256$, etc.

\begin{example}[Continuation of \cref{ex:cubicParameterization}] \label{ex:cubicPrincipalAdetFromTriangs}
    We will continue the example from \cref{ssec:Adiscriminants}. From the set $A=(0,1,2,3)$ we can determine $4$ triangulations with the weights $\varphi_{\Tt_1} = (3,0,0,3)$, $\varphi_{\Tt_2}=(1,3,0,2)$, $\varphi_{\Tt_3}=(2,0,3,1)$ and $\varphi_{\Tt_4}=(1,2,2,1)$. Adding the only possible interior point $(2,1,1,2)$, we obtain the following ansatz by use of \cref{thm:NewtSec}
    \begin{align}
        E_A(f) = 27 z_0^3 z_3^3 + 4 z_0 z_1^3 z_3^2 + 4 z_0^2 z_2^3 z_3 - z_0 z_1^2 z_3 + \alpha z_0^2 z_1 z_2 z_3^2
    \end{align}
    where we have to determine $\alpha\in\mathbb Z$. By considering the faces of $\Newt(f)$ we can split $E_A(f)$ into discriminants
    \begin{align}
        E_A(f) = z_0 z_3 \Delta_A(f) = z_0 z_3 \frac{z_1^6}{z_0^2} \Delta_\Bb(f)
    \end{align}
    with the reduced $A$-discriminant $\Delta_\Bb(f) = 27 y_2^2 + 4y_1^3 +4y_2 -y_1^2 + \alpha y_1 y_2$ with the same conventions as in \cref{ex:cubicGale}.  By the \HKP (see \cref{ex:cubicParameterization}) we can calculate those points $(y_1,y_2)$ which satisfy $\Delta_\Bb(f)=0$. Choosing for example $t_1=-1$, we obtain the point $(y_1,y_2)=(-1,-1)$, which leads to $\alpha=-18$. Hence, we determined the principal $A$-determinant $E_A(f)$ by means of the triangulations of the point configuration $A$ and the \HKp.
\end{example}

In general, we can replace $E_A(f)$ by means of Cayley's trick (\cref{lem:CayleysTrick}) by one single $A$-discriminant in order to simplify the usage of Horn-Kapranov-parameterization. However, we have to mention that the number of vertices of the secondary polytope -- or equivalently the number of regular triangulations -- grows very fast. To determine the Landau variety in the application of Feynman integrals (see \cref{ch:singularities}) of the $2$-point, $3$-loop Feynman diagram (also known as $3$-loop banana) we have to consider $78.764$ possible regular triangulations. Hence, the principal $A$-determinant will have more than $78\,764$ monomials. The $2$-loop double-edged triangle graph (or dunce's cap) generates even $885\,524$ triangulations. Further numbers to regular triangulations of Feynman diagrams can be found in the \cref{sec:AppendixCharacteristics}. Nevertheless, this approach could be faster, than the direct calculation of principal $A$-determinants by standard algorithms, since there are very efficient methods known for triangulations \cite{DeLoeraTriangulations2010, RambauTOPCOMTriangulationsPoint2002}. However, we want to remark that the triangulations are not the only way to construct the secondary polytope. We will refer to \cite{BilleraConstructionsComplexitySecondary1990} for details.


\section{Holonomic \texorpdfstring{$D$}{D}-modules} \label{sec:holonomicDmodules}
%

The theory of $\Aa$-hypergeometric functions will be expressed by solutions of systems of partial linear differential equations. We therefore will rephrase certain basic terms of $D$-modules. In this short overview we will forgo to have a recourse to the notion of sheaves. This section will be oriented towards \cite{SaitoGrobnerDeformationsHypergeometric2000, TakayamaGrobnerBasisRings2013, SattelbergerDModulesHolonomicFunctions2019} which we also refer for a more detailed description together with \cite{OakuComputationCharacteristicVariety1994}. For a general introduction to the subject of $D$-modules we suggest \cite{BjorkAnalyticDModulesApplications1993}.

The consideration of $D$-modules takes place in the \textit{Weyl algebra}
\begin{align}
	\gls{Weyl} = \mathbb K \langle z_1,\ldots,z_N, \partial_1,\ldots,\partial_N \rangle 
\end{align}
where $\mathbb K$ is any field and all generators suppose to commute except for $\partial_i z_i = z_i\partial_i + 1$. The Weyl algebra is isomorphic to the ring of differential operators on the affine space $\mathbb A_{\mathbb K}^N$ \cite{SaitoGrobnerDeformationsHypergeometric2000}, which is the reason for our special interest. However, when we want to consider differential operators with rational functions as coefficients we will use the so-called \textit{rational Weyl algebra}
\begin{align}
	R = \mathbb K (z_1,\ldots,z_N) \langle \partial_1,\ldots,\partial_N\rangle
\end{align}
where $\mathbb K(z_1,\ldots,z_N)$ is the field of rational functions and $\partial_i r(z) = r(z)\partial_i + \pd{r(z)}{z_i}$ will be the analogue commutator relation for any rational function $r(z)\in\mathbb K(z_1,\ldots,z_N)$. Note, that $D$ is a subalgebra of $R$.

Any left module over the ring $D$ will be called a $D$-module $M$. Many function spaces can be considered as $D$-modules, for example the ring of polynomials $\mathbb K[z_1,\ldots,z_N]$, the space of formal power series $\mathbb K[[z_1,\ldots,z_N]]$ or the space of holomorphic functions $\mathcal O(U)$ on an open set $U\subseteq \mathbb C^N$ by using the natural identification of $\partial_i$ as a derivative and $z_i$ as a multiplication
\begin{align}
	\bullet : D \times M \rightarrow M,\quad \partial_i \bullet f = \pd{f}{z_i}, \quad z_i\bullet f = z_i f \quad \text{ for } f\in M \point
\end{align}
In order to distinguish the module action $\bullet$ from the multiplication in the Weyl algebra $\cdot : D \times D \rightarrow D$ we will use different symbols. Thus, a system of linear differential equations with polynomial coefficients can be identified with a left ideal $I\subset D$ in the Weyl algebra. The \textit{solution space} for such an ideal $I\subset D$ is the $\mathbb C$-vector space 
\begin{align}
    \gls{Sol}= \{ f\in\mathcal O(U) \,\rvert\, P \bullet f = 0 \quad \forall\, P\in I \}
\end{align}
where $\mathcal O(U)$ is the $D$-module of holomorphic functions on the domain $U\subseteq\mathbb C^N$. 
%

\begin{example} \label{ex:D1ideal}
	To illustrate the definitions above we will consider an example in the well-known univariate case. Thus, let $I_1 = \langle (z \partial + 1)^2 (z\partial - 2) \rangle=\langle z^3 \partial^3 + 3 z^2\partial^2 - 2 z \partial - 2\rangle$ be a $D$-ideal. It is not hard to solve the corresponding differential equation, and we will obtain $\Sol(I_1) = \mathbb C \{ z^{-1}, z^{-1}\ln(z), z^2 \}$. Surely, in order to have a well-defined solution space, we will consider $z\in U$ where $U\subseteq \mathbb C^*$ is a simply connected domain.
\end{example}

As seen in this example the solutions are not necessarily entire functions, they may have singularities. Let us first recall the situation in the univariate case to get an intuition for the multivariate generalization. Let $P=c_m(z) \partial^m + \ldots +c_1(z) \partial + c_0(z)\in D$ be a differential operator in a single variable $z\in\mathbb C$ with polynomials $c_i(z)\in\mathbb C[z]$ as coefficients and $c_m(z) \not\equiv 0$. We will call the roots of the leading coefficient $c_m(z)$ the singular points of $P$ and the set of all roots is known as the singular locus $\Sing(D\cdot P)$. Standard existence theorems state then \cite{CoddingtonTheoryOrdinaryDifferential2012,SaitoGrobnerDeformationsHypergeometric2000}, that for a simply connected domain $U\subseteq\mathbb C\setminus\Sing(D\cdot P)$ there exist holomorphic functions on $U$, which are solutions of $P\bullet f=0$. Moreover, the dimension of the solution space (for holomorphic functions on $U$) will be equal to $m$. Therefore, the singular locus will describe potential singularities in the analytic continuation of the solutions. In the univariate case one usually distinguish further between irregular and regular singular points. By Frobenius' method, one can construct series solutions for the latter. A compact summary of this method in the univariate case can be found e.g.\ in \cite{CattaniThreeLecturesHypergeometric2006}. \bigskip

We will now turn to the multivariate case. Every element $P\in D$ has a \textit{normally ordered expression}
\begin{align}
	P = \sum_{(\alpha,\beta) \in E} c_{\alpha\beta} z^\alpha \partial^\beta \label{eq:normallyOrderedEx}
\end{align}
where we use the multi-index notation as usual and $c_{\alpha\beta} \in \mathbb K^*$. Since we have more than one generator in the Weyl algebra $D$, generators can enter with different weights, when introducing an order on $D$. Hence, the vector $(u,v)\in\mathbb R^{2N}$ will be a \textit{weight} for the Weyl algebra $D$, where we associate $u$ with the weights of the generators $z_1,\ldots,z_N$ and $v$ with the weights of the generators $\partial_1,\ldots,\partial_N$. In doing so, we can define an \textit{order} for every element $P\in D$ by $\gls{order} := \max_{(\alpha,\beta)\in E} (u\alpha+v\beta)$, where $E$ refers to the normally ordered expression \cref{eq:normallyOrderedEx}. The order allows to define a filtration $\ldots \subseteq F_{-1} \subseteq F_0\subseteq F_1 \subseteq \ldots$ of $D$ by
\begin{align}
	F_m = \left\{ P\in D \,\rvert\, \ord(P)\leq m \right\} \point
\end{align}
If $u_i\geq 0$ and $v_i\geq 0$ the filtration is bounded by below $\{0\}\subseteq F_0 \subseteq \ldots$, and we will always assume that case in this work. The \textit{associated graded ring} $\gls{assocGrRing}$ of this filtration will be generated by
\begin{align}
	\{z_1,\ldots,z_N\} \cup \{ \partial_i \,\rvert\, u_i+v_i=0\} \cup \{ \xi_i \,\rvert\, u_i + v_i > 0 \}
\end{align}
where we replace the generators $\partial_i$ in the Weyl algebra $D$ by commuting generators $\xi_i$ whenever$u_i+v_i > 0$. For example, we have $\gr_{(\mathbf 0,\mathbf 0)}(D) = D$ and $\gr_{(\mathbf 0, \mathbf 1)}(D) = \mathbb K [z,\xi] = \mathbb K[z_1,\ldots,z_N,\xi_1,\ldots,\xi_N]$, where $\mathbf 0:=(0,\ldots,0)$ and $\mathbf 1:=(1,\ldots,1)$ stand for the constant zero and the constant one vector, respectively.

The order allows us also to define an analogue of the leading term. By the \textit{initial form} of an element from the Weyl algebra $P\in D$ we understand the following element in $\gr_{(u,v)} (D)$ 
\begin{align}
	\initial_{(u,v)} (P) = \sum_{\substack{ (\alpha,\beta)\in E \\ \alpha u +\beta v = \ord (P)}} c_{\alpha\beta} \prod_{i:\ u_i+v_i > 0} z_i^{\alpha_i} \xi_i^{\beta_i} \prod_{i:\ u_i+v_i = 0} z_i^{\alpha_i} \partial_i^{\beta_i} \in\gr_{(u,v)} (D) \point \label{eq:initalForm}
\end{align}
Furthermore, for any left ideal $I\subseteq D$ we define $\gls{initialuv} := \mathbb K \cdot \{ \initial_{(u,v)} (P) \,\rvert\, P \in I \} \subseteq \gr_{(u,v)} (D)$ as the \textit{initial ideal} of $I$ with respect to $(u,v)$. \bigskip

However, $\ord(P)$ defines only a partial order on $D$ because different elements of the Weyl algebra can have the same order. A total order $\prec$ is called a \textit{term order}, if (i) $z^\alpha\partial^\beta \prec z^{\alpha^\prime} \partial^{\beta^\prime}$ implies $z^{\alpha+s}\partial^{\beta+t} \prec z^{\alpha^\prime+s} \partial^{\beta^\prime+t}$ for any $(s,t)\in\mathbb N^{2N}$ and (ii) $1=z^0\partial^0$ is the smallest element with respect to $\prec$. When $z^\alpha\partial^\beta$ is the largest monomial of $P\in D$ in the normally ordered expression \cref{eq:normallyOrderedEx} with respect to a term order $\prec$, we will call $z^\alpha\xi^\beta \in\mathbb K[z,\xi]$ the \textit{initial monomial} $\initial_\prec(P)$ of $P$. By slightly abuse of notation, we call the ideal of all initial monomials $\initial_\prec(I) := \mathbb K \cdot \{ \initial_\prec (P) \,\rvert\, P \in I \}\subseteq\mathbb K[z,\xi]$ the \textit{initial ideal} with respect to $\prec$. The monomials which do not lie in $\initial_\prec(I)$ are called the \textit{standard monomials} of $I$. The number of standard monomials can be finite or infinite. \bigskip

Analogue to the commuting case we call a generating set $G=\{g_1,\ldots,g_m\}$ of the ideal $I\subseteq D$ a \textit{Gröbner basis} of $I$ with respect to a term order $\prec$, if the initial monomials $\{\initial_\prec(g_1),\ldots,\initial_\prec(g_m)\}$ generate the initial ideal $\initial_\prec (I)$. Note, that $\initial_\prec (I)$ is an ideal in the commutative ring $\mathbb K[z,\xi]$, whereas $I\subseteq D$ is an ideal in a noncommutative ring. Let $\prec$ be a term order sorting monomials by their order $\ord(P)$ in case when the orders a different and by a further term order, e.g.\ a lexicographic order, in case where the orders are equal. A Gröbner basis $G$ with respect to such a term order satisfies $\langle\initial_{(u,v)}(G)\rangle = \initial_{(u,v)} (I)$ \cite{SaitoGrobnerDeformationsHypergeometric2000}. The benefit of Gröbner basis is, that every element of $I$ has a standard representation in terms of its Gröbner basis, and we can reduce ideal operations to operations on its Gröbner basis. We refer to \cite{SaitoGrobnerDeformationsHypergeometric2000} for the determination of those Gröbner bases by an extended version of Buchberger's algorithm. For an actual calculation of Gröbner bases we recommend the use of software e.g.\ \softwareName{Macaulay2} \cite{GraysonMacaulay2SoftwareSystem} or \softwareName{Singular} \cite{DeckerSingular421Computer2021}. We included in the \cref{sec:SoftwareTools} several examples on the usage of those programs. \bigskip

In the following we will specify the weight to $(\mathbf 0,\mathbf 1)$, i.e.\ we will give the generators $\partial_i$ the weight $1$, whereas the generators $z_i$ get the weight zero. In this case the associated graded algebra is a commutative polynomial ring $\gr_{(\mathbf0,\mathbf 1)}(D) = \mathbb K [z,\xi]$. For this specific weight, the initial forms $\initial_{(\mathbf 0,\mathbf 1)}(P)$ are also called \textit{principal symbols} and the initial ideal is known as \textit{characteristic ideal} $\initial_{(\mathbf 0,\mathbf 1)}(I)\subseteq \mathbb K[z,\xi]$. Furthermore, the variety generated by this ideal is called the \textit{characteristic variety}
\begin{align}
	\gls{char} := \Var\!\left(\initial_{(\mathbf 0,\mathbf 1)}(I)\right) \subseteq \mathbb A_{\mathbb K}^{2N} \point
\end{align}

\begin{example} \label{ex:D2ideal}
	In order to illustrate these definitions, we will study two more examples. The first will be the ideal $I_2 = \langle z_1 \partial_1 - 1, z_2\partial_2^2 + \partial_1^2 \rangle$. Its (reduced) Gröbner basis is $\{z_1\partial_1 - 1, \partial_1^2, z_2\partial_2^2\}$, which can be determined by a convenient software, e.g.\ \softwareName{Macaulay2} \cite{GraysonMacaulay2SoftwareSystem} or \softwareName{Singular} \cite{DeckerSingular421Computer2021}. Thus, its characteristic ideal is given by $\initial_{(\mathbf 0,\mathbf 1)}(I_2) = \langle z_1 \xi_1, \xi_1^2, z_2\xi_2^2\rangle$. 
\end{example}
\begin{example} \label{ex:D3ideal}
	For a slightly more extensive example with $N=3$ consider $I_3 = \langle \partial_1\partial_3 - \partial_2^2, z_1\partial_1 + z_2\partial_2+z_3\partial_3 + 2, z_2 \partial_2 + 2 z_3 \partial_3 - 1\rangle$. The Gröbner basis with respect to a degree reverse lexicographic ordering can be calculated with \softwareName{Macaulay2} or \softwareName{Singular} and one obtains
	\begin{align}
		\{&\partial_2^2-\partial_1\partial_3,   z_2\partial_2+2z_3\partial_3-1,   z_1\partial_1-z_3\partial_3+3,    z_2\partial_1\partial_3+2z_3\partial_2\partial_3, \nonumber \\
		&\quad 2z_1z_3\partial_2\partial_3+z_2z_3\partial_3^2-2z_2\partial_3,    z_2^2z_3\partial_3^2-4z_1z_3^2\partial_3^2-2z_2^2\partial_3-2z_1z_3\partial_3\} \point
	\end{align}
	Therefore, its characteristic ideal is given by 
	\begin{align}
		\initial_{(\mathbf 0,\mathbf 1)} (I_3) &= \langle \xi_2^2-\xi_1\xi_3,  z_2\xi_2+2z_3\xi_3,   z_1\xi_1-z_3\xi_3,    z_2\xi_1\xi_3 + 2z_3\xi_2\xi_3,  \nonumber \\
		&\quad 2z_1z_3\xi_2\xi_3 + z_2z_3\xi_3^2, z_2^2z_3\xi_3^2 - 4z_1z_3^2\xi_3^2 \rangle \nonumber \\
		& = \langle  \xi_1\xi_3 - \xi_2^2, z_1\xi_1  - z_3\xi_3, z_2 \xi_2 + 2 z_3 \xi_3 \rangle \point
	\end{align}
\end{example}

The dimension of the characteristic variety will be crucial for the behaviour of the $D$-ideal. The so-called \textit{weak fundamental theorem of algebraic analysis} going back to \cite{BernshteinAnalyticContinuationGeneralized1973} states that for an ideal $I\subsetneq D$
\begin{align}
	\dim (\ch (I)) \geq N \label{eq:weakFTAA}
\end{align}
the dimension of the characteristic ideal is at least $N$. Moreover, it can be shown that even every irreducible component of $\ch(I)$ has at least dimension $N$, which is the statement of the \textit{strong fundamental theorem of algebraic analysis} \cite{SaitoGrobnerDeformationsHypergeometric2000, KomatsuMicrofunctionsPseudodifferentialEquations1973}. We will call an $D$-ideal $I$ \textit{holonomic} if its characteristic variety has the minimal dimension $N$. Holonomic $D$-ideals will behave in a certain way much better than non-holonomic $D$-ideals. A module $M=\normalslant{D}{I}$ is called holonomic if the ideal $I$ is holonomic.

Furthermore, we can classify ideals $I$ by the number of their standard monomials. Thus, we will slightly change our perspective, and we will consider the $D$-ideal as an ideal of the rational Weyl algebra $R$. For any term order $\prec$, we call the number of standard monomials of $\initial_\prec(I)$ the \textit{holonomic rank} $\rank (I)$ of $I\subseteq R$. We will see in \cref{thm:CKKT}, that the holonomic rank will be the generalization of the order of ordinary differential equations. Note, that if $I$ is holonomic, then its holonomic rank is finite \cite{SaitoGrobnerDeformationsHypergeometric2000}. However, the converse is not necessarily true.

\begin{example}
	We will continue the \cref{ex:D1ideal,,ex:D2ideal,,ex:D3ideal}. All ideals $I_1$, $I_2$ and $I_3$ are holonomic ideals. For a convenient term order $\prec$ we have $\initial_\prec(I_1) = \langle \xi^3 \rangle$ as the initial ideal of $I_1$ in the rational Weyl algebra $R$. Thus, its standard monomials are $\{1,\xi,\xi^2\}$ and its holonomic rank is $3$. Note, that the holonomic rank equals the order of the differential operator.
	
	For the second example, the initial ideal in the rational Weyl algebra will be $\initial_\prec(I_2) = \langle \xi_1,\xi_1^2,\xi_2^2\rangle$. Therefore, the standard monomials are $\{1,\xi_2\}$ and the holonomic rank is $2$. 
	
	For the last example with $N=3$, we obtain $\rank(I_3) = 2$, where the standard monomials are $\{1,\xi_3\}$.
\end{example}

By the holonomic rank we introduced a generalization of the order of ordinary differential equations. The last object we have to generalize to the multivariate case is the singular locus. Let $\pi:\mathbb C^{2N} \rightarrow \mathbb C^N$ denotes the projection on the first coordinates $(z,\xi)\mapsto z$. The Zariski closure of such a projection of the characteristic variety without the trivial solution $\xi_1=\ldots=\xi_N=0$ will be defined to be the \textit{singular locus}
\begin{align}
	\gls{Sing} := \overline{ \pi (\ch(I)\setminus\{\xi=0\})} \subseteq \mathbb C^N \point \label{eq:SingLocusDef}
\end{align}
If $I$ is a holonomic $D$-ideal, the singular locus $\Sing(I) \subsetneq \mathbb C^N$ is a proper subset of $\mathbb C^N$ \cite{SattelbergerDModulesHolonomicFunctions2019}. \bigskip

We can now turn to the central existence and uniqueness theorem for partial differential equations. It is a special case of the Cauchy-Kovalevskaya-Kashiwara theorem.

\begin{theorem}[Cauchy-Kovalevskaya-Kashiwara \cite{SaitoGrobnerDeformationsHypergeometric2000}] \label{thm:CKKT}
	Let $I\subseteq D$ be a holonomic $D$-ideal and $U\subset \mathbb C^N\setminus\Sing(I)$ a simply connected domain. Then the dimension of the $\mathbb C$-vector space of holomorphic solutions on $U$ is equal to its holonomic rank
	\begin{align}
		\dim\!\left(\Sol(I)\right) = \rank (I) \point
	\end{align}
\end{theorem}

Therefore, the notions from the univariate case will transfer also to the multivariate case. However, the involved objects are much harder to determine in most cases, as the calculation of Gröbner basis can be very time-consuming. We will conclude this section by continuing the examples.

\begin{example}
	Let $I_1$, $I_2$ and $I_3$ be the $D$-ideals defined in the \cref{ex:D1ideal,,ex:D2ideal,,ex:D3ideal}, respectively. As expected, the singular locus for the first example is simply $\Sing(I_1) = \Var (z)$. Thus, $z=0$ is the only singular point for $I_1$ and also the only possible singularity of its solutions. Quite similar is the situation for $I_2$, where we obtain $\Sing(I_2) = \Var (z_2)$.
	
	For the last example we have $\Sing(I_3) = \Var (z_1 z_3 (z_2^2 - 4z_1 z_3))$, which can be determined e.g.\ by \softwareName{Macaulay2} or by eliminating $\xi$ from $\ch(I_3)$. Note that the singular locus is generated by the principal $A$-determinant of the quadratic polynomial $f = z_1 + z_2 x + z_3 x^2$. Thus, we have $E_A(f) = z_1 z_3 (z_2^2 - 4z_1 z_3)$. In the following section it will turn out that this is not just an arbitrary coincidence but rather an instance of a more general correspondence between principal $A$-determinants and singular loci.	
\end{example}


\section{\texorpdfstring{$\Aa$}{A}-hypergeometric systems} \label{sec:AHypSystems}
%

Since the first hypergeometric function was studied by Euler and Gauss more than 200 years ago, many different generalization of hypergeometric functions were introduced: Pochhammer series ${}_pF_q$, Appell's, Lauricella's and Kampé-de-Fériet functions, to name a few. Those functions can be characterized in three different ways: by series representations, by integral representations and as solutions of partial differential equations. Starting with Gauss' hypergeometric function ${_2}F_1$ we have
\begin{align}
	{_2}F_1 (a,b,c|z) = \sum_{k\geq 0} \frac{(a)_k(b)_k}{(c)_k} \frac{z^k}{k!} \label{eq:GaussSeries}
\end{align}
as a series representation. We call $\gls{Pochhammer} := \frac{\Gamma(a+k)}{\Gamma(a)}$ the Pochhammer symbol, and we assume appropriate limits, if both $\Gamma$-functions have a pole. 

Alternatively, there are many integral representations known for the Gauss' hypergeometric function, e.g.\ \cite{OlverNISTHandbookMathematical2010, BerkeschEulerMellinIntegrals2013}
\begin{align}
	{_2}F_1 (a,b,c|z) &= \frac{\Gamma(c)}{\Gamma(b)\Gamma(c-b)} \int_0^1 \frac{t^{b-1} (1-t)^{c-b-1}}{(1-zt)^a} \dif t \label{eq:GaussEuler}\\
	{_2}F_1 (a,b,c|z) &= \frac{\Gamma(c)^2}{\Gamma(a)\Gamma(b)\Gamma(c-a)\Gamma(c-b)} \int_{\mathbb R^2_+} \frac{x_1^{a-1}x_2^{b-1}}{\left(1+x_1+x_2+(1-z)x_1x_2\right)^c} \dif x_1 \dif x_2 \label{eq:GaussEulerMellin}\\
	{_2}F_1 (a,b,c|z) &= \frac{\Gamma(c)}{\Gamma(a)\Gamma(b)} \frac{1}{2\pi i} \int_{-i\infty}^{i\infty} \frac{\Gamma(a+t)\Gamma(b+t)}{\Gamma(c+t)} \Gamma(-t) (-z)^t \dif t \label{eq:GaussMellinBarnes}
\end{align}
for convenient domains, which are known as Euler integral, Euler-Mellin integral and Mellin-Barnes integral, respectively. And finally Gauss' hypergeometric function can also be considered as a solution of the differential equation
\begin{align}
	\big[ z(z-1) \partial^2_z + ((a+b+1)z-c)\partial_z + ab\big] \bullet {_2}F_1(a,b,c|z) = 0 \point \label{eq:GaussODE}
\end{align}
We can reformulate this differential equation also in a more symmetric form, which shows the connection to \cref{eq:GaussSeries}
\begin{align}
	\big[ (\theta_z+a)(\theta_z+b) - (\theta_z+c)\partial_z \big] \bullet  {_2}F_1(a,b,c|z) = 0 \label{eq:GaussODE2}
\end{align}
where $\theta_z:= z\partial_z$ is the Euler operator.

Therefore, there are in principle three different branches to generalize the notion of a hypergeometric function: by generalizing the series representation \cref{eq:GaussSeries}, the various integral representations \cref{eq:GaussEuler}, \cref{eq:GaussEulerMellin}, \cref{eq:GaussMellinBarnes} or by generalizing the differential equations \cref{eq:GaussODE}, \cref{eq:GaussODE2}. However, these generalizations do not necessarily agree. Nonetheless, we would like to have the three different kinds of representation also for generalized hypergeometric functions. \bigskip

The most general series representation goes back to Horn \cite{HornUberHypergeometrischeFunktionen1940} and was later investigated by Ore and Sato (a summarizing discussion can be found in \cite{GelfandGeneralHypergeometricSystems1992}). A \textit{Horn hypergeometric series} is a multivariate power series in the variables $y_1,\ldots,y_r\in\mathbb C$
\begin{align}
    \sum_{k\in\mathbb N_0^r} c(k) y^k \label{eq:DefHornHypergeometric}
\end{align}
where ratios of the coefficients $\frac{c(k+e_i)}{c(k)}$ are supposed to be rational functions in $k_1,\ldots k_r$. By $e_1,\ldots,e_r$ we denote the elements of the standard basis in Euclidean space. Thus, the coefficients $c(k)$ can be represented mainly by a product of Pochhammer symbols\footnote{By the property of the Pochhammer symbols to satisfy $(a)_n^{-1} = (-1)^n (1-a)_{-n}$ for $n\in\mathbb Z$ one can convert Pochhammer symbols in the denominator to Pochhammer symbols in the numerator and vice versa. The most general form of those coefficients $c(k)$ is given by the Ore-Sato theorem \cite{GelfandGeneralHypergeometricSystems1992}.} $\prod_i (a_i)_{l_i(k)}$. Thereby, we consider $a_i\in\mathbb C$ as complex numbers and $l_i(k)$ as integer linear combinations of the summation indices $k_1,\ldots,k_r$. Since derivatives of Pochhammer symbols with respect to $a_i$ can be expressed by a sum of other Pochhammer symbols, derivatives of Horn hypergeometric functions with respect to their parameters are again Horn hypergeometric functions \cite{BytevDerivativesHorntypeHypergeometric2017}. For further studies of Horn hypergeometric functions we refer to \cite{SadykovHypergeometricFunctionsSeveral2002}. However, we want to remark that for $r>2$ not all Horn hypergeometric functions can be expressed as a solution of a holonomic $D$-ideal \cite{GelfandGeneralHypergeometricSystems1992}. \bigskip

Another option to generalize the notion of hypergeometric functions is to start with the integral representations. This was worked out by Aomoto, among others, who generalized the Euler integral representation to integrals over a product of linear forms up to certain powers \cite{AomotoStructureIntegralsPower1977}. However, the integration region for multivariate hypergeometric integrals can be very intricate. We refer to \cite{AomotoTheoryHypergeometricFunctions2011} for a comprehensive overview about those hypergeometric integrals. \bigskip

Finally, we can choose the differential equation as starting point for a generalized meaning of hypergeometric functions. This approach was initiated in the late 1980s by Gelfand, Graev, Kapranov, Zelevinsky and collaborators \cite{GelfandCollectedPapersVol1989, GelfandHypergeometricFunctionsToric1991, GelfandGeneralizedEulerIntegrals1990, GelfandGeneralHypergeometricSystems1992, GelfandDiscriminantsResultantsMultidimensional1994}. Due to their dependence on a finite set of lattice vectors $\Aa$, they are known as $\Aa$-hypergeometric functions or occasionally as Gelfand-Kapranov-Zelevinsky (GKZ) hypergeometric functions. 

$\Aa$-hypergeometric functions are defined as solutions of certain holonomic $D$-ideals and include Aomoto's hypergeometric functions \cite{OpdamMultivariableHypergeometricFunctions2001}, as well as all holonomic Horn hypergeometric functions \cite{CattaniThreeLecturesHypergeometric2006}. And vice versa, we can express $\Aa$-hypergeometric functions also in terms of series or integrals as in the case of Gauss' hypergeometric function. The numerous representations for Feynman integrals thus appear naturally in the light of $\Aa$-hypergeometric functions. In the application to Feynman integrals we can identify a generalization of \cref{eq:GaussEulerMellin} with the parametric representation of Feynman integrals, and we will find its connection to the Mellin-Barnes representations \cref{eq:GaussMellinBarnes} which are ubiquitous in Feynman calculus (see e.g.\ \cite{SmirnovFeynmanIntegralCalculus2006}) in \cref{thm:MellinBarnesRepresentation}. Also, a relation to Euler type integrals \cref{eq:GaussEuler} can be given, which we will elaborate in \cref{sec:EulerIntegrals}. Furthermore, the whole \cref{ch:seriesRepresentations} will be devoted to the connection of Feynman integrals to hypergeometric series representations.

The theory of $\Aa$-hypergeometric functions not only characterizes those functions, it will also give us deep insights into the structure of hypergeometric functions. This can also have profit in the application to Feynman integrals, which we want to demonstrate by the examination of the singular locus also known as Landau variety. Moreover, $\Aa$-hypergeometric functions are connected to many branches in mathematics, e.g.\ combinatorics, number theory, motives, and Hodge theory (see exemplarily \cite{ReicheltAlgebraicAspectsHypergeometric2020, RobertsHypergeometricMotives2021}). Therefore, we will have the reasonable hope, that these various connections will also contain useful insights for Feynman integrals. \bigskip

In this section we will introduce $\Aa$-hypergeometric systems, and we will work out a representation of $\Aa$-hypergeometric functions by multivariate series. Moreover, we will draw their connection to $A$-discriminants. The following collection can only give a small glimpse of this rich theory. We refer to \cite{GelfandCollectedPapersVol1989, GelfandGeneralHypergeometricSystems1992, AomotoTheoryHypergeometricFunctions2011, SaitoGrobnerDeformationsHypergeometric2000, StienstraGKZHypergeometricStructures2005, CattaniThreeLecturesHypergeometric2006, VilenkinGelFandHypergeometric1995} for a more detailed description.


\subsection{Basic properties of \texorpdfstring{$\Aa$}{A}-hypergeometric systems} \label{ssec:BasicAhypergeometricSystems}
%
%
%
%

Let $\Aa=\{a^{(1)},\ldots,a^{(N)}\}\subset \mathbb Z^{n+1}$ be a finite subset of lattice points, which span $\mathbb R^{n+1}$ as a vector space $\operatorname{span}_{\mathbb R}(\Aa) = \mathbb R^{n+1}$. Thus, we always want to consider the case $n+1\leq N$. As before, we will denote by $\Aa\in\mathbb Z^{(n+1)\times N}$ also the matrix generated by the elements of the subset $\Aa$ as column vectors. Therefore, $\operatorname{span}_{\mathbb R}(\Aa) = \mathbb R^{n+1}$ is nothing else than the requirement that $\Aa$ has full rank. Further, we will assume that there exists a linear map $h:\mathbb Z^{n+1}\rightarrow \mathbb Z$, such that $h(a)=1$ for any $a\in\Aa$. Thus, all elements of $\Aa$ lie on a common hyperplane off the origin, which allows us to consider $\Aa$ as describing points of an affine space in homogenization. We write $A\subset\mathbb Z^n$ for the dehomogenized point configuration of $\Aa$ (see \cref{ssec:vectorConfigurations}). Equivalent to requiring such a linear map $h$, we can also demand $f=\sum_{a\in\Aa} z_a x^a$ to be a quasi-homogeneous polynomial. The integer kernel of $\Aa$ will be denoted by
\begin{align}
    \gls{Ll} := \ker_{\mathbb Z}(\Aa) = \left\{ (l_1,\ldots,l_N)\in\mathbb Z^N \, \big\rvert \, l_1 a^{(1)} + \ldots + l_N a^{(N)} = 0 \right\} = \Dep(\Aa)\cap\mathbb Z^N \point
\end{align}
By the notions of \cref{ssec:GaleDuality} this is the space of integer linear dependences and every Gale dual of $\Aa$ provides a basis of $\mathbb L$.\bigskip

The $\Aa$-hypergeometric system will be generated by a left ideal in the Weyl algebra $D := \langle z_1,\ldots,z_N,\allowbreak\partial_1,\ldots,\partial_N\rangle$, which consists of two types of differential operators, called toric and homogeneous operators
\begin{align}
    \gls{toricOp} &:= \prod_{l_j>0} \partial_j^{l_j} - \prod_{l_j<0} \partial_j^{-l_j} \qquad\text{for}\quad l\in\mathbb L \label{eq:toricOperators}\\
    \gls{homOp} &:= \sum_{j=1}^N a_i^{(j)} z_j \partial_j + \beta_i \qquad\text{for}\quad i=0,\ldots,n \label{eq:homogeneousOperators}
\end{align}
where $\gls{GKZpar}\in\mathbb C^{n+1}$ is an arbitrary complex number. We call the $D$-ideal
\begin{align}
    \gls{GKZIdeal} = \sum_{i=0}^{n} D\cdot E_i(\beta) + \sum_{l\in\mathbb L} D\cdot \square_l \label{eq:AhypIdeal}
\end{align}
the \textit{$\Aa$-hypergeometric ideal}. Note that this definition agrees with the definition in \cref{eq:AHypIdeal1}. We will call the latter part $\gls{toricId} = \sum_{l\in\mathbb L} D\cdot \square_l$ the toric ideal. This is an ideal in the (commutative) ring $\mathbb K[\partial_1,\ldots,\partial_N]$. The $D$-module of equivalence classes $\gls{GKZModule} = \normalslant{D}{H_\Aa(\beta)}$ is referred as the \textit{$\Aa$-hypergeometric system} or the \textit{$\Aa$-hypergeometric module}\footnote{Note, that there is a unique isomorphism between $\Aa$-hypergeometric systems and the $n$-th relative de Rham cohomology group $\mathbb H^n$, i.e.\ $\mathcal M_\Aa(\beta) \cong \mathbb H^n$, see \cite[prop. 2.3]{ChestnovMacaulayMatrixFeynman2022}.}. Holomorphic solutions on convenient domains $U\subseteq\mathbb C^N$ of these differential equation systems will be called \textit{$\Aa$-hypergeometric functions}. Thus, $\Aa$-hypergeometric functions are the elements of $\Sol (H_\Aa(\beta))$.\bigskip

One of the most important properties of those hypergeometric ideals \cref{eq:AhypIdeal} is to be always holonomic. Therefore, the dimension of $\Sol(H_\Aa(\beta))$ will be finite and the definition of $\Aa$-hypergeometric functions is meaningful. Furthermore, the holonomic rank can be characterized by the volume of the polytope $\Conv(A)$. In the following theorem, we will collect the most essential properties about $\Aa$-hypergeometric systems, which summarize the results of over more than $10$ years of research in that field.

\begin{theorem}[Holonomic rank of $\Aa$-hypergeometric systems \cite{GelfandHolonomicSystemsEquations1988, GelfandEquationsHypergeometricType1988,  GelfandGeneralizedEulerIntegrals1990, AdolphsonHypergeometricFunctionsRings1994, SaitoGrobnerDeformationsHypergeometric2000, MatusevichHomologicalMethodsHypergeometric2004}] \label{thm:HolRankAHyp}
    The $\Aa$-hypergeometric ideal $H_\Aa(\beta)$ is always a holonomic $D$-ideal. Its holonomic rank is bounded by the volume of a polytope
    \begin{align}
        \rank (H_\Aa(\beta)) \geq \vol (\Conv (A)) \point \label{eq:AHypHolRank}
    \end{align}
    For generic values of $\beta\in\mathbb C^{n+1}$ equality holds in \cref{eq:AHypHolRank}. Furthermore, equality in \cref{eq:AHypHolRank} holds for all values of $\beta\in\mathbb C^{n+1}$ if and only if $I_\Aa $ is Cohen-Macaulay.
\end{theorem}

We will justify the statement for (very) generic $\beta$ by an explicit construction in \cref{ssec:GammaSeries}.\bigskip

By the restriction to \textit{generic} values $\beta\in\mathbb C^{n+1}$ we mean, that $\rank (H_\Aa(\beta)) = \vol (\Conv (A))$ holds except for values $\beta$ of a proper Zariski closed subset of $\mathbb C^{n+1}$. Hence, the set of generic points of $\mathbb C^{n+1}$ is a non-empty open Zariski set, which is dense in $\mathbb C^{n+1}$. Occasionally, we consider the stronger restriction \textit{very generic}, which describes elements from a countable intersection of non-empty Zariski open sets. Typically, non-generic points $\beta$ concern integer values. We will specify this statement in \cref{ssec:GammaSeries}. \bigskip

\begin{example}[Gauss' hypergeometric function] \label{ex:GaussAhyp}
	We want to consider Gauss' hypergeometric function from \cref{sec:AHypSystems} in the context of $\Aa$-hypergeometric systems. For this function the $\Aa$-hypergeometric system will be generated by four points in $\mathbb Z^3$
	\begin{align}
		\Aa = \begin{pmatrix}
			1 & 1 & 1 & 1 \\
			0 & 1 & 0 & 1 \\
			0 & 0 & 1 & 1
		\end{pmatrix} , \qquad \beta = \begin{pmatrix} a + b - c \\ a \\ b \end{pmatrix} \point \label{eq:ABetaGauss}
	\end{align}
	The lattice of linear dependences of $\Aa$ is given by $\mathbb L = (1,-1,-1,1)\ \mathbb Z$. Hence, the $\Aa$-hypergeometric ideal $H_\Aa(\beta)$ will be spanned by the differential operators
	\begin{align}
		\partial_1 \partial_4 - \partial_2\partial_3 = 0, & \qquad z_1\partial_1 + z_2\partial_2 + z_3\partial_3 + z_4\partial_4 + (a+b-c) = 0 \\
		z_2\partial_2 + z_4\partial_4 + a = 0, & \qquad z_3\partial_3 + z_4\partial_4 + b = 0 \point		
	\end{align}
	We can combine these four differential operators to a single one
	\begin{align}
		z_4 \left(\frac{z_4}{z_2z_3} - \frac{1}{z_1} \right) \partial_4^2 + \left( \frac{1+a+b}{z_2z_3} z_4 - \frac{c}{z_1}\right) \partial_4 + \frac{ab}{z_2z_3} = 0 \point \label{eq:GaussAhypDiff}
	\end{align}
	Thus, for $z_1 = z_2 = z_3 = 1$ we obtain the differential equation for Gauss' hypergeometric function \cref{eq:GaussODE}. Hence, the Gauss' hypergeometric function ${_2}F_1(a,b,c|z)$ is an $\Aa$-hypergeometric function $\Phi(1,1,1,z)$, with $\Phi\in\Sol(H_\Aa(\beta))$ where $\Aa$ and $\beta$ are given by \cref{eq:ABetaGauss}. The hypergeometric differential equation has two basic solutions, which will agree to the $\Aa$-hypergeometric description, since $\vol\!\left(\Conv(A)\right)=2$.
\end{example}

At first sight, it seems that we have introduced too many variables in $\Aa$-hyper\-geometric systems as it appears in \cref{ex:GaussAhyp}. However, closer inspection of this example shows, that there is actual only one effective variable $y = \frac{z_1z_4}{z_2 z_3}$. It is the beautiful observation of Gelfand, Kapranov and Zelevinsky that the introduction of these extra variables will greatly simplify the hypergeometric structures \cite{StienstraGKZHypergeometricStructures2005}. Note that this behaviour is the same as observed for the $\beta$-projection of secondary fans (see \cref{ssec:secondaryPolytope}) and the reduced $\Aa$-discriminants of \cref{ssec:Adiscriminants}.\bigskip

A simple but nevertheless useful property of $\Aa$-hypergeometric functions is that they obey certain shift relations.
\begin{lemma}
	Let $\Phi(z)\in\Sol\!\left(H_\Aa(\beta)\right)$ an $\Aa$-hypergeometric function. Then its derivative $\pd{\Phi(z)}{z_j} \in \Sol\!\left(H_\Aa\!\left(\beta + a^{(j)}\right)\right)$ is an $\Aa$-hypergeometric function, where the parameter $\beta$ is shifted by $a^{(j)}$.
\end{lemma}
\begin{proof}
	The generator $\partial_j$ commutes with $\square_l$ of equation \cref{eq:toricOperators}. By the use of the commutator relations of the Weyl algebra we can verify easily $\partial_j E_i(\beta) = E_i\!\left(\beta + a_i^{(j)}\right) \partial_j$, which shows the assertion.
\end{proof}

Another simple but useful consequence of the existence of a linear map $h : \mathbb Z^{n+1}\rightarrow \mathbb Z$ with $h(a)=1$ for all $a\in\Aa$ is the following observation.
\begin{lemma} \label{lem:HomogenityRelationsForAa} 
    Let $\Aa = \{a^{(1)},\ldots,a^{(N)} \} \subset\mathbb Z^{n+1}$ be a vector configuration satisfying the conditions for an $\Aa$-hypergeometric system (i.e.\ $\operatorname{span}_{\mathbb R}(\Aa) = \mathbb R^{n+1}$ and all elements of $\Aa$ lying on a common hyperplane off the origin) and let $\sigma \subseteq \{1,\ldots,N\}$ be an index set of cardinality $n+1$, such that $\Aas := \{ a^{(j)} \}_{j\in\sigma}$ is also full dimensional. In other words, $\sigma$ describes a full dimensional simplex from points of $\Aa$. Then we have
    \begin{align}
        \sum_{i\in\sigma} \left(\Aas^{-1}\Aabs\right)_{ij} = 1 \label{eq:AA1a}
    \end{align}
    for all $j=1,\ldots,N$ where $\bar\sigma$ is the complement of $\sigma$. If additionally $\Aa$ is a homogenized point configuration, i.e.\ $\Aa$ is of the form $\Aa = \begin{psmallmatrix} 1 & \ldots & 1 \\ & A & \end{psmallmatrix}$, it follows
    \begin{align}
        \sum_{i\in\sigma} \left(\Aas^{-1}\nuu\right)_i = \nu_0 \label{eq:AA1b}
    \end{align}
    for any complex vector $\nuu = (\nu_0,\nu_1,\ldots,\nu_n)\in\mathbb C^{n+1}$.
\end{lemma}
\begin{proof}
    $\Bb=\begin{psmallmatrix} -\Aas^{-1}\Aabs \\ \mathbbm{1}\end{psmallmatrix}\in\mathbb{Q}^{N\times r}$ is a possible Gale dual of $\Aa$, which can be verified by the direct computation $\Aa \Bb = 0$. From the existence of a linear form $h$ with $h(a) = 1$ for all $a\in\Aa$ it follows in particularly that $\sum_j \Bb_{jk}=0$ and thereby equation \cref{eq:AA1a}. The second statement \cref{eq:AA1b} follows trivially from $\sum_{j,k} (\Aas)_{ij} (\Aas^{-1})_{jk} \nuu_k=\nuu_i$.
\end{proof}

As seen in \cref{thm:HolRankAHyp}, $\Aa$-hypergeometric systems are in a certain sense well-behaved, and we can construct a basis of its solution space $\Sol(H_\Aa(\beta))$. There are several ways to construct those bases. In the following section we will give a basis in terms of multivariate power series, which was the first solution of $\Aa$-hypergeometric systems invented in \cite{GelfandHolonomicSystemsEquations1988}. Thereby, the observations in \cref{lem:HomogenityRelationsForAa} will help us to simplify the convergence conditions of those series. There are many alternatives to construct a basis of the solution space of $\Aa$-hypergeometric systems. We refer to \cite{Matsubara-HeoLaplaceResidueEuler2018} for a collection of various types of bases and to \cite{SaitoGrobnerDeformationsHypergeometric2000} for a systematic treatment of series solutions.


\vspace{.3\baselineskip}
\subsection{\texorpdfstring{$\Gamma$}{\unichar{"0393}}-series}\label{ssec:GammaSeries}
\vspace{.3\baselineskip}

%

Let $\Aa\in\mathbb Z^{(n+1)\times N}$ be a full rank integer matrix with $n+1\leq N$ and $\mathbb {L} = \operatorname{ker}_\mathbb{Z} (\Aa)$ its kernel with $\rank (\mathbb L) = N-n-1 =: r$ as before. Then for every $\gamma\in\mathbb{C}^N$ the formal series invented in \cite{GelfandHolonomicSystemsEquations1988}\glsadd{gammaSer}\glsunset{gammaSer}
\begin{align}
    \varphi_\gamma(z) = \sum_{l\in\mathbb L} \frac{z^{l+\gamma}}{\Gamma(\gamma+l+1)} \label{eq:GammaSeriesOriginalDef}
\end{align}
is called a \textit{$\mathit\Gamma$-series}. We recall the use of a multi-index notation, i.e.\ we write $z^{l+\gamma} = \prod_{i=1}^N z_i^{l_i+\gamma_i}$ and $\Gamma(\gamma+l+1) = \prod_{i=1}^N \Gamma(\gamma_i+l_i+1)$. For an appropriate choice of $\gamma$, those $\Gamma$-series are formal solutions of the $\Aa$-hypergeometric systems \cref{eq:AhypIdeal}.

\begin{lemma}[$\Gamma$-series as formal solutions of GKZ hypergeometric systems \cite{GelfandHolonomicSystemsEquations1988,GelfandHypergeometricFunctionsToral1989}]
    Let $\gamma\in\mathbb C^N$ be a complex vector satisfying $\Aa \gamma + \beta = 0$. Then the $\Gamma$-series $\varphi_\gamma(z)$ is a formal solution of the $\Aa$-hypergeometric system
    \begin{align*}
    H_\Aa (\beta) \bullet \varphi_\gamma(z)=0 \point
    \end{align*}
\end{lemma}
\begin{proof}
    For any $u\in\mathbb N_0^N$ and any $r\in\mathbb C^N$ it is $\left(\frac{\partial}{\partial z}\right)^u z^r = \frac{\Gamma(r+1)}{\Gamma(r-u+1)} z^{r-u}$ (with an appropriate limit, respectively). Furthermore, one can add an element of $\mathbb L$ to $\gamma$, without changing the $\Gamma$-series. Every element of $\mathbb L$ can be written as $u-v$, where $u,v\in\mathbb N_0^N$ are two vectors satisfying $\Aa u = \Aa v$. Therefore, we have
    \begin{align}
        \partial^u \bullet \varphi_{\gamma}(z) &= \sum_{l\in\mathbb L} \frac{z^{l+\gamma-u}}{\Gamma(\gamma+l-u +1)} = \varphi_{\gamma-u}(z) = \varphi_{\gamma-v}(z) = \partial^v \bullet \varphi_{\gamma}(z)
    \end{align}
    which shows that the $\Gamma$-series fulfil the toric operators \cref{eq:toricOperators}. For the homogeneous operators \cref{eq:homogeneousOperators} one considers
    \begin{align}
    \sum_{j=1}^N a^{(j)} z_j \partial_j \bullet \varphi_\gamma(z) = \sum_{l\in\mathbb L} \left(\sum_{j=1}^N a^{(j)} (\gamma_j+l_j)\right) \frac{z^{l+\gamma}}{\Gamma(\gamma+l+1)} = - \beta \varphi_\gamma(z) \point
    \end{align}
\end{proof}

The restriction $\Aa \gamma + \beta=0$ allows in general many choices of $\gamma$. Let $\sigma \subseteq \{1,\ldots,N\}$ be an index set with cardinality $n+1$, such that the matrix $\Aa_\sigma := \{a^{(j)}\}_{j\in \sigma}$ restricted to columns corresponding to $\sigma$ is non-singular $\det (\Aas) \neq 0$. Denote by $\bar\sigma = \{1,\ldots,N\} \setminus \sigma$ the complement of $\sigma$ and $\Aabs:=\{a^{(j)}\}_{j\in \bar\sigma}$. If one sets $\gamma_\sigma = -\Aas^{-1} (\beta + \Aabs k)$ and $\gamma_{\bar\sigma} = k$ the condition $\Aa \gamma + \beta = 0$ is satisfied for any $k\in\mathbb{C}^{r}$.

On the other hand we can split the lattice $\mathbb L = \left\{ l\in\mathbb Z^N \,\rvert\, \Aa l =0\right\}$ in the same way $\Aas l_\sigma + \Aabs l_{\bar\sigma}=0 $ and obtain a series only over $l_{\bar\sigma}$
\begin{align}
    \varphi_{\gamma_\sigma} (z) &= \sum_{\substack{l_{\bar\sigma}\in\mathbb Z^r \\ \textrm{s.t. } \Aas^{-1}\Aabs l_{\bar\sigma}\in\mathbb Z^{n+1} }} \frac{z_\sigma^{-\Aas^{-1} (\beta+\Aabs k + \Aabs l_{\bar\sigma}) } z_{\bar\sigma}^{k+l_{\bar\sigma}}}{\Gamma\!\left(-\Aas^{-1}(\beta+\Aabs k +\Aabs l_{\bar\sigma})+1\right) \Gamma(k+l_{\bar\sigma}+1)} \point
\end{align}
In order to simplify the series we will restrict $k\in\mathbb Z^r$ to integers, since terms with $(k+l_{\bar\sigma})_i\in\mathbb{Z}_{< 0}$ will vanish. The $\Gamma$-series depends now on $k$ and $\sigma$\glsadd{gammaSer}\glsunset{gammaSer}
\begin{align}
    \varphi_{\sigma,k}(z) =z_\sigma^{-\Aas^{-1}\beta} \sum_{\lambda\in\Lambda_k} \frac{z_\sigma^{-\Aas^{-1}\Aabs\lambda} z_{\bar\sigma}^\lambda}{\lambda!\ \Gamma\!\left(-\Aas^{-1}(\beta+\Aabs \lambda)+1\right)} \label{eq:GammaPowerSeriesVarphi}
\end{align}
where $\gls{Lambdak}=\left\{ k + l_{\bar\sigma} \in \mathbb{N}_0^r \,\rvert\, \Aabs l_{\bar\sigma} \in \mathbb Z \Aas \right\}\subseteq \mathbb{N}_0^r$. Therefore, the $\Gamma$-series is turned into a power series in the variables\glsadd{effectiveVars}
\begin{align}
	y_j = \frac{(z_{\bar\sigma})_j}{\prod_i (z_\sigma)_i^{(\Aas^{-1} \Aabs)_{ij}}} \quad\text{for } j=1,\ldots,r \text{ .}\label{eq:effectiveVarsGammaSeries}
\end{align}
Note that these are the same ``effective'' variables, which we introduced in \cref{eq:effectiveVarsADisc}.

However, for certain values of $\beta\in\mathbb C^{n+1}$ some of the series \cref{eq:GammaPowerSeriesVarphi} may vanish. A $\Gamma$-series $\varphi_{\sigma,k}$ vanishes if and only if for every $\lambda\in\Lambda_k$ the vector $\Aas^{-1}(\beta+\Aabs\lambda)$ has at least one (strictly) positive integer component. Hence, a $\Gamma$-series will vanish if and only if $\beta$ takes values on one of certain countable hyperplanes. Therefore, we will demand $\beta$ to be very generic in order to avoid those cases. Note, that the set of very generic $\beta$ is dense in $\mathbb C^{n+1}$.\bigskip

At next, we will consider the dependence of $\varphi_{\sigma,k}$ on the choice of $k\in\mathbb Z^r$. Note first, that $\Lambda_k = \Lambda_{k+k'}$ if and only if $\Aas^{-1}\Aabs k' \in \mathbb Z^{n+1}$. Furthermore, there is always such a positive integer vector $k'\in\mathbb N^r$ satisfying $\Aas^{-1}\Aabs k' \in \mathbb Z^{n+1}$, e.g.\ by choosing $k'\in |\det(\Aas)|\cdot \mathbb N^r$. Thus, we can shift every $k\in\mathbb Z^r$ to positive integers, which is why we want to restrict our consideration to $k\in\mathbb N_0^r$. Moreover, the cardinality of the set $\left\{\Lambda_k \,\rvert\, k \in \mathbb Z^r\right\} = \left\{\Lambda_k \,\rvert\, k \in \mathbb N_0^r\right\}$ is given by $|\det(\Aas)|$, which is nothing else than the volume of $\Conv(\Aas)$ \cite{Fernandez-FernandezIrregularHypergeometricDmodules2009}. In order to generate a basis of the solution space $\Sol\!\left(H_\Aa(\beta)\right)$, we want to pick those elements from $\left\{\Lambda_k \,\rvert\, k \in \mathbb N_0^r\right\}$ such that the resulting $\Gamma$-series $\varphi_{\sigma,k}$ are linearly independent. Thus, let $k^{(1)},\ldots,k^{(s)}$ be representatives of $\bigslant{\mathbb Z^{n+1}}{\mathbb Z\Aas}$, i.e.\
\begin{align}
	\bigslant{\mathbb Z^{n+1}}{\mathbb Z\Aas} = \left\{ \left[\Aabs k^{(j)}\right] \, \rvert \, j = 1, \ldots, s = \vol\!\left(\Conv (\Aas)\right) \right\} \point \label{eq:representativesK}
\end{align}
Hence, $\left\{\Lambda_{k^{(j)}} \,\rvert\, j=1,\ldots,s\right\}$ defines a partition of $\mathbb N_0^r$ \cite{Fernandez-FernandezIrregularHypergeometricDmodules2009}:
\begin{align}
	\Lambda_{k^{(i)}} \cap \Lambda_{k^{(j)}} = \emptyset\quad \text{for all} \quad i\neq j \quad\text{and} \qquad \bigcup_{j=1}^s \Lambda_{k^{(j)}} = \mathbb N^r_0 \point \label{eq:LambdaPartition}
\end{align}
Therefore, the $\Gamma$-series $\varphi_{\sigma,k^{(1)}},\ldots,\varphi_{\sigma,k^{(s)}}$ have different support and are linear independent if none of them is identically zero, which we have avoided by assuming very generic $\beta$.

\begin{example}
	We want to illustrate the construction of partitions of $\mathbb N_0^r$ by means of the following example. Let
	\begin{align}
		\Aa = \begin{pmatrix}
		      	1 & 1 & 1 & 1 \\
		      	1 & 0 & 2 & 0 \\
		      	0 & 1 & 0 & 2 
		      \end{pmatrix}
	\end{align}
	be a vector configuration and $\sigma = \{1,3,4\}$ an index set, where the corresponding full dimensional simplex has volume $\vol\!\left(\Conv(\Aas)\right) = |\det \Aas | = 2$. As we will later see, this is a specific configuration for a fully massive bubble self-energy graph. Considering the inverse of $\Aas$, we find that the ideal $\mathbb Z \Aas = (\mathbb Z,\mathbb Z,2\mathbb Z)^\top$ contains only even integers in its last component. Hence, the quotient ring $\bigslant{\mathbb Z^3}{\mathbb Z \Aas} = \{ (\mathbb Z,\mathbb Z,2\mathbb Z)^\top, (\mathbb Z,\mathbb Z,2\mathbb Z + 1)^\top \}$ consists in the two equivalence classes with even and odd integers in the last component. Thus, we can choose the representatives $[0,0,0]^\top, [0,0,1]^\top$ or equivalently  $[0,0,0]^\top, [1,0,1]^\top$. Due to \cref{eq:representativesK}, these representatives will be generated by $k^{(1)} = 0$ and $k^{(2)} = 1$. The two resulting summation regions $\Lambda_{k^{(1)}} = 2\mathbb N_0$ and $\Lambda_{k^{(2)}} = 2\mathbb N_0+1$ consist in even and odd natural numbers, respectively.
\end{example}

$\Gamma$-series were hitherto only treated as formal series. Therefore, we want to fill this gap now and examine the convergence of $\Gamma$-series. Thus, we will show in the following that there is a non-vanishing convergence radius $R\in\mathbb R^r_{>0}$, such that \cref{eq:GammaPowerSeriesVarphi} converges absolutely for $|y_j|< R_j$ with $j=1,\ldots,r$. Hence, we have to estimate the summands of $\Gamma$-series. As an application of the Stirling formula, one can state the following lemma.

\begin{lemma}[Bounds of $\Gamma$-functions (similar to \cite{StienstraGKZHypergeometricStructures2005})] \label{lem:gammaest}
    For every $C\in\mathbb C$ there are constants $\kappa,R\in\mathbb R_{>0}$ independent of $M$, such that
    \begin{align}
    \frac{1}{|\Gamma (C+M)|} \leq \kappa R^{|M|} |M|^{-M} \label{eq:gammest}
    \end{align}
    for all $M\in\mathbb Z$.
\end{lemma}
\begin{proof}
    Firstly, consider the non-integer case $C\notin \mathbb{Z}$. For $M>0$ it is
    \begin{align}
        |\Gamma(C+M)| &= |\Gamma(C)| \prod_{j=0}^{M-1} |C+j| \geq |\Gamma(C)| \prod_{j=0}^{M-1} \big| |C|- j\big| \nonumber\\
        &= |\Gamma(C)| \prod_{j=1}^M j \left|\frac{|C|-j+1}{j}\right| \geq M! \ Q^M |\Gamma(C)|
    \end{align}
    where $Q=\min \left|\frac{|C|-j+1}{j}\right| > 0$ and we used a variation of the triangle inequality $|a+b| \geq  \big| |a| - |b|\big|$ (also known as reverse triangle inequality). With Stirling's approximation one obtains further
    \begin{align}
        |\Gamma(C+M)| &\geq |\Gamma(C)| \sqrt{2\pi} Q^M M^{M+\frac{1}{2}} e^{-M} \geq |\Gamma(C)| \sqrt{2\pi} \left(\frac{Q}{e}\right)^M M^M \point
    \end{align}
    In contrast, for $M<0$ using the triangle inequality $|a-b| \leq |a| + |b|$ we have
    \begin{align}
        |\Gamma(C+M)| &= |\Gamma(C)| \prod_{j=1}^{|M|} \left|\frac{1}{C-j}\right| \geq |\Gamma(C)| \prod_{j=1}^{|M|} \left|\frac{1}{|C|+j}\right| \geq |\Gamma(C)| \prod_{j=1}^{|M|} \left|\frac{1}{|C|+|M|}\right| \nonumber\\
        &= \left(1+\frac{|C|}{|M|}\right)^{-|M|} |M|^{M} |\Gamma(C)| \geq |\Gamma(C)| \left(1+|C|\right)^{-|M|} |M|^{M} \point
    \end{align}
    By the setting $\kappa = |\Gamma(C)|^{-1}$ and $R=\max \left( 1+|C|,e Q^{-1}\right)$ one can combine both cases to equation \cref{eq:gammest}. The case $M=0$ is trivially satisfied whereat we set $0^0:=1$.
    
    Consider now the case where $C\in\mathbb Z$. If $C+M\leq 0$ the $\Gamma$-function has a pole and the lemma is satisfied automatically. For $C+M\geq 1$ it is $\Gamma(C+M) = \prod_{j=0}^{C+M-2} (j+1) \geq \prod_{j=0}^{C+M-2} \left(j+\frac{1}{2}\right) = \frac{\Gamma(C+M-\frac{1}{2})}{\Gamma(\frac{1}{2})} = \frac{1}{\sqrt \pi} \Gamma(C+M-\frac{1}{2})$ which recurs to the non-integer case with $C^\prime=C-\frac{1}{2}\notin \mathbb Z$.
\end{proof}

To apply this estimation to a product of $\Gamma$ functions the following lemma is helpful.

\begin{lemma} \label{lem:apowera}
    Let $a_1,\ldots,a_N\in\mathbb{R}_{>0}$ be a set of positive real numbers. Then it holds
    \begin{align}
        \left(\sum_{i=1}^N a_i\right)^{\sum_{i=1}^N a_i} \geq \prod_{i=1}^N a_i^{a_i} \geq \left(\frac{1}{N} \sum_{i=1}^N a_i\right)^{\sum_{i=1}^N a_i} \point
    \end{align}
\end{lemma}
\begin{proof}
    The left inequality is trivially true. The right inequality will be proven by considering different cases of $N$. For $N=2$ without loss of generality it is $\frac{a_1}{a_2} =: \rho \geq 1$. The latter is then equivalent to the Bernoulli inequality
    \begin{align}
        a_1^{a_1} a_2^{a_2} \geq \left(\frac{a_1+a_2}{2}\right)^{a_1+a_2} \Leftrightarrow  \left(1+\frac{\rho-1}{\rho+1}\right)^{\rho+1} \geq \rho \point
    \end{align}
    For $N=2^j$ with $j\in\mathbb N$ the lemma can be reduced to the case $N=2$ by an iterative use. All the other cases can be reduced to a $2^j$-case by adding the mean value $\mu := \frac{1}{N} \sum_{i=1}^N a_i$
    \begin{align}
        \left(\mu^\mu\right)^{2^j-N}  \prod_{i=1}^N a_i^{a_i} &\geq \left(\frac{\sum_{i=1}^N a_i + (2^j-N)\mu}{2^j}\right)^{\sum_{i=1}^N a_i + (2^j-N)\mu} \nonumber \\
        &= \left(\frac{\sum_{i=1}^N a_i}{N}\right)^{\sum_{i=1}^N a_i} \mu^{(2^j-N)\mu}
    \end{align}
    with $2^j-N>0$.
\end{proof}

Combining the estimations in \cref{lem:gammaest} and \cref{lem:apowera}, we can show the convergence of $\Gamma$-series by a convergent majorant.

\begin{theorem}[Convergence of $\Gamma$-series \cite{GelfandHolonomicSystemsEquations1988, GelfandHypergeometricFunctionsToral1989,StienstraGKZHypergeometricStructures2005}] \label{thm:GammaConverge}
    There is always a positive real tuple $R\in\mathbb R_{>0}^r$, such that the series $\varphi_{\sigma,k}$ \cref{eq:GammaPowerSeriesVarphi} converges absolutely for any $y\in\mathbb C^r$ with $|y_j| < R_j$ for $j=1,\ldots,r$.
\end{theorem}
\begin{proof}
    The series in equation \cref{eq:GammaPowerSeriesVarphi} can be written in the form
    \begin{align}
        \sum_{\lambda\in\Lambda_k} \frac{y^\lambda}{\Gamma(\mathbf 1_N + C + \Bb\lambda)}
    \end{align}
    where $\Bb =\begin{psmallmatrix} -\Aas^{-1}\Aabs \\ \mathbbm{1}_r \end{psmallmatrix}\in\mathbb Q^{N\times r}$ is a Gale dual of $\Aa$, $y$ are the effective variables from \cref{eq:effectiveVarsGammaSeries} and $C = (-\Aas^{-1}\beta,\mathbf 0_r)$. From the definition of $\Lambda_k$ we see that $\Bb \lambda\in\mathbb Z^N$ for all $\lambda\in\Lambda_k$. Thus, one can estimate by \cref{lem:gammaest}
    \begin{align}
        \left| \prod_{i=1}^N \frac{1}{\Gamma(1 + C_i + \sum_{j=1}^r b_{ij} \lambda_j)}\right| \leq \kappa \prod_{i=1}^N R_i^{\left|\sum_{j=1}^r b_{ij}\lambda_j\right|} \left|\sum_{j=1}^r b_{ij}\lambda_j\right|^{-\sum_{j=1}^r b_{ij} \lambda_j} \comma
    \end{align}
    where $b_{ij}$ are the components of $\Bb$. Furthermore, by \cref{lem:HomogenityRelationsForAa} it is also $\sum_{i=1}^N b_{ij} = 0$, which will imply $D:=\sum_{b_{ij}>0} b_{ij}\lambda_j = - \sum_{b_{ij}<0} b_{ij}\lambda_j = \frac{1}{2} \left|\sum_{ij} b_{ij}\lambda_j\right|$. With \cref{lem:apowera} we can continue to estimate
    \begin{align}
        \prod_{i=1}^N \left|\sum_{j=1}^r b_{ij}\lambda_j\right|^{-\sum_{j=1}^r b_{ij} \lambda_j} \leq D^D \ N_+^D D^{-D} = N_+^D
    \end{align}
    by splitting the product in $N_+$ factors where $\sum_{j=1}^r b_{ij}\lambda_j>0$ is positive and the $N_-$ factors where $\sum_{j=1}^r b_{ij}\lambda_j<0$ is negative. With $R_{\textrm{max}}=\max_i (R_i)$ we obtain
    \begin{align}
        \left| \prod_{i=1}^N \frac{1}{\Gamma(a_i+\sum_{j=1}^r C_{ij} \lambda_j)}\right| \leq \kappa N_+^D R_{\textrm{max}}^{2D} \point
    \end{align}
    Thus, the $\Gamma$-series will be bounded by a geometric series and there is always a non-vanishing region of absolute convergence.
\end{proof}

Applying this result to the variables $z_1,\ldots,z_N$, a $\Gamma$-series $\varphi_{\sigma,k}$ converges absolutely if those variables satisfy
\begin{align}
	\ln |y_j| = \sum_{i=1}^N b_{ij} \ln |z_i| < \rho_j \label{eq:convergenceConditionLog}
\end{align}
for $j=1,\ldots,r$ where $\rho_j\in\mathbb R$ are real numbers and $b_{ij}$ denote the elements of the Gale dual $\Bb = \Bb(\sigma) =\begin{psmallmatrix} -\Aas^{-1}\Aabs \\ \mathbbm{1}_r \end{psmallmatrix}$. Let
\begin{align}
	C(\sigma) = \left\{ \omega\in\mathbb R^N \,\rvert\, \omega\Bb(\sigma) \geq 0 \right\}
\end{align}
be the cone generated by the Gale dual $\Bb(\sigma)$. As assumed, $\Aa$ is an acyclic vector configuration and therefore $\Bb$ is totally cyclic (see \cref{ssec:GaleDuality}). Thus, there is a vector $p\in\mathbb R^N$ such that $- p \Bb(\sigma) = \rho$, and we can reformulate the convergence condition \cref{eq:convergenceConditionLog} into points of a translated cone 
\begin{align}
    \left(-\ln|z_1|,\ldots,-\ln|z_N|\right) \in \relint\!\left(C(\sigma)\right) + p \point \label{eq:logMapTranslatedCone}
\end{align}

Comparing the cone $C(\sigma)$ with the results from \cref{eq:subdivisionCondGale} and \cref{eq:subdivisionCondGaleBeta} we obtain $\relint\!\left(C(\sigma)\right) = \{\omega\in\mathbb R^N \,\rvert\, \sigma\in \Ss(\Aa,\omega)\}$. Therefore, for any regular triangulation $\Tt = \Ss(\Aa,\omega)$ of $\Aa$ generated by a height $\omega\in\mathbb R^N$ we consider the intersection $C(\Tt) := \cap_{\sigma\in\Tt} C(\sigma)$ which is nothing else than a secondary cone. Hence, $C(\Tt)$ is full dimensional \cite{DeLoeraTriangulations2010} and thus also the intersection of translated cones is full dimensional. Therefore, there is a common convergence region of all $\{\varphi_{\sigma,k}\}_{\sigma\in\Tt}$. Note that regularity of triangulations is necessary, as otherwise $C(\Tt)$ will not be full dimensional \cite{DeLoeraTriangulations2010}. Furthermore, all those series will be linearly independent. Recall from \cref{thm:HolRankAHyp} that the holonomic rank for generic $\beta$ is given by the volume of $\Conv(A)$. Hence, by collecting all $\Gamma$-series $\varphi_{\sigma,k}$ corresponding to maximal cells $\sigma$ of a triangulation $\Tt$ and varying $k$ according to \cref{eq:representativesK} we will obtain exactly $\vol(\Conv(A))$ independent series $\varphi_{\sigma,k}$. We will combine all these results in the following theorem.

\begin{theorem}[Solution space with $\Gamma$-series \cite{GelfandHolonomicSystemsEquations1988, GelfandHypergeometricFunctionsToral1989, GelfandGeneralHypergeometricSystems1992}] \label{thm:SolutionSpaceGammaSeries}
    Let $\Tt$ be a regular triangulation of the vector configuration $\Aa\subset\mathbb Z^{n+1}$ and $\beta\in\mathbb C^{n+1}$ very generic with respect to every $\sigma\in \hatT$, where $\hatT$ are the maximal cells of $\Tt$. Further, let $\gls{Ksigma} = \left\{ k^{(1)},\ldots, k^{(\vol(\Conv(\Aas)))} \right\}\subset \mathbb N_0^r$ be a set of representatives of $\bigslant{\mathbb Z^{n+1}}{\mathbb Z\Aas}$ for any $\sigma\in\hatT$ according to \cref{eq:representativesK}. Then the set of power series
    \begin{align}
    	\left\{ \left\{ \varphi_{\sigma,k} \right\}_{k\in K_\sigma} \right\}_{\sigma\in\hatT}
    \end{align}
    is a basis of the solution space $\Sol\!\left(H_\Aa(\beta)\right)$ and all those power series have a common, non-empty domain of absolute convergence.
\end{theorem}
Above we only showed that a non-empty, common domain of convergence exists. The investigation of the exact shape of this domain therefore starts at \cref{eq:logMapTranslatedCone} and led to the introduction of the so-called \textit{amoeba}. For further reading about amoebas we refer \cite[ch. 6.1]{GelfandDiscriminantsResultantsMultidimensional1994}, \cite{PassareSingularitiesHypergeometricFunctions2005, ForsbergLaurentDeterminantsArrangements2000}. See also \cref{fig:ComplexLogarithmMaps}.\bigskip

For very generic $\beta\in\mathbb C^N$ we can normalize the $\Gamma$-series\glsadd{norGammaSer}\glsunset{norGammaSer}
\begin{align}
    \phi_{\sigma,k} (z) := \Gamma(-\Aas^{-1}\beta+1)\ \varphi_{\sigma,k}(z) = z_\sigma^{-\Aas^{-1}\beta} \sum_{\lambda\in\Lambda_k} \frac{z_\sigma^{-\Aas^{-1}\Aabs\lambda} z_{\bar\sigma}^\lambda}{\lambda!\ (1 -\Aas^{-1}\beta)_{-\Aas^{-1}\Aabs \lambda}}
\end{align}
such that the first term in the series is equal to $1$. Note that this definition of $\Gamma$-series agrees with the definition in \cite{SaitoGrobnerDeformationsHypergeometric2000}.

Especially, for unimodular triangulations ($|\det \Aas| = 1$) we can simplify the $\Gamma$-series further. Note that in this case $\Lambda_k = \mathbb N_0^r$ for any $k\in\mathbb Z^r$ since $\Aas^{-1}\in\mathbb Z^{(n+1)\times (n+1)}$, and we obtain\glsadd{norGammaSer}\glsunset{norGammaSer}
\begin{align}
    \phi_{\sigma} (z) = z_\sigma^{-\Aas^{-1}\beta} \sum_{\lambda\in\mathbb N_0^r} \frac{(\Aas^{-1}\beta)_{\Aas^{-1}\Aabs\lambda}}{\lambda!} \frac{z_{\bar\sigma}^\lambda}{(-z_\sigma)^{\Aas^{-1}\Aabs\lambda} } 
\end{align}
by the properties of Pochhammer symbols. \bigskip

A slight variation of $\Gamma$-series are the \textit{Fourier $\Gamma$-series} \cite{StienstraGKZHypergeometricStructures2005} where we replace the variables $z_j \mapsto e^{2\pi i w_j}$ for $j=1,\ldots,N$. Those Fourier $\Gamma$-series are more flexible than the original definition \cref{eq:GammaSeriesOriginalDef}. This replacement will simplify the convergence criterion \cref{eq:convergenceConditionLog} and also considerations about the monodromy of $\Aa$-hypergeometric functions are more accessible. In this context we refer also to the coamoeba (see \cref{sec:Coamoebas}), which adopts the spirit of this idea.


\subsection{Singular locus of \texorpdfstring{$\Aa$}{A}-hypergeometric systems} \label{ssec:SingularLocusAHyp}
%

In the previous section we constructed a basis of the solution space in terms of power series. However, we restricted the domain of these functions in order to ensure the convergence of these series. In this section we will ask for the analytic continuation of solutions of $H_\Aa(\beta)$, i.e.\ we will look for a maximal domain. By the Cauchy-Kashiwara-Kovalevskaya \cref{thm:CKKT} we have to consider the singular locus of $H_\Aa(\beta)$ for the analytic continuation. Let us first remark, that we can restrict ourselves to the codimension $1$ part of the singular locus, since all singularities in higher codimensions are removable singularities due to Riemann's second removable theorem \cite{KaupHolomorphicFunctionsSeveral1983, BerkeschZamaereSingularitiesHolonomicityBinomial2014}.

This section will also establish the link between $\Aa$-hypergeometric systems and $A$-discriminants. This connection was developed in a series of articles by Gelfand, Kapranov and Zelevinsky, mainly in \cite{GelfandEquationsHypergeometricType1988, GelfandAdiscriminantsCayleyKoszulComplexes1990, GelfandHypergeometricFunctionsToric1991, GelfandDiscriminantsResultantsMultidimensional1994}. A major part of this correspondence was also shown in \cite{AdolphsonHypergeometricFunctionsRings1994}, where the following deduction is mostly based on. A generalization of this relation can be found in \cite{BerkeschZamaereSingularitiesHolonomicityBinomial2014} and \cite{SchulzeIrregularityHypergeometricSystems2006}.

Recall, that we will always assume for $\Aa$-hypergeometric systems, that $\Aa$ describes points lying on a common hyperplane off the origin (see \cref{ssec:BasicAhypergeometricSystems}). This will imply a certain regularity of $H_\Aa(\beta)$, i.e.\ local solutions have at worst logarithmic singularities near the singular locus \cite{CattaniThreeLecturesHypergeometric2006}. This is a generalization of the behaviour of regular singular points in ordinary differential equations (see also \cref{sec:holonomicDmodules}). \bigskip
    
As a first step we want to establish a connection between the faces of $\Conv(\Aa)$ and the characteristic variety of $H_\Aa(\beta)$. The following two lemmata are inspired by \cite{AdolphsonHypergeometricFunctionsRings1994} with slight adjustments.

\begin{lemma} \label{lem:face-characteristic}
    For every point $(\hat z,\hat \xi)\in\ch\!\left(H_\Aa(\beta)\right)$ of the characteristic variety, there exists a unique face $\tau\subseteq\Conv(\Aa)$ such that $\hat \xi_j\neq 0$ if and only if $j\in\tau$.
\end{lemma}
\begin{proof}
    The case $\hat\xi = (0,\ldots,0)$ is trivially satisfied by $\tau=\emptyset$, and we will exclude this case in the following. Denote by $\emptyset\neq J\subseteq\{1,\ldots,N\}$ the index set for which $\hat\xi_j \neq 0$ for all $j\in J$ and let $\tau$ be the carrier of $J$, i.e.\ the smallest face of $\Conv(\Aa)$ containing all points with labels in $J$. Thus, we want to show $J=\tau$. Let us first show that $J$ spans affinely the supporting hyperplane of $\tau$, i.e.\ that $\Conv(J)$ and $\tau$ have the same dimension\footnote{One can see easily, that $\tau\subseteq\Aff(J)$ follows from $\dim(\Conv(J))=\dim(\tau)$. Note first, that $J\subseteq\tau$ implies also $\Aff(J)\subseteq\Aff(\tau)$. When $\dim(\Conv(J))=k$, there are $k+1$ affinely independent points in $\Conv(J)$ which span a basis of $\Aff(J)$. If $\tau\not\subseteq\Aff(J)$ there would be more than $k+1$ affinely independent points in $\tau$, which gives a contradiction.}. 
    Suppose that $\dim(\tau) > \dim (\Conv(J))$. Then we can find two points $\alpha,\beta\in\tau\setminus J$ with $\hat\xi_\alpha = \hat\xi_\beta=0$, such that the line segment from $\alpha$ to $\beta$ has an intersection point with $\Conv(J)$. Thus, there exist a rational number $0<\gamma<1$ and rational numbers $\lambda_j\geq 0$ describing this intersection point
    \begin{align}
        \gamma a^{(\alpha)} + (1-\gamma) a^{(\beta)} = \sum_{j\in J} \lambda_j a^{(j)} \qquad \text{with} \quad \sum_{j\in J} \lambda_j = 1 \point 
    \end{align}
    Denote by $m\in\mathbb Z_{>0}$ the least common multiple of all denominators of $\gamma$ and $\lambda_j$ for all $j\in J$. Then we can generate an element in $\mathbb L$ or in $H_\Aa(\beta)$, respectively
    \begin{align}
        \square = \partial_\alpha^{m\gamma} \partial_\beta^{m(1-\gamma)} - \prod_{j\in J} \partial_j^{m\lambda_j} \in H_\Aa(\beta) \point
    \end{align}
    Since its principal symbol
    \begin{align}
        \initial_{(\mathbf 0,\mathbf 1)} (\square) = \hat\xi_\alpha^{m\gamma} \hat\xi_\beta^{m(1-\gamma)} - \prod_{j\in J} \hat\xi_j^{m\lambda_j}
    \end{align}
    has to vanish for all values $(\hat z,\hat \xi)\in\ch(H_\Aa(\beta))$ we get a contradiction, since $\hat\xi_\alpha=\hat\xi_\beta=0$ and $\hat\xi_j\neq 0$ for all $j\in J$.
    
    Hence, $\Conv(J)$ and $\tau$ have the same dimension, and thus they share also the same supporting hyperplane. The second step will be to show equality $J=\tau$. Let $k\in\tau$ be an arbitrary point of the face $\tau$. We then have to prove that $\hat\xi_k\neq 0$. Since $\tau$ lies in the affine span of $J$, we will find some rational numbers $\lambda_j$ such that
    \begin{align}
        a^{(k)} = \sum_{j\in J} \lambda_j a^{(j)} = \sum_{\substack{j\in J \\ \lambda_j<0}} \lambda_j a^{(j)} + \sum_{\substack{j\in J \\ \lambda_j > 0}} \lambda_j a^{(j)} \qquad \text{with} \quad \sum_{j\in J} \lambda_j = 1 \point
    \end{align}      
    Again, let $m\in\mathbb Z_{>0}$ be the least common multiple of denominators of all $\lambda_j$ with $j\in J$, which will generate an element in $\mathbb L$, and we obtain
    \begin{align}
        \square = \partial_{k}^m \prod_{\substack{j\in J \\ \lambda_j<0}} \partial_j^{-m \lambda_j} - \prod_{\substack{j\in J \\ \lambda_j > 0}} \partial_j^{m\lambda_j} \in H_\Aa(\beta) \point
    \end{align}
    Both terms have the same order, since $1-\sum_{\lambda_j<0} \lambda_j = \sum_{\lambda_j > 0} \lambda_j$, which results in a principal symbol
    \begin{align}
        \initial_{(\mathbf 0,\mathbf 1)}(\square) = \xi_{k}^m \prod_{\substack{j\in J \\ \lambda_j<0}} \xi_j^{-m \lambda_j} - \prod_{\substack{j\in J \\ \lambda_j > 0}} \xi_j^{m\lambda_j} \point
    \end{align}
    Thus, it is $\hat\xi_{k}^m \prod_{\lambda_j<0} \hat\xi_j^{-m \lambda_j} = \prod_{\lambda_j> 0} \hat\xi_j^{m\lambda_j}$. By assumption, it is $\hat\xi_{j}\neq 0$ for all $j\in J$, and therefore it has to be also $\hat\xi_{k}\neq 0$. 
\end{proof}

In order to give a relation between $A$-discriminants and the characteristic varieties, we will associate to every finite subset $\Aa\subset\mathbb Z^{n+1}$ a multivariate polynomial
\begin{align}
    f_z(x) = \sum_{a^{(j)}\in\Aa} z_j x^{a^{(j)}} \in \mathbb C[x_0^{\pm 1},\ldots,x_n^{\pm 1}] \point
\end{align}
Recall, that for every face $\tau\subseteq\Conv(\Aa)$ we understand by\glsadd{truncpoly}
\begin{align}
    f_{\tau,z}(x) = \sum_{j\in\tau} z_j x^{a^{(j)}} \in \mathbb C[x_0^{\pm 1},\ldots,x_n^{\pm 1}]
\end{align}
the truncated polynomial with respect to the face $\tau$.

\begin{lemma} \label{lem:characteristic-disc}
    Let $\Aa\subset\mathbb Z^{n+1}$ be a vector configuration describing points on a hyperplane off the origin and let $\emptyset\neq\tau\subseteq\Conv(\Aa)$ be an arbitrary face. Then the following two statements are equivalent:
    \begin{enumerate}[i)]
        \item the point $(\hat z,\hat \xi)\in\ch\!\left(H_\Aa(\beta)\right)$ is a point of the characteristic variety and $\tau$ is the face of $\Conv(\Aa)$ corresponding to this point according to \cref{lem:face-characteristic}, i.e.\ $\hat \xi_j \neq 0$ if and only if $j\in\tau$
        \item the polynomials $\frac{\partial f_{\tau,\hat z}}{\partial x_0},\ldots,\frac{\partial f_{\tau,\hat z}}{\partial x_n}$ have a common zero in $x\in (\mathbb C^\star)^{n+1}$.
    \end{enumerate}        
\end{lemma}
\begin{proof}
    ``$ii)\Rightarrow i)$'': Let $\hat x\in (\mathbb C^\star)^{n+1}$ be a common solution of $\frac{\partial f_{\tau,\hat z}}{\partial x_0}=\ldots=\frac{\partial f_{\tau,\hat z}}{\partial x_n}=0$ which implies
    \begin{align}
        \hat x_i \frac{f_{\tau,\hat z}(\hat x)}{\partial x_i} = \sum_{j\in \tau} a_i^{(j)} \hat z_j \hat x^{a^{(j)}} = 0 \point
    \end{align}
    Consider the principal symbol of the homogeneous operators $E_i(\beta)\in H_\Aa(\beta)$ from \cref{eq:homogeneousOperators}. By setting all $\hat\xi_j=0$ for $j\notin\tau$ and $\hat\xi_j = x^{a^{(j)}}$ for all $j\in\tau$ we obtain
    \begin{align}
        \initial_{(\mathbf 0,\mathbf 1)}\!\left(E_i(\beta)\right)(\hat z,\hat\xi) = \sum_{a^{(j)}\in\Aa} a_i^{(j)} \hat z_j \hat\xi_j = 0\point
    \end{align}
    It remains to prove that $\initial_{(\mathbf 0,\mathbf 1)}(\square_l) (\hat z,\hat \xi) = 0$ for all $l\in\mathbb L$, where $\square_l$ was defined in equation \cref{eq:toricOperators}. Since all points described by $\Aa$ lying on a hyperplane off the origin, all monomials in $\square_l$ having the same order. Therefore, we have to show
    \begin{align}
        \prod_{l_j>0} \hat\xi_j^{l_j} = \prod_{l_j<0} \hat\xi_j^{-l_j} \quad\text{for all}\quad l\in\mathbb L \point \label{eq:xi=xi}
    \end{align}
    According to \cref{lem:face-kernel}, if $\tau\subseteq\Conv(\Aa)$ is a face\footnote{Since all points of $\Aa$ lying on a common hyperplane, the faces of $\Conv(A)$ and $\Conv(\Aa)$ are in one to one correspondence, where $A$ is a dehomogenization of $\Aa$.}, then $l_{\bar\tau} := (l_j)_{j\neq\tau}$ is either zero $l_{\bar\tau} = 0$ or it contains positive and negative components (or empty if $\tau=\Conv(\Aa)$). Thus, for the first case we insert $\hat\xi_j=\hat x^{a^{(j)}}$ for all $j\in\tau$
    \begin{align}
        \hat x^{\sum_{l_j>0} l_j a^{(j)}} = \hat x^{-\sum_{l_j<0} l_j a^{(j)}}
    \end{align}
    which is true, since all $l\in\mathbb L$ satisfy $\sum_j l_j a^{(j)} = \sum_{l_j>0} l_j a^{(j)} + \sum_{l_j<0} l_j a^{(j)} = 0$. In the second case, there are elements with $l_j<0$ as well as with $l_j>0$ corresponding to points outside of $\tau$ and equation \cref{eq:xi=xi} is trivially satisfied by $0=0$.\bigskip
    
    ``$i)\Rightarrow ii)$'': If $(\hat z,\hat\xi)\in\ch\!\left(H_\Aa(\beta)\right)$ that implies
    \begin{align}
        \initial_{(\mathbf 0,\mathbf 1)}\!\left(E_i(\beta)\right) = \sum_{a^{(j)}\in\Aa} a^{(j)} \hat z_j \hat\xi_j = \sum_{j\in\tau} a^{(j)} \hat z_j \hat\xi_j = 0 \point
    \end{align}
    Thus, $\frac{\partial f_{\tau,\hat z}}{\partial x_0},\ldots,\frac{\partial f_{\tau,\hat z}}{\partial x_n}$ have a common zero in $\hat x\in(\mathbb C^*)^{n+1}$ if the system of equations
    \begin{align}
        \hat x^{a^{(j)}} = \hat \xi_j \qquad\text{for all}\quad j\in\tau \label{eq:xxi}
    \end{align}
    has a solution in $\hat x\in(\mathbb C^*)^{n+1}$. Hence, we have to show that it is impossible to construct a contradicting equation by combining the equations of \cref{eq:xxi}. In other words for all integers $l_j\in \mathbb Z$ satisfying 
    \begin{align}
        \sum_{j\in\tau} l_j a^{(j)} = 0 \label{eq:lcomb}
    \end{align}
    we have to show that $\prod_{j\in\tau} (\hat\xi_j)^{l_j} = 1$. Note, that \cref{eq:lcomb} directly gives rise to an element in $\mathbb L$, by setting the remaining $l_j=0$ for all $j\notin\tau$. Therefore, we can construct
    \begin{align}
        \square = \prod_{l_j>0} \partial_j^{l_j} - \prod_{l_j<0} \partial_j^{-l_j} \in H_\Aa(\beta) \point \label{eq:squarel}
    \end{align}
    Again, by the fact that all points described by $\Aa$ lie on a common hyperplane off the origin, both terms in \cref{eq:squarel} have the same order. Thus,
    \begin{align}
        \initial_{(\mathbf 0,\mathbf 1)}(\square) = \prod_{l_j>0} \xi_j^{l_j} - \prod_{l_j<0} \xi_j^{-l_j} \comma
    \end{align}
    which completes the proof since $\initial_{(\mathbf 0,\mathbf 1)}(\square)(\hat z,\hat \xi) = 0$.
\end{proof}

Combining \cref{lem:characteristic-disc} with the results from \cref{sec:ADiscriminantsReultantsPrincipalADets} we can conclude directly the following theorem.

\begin{theorem}[Singular locus of $\Aa$-hypergeometric systems \cite{GelfandAdiscriminantsCayleyKoszulComplexes1990, AdolphsonHypergeometricFunctionsRings1994}] \label{thm:SingularLocusPrincipalAdet}
    Let $\Aa\subset\mathbb Z^{n+1}$ be a finite subset, which spans $\mathbb R^{n+1}$ as a vector space and describes points on a common hyperplane off the origin in $\mathbb R^{n+1}$. Furthermore, let $f=\sum_{a^{(j)}\in \Aa} z_j x^{a^{(j)}}\in\mathbb C[x_0^{\pm 1},\ldots,x_n^{\pm 1}]$ be the corresponding polynomial to $\Aa$. Then we have the equality
    \begin{align}
        \operatorname{Sing}\!\left(H_\Aa(\beta)\right) = \Var\!\left(E_\Aa(f)\right) \point
    \end{align}
\end{theorem}
\begin{proof}
    Let us assume for the moment that $\Aa$ has the form of a homogenizated point configuration, i.e.\ $\Aa=\{(1,a^{(1)}),\ldots,(1,a^{(N)})\}$ where $A=\{a^{(1)},\ldots,a^{(N)}\}\subset\mathbb Z^n$. Then the statement ii) in \cref{lem:characteristic-disc} is equal to a common zero of $f_{\tau,\hat z}, \pd{f_{\tau,\hat z}}{x_1},\ldots,\pd{f_{\tau,\hat z}}{x_n}$ in $x\in \Csn$. Thus, according to its definition \cref{eq:SingLocusDef}, the singular locus $\Sing\!\left(H_\Aa(\beta)\right)$ is then given by the Zariski closure of
    \begin{align}
        \bigcup_{\emptyset\neq \tau\subseteq\Conv(A)} \left\{ \hat z \in\mathbb C^N \,\rvert\, \Var\!\left(f_{\tau,\hat z},\frac{\partial f_{\tau,\hat z}}{\partial x_1},\ldots,\frac{\partial f_{\tau,\hat z}}{\partial x_n}\right) \neq\emptyset \text{ in } (\mathbb C^*)^n \right\} \point
    \end{align}
    But this is nothing else than the union of $A$-discriminants
    \begin{align}
        \Sing(H_\Aa(\beta)) = \bigcup_{\emptyset\neq \tau\subseteq\Conv(A)} \Var\!\left(\Delta_{A\cap\tau} (f_\tau)\right) \point
    \end{align}
    The application of \cref{thm:pAdet-factorization} concludes the proof for $\Aa=\{(1,a^{(1)}),\ldots,(1,a^{(N)})\}$. If $\Aa$ is not of that form we will always find a non-singular matrix $D$, such that $D\Aa$ will have this form. However, the $\Aa$-hypergeometric system as well as the $A$-discriminants are independent of such a transformation and thus the theorem applies also to all the other configurations $\Aa$.
\end{proof}

Thus, we have characterized the singular locus of $\Aa$-hypergeometric systems, which describes the possible singularities of the $\Aa$-hypergeometric functions, by the principal $A$-determinant. In general, it is a hard problem to calculate these principal $A$-determinants. However, by the \HKP we have a way to describe these possible singularities very efficiently in an indirect manner.


\chapter{Feynman integrals}
\label{ch:FeynmanIntegrals}
\label{sec:FeynmanIntegralsIntro}
%

After introducing the mathematical basis in the previous chapter, we now come to the actual main object of this investigation: the Feynman integral. To contextualize Feynman integrals, we will briefly sketch its origins in quantum field theory (QFT). For a detailed introduction to the topic of Feynman integrals we refer to \cite{WeinzierlFeynmanIntegrals2022, WeinzierlArtComputingLoop2006, SmirnovFeynmanIntegralCalculus2006, SmirnovAnalyticToolsFeynman2012, BognerMathematicalAspectsFeynman2009}.

The prototypical situation to study quantum field theories experimentally are scattering experiments, where particles are brought to collision under very high energies. Therefore, the scattering of particles is also immanent in the theoretical description of QFTs and the probability for certain events in such a scattering process is described by an operator called the \textit{$S$-matrix}. By means of the LSZ-reduction formula, the calculation of $S$-matrix elements reduces to the \textit{Green's function} or the \textit{correlator function} (see e.g.\ \cite{SchwartzQuantumFieldTheory2014}). Except for a very few toy models in low spacetime dimensions, the calculation of those Green's functions is only feasible by a perturbative approach. Thus, we will assume the Green's function as a formal power series in a coupling constant $g$, which describes the strength of interaction.

It was already observed in the early days of QFT \cite{FeynmanSpaceTimeApproachQuantum1949}, that the many terms arising in such a perturbation series can be represented and ordered by certain graphs\footnote{Richard Feynman introduced his graphs for the first time in 1948 during the Pocono conference about QED \cite{KaiserPhysicsFeynmanDiagrams2005}. Although, Feynman graphs have structural similarities to diagrams of Wentzel \cite{WentzelSchwereElektronenUnd1938} and Stueckelberg \cite{StueckelbergRemarqueProposCreation1941} their genuine connection to perturbation theory was fundamentally new.}. Depending on the given Lagrangian density $\mathcal L$ defining the QFT, those \textit{Feynman graphs} can consist in different types of edges and vertices. The map turning those graphs into algebraic expressions in the perturbation series is known as \textit{Feynman rules}. As the Feynman rules for vertices usually contain the coupling constant $g$, more complex Feynman graphs will stand for higher order terms in the perturbation series. Thus, for the prediction of an outcome of a scattering process we will calculate all possible Feynman graphs from a given theory up to a certain order. Thereby, Feynman graphs can be understood as a depiction of possibilities how certain particles can interact with each other. 

It is known for a long time \cite{DysonDivergencePerturbationTheory1952} that the perturbation series will diverge for the most QFTs. However, a finite number of terms of this series may still be a good approximation (in the sense of an asymptotic expansion). And indeed, the obtained results from this procedure are in surprisingly good agreement with experimental data. This overwhelming success of the predictions of QFT is also the reason why -- despite the mathematical and fundamental difficulties (e.g.\ Haag's theorem) -- it is worthwhile to further explore and develop QFT. \bigskip

Before introducing the language of graphs in more detail in the following \cref{sec:FeynmanGraphs}, we will summarize the most important characteristics of Feynman graphs. These graphs should be understood as networks, where we will assign a momentum flowing through every edge. By a momentum we mean a $d$-dimensional vector\glsadd{mink}\glsunset{mink} $\mink p = (p^0,\ldots, p^{d-1})$ in Minkowski space. Thus, a scalar product of those momentas $\mink p$, $\mink q$ is given by
\begin{align}
	\mink p \cdot \mink q = p^0 q^0 - p^1 q^1 - \ldots - p^{d-1} q^{d-1} \point
\end{align}
From these momenta every edge in the graph obtains an orientation. However, those orientations are not fixed and can be flipped by a change of the sign of the corresponding momentum. Furthermore, we will distinguish between internal and external edges. The \textit{external edges} are those, which are incident to a pendant vertex (i.e.\ a vertex of degree $1$) and are also known as \textit{legs} in the terminology of Feynman graphs. The momenta assigned to the legs are called \textit{external momenta}, and they are assumed to be the variables given by the experimental setting. For convenience, we will usually choose the external momenta to be incoming.

The internal edges are the essential part for a Feynman integral. Following the notion, that Feynman graphs represent possible ways of interaction of elementary particles, those internal edges are said to represent ``virtual particles'', as they carry a momentum which is off-shell, i.e.\ the assigned momentum does not satisfy the energy-momentum relation. 

Depending on the considered theory, edges, and vertices may also get additional weights or colorings. These decorations of graphs will constitute a further combinatorial difficulty but no general obstacles. However, we will not discuss those cases. \bigskip

Exemplary, we will rephrase the Feynman rules for the $\phi^4$-theory, which is one of the simplest possible theories consisting in one type of scalar particles with mass $m$ having the Lagrangian density $\mathcal L = \frac{1}{2} (\partial_\mu \phi)^2 - \frac{m^2}{2} \phi^2 - \frac{g}{4!} \phi^4$. For a summary of the Feynman rules appearing in the standard model, we refer to \cite{RomaoResourceSignsFeynman2012} which also collects the different choices of signs and prefactors.

\begin{example}[Feynman rules of $\phi^4$-theory in momentum space] \label{ex:phi4theory}
\hspace{0cm}
    \begin{enumerate}[I)]
        \item Obey momentum conservation at every vertex (except for the pendant vertices). This implies also an overall momentum conservation of the external momenta.
        \item for every internal edge: 
            \begin{tikzpicture}
                \coordinate[draw, shape=circle, fill=black, scale=.5] (A) at (0,0); 
                \coordinate[draw, shape=circle, fill=black, scale=.5] (B) at (3,0);
                \coordinate (a) at (1,0.3); \coordinate (b) at (2,0.3);
                \draw[thick] (A) -- (B);
                \draw[thick,->] (a) -- node [above] {$\mink q$} (b); 
            \end{tikzpicture}
            ${\displaystyle \qquad \longmapsto\quad\frac{-i}{- \mink q^2+m^2}}$
        \item for every vertex: \hspace{1cm} \vcenteredhbox{\begin{tikzpicture}[scale=0.5]
                \coordinate (A) at (-1,-1);
                \coordinate (B) at (1,-1);
                \coordinate (C) at (1,1);
                \coordinate (D) at (-1,1);
                \coordinate[draw, shape=circle, fill=black, scale=.5] (O) at (0,0);
                \draw[thick] (A) -- (C); \draw[thick] (B) -- (D);
            \end{tikzpicture}}
            ${\displaystyle \qquad \longmapsto\quad i g}$ 
        \item integrate over every indeterminated (internal) momentum $\mink k$: ${\displaystyle \int_{\mathbb R^d} \frac{\ddif \mink k}{(2\pi)^d}}$
    \end{enumerate}	
\end{example}

Hence, for the Green's function we will sum all possible Feynman graphs, which can be built from these rules weighted with an additional symmetry factor. However, we can make three major simplifications in this sum. First of all, we can omit all vacuum graphs, i.e.\ we neglect all those graphs, which have no legs, since their contribution can be comprised in a normalization factor. Second, we will exclude all disconnected graphs in our considerations. In general a disconnected graph would evaluate to the product of its components. However, the contribution of disconnected graphs to the perturbation series is usually disregarded in most QFTs by the cluster decomposition principle\footnote{The cluster decomposition principle is based on locality considerations, that assume for disconnected graphs to represent separated processes which do not influence each other. This is equivalent to claim that the $S$-matrix contains no singularities worse than poles and branch points besides one single overall momentum conservation $\delta$-distribution \cite{WeinbergQuantumTheoryFields1995}. However, it is by no means clear whether the cluster decomposition principle is always fulfilled (see also the discussion in \cite{SchwartzQuantumFieldTheory2014, WeinbergWhatQuantumField1997}). For example, the cluster decomposition principle could be violated by colour confinement in QCD \cite{LowdonConditionsViolationCluster2016}. Note that in the case of vacuum components, these can also be neglected by the aforementioned normalization factor.} \cite{WeinbergQuantumTheoryFields1995}. Third, we will restrict ourselves to the so-called \textit{amputated graphs} or \textit{truncated graphs}. Thus, we will not consider any graph, where a leg joins a subgraph of self-energy type (i.e.\ a subgraph having only two legs). Those cases are also ruled out by LSZ-reduction formula \cite{SchwartzQuantumFieldTheory2014, ItzyksonQuantumFieldTheory1980}. 

The amputated graphs are closely related to the so-called \textit{$1$-particle irreducible} graphs ($1$PI) or \textit{bridgefree} graphs, i.e.\ graphs which are still connected, when cutting an arbitrary edge. In general, a Feynman integral of a $1$-particle reducible graph evaluates to a product of all its components when cutting its bridges. Even though $1$-particle reducible graphs contribute to the Green's function, we want to exclude them in this consideration. Due to the factorization property, this is not a restriction to generality. 

Also, the contributions of graphs having a cut vertex can be factorized into their components. Thereby, a \textit{cut vertex} is a vertex which increases the number of components when it will be removed. According to \cite{SmirnovAnalyticToolsFeynman2012} we will call a graph without any cut vertex, a \textit{$1$-vertex irreducible graph} ($1$VI). If additionally at least one of those components is a vacuum graph, we will call it a \textit{tadpole-like} graph. If those vacuum components not only have no legs but also no masses, they will be independent of any variable. Therefore, we will call those types of components \textit{scaleless}. Scaleless components can be renormalized to zero. Furthermore, tadpole-like graphs can also be omitted due to the renormalization procedure\footnote{In \cite{SchwartzQuantumFieldTheory2014} and \cite[cor. 3.49]{PrinzAlgebraicStructuresCoupling2021} this was shown for a subclass of the tadpole-like graphs. However, one can directly extend their results to the entire class of tadpole-like graphs using \cite[prop. 3.48]{PrinzAlgebraicStructuresCoupling2021}.}. Thus, we will lastly also omit all tadpole-like graphs in the Green's function. \bigskip

\begin{figure}
	\centering
	\begin{subfigure}{.3\textwidth}
		\centering
		\begin{tikzpicture}[thick, dot/.style = {draw, shape=circle, fill=black, scale=.5}, scale=0.7]
            \coordinate[dot] (A) at (0,0);
            \coordinate[dot] (B) at (2,0);
            \draw (A) -- (B);  
            \draw (-.7,.7) -- (A); \draw (-.7,0) -- (A); \draw (-.7,-.7) -- (A); 
            \draw (2.7,.7) -- (B); \draw (2.7,0) -- (B); \draw (2.7,-.7) -- (B);
        \end{tikzpicture}
        \caption{a tree graph}
	\end{subfigure}
	\begin{subfigure}{.3\textwidth}
        \centering
        \begin{tikzpicture}[thick, dot/.style = {draw, shape=circle, fill=black, scale=.5}, scale=0.7]
            \coordinate[dot] (A) at (0,0);
            \coordinate[dot] (B) at (2,0);
            \draw (1,0) circle (1);
            \draw (-.7,.7) -- (A); \draw (-.7,-.7) -- (A); 
            \draw (2.7,.7) -- (B); \draw (2.7,-.7) -- (B);
        \end{tikzpicture}
        \caption{a $1$-loop graph}
	\end{subfigure}
	\begin{subfigure}{.3\textwidth}
		\centering
		\begin{tikzpicture}[thick, dot/.style = {draw, shape=circle, fill=black, scale=.5}, scale=0.7]
            \coordinate[dot] (A) at (0,0);
            \coordinate[dot] (B) at (2,0);
            \draw (1,0) circle (1);
            \draw (-.7,.7) -- (A); \draw (-.7,-.7) -- (A); 
            \draw (2.7,.7) -- (B); \draw (2.7,-.7) -- (B);
            \coordinate[dot] (C) at (4,0);
            \draw (3.3,.7) -- (C); \draw (4.7,.7) -- (C);
            \draw (3.3,-.7) -- (C); \draw (4.7,-.7) -- (C);
        \end{tikzpicture}
        \caption{a disconnected graph}
	\end{subfigure}
	
	\vspace{.7cm}
	
	\begin{subfigure}{.3\textwidth}
		\centering
        \begin{tikzpicture}[thick, dot/.style = {draw, shape=circle, fill=black, scale=.5}, scale=0.7]
            \coordinate[dot] (A) at (0,0);
            \draw (-1,0) circle (1);
            \draw (1,0) circle (1);
        \end{tikzpicture}
        \caption{a vacuum graph}
	\end{subfigure}
	\begin{subfigure}{.3\textwidth}
		\centering
        \begin{tikzpicture}[thick, dot/.style = {draw, shape=circle, fill=black, scale=.5}, scale=0.7]
            \coordinate[dot] (A) at (0,0);
            \coordinate[dot] (B) at (2,0);
            \draw (1,0) circle (1);
            \draw (A) -- (B);
            \draw (-.7,0) -- (A); 
            \draw (2.7,0) -- (B);
        \end{tikzpicture}
        \caption{a $2$-loop graph}
	\end{subfigure}
	\begin{subfigure}{.3\textwidth}
		\centering
        \begin{tikzpicture}[thick, dot/.style = {draw, shape=circle, fill=black, scale=.5}, scale=0.7]
            \coordinate[dot] (A) at (0,0);
            \coordinate[dot] (B) at (2,0);
            \coordinate[dot] (C) at (3,0);
            \coordinate[dot] (D) at (5,0);
            \draw (1,0) circle (1); \draw (4,0) circle (1); 
            \draw (A) -- (B) -- (C) -- (D);
            \draw (-.7,0) -- (A); 
            \draw (5.7,0) -- (D);
        \end{tikzpicture}
        \caption{a $1$-particle reducible graph (also a non-amputated graph)}
	\end{subfigure}

	\vspace{.7cm}
	
	\begin{subfigure}{.3\textwidth}
		\centering
        \begin{tikzpicture}[thick, dot/.style = {draw, shape=circle, fill=black, scale=.5}, scale=0.7]
            \coordinate[dot] (A) at (0,0);
            \coordinate[dot] (B) at (2,0);
            \coordinate[dot] (C) at (4,0);
            \draw (1,0) circle (1); \draw (3,0) circle (1); 
            \draw (-.7,.7) -- (A); \draw (-.7,-.7) -- (A); 
            \draw (4.7,.7) -- (C); \draw (4.7,-.7) -- (C);
        \end{tikzpicture}
        \caption{a $1$-vertex reducible graph}
	\end{subfigure}		
	\begin{subfigure}{.3\textwidth}
		\centering
        \begin{tikzpicture}[thick, dot/.style = {draw, shape=circle, fill=black, scale=.5}, scale=0.7]
            \coordinate[dot] (A) at (0,0);
            \draw (0,1) circle (1);
            \draw (-1.7,0) -- (A); 
            \draw (1.7,0) -- (A);
        \end{tikzpicture}
        \caption{a tadpole graph} \label{fig:tadpoleGraph}
	\end{subfigure}	
	\begin{subfigure}{.3\textwidth}
		\centering
        \begin{tikzpicture}[thick, dot/.style = {draw, shape=circle, fill=black, scale=.5}, scale=0.7]
            \coordinate[dot] (A) at (0,0);
            \coordinate[dot] (B) at (2,0);
            \coordinate[dot] (C) at (1,-1);
            \coordinate[dot] (D) at (3,1);
            \coordinate[dot] (E) at (3,-1);
            \draw (1,0) circle (1); \draw (3,0) circle (1); 
            \draw (C) arc[start angle=0, end angle=90, radius=1];
            \draw (D) arc[start angle=130, end angle=230, radius=1.305];
            \draw (E) arc[start angle=-50, end angle=50, radius=1.305];
            \draw (-.7,0) -- (A); 
            \draw (1,-1.7) -- (C);
        \end{tikzpicture}
        \caption{a tadpole-like graph}
	\end{subfigure}
	
	\caption[Examples of Feynman graphs in $\phi^4$-theory]{Examples of certain Feynman graphs in $\phi^4$-theory. Pendant vertices are not drawn explicitly.}
\end{figure}
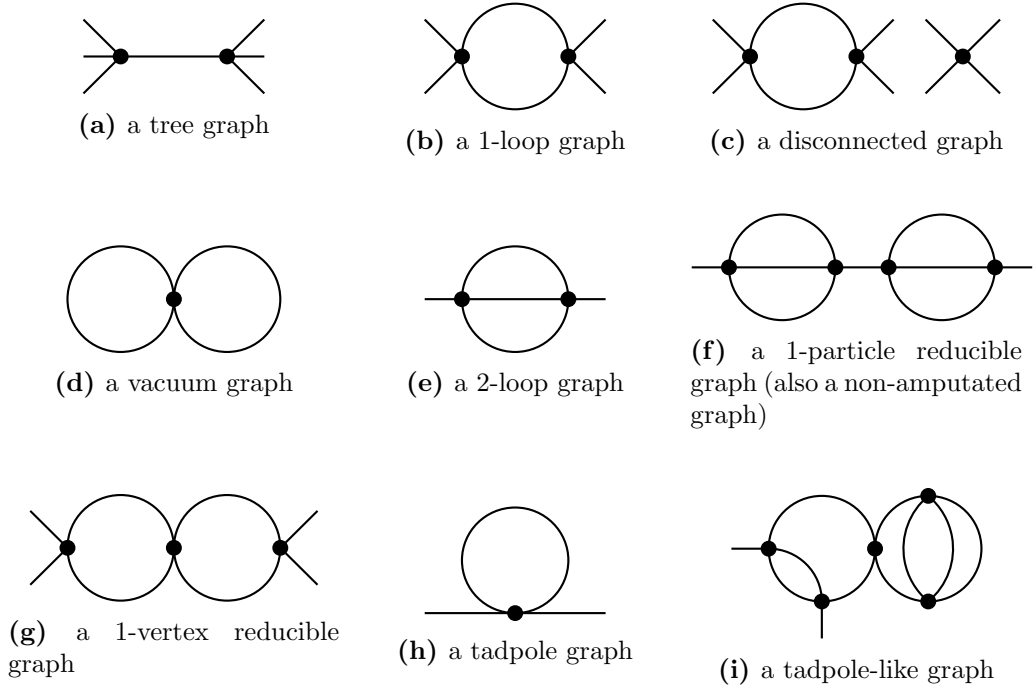

Our special interest here is the last Feynman rule from \cref{ex:phi4theory}, which also appears analogously for any other QFT. Hence, we will obtain certain integrals in the evaluation of Feynman graphs, whenever the graph contains a loop. A graph with $n$ internal edges and $L$ indeterminated internal momenta $\mink k_1,\ldots, \mink k_L$ results in a \textit{Feynman integral} of the form
\begin{align}
    \int_{\mathbb R^{d\times L}} \left(\prod_{j=1}^L \frac{\ddif \mink k_j}{i\pi^{d/2}} \right) \prod_{i=1}^n \frac{1}{- \mink q_i^2 + m_i^2} \comma \label{eq:FeynmanMomSpMinkowski}
\end{align}
where $\mink q_i$ is the (Minkowski-)momentum flowing through the edge $e_i$ and $m_i$ is the corresponding mass. For convenience reasons we adjusted the prefactor in \cref{eq:FeynmanMomSpMinkowski} slightly, see also \cite{WeinzierlFeynmanIntegrals2022}. As momenta and masses come with a specific unit, we often introduce an additional parameter in \cref{eq:FeynmanMomSpMinkowski} to make the Feynman integral dimensionless, which becomes important in the renormalization procedure. For the sake of notational simplicity we will omit this parameter. 

Since the denominator in the integrand vanishes on the integration contour, the integral in \cref{eq:FeynmanMomSpMinkowski} is ill-defined. This issue is typically solved by introducing a small imaginary part $-i\epsilon$ with $\epsilon>0$ in the denominator of \cref{eq:FeynmanMomSpMinkowski}. As we will assume $\mink q_i$ and $m_i$ to be real valued in the physical application, the poles of the integrand can be omitted for generic external momenta in that way. Equivalently, we can slightly change the integration contour. This idea will be made more explicitly in \cref{sec:Coamoebas} by means of coamoebas. However, there are certain cases where this procedure fails and we have to expect singularities depending on the external momenta and the masses. The precise description of those singularities is very subtle and the whole \cref{ch:singularities} is devoted to examine the analytic behaviour caused by those singularities. \bigskip

To exclude the problems arising from these singularities for now, we will consider a slightly different version of \cref{eq:FeynmanMomSpMinkowski} where we replace the Minkowskian kinematics by Euclidean kinematics. Therefore, we will 
define the \textit{Feynman integral} of a Feynman graph $\Gamma$ to be
\begin{align}
    \gls{FI} = \int_{\mathbb R^{d\times L}} \left(\prod_{j=1}^L \frac{\ddif k_j}{\pi^{d/2}}\right) \prod_{i=1}^n \frac{1}{(q_i^2+m_i^2)^{\nu_i}} \label{eq:FeynmanMomSp} \comma
\end{align}
where the momentum $q_i$\glsadd{internalmomenta}\glsunset{internalmomenta} attached to an edge $e_i$ is now a $d$-dimensional Euclidean vector, i.e.\ $q_i^2 := (q_i^0)^2 + \ldots (q_i^{d-1})^2$. Additionally, we introduced \textit{indices} $\nu = (\gls{nu})$ to the propagators, which will turn out to be convenient, e.g.\ for the application in the so-called integration-by-parts (IBP) methods \cite{ChetyrkinIntegrationPartsAlgorithm1981, TkachovTheoremAnalyticalCalculability1981} and general linear relations of Feynman integrals \cite{BitounFeynmanIntegralRelations2019} as well as in the analytical regularization \cite{SpeerGeneralizedFeynmanAmplitudes1969}. We will see in \cref{sec:DimAnaReg} and \cref{sec:FeynmanIntegralsAsAHyp} the dependence of the Feynman integral from those indices. \bigskip

For the Feynman integral \cref{eq:FeynmanMomSp} we will treat indices $\nu$ and the spacetime dimension $d$ as parameters. Originally restricted to positive integer values, we will meromorphically continue the Feynman integral $\mathcal I_\Gamma$ to complex values $\nu\in\mathbb C^n$ and $d\in\mathbb C$ in \cref{sec:DimAnaReg}.

Apart from the singularities arising from poles of the integrand, which we omitted for the massive case by the Euclidean kinematics, we may also have to worry about the behaviour of the integrand for large momenta $|k_j|\rightarrow \infty$. A divergence stemming from large momenta is called an \textit{UV-divergence}. In the massless case $m_i=0$, we may also have divergences for small momenta, which are called \textit{IR-divergences}. It will turn out, that the convergence behaviour can be controlled by the parameters $\nu$ and $d$. For now, we will treat the Feynman integral \cref{eq:FeynmanMomSp} as a formal integral, and we will answer the question of convergence in \cref{sec:DimAnaReg}.

The external momenta $p = (p_1,\ldots,p_m)^\top$ and the internal masses $m = (m_1,\ldots,m_n)^\top$ are considered as variables of \cref{eq:FeynmanMomSp}. Analogue to the parameters we will analytically continue those variables to complex values which will be accomplished in \cref{ch:singularities}. However, for this section we will assume these variables to be real. Instead of the $d$-dimensional momenta $p_1,\ldots,p_m$ it will be often more convenient to consider the Feynman integral to be depending on scalar products $s_{ij} = p_i p_j$ of those external momenta. This will be more apparent by the parametric representations in \cref{sec:ParametricFeynmanIntegrals}.\bigskip

When the momenta (or rather their scalar products) continued to complex values we can also relate \cref{eq:FeynmanMomSpMinkowski} with \cref{eq:FeynmanMomSp}. By inserting imaginary values for the zeroth component of every momentum, we will change the Euclidean vectors to Minkowskian vectors. This technical ``trick'' is known as \textit{Wick rotation}. Therefore, by the analytic continuation of the momenta to complex values, the Euclidean Feynman integral \cref{eq:FeynmanMomSp} will also include the Minkowskian case \cref{eq:FeynmanMomSpMinkowski}. \bigskip

The Feynman integrals in the form of \cref{eq:FeynmanMomSpMinkowski} and \cref{eq:FeynmanMomSp} are known as \textit{scalar Feynman integrals}. However, when considering more complex theories Feynman integrals will may obtain a tensorial structure, i.e.\ the Feynman integrals contain additional momenta in the numerator. It is known for a long time, that those tensorial Feynman integrals can be reduced to a linear combination of scalar integrals. Therefore, we can restrict our considerations fully to the case of \cref{eq:FeynmanMomSp}. We will present a method for such a reduction in the end of \cref{sec:ParametricFeynmanIntegrals} by \cref{thm:TensorReduction}.


\section{Feynman graphs}
\label{sec:FeynmanGraphs}
    

As pointed out in the previous section, Feynman integrals are based on graphs. Therefore, we will shortly rephrase the most important notions in graph theory related to Feynman graphs. Further information about classical graph theory can be found in e.g.\ \cite{HarayGraphTheory1969, TutteGraphTheory1984, ThulasiramanHandbookGraphTheory2016} and we will refer \cite{NakanishiGraphTheoryFeynman1971} for a comprehensive introduction to Feynman graphs. \cite{BapatGraphsMatrices2014} has its main focus on graph matrices and for a treatment of Feynman graph polynomials we suggest \cite{BognerFeynmanGraphPolynomials2010}.

Oriented graphs as we will introduce them in this section are very closely related to convex polytopes from \cref{sec:ConvexPolytopes}. For example there is also a version of Farkas' \cref{lem:Farkas} applying to oriented graphs \cite{BachemLinearProgrammingDuality1992}. The underlying reason is, that both objects can be described via oriented matroids \cite{BachemLinearProgrammingDuality1992, BjornerOrientedMatroids1999}. However, we will introduce graphs without relying on oriented matroids, to keep this summary short. \bigskip

As aforementioned a Feynman graph consists in internal and external edges. Since the external edges do not contribute to the following combinatorial considerations, we will omit them in our description to get a simpler notation. Thus, for our purpose a graph $\Gamma=(E,V,\varphi)$ consists in a set of $n$ edges $E=\{e_1,\ldots,e_n\}$, a  set of $m$ vertices $V=\{v_1,\ldots,v_m\}$ and an incidence map $\varphi : E \rightarrow V \times V$ relating edges $e$ and vertices $u,v$ by $\varphi(e) = (u,v)$. We will orient the edges, and consequentially we will distinguish between the \textit{start point} $u$ and the \textit{end point} $v$ of an edge $e$. Except for graphs where the start point and the end point of an arbitrary edge coincide, it is convenient to formulate the incidence relations by a $m\times n$ matrix
\begin{align}
	I_{ij} = \begin{cases}
	         	+ 1 & v_i \text{ start point of } e_j  \\
	         	-1 & v_i \text{ end point of } e_j \\
	         	0 & e_j \text{ not incident with } v_i
	         \end{cases}
\end{align}
called the \textit{incidence matrix} \gls{incidence}. It is not hard to show, that those incidence matrices are totally unimodular (i.e.\ the determinant of every non-singular square submatrix of $I$ is equal $\pm 1$) \cite[lem. 2.6]{BapatGraphsMatrices2014}, \cite[thm. 8.13]{ThulasiramanHandbookGraphTheory2016} and that their rank is given by 
\begin{align}
	\rank (I) = m - b_0
\end{align}
with the number of vertices $m$ and the number of connected components \gls{b0} of the graph \cite[thm. 2.3]{BapatGraphsMatrices2014}, \cite[cor. 8.1]{ThulasiramanHandbookGraphTheory2016}. \bigskip

A graph $\Gamma^\prime=(E^\prime,V^\prime,\varphi|_{E^\prime})$ satisfying $E^\prime\subseteq E$ and $V^\prime\subseteq V$ with an incidence relation $\varphi|_{E^\prime}$ restricted to $E^\prime$ is called a \textit{subgraph} of $\Gamma=(E,V,\varphi)$, symbolically $\Gamma^\prime \subseteq\Gamma$. In case of $V^\prime = V$ the subgraph is said to be \textit{spanning}. By a \textit{loop}\footnote{We use the nomenclature which is ubiquitous in the theory of Feynman graphs. However, in graph theory this object is usually called a \textit{circuit}, whereas a ``loop'' in graph theory is what physicists often call a ``tadpole'' or a ``self-loop''.} of $\Gamma$ we understand a subgraph where every vertex is incident with exactly two (not necessarily distinct) edges. A \textit{tadpole} refers to a loop consisting in only one edge. A graph $\Gamma$ which does not contain any loop is called a \textit{forest}. Consequentially, a \textit{tree} is a forest consisting in only one component. We denote the set of all spanning forests with $k$ components of a given graph by \gls{spanT}.

A tree $T$ will have one vertex more than edges $|E_T|+1=|V_T|$. Therefore, for a connected graph $\Gamma$ we have to delete at least $L=|E|-|V|+1$ edges to turn $\Gamma$ into a spanning tree. For a graph consisting in $b_0$ components we obtain similarly at least $L=|E|-|V|+b_0$ edges which have to deleted to turn $\Gamma$ into a spanning forest.\bigskip

Let \gls{circuits} be the set of all loops of a graph $\Gamma$. In addition to the orientation of edges, we will also introduce an (arbitrary) orientation of every loop and define
\begin{align}
	C_{ij}  = \begin{cases}
	         	+ 1 & e_j \in \mathscr C_i \text{ and }  \mathscr C_i, e_j \text{ have the same orientation} \\
	         	- 1 & e_j \in \mathscr C_i \text{ and }  \mathscr C_i, e_j \text{ have the opposite orientation} \\
	         	0 & e_j \notin \mathscr C_i
	         \end{cases} \label{eq:loopmatrix}
\end{align}
the $r\times n$ \textit{loop matrix} or \textit{circuit matrix}. \bigskip

Similar to the linear evaluations and linear dependences of \cref{ssec:GaleDuality} we obtain the following fact.
\begin{lemma} \label{lem:IncidenceLoopOrthogonality}
	For a graph without tadpoles, the incidence matrix and the loop matrix are orthogonal
	\begin{align}
		I \cdot C^\top = 0 \point
	\end{align}
\end{lemma}
\begin{proof}
	By definition, the term $I_{ij}C_{kj}$ is only non-zero if the vertex $v_i$ is incident with an edge contained in the loop $\mathscr C_k$. In other words, the vertex $v_i\in\mathscr C_k$ is contained in the loop. However, if this is the case, there are exactly two edges incident with $v_i$ which are both in $\mathscr C_k$. By the sign conventions the two contributions will always have different sign and sum up to zero.
\end{proof}

Let $T\in\mathscr T_1$ be a spanning tree of the graph $\Gamma$. By slight abuse of notation we will denote by $T$ and $\Gamma$ also the edge set of the tree and the graph, respectively. Furthermore, let $\gls{chord} = \{e_1,\ldots,e_L\} = \Gamma\setminus T$ be its complement, which is also known as \textit{cotree} or \textit{chord set}. For every two vertices in the tree $T$, there will be a unique path in $T$ connecting these vertices. Therefore, for every edge $e_j\in T^*$ we can uniquely construct a loop $\mathscr C_j := P_j \cup e_j$, where $P_j$ is the path in $T$ connecting the start and the end point of $e_j$. For convenience we will choose the orientation of $\mathscr C_j$ in the same direction as $e_j$. 

For a chord set $T^*$, we call the set $\mathscr C_1,\ldots,\mathscr C_L$ of loops constructed in that way a \textit{fundamental set of loops}. By definition, every loop in this set will have at least one edge which is not contained in any other loop $\mathscr C_j \nsubseteq \cup_{k\neq j} \mathscr C_k$. Furthermore, one can show that for a maximal set of loops this property is also sufficient to be a fundamental set of loops \cite{NakanishiGraphTheoryFeynman1971}. We will denote by \gls{fundloop} the loop matrix \cref{eq:loopmatrix} restricted to rows corresponding to the loops of a fundamental set generated by the chord set $T^*$. Owing to its construction this matrix will have the form
\begin{align}
	C_{T^*} = \begin{pmatrix}
	          	\, (C_{T^*})_T, & \!\!\!\!\mathbbm 1 \,\,
	          \end{pmatrix} \label{eq:LoopMatrixUnit}
\end{align}
where the columns of the unit matrix correspond to edges in $T^*$, whereas $(C_{T^*})_T$ denotes the restriction of columns corresponding to edges in $T$. Similar to the incidence matrix, it can be shown that any fundamental loop matrix $C_{T^*}$ is totally unimodular \cite[lem. 5.12]{BapatGraphsMatrices2014} \cite[thm. 8.15]{ThulasiramanHandbookGraphTheory2016}.

Moreover, it is not hard to deduce a representation of $(C_{T^*})_T$ in terms of the incidence matrix by means of \cref{lem:IncidenceLoopOrthogonality}, whenever the graph does not contain tadpoles. Let $I'$ be the incidence matrix, where we removed $b_0$ rows such that $I'$ has full rank. With the same convention as in \cref{eq:LoopMatrixUnit} we will split $I'= (I'_{T}, I'_{T^*})$ into columns corresponding to $T$ and to $T^*$. We will then obtain \cite[thm. 5.6]{BapatGraphsMatrices2014}
\begin{align}
	C_{T^*} = \begin{pmatrix}
	          	\, - I'^\top_{T^*}\!\left(I'^\top_{T}\right)^{-1}, & \!\!\!\!\mathbbm 1\,\,
	          \end{pmatrix} \point \label{eq:LoopMatrixUnitIncidence} 
\end{align}

Furthermore, \cref{eq:LoopMatrixUnit} shows that the fundamental loop matrix $C_{T^*}$ has always full rank $L$. Hence, $L$ is also a lower bound for the rank of the full loop matrix $\rank (C) \geq L$. On the other hand, \cref{lem:IncidenceLoopOrthogonality} and the rank-nullity theorem lead to an upper bound $\rank (C) \leq n - \rank (I) = L$. Thus, we have $\rank (C) = L$ and every fundamental set of loops $C_{T^*}$ will provide a basis for the space of all loops. We call the dimension of the space of loops the \textit{loop number} \gls{loop}. From a topological point of view, $L$ is the first Betti number of $\Gamma$.

\begin{lemma}[Minors of the loop matrix {\cite[lem 5.9]{BapatGraphsMatrices2014}} ] \label{lem:LoopMatrixMinors}
    Let $\Gamma=(E,V,\varphi)$ be a connected graph with loop number $L$ and $S\subseteq E$ be a subset of $L$ edges. The restriction $(C_{T^*})_S$ of a fundamental loop matrix $C_{T^*}$ to columns of $S$ is non-singular if and only if $\Gamma\setminus S$ forms a spanning tree.
\end{lemma}
\begin{proof}
	``$\Leftarrow$'': As considered $S$ is a chord set. Therefore, we can also construct a fundamental set of loops from $S$ with a fundamental loop matrix $C_S$. By changing the basis we will always find a non-singular matrix $R$ such that $C_{T^*} = R\, C_S$. Using the representation \cref{eq:LoopMatrixUnit} for $C_S$, we will find $\left(C_{T^*}\right)_S = R$.
	
	``$\Rightarrow$'': Assume that $\Gamma\setminus S$ is not a spanning tree. Then $\Gamma\setminus S$ has to contain a loop\footnote{The case of a forest or a non-spanning tree can be excluded by analysing the relation $L=|E|-|V|+b_0$ for $\Gamma$ and $\Gamma\setminus S$.}. Choose a fundamental set of loops of $\Gamma$ containing this loop and let $C_{T_0^*}$ the corresponding fundamental loop matrix. Again, we will find a basis transformation $C_{T^*} = R\, C_{T_0^*}$ with a non-singular matrix $R$. As considered, there is a loop in the fundamental set which has no edges from $S$. Therefore, $\left(C_{T_0^*}\right)_S$ has a row of zeros. Hence, $\left(C_{T_0^*}\right)_S$ is singular and thus also $\left(C_{T^*}\right)_S$ is singular.
\end{proof}

By means of this lemma about the minors of $C_{T^*}$, we are able to show the main result of this section, which will establish a connection between graphs and polynomials. For this purpose we will introduce a variable $x_i$ for every edge $e_i$ of the graph.

\begin{theorem}[Matrix-tree theorem] \label{thm:MatrixTree}
    Let $X = \diag (x_1,\ldots,x_n)$ be the diagonal matrix of the edge variables. Then
    \begin{align}
    	\det\!\left(C_{T^*} X C_{T^*}^\top \right) = \sum_{T \in \mathscr T_1} \prod_{x_e\notin T} x_e \label{eq:MatrixTree1}
    \end{align}
    is independent of the chosen chord set $T^*$.
\end{theorem}
\begin{proof}
	By the Cauchy-Binet formula we can split the determinant into
	\begin{align}
		\det\!\left(C_{T^*} X C_{T^*}^\top \right) = \sum_S \det\!\left(C_{T^*} X\right)_S  \det\!\left(C_{T^*}\right)_S = \sum_S \left(\det\!\left(C_{T^*}\right)_S \right)^2 \prod_{e\in S} x_e
	\end{align}
	where the sum goes over all edge subsets $S\subseteq E$ of length $L$. Due to \cref{lem:LoopMatrixMinors} the sum will reduce to terms where $\Gamma\setminus S$ is a spanning tree. Total unimodularity of $C_{T^*}$ concludes the proof.
\end{proof}

The polynomial on the right hand side of \cref{eq:MatrixTree1} will be introduced in the next section as the first Symanzik polynomial. It is also known as dual Kirchhoff polynomial, due to a polynomial appeared in Kirchhoff's analysis of electrical networks \cite{KirchhoffUeberAuflosungGleichungen1847}. This coincidence comes with no surprise, as electrical networks and Feynman graphs are structurally very similar. We refer to \cite{SeshuLinearGraphsElectrical1961} for a summary of graphs in electrical networks, which contains also the most of the results that are of significance for Feynman graphs. For the various ways to represent Symanzik polynomials we suggest \cite{BognerFeynmanGraphPolynomials2010}.

Often the matrix-tree \cref{thm:MatrixTree} is known in a slightly different variant. 
\begin{cor}
	Let $I'$ be the incidence matrix of a tadpole free graph, where we removed $b_0$ rows to obtain a full ranked matrix. Further, let $\hat{\mathscr L} = I' X^{-1} I'^\top$ be the so-called reduced, dual Laplacian. Then we have
	\begin{align}
		 \sum_{T \in \mathscr T_1} \prod_{x_e\notin T} x_e = \det (X) \det (\hat{\mathscr L}) \point
	\end{align}
\end{cor}
\begin{proof}
	Using the representation \cref{eq:LoopMatrixUnitIncidence} we can write 
	\begin{align}
		\det\!\left(C_{T^*} X C_{T^*}^\top\right) = \det\!\left(X_{T^*}\right) \det\!\left[ X_{T^*}^{-1} I'^\top_{T^*} \left(I'^\top_{T}\right)^{-1} X_T \left(I'_{T}\right)^{-1} I'_{T^*} + \mathbbm 1 \right]
	\end{align}
	where we split the diagonal matrix $X$ in the same way, i.e.\ $X_{T} = \diag(x_i)_{e_i\in T}$ and $X_{T^*} = \diag(x_i)_{e_i\in T^*}$. Using Sylvester's determinant identity $\det(\mathbbm 1 + AB) = \det(\mathbbm 1 + BA)$ we can rewrite
	\begin{align}
		&\det\!\left(X_{T^*}\right) \det\!\left[ X_T \left(I'_{T}\right)^{-1} I'_{T^*} X_{T^*}^{-1} I'^\top_{T^*} \left(I'^\top_{T}\right)^{-1}  + \mathbbm 1 \right] \nonumber \\
		&= \det\!\left(X_{T^*}\right) \det\!\left(X_{T}\right)  \det\!\left[ \left(I'_{T}\right)^{-1} \left( I'_{T} X_{T}^{-1} I'^\top_T + I'_{T^*} X_{T^*}^{-1} I'^\top_{T^*} \right) \left(I_T^\top\right)^{-1} \right] \nonumber \\
		&= \det(X) \det\!\left[ I'_{T} X_{T}^{-1} I'^\top_T + I'_{T^*} X_{T^*}^{-1} I'^\top_{T^*} \right] = \det (X) \det\!\left(I' X^{-1} I'^\top\right) \point
	\end{align}
\end{proof}

Similar to the deletion and contraction of vector configurations (see \cref{ssec:GaleDuality}), we can introduce the operations deletion and contraction on graphs. Let $e_i\in E$ be an edge of the graph $\Gamma=(E,V,\varphi)$. Then $\contraction{\Gamma}{e_i}$ denotes the \textit{contraction}\glsadd{contraction} of $e_i$, which means that we shrink the edge $e_i$ such that the endpoints of $e_i$ become a single vertex. For the fundamental loop matrix $C_{T^*}$ of $\Gamma$ this implies that we remove the column corresponding to $e_i$ to arrive at the fundamental loop matrix of $\contraction{\Gamma}{e_i}$. Denote by $\Uu(\Gamma)=\det\!\left(C_{T^*}XC_{T^*}^\top\right)$ the polynomial of equation \cref{eq:MatrixTree1} for a graph $\Gamma$. From the previous considerations we immediately obtain
\begin{align}
	\Uu (\contraction{\Gamma}{e_i}) = \Uu(\Gamma)\big|_{x_i=0} \point \label{eq:UContracted}
\end{align}
By $\Gamma\setminus e_i$ we mean the \textit{deletion}\glsadd{deletion} of $e_i$, i.e.\ $\Gamma\setminus e_i = \Gamma(E\setminus e_i,V,\varphi |_{E\setminus e_i})$. Assuming that $\Gamma$ is bridgefree, we can choose a chord set $T^*$ such that $e_i\in T^*$. Then the fundamental loop matrix $C_{T^*}$ of $\Gamma \setminus e_i$ can be obtained from that of $\Gamma$ by removing the column and the row corresponding to $e_i$.

For a bridgefree graph, the polynomial $\Uu(\Gamma)$ will be linear in each variable $x_i$. Therefore, we obtain for its derivative
\begin{align}
	\pd{\Uu}{x_i} = \sum_{\substack{T\in \mathscr T_1 \\ e_i\notin T}} \prod_{\substack{e\notin T \\ e \neq e_i}} x_e = \Uu (\Gamma\setminus e_i) \label{eq:UDeleted}
\end{align}
where the spanning trees of $\Gamma$ not containing $e_i$ are precisely those spanning trees of $\Gamma\setminus e_i$. When $\Gamma$ is a bridgefree graph, we can combine \cref{eq:UContracted} and \cref{eq:UDeleted} to 
\begin{align}
	\Uu (\Gamma) = \Uu (\contraction{\Gamma}{e_i} ) + x_i \, \Uu (\Gamma \setminus e_i) \label{eq:UContractedDeleted}
\end{align}
which is known as ``deletion-and-contraction'' relation.


\section{Parametric Feynman integrals} \label{sec:ParametricFeynmanIntegrals}



In this section we will present various integral representations of the Feynman integral. For this reason we will attach so-called \textit{Schwinger parameters} or \textit{Feynman parameters} $x_i$ to every edge $e_i$ of a Feynman graph $\Gamma$ and express the Feynman integral \cref{eq:FeynmanMomSp} as an integral over those parameters $\gls{Schwingers}$. Those parameter integrals will contain certain graph polynomials, which will allow us to use the many tools of algebraic geometry. Therefore, parametric representations are much more convenient than the momentum space representation \cref{eq:FeynmanMomSp} of the Feynman integral for the purpose of this thesis. Further, the parametric representations will enable us to continue the spacetime dimension $d$ to complex numbers, which will be used in the dimensional regularization (see \cref{sec:DimAnaReg}). In addition, parametric representations are also helpful when reducing tensorial Feynman integrals to scalar Feynman integrals, which we will demonstrate in the end of this section. The parametric Feynman integral will also be the starting point for the application of hypergeometric theory. 

For general Feynman graphs the parametric Feynman integrals were first considered by Chisholm \cite{ChisholmCalculationMatrixElements1952}. In this procedure there appear certain polynomials which can be considered as graph polynomials, which was found by Symanzik \cite{SymanzikDispersionRelationsVertex1958} and was later proven by Nakanishi \cite{NakanishiIntegralRepresentationsScattering1961}. We will refer to \cite{NakanishiGraphTheoryFeynman1971} for a comprehensive presentation of all these results and to \cite{BognerFeynmanGraphPolynomials2010} for a modern summary. \bigskip

There are several ways to deduce parametric representations from \cref{eq:FeynmanMomSp}. In any case we will use certain integral representations of the integrand of \cref{eq:FeynmanMomSp} to change from momentum to parameter space. The simplest way of this rewriting can be made with ``Schwinger's trick''. Alternatively, one can also use ``Feynman's trick'', which we included in the appendix (\cref{lem:FeynmanTrick}) for the sake of completeness.

\begin{lemma}[Schwinger's trick] \label{lem:SchwingersTrick}
	Let $D_1,\ldots,D_n$ be positive, real numbers $D_i>0$ and $\nu\in\mathbb C^n$ with $\Re(\nu_i)>0$. Then we have the following identity
	\begin{align}
		\frac{1}{\prod_{i=1}^n D_i^{\nu_i}} = \frac{1}{\Gamma(\nu)} \int_{\mathbb R^n_+} \dif x\, x^{\nu-1} e^{-\sum_{i=1}^n x_i D_i} \point \label{eq:SchwingersTrick}
	\end{align}
	As before, we use a multi-index notation to keep formulas short. This implies $\Gamma(\nu):= \prod_{i=1}^n \Gamma(\nu_i)$ and $\dif x\,x^{\nu-1} := \prod_{i=1}^n \dif x_i\,x_i^{\nu_i-1}$.
\end{lemma}
\begin{proof}
	The lemma follows directly from the integral representation of the $\Gamma$-function $\Gamma(\nu_i) = \int_0^\infty \dif x_i \, x_i^{\nu_i-1} e^{-x_i}$ by substituting $x_i \mapsto D_i x_i$.
\end{proof}

After applying Schwinger's trick to the momentum space representation \cref{eq:FeynmanMomSp}, there will be certain Gaussian integrals left.
\begin{lemma}[Gaussian integrals] \label{lem:GaussianIntegrals}
	Let $M$ be a real, positive-definite  $(L\times L)$-matrix. Then
	\begin{align}
		\int_{\mathbb R^{d\times L}} \left(\prod_{i=1}^L \frac{\ddif k_j}{\pi^{d/2}} \right) e^{-k^\top\! M k - 2 Q^\top k} = e^{Q^\top\! M^{-1} Q} \det(M)^{-\dhalf}
	\end{align}
	holds for any $Q\in\mathbb R^L$.
\end{lemma}
\begin{proof}
	As considered $M$ is a non-singular matrix, and we can substitute $k\mapsto k - M^{-1} Q$ to remove the linear part, which results in
	\begin{align}
		e^{Q^\top\! M^{-1} Q} \int_{\mathbb R^{d\times L}} \left(\prod_{i=1}^L \frac{\ddif k_j}{\pi^{d/2}} \right) e^{-k^\top\! M k} \point
	\end{align}
	Furthermore, we can diagonalize the matrix $M=S R S^\top$, where $R=\diag(r_1,\ldots,r_L)$ consisting in the eigenvalues of $M$ and $S^{-1}=S^\top$ is orthogonal. Substituting $l = S^\top\! k$ we obtain
	\begin{align}
		e^{Q^\top\! M^{-1} Q} \int_{\mathbb R^{d\times L}} \left(\prod_{i=1}^L \frac{\ddif l_j}{\pi^{d/2}} \right) e^{- l^\top\! R l } = e^{Q^\top\! M^{-1} Q} \prod_{i=1}^L r_i^{-\dhalf} = e^{Q^\top\! M^{-1} Q} \det(M)^{-\dhalf}
	\end{align}
	where the integration with respect to $l$ factorizes in $L$ Gaussian integrals $\int_{\mathbb R^d} \frac{\ddif l_i}{\pi^{d/2}} e^{-r_i l_i^2} = r_i^{-d/2}$.
\end{proof}

Combining \cref{lem:SchwingersTrick} and \cref{lem:GaussianIntegrals} we can give the first parametric integral representation of Feynman integrals of this section. As aforementioned, we attach a momentum $q_i$\glsadd{internalmomenta}\glsunset{internalmomenta} to every edge $e_i$. Since we improve momentum conservation at every vertex, $q_i$ will consist of a linear combination of external momenta $p_1,\ldots,p_m$ and a linear combination of indeterminated loop momenta $k_1,\ldots,k_L$. A possible choice can be made by a fundamental loop matrix $C_{T^*}$ for any chord set $T^*$ as
\begin{align}
	q_i = \sum_{j=1}^L k_j \!\left(C_{T^*}\right)_{ji} + \widehat q_i \label{eq:qjhatqj}
\end{align}
where we denote by $\widehat q_i$\glsadd{internalmomentaHat}\glsunset{internalmomentaHat} the part coming from external momenta.

\begin{theorem}[Schwinger representation] \label{thm:SchwingerRepresentation}
	Let $\Gamma$ be a (1PI) Feynman graph and $C_{T^*}$ its fundamental loop matrix with respect to any chord set $T^*$. Further, $X=\diag(x_1,\ldots,x_n)$ collects the Schwinger parameters and $\widehat q=(\widehat q_1,\ldots,\widehat q_n)$ represent the amount of external momenta on every edge $e_1,\ldots,e_n$ according to \cref{eq:qjhatqj}. When the Feynman integral $\mathcal I_\Gamma (d,\nu, p, m)$ from \cref{eq:FeynmanMomSp} converges absolutely, it can be rewritten as
	\begin{align}
		\mathcal I_\Gamma(d,\nu,p,m) = \frac{1}{\Gamma(\nu)} \int_{\mathbb R^n_+} \dif x\, x^{\nu-1} \Uu^{-\dhalf} e^{-\frac{\Ff}{\Uu}} \comma \label{eq:SchwingerRepresentation}
	\end{align}
	where we will assume $\Re(\nu_i) >0$. Thereby, the polynomials $\Uu\in\mathbb R [x_1,\ldots,x_n]$ and $\Ff\in\mathbb R [x_1,\ldots,x_n]$ are given by
	\begin{align}
		\gls{Uu} = \det (M) \qquad \text{and} \qquad \gls{Ff} = \det (M) (J-Q^\top M^{-1} Q) \label{eq:SymanzikPolynomialsMatrices}
	\end{align}
	with $\gls{M}=C_{T^*} X C_{T^*}^\top$, $\gls{Q}=C_{T^*}X\widehat q$ and $\gls{J}=\sum_{i=1}^n x_i \!\left(\widehat q_i^2 + m_i^2\right)$.
\end{theorem}
\begin{proof}
	Starting with \cref{eq:FeynmanParSpFeynman} and using Schwinger's trick (\cref{lem:SchwingersTrick}) we arrive at
	\begin{align}
		\mathcal I_\Gamma(d,\nu,p,m) = \frac{1}{\Gamma(\nu)} \int_{\mathbb R^n_+} \dif x\, x^{\nu-1} \int_{\mathbb R^{d\times L}} \left(\prod_{j=1}^L \frac{\ddif k_j}{\pi^{d/2}}\right) e^{-\sum_{i=1}^n x_i D_i} \label{eq:SchwingerRepresentationProof1}
	\end{align}
	with the inverse propagators\glsadd{propagators}\glsunset{propagators} $D_i = q_i^2 + m_i^2$. Sorting the exponent in \cref{eq:SchwingerRepresentationProof1} in terms being quadratic, linear, and constant in $k$ we obtain by means of \cref{eq:qjhatqj}
	\begin{align}
        \gls{Lambda} :=\sum_{i=1}^n x_i D_i &= k^\top\! C_{T^*} X C_{T^*}^\top\! k + 2 k^\top\! C_{T^*} X \widehat q + \sum_{i=1}^n x_i \!\left(\widehat q_i^2 + m_i^2\right) \nonumber \\
        & = k^\top\! M k + 2 Q^\top\! k + J \point \label{eq:LambdaxDrelation}
	\end{align}
	By construction $M$ is symmetric and its $i$-th leading principal minor is given by $\Uu(\Gamma\setminus e_i) = \pd{\Uu}{x_i}$ due to the considerations in \cref{eq:UDeleted}. Thus, $M$ is positive-definite inside the integration region $x\in (0,\infty)^n$, and we can apply \cref{lem:GaussianIntegrals} to obtain \cref{eq:SchwingerRepresentation}.	
\end{proof}

\begin{theorem}[Symanzik polynomials] \label{thm:SymanzikPolynomialsTreeForest}
	The polynomials appearing in \cref{thm:SchwingerRepresentation} can be rewritten as
	\begin{align}
		\gls{Uu} &= \sum_{T\in\mathscr T_1} \prod_{e\notin T} x_e \label{eq:FirstSymanzik} \\
		\gls{Ff} &= \sum_{F\in\mathscr T_2} s_F \prod_{e\notin F} x_e + \Uu (x) \sum_{i=1}^n x_i m_i^2 \label{eq:SecondSymanzik}
	\end{align}
	where $s_F = \left(\sum_{e_i\notin F} \pm q_i\right)^2 = \left(\sum_{e_i\notin F} \pm \widehat q_i\right)^2$ is the squared sum of signed momenta of the cutted edges. Denote by $F_1$ and $F_2$ the two components of the spanning $2$-forest $F$. Then we choose the signs in such a way, that the momenta on $e$ with $e\notin F$ flow from $F_1$ to $F_2$. Thus, $s_F$ is the squared total momentum flowing from $F_1$ to $F_2$. By momentum conservation this is equivalent to $s_F = \left(\sum_{a\in V_{F_1}} p_a\right)^2$ the sum of external momenta flowing into $F_1$ (or equivalently $F_2$). The polynomials $\Uu (x)$ and $\Ff (x)$ are known as the \emph{first} and the \emph{second Symanzik polynomial}.
\end{theorem}
\begin{proof}
	The representation \cref{eq:FirstSymanzik} was already shown in \cref{thm:MatrixTree}. For the representation of \cref{eq:SecondSymanzik} the proof is oriented towards \cite{NakanishiGraphTheoryFeynman1971}. Note, that the massless part of \cref{eq:SecondSymanzik} as well as the massless part of the representation in \cref{eq:SymanzikPolynomialsMatrices} is quadratic in the external momenta. Thus, we will write $\sum_{a,b} W_{ab} p_a p_b$ for the massless part of \cref{eq:SymanzikPolynomialsMatrices} and $\sum_{a,b} \widetilde W_{ab} p_a p_b$ for the massless part of \cref{eq:SecondSymanzik}. Hence, we want to compare the coefficients $W_{ab}$ and $\widetilde W_{ab}$. For this reason we can choose arbitrary external momentum configurations, e.g.\ we can set all but two external momenta to zero. Therefore, it is sufficient to show equality in the case $p_a = - p_b$.
	
	Let us add a new edge $e_0$ linking the external vertices\footnote{Without loss of generality we can assume that $v_a\neq v_b$, since the case $v_a=v_b$ corresponds to a ``tadpole-like'' graph, where the massless part of \cref{eq:SymanzikPolynomialsMatrices} and \cref{eq:SecondSymanzik} vanishes identically.} $v_a$ and $v_b$ (in this direction). We will call the resulting graph $\Gamma_0 = \Gamma \cup \{e_0\}$. To construct a fundamental loop matrix of $\Gamma_0$ we will select a chord set $T_0^*=T^*\cup\{e_0\}$, where $T^*$ is a chord set of $\Gamma$. The loop $\mathscr C_0$ will then consist of $e_0$ and a path in $\Gamma$ connecting the vertices $v_a$ and $v_b$. Therefore, we will have
	\begin{align}
		\widehat q_i = \left(C_{T_0^*}\right)_{0,i} p_b \point \label{eq:hatqi}
	\end{align}
	That implies $Q_i = \sum_{j=1}^n \!\left(C_{T^*}\right)_{ij} x_j \widehat q_j = p_b M^0_{0,i}$ where $M^0 = C_{T_0^*} X C_{T_0^*}^\top$ is the corresponding matrix for the graph $\Gamma_0$.
	
	On the other hand, we can write the massless part of \cref{eq:SymanzikPolynomialsMatrices} as\glsadd{Ff0}\glsunset{Ff0} $\Ff_0 = - Q^\top\! \Adj(M) Q + \det(M) \sum_{i=1}^n x_i \widehat q_i^2 = \det (V)$, where
	\begin{align}
		V = \begin{pmatrix}
		    	\sum_{i=1}^n x_i \widehat q_i^2 & Q \\
		    	Q^\top & M
		    \end{pmatrix}
	\end{align}
	by using Laplace expansion. Inserting \cref{eq:hatqi} this is nothing else than $V = \left. p_b^2 M^{0}\right|_{x_0 = 0}$. Using \cref{thm:MatrixTree} again we obtain
	\begin{align}
		\det (V) = \left. p_b^2 \det(M^{0})\right|_{x_0 = 0} = p_b^2 \sum_{\substack{T\in\mathscr T_1(\Gamma_0) \\ e_0 \notin T}} \prod_{e\notin T} x_e = p_b^2 \sum_{F\in \mathscr T_2(\Gamma)} \prod_{e\notin F} x_e
	\end{align}
	where every spanning tree of $\Gamma_0$ which does not contain $e_0$, will split the graph $\Gamma$ into two components. This is nothing else than a spanning $2$-forest of $\Gamma$.
\end{proof}

As obviously by \cref{eq:FirstSymanzik} and \cref{eq:SecondSymanzik} the Symanzik polynomials are homogeneous of degree $L$ and $L+1$, respectively. In addition to it the first Symanzik polynomial $\Uu$ and the massless part of the second Symanzik polynomial $\Ff_0$ are linear in each Schwinger parameter. The complete second Symanzik polynomial $\Ff$ is at most quadratic in every Schwinger parameter. Beside the representations in \cref{eq:FirstSymanzik}, \cref{eq:SecondSymanzik} and \cref{eq:SymanzikPolynomialsMatrices}, there are many alternatives known to write Symanzik polynomials. We refer to \cite{BognerFeynmanGraphPolynomials2010, NakanishiGraphTheoryFeynman1971} for an overview.

Alternatively, we can consider the second Symanzik polynomial also as the discriminant of $\Lambda(k,p,x)\, \Uu(x)$ with respect to $k$, where $\Lambda(k,p,x)$ was defined in \cref{eq:LambdaxDrelation}. Thus, we will obtain the second Symanzik polynomial by eliminating $k$ in $\Lambda(k,p,x)\, \Uu(x)$ by means of the equation $\frac{\partial \Lambda}{\partial k_j} = 0$, which was noted in \cite{EdenAnalyticSmatrix1966}.

As we can reduce $\Ff_0 := \sum_{F\in\mathscr T_2} s_F \prod_{e\notin F} x_e$ to a determinant of a matrix $C_{T_0^*} X C_{T_0^*}^\top$ the formulas \cref{eq:UContracted}, \cref{eq:UDeleted} and \cref{eq:UContractedDeleted} do also apply to $\Ff_0$.\bigskip

We already restricted us to ($1$PI) graphs (i.e.\ bridgefree graphs) in \cref{thm:SchwingerRepresentation}. Another specific case happens when the graph $\Gamma$ contains a cut vertex, i.e.\ if $\Gamma$ is a $1$-vertex reducible graph. Whenever a Feynman graph contains those cut vertices, the Feynman integral will factorize into the Feynman integrals of the components which are held together by the cut vertices. This can be seen easily by the momentum space representation \cref{eq:FeynmanMomSp} as we can choose the loop momenta separately for all those components. If additionally one of these components is a vacuum graph (i.e.\ it contains no legs), the value of this component is independent of the external momenta, and we call those graphs tadpole-like. We already excluded this kind of graphs to appear in perturbation series in the beginning of \cref{sec:FeynmanIntegralsIntro}. If such a component not only has no legs but also no masses, it will be independent of any variable and is scaleless. Consequentially, the second Symanzik polynomial of every scaleless graph or subgraph will vanish $\Ff\equiv 0$. Hence, Feynman integrals containing a scaleless subgraph connected only by a cut vertex will diverge (in any representation and for any choice of $d$ and $\nu$). In \cref{sec:DimAnaReg} we will derive a similar result. Furthermore, we will show, that this is the only case, where the Feynman integral diverges for any choice of $d$ and $\nu$. The problem of divergences of scaleless graphs can be ``cured'' also in renormalization procedures. By a convenient choice of the counter term, the renormalization procedure will map any tadpole-like graph to zero \cite{SchwartzQuantumFieldTheory2014}.\bigskip

From \cref{thm:SchwingerRepresentation} we can deduce another representation of the Feynman integral in parameter space.

\begin{lemma}[Feynman representation] \label{lem:FeynmanRepresentation}
	If the Feynman integral \cref{eq:SchwingerRepresentation} converges absolutely, we can rewrite it as the following integral 
	\begin{align}
		\mathcal I_\Gamma(d,\nu,p,m) = \frac{\Gamma(\omega)}{\Gamma(\nu)} \int_{\mathbb R^n_+} \dif x\, x^{\nu-1} \delta\big(1-H(x)\big) \frac{\Uu^{\omega-\dhalf}}{\Ff^\omega} \label{eq:FeynmanParSpFeynman}
	\end{align}
	where $\gls{sdd} = \sum_{i=1}^n \nu_i - L\dhalf$ is called the \textit{superficial degree of divergence} and $H(x) = \sum_{i=1}^n h_i x_i$ is an arbitrary hyperplane with $h_i\geq 0$ not all zero. $\delta(x)$ denotes Dirac's $\delta$-distribution, and we will assume $\Re (\omega)>0$ and $\Re(\nu_i)>0$.
\end{lemma}
\begin{proof}
	We will insert $1 = \int_0^\infty \dif t \, \delta(t-H(x))$ in the integral \cref{eq:SchwingerRepresentation}, which results in
	\begin{align}
		\mathcal I_\Gamma(d,\nu,p,m) = \frac{1}{\Gamma(\nu)} \int_0^\infty \dif t \int_{\mathbb R^n_+} \dif x\, x^{\nu-1} \frac{1}{t} \delta\left(1-\sum_{i=1}^n h_i \frac{x_i}{t} \right) \Uu^{-\dhalf} e^{-\frac{\Ff}{\Uu}} \point
	\end{align}
	Due to the homogenity of $\Uu$ and $\Ff$, a substitution $x_i \mapsto t x_i$ leads to
	\begin{align}
		\mathcal I_\Gamma(d,\nu,p,m) = \frac{1}{\Gamma(\nu)} \int_0^\infty \dif t \, t^{\omega-1} \int_{\mathbb R^n_+} \dif x\, x^{\nu-1} \delta\big(1-H(x)\big) \Uu^{-\dhalf} e^{- t \frac{\Ff}{\Uu}} \point
	\end{align}
	The integration with respect to $t$ can be performed by \cref{eq:SchwingersTrick}.
\end{proof}

The freedom of the choice of a hyperplane $H(x)$ in \cref{eq:FeynmanParSpFeynman} is sometimes referred as Cheng-Wu theorem and expresses the projective nature of the integral due to the homogenity of the Symanzik polynomials. We can also use this freedom to exclude one of the Symanzik polynomials in the integrand, as done for the first Symanzik polynomial in the following corollary. However, note that the evaluation of the $\delta$-distribution will produce terms of $\Uu(\Gamma \setminus e_i)$ and $\Uu(\contraction{\Gamma}{e_i})$.

\begin{cor}
	In case of absolute convergence, the Feynman integral can also be expressed as
	\begin{align}
		\mathcal I_\Gamma(d,\nu,p,m) = L \frac{\Gamma(\omega)}{\Gamma(\nu)} \int_{\mathbb R^n_+}  \dif x \, x^{\nu-1} \delta\big(1-\Uu(x)\big) \Ff^{-\omega} \label{eq:FeynmanRepresentationCorollary} \point
	\end{align}
\end{cor}
\begin{proof}
	Similar to the proof of \cref{lem:FeynmanRepresentation} we will insert $1 = \int_0^\infty \dif t \, \delta\big(t - \Uu(x)\big)$ for $x\in (0,\infty)^n$ in the representation \cref{eq:SchwingerRepresentation}. The remaining steps are completely analogue to the proof of \cref{lem:FeynmanRepresentation} where we will substitute $x_i\mapsto t^{1/L} x_i$. The relation \cref{eq:UContractedDeleted} can be used to evaluate the integration of the $\delta$-distribution.
\end{proof}

Clearly, those integrals converge not for every value of $\nu\in\mathbb C^n$ and $d\in\mathbb C$. Moreover,  \cref{eq:SchwingerRepresentation} and \cref{eq:FeynmanParSpFeynman} will have in general different convergence regions. However, it is only important that their convergence regions intersect with each other (which was shown implicitly by the proof of \cref{lem:FeynmanRepresentation}). Afterwards we will meromorphically continue \cref{eq:SchwingerRepresentation} and \cref{eq:FeynmanParSpFeynman} to the whole complex plane. The convergence as well as the meromorphic continuation will be discussed in more detail in the next \cref{sec:DimAnaReg}.\bigskip

The previous parametric representations contained two polynomials. With the following parametric Feynman integral we will give a representation, which contains only one polynomial. This representation was first noted in \cite{LeeCriticalPointsNumber2013} and will be very convenient when applied to hypergeometric theory. The following representation was also found independently in \cite{NasrollahpoursamamiPeriodsFeynmanDiagrams2016}.

\begin{lemma}[Lee-Pomeransky representation \cite{LeeCriticalPointsNumber2013,BitounFeynmanIntegralRelations2019}] \label{lem:LeePomeranskyRepresentation}
	In case of absolute convergence we can rewrite the Feynman integral as
	\begin{align}
		\mathcal I_\Gamma(d,\nu,p,m) = \frac{\Gamma\left(\dhalf\right)}{\Gamma\left(\dhalf-\omega\right)\Gamma(\nu)} \int_{\mathbb R^n_+} \dif x\, x^{\nu-1} \Gg^{-\dhalf} \label{eq:LeePomeranskyRepresentation}
	\end{align}
	where $\Gg = \Uu + \Ff$\glsadd{Gg}\glsunset{Gg} denotes the sum of the first and the second Symanzik polynomial. We will assume\footnote{Those restrictions can later be relaxed by meromorphic continuation, see \cref{sec:DimAnaReg}.} that $\Re\!\left(\dhalf\right)>0$, $\Re\!(\nu_i)>0$ and $\Re\left(\dhalf-\omega\right)>0$.
\end{lemma}
\begin{proof}
	Note that for every real, positive $D>0$ we have 
	\begin{align}
		D^{-\omega} = \frac{\Gamma(\alpha)}{\Gamma(\alpha-\omega)\Gamma(\omega)} \int_0^\infty \dif t \, t^{\omega-1} (1 + D t)^{-\alpha} \label{eq:LeePomTrick}
	\end{align}
	from the Euler beta function, where we also assume $\Re(\alpha)>0$, $\Re(\omega)>0$ and $\Re(\alpha-\omega)>0$. Applying \cref{eq:LeePomTrick} with $D=\frac{\Ff}{\Uu}$ and $\alpha=\dhalf$ to the representation \cref{eq:FeynmanParSpFeynman} we obtain
	\begin{align}
		\mathcal I_\Gamma(d,\nu,p,m) = \frac{\Gamma\left(\dhalf\right)}{\Gamma\left(\dhalf-\omega\right)\Gamma(\nu)} \int_0^\infty \dif t\,t^{\omega-1} \int_{\mathbb R^n_+} \dif x \, x^{\nu-1} \delta\big(1-H(x)\big) (\Uu + t \Ff)^{-\dhalf} \point
	\end{align}
	Again, by a substitution $x_i = \frac{y_i}{t}$ we will arrive at the assertion.
\end{proof}

Hence, the Feynman integral can be expressed as the Mellin transform of a polynomial up to a certain power. For this reason, integrals of the form \cref{eq:LeePomeranskyRepresentation} are often called \textit{Euler-Mellin integrals}, which were systematically investigated in \cite{NilssonMellinTransformsMultivariate2010, BerkeschEulerMellinIntegrals2013}. 

On a more abstract level, \cref{lem:LeePomeranskyRepresentation} can be considered as an application of the ``Cayley trick'' (\cref{lem:CayleysTrick}) for integrals. Indeed, we will see in \cref{sec:FeynmanIntegralsAsAHyp} that \cref{eq:LeePomeranskyRepresentation} arise from the Cayley embedding of \cref{eq:FeynmanParSpFeynman}.\bigskip

Instead of the external momenta $p$ and the masses $m$, the parametric representations indicate another choice of what we want to use as the variables of the Feynman integral. Thus, it is much more convenient to use the coefficients of the Symanzik polynomials as the variables of Feynman integrals. Hence, we will write\glsadd{Gg}
\begin{align}
    \Gg = \sum_{a\in A} z_a x^a = \sum_{j=1}^N z_j \prod_{i=1}^n x_i^{a^{(j)}_i} \in\mathbb C[x_1,\ldots,x_n]\label{eq:Gsupport}
\end{align}
where $A=\left\{a^{(1)},\ldots,a^{(N)}\right\}\subset\mathbb Z_{\geq 0}^n$ is the set of exponents and $z\in\mathbb C^N$ are the new variables of the Feynman integral. To avoid redundancy we will always assume that $z_j\not\equiv 0$ and that all elements of $A$ are pairwise disjoint. Therefore, we will change our notation of the Feynman integral to $\mathcal I_\Gamma(d,\nu,z)$.

In \cref{eq:Gsupport} we have introduced in fact a generalization of Feynman integrals by way of the back door. Equation \cref{eq:Gsupport} gives also coefficients to the first Symanzik polynomial, and it is also implicitly assumed that the coefficients in the second Symanzik polynomial are independent of each other. We will call such an extension of the Feynman integral a \textit{generalized Feynman integral}. By specifying the variables $z$ to the physically relevant cases, we can always lead back to the original case. However, it should not go unmentioned, that such a limit from generalized Feynman integrals to physical Feynman integrals will not always be unproblematic. We will discuss this problem more in detail in the following sections. Especially, it is not guaranteed that specific representations of generalized Feynman integrals e.g.\ in terms of power series or integrals converge also for the physically relevant specification of the variables $z$, albeit the analytic continuation of those power series or integral representations will still lead to finite values.\bigskip

As before, we will write $\Aa = \left\{\left(1,a^{(1)}\right),\ldots,\left(1,a^{(N)}\right)\right\}$ for the homogenized point configuration of $A$ (see \cref{sec:affineSpace} and \cref{ssec:vectorConfigurations}). In addition, we will also transfer the parameters $d$ and $\nu$ to the vector space $\mathbb C^{n+1}$. Therefore, we will write 
\begin{align}
	\gls{nuu} := (\nu_0, \nu)\in \mathbb C^{n+1} \qquad\text{with}\quad \nu_0 := \dhalf \point \label{eq:nuuDefinition}
\end{align}
As we will see later, the spacetime dimension $d$ and the indices $\nu$ will take a similar role in parametric Feynman integrals. Therefore, it is meaningful from a mathematical perspective to combine those parameters into a single parameter $\nuu$. Their similar mathematical role is also the reason why instead of using dimensional regularization one can regularize the integral by considering the indices as done in analytic regularization \cite{SpeerGeneralizedFeynmanAmplitudes1969}. We will develop this consideration more in detail in the following two \cref{sec:DimAnaReg,sec:FeynmanIntegralsAsAHyp}. 

Hence, the (generalized) Feynman integral becomes finally to $\gls{gFI}$, where the essential part of the graph structure, which is necessary to evaluate the Feynman integral is contained in the vector configuration $\Aa\in\mathbb Z^{(n+1)\times N}$. The $\nuu\in\mathbb C^n$ are the parameters and $z\in\mathbb R^N$ are the variables, encoding the dependencies from external momenta and masses according to \cref{eq:Gsupport}. \bigskip

It is a genuine aspect of hypergeometric functions to be also representable by Mellin-Barnes integrals. By use of a multivariate version of Mellin's inversion theorem, we will derive such a representation of the Feynman integral.

\begin{theorem}[Mellin-Barnes representation] \label{thm:MellinBarnesRepresentation}
    Let $\sigma\subset \{1,\ldots,N\}$ be an index subset with cardinality $n+1$, such that the matrix $\Aa$ restricted to columns of $\sigma$ is invertible, $\det(\Aas) \neq 0$. Then the Feynman integral can be written as a multi-dimensional Mellin-Barnes integral
    \begin{align}
        \mathcal I_\Aa (\nuu, z) =  \frac{1}{\Gamma(\nu_0-\omega)\Gamma(\nu)} \frac{ z_\sigma^{-\Aas^{-1}\nuu}}{\left|\det(\Aas)\right|} \int_\gamma \frac{\dif t}{(2\pi i)^r} \Gamma( t) \Gamma\!\left(\Aas^{-1}\nuu-\Aas^{-1}\Aabs t\right)  z_{\bar\sigma}^{- t}  z_\sigma^{\Aas^{-1}\Aabs  t} \label{eq:MellinBarnesRepresentation}
    \end{align}
    wherever this integral converges. The set $\bar\sigma:=\{1,\ldots N\} \setminus \sigma$ denotes the complement of $\sigma$, containing $r:=N-n-1$ elements. Restrictions of vectors and matrices to those index sets are similarly defined as $ z_\sigma := (z_i)_{i\in\sigma}$, $ z_{\bar\sigma} := (z_i)_{i\in\bar\sigma}$, $\Aabs := (a_i)_{i\in\bar\sigma}$. Every component of the integration contour $\gamma\in\mathbb C^r$ goes from $-i\infty$ to $i\infty$ such that the poles of the integrand are separated. 
\end{theorem}
\begin{cor} \label{cor:MellinBarnesRepresentation}
    Let $N=n+1$ or in other words let $\Aa$ be quadratic or equivalently $\Newt(\Gg)=\Conv(A)$ forming a simplex. If the Feynman integral $\mathcal I_\Aa(\nuu, z)$ from \cref{eq:LeePomeranskyRepresentation} converges absolutely, it can be expressed by a combination of $\Gamma$-functions
    \begin{align}
        \mathcal I_\Aa (\nuu, z) = \frac{1}{\Gamma(\nu_0-\omega)\Gamma(\nu)}  \frac{\Gamma(\Aa^{-1}\nuu)}{ \left|\det (\Aa)\right|}  z^{-\Aa^{-1}\nuu} \point \label{eq:MellinBarnesCorollary}
    \end{align}
\end{cor}  
\begin{proof}
    For simplicity, we will write $\mathcal I_\Aa(\nuu,z) = \frac{1}{\Gamma(\nu_0-\omega)\Gamma(\nu)} \gls{gFJ}$. Then by Schwinger's trick (\cref{lem:SchwingersTrick}) we can reformulate \cref{eq:LeePomeranskyRepresentation} as
    \begin{align}
        \mathcal J_\Aa(\nuu, z) = \int_{\mathbb{R}^{n+1}_+} \dif x_0 \, x_0^{\nu_0-1} \dif x \,  x^{\nu-1} e^{-x_0 \Gg} \comma
    \end{align}
    where $\nu_0=\dhalf$ was defined in \cref{eq:nuuDefinition}. Writing $\underline{x} = (x_0, x)$, $\nuu = (\nu_0,\nu)$ and using the Cahen-Mellin integral representation of exponential function one obtains
    \begin{align}
        \mathcal J_\Aa(\nuu, z) = \int_{\mathbb{R}^{n+1}_+} \dif \underline{x} \, \underline{x}^{\nuu-1} \int_{\delta+i\mathbb R^{n+1}} \frac{\dif u}{(2\pi i)^{n+1}} \Gamma(u) z_\sigma^{-u} \underline{x}^{-\Aas u} \int_{\eta+i\mathbb R^r} \frac{\dif t}{(2\pi i)^r} \Gamma(t) z_{\bar\sigma}^{-t} \underline{x}^{-\Aabs t} 
    \end{align}
    with $u\in\mathbb{C}^{n+1}$, $t\in\mathbb{C}^r$, some arbitrary positive numbers $\delta_i >0$, $\eta_i>0$ and where we split the polynomial $\Gg$ into a $\sigma$ and a $\bar\sigma$ part. By a substitution $u\mapsto \Aas^{-1} u^\prime$ it is
    \begin{align}
        \mathcal J_\Aa(\nuu,z) = \left|\det\!\left(\Aas^{-1}\right)\right| &\int_{\eta+i\mathbb R^{r}} \frac{\dif t}{(2\pi i)^r} \Gamma(t) z_{\bar\sigma}^{-t} \nonumber \\ 
        & \: \int_{\mathbb{R}^{n+1}_+}\! \dif \underline{x} \, \int_{\Aas \delta+i \Aas \mathbb R^{n+1}} \frac{\dif u^\prime}{(2\pi i)^{n+1}}  \Gamma\!\left(\Aas^{-1}u^\prime\right) z_\sigma^{-\Aas^{-1}u^\prime} \underline{x}^{\nuu-u^\prime -\Aabs t- 1} \:\text{.}
    \end{align}
    Since the matrix $\Aas$ contains only positive values, the integration region remains the same $\Aas \delta + i \Aas \mathbb R^{n+1} \simeq \delta^\prime+i\mathbb R^{n+1}$ with some other positive numbers $\delta^\prime\in\mathbb R^{n+1}_{>0} $, which additionally have to satisfy $\Aas^{-1}\delta^\prime >0$. By Mellin's inversion theorem \cite{AntipovaInversionMultidimensionalMellin2007} only the $t$-integrations remain and one obtains equation \cref{eq:MellinBarnesRepresentation}.

    Thereby, the integration contour has to be chosen, such that the poles are separated from each other in order to satisfy $\Aas^{-1}\delta^\prime>0$. More specific this means that the contour $\gamma$ has the form $c+i\mathbb R^{n+1}$ where $c\in\mathbb R^{n+1}_{>0}$ satisfies $\Aas^{-1}\nuu - \Aas^{-1}\Aabs c>0$. Clearly, in order for those $c$ to exist, the possible values of parameters $\nuu$ are restricted.

    The proof of the corollary is a special case, where one does not have to introduce the integrals over $t$. The existence of the inverse $\Aa^{-1}$ will be ensured by \cref{thm:FIconvergence}.
\end{proof}

A more general version of this theorem can be found in \cite[thm. 5.6]{BerkeschEulerMellinIntegrals2013} with an independent proof. In \cite{SymanzikCalculationsConformalInvariant1972} a similar technique was used to obtain Mellin-Barnes representations from Feynman integrals. It should not go unmentioned, that the precise contours of Mellin-Barnes integrals in higher dimensions contain many subtleties. We refer to \cite{ParisAsymptoticsMellinBarnesIntegrals2001} for general aspects in the study of Mellin-Barnes integrals. Furthermore, Mellin-Barnes integral are well suited to investigate the monodromy of $\Aa$-hypergeometric functions \cite{BeukersMonodromyAhypergeometricFunctions2013}.

According to \cite{HaiConvergenceProblemCertain1995}, integrals of the form \cref{eq:MellinBarnesRepresentation} are also known as multivariate Fox's $H$-functions, where also convergence criteria of those functions can be found. The connection between Feynman integrals and Fox's $H$-function was studied before \cite{Inayat-HussainNewPropertiesHypergeometric1987, Inayat-HussainNewPropertiesHypergeometric1987a, BuschmanFunctionAssociatedCertain1990}.

We have to remark, that the representation \cref{eq:MellinBarnesRepresentation} is not necessarily made for an efficient computation of Feynman integrals. There are much more involved methods to derive more specific Mellin-Barnes representations for certain types of Feynman integrals \cite{UsyukinaRepresentationThreepointFunction1975, BergereAsymptoticExpansionFeynman1974, SmirnovFeynmanIntegralCalculus2006}. The advantage of \cref{eq:MellinBarnesRepresentation} is rather the simplicity of its representation, especially the simplicity in the situation of the \cref{cor:MellinBarnesRepresentation}, which will be of great importance, when constructing series representations later.

Feynman integrals, which satisfy the conditions of \cref{cor:MellinBarnesRepresentation} are the so-called massless ``banana graphs'', i.e.\ graphs consisting in $L$ loops and having the minimal number $L+1$ of edges. However, we can apply \cref{cor:MellinBarnesRepresentation} also to any Euler-Mellin integral where there is exactly one monomial more than variables, i.e.\ where $\Newt(\Gg)$ forms a simplex. 

\Cref{cor:MellinBarnesRepresentation} can alternatively be shown using Ramanujan's master theorem. For this purpose, one splits $\Gg^{-\nu_0}$ by means of the multinomial theorem $(m_1 + \ldots + m_N)^s = \sum_{k_1=0}^\infty \cdots \sum_{k_{N-1}=0}^\infty \frac{(-1)^{|k|} (-s)_{|k|}}{k!} m_1^{k_1} \cdots m_{N-1}^{k_{N-1}} m_N^{s-|k|}$ and then applies a multivariate version of Ramanujan's master theorem \cite{GonzalezGeneralizedRamanujanMaster2011}.\bigskip

Let us illustrate the application of \cref{thm:MellinBarnesRepresentation} about Mellin-Barnes representations with a small example.    

\begin{example} \label{ex:1loopbubbleA}
    Consider the self-energy $1$-loop $2$-point function with one mass (see \cref{fig:bubble1}) having the Symanzik polynomials $\Uu=x_1 + x_2$ and $\Ff=(m_1^2+p^2) x_1 x_2 + m_1^2 x_1^2$ in Euclidean kinematics. Thus, the matrix $\Aa$ and the vector $z$ are given by
    \begin{align}
        \Aa = \begin{pmatrix} 
            1 & 1 & 1 & 1 \\
            1 & 0 & 1 & 2 \\
            0 & 1 & 1 & 0
        \end{pmatrix} \qquad z = (1,1,m_1^2+p^2,m_1^2) \point
    \end{align}
    
    \begin{figure}[bt] 
        \begin{center}
            \includegraphics[width=.38\textwidth, trim = 0 1.5cm 0 0.5cm,clip]{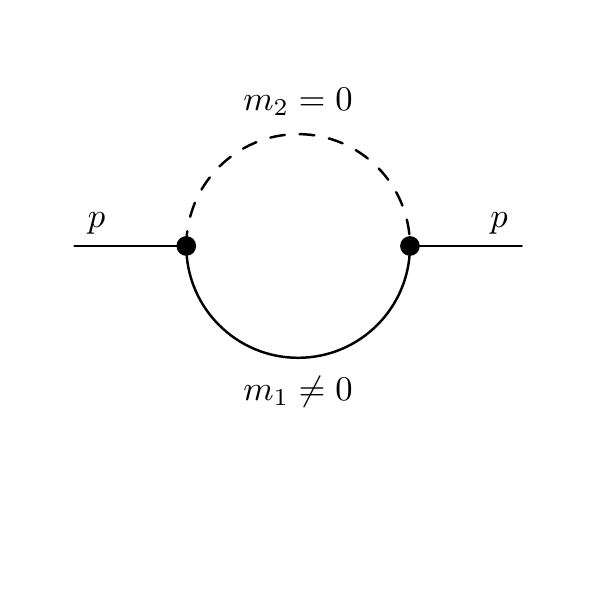}
        \end{center}
        \caption[The $1$-loop self-energy Feynman graph with one mass]{The $1$-loop self-energy Feynman graph with one mass (self-energy/bubble graph).} \label{fig:bubble1}
    \end{figure}
    
    Choosing an index set $\sigma = \{1,2,3\}$, the corresponding Feynman integral in the Mellin-Barnes representation of \cref{thm:MellinBarnesRepresentation} is given by
    \begin{align}
        \mathcal J_\Aa (\nuu,z) &= z_1^{-\nu_0+\nu_2} z_2^{-\nu_0+\nu_1} z_3^{\nu_0-\nu_1-\nu_2} \int_{\delta-i\infty}^{\delta+i\infty} \frac{\dif t}{2\pi i} \Gamma(t) \Gamma(\nu_0-\nu_1+t)\nonumber\\
        &\qquad \qquad \Gamma(\nu_0-\nu_2-t) \Gamma(-\nu_0+\nu_1+\nu_2-t) \left(\frac{z_2 z_4}{z_1 z_3}\right)^{-t} \comma
    \end{align}
    where we write again $\mathcal J_\Aa(\nuu,z) =  \Gamma(\nu_0-\omega)\Gamma(\nu) \mathcal I_\Aa(\nuu,z)$ to omit the prefactor. For the correct contour prescription the poles have to be separated such that there exist values $\delta$ satisfying $\max \{0,-\nu_0+\nu_1\} < \delta < \min \{\nu_0-\nu_2,-\nu_0+\nu_1+\nu_2\}$. Therefore, we can extract $4$ conditions on the values of $\nuu$. As we see in this example, these conditions are equivalent to demand $\delta\in\relint\!\left(\Re(\nu_0)\Newt(\Gg)\right)$, where $\Re(\nu_0)\Newt(\Gg)$ is the Newton polytope of $\Gg$ dilated by $\Re(\nu_0)$ (see \cref{sec:ConvexPolytopes}). Hence, $\Newt(\Gg)$ should be full dimensional in order to allow values for $\delta$. As we will see in \cref{sec:DimAnaReg} the full dimensionalness of $\Newt(\Gg)$ relates directly to the convergence of $\mathcal I_\Aa(\nuu,z)$. In this case, by Cauchy's theorem the integral evaluates simply to a Gaussian hypergeometric function
    \begin{align}
        \mathcal J_\Aa (\nuu,z) &=  \frac{\Gamma(2\nu_0-\nu_1-\nu_2)\Gamma(\nu_2)\Gamma(\nu_0-\nu_2)\Gamma(-\nu_0+\nu_1+\nu_2)}{\Gamma(\nu_0)} \nonumber\\
        & z_1^{-\nu_0+\nu_2} z_2^{-\nu_0+\nu_1} z_3^{\nu_0-\nu_1-\nu_2} \HypF{\nu_0-\nu_2,-\nu_0+\nu_1+\nu_2}{\nu_0}{1-\frac{z_2 z_4}{z_1z_3}} \point
    \end{align}
    After restoring the original prefactors and coefficients $z_1=z_2=1$, $z_3=p^2-m_1^2$, $z_4=m_1^2$ and $\nu_0=\frac{d}{2}$ it agrees with the expected result
    \begin{align}
        \mathcal I_\Gamma (\nu_1,\nu_2,d,m_1^2,p^2) &= \frac{\Gamma\left(\frac{d}{2}-\nu_2\right)\Gamma\left(-\frac{d}{2}+\nu_1+\nu_2\right)}{\Gamma(\nu_1)\Gamma\left(\frac{d}{2}\right)} \nonumber\\
        & (p^2-m_1^2)^{\frac{d}{2}-\nu_1-\nu_2} \HypF{\frac{d}{2}-\nu_2,-\frac{d}{2}+\nu_1+\nu_2}{\frac{d}{2}}{1-\frac{m_1^2}{p^2-m_1^2}} \point
    \end{align}
\end{example}
\bigskip

As was mentioned in the beginning of this chapter, scalar Feynman integrals arise only in the simplest possible QFTs and in general Feynman integrals will carry an additional Lorentz-tensorial structure. More precisely, in a more general QFT we also have to expect terms in the numerator of \cref{eq:FeynmanMomSp}, i.e.\ Feynman integrals appear in the form
\begin{align}
	\mathcal I_\Gamma (s,d,\nu,p,m) = \int_{\mathbb R^{d\times L}} \left(\prod_{j=1}^L \frac{\ddif k_j}{\pi^{d/2}} \right) \frac{q_1^{s_1} \cdots q_n^{s_n}}{D_1^{\nu_1} \cdots D_n^{\nu_n}} \label{eq:FeynmanMomSpTensor}
\end{align}
where $s_i\in\mathbb N_0^d$, $q_i$ is the momentum of the edge $e_i$, $D_i=q_i^2+m_i^2$ are the inverse propagators and we will also use a multi-index notation $q_i^{s_i} = \prod_{\mu=1}^d \!\left(\!\left(q_i\right)^\mu\right)^{(s_i)^\mu}$ for Lorentzian vectors. We deviate here from the common representation of Lorentz tensors by small Greek letters as indices, in order that the symbols do not become overloaded by the many indices. Thus, for a fixed value of $s\in\mathbb N_0^{d\times n}$, equation \cref{eq:FeynmanMomSpTensor} represents an element of a Lorentz tensor. Moreover, equation \cref{eq:FeynmanMomSpTensor} often shows up restricted to the loop momenta $k_1,\ldots,k_L$ in the numerator. Note that \cref{eq:FeynmanMomSpTensor} includes also this case. By a choice of a chord set $T^*$, we can assign the momenta $q_i=k_i$ for edges $e_i\in T^*$. The remaining momenta can be excluded by setting $s_i=0$ whenever $e_i\notin T^*$.

However, one can always reduce those integrals to a linear combination of scalar Feynman integrals. For $1$-loop integrals such a reduction is known for a long time \cite{PassarinoOneloopCorrectionsAnnihilation1979}, which wrotes the tensorial Feynman integral as a linear combination of the physically relevant tensors and find the coefficients by considering certain special cases. Another idea to reduce tensorial Feynman integrals by certain differential operators was developed in \cite{DavydychevSimpleFormulaReducing1991} for $1$-loop integrals and extended to higher loops in \cite{TarasovConnectionFeynmanIntegrals1996, TarasovGeneralizedRecurrenceRelations1997}. Thus, we can consider \cref{eq:FeynmanMomSpTensor} as a derivative
\begin{align}
	& \mathcal I_\Gamma (s,d,\nu,p,m) = \left. \prod_{i=1}^n \left(\pd{}{c_i}\right)^{s_i}  \int_{\mathbb R^{d\times L}} \left(\prod_{j=1}^L \frac{\ddif k_j}{\pi^{d/2}} \right) \frac{e^{c_1 q_1 + \ldots + c_n q_n}}{D_1^{\nu_1} \cdots D_n^{\nu_n}}\right|_{c=0} \label{eq:FeynmanTensorDerivative}
\end{align}
Repeating the rewriting from momentum space to parametric representation (\cref{thm:SchwingerRepresentation}) with the additional exponential function, one can see that all those tensorial Feynman integrals can be expressed by means of scalar Feynman integrals with shifted parameters.

\begin{theorem}[Tensor reduction, similar to \cite{TarasovConnectionFeynmanIntegrals1996, TarasovGeneralizedRecurrenceRelations1997}] \label{thm:TensorReduction}
	Let $\mathcal I_\Gamma(s,d,\nu,p,m)$ be the tensorial Feynman integral defined in \cref{eq:FeynmanMomSpTensor} for a fixed value of $s\in\mathbb N_0^{d\times n}$. Then there exist polynomials $\alpha_1,\ldots,\alpha_t\in\mathbb Q[p,m]$ and vectors $\rho_1,\ldots,\rho_t\in\mathbb Z_{\geq0}$ such that
	\begin{align}
		\mathcal I_\Gamma(s,d,\nu,p,m) = (-1)^{|s|} \sum_{j=1}^t \alpha_j(p,m)\ (\nu)_{\rho_j}\ \mathcal I_\Gamma (d+ 2 |s|, \nu+\rho_j,p,m) \label{eq:tensorReduction}
	\end{align}
	the tensorial Feynman integral is a linear combination of scalar Feynman integrals with shifted parameters. $(\nu)_{\rho_j} := \frac{\Gamma(\nu+\rho_j)}{\Gamma(\nu)}$ denotes the Pochhammer symbol and $|s|:= \sum_{i=1}^n\sum_{\mu=1}^d (s_i)^\mu \in \mathbb N$ is the number of derivatives in \cref{eq:FeynmanTensorDerivative}.
\end{theorem}
\begin{proof}
	As a first step, we will show that
	\begin{align}
        \overline {\mathcal I}_\Gamma (d,\nu,p,m,c) := \int_{\mathbb R^{d\times L}} \left(\prod_{j=1}^L \frac{\ddif k_j}{\pi^{d/2}} \right) \frac{e^{c^\top\! q}}{D_1^{\nu_1} \cdots D_n^{\nu_n}} = \frac{1}{\Gamma(\nu)} \int_{\mathbb R_+^n} \dif x\, x^{\nu-1} \Uu^{-\dhalf} e^{-\frac{\overline \Ff}{\Uu}} \label{eq:tensorReductionStep1}
	\end{align}
    we can reduce $\overline{\mathcal I}_\Gamma (d,\nu, p, m, c)$ to a scalar Feynman integral with slightly different external momenta and masses where $\overline \Ff(p,m,c) = \Ff(\overline p,\overline m)$ with $\overline p = p + \frac{1}{2} I X^{-1} c$ and $\overline m_i^2 = m_i^2 - \frac{c_i^2}{4x_i^2}$. This is analogue to the deduction of \cref{thm:SchwingerRepresentation}. Applying Schwinger's trick (\cref{lem:SchwingersTrick}) to the left hand side of \cref{eq:tensorReductionStep1} we obtain
	\begin{align}
		\overline {\mathcal I}_\Gamma (d,\nu,p,m,c) = \frac{1}{\Gamma(\nu)} \int_{\mathbb R^n_+} \dif x\,x^{\nu-1} \int_{\mathbb R^{d\times L}} \left(\prod_{j=1}^L \frac{\ddif k_j}{\pi^{d/2}} \right) e^{-\overline\Lambda} \comma
	\end{align}
	where $\overline\Lambda = \Lambda + c^\top\! q = k^\top\! M k + 2 \overline Q^\top\! k + \overline J$ with $\overline Q:= Q + \frac{1}{2} \left(C_{T^*}\right) c$ and $\overline J := J + c^\top\! \hat q$. Further, we used the relation \cref{eq:qjhatqj} and $\Lambda, Q, M$ and $J$ were defined in \cref{thm:SchwingerRepresentation}. Applying \cref{lem:GaussianIntegrals}, we will get
	\begin{align}
		\overline{\mathcal I}_\Gamma (d,\nu, p, m, c) = \frac{1}{\Gamma(\nu)} \int_{\mathbb R^n_+} \dif x\, x^{\nu-1} \Uu^{-\dhalf} e^{-\frac{\overline \Ff}{\Uu}}
	\end{align}
	where $\overline \Ff(p,m,c) := \det (M) \left(-\overline Q^\top M^{-1} \overline Q + \overline J\right)$. Hence, it remains to show $\overline \Ff(p,m,c) = \Ff(\overline p,\overline m)$, which means in particular to show $\overline Q (p) = Q(\overline p)$ and $\overline J(p,m) = J (\overline p,\overline m)$. Due to \cref{lem:IncidenceLoopOrthogonality} it is $I q = I \widehat q = p$. Hence, we will get $\widehat{\overline q } = \widehat q + \frac{1}{2} X^{-1} c$ from $I\left(\widehat{\overline q} - \widehat q\right) = \overline p - p$. Inserting $\widehat{\overline q}$ and $\overline m$ in the definitions of $Q$ and $J$ concludes the statement of \cref{eq:tensorReductionStep1}.
	
	Therefore, the tensorial Feynman integral can be expressed as a derivative with respect to $c$ of $\overline{\mathcal I}_\Gamma (d,\nu, p, m, c)$
	\begin{align}
		\mathcal I_\Gamma(s,d,\nu,p,m) = \frac{1}{\Gamma(\nu)} \int_{\mathbb R_+^n} \dif x\, x^{\nu-1} \Uu^{-\dhalf} \prod_{i=1}^n \left.\left(\pd{}{c_i}\right)^{s_i} e^{-\frac{\overline\Ff}{\Uu}}\right|_{c=0} \point \label{eq:tensorReductionStep2}
	\end{align}	
	Note, that $\overline \Ff(p,m,c)$ is a polynomial in $x$ and $c$, because $\det(M) M^{-1} = \Adj(M), \overline Q$ and $\overline J$ contain only polynomials in $x$ and $c$. Hence, derivatives of $\overline\Ff(p,m,c)$ with respect to $c$ are polynomials in $x$ and $c$ as well. Therefore, we will write
	\begin{align}
		h(x) := \left.\left(\pd{}{c_i}\right)^{s_i} \overline\Ff(p,m,c)\right|_{c=0} = \sum_{j=1}^t \alpha_j(p,m)\ x^{\rho_j} \point \label{eq:tensorReductionStep3}
	\end{align}
	The insertion of \cref{eq:tensorReductionStep3} in \cref{eq:tensorReductionStep2}, together with $\overline\Ff(p,m,c)|_{c=0} = \Ff(p,m)$ will show the assertion. The polynomial $h(x)$ from \cref{eq:tensorReductionStep3} will specify the polynomials $\alpha_1,\ldots,\alpha_t$ and the vectors $\rho_1,\ldots,\rho_t$ in the linear combination \cref{eq:tensorReduction}.
\end{proof}

An advanced version of such a tensor reduction was proposed in \cite{KreimerQuantizationGaugeFields2013, KreimerPropertiesCorollaPolynomial2012}, where a specific differential operator turns scalar Feynman integrals in the corresponding Feynman integral of a gauge theory. This differential operator can be constructed from a graph polynomial, known as the \textit{corolla polynomial}.\bigskip

Therefore, when focusing to scalar Feynman integrals (with general dimension), one can still describe the full range of Feynman integrals. Hence, we will restrict ourselves in the following discussion to the scalar Feynman integrals only.


\section{Dimensional and analytic regularization}
\label{sec:DimAnaReg}
%

In the definition of Feynman integrals in momentum space as well in their reformulation in parametric space we have omitted to discuss the convergence of those integrals. Sufficient criteria of the absolute convergence of the Feynman integral in momentum space \cref{eq:FeynmanMomSp} can be derived by power counting. Hence, the (1PI) Feynman integral in momentum space with Euclidean kinematics \cref{eq:FeynmanMomSp} converges absolutely if the superficial degree of divergence $\omega(\gamma) = \sum_{i=1}^n \Re(\nu_i) - L \Re\!\left(\dhalf\right) >0$ is positive for every subgraph $\gamma\subseteq\Gamma$ \cite{WeinbergHighEnergyBehaviorQuantum1960} and where we have to assume massive edges to exclude IR-divergences. An extension of this result which holds also for non-massive edges was given in \cite{LowensteinPowerCountingTheorem1975}. The convergence region for the Feynman parametric representation \cref{eq:FeynmanParSpFeynman} was worked out in \cite{SpeerUltravioletInfraredSingularity1975}. A short summary of those results can be found in \cite{PanzerFeynmanIntegralsHyperlogarithms2015}. Therefore, we will discuss the convergence of Feynman integrals in the Lee-Pomeransky representation \cref{eq:LeePomeranskyRepresentation}, where we can give a necessary and sufficient condition for absolute convergence. These convergence criteria will also nicely rely on polytopes.\bigskip

The following theorem is mostly a direct implication of the work of \cite{NilssonMellinTransformsMultivariate2010, BerkeschEulerMellinIntegrals2013, SchultkaToricGeometryRegularization2018} and proofs can be found there.
\begin{theorem}[following from {\cite[thm. 2.2]{BerkeschEulerMellinIntegrals2013}}, the second statement was proven in {\cite[thm. 3.1]{SchultkaToricGeometryRegularization2018}}] \label{thm:FIconvergence}
    Consider the Feynman integral in the Lee-Pomeransky representation \cref{eq:LeePomeranskyRepresentation} with the conventions from \cref{eq:Gsupport} and \cref{eq:nuuDefinition} in the Euclidean region $\Re (z_j)>0$ and with positive dimensions $\Re (\nu_0) >0$. Then the Feynman integral converges absolutely, if the real parts of $\nu$ scaled componentwise by the real part of $\nu_0=\frac{d}{2}$ lie inside the relative interior of the Newton polytope of $\Gg$
    \begin{align}
        \Re(\nu)/\Re(\nu_0) \in \relint (\Newt(\Gg)) \point \label{eq:ConvergenceCriteriaFeynmanNewton}
    \end{align}
    Furthermore, if the Newton polytope $\Newt(\Gg)$ is not full dimensional, the Feynman integral does not converge absolutely for any choice of $\nu_0\in\mathbb C$ and $\nu\in\mathbb C^n$.
\end{theorem}
    
Hence, parametric Feynman integrals are well-defined, with the exception of the case where the Newton polytope $\Newt(\Gg)$ is not full dimensional. A degenerated Newton polytope $\Newt(\Gg)$ means, that $\Gg$ is quasi-homogeneous or equivalently that the homogenized point configuration $\Aa$ has not the full rank. For Feynman integrals this case happens if $\Gamma$ is a scaleless graph or contains a scaleless subgraph, i.e.\ a massless subgraph without legs which is connected to the remaining part of $\Gamma$ only by a cut vertex. As aforementioned for those Feynman graphs also the momentum representation diverges for any choice of $d$ and $\nu$. As those graphs will be excluded by the renormalization procedure \cite{StermanIntroductionQuantumField1993}, they are not of relevance in a physical treatment. Therefore, the parametric space representation \cref{eq:LeePomeranskyRepresentation} and the momentum space representation \cref{eq:FeynmanMomSp} are regularizable by $d$ or $\nu$ for the same graphs, albeit the convergence regions may differ.\bigskip

Note, that we can also reformulate \cref{eq:ConvergenceCriteriaFeynmanNewton} by the polytope representation \cref{eq:HPolytope} as  $b_j \Re (\nu_0) - m_j^\top \cdot \Re (\nu) > 0$ for all $j=1,\ldots,k$ where $m_j^\top$ denotes the $j$-th row of $M$. However, these regions of convergence will only cover a small part of the domain where we can analytically continue the Feynman integral representations, even though the Feynman integral is not an entire function. We can distinguish two different kinds of singularities which will appear in the analytic continuation: singularities in the parameters $\nuu$ and singularities in the variables $z$. For the parameters $\nuu$ these singularities are known as UV- and IR-divergences and the Feynman integral has only poles in these parameters $\nuu$. The singularities with respect to the variables $z$ are known as Landau singularities and will only appear if we leave the Euclidean region $\Re(z_j)>0$. Their nature is much more intricate than that of the parametric singularities, and we will discuss them extensively in \cref{ch:singularities}. For now, we will avoid them by restricting ourselves to the Euclidean region.

The possible poles in the parameters $\nuu$ can simply be described with the facets of the Newton polytope $\Newt(\Uu+\Ff)$ by a clever use of integration by parts based on the representation \cref{eq:LeePomeranskyRepresentation}. The following theorem is an application of theorem \cite[thm. 2.4, rem. 2.6]{BerkeschEulerMellinIntegrals2013}.

\begin{theorem}[Meromorphic continuation in parameters $\nuu$ \cite{KlausenHypergeometricSeriesRepresentations2019, BerkeschEulerMellinIntegrals2013}] \label{thm:meromorphicContinuation}
    Describe the non-degenerated Newton polytope $\Newt(\Gg)$ as a minimal number of intersections of half-spaces according to equation \cref{eq:HPolytope} and assume that the components of $m_j^\top\in\mathbb Z^n$ and $b_j\in\mathbb Z$ are relatively prime. Then any Feynman integral $\mathcal I_\Gamma(\nuu,z)$ in the Euclidean region $\Re(z_j)>0$ can be written as
    \begin{align}
        \mathcal I_\Aa(\nuu,z) = \gls{PhiMero} \frac{\prod_{j=1}^k \Gamma( b_j \Re \nu_0 - m_j^\top \cdot \Re \nu)}{\Gamma(\nu_0-\omega)\Gamma(\nu)} \label{eq:meromorphicContinuation}
    \end{align}
    where $\Phi_\Aa(\nuu,z)$ is an entire function with respect to $\nuu\in\mathbb C^{n+1}$. As before we use $\nu_0 = \frac{d}{2}$ and $\nuu = (\nu_0,\nu)$.
\end{theorem}

Hence, we can continue the Feynman integral meromorphically with respect to its parameters $d,\nu$, and we can easily give a necessary condition for its poles. Similar results were also found in \cite{SpeerUltravioletInfraredSingularity1975} based on the Feynman representation \cref{eq:FeynmanParSpFeynman}. Note, that we can apply the results of \cref{thm:FIconvergence,thm:meromorphicContinuation} also to the Feynman representation \cref{eq:FeynmanParSpFeynman} e.g.\ by choosing the hyperplane $H(x) = x_n$.

By means of \cref{thm:meromorphicContinuation} we see that we can avoid the poles of the Feynman integral by considering non-rational values of $\nuu$. More specific it is sufficient to consider either $d$ to be non-rational or $\nu$ to be non-rational. These two options are also known as \textit{dimensional regularization} \cite{EtingofNoteDimensionalRegularization1999, BolliniDimensionalRenorinalizationNumber1972, THooftRegularizationRenormalizationGauge1972} and \textit{analytical regularization} \cite{SpeerAnalyticRenormalization1968, SpeerGeneralizedFeynmanAmplitudes1969}.

\begin{example}
    Consider the \cref{ex:1loopbubbleA} from above which corresponds to \cref{fig:bubble1}. For the relative interior of the Newton polytope one obtains from the facet representation, the region of convergence (with $\Re(\nu_0)>0$)
    \begin{align*}
        & \qquad \Re(\nu_0-\nu_2) > 0 & \Re(-\nu_0+\nu_1+\nu_2)>0 \\
        & \qquad \Re(\nu_2) > 0 & \Re(2\nu_0-\nu_1-\nu_2) > 0
    \end{align*}
    which enables us to separate the poles of the Feynman integral in the $\Gamma$ functions
    \begin{align*}
        \mathcal I_\Aa (\nuu,z) = \Phi_\Aa (\nuu,z) \frac{\Gamma(-\nu_0+\nu_1+\nu_2) \Gamma(\nu_0-\nu_2)}{\Gamma(\nu_1)}
    \end{align*}
    with an entire function $\Phi_\Aa(\nuu,z)$. \Cref{fig:exampleMeromorphicContinuation} shows the original convergence region, as well as the meromorphic continuation. In the case $\nu_1=\nu_2=1$ and $\nu_0=\dhalf=2-\epsilon$ we will obtain
    \begin{align}
    	\mathcal I_\Aa (\nuu,z) = \frac{\Phi^{(0)}_\Aa(\nuu,z)}{\epsilon} + \Phi^{(1)}_\Aa(\nuu,z) + \left( \Phi^{(2)}_\Aa(\nuu,z) + \zeta(2) \Phi^{(0)}_\Aa(\nuu,z) \right) \epsilon + \mathcal O(\epsilon^2)
    \end{align}
    with $\Phi^{(i)}_\Aa(\nuu,z) = \left.\pd{\Phi_\Aa(\nuu,z)}{\epsilon}\right|_{\epsilon=0}$.
    
    \begin{figure}
    	\centering
        \begin{tikzpicture}
            \draw[thick,->] (0,0) -- (4.8,0) node[anchor=north west] {$\frac{\Re(\nu_1)}{\Re(\nu_0)}$};
            \draw[thick,->] (0,0) -- (0,3.8) node[anchor=south east] {$\frac{\Re(\nu_2)}{\Re(\nu_0)}$};
            \coordinate[circle,inner sep=1pt,fill] (A) at (1,0);
            \coordinate[circle,inner sep=1pt,fill] (B) at (2,0);
            \coordinate[circle,inner sep=1pt,fill] (C) at (1,1);
            \coordinate[circle,inner sep=1pt,fill] (D) at (0,1);
            \filldraw [thick, fill=gray, fill opacity = 0.5] (A)  coordinate (GeneralStart) -- ++(1,0) -- ++(-1,1) -- ++(-1,0) -- ++(1,-1) -- cycle;
            \draw (0.1,-2.1) -- (-2.1,0.1);
            \draw (-0.9,-2.1) -- (-2.1,-0.9);
            \draw ($(A) + (-2,0)$) -- ++(2.1,-2.1) -- ($(D) + (-2,0)$) -- ++(-0.1,0.1);
            \draw ($(A) + (-1,0)$) -- ++(2.1,-2.1) -- ($(D) + (-1,0)$) -- ++(-1.1,1.1);
            \draw ($(A) + (0,0)$) -- ++(2.1,-2.1) -- ($(D) + (0,0)$) -- ++(-2.1,2.1);
            \draw ($(A) + (1,0)$) -- ++(2.1,-2.1) -- ($(D) + (1,0)$) -- ++(-2.1,2.1);
            \draw ($(A) + (2,0)$) -- ++(1.1,-1.1) -- ($(D) + (2,0)$) -- ++(-2.1,2.1);
            \draw ($(A) + (3,0)$) -- ++(0.1,-0.1) -- ($(D) + (3,0)$) -- ++(-2.1,2.1);
            \draw (1.9,3.1) -- (4.1,0.9); 
            \draw (2.9,3.1) -- (4.1,1.9);
            
            \draw ($(A) + (0,-2)$) -- ++(-3.1,0) -- ($(B) + (0,-2)$) -- ++(2.1,0);
            \draw ($(A) + (0,-1)$) -- ++(-3.1,0) -- ($(B) + (0,-1)$) -- ++(2.1,0);
            \draw ($(A) + (0,0)$) -- ++(-3.1,0) -- ($(B) + (0,0)$) -- ++(2.1,0);
            \draw ($(A) + (0,1)$) -- ++(-3.1,0) -- ($(B) + (0,1)$) -- ++(2.1,0);
            \draw ($(A) + (0,2)$) -- ++(-3.1,0) -- ($(B) + (0,2)$) -- ++(2.1,0);
            \draw ($(A) + (0,3)$) -- ++(-3.1,0) -- ($(B) + (0,3)$) -- ++(2.1,0);
        \end{tikzpicture}
        \caption[Meromorphic continuation of Feynman integrals w.r.t. parameters $\nuu$]{The original convergence region of the Feynman integral is the gray shaded tetragon. By meromorphic continuation one can extend the Feynman integral to the whole plane. The lines characterize the poles in the parameters. We omitted the cancellations of $\Gamma$-functions with the denominator $\Gamma(\nu_0-\omega)\Gamma(\nu)$ in \cref{eq:meromorphicContinuation}. Thus, this figure shows rather the poles of $\mathcal J_\Aa(\nuu,z)$.} \label{fig:exampleMeromorphicContinuation}
    \end{figure}
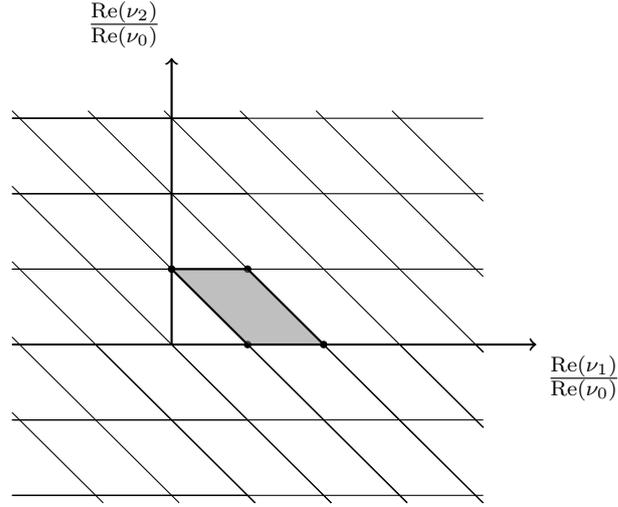
\end{example}

Hence, \cref{thm:meromorphicContinuation} not only vindicates dimensional and analytical regularization, it will also allow us in the $\epsilon$-expansion around integer values of $\nu_0=\dhalf$ to focus on the Taylor expansion of $\Phi_\Aa$ instead of a Laurent expansion of $\mathcal I_\Aa$. Thus, one can determine the coefficients in the $\epsilon$-expansion by differentiating, which makes the procedure much easier.\bigskip

The regularization of Feynman integrals is an intermediate step in the renormalization procedure, which makes the divergences visible. However, we would like to remark that there are also renormalization procedures in which this intermediate step does not have to be carried out explicitly. Renormalization is essential for perturbative QFTs to be formulated in a meaningful way. This is because the original Lagrangian density $\mathcal L$ contains certain ambiguities that are resolved with renormalization. As noted in \cite{DelamotteHintRenormalization2004}, these ambiguities, rather than the divergences of Feynman integrals, are the main reason that renormalization is required in pQFTs. Nevertheless, we do not intend to give an introduction to the very extensive field of renormalization here. Instead, we will refer exemplarily to \cite{DelamotteHintRenormalization2004} for an elementary but illustrative introduction, to \cite{CollinsRenormalizationIntroductionRenormalization1984} for a classical overview and to \cite{ConnesRenormalizationQuantumField1999, ConnesRenormalizationQuantumField2001} for a mathematically rigorous treatment.


\section{Feynman integrals as \texorpdfstring{$\Aa$}{A}-hypergeometric functions} \label{sec:FeynmanIntegralsAsAHyp}
%

It is one of the first observations in the calculation of simple Feynman amplitudes, that Feynman integrals evaluate to classical hypergeometric functions in many cases. This observation was leading Tullio Regge to the conjecture that Feynman integrals are always hypergeometric functions in a general sense \cite{ReggeAlgebraicTopologyMethods1968}. For such a generalization of hypergeometric functions, he suggested to take the analytic behaviour of Feynman integrals as a starting point. In the language of $D$-modules, he proposed that Feynman integrals are holonomic functions whose singular locus is given by the Landau variety. As Regge also noted, this criterion can be transferred to certain partial differential equations, which can be understood as generalized Picard-Fuchs equations.

This idea was later refined by Kashiwara and Kawai \cite{KashiwaraHolonomicSystemsLinear1976}, who showed that Feynman integrals are indeed holonomic functions, i.e.\ they always satisfy holonomic differential equations (see \cref{sec:holonomicDmodules} for the basic notions of $D$-modules). For the $1$-loop case Regge's idea was worked out partly by Kershaw \cite{KershawFeynmanAmplitudesPower1973} and Wu \cite{WuGeneralizedEulerPochhammerIntegral1974}. However, this development happened at a time when there was no consistent theory of general hypergeometric functions. Such a theory was only started in the late '80s by Gelfand, Kapranov and Zelevinsky (GKZ) and their collaborators (we summarized their approach in \cref{sec:AHypSystems}). As already remarked in \cite{GolubevaReggeGelfandProblem2014}, the GKZ theory is the consequential extension of Regge's ideas. Thus, we can give a revisited view on the idea of Regge by means of GKZ theory and within this framework, Regge's conjecture can also be confirmed. \bigskip

Beside the general question of a sufficient functional class of Feynman integrals, specific hypergeometric functions play also an important role in many approaches of the calculation of Feynman integrals. Typically, those hypergeometric functions appear in the often used Mellin-Barnes approach \cite{UsyukinaRepresentationThreepointFunction1975, BergereAsymptoticExpansionFeynman1974, SmirnovFeynmanIntegralCalculus2006}. This appearance is a consequence of Mellin-Barnes representations with integrands consisting in a product of $\Gamma$-functions, which can be identified by the hypergeometric Fox's $H$-functions \cite{Inayat-HussainNewPropertiesHypergeometric1987, Inayat-HussainNewPropertiesHypergeometric1987a, BuschmanFunctionAssociatedCertain1990}. Further, there are different techniques for the representation of $1$-loop integrals in terms of simple hypergeometric functions \cite{BoosMethodCalculatingMassive1991, FleischerNewHypergeometricRepresentation2003, PhanScalar1loopFeynman2019}, which rely on those Mellin-Barnes representations by means of residue theorem. Beside the $1$-loop case, there are not many results known, which is also due to the fact that multivariate Mellin-Barnes integrals can be highly non-trivial \cite{ParisAsymptoticsMellinBarnesIntegrals2001}. We also refer to \cite{KalmykovHypergeometricFunctionsTheir2008} for a review of hypergeometric functions appearing in Feynman integral calculus. \bigskip

As aforementioned the Gelfand-Kapranov-Zelevinsky approach is a convenient opportunity to examine the correspondence between hypergeometric functions and Feynman integrals. It was already stated by Gelfand and his collaborators themselves, that ``practically all integrals which arise in quantum field theory'' \cite{GelfandHypergeometricFunctionsToral1989} can be treated with this approach. However, this insight seems to have been forgotten for a long time and was not pursued further. In 2014, the connection between Feynman integrals and GKZ-hypergeometric theory was mentioned and discussed again by Golubeva \cite{GolubevaReggeGelfandProblem2014}. In \cite{NasrollahpoursamamiPeriodsFeynmanDiagrams2016} it was proven, that the Feynman integrals satisfy a system of differential equations which is isomorphic to an $\Aa$-hypergeometric system. Recently, the fact that scalar Feynman integrals are $\Aa$-hypergeometric functions was independently shown in 2019 by \cite{DeLaCruzFeynmanIntegralsAhypergeometric2019} and \cite{KlausenHypergeometricSeriesRepresentations2019} based on the Lee-Pomeransky representation \cref{eq:LeePomeranskyRepresentation}.

\begin{theorem}[Feynman integrals as $\Aa$-hypergeometric functions] \label{thm:FeynmanAHypergeometric}
    Consider a generalized scalar Feynman integral $\mathcal I_\Aa(\nuu,z)$ in the representation \cref{eq:LeePomeranskyRepresentation}. $\nuu\in\mathbb C^{n+1}$ was defined in \cref{eq:nuuDefinition} and $\Aa\in\mathbb Z^{(n+1)\times N}$ is the homogenization of $A$, which was defined by \cref{eq:Gsupport}, i.e.\ we interpret $A$ as a set of column vectors building an $n\times N$ integer matrix and adding the row $(1,\ldots,1)$. Then $\mathcal I_\Aa(\nuu,z)$ is $\Aa$-hypergeometric, i.e.\ it satisfies the $\Aa$-hypergeometric system \cref{eq:AhypIdeal}
    \begin{align}
        H_\Aa (\nuu) \bullet \mathcal I_\Aa(\nuu,z)=0 \point      
    \end{align}
    Thus, the generalized Feynman integral is an $\Aa$-hypergeometric function.
\end{theorem}
\begin{proof}
    Instead of $\mathcal I_\Aa(\nuu,z)$ we will consider $\mathcal J_\Aa(\nuu,z) := \frac{1}{\Gamma(\nu_0-\omega)\Gamma(\nu)} \mathcal I_\Aa(\nuu,z)$ to avoid unnecessary prefactors. Firstly, we want to show that Feynman integrals satisfy the toric operators $\square_l$ for all $l\in\mathbb L$ or equivalently $ \left\{\partial^u - \partial^v \,\rvert\, \Aa u = \Aa v,\: u,v\in\mathbb N^N\right\}$. Derivatives of the Feynman integral with respect to $z$ result in
    \begin{align}
        \partial^u \bullet \Gamma(\nu_0) \int_{\mathbb R^{n}_+} \dif x \, x^{\nu-1}\Gg^{-\nu_0} = \Gamma\!\left(\nu_0-|u|\right) \int_{\mathbb{R}^n_+} \dif x \, x^{\nu-1} x^{\Aa u} \Gg^{-\nu_0- |u|}
    \end{align}
    where $|u|:= \sum_i u_i$. Since $\Aa$ contains the row $(1,1,\ldots,1)$ it follows immediately that $|u|=|v|$. Therefore, one obtains the same equation for $v$.
    
    Secondly, we want to show that Feynman integrals satisfy the homogeneous operators $E_i(\nuu)$ for $i=0,\ldots,n$. For this purpose note that $\mathcal J_\Aa\!\left(\nuu,s^{a_{b}^{(1)}} \! z_1,\ldots,s^{a_{b}^{(N)}} \! z_N\right) = \Gamma(\nu_0) \int_{\mathbb R^n_+} \dif x\, x^{\nu-1} \Gg (x_1,\ldots,s x_b,\ldots,x_n)^{-\nu_0}$. After a substitution $x_b \mapsto \frac{1}{s} x_b$ for $s>0$ it is
    \begin{align}
        \mathcal J_\Aa\!\left(\nuu,s^{a_{b}^{(1)}} \! z_1,\ldots,s^{a_{b}^{(N)}} \! z_N\right) = s^{-\nu_b} \mathcal J_\Aa (\nuu,z) \point \label{eq:FeynAhypProof}
    \end{align}
    A derivative of \cref{eq:FeynAhypProof} with respect to $s$ concludes the proof for $s=1$.
\end{proof}

Hence, as conjectured by Regge and suggested already by Gelfand and his collaborators every scalar, generalized Feynman integral in Euclidean kinematics satisfies the $\Aa$-hypergeometric system and can be treated within the framework of GKZ. 

This fact is not quite surprising as the parametric representations \cref{eq:FeynmanParSpFeynman} and \cref{eq:LeePomeranskyRepresentation} belong both to the class of Euler-Mellin integrals \cite{BerkeschEulerMellinIntegrals2013}, which are defined as Mellin transforms of a product of polynomials up to certain powers. As every Euler-Mellin integral, also the Feynman integral is an $\Aa$-hypergeometric function. 

Therefore, we can also write an $\Aa$-hypergeometric systems for the Feynman representation \cref{eq:FeynmanParSpFeynman} in the following way. Without loss of generality we will set $x_n=1$ by evaluating the $\delta$-distribution in \cref{eq:FeynmanParSpFeynman}. Denote by $A_\Uu$ and $A_\Ff$ the support of the first and the second Symanzik polynomial after setting $x_n=1$. In doing so, we can construct the following matrix      
\begin{align}
    \Aa^\prime = \begin{pmatrix}
        1 & \cdots & 1 & 0 & \cdots & 0\\
        0 & \cdots & 0 & 1 & \cdots & 1\\
        \\
        & A_\Uu      &   &   &  A_\Ff &\\
        \\            
    \end{pmatrix} \label{eq:AprimeUF}
\end{align}
which defines together with $\beta = \left(\frac{d}{2}-\omega,\omega,\nu_1,\ldots,\nu_{n-1}\right)$ the $\Aa$-hypergeometric system $H_{\Aa^\prime}(\beta)$ of \cref{eq:FeynmanParSpFeynman}. A matrix of the form \cref{eq:AprimeUF} is also known as \textit{Cayley embedding} of $A_\Uu$ and $A_\Ff$ \cite[def. 9.2.11]{DeLoeraTriangulations2010}. As expected the $\Aa$-hypergeometric systems for \cref{eq:FeynmanParSpFeynman} and \cref{eq:LeePomeranskyRepresentation} are equivalent, which can be verified by the matrix
\begin{align}
    T = \left(\begin{array}{cccc}
        L+1 & -1 & \cdots & -1 \\
        -L  & 1  & \cdots & 1  \\
        0   &    &        & 0 \\  
        \vdots &   \multicolumn{2}{c}{\smash{ \scalebox{1.5}{$\mathbbm{1}$}}_{n-1} }       & \vdots \\
        0      & &        & 0
        \end{array}\right)\quad\text{,}\qquad 
    T^{-1} =  \left(\begin{array}{ccccc}
        1 & 1 & 0 & \cdots & 0 \\
        0 & 0 & & & \\
        \vdots & \vdots &   \multicolumn{3}{c}{\smash{\scalebox{1.5}{$\mathbbm{1}$}}_{n-1} } \\
        0 & 0 & & \\
        L & L+1 & -1 & \cdots & -1
    \end{array}\right)
\end{align}
which transforms $\Aa^\prime = T \Aa$ and $\beta = T \nuu$. Moreover, by Laplace expansion we can see that $T$ is an unimodular matrix, whence $\Conv(\Aa)$ and $\Conv(\Aa^\prime)$ are equivalent polytopes. According to \cite{BerkeschEulerMellinIntegrals2013}, when working with the representation \cref{eq:FeynmanParSpFeynman} we will consider the Newton polytope $\Newt\!\left(\left(\Uu\cdot\Ff\right)\!|_{x_n=1}\right) = \Newt\!\left(\Uu|_{x_n=1}\right) + \Newt\!\left(\Ff|_{x_n=1}\right)$ instead of $\Newt(\Gg)$, where the sum denotes the Minkowski addition. For the general relation between Cayley embedding and Minkowski sums we refer to \cite[lem. 3.2]{HuberCayleyTrickLifting2000}. In the following we will mostly prefer the Lee-Pomeransky representation \cref{eq:LeePomeranskyRepresentation} due to its plainer structure. \bigskip

Feynman integrals will only be a subclass of $\Aa$-hypergeometric functions. We want to list certain characteristics that distinguish Feynman integrals within this class of arbitrary $\Aa$-hypergeometric functions. This characterization does not claim to be exhaustive. When considering $\Aa$-hypergeometric functions, the behaviour will be determined by the vector configuration $\Aa\subset\mathbb Z^{n+1}$ or equivalently by the Newton polytope $\Newt(\Gg)$. Therefore, we will examine the special properties of these objects in case where $\Aa$ comes from a scalar Feynman integral. Obviously, from the definitions of $\Uu$ and $\Ff$ \cref{eq:FirstSymanzik}, \cref{eq:SecondSymanzik}, the entries of the matrix $\Aa$ are restricted to $\Aa\in \{0;1;2\}^{(n+1)\times N}$. In addition, every column of $\Aa$ contains at most one entry equals $2$. In case of massless Feynman integrals, $\Newt(\Gg)$ will even be a $0/1$-polytope, i.e.\ $\Aa\in\{0;1\}^{(n+1)\times N}$. Furthermore, due to the homogeneous degrees of $\Uu$ and $\Ff$, all points in $\Newt(\Gg)$ are arranged on two parallel hyperplanes of $\mathbb R^{n}$. These hyperplanes will be of the form $\left\{\mu\in\mathbb R^n \,\rvert\, \sum_{i=1}^n \mu_i = L \right\}$ and  $\left\{\mu\in\mathbb R^n \,\rvert\, \sum_{i=1}^n \mu_i = L+1 \right\}$, respectively. Thus, $\Newt(\Gg)$ is compressed in one direction. As a further consequence, $\Newt(\Gg)$ will have no interior points. \bigskip

For fully massive Feynman integrals it was noticed in \cite{TellanderCohenMacaulayPropertyFeynman2021}, that every monomial of $\Ff_0$ will be contained in the massive part of the second Symanzik polynomial $\Uu\sum_i x_i m_i^2$. Hence, we will have
\begin{align}
    \Newt(\Ff) = \Newt\!\left(\Uu \sum_{i=1}^n x_i m_i^2\right) = \Newt(\Uu) + \Delta_n	
\end{align}
where $\Delta_n = \Conv(e_1,\ldots,e_n)$ is the $(n-1)$-simplex and where the sum denotes the Minkowski addition. In consequence, it is
\begin{align}
	\Newt(\Gg) = \Newt\!\left(\Uu \left(1 + \sum_{i=1}^n x_i m_i^2\right)\right) = \Newt(\Uu) + \tilde\Delta_n
\end{align}
with the $n$-simplex $\tilde\Delta_n = \Conv(0,e_1,\ldots,e_n)$. Thus, it is remarkable that for fully massive graphs, the precise form of the second Symanzik polynomial does not play a role for the most properties of the Feynman integral. Only for the specification to the variables $z$ to the physical values we have to consider the second Symanzik polynomial. \bigskip

It was conjectured\footnote{In fact, it was even conjectured that Feynman configurations always admit unimodular triangulations, which is a slightly stronger assumption.} in \cite{KlausenHypergeometricSeriesRepresentations2019}, that Feynman configurations are always \textit{normal}, i.e.\ that they satisfy
\begin{align}
    k \Newt(\Gg) \cap \mathbb Z^n = (k-1) \Newt(\Gg) \cap \mathbb Z^n + \Newt(\Gg)\cap\mathbb Z^n
\end{align}
for any $k\in\mathbb N$, which is also known as the \textit{integer decomposition property}. This property is particularly interesting because it implies that the toric ideal $I_\Aa$ defined by $\Newt(\Gg)=\Conv(A)$ is Cohen-Macaulay. 

This assumption was shown for two specific classes of Feynman graphs in \cite{TellanderCohenMacaulayPropertyFeynman2021} and was largely extended in \cite{WaltherFeynmanGraphsMatroids2022}.
\begin{theorem}[Cohen-Macaulayness for Feynman integrals {\cite[thm. 3.1 and thm. 3.4]{TellanderCohenMacaulayPropertyFeynman2021},\cite[thm. 4.3, thm. 4.5 and thm. 4.9]{WaltherFeynmanGraphsMatroids2022}}] \label{thm:CohenMacaulayFeynman}
	Assuming that $\Gamma$ is a (1PI), (1VI) Feynman graph with sufficiently generic\footnote{``Sufficiently generic'' means here that the external momenta should not lead to a cancellation of monomials in $\Gg$, which could happen outside the Euclidean region for specific momenta.} external momenta. Then in the following cases the Feynman configuration is normal
	\begin{enumerate}[a)]
		\item all edges of $\Gamma$ are massive
		\item all edges of $\Gamma$ are massless
		\item every $2$-forest of $\Gamma$ induces a non-zero term in $\Gg$.
	\end{enumerate}
	The last case c) can be rephrased as follows: For every internal vertex of $\Gamma$ there is a path to an external vertex consisting in massive edges only.
\end{theorem}

As seen in \cref{thm:HolRankAHyp}, the Cohen-Macaulayness of the toric ideal $I_\Aa$ ensures that the holonomic rank is given by $\vol(\Newt(\Gg))$ for all values of $\nuu$ and not only for generic ones. Therefore, the $\Aa$-hypergeometric system is well-behaved also for the points $\nuu\in\mathbb Z^{n+1}$ and has no rank jumpings there. Thus, even though the Feynman integral may diverge for $d\rightarrow 4$, the structure of the Feynman integral remains the same in this limit.\bigskip

As a further characterization of Feynman configurations, we would like to draw attention to the insight found in \cite{SchultkaToricGeometryRegularization2018}. According to it, the polytope $\Newt(\Uu\Ff) = \Newt(\Uu) + \Newt(\Ff)$ is a generalized permutahedron, which means that all edges of $\Newt(\Uu\Ff)$ are parallel to an edge of the form $e_i-e_j$ for all $i,j\in \{1,\ldots,n\}$. In addition, Symanzik polynomials satisfy certain useful relations going beyond \cref{eq:UContractedDeleted} as stated in \cite[prop. 4.11]{SchultkaToricGeometryRegularization2018}, \cite{BrownFeynmanAmplitudesCosmic2017}.\bigskip

To conclude this chapter we want to highlight the most important points for the following. For a ($1$PI) and ($1$VI) Feynman graph, the scalar Feynman integral in the Euclidean region $\Re(z_j)>0$ is a meromorphic function in the parameters $\nuu=\left(\dhalf,\nu_1,\ldots,\nu_n\right)\in\mathbb C^{n+1}$. Moreover, the generalized Feynman integral is an $\Aa$-hyper\-ge\-o\-me\-tric function and the (physical) Feynman integral is a certain restriction of those $\Aa$-hypergeometric functions. As we will see in \cref{ch:singularities} we can relax the restriction to the Euclidean region after a rigorous treatment of kinematic singularities. Furthermore, in \cref{cor:MellinBarnesRepresentation} we found a class of Feynman-like integrals which provides a simple solution in terms of $\Gamma$-functions. These integrals will be helpful to fix the boundary values for the $\Aa$-hypergeometric systems.


\chapter{Series representations} \label{ch:seriesRepresentations} 
%

In this chapter we will concern with power series representations of Feynman integrals. It will be found\footnote{That Feynman integrals can always be expressed by Horn hypergeometric functions has been assumed for a long time, with good reasons, see e.g.\ \cite{KalmykovFeynmanDiagramsDifferential2009}. However, a rigorous proof has been lacking so far.} that any given generalized Feynman integral $\mathcal I_\Aa(\nuu,z)$ can be expressed in terms of Horn hypergeometric functions \cref{eq:DefHornHypergeometric} for every regular triangulation of the Newton polytope $\Newt(\Gg)$. Hence, we will generate a set of different series representations for each Feynman integral. This will be done by considering the $\Gamma$-series solutions of $\Aa$-hypergeometric systems from \cref{ssec:GammaSeries}. To fix the Feynman integral as a specific element in the solution space $\Sol(H_\Aa(\nuu))$ we will make use of \cref{cor:MellinBarnesRepresentation}, which provides the boundary values for the $\Aa$-hypergeometric system. In this way, we are able to give a closed formula for series representations of generalized Feynman integrals in \cref{thm:FeynSeries}.\bigskip

After discussing the series representations for generalized Feynman integrals, we will answer the question of how one can transform statements for the generalized Feynman integral into statements of the Feynman integral restricted to physical values (\cref{sec:AnalyticContinuation}). Furthermore, we will present techniques for the Laurent expansion of those series in their parameters $\nuu$ in \cref{sec:epsilonExpansion}, which is necessary in the dimensional and analytic regularization.

Since the handling of multivariate series can be very elaborate, we will provide a small amount of tools for the treatment of those series in \cref{sec:ManipulationSeries}. Those techniques can in principle also be used for a symbolic reduction of Horn hypergeometric series to multiple polylogarithms and related functions. However, series representations are notably efficient in a numerical evaluation (\cref{sec:numerics}). For convenient kinematical regions those series will converge very fast and one can give a sufficient approximation after a few summands.\bigskip

In principle there are many known ways to span the solution space of $\Aa$-hyper\-ge\-o\-me\-tric systems \cite{Matsubara-HeoLaplaceResidueEuler2018}. Thus, one can use the $\Aa$-hypergeometric theory also to write the Feynman integral in terms of Euler integrals or Laplacian integrals to name a few. We will address those alternative approaches in \cref{sec:EulerIntegrals}.

The procedure presented in this chapter is not only restricted to Feynman integrals. One can derive series representations for any Euler-Mellin integral. We will show some of these applications appearing in Feynman calculus in \cref{sec:periodMarginal}. In order to illustrate the method of Horn hypergeometric series, we will conclude this chapter by deriving a series representation of the fully massive sunset graph in \cref{sec:ExampleSunset}.


\section[Series representations for generalized Feynman integrals]{Series representations for generalized Feynman \\ integrals} 
\sectionmark{... for generalized Feynman integrals}
\label{sec:seriesRepresentationsSEC}


%

As it was observed in \cref{thm:FeynmanAHypergeometric}, every generalized Feynman integral is an $\Aa$-hypergeometric function. Thus, we can directly apply the results of $\Aa$-hypergeometric systems from \cref{ch:AHypergeometricWorld}. In addition to the genericity of the variables $z$, we will also consider preliminary generic\footnote{As before, we will call $\nuu\in \mathbb C^{n+1}$ generic, if $\nuu \in D \subseteq \mathbb C^{n+1}$ attains only values in a non-empty, Zariski open set $D$. This means in particular that $D$ is dense in $\mathbb C^{n+1}$. When working with $\Gamma$-series \cref{eq:GammaSeriesDefCh4}, we will often consider a slightly stronger restriction, where we will assume $\nuu$ to be in a countable intersection of non-empty, Zariski open sets. In other words, we want $\nuu$ to attend values in the complement of a countable union of hyperplanes of $\mathbb C^{n+1}$. This slightly stronger version is sometimes referred as ``very generic''. Hence, for very generic $\nuu$, we can assume that $\Aas^{-1}\nuu+\Aas^{-1}\Aabs \lambda$ will not contain any integer coordinate for all $\lambda\in\Lambda_k$, see also \cref{thm:meromorphicContinuation}. Note that for the assumptions in \cref{thm:HolRankAHyp} generic values of $\nuu$ are sufficient.} values for $\nuu\in \mathbb C^{n+1}$ in order to avoid singularities in the $\Gamma$-functions. The latter assumption will also exclude the possibilities of rank jumpings, i.e.\ the holonomic rank of $H_\Aa(\nuu)$ is precisely given by the volume of $\Newt(\Gg)$ according to \cref{thm:HolRankAHyp}. However, as stated in \cref{thm:CohenMacaulayFeynman}, the toric ideal generated by Feynman integrals is Cohen-Macaulay in many cases. Therefore, there will be no rank jumpings even for non-generic $\nuu$ for the most Feynman integrals. Hence, this will be not a serious restriction for Feynman integrals and is only done for simplicity in the handling of $\Gamma$-functions.\bigskip

In \cref{ssec:GammaSeries} we already constructed series solutions for any $\Aa$-hypergeometric system. In particular, a basis of the solution space $\Sol(H_\Aa(\nuu))$ was provided by $\Gamma$-series. Thus, let $\Tt$ be a regular triangulation of the Newton polytope $\Newt(\Gg)=\Newt(\Uu+\Ff)$ and $\hatT$ the set of maximal cells of $\Tt$. Then according to \cref{thm:SolutionSpaceGammaSeries} we can write every element of $\Sol(H_\Aa(\nuu))$ as a linear combination of $\Gamma$-series and especially the Feynman integral. For the sake of simplicity, we will drop the prefactors of the Feynman integral \cref{eq:LeePomeranskyRepresentation} and consider instead
\begin{align}
    \gls{gFJ} := \Gamma (\nu_0-\omega)\Gamma(\nu) \mathcal I_\Aa(\nuu,z) = \Gamma(\nu_0)\int_{\mathbb R^n_+} \! \dif x\, x^{\nu-1} \Gg^{-\nu_0} \label{eq:DefinitionFeynmanJ}
\end{align}
where the corresponding definitions can be found in \cref{eq:Gsupport} and \cref{eq:nuuDefinition}. Hence, the Feynman integral can be written as a linear combination of $\Gamma$-series
\begin{align}
	\mathcal J_\Aa(\nuu,z) = \sum_{\sigma\in\hatT} \sum_{k\in K_\sigma} \gls{Cfactors} \varphi_{\sigma,k}(\nuu,z) \label{eq:LinearCombinationGammaSeries}
\end{align}
where $K_\sigma=\left\{k^{(1)},\ldots,k^{(s)}\right\}$ is a set of representatives of $\bigslant{\mathbb Z^{n+1}}{\mathbb Z\Aas} = \big\{\!\left[\Aabs k^{(j)}\right] \, \big\rvert \, j = 1,\ldots, s = \vol(\Conv(\Aas))\big\}$ for every simplex $\sigma\in\hatT$ according to \cref{eq:representativesK}. Thereby $\varphi_{\sigma,k}$ denote the $\Gamma$-series which were defined in \cref{ssec:GammaSeries}, i.e.\
\begin{align}
	\varphi_{\sigma,k} (\nuu,z) = \zsigma^{-\Aas^{-1}\nuu} \sum_{\lambda\in\Lambda_k} \frac{\zsigma^{-\Aas^{-1}\Aabs \lambda} \zbarsigma^\lambda}{\lambda! \ \Gamma\!\left(1-\Aas^{-1}\nuu - \Aas^{-1}\Aabs \lambda\right)} \label{eq:GammaSeriesDefCh4}
\end{align}
with $\Lambda_k=\left\{ k + m \in \mathbb{N}_0^r \, \rvert\, \Aas^{-1}\Aabs m \in \mathbb Z^{n+1} \right\}\subseteq \mathbb{N}_0^r$. By $\bar\sigma = \{1,\ldots,N\}\setminus\sigma$ we were denoting the complement of $\sigma$ and $\Aas$, $\Aabs$, $\zsigma$, $\zbarsigma$ and all related objects indicate the restriction to columns corresponding to $\sigma$ and $\bar\sigma$, respectively. Further, we will assume Euclidean kinematics $\Re(z_j)>0$ in a convenient region, such that all $\Gamma$-series $\varphi_{\sigma,k}$ are convergent. As pointed out in \cref{ssec:GammaSeries}, those regions always exist. 

Therefore, in order to get a series representation of Feynman integrals, we have to determine the meromorphic functions $C_{\sigma,k} (\nuu)$ of \cref{eq:LinearCombinationGammaSeries}. This can be done by considering specific values for $z$ in \cref{eq:LinearCombinationGammaSeries}. In order to get sufficient boundary values also for non-unimodular triangulations $\Tt$, we have to include derivatives of $\mathcal J_\Aa(\nuu,z)$ with respect to $z$. As aforementioned, derivatives with respect to $z$ will correspond to a shift in the parameters $\nuu$, and we will have
\begin{align}
	\partial^u \varphi_{\sigma,k} = \varphi_{\sigma,k-u_{\bar\sigma}} (\nuu + \Aa u,z) 
\end{align}
with $\partial^u = \prod_{i=1}^N \left(\pd{}{z_i}\right)^{u_i}$ and $u\in\mathbb N_{\geq 0}^N$. Up to a sign, also the derivatives of the Feynman integral relate simply in a shift of the parameters $\nuu$
\begin{align}
	\partial^u \mathcal J_\Aa(\nuu,z) = (-1)^{|u|} \mathcal J_\Aa (\nuu+\Aa u,z) \label{eq:FeynmanJDerivative}
\end{align}
where $|u| := \sum_{i=1}^N u_i$, which follows directly from the definition \cref{eq:DefinitionFeynmanJ}. \bigskip

For the purpose of considering boundary values of \cref{eq:LinearCombinationGammaSeries} where certain variables $z_j$ are set to zero, we will examine the behaviour of $\Gamma$-series when they are restricted to subconfigurations in the following slightly technical lemma. As we will see below, Feynman integrals transmit their meromorphic functions $C_{\sigma,k}(\nuu)$ to simpler Feynman integrals. This enables us to reduce every Feynman integral to the case described in \cref{cor:MellinBarnesRepresentation} and derive an analytic expression of the functions $C_{\sigma,k}(\nuu)$.

\begin{lemma}[linear combinations for subtriangulations] \label{lem:linCombSubtriangs}
	Let $\Aa\in\mathbb Z^{(n+1)\times N}$ and $\Aa^\prime := \Aa\setminus i \in\mathbb Z^{(n+1)\times (N-1)}$ be two acyclic, full dimensional vector configurations with $r=N-n-1>1$. Further, let $\Tt$ be a regular triangulation of $\Aa$ and $\Tt^\prime$ a regular triangulation of $\Aa^\prime$ being subtriangulations of each other $\Tt^\prime \subseteq \Tt$. Moreover, we will consider that the representatives $K_\sigma$ and $K^\prime_\sigma$ are chosen compatible with each other for all $\sigma\not\ni i$, i.e.\ for every ${k^\prime}^{(j)}\in K_\sigma^\prime$ we will construct $k^{(j)}\in K_\sigma$ by adding $k^{(j)}_i = 0$ as the $i$-th component. Under the assumption $\lim_{z_i\rightarrow 0} \mathcal J_\Aa(\nuu,z) = \mathcal J_{\Aa^\prime}(\nuu,z^\prime)$ we will have an equality of the meromorphic functions from \cref{eq:LinearCombinationGammaSeries}
	\begin{align}
		C_{\sigma,k}(\nuu) = C^\prime_{\sigma,k^\prime} (\nuu) \qquad\text{for all } \sigma \in \hatT^\prime
	\end{align}
	where the primed objects are all related to the system $H_{\Aa^\prime}(\nuu)$.
\end{lemma}
\begin{proof}
    Due to the relation of derivatives and shifts in the parameters \cref{eq:FeynmanJDerivative} we can extend $\lim_{z_i\rightarrow 0} \mathcal J_\Aa(\nuu,z) = \mathcal J_{\Aa^\prime}(\nuu,z^\prime)$ to
    \begin{align}
        \lim_{z_i\rightarrow 0} \partial^u \mathcal J_\Aa(\nuu,z) = (-1)^{|u|} \mathcal J_{\Aa^\prime} (\nuu + \Aa u,z^\prime) \stackrel{u_i=0}{=} \partial^{u^\prime} \mathcal J_{\Aa^\prime} (\nuu , z^\prime)
    \end{align}
    where we use $z^\prime := (z_j)_{j\neq i}$ as before for the variables of the system $H_{\Aa^\prime}(\nuu)$ and equivalently $u^\prime := (u_j)_{j\neq i}$. So the task will be to compare the limits of $\partial^u \mathcal J_\Aa(\nuu,z)$ with the limits of $\Gamma$-series $\partial^u \varphi_{\sigma,k}(\nuu,z)$. For the latter we have to distinguish between the two cases $\sigma\ni i$ and $\sigma\not\ni i$. Starting with the second case we have
    \begin{align}
        \lim_{z_i\rightarrow 0} \partial^u \varphi_{\sigma,k}(\nuu,z) \stackrel{i\notin\sigma}{=} \zsigma^{- \Aas^{-1}(\nuu+\Aa u)} \displaystyle\sum_{\substack{\lambda\in\Lambda_{k-u_{\bar\sigma}} \\ \text{with } \lambda_i = 0}} \frac{\zsigma^{-\Aas^{-1}\Aabs \lambda} \zbarsigma^\lambda}{\lambda! \ \Gamma\!\left(1-\Aas^{-1} (\nuu + \Aa u + \Aabs \lambda)\right)} \label{eq:derivativeGammaSeriesLimit} \point
    \end{align}
    Note, that there does not necessarily exist a $\lambda\in\Lambda_{k-u_{\bar\sigma}}$ satisfying $\lambda_i=0$. In this case the sum in \cref{eq:derivativeGammaSeriesLimit} will be empty and $\lim_{z_i\rightarrow 0} \partial^u \varphi_{\sigma,k}(\nuu,z) = 0$. To avoid this case, let us assume that we have chosen $u$ in such a way that $(u_{\bar\sigma})_i = k_i$. This will agree with the previous formulated assumptions $(u_{\bar\sigma})_i = 0$ and $k_i=0$. Hence, we can reformulate the summation region
    \begin{align}
        \left\{\lambda\in\Lambda_{k-u_{\bar\sigma}} \,\rvert\, \lambda_i = 0 \right\} &= \left\{ k^\prime - u_{\bar\sigma}^\prime + m^\prime \in \mathbb{N}_0^{r-1} \, \rvert\, \Aas^{-1}\Aabs^\prime m^\prime \in \mathbb Z^{n+1} \right\} \times \left\{\lambda_i = 0\right\} \nonumber \\
        & = \Lambda^\prime_{k^\prime - u_{\bar\sigma}^\prime}  \times \left\{\lambda_i = 0 \right\}
    \end{align}
    where the primed objects belong to the vector configuration $\Aa^\prime$. Thus, we will obtain for $i\notin\sigma$ with $(u_{\bar\sigma})_i = k_i=0$
    \begin{align}
        \lim_{z_i\rightarrow 0} \partial^u \varphi_{\sigma,k} (\nuu,z) = \varphi^\prime_{\sigma,k^\prime - u_{\bar\sigma}^\prime} (\nuu+\Aa u, z^\prime) = \partial^{u^\prime} \varphi^\prime_{\sigma,k^\prime} (\nuu,z^\prime) \comma
    \end{align}
    which is nothing else than the $\Gamma$-series we would expect for a system $\Aa^\prime$. Therefore, in the linear combination \cref{eq:LinearCombinationGammaSeries} we will have
    \begin{align}
        \lim_{z_i\rightarrow 0} & \partial^u \mathcal J_\Aa(\nuu,z) = \partial^{u^\prime} \mathcal J_{\Aa^{\prime}} (\nuu, z^\prime) \\&= \sum_{\sigma\in\hatT^\prime} \sum_{k^\prime\in K^\prime_\sigma} C_{\sigma,k^\prime}(\nuu) \partial^{u^\prime} \varphi^\prime_{\sigma,k^\prime}(\nuu,z^\prime) + \sum_{\substack{ \sigma\in\hatT \\ \text{s.t. } i\in\sigma}} \sum_{k\in K_\sigma} C_{\sigma,k} (\nuu) \lim_{z_i\rightarrow 0}\partial^{u^\prime} \varphi_{\sigma,k}(\nuu,z) \point \nonumber
    \end{align}
    Since the $\Gamma$-series $\varphi^\prime_{\sigma,k^\prime}$ with $\sigma\not\ni i$ already describe the full $\Aa^\prime$-hypergeometric system $H_{\Aa^\prime}(\nuu)$ the $\Gamma$-series $\lim_{z_i\rightarrow 0} \varphi_{\sigma,k}$ with $i\in\sigma$ are either linear dependent of the latter or they are zero. Suppose that they are linear dependent. This implies in particular that the limit of the series 
    \begin{align}
        \lim_{z_i\rightarrow 0} \sum_{\lambda\in\Lambda_{k-u_{\bar\sigma}}} \frac{\zsigma^{-\Aas^{-1}\Aabs\lambda} \zbarsigma^\lambda}{\lambda! \ \Gamma\!\left(1-\Aas^{-1} (\nuu + \Aa u) - \Aas^{-1}\Aabs\lambda\right)} \quad\text{with } i\in\sigma
    \end{align}
    is finite. On the other hand there exists a (full dimensional) region $D\subset\mathbb C^{n+1}$ such that $\left(\Aas^{-1}\nuu\right)_i < 0$ for all $\nuu\in D$. Therefore, we will have $\lim_{z_i\rightarrow 0} \partial^u \varphi_{\sigma,k}(\nuu,z)=0$ for all $\nuu\in D$. Standard arguments of complex analysis lead to the result that $\lim_{z_i\rightarrow 0} \partial^u \varphi_{\sigma,k}(\nuu,z)$ vanishes for all generic values of $\nuu$ \cite[thm. 4.8, ch. 2]{SteinComplexAnalysis2003}. Therefore, the assumption of linear dependency was false and all the $\Gamma$-series $\partial^u \varphi_{\sigma,k}$ with $\sigma\ni i$ will identically vanish in the limit $z_i\rightarrow 0$. Note, that these series may not be defined in that limit, and it is rather the analytic continuation of the $\Gamma$-series which will vanish then.
    
    Hence, the $\Gamma$-series of $\Aa$ with $\sigma \not\ni i$ will transform to the $\Gamma$-series of $\Aa^\prime$ in the limit $z_i\rightarrow 0$, whereas the $\Gamma$-series with $\sigma\ni i$ will vanish. Thus, the meromorphic functions $C_{\sigma,k}(\nuu)$ will stay the same in such a limit. 
\end{proof}

Therefore, the set of $\Gamma$-series for the vector configuration $\Aa$ becomes the set of $\Gamma$-series for the vector configuration $\Aa^\prime$ in the limit $z_i\rightarrow 0$ if the considered triangulations are compatible. Such a behaviour was also mentioned in \cite{GelfandGeneralHypergeometricSystems1992}. Thus, the meromorphic functions $C_{\sigma,k}(\nuu)$ for a Feynman integral with vector configuration $\Aa$ will be transmitted to the one for a vector configuration $\Aa^\prime$, and we can reduce the determination of the functions $C_{\sigma,k}(\nuu)$ to simpler situations.

Note, that those pairs of regular triangulations as considered in \cref{lem:linCombSubtriangs} always exist, as argued in the end of \cref{ssec:TriangulationsPolyhedra}. E.g.\ we can construct $\Tt$ from $\Tt^\prime$ as a placing triangulation with respect to the point labelled by $i$. Also the assumption of the choice of representatives $K_\sigma$ and $K^\prime_\sigma$ in \cref{lem:linCombSubtriangs} is always possible, as the quotient rings $\normalslant{\mathbb Z^{n+1}}{\mathbb Z \Aas}$ and $\normalslant{\mathbb Z^{n+1}}{\mathbb Z \Aas^\prime}$ are obviously the same for all $\sigma\not\ni i$. \bigskip


Thus, we will determine the meromorphic functions $C_{\sigma,k} (\nuu)$ by considering Euler-Mellin integrals corresponding to subtriangulations. In this way one can define ancestors and descendants of Feynman integrals by deleting or adding monomials to the Lee-Pomeransky polynomial $\Gg$. E.g.\ the massless one-loop bubble graph is a descendant of the one-loop bubble graph with one mass, which itself is a descendant of the fully massive one-loop bubble. Those ancestors and descendants do not necessarily correspond to Feynman integrals in the original sense, since one can also consider polynomials $\Gg$ which are not connected to graph polynomials any more. For an arbitrary acyclic vector configuration $\Aa$, we can choose any full dimensional simplex of a triangulation of $\Aa$ as a possible descendant. Hence, the trivial triangulation of that simplex and the triangulation of $\Aa$ are compatible in the sense of \cref{lem:linCombSubtriangs}. In doing so, one can relate the prefactors $C_{\sigma,k} (\nuu)$ to the problem where only one simplex is involved. This consideration results in the following theorem.

\vspace{2em}
\begin{theorem}[Series representation of Feynman integrals] \label{thm:FeynSeries}
    Let $\Tt$ be a regular triangulation of the Newton polytope $\Newt(\Gg)$ corresponding to a generalized Feynman integral $\mathcal I_\Aa(\nuu,z)$ and $\hatT$ its maximal cells. Then the generalized Feynman integral can be written as
    \begin{align}
        \mathcal I_\Aa(\nuu,z) = \frac{1}{\Gamma(\nu_0-\omega)\Gamma(\nu)} \sum_{\sigma\in \hatT} \frac{z_\sigma^{-\Aas^{-1}\nuu}}{\left|\det (\Aas)\right|} \sum_{\lambda\in\mathbb N_0^r} \frac{(-1)^{|\lambda|}}{\lambda!} \Gamma\!\left(\Aas^{-1}\nuu+\Aas^{-1}\Aabs\lambda\right) z_\sigma^{-\Aas^{-1}\Aabs\lambda} z_{\bar\sigma}^\lambda \label{eq:SeriesRepresentationFeynmanIntegral}
    \end{align}
    where $r=N-n-1$. In order to avoid poles of the Feynman integral, we will assume $\nuu\in\mathbb C^{n+1}$ to be very generic. All series in \cref{eq:SeriesRepresentationFeynmanIntegral} have a common region of absolute convergence. 
\end{theorem}
\begin{proof}
    As pointed out above, we only have to determine the meromorphic functions $C_{\sigma,k}(\nuu)$ which were defined by \cref{eq:LinearCombinationGammaSeries}. Due to \cref{lem:linCombSubtriangs} this can be accomplished by considering simpler cases. Especially, we want to reduce it to the case of simplices\footnote{The reduction to a simplex (i.e.\ $r=1$) does not satisfy the original assumptions which were considered in \cref{lem:linCombSubtriangs}. However, by duplicating an arbitrary vertex in $\Aa^\prime$ and adjusting the corresponding $z$-variable we can treat also simplices by \cref{lem:linCombSubtriangs}. Hence, we will split a monomial of $\Gg$ into two equal monomials with halved coefficients.}, i.e.\ we examine the limit $\zbarsigma \rightarrow 0$. Note that $\lim_{\zbarsigma\rightarrow 0} \partial^u \varphi_{\sigma,k}(\nuu,z)$ vanishes if and only if $0\in\Lambda_{k-u_{\bar\sigma}}$. This is satisfied if $k=u_{\bar\sigma}$, which follows directly from the definition of $\Lambda_{k-u_{\bar\sigma}}$. On the other hand, if $0\notin\Lambda_{k-u_{\bar\sigma}}$, $k$ and $u_{\bar\sigma}$ have to correspond to different equivalence classes described by $K_\sigma$. Therefore, we obtain by a choice of $u_\sigma=0$
    \begin{align}
	   \lim_{\zbarsigma \rightarrow 0} \partial^{(0,u_{\bar\sigma})} \varphi_{\sigma,k}(\nuu,z) = \delta_{k,u_{\bar\sigma}} \frac{\zsigma^{-\Aas^{-1}(\nuu+\Aabs k)}}{\Gamma\!\left(1- \Aas^{-1} ( \nuu + \Aabs k)\right)} \label{eq:proofSeriesA}
    \end{align}
    for all $u_{\bar\sigma},k\in K_\sigma$ where $\delta_{k,u_{\bar\sigma}}$ denotes the Kronecker delta.
    
    If $\Conv(\Aas)$ describes only a simplex, we can give a simple analytic expression for the Feynman integral $\mathcal I_{\Aas}$ by means of \cref{cor:MellinBarnesRepresentation}. For the derivative we obtain by \cref{eq:FeynmanJDerivative} and \cref{eq:MellinBarnesCorollary}
    \begin{align}
    	\lim_{\zbarsigma \rightarrow 0} \partial^{(0,u_{\bar\sigma})} \mathcal J_\Aa(\nuu,z) = (-1)^{|u_{\bar\sigma}|} \frac{\Gamma(\Aas^{-1}\nuu + \Aas^{-1}\Aabs u_{\bar\sigma})}{\left|\det (\Aas)\right|} \zsigma^{-\Aas^{-1}\nuu - \Aas^{-1}\Aabs u_{\bar\sigma}} \point \label{eq:proofSeriesB}
    \end{align}
    Comparing the coefficients of \cref{eq:proofSeriesA} and \cref{eq:proofSeriesB} we will get
    \begin{align}
    	\frac{(-1)^{|k|}}{\left|\det(\Aas)\right|} \Gamma\!\left(\Aas^{-1}\nuu + \Aas^{-1}\Aabs k\right) = \frac{C_{\sigma,k}(\nuu)}{\Gamma\!\left(1-\Aas^{-1}\nuu-\Aas^{-1}\Aabs k\right)}\point
    \end{align}
    By means of Euler's reflection formula, $\Aas^{-1}\Aabs(\lambda-k)\in\mathbb Z^{n+1}$ for all $\lambda\in\Lambda_k$ and \cref{lem:HomogenityRelationsForAa} we have the equality 
    \begin{align}
        \frac{\Gamma\!\left(\Aas^{-1} \nuu +\Aas^{-1}\Aabs k\right) \Gamma\!\left(1-\Aas^{-1}\nuu - \Aas^{-1}\Aabs k\right)}{\Gamma\!\left(\Aas^{-1}\nuu + \Aas^{-1}\Aabs\lambda\right)\Gamma\!\left(1-\Aas^{-1}\nuu-\Aas^{-1}\Aabs\lambda\right)} = (-1)^{\Aas^{-1}\Aabs (\lambda-k)} = (-1)^{|\lambda|-|k|} \point
    \end{align}   
    This results in
    \begin{align}
        \mathcal J_\Aa(\nuu,z) = \sum_{\sigma\in \hatT} \sum_{k\in K_\sigma} \frac{z_\sigma^{-\Aas^{-1}\nuu}}{\left|\det (\Aas)\right|} \sum_{\lambda\in\Lambda_k} \frac{(-1)^{|\lambda|}}{\lambda!} \Gamma\!\left(\Aas^{-1}\nuu+\Aas^{-1}\Aabs\lambda\right) z_\sigma^{-\Aas^{-1}\Aabs\lambda} z_{\bar\sigma}^\lambda \point
    \end{align}
    Since the summands have no $k$ dependence any more, we can combine the summation over $k$. As $K_\sigma$ describes a partition of $\mathbb N_0^r$ according to \cref{eq:LambdaPartition} we obtain \cref{eq:SeriesRepresentationFeynmanIntegral}. The common convergence region was shown in \cref{thm:GammaConverge}.
\end{proof}

\Cref{thm:FeynSeries} works for regular unimodular, as well as for regular non-unimodular triangulations and extends the results of \cite{KlausenHypergeometricSeriesRepresentations2019}. Although, there are less series in the non-unimodular case (we obtain a Horn hypergeometric series for every simplex in the triangulation) we will mostly prefer the unimodular case. For those triangulations we will always have $\Aas^{-1}\Aabs\in\mathbb Z^{(n+1)\times r}$, such that the summation indices will appear only as integer combinations in the $\Gamma$-functions. This will simplify the following treatments, e.g.\ the Laurent expansion in dimensional regularization. However, the rational combinations of summation indices appearing in the non-unimodular case are no general obstacle and can be cured by splitting the series into $|\det(\Aas)|$ many series with summation region $\Lambda_k$ for $k\in K_\sigma$.\bigskip

We want to emphasize that many Feynman integrals allow unimodular triangulations. Therefore, unimodular triangulations will be the relevant case in most applications. As we have not found any Feynman integral yet which does not allow any unimodular triangulation, we formulated the conjecture that all Feynman integrals enable unimodular triangulations \cite{KlausenHypergeometricSeriesRepresentations2019}. The only exception are certain momentum constraints, where monomials in the second Symanzik polynomial drop out. E.g.\ the fully massive $1$-loop bubble with $\Uu=x_1+x_2$ and $\Ff = (p^2+m_1^2+m_2^2) x_1 x_2 + m_1^2 x_1^2 + m_2^2 x_2^2$ does not allow an unimodular triangulation for the partial case $p^2 + m_1^2 + m_2^2 = 0$ and $m_1^2 \neq 0$, $m_2^2\neq 0$. As usual (see e.g.\ \cite{SchultkaToricGeometryRegularization2018}), we want to exclude those situations, where external momenta have specific, fixed values. Hence, we treat the external momenta as variables that are set to specific values only after renormalization.

The conjecture that any Feynman integral (without the aforementioned specific momentum constraints) allows at least one unimodular triangulation is neither proven nor disregarded. A first step to answer this question was done in \cite{TellanderCohenMacaulayPropertyFeynman2021} and \cite{WaltherFeynmanGraphsMatroids2022}. In these two articles the slightly weakened conjecture that all Feynman integrals generate Cohen-Macaulay ideals was shown for a wide class of Feynman graphs (see also \cref{thm:CohenMacaulayFeynman}). However, we want to emphasize, that the existence of unimodular triangulations is not necessary for series representations as seen in \cref{thm:FeynSeries}.\bigskip

In order to illustrate those series representations we will present two small examples. In \cref{sec:ExampleSunset} we will show a more extensive example to sketch the scope of this approach.
\begin{example} \label{ex:1loopbubbleB}
    To exemplify the series representation in the unimodular case, we want to continue the \cref{ex:1loopbubbleA} of the $1$-loop bubble graph with one mass corresponding to \cref{fig:bubble1}. The vector configuration for this Feynman graph was given by
    \begin{align}
        \Aa = \begin{pmatrix} 
            1 & 1 & 1 & 1 \\
            1 & 0 & 1 & 2 \\
            0 & 1 & 1 & 0
        \end{pmatrix} \qquad z = (1,1,p^2+m_1^2,m_1^2) \point
    \end{align}
    For the triangulation $\widehat\Tt_1 = \big\{\{1,2,4\},\{2,3,4\}\big\}$ one obtains the series representation
    \begin{align}
        \mathcal J_\Aa(\nuu,z) &= z_1^{-2 \nu_0 + \nu_1 + 2 \nu_2} z_2^{-\nu_2} z_4^{\nu_0 - \nu_1 - \nu_2} \sum_{\lambda\in\mathbb N_0} \frac{1}{\lambda !} \left(-\frac{z_1 z_3}{z_2 z_4}\right)^\lambda \Gamma (\nu_2+\lambda ) \nonumber\\
        &\qquad\qquad \Gamma (2 \nu_0-\nu_1-2 \nu_2-\lambda ) \Gamma (-\nu_0+\nu_1+\nu_2+\lambda ) \nonumber  \\
        &\qquad + z_4^{\nu_2-\nu_0} z_2^{-2 \nu_0+\nu_1+\nu_2} z_3^{2 \nu_0-\nu_1-2 \nu_2} \sum_{\lambda\in\mathbb N_0} \frac{1}{\lambda !} \left(-\frac{z_1 z_3}{z_2 z_4}\right)^\lambda  \Gamma (\nu_0-\nu_2+\lambda ) \nonumber \\
        &\qquad\qquad \Gamma (-2 \nu_0+\nu_1+2 \nu_2-\lambda ) \Gamma (2 \nu_0-\nu_1-\nu_2+\lambda ) \point \label{eq:exBubble1MassSeriesRepr}
    \end{align}
    In the physically relevant case $z = (1,1,p^2+m_1^2,m_1^2)$ and $\nuu = (2-\epsilon,1,1)$ one can easily evaluate the series
    \begin{align}
        &\mathcal I_\Aa\!\left(2-\epsilon,1,1,1,1,m_1^2+p^2,m_1^2\right) = (m_1^2)^{-\epsilon} \frac{\Gamma(1-2\epsilon)\Gamma(\epsilon)}{\Gamma(2-2\epsilon)}\ \HypF{1,\epsilon}{2\epsilon}{\frac{ m_1^2+p^2}{m_1^2}} \nonumber \\
        &\qquad +  \left(m_1^2+p^2\right)^{1-2\epsilon} \left(-p^2\right)^{-1+\epsilon} \Gamma(1-\epsilon) \Gamma(-1+2\epsilon) \label{eq:Example1loopBubble2F1}
    \end{align}
    which agrees with the expected result. The power series representation which can be obtained by the triangulation $\widehat\Tt_2=\big\{\{1,2,3\},\{1,3,4\}\big\}$ will contain $-\frac{z_2 z_4}{z_1z_3}$ as its argument. Hence, depending on the kinematic region we like to consider, it is possibly worthwhile for numerical issues to take a different triangulation. For a region where $0 \ll |p^2|$ we should prefer $\Tt_2$, whereas $\Tt_1$ is convenient when $|p^2| \approx m_1^2$. The series representation for $\Tt_2$, as well as the former result in \cref{ex:1loopbubbleA}, are equivalent to \cref{eq:exBubble1MassSeriesRepr} by transformation rules of the ${_2}F_1$ function \cref{eq:2F1trafoA}.
\end{example}

\begin{example} \label{ex:SeriesRepresentationNonUnimod}
	To illustrate how \cref{thm:FeynSeries} works in the non-unimodular case we want to consider the $1$-loop bubble graph with two masses. Hence, we will have the Lee-Pomeransky polynomial $\Gg=\Uu+\Ff$
	\begin{align}
		\Gg = x_1 + x_2 + \left(p^2 + m_1^2 + m_2^2\right) x_1 x_2 + m_1^2 x_1^2 + m_2^2 x_2^2 \comma
	\end{align}
	where the corresponding data for $\Aa$ and $z$ are 
	 \begin{align}
        \Aa = \begin{pmatrix} 
            1 & 1 & 1 & 1 & 1\\
            1 & 0 & 1 & 2 & 0 \\
            0 & 1 & 1 & 0 & 2
        \end{pmatrix} \qquad z = (1,1,p^2+m_1^2+m_2^2,m_1^2,m_2^2) \point
    \end{align}
    The vector configuration $\Aa$ has five (regular) triangulations, where two of them are non-unimodular. We will choose the triangulation $\hatT = \big\{\{1,2,5\},\{1,4,5\}\big\}$, where the simplex $\sigma_2 = \{1,4,5\}$ has volume $2$. By means of those data we can put together the series representation according to \cref{eq:SeriesRepresentationFeynmanIntegral}
    \begin{align}
    	\mathcal J_\Aa (\nuu,z) &= z_1^{-\nu_1} z_2^{-2\nu_0 + 2\nu_1 + \nu_2} z_5^{\nu_0 - \nu_1 - \nu_2} \sum_{\lambda\in\mathbb N_0^2} \frac{1}{\lambda!} \left(-\frac{z_2 z_3}{z_1 z_5}\right)^{\lambda_1} \!\left(-\frac{z_2^2 z_4}{z_1^2 z_5}\right)^{-\lambda_2} \!\! \Gamma (\nu_1 + \lambda_1 + 2 \lambda_2) \nonumber \\
    	&\quad \Gamma(2 \nu_0 - 2 \nu_1 - \nu_2 -\lambda_1 - 2 \lambda_2) \Gamma(- \nu_0 + \nu_1 + \nu_2 + \lambda_1 + \lambda_2) \nonumber \\
    	& + \frac{1}{2} z_1^{-2\nu_0 + \nu_1 + \nu_2} z_4^{\nu_0 - \nu_1 - \frac{\nu_2}{2}} z_5^{-\frac{\nu_2}{2}} \sum_{\lambda\in\mathbb N_0^2} \frac{1}{\lambda!} \left(-\frac{z_2 \sqrt{z_4}}{z_1 \sqrt{z_5}}\right)^{\lambda_1} \!\left(-\frac{z_3}{\sqrt{z_4 z_5}}\right)^{\lambda_2}  \nonumber \\
    	&\quad \Gamma\!\left(2\nu_0 - \nu_1 - \nu_2 + \lambda_1\right) \Gamma\!\left(\frac{\nu_2}{2} + \frac{\lambda_1}{2} + \frac{\lambda_2}{2} \right) \Gamma\!\left(- \nu_0 + \nu_1 + \frac{\nu_2}{2} - \frac{\lambda_1}{2} + \frac{\lambda_2}{2} \right)   \point
    \end{align}
    For the physical case $z=(1,1,p^2+m_1^2+m_2^2,m_1^2,m_2^2)$, $\nuu=(2-\epsilon,1,1)$ we can rewrite
      \begin{align}
    	\mathcal I_\Aa &= \frac{1}{\Gamma(2-2\epsilon)} \left\{ 
    	 z_5^{-\epsilon} \sum_{\lambda\in\mathbb N_0^2} \frac{(\lambda_1 + 2 \lambda_2)!}{\lambda!}  \Gamma(1-2\epsilon -\lambda_1 - 2 \lambda_2) \Gamma(\epsilon + \lambda_1 + \lambda_2) \frac{\left(-y_1\right)^{\lambda_1}}{\left(-y_2\right)^{\lambda_2}} \right. \nonumber \\
    	 & + \left. \frac{z_4^{\frac{1}{2}-\epsilon}  }{2 \sqrt{z_5}} \sum_{\lambda\in\mathbb N_0^2} \frac{(-1)^\lambda}{\lambda!} \Gamma\!\left(2-2\epsilon + \lambda_1\right) \Gamma\!\left(\frac{1+\lambda_1+\lambda_2}{2}\right) \Gamma\!\left(-\frac{1}{2} + \epsilon + \frac{\lambda_2 - \lambda_1}{2} \right)  y_2^{\frac{\lambda_1}{2}} y_3^{\frac{\lambda_2}{2}} \right\}
    \end{align}
    with $y_1 = \frac{z_3}{z_5} = \frac{p^2+m_1^2+m_2^2}{m_2^2}$, $y_2 = \frac{z_4}{z_5} = \frac{m_1^2}{m_2^2}$ and $y_3 = \frac{z_3^2}{z_4 z_5} = \frac{(p^2+m_1^2+m_2^2)^2}{m_1^2m_2^2}$. In order to avoid half-integer summation indices in the second series one can split this series into two parts as indicated by \cref{eq:representativesK}. Thus, we will have $\mathbb Z\Aa_{\sigma_2} = (\mathbb Z,\mathbb Z,2\mathbb Z)^\top$ for the ideal spanned by $\Aa_{\sigma_2}$ and therefore $\normalslant{\mathbb Z\Aa}{\mathbb Z\Aa_{\sigma_2}} = \left\{ [0,0,0]^\top ; [1,0,1]^\top\right\}$ which results in $K_{\sigma_2} = \left\{ (0,0)^\top ; (1,0)^\top \right\}$ as a possible choice. Hence, we have to split the second series into the two summation regions $\Lambda_{k_1} = \begin{Bsmallmatrix} 2\mathbb N_0 \\ 2\mathbb N_0 \end{Bsmallmatrix} \cup \begin{Bsmallmatrix} 2\mathbb N_0 + 1\\ 2\mathbb N_0 + 1\end{Bsmallmatrix}$ and  $\Lambda_{k_2} = \begin{Bsmallmatrix} 2\mathbb N_0 + 1\\ 2\mathbb N_0 \end{Bsmallmatrix} \cup \begin{Bsmallmatrix} 2\mathbb N_0\\ 2\mathbb N_0 + 1\end{Bsmallmatrix}$. This is a partition of $\mathbb N_0^2$ which transforms the half-integer combinations of summation indices into integer combinations. 
\end{example}

Thus, we found series representations for generalized Feynman integrals for every regular triangulation of $\Newt(\Gg)$. We want to recall from \cref{ssec:TriangulationsPolyhedra} that every point/vector configuration has at least one regular triangulation. Typically, a Feynman graph admits many different possibilities to triangulate its corresponding Newton polytope. Therefore, one usually obtains a large number of those series representations. We included the number of triangulations for certain Feynman graphs in the appendix in \cref{tab:characteristics}. It is not surprising that there are a lot of different series representations, since hypergeometric functions satisfy many transformation formulas and can be converted to other hypergeometric functions. Therefore, for the sake of numerical computations one can choose a series representation, which converges fast for the given kinematics, such that we can evaluate the Feynman integral numerically by considering the first summands of every series. Thereby, one can directly construct a convenient triangulation by means of the height (see \cref{ssec:TriangulationsPolyhedra}).

\begin{lemma} \label{lem:heightEffectiveVars}
	Let $\omega=-\ln |z| = (-\ln |z_1|,\ldots,-\ln |z_N|)$ be a height vector of a regular triangulation $\Tt$ of an acyclic vector configuration $\Aa$. Then the effective variables $y=\zbarsigma \zsigma^{-\Aas^{-1}\Aabs}\in\mathbb C^r$ for every simplex $\sigma \in\hatT$ will satisfy $|y_j|<1$ for $j=1,\ldots,r$.
\end{lemma}
\begin{proof}
	According to \cref{eq:subdivisionCondGaleBeta} we have 
	\begin{align}
		\sigma \in \Ss\!\left(\Aa,-\ln |z|\right) \Leftrightarrow - \ln |z| \cdot \Bb(\sigma) > 0
	\end{align}
	with $\Bb(\sigma) = \begin{psmallmatrix} -\Aas^{-1}\Aabs \\ \mathbbm 1 \end{psmallmatrix}$ a Gale dual of $\Aa$ constructed from the simplex $\sigma$ as before. Since $y = z^{\Bb(\sigma)}$ we arrive at the assertion.
\end{proof}

Note that height vectors $\omega$ will produce in general only a regular subdivision. Only for generic heights $\omega$ the regular subdivision $\Ss(\Aa,\omega)$ will be a triangulation. Thus, for specific values $z$ we will not always find an optimal triangulation. It will not be guaranteed that for certain values of $z$, we will find a triangulation which results in power series with all $|y_j|<1$. In this case the series have to be continued analytically by means of hypergeometric transformation rules, which we consider in the subsequent \cref{sec:AnalyticContinuation}. Additionally, it is often worthwhile to use classical series accelerations as e.g.\ Euler's transformation to speed up the numerical calculation of those series. \Cref{sec:numerics} will take a closer look on numerical issues.

However, \cref{lem:heightEffectiveVars} will give us a simple guideline to construct appropriate triangulations, and we can perturb $\omega=-\ln|z|$ slightly in order to generate an almost optimal triangulation, whenever $-\ln|z|$ is not generic enough.


\section{Analytic continuation of series representations} \label{sec:AnalyticContinuation}
%
%

Up to this point, most of the theorems were statements about the generalized Feynman integral, i.e.\ we assumed the coefficients of Symanzik polynomials $z$ to be generic. When we want to consider non-generic values of $z$, as necessary in a physical application, we will have two possible strategies. The first option is to specify the $\Aa$-hypergeometric system, calculating then the restriction to a coordinate subspace and find solutions of this restricted system. In principle, it is known how to restrict $\Aa$-hypergeometric systems to coordinate subspaces and one can find a comprehensive description\footnote{Typically, one considers restrictions where variables $z_j=0$ vanish. This is no general obstacle as one can change the $D$-ideal by a convenient coordinate transformation $z_j\mapsto z_j-1$.} in \cite[sec. 5.2]{SaitoGrobnerDeformationsHypergeometric2000}. However, it can be algorithmically very hard to calculate those restriction ideals and for practical applications these algorithms seem hopelessly slow. Another drawback of that strategy is that we will lose the information about the explicit form of the meromorphic functions $C_{\sigma,k}(\nuu)$ from \cref{thm:FeynSeries}. We refer to \cite{WaltherAlgorithmicComputationRham1998,OakuLocalizationAlgorithmModules1999,OakuAlgorithmsDmodulesRestriction1998} for algorithms as well as the included algorithms in \softwareName{Singular} \cite{GreuelSINGULARComputerAlgebra2009}.

A similar strategy was also used in \cite{KlemmLoopBananaAmplitude2020} for the computation of marginal banana graphs. Instead of computing the restriction ideal by the above-mentioned algorithms, the authors made ansatzes for the restriction ideal based on the first terms of the series solutions of the generic case.\bigskip

The second option is to calculate solutions of the generic $\Aa$-hypergeometric system and specify the $\Aa$-hypergeometric functions afterwards by analytic continuation. This way is substantially faster, as one can use well known relations of hypergeometric functions for the implementation of the analytic continuation. Hence, we will consider the series representation due to \cref{thm:FeynSeries} and calculate the limit from generic $z$ to the physically relevant values as specified by the Symanzik polynomials. In particular, we will set the variables corresponding to the first Symanzik polynomial $\Uu$ equal to $1$ and certain variables related to the second Symanzik polynomial $\Ff$ may get the same value. In this limit the convergence behaviour of the $\Gamma$-series can be changed. Consider a region $D\subseteq \mathbb C^{n+1}$ where the Feynman integral \cref{eq:LeePomeranskyRepresentation} has no poles for $\nuu\in D$ and assume that the masses and momenta of the Feynman integral do not correspond to a Landau singularity (see \cref{ch:singularities}). Due to \cref{thm:meromorphicContinuation} the Feynman integral then has an analytic continuation and hence also the linear combination of $\Gamma$-series \cref{eq:LinearCombinationGammaSeries} has a (finite) analytic continuation. Thus, in the limit from generic to physical values there can arise only two issues: a) every series converges separately, but they do not have a common convergence region anymore or b) certain $\Gamma$-series diverge, but the linear combination is still finite. In the first case a) the convergence criteria for the effective variables of the $\Gamma$-series $y = z_{\bar\sigma} z_\sigma^{-\Aas^{-1} \Aabs}$ contradict each other for the different simplices $\sigma\in \hatT$ due to the specification of the variables $z$. In the second case b) certain variables $y_j$ become constant (usually equals $1$) after the limit to physical values of $z$, which can be outside of the convergence region. In practice, problem a) can be usually avoided by the choice of an appropriate triangulation (see \cref{lem:heightEffectiveVars}). However, in the following we will discuss this case equally since it is to be solved analogously to b).

Note that the analytic continuation with respect to the variables $z$ of the Feynman integral with $\nuu\in D$ is unique as long as the variables $z$ are in the Euclidean region $\Re(z_j)>0$. Therefore, we will assume Euclidean kinematics in this chapter in order that also the $\Gamma$-series have a unique analytic continuation. Outside of the Euclidean region, the Feynman integral will be potentially multi-valued, which will be discussed in \cref{ch:singularities}. 

Fortunately, analytic continuations of $\Gamma$-series can be calculated explicitly by means of well-known transformation formulas of standard hypergeometric functions. We can always arrange the $\Gamma$-series in an arbitrary way from inner to outer power series. Therefore, we can sort these Horn hypergeometric series such that all series with certain convergence issues appear as the innermost series and focus on them separately. Note that every of these series parts is itself a Horn hypergeometric series. Let us start by discussing the case in which this inner series is a one-dimensional power series in one variable. We can distinguish then between the situation where we have to transform $y_j \mapsto \frac{1}{y_j}$ to solve issues coming from case a) or we have to transform $y_j \mapsto y_j-1$ to solve the problem b). Hence, we have to look at transformation formulas of hypergeometric functions of those types. Relationships between hypergeometric functions have always played an important role in the Feynman calculus, and we refer exemplarily to \cite{KniehlFindingNewRelationships2012}.\bigskip

We will start with the simplest, but also most often appearing case. Hence, consider that the innermost series where convergence issues appear is of the type of a ${_2}F_1$ hypergeometric function. Those transformation formulas are well known \cite[sec. 15.8]{OlverNISTHandbookMathematical2010}
\begin{align}
    \HypF{a,b}{c}{t} &= \frac{\Gamma(c)\Gamma(b-a)}{\Gamma(b)\Gamma(c-a)} (-t)^{-a} \HypF{a,a-c+1}{a-b+1}{\frac{1}{t}} \nonumber \\
    & + \frac{\Gamma(c)\Gamma(a-b)}{\Gamma(a)\Gamma(c-b)} (-t)^{-b} \HypF{b,b-c+1}{b-a+1}{\frac{1}{t}} \label{eq:2F1trafoA} \\
    \HypF{a,b}{c}{t} &= \frac{\Gamma(c)\Gamma(c-a-b)}{\Gamma(c-a)\Gamma(c-b)} \HypF{a,b}{a+b-c+1}{1-t} \nonumber \\
    & + (1-t)^{c-a-b} \frac{\Gamma(c)\Gamma(a+b-c)}{\Gamma(a)\Gamma(b)} \HypF{c-a,c-b}{c-a-b+1}{1-t} \point\label{eq:2F1trafoB}
\end{align}
Note that the analytic continuations by means of those formulas are again of Horn hypergeometric type, as the prefactors of \cref{eq:2F1trafoA} and \cref{eq:2F1trafoB} consist in $\Gamma$-functions.

We will illustrate the application of those transformation formulas with an example. 
\begin{example}
    For the $2$-loop sunset graph with two different masses in dimensional regularization, inter alia there appears the $\Gamma$-series
    \begin{align}
        \phi_2 =& \sum_{k\in\mathbb N_0^4} (1-\epsilon)_{k_3+k_4} (\epsilon)_{k_1+2 k_2+k_3} (\epsilon -1)_{-k_1-k_2+k_4} (2-2 \epsilon)_{k_1-k_3-k_4} \nonumber \\
        &\qquad \frac{1}{k_1!\, k_2!\, k_3!\, k_4!}  \left(-\frac{z_1 z_6}{z_5 z_2}\right)^{k_1} \!\left(-\frac{z_4 z_6}{z_5^2}\right)^{k_2} \! \left(-\frac{z_2 z_7}{z_3 z_5}\right)^{k_3} \! \left(-\frac{z_2 z_8}{z_3 z_6}\right)^{k_4} \label{eq:example2MassesSunsetHypTrafoPhi2}
    \end{align}
    where one has to consider the limit $(z_1,z_2,z_3,z_4,z_5,z_6,z_7,z_8)\rightarrow (1,1,1,m_2^2,m_1^2+m_2^2+p_1^2,m_1^2,m_2^2,m_1^2)$. For simplicity, we dropped prefactors in \cref{eq:example2MassesSunsetHypTrafoPhi2} and $(a)_n = \Gamma(a+n)/\Gamma(a)$ denotes the Pochhammer symbol as before. In this limit we will have  $\left(-\frac{z_2 z_8}{z_3 z_6}\right)^{k_4} \rightarrow \left(-1\right)^{k_4}$, which is not in the convergence region for small values of $\epsilon>0$ any more. Therefore, we evaluate the $k_4$ series carefully and write
    \begin{align}
        \phi_2 &= \lim_{t\rightarrow 1} \sum_{(k_1,k_2,k_3)\in\mathbb N_0^3} (1-\epsilon)_{k_3}  (\epsilon-1)_{-k_1-k_2} (2-2 \epsilon)_{k_1-k_3} (\epsilon)_{k_1+2 k_2+k_3} \frac{1}{k_1!\, k_2!\, k_3!} \nonumber \\
        & \qquad\left(-y_1\right)^{k_1} \left(-y_1 y_2\right)^{k_2} \left(-y_2\right)^{k_3}  \HypF{-\epsilon+k_3+1,\epsilon-k_1-k_2-1}{2 \epsilon-k_1+k_3-1}{t}
    \end{align}
    where $y_i=\frac{ m_i^2}{m_1^2+m_2^2+p_1^2}$. With the transformation formula \cref{eq:2F1trafoB} for the ${_2}F_1$ hypergeometric function, one can split the series in a convergent and a divergent part
    \begin{align}
        &\phi_2 = \!\sum_{(k_1,k_2,k_3)\in\mathbb N_0^3} \! \frac{\Gamma (k_2+2 \epsilon -1) \Gamma (-k_1+k_3+2 \epsilon -1) }{ \Gamma (-k_1+3 \epsilon -2) \Gamma (k_2+k_3+\epsilon )} (1-\epsilon )_{k_3} (\epsilon -1)_{-k_1-k_2} (2-2 \epsilon )_{k_1-k_3} \nonumber \\
        & \quad (\epsilon )_{k_1+2 k_2+k_3} \frac{(-y_1)^{k_1} (-y_2)^{k_3} (-y_1 y_2)^{k_2}}{k_1!\, k_2!\, k_3!}  + \lim_{t\rightarrow 1} \sum_{(k_1,k_2,k_3,k_4)\in\mathbb N_0^4} \frac{(1-t)^{k_2+k_4+2 \epsilon -1} }{k_1!\, k_2!\, k_3!\, k_4! } \nonumber \\
        & \quad \frac{\Gamma (1 - 2\epsilon - k_2) \Gamma (2\epsilon - 1 - k_1 + k_3) \Gamma (\epsilon + k_2 + k_3 + k_4) \Gamma (2\epsilon + k_2) \Gamma (3\epsilon - 2 - k_1 + k_4)}{ \Gamma (1 - \epsilon + k_3) \Gamma (-1 + \epsilon - k_1 - k_2) \Gamma (\epsilon + k_2 + k_3) \Gamma(2\epsilon + k_2 + k_4)\Gamma (- 2 + 3\epsilon - k_1)} \nonumber \\
        & \quad (\epsilon )_{k_1+2 k_2+k_3}  (1-\epsilon )_{k_3} (\epsilon -1)_{-k_1-k_2} (2-2 \epsilon )_{k_1-k_3} (-y_1)^{k_1} (-y_2)^{k_3} (-y_1 y_2)^{k_2} \nonumber  \allowdisplaybreaks \\
        & =\sum_{(k_1,k_2,k_3)\in\mathbb N_0^3}  \frac{\Gamma (k_2+2 \epsilon -1) \Gamma (-k_1+k_3+2 \epsilon -1) }{ \Gamma (-k_1+3 \epsilon -2) \Gamma (k_2+k_3+\epsilon )} (1-\epsilon )_{k_3} (\epsilon -1)_{-k_1-k_2} (2-2 \epsilon )_{k_1-k_3} \nonumber \\
        & \quad (\epsilon )_{k_1+2 k_2+k_3}  \frac{(-y_1)^{k_1} (-y_2)^{k_3} (-y_1 y_2)^{k_2}}{k_1!\, k_2!\, k_3!}  + \lim_{t\rightarrow 1} (1-t)^{2 \epsilon -1} \!\sum_{(k_1,k_3)\in\mathbb N_0^2}\! \frac{  (-y_1)^{k_1} (-y_2)^{k_3}}{k_1!\, k_3!} \nonumber \\
        & \quad \frac{\Gamma (-2 \epsilon +1)  \Gamma (-k_1+k_3+2 \epsilon -1)}{ \Gamma (1-\epsilon) \Gamma (\epsilon -1)} (\epsilon )_{k_1+k_3}   (2-2 \epsilon )_{k_1-k_3} \point
    \end{align}
    Note that all Horn hypergeometric functions at zero are trivial to evaluate, and we have $\HypF{a,b}{c}{0}=1$. Comparing the divergent part with the other $\Gamma$-series, which occur in the calculation of the sunset graph with two masses, one can find another divergent series which exactly cancels this divergence. This cancellation always has to happen, since the linear combination has to be finite.
\end{example}

Although, those identities of the ${_2}F_1$ functions will fix the most issues appearing in practice, there potentially appear more complicated situations. The next more general type of innermost series are the ${}_{p+1}F_p$ hypergeometric functions. Those transformation formulas were considered in \cite{BuhringGeneralizedHypergeometricFunctions1992}. We have the recursion
\begin{align}
	\HypFpq{p+1}{p}{a_1,\ldots,a_{p+1}}{b_1,\ldots,b_p}{t} = \sum_{k=0}^\infty \tilde A_k^{(p)}(a,b)\ \HypF{a_1,a_2}{|b|-|a|+a_1+a_2+k}{t} \label{eq:pFqReduction}
\end{align}
where $\tilde A_k^{(p)}(a,b)$ are rational functions in $a$ and $b$. Expressions of those functions $\tilde A_k^{(p)}(a,b)$ can be found in \cite[sec. 2]{BuhringGeneralizedHypergeometricFunctions1992}. Hence, the case where we need transformation formulas of ${}_{p+1}F_p$ functions can be reduced to the ${_2}F_1$ case.\bigskip

As a next step we consider any Horn hypergeometric function depending on one variable, where the summation index appears only as an integer multiple in the $\Gamma$-functions. By the identities of Pochhammer symbols
\begin{align}
	(a)_{-n} &= (-1)^n \frac{1}{(1-a)_n} \qquad \text{for } n\in\mathbb Z \label{eq:PochhammerIdentityA} \\ 
	(a)_{m+n} &= (a)_m (a+m)_n \label{eq:PochhammerIdentityB} \\
	(a)_{mn} &= m^{mn} \prod_{j=0}^{m-1} \left(\frac{a+j}{m}\right)_n \qquad \text{for } m\in\mathbb Z_{>0} \label{eq:PochhammerIdentityC}
\end{align}
one can always bring that Horn hypergeometric series in the form of an ${}_pF_q$ function. Due to \cref{lem:HomogenityRelationsForAa}, in particular equation \cref{eq:AA1a}, there will only appear functions with $p=q+1$, which relates back to the previous discussed situation.

As the most general case, we will expect Horn hypergeometric functions in many variables. Again, we will assume that summation indices only appear as linear integer combinations, which is always the case for unimodular triangulations and which can be achieved for non-unimodular triangulations as well by a convenient splitting of the series. The analytic continuation of those functions can be accomplished iteratively, as the analytically continuated functions will be again from hypergeometric type. This is due to the rationality of the coefficients in \cref{eq:2F1trafoA}, \cref{eq:2F1trafoB} and \cref{eq:pFqReduction}.\bigskip

Therefore, one can derive convergent series representations also for physical Feynman integrals by considering convenient transformation formulas of hypergeometric functions. Since those transformation formulas consist in a simple replacement of $\Gamma$-functions, the analytic continuation can be done algorithmically efficient.\bigskip

We want to remark that the limit from generic to specific values of $z$ may reduce the dimension of $\Sol(H_\Aa(\nuu))$. Hence, it may appear that $\Gamma$-series become linearly dependent after such a limit (as e.g.\ \cref{sec:ExampleSunset}). However, this will be no general obstacle and shows only that the series representations are not necessarily optimal in the sense that they are not always the simplest possible series representation of the given Feynman integral. Or in other words, the holonomic rank of the restriction ideal may differ from the one of the original ideal. We refer to \cite{BitounFeynmanIntegralRelations2019}, where a similar decreasing of the dimension of a solution space was described, which also occurred in the specification from generic to physical values of $z$.


\section{Laurent expansion of hypergeometric series} \label{sec:epsilonExpansion}


As outlined in \cref{sec:DimAnaReg}, one has to renormalize Feynman integrals to fix certain ambiguities. In addition, renormalization removes also the divergences in Feynman integrals. Hence, in this process it is necessary for the most renormalization schemes to handle the singularities of the Feynman integrals by regularization. Therefore, in the widely spread dimensional and analytic regularization (see \cref{sec:DimAnaReg}) one is rather interested in the Laurent expansion of a Feynman integral around certain integer values of $\nuu$ instead in the Feynman integral itself. Namely, in dimensional regularization we assume all indices to be integer values and consider the spacetime dimension to be close to an integer, i.e.\ we set $\nu_i\in\mathbb Z_{>0}$ for $i=1,\ldots,n$, $\nu_0=\dhalf = \frac{D}{2} - \epsilon$ with $\frac{D}{2}\in\mathbb Z_{>0}$ (usually $D=4$) and expand around $\epsilon =0$. In analytic regularization one would instead fix $\nu_0 = \frac{D}{2}\in\mathbb Z_{>0}$ and introduce a single parameter $\epsilon$ controlling the distance of $\nu\in(\mathbb C\setminus\mathbb Z)^n$ to integer values.\bigskip

Due to the \cref{thm:meromorphicContinuation} one can relate this task to the Taylor expansion of the hypergeometric series representation. Thus, one has simply to differentiate the Horn hypergeometric series with respect to their parameters $\nuu$. As pointed out in \cite{BytevDerivativesHorntypeHypergeometric2017}, those derivatives of Horn hypergeometric series are again Horn hypergeometric series of higher degree. By the identities of Pochhammer symbols \cref{eq:PochhammerIdentityA}, \cref{eq:PochhammerIdentityB} and \cref{eq:PochhammerIdentityC} one can reduce all\footnote{We will assume for simplicity that all summation indices in \cref{eq:SeriesRepresentationFeynmanIntegral} appear only as integer combinations. This is true for all unimodular triangulations due to $\Aas^{-1}\Aabs\in\mathbb Z^{(n+1)\times r}$. For non-unimodular triangulations one can always find a partition of the summation region to arrive at this situation, similar to \cref{ex:SeriesRepresentationNonUnimod}.} derivatives to two cases \cite{BytevDerivativesHorntypeHypergeometric2017}
\begin{align}
    \frac{\partial}{\partial \alpha} \sum_{n=0}^\infty f(n) (\alpha)_n t^n &= t \sum_{k=0}^\infty\sum_{n=0}^\infty f(n+k+1) \frac{(\alpha+1)_{n+k} (\alpha)_k}{(\alpha+1)_k}\ t^{n+k} \label{eq:HornDerivativeA} \\
    \frac{\partial}{\partial \alpha} \sum_{n=0}^\infty f(n) (\alpha)_{-n} t^n &= -t \sum_{k=0}^\infty\sum_{n=0}^\infty f(n+k+1) \frac{(\alpha)_{-n-k-1} (\alpha)_{-k-1}}{(\alpha)_{-k}}\ t^{n+k} \point \label{eq:HornDerivativeB}
\end{align}
Thus, Horn hypergeometric functions do not only appear as series representations of Feynman integrals, but also in every coefficient of the Laurent expansion of those Feynman integrals. Therefore, the class of Horn hypergeometric functions is sufficient to describe all Feynman integrals as well as their Laurent expansions. \bigskip

However, it seems much more efficient to expand the $\Gamma$-functions in the series representation around $\epsilon = 0$ instead of a determination of the derivatives by \cref{eq:HornDerivativeA} and \cref{eq:HornDerivativeB}. For this purpose we want to introduce the (unsigned) \textit{Stirling numbers of the first kind}. Those numbers $\StirlingFirstSmall{n}{k}$ \glsadd{StirlingFirst}\glsunset{StirlingFirst} count the permutations of $\{1,\ldots,n\}$ with precisely $k$ cycles. We can define them recursively by
\begin{align}
	\StirlingFirst{n+1}{k} = n \StirlingFirst{n}{k} + \StirlingFirst{n}{k-1} \quad\text{with}\quad \StirlingFirst{0}{0} = 1, \quad \StirlingFirst{n}{0} = \StirlingFirst{0}{k} = 0 \quad \text{for } n,k\in\mathbb N_{>0} \point \label{eq:StirlingFirstRecursion}
\end{align}
Stirling numbers are related to many other functions. E.g.\ they can be expressed by the $Z$-sums invented in \cite{MochNestedSumsExpansion2002}, which are in this case also known as Euler-Zagier sums appearing in the study of multiple zeta values \cite{ZagierValuesZetaFunctions1994}. We will consider this relation more in detail in \cref{sec:ManipulationSeries}. Furthermore, the generating function of Stirling numbers are the Pochhammer symbols
\begin{align}
	(\alpha)_n = \prod_{j=0}^{n-1} (\alpha+j) = \sum_{k=0}^n \StirlingFirst{n}{k} \alpha^k \label{eq:StirlingPochhammer} \point
\end{align}
In particular, we have for $n\in\mathbb N_{>0}$
\begin{align}
	&\StirlingFirst{n}{0} = 0\, , \quad \StirlingFirst{n}{1} = (n-1)!\, , \quad \StirlingFirst{n}{2} = (n-1)!\, H_{n-1}\, , \quad \StirlingFirst{n}{3} = \frac{(n-1)!}{2!}  \left(H_{n-1}^2 - H_{n-1}^{(2)}\right) \label{eq:StirlingPartPos}
\end{align}
where $\gls{harmonicNumber} := \sum_{i=1}^n \frac{1}{i^k}$ denotes the $k$-th harmonic number and $H_{n} := H_n^{(1)}$. We collected further relations in \cref{sec:StirlingNumbers}. Those relations can be derived efficiently by a recursion with respect to $k$ \cite{AdamchikStirlingNumbersEuler1997}
\begin{align}
	&\StirlingFirst{n}{k} = \frac{(n-1)!}{(k-1)!}\, w(n,k-1) \nonumber \\
	&\text{with } w(n,0) = 1 \, , \quad w(n,k) = \sum_{i=0}^{k-1} (1-k)_i\, H_{n-1}^{(i+1)}\, w(n,k-1-i) \point \label{eq:StirlingNumbersRecursionPos}
\end{align}
By means of \cref{eq:StirlingPochhammer} one can extend the Stirling numbers also to negative $n$. In this case we will have \cite{BransonExtensionStirlingNumbers1996, KnuthTwoNotesNotation1992}
\begin{align}
	\StirlingFirst{-n}{k} = \frac{(-1)^n}{n!} \sum_{i=1}^n \frac{(-1)^{i+1}}{i^k} \begin{pmatrix} n \\ i \end{pmatrix} \quad\text{for } n\in\mathbb N_{\geq 0}, k\in\mathbb N_{>0} \label{eq:StirlingFirstNegativeDef}
\end{align}
and in particular
\begin{align}
	&\StirlingFirst{-n}{0} = \frac{(-1)^n}{n!}\, , \quad \StirlingFirst{-n}{1} = \frac{(-1)^n}{n!} H_n\, , \quad \StirlingFirst{-n}{2} = \frac{(-1)^n}{2!\,n!} \left(H_{n}^2 + H_{n}^{(2)}\right) \label{eq:StirlingPartNeg} \point
\end{align}
This list will be continued in \cref{sec:StirlingNumbers}. The Stirling numbers with negative and positive integers $n$ are the natural analytic continuations of each other. Thus, one can show that definition \cref{eq:StirlingFirstNegativeDef} satisfies the recursion \cref{eq:StirlingFirstRecursion}. Hence, when writing \cref{eq:StirlingPartNeg} by means of the more general functions $n! = \Gamma(n+1)$, $H_n=\psi^{(0)} (n+1) + \gamma$ and $H_n^{(r)} = \zeta(r) - \frac{(-1)^r}{(r-1)!} \psi^{(r-1)} (n+1)$ for $r>1$ we will obtain \cref{eq:StirlingPartPos} in the corresponding limits. By these replacements we can also continue the Stirling numbers to complex arguments. However, such a generalization is not necessary in the following approach.\bigskip


By those definitions we can expand the $\Gamma$-functions with respect to a small parameter $|\epsilon|\ll 1$ appearing in the (unimodular) series representation \cref{eq:SeriesRepresentationFeynmanIntegral} by
\begin{align}
	\Gamma(n+\epsilon) = \Gamma(\epsilon) \left( \StirlingFirst{n}{0} + \StirlingFirst{n}{1} \epsilon + \StirlingFirst{n}{2} \epsilon^2 + \StirlingFirst{n}{3} \epsilon^3 + \ldots \right) \qquad n\in\mathbb Z \label{eq:GammaSeriesExpansionStirling}
\end{align}
which works for positive as well as for negative integers $n$. Equation \cref{eq:GammaSeriesExpansionStirling} is a direct consequence of \cref{eq:StirlingPochhammer}. A similar expansion in terms of $S$-sums and $Z$-sums was suggested in \cite{MochNestedSumsExpansion2002}, where the cases for $n\in\mathbb Z_{>0}$ and $n\in \mathbb Z_{<0}$ were treated separately. The expansion \cref{eq:GammaSeriesExpansionStirling} is especially appropriate for a numerical evaluation, as we have not to distinguish between positive and negative cases in a symbolic preprocessing. We will consider both aspects below: a numerical evaluation of those series in \cref{sec:numerics} and a symbolical evaluation by means of $Z$-sums in \cref{sec:ManipulationSeries}. In \cref{sec:StirlingNumbers} we collected further useful facts about Stirling numbers.


\section{Manipulation of series} \label{sec:ManipulationSeries}


In this section we want to collect several aspects, how to treat multivariate series as appearing in \cref{thm:FeynSeries}. However, this overview is not intended to be exhaustive, and we will only give a small amount of useful relations and references for further information and algorithms. In particular, the methods presented below do not claim to be an efficient rewriting of the Horn type series into multiple polylogarithms and related functions in more complex examples.\bigskip

As seen in the previous \cref{sec:epsilonExpansion}, Stirling numbers $\StirlingFirstSmall{n}{k}$ behave slightly different for positive and negative $n\in\mathbb Z$. Therefore, it will be necessary for a symbolic processing of the Laurent coefficients of the Feynman integral to determine the signs in the Stirling numbers. In order to avoid fixed integer shifts in the argument of Stirling numbers it is recommendable to normalize the $\Gamma$-functions before introducing the Stirling numbers. Thus, we will shift $\Gamma$-functions by means of $\Gamma(n+c) = (n)_c \Gamma(n)$ before applying \cref{eq:GammaSeriesExpansionStirling}. Therefore, we will be left with Stirling numbers depending only on a linear combination of summation indices. In order to generate definite signs of these linear combinations, we can use series rearrangement techniques. To determine for example the signs in a Stirling number $\StirlingFirstSmall{i-j}{k}$ we can split a double series in all its cases $i<j$, $i>j$ and $i=j$, i.e.\
\begin{align}
	\sum_{j=a}^n \sum_{i=a}^n f(i,j) = \sum_{j=a}^n \sum_{i=a}^{j-1} f(i,j) + \sum_{i=a}^n \sum_{j=a}^{i-1} f(i,j) + \sum_{j=a}^n f(j,j) \label{eq:splittingNested}
\end{align}
which holds for $n\in\mathbb N \cup \{\infty\}$. The structure of relations like \cref{eq:splittingNested} is also known as \textit{quasi-shuffle product} \cite{WeinzierlFeynmanIntegrals2022}. Especially, for infinite sums we can also rearrange the summands instead of a splitting of the summation region, which results in the following identities \cite{SrivastavaTreatiseGeneratingFunctions1984}
\begin{align}
	\sum_{j=0}^\infty \sum_{i=0}^\infty f(i,j) &= \sum_{j=0}^\infty \sum_{i=0}^j f(i,j-i) \\
	\sum_{j=0}^\infty \sum_{i=0}^j f(i,j) &= \sum_{j=0}^\infty \sum_{i=0}^\infty f(i,j+i) = \sum_{i=0}^\infty \sum_{j=0}^i f(i-j,i) \point
\end{align}
More complicated situations of indetermined signs in Stirling numbers can be solved by applying \cite{SrivastavaTreatiseGeneratingFunctions1984}
\begin{align}
	\sum_{j=0}^\infty \sum_{i=0}^\infty f(i,j) &= \sum_{j=0}^\infty \sum_{i=0}^{\left\lfloor\frac{j}{m}\right\rfloor} f(i,j-mi)
\end{align}
where $\gls{floor}$ denotes the greatest integer less or equal $x$ with $m\in\mathbb N_{>0}$. All those identities can be used recursively to disentangle more involved combinations of indetermined signs in Stirling numbers. \bigskip

It is often also worthwhile to split the first summands of each series off, since they behave different from the remaining series due to the positive Stirling numbers \cref{eq:StirlingPartPos}. This can be done by an application of
\begin{align}
	\sum_{N\geq i_k \geq \ldots \geq i_1 \geq 0} \hspace{-1em} f(i_1,\ldots,i_k) \hspace{.5em} &= \hspace{-1em} \sum_{N\geq i_k \geq \ldots \geq i_1 \geq 1} \hspace{-1em} f(i_1,\ldots,i_k) \hspace{.5em} + \hspace{-1em} \sum_{N\geq i_k \geq \ldots \geq i_2 \geq 1} \hspace{-1em} f(0,i_2,\ldots,i_k) \nonumber \\
	& + \ldots + \sum_{N\geq i_k\geq 1} f(0,\ldots,0,i_k) + f(0,\ldots,0)	\point
\end{align}

After these steps we will arrive at nested sums containing three types of factors in their summands: a) powers of variables, b) Horn rational expressions in the summation indices, stemming from the factorials and the normalization of $\Gamma$-functions and c) harmonic numbers. We will call an expression $c(k)$ \textit{Horn rational} in $k$, when $\frac{c(k+e_i)}{c(k)}$ is a rational function in $k$ for all directions $i$ (see also \cref{sec:AHypSystems}). There are different techniques known to simplify those sums. For example Zeilberger's algorithm, also known as creative telescoping \cite{ZeilbergerMethodCreativeTelescoping1991, Petkovsek1996} can be used to find expressions for finite sums consisting in certain Horn rational summands. In principle, one can use this technique also for multivariate sums. However, it is not very efficient then \cite{ZeilbergerMethodCreativeTelescoping1991}. Hence, this algorithm is mostly restricted to a single summation where no harmonic numbers are involved. 

Certain cases can also be solved by introducing integral expressions of the harmonic numbers, see e.g.\ \cite{AdamchikStirlingNumbersEuler1997}. Further, we want to refer to \cite{BorweinEvaluationsKfoldEuler1996} where different techniques are used to find expressions for those sums.\bigskip

Another approach was suggested in \cite{MochNestedSumsExpansion2002}. They introduced so-called $Z$-sums as nested sums
\begin{align}
	\gls{Zsum} := \sum_{n \geq i_k > \ldots > i_1 > 0} \frac{y_1^{i_1} \cdots y_k^{i_k}}{i_1^{m_1} \cdots i_k^{m_k}}
\end{align}
where $n$ can also be infinite and $m_j\in\mathbb N_{\geq 0}$. Note that we adapted the notation from \cite{MochNestedSumsExpansion2002} slightly. On the one hand, those $Z$-sums can be understood as generalizations of multiple polylogarithms\footnote{Note, that the order of variables in multiple polylogarithms as well as in multiple zeta functions differs in various literature. We supposed the definitions in \cite{ZagierValuesZetaFunctions1994, GoncharovMultiplePolylogarithmsMixed2001}.}
\begin{align}
	\Zsum{\infty}{m_1, \ldots, m_k}(y_1, \ldots, y_k) = \Li_{m_1, \ldots, m_k} (y_1, \ldots, y_k) \point
\end{align}
On the other hand, they generalize also so-called Euler-Zagier sums, and we write $\gls{Zsum1} := \Zsum{n}{m_1, \ldots, m_k}(1, \ldots, 1)$ for short. Euler-Zagier sums \cite{ZagierValuesZetaFunctions1994} appear in the study of multiple $\zeta$-values due to $\Zsum{\infty}{m_1, \ldots, m_k} = \zeta(m_1,\ldots,m_k)$. 
 
Stirling numbers as defined in \cref{sec:epsilonExpansion} are special cases of Euler-Zagier sums, and we have
\begin{align}
	\StirlingFirst{n}{k+1} = (n-1)! \ \Zsum{n-1}{\underbrace{\scriptstyle 1,\ldots,1}_{k}} \quad\text{for } n,k\in\mathbb Z_{>0} \point
\end{align} 
Stirling numbers with negative first arguments $n$ will correspond to $S$-sums, which were also defined in \cite{MochNestedSumsExpansion2002}, see \cref{sec:StirlingNumbers} for details. Note that $Z$-sums are also closely related to harmonic numbers $H_n^{(k)} = \Zsum{n}{k}$ and their generalizations \cite{VermaserenHarmonicSumsMellin1999}.

As pointed out in \cite{MochNestedSumsExpansion2002} $Z$-sums obey a nice algebraic structure for fixed upper bound $n$. In particular, products of $Z$-sums with the same upper bound $n$ can be expressed in terms of single $Z$-sums by an iterative use of \cref{eq:splittingNested}. Moreover, certain sums over (Horn) rational functions, powers of variables and $Z$-sums can be converted into linear combinations of $Z$-sums as well \cite{MochNestedSumsExpansion2002}. Implementations of those algorithms are available in \cite{WeinzierlExpansionHalfintegerValues2004, MochXSummerTranscendentalFunctions2006}. It should be mentioned, that the known algorithms only cover a small part of the possible situations.

Hence, apart from very simple Feynman integrals it is currently not possible to convert Horn type hypergeometric series from \cref{eq:SeriesRepresentationFeynmanIntegral} efficiently into multiple polylogarithms and related functions by existing algorithms. Nevertheless, we will provide a small amount of tools for this problem in \cref{sec:StirlingNumbers}. However, there are already very efficient algorithms known \cite{PanzerFeynmanIntegralsHyperlogarithms2015, BrownMasslessHigherloopTwopoint2009, BognerSymbolicIntegrationMultiple2012, BognerFeynmanIntegralsIterated2015} for those Feynman integrals which can be evaluated by means of multiple polylogarithms. Therefore, for symbolical expressions of the Feynman integral, the hypergeometric approach should be understood as an alternative way, which could be used, when the current algorithms can not be applied. Furthermore, the hypergeometric approach is well-made for a fast numerical evaluation.

\begin{example}
	To illustrate the techniques from above in a very simple case, we want to continue \cref{ex:1loopbubbleB} and write the Laurent expansion of the appearing Gauss' hypergeometric function in terms of multiple polylogarithms. In \cref{eq:Example1loopBubble2F1} we had the following hypergeometric function
	\begin{align}
		\HypF{1,\epsilon}{2\epsilon}{t} &= \sum_{k\geq 0} (-t)^k (\epsilon)_k\, (1-2\epsilon)_{-k} = \sum_{k\geq 0} (-t)^k \frac{2\epsilon+k}{2\epsilon}  \left( \StirlingFirst{k}{0} + \StirlingFirst{k}{1} \epsilon \right. \nonumber \\
		& \qquad \left. + \StirlingFirst{k}{2} \epsilon^2 + \order{\epsilon^3} \!\right) \!\left( \StirlingFirst{-k}{0} - \StirlingFirst{-k}{1} 2 \epsilon + \StirlingFirst{-k}{2} 4 \epsilon^2 + \order{\epsilon^3} \! \right) \nonumber \\
		& = \frac{1}{\epsilon} c_{-1} + c_0 + c_1 \epsilon + c_2 \epsilon^2 + c_3 \epsilon^3 + \order{\epsilon^4} 
	\end{align}
	where we used \cref{eq:PochhammerIdentityA} and introduced Stirling numbers due to \cref{eq:GammaSeriesExpansionStirling} after a normalization of $\Gamma$-functions. We find then
	\begin{align}
		c_{-1} &= \frac{1}{2} \sum_{k\geq 0} (-t)^k k \StirlingFirst{k}{0} \StirlingFirst{-k}{0} = 0 \\
		c_0 &= \sum_{k\geq 0} (-t)^k \left\{ \StirlingFirst{k}{0} \StirlingFirst{-k}{0} - k \StirlingFirst{k}{0} \StirlingFirst{-k}{1} + \frac{1}{2} k \StirlingFirst{k}{1} \StirlingFirst{-k}{0} \right\} \nonumber \\
		&= 1 + \frac{1}{2} \sum_{k\geq 1} t^k = 1 + \frac{1}{2} \Li_0(t)
	\end{align}
	where we treated the first summand with $k=0$ separately and inserted the special representations of Stirling numbers from \cref{eq:StirlingPartPos} and \cref{eq:StirlingPartNeg}. By the same procedure we will obtain
	\begin{align}
		c_1 &= \sum_{k\geq 1} t^k \left( \frac{1}{k} - H_k + \frac{1}{2} H_{k-1} \right) = - \frac{1}{2} \sum_{k\geq 1} t^k H_{k-1} = -\frac{1}{2} \Li_{1,0}(1,t)
	\end{align}
	where we used $H_{k}^{(i)} = H_{k-1}^{(i)} + \frac{1}{k^i}$. For the next coefficient we have
	\begin{align}
		c_2 &= \sum_{k\geq 1} t^k \left( - \frac{2}{k} H_k + H_k^2 + H_k^{(2)} + \frac{1}{k} H_{k-1} - H_{k-1} H_k + \frac{1}{4} H_{k-1}^2 - \frac{1}{4} H_{k-1}^{(2)} \right) \nonumber \\
		&= \sum_{k\geq 1} t^k \left( \frac{3}{4} H_{k-1}^{(2)} + \frac{1}{4} H_{k-1}^2 \right) = \frac{3}{4} \Li_{2,0}(1,t) + \frac{1}{4} \sum_{k\geq 1} t^k \left( 2 \sum_{i=1}^{k-1} \sum_{j=1}^{i-1} \frac{1}{ij} + \sum_{i=1}^{k-1} \frac{1}{i^2} \right) \nonumber \\
		&= \Li_{2,0}(1,t) + \frac{1}{2} \Li_{1,1,0} (1,1,t)
	\end{align}
	where we additionally made use of \cref{eq:splittingNested}. By the same techniques we can continue 
	\begin{align}
		c_3 = -\frac{1}{2} \Li_{1,1,1,0}(1,1,1,t) - \Li_{1,2,0}(1,1,t) -\Li_{2,1,0} (1,1,t) - 2\Li_{3,0} (1,t) \point
	\end{align}
\end{example}

We want to conclude this section with a remark on another promising but immature method. As every coefficient in the Laurent expansion is again a Horn hypergeometric function one can apply the machinery of $\Aa$-hypergeometric systems a second time. We refer to \cite{SadykovHypergeometricFunctionsSeveral2002} for a comprehensive treatment of Horn hypergeometric functions within $D$-modules. In comparison to the original $\Aa$-hypergeometric system for the generalized Feynman integral $H_\Aa(\nuu)$, those new $\Aa$-hypergeometric systems will describe (a part of) a single coefficient in the Laurent expansion and will not contain additional variables any more. Hence, these second $\Aa$-hypergeometric systems will describe the problem in an efficient manner. Note further that a Horn hypergeometric series collapses to a Horn hypergeometric series with reduced depth when setting a variable equals zero. Hence, one can once again find boundary values by means of simpler cases. However, one source of potential difficulties in this approach lies in the singularity structure of the Horn hypergeometric functions \cite{PassareSingularitiesHypergeometricFunctions2005}, i.e.\ one has to ensure that all Horn hypergeometric functions are on their principle sheets.

    
\section{Notes on numerical evaluation} \label{sec:numerics}
%

The series representations from \cref{thm:FeynSeries} are very suitable for a numerical evaluation of Feynman integrals. Therefore, we will collect certain considerations on numerics in this section from the perspective of a practitioner. As described in \cref{lem:heightEffectiveVars} we can choose a regular triangulation which is convenient for the considered kinematical region. Furthermore, if the restriction to physical values will produce convergence issues of those series, we can solve those problems by the methods from \cref{sec:AnalyticContinuation}. Therefore, after the Laurent expansion described in \cref{sec:epsilonExpansion} we will obtain series of the form
\begin{align}
	\sum_{k\in\mathbb N^r_0} c(k) \frac{y^k}{k!} \label{eq:seriesNumerics}
\end{align}
where $|y_i|<1$. The function $c(k)$ is a product of Stirling numbers and additionally may contain certain factorials (which are Stirling numbers as well). Hence, we can calculate an approximation of those series simply by summing the first terms. Even this first naive attempt often creates surprisingly precise results after a few amount of terms. Classical techniques of series accelerations will enhance this approach and are worthwhile when $|y_i| \approx 1$. E.g.\ one can speed up the convergence rate by Euler's transformation. This is particularly appropriate since those series \cref{eq:seriesNumerics} are typically alternating. Hence, we will use the identity
\begin{align}
	&\sum_{k=0}^\infty (-1)^k f(k) = \sum_{k=0}^\infty \frac{(-1)^k}{2^{k+1}} \fd[k]{f(0)} \comma \qquad \text{where } \label{eq:EulerSummation} \\
	&\fd[k]{f(0)} = \sum_{i=0}^k (-1)^i \begin{pmatrix} k \\ i \end{pmatrix} f(k-i) = \sum_{i=0}^k (-1)^{k+i} \begin{pmatrix} k \\ i \end{pmatrix} f(i) \point \label{eq:kthFiniteDifference}
\end{align}
Thereby $\gls{fd}$ is known as the \textit{$k$-th forward difference}, i.e.\ we have $\fd{f(i)} = f(i+1)-f(i)$ and $ \fd[k]{f(i)} = \Delta^{k-1} \!\left(\fd{f(i)}\right)$. Hence, we can transform an alternating series into a series where subsequent summands are suppressed by a power of $\frac{1}{2}$. The series acceleration \cref{eq:EulerSummation} can be used iteratively to multivariate series as \cref{eq:seriesNumerics}. Moreover, there are many more known techniques for series acceleration. We refer to \cite{CohenConvergenceAccelerationAlternating2000} for an overview. \bigskip

When truncating the series \cref{eq:seriesNumerics}, the first question should be how many terms we take into account. We will denote by 
\begin{align}
	S(K) = S(K_1,\ldots,K_r) := \sum_{k_1=0}^{K_1} \cdots \sum_{k_r=0}^{K_r}  f(k)
\end{align}
the partial sum of \cref{eq:seriesNumerics} or an accelerated version of \cref{eq:seriesNumerics}, respectively. The relative increase for a summation step in direction $i$ will be denoted by
\begin{align}
	t_i(K) := \left| \frac{f(K+e_i)}{S(K)}\right| \point
\end{align}
A good criterion to truncate the series is to stop the summation in direction $i$ when $t_i(K)$ becomes less than a given tolerance threshold. To avoid possible effects from the first summands it is recommendable to fix also a minimal number of terms.

Another criterion for series truncation that is often used is to set a tolerance threshold for the difference of two summands $\fpd{i}{f(K)} = f(N+e_i) - f(N)$. Depending on the series, one can relate this difference threshold to an estimation of the error term \cite{CalabreseNoteAlternatingSeries1962, PinskyAveragingAlternatingSeries1978, BradenCalculatingSumsInfinite1992}.\bigskip

As aforementioned, the summands of \cref{eq:seriesNumerics} consist only in a product of Stirling numbers (see \cref{sec:epsilonExpansion}). As we need only a small amount of those Stirling numbers, we suggest to create a table of those numbers instead of a repetitive calculation. This will enhance the evaluation speed further. Depending on the used data types, it is possibly reasonable to replace $\StirlingFirstSmall{n}{k} = (n-1)!\ \sigma (n,k-1)$ for $n,k\in\mathbb Z_{>0}$ and $\StirlingFirstSmall{-n}{k} = \frac{(-1)^n}{(n)!} \sigma(-n,k)$ for $n,k\in\mathbb Z_{\geq 0}$ and store $\sigma(n,k)$ instead. It can be shown, that $\sigma(n,k)$ grows only moderately for $n\rightarrow\infty$ (see \cref{sec:StirlingNumbers}).\bigskip

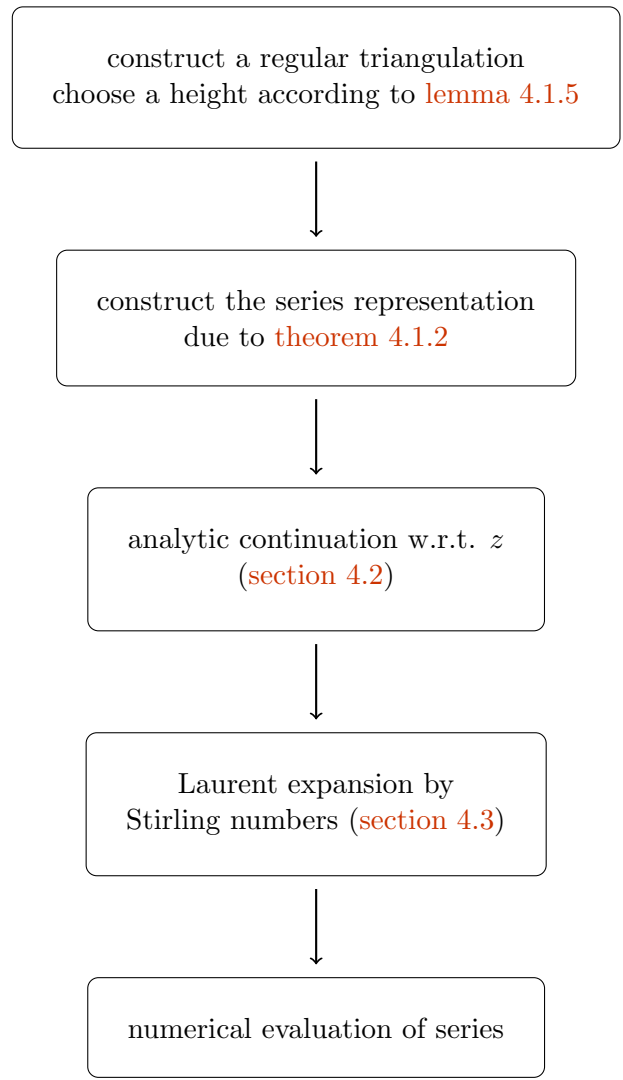
\begin{figure}[htb]
	\centering
	\begin{tikzpicture}
        \tikzstyle{every node}=[draw, rounded corners, inner sep = 15pt, outer sep = 5pt, align=center]
        \node (1) at (0,0) {construct a regular triangulation \\ choose a height according to \cref{lem:heightEffectiveVars}};
        \node (2) [below=of 1] {construct the series representation \\ due to \cref{thm:FeynSeries}};
        \node (3) [below=of 2] {analytic continuation w.r.t. $z$ \\ (\cref{sec:AnalyticContinuation})};
        \node (4) [below=of 3] {Laurent expansion by \\ Stirling numbers (\cref{sec:epsilonExpansion})};
        \node (5) [below=of 4] {numerical evaluation of series};
        \draw[->, thick] (1) -- (2); \draw[->, thick] (2) -- (3); \draw[->, thick] (3) -- (4); \draw[->, thick] (4) -- (5);
    \end{tikzpicture}
    \caption[Structure of a numerical evaluation of Feynman integrals by means of series representations]{Structure of a numerical evaluation of Feynman integrals by means of series representations.}
    \label{fig:structureEvaluation}
\end{figure}

By those methods one can construct a very efficient way for a numerical evaluation of Feynman integrals. \Cref{fig:structureEvaluation} shows the different steps in such an evaluation. Let us shortly comment each step from an algorithmic point of view. The construction of a single triangulation by means of a given height vector is algorithmically a relatively simple task and implementations can be found in various software programs (see \cref{sec:SoftwareTools} for details). We want to remark, that it is simple to construct a single triangulation. However, it is very intricate to determine all triangulations of a given configuration. For the next step, the construction of the series representation, one only has to calculate the Gale duals for each simplex, which is mainly an inversion of an $(n+1)\times (n+1)$ integer matrix. The complexity of the analytic continuation of those series, which is necessary for the most except of very simple Feynman integrals, depends highly on the problem. However, it seems that many Feynman integrals can be solved even by the ${_2}F_1$ transformation rules \cref{eq:2F1trafoA} and \cref{eq:2F1trafoB}. Hence, this step is in most cases a simple replacement of algebraic expressions. However, an efficient implementation of this step seems to us to be the comparatively greatest difficulty in an algorithmic realization. After the analytic continuation one would usually expand the Feynman integrals due to a small parameter $\epsilon$ for dimensional or analytical regularization. The algorithmic effort consists in a simple replacement of algebraic expressions and a reordering of terms in powers of $\epsilon$. Further, it is worthwhile to simplify all those terms where Stirling numbers $\StirlingFirstSmall{n}{k}$ appear with $n\geq 0$. Last but not least the computation of the series heavily depends on the kinematical region. Thus, for $|y_i|\ll 1$ the series converges very fast. However, by series accelerations one can also achieve good results when $|y_i|\approx 1$.


\section{Euler integrals and other representations} \label{sec:EulerIntegrals}


There are different ways to span the solution space $\Sol(H_\Aa(\nuu))$ of an $\Aa$-hypergeometric system. In principle, we can use all types of functions which belong to the class of $\Aa$-hypergeometric functions to generate this space. Hence, we will also have various possibilities to write the Feynman integral. All these representations will have slightly different characteristics. Therefore, the investigation of the various formulations of Feynman integrals is potentially helpful for certain aspects of those integrals. When spanning $\Sol(H_\Aa(\nuu))$ by $\Gamma$-series (see \cref{ssec:GammaSeries}) we will obtain the series representation of \cref{thm:FeynSeries}. Another type of functions which can be used to span the solution space are Laplace integrals $\int_\gamma \dif x\, x^{\nu-1} e^{-\Gg(x)}$. We refer to \cite{Matsubara-HeoLaplaceResidueEuler2018} for the construction of integration cycles $\gamma$. In this reference one can also find a direct transformation between $\Gamma$-series and Laplace integrals. Hence, we can also adopt the prefactors $C_{\sigma,k}(\nuu)$ which were derived in \cref{sec:seriesRepresentationsSEC}. 

A further option to generate $\Sol(H_\Aa(\nuu))$ are Mellin-Barnes integrals. For Feynman integrals this was derived by the results of \cref{thm:MellinBarnesRepresentation}. In general and more extensively those representations were investigated in \cite{BeukersMonodromyAhypergeometricFunctions2013, BerkeschEulerMellinIntegrals2013, Matsubara-HeoMellinBarnesIntegralRepresentations2018}. We want to remark, that Mellin-Barnes integrals are in particular useful when considering the monodromy of Feynman integrals \cite{BeukersMonodromyAhypergeometricFunctions2013}.\bigskip

We can also give a representation of Feynman integrals by so-called Euler integrals. Since the Feynman integral in parametric space \cref{eq:FeynmanParSpFeynman} is an Euler integral by itself, this type of integrals play a central role for many approaches, see exemplarily \cite{AbreuGeneralizedHypergeometricFunctions2019}. To find a relation between Horn type series and Euler integrals one can use Kummer's method \cite{AomotoTheoryHypergeometricFunctions2011}. Hence, we will use a convenient integral representation of the summand and change summation and integration. A classical result for iterated integrals going back to Kronecker and Lagrange is the following \cite[sec. 12.5]{WhittakerCourseModernAnalysis2021}
\begin{align}
	\int_{\Delta^n} \dif t \, t^{\alpha-1} f(t_1+\ldots+t_n)  = \frac{\Gamma(\alpha)}{\Gamma(|\alpha|)} \int_0^1 \dif \tau \, \tau^{|\alpha|-1} f(\tau)
\end{align}
where $t=(t_1,\ldots,t_n)\in\mathbb R^n$, $\alpha=(\alpha_1,\ldots,\alpha_n)\in\mathbb C^n$ with $\Re(\alpha_i)>0$, $|\alpha|:=\sum_{i=1}^n \alpha_i$ and $\Delta^n := \left\{ t\in\mathbb R^n \, |\, t_i > 0,\, \sum_{i=1}^n t_i < 1 \right\}$ denotes an $n$-simplex. Especially, for $f(\tau) = (1-\tau)^{\alpha_0-1}$ we will find a generalization of the beta function
\begin{align}
	\int_{\Delta^n} \dif t \, t^{\alpha-1} \left(1-\sum_{i=1}^n t_i \right)^{\alpha_0-1} = \frac{\Gamma(\underline\alpha)}{\Gamma(|\underline\alpha|)} \label{eq:DirichletIntegral}
\end{align}
where $\underline\alpha = (\alpha_0,\alpha)\in\mathbb C^{n+1}$ with $\Re(\underline\alpha_i)>0$. Once again we make use of a multi-index notation for compact expressions (see for example \cref{lem:SchwingersTrick}). We can use this integral representation of a product of $\Gamma$-functions to derive the following theorem.

\begin{theorem}[Euler integrals \cite{AomotoTheoryHypergeometricFunctions2011}] \label{thm:EulerHorn}
	We have the following relation between Horn hypergeometric series and Euler integrals for all $\gamma\notin \mathbb Z$
	\begin{align}
		&\sum_{k\in\mathbb N_0^r} \frac{y^k}{k!} \prod_{i=1}^{n} \Gamma\!\left(\alpha_i + \sum_{j=1}^r B_{ij} k_j\right) =  \nonumber\\
		&\qquad \Gamma(\gamma)\Gamma(1+|\alpha|-\gamma) \int_{\Delta^n} \dif t \, t^{\alpha-1} \left(1-\sum_{i=1}^n t_i \right)^{-\gamma} \left( 1+ \sum_{j=1}^r y_j \prod_{i=1}^n t_i^{B_{ij}} \right)^{-\gamma} \label{eq:EulerRepresentation}
	\end{align}
	whenever $\Re(\alpha_i)>0$ and $\sum_{i=1}^n B_{ij} = 1$ for $j=1,\ldots,r$. Furthermore, $\alpha$ is assumed to be generic such that the expressions converge absolutely.
\end{theorem}
\begin{proof}
	The proof follows \cite{AomotoTheoryHypergeometricFunctions2011}. Introducing $1=\Gamma(1-\gamma-|k|) \Gamma(1-\gamma-|k|)^{-1}$ we can reformulate the left hand side of \cref{eq:EulerRepresentation} by means of \cref{eq:DirichletIntegral} to
	\begin{align}
		\Gamma(1+|\alpha|-\gamma) \int_{\Delta^n} \dif t\, t^{\alpha-1} \left(1-\sum_{i=1}^n t_i \right)^{-\gamma} \sum_{k\in\mathbb N_0^r} \frac{y^k t^{Bk}}{k! \, \Gamma(1-\gamma-|k|)} \point \label{eq:EulerProof1}
	\end{align}
	For the remaining series in \cref{eq:EulerProof1} note that
	\begin{align}
		 \sum_{k\in\mathbb N_0^r} \frac{x^k}{k! \, \Gamma(1-\gamma-|k|)} =\frac{(1+|x|)^{-\gamma}}{\Gamma(1-\gamma)} \label{eq:EulerProof2}
	\end{align}
	which can be shown by an induction over $r$. For $r=1$ this follows from the binomial theorem. The induction step can be shown again by the binomial theorem. Inserting \cref{eq:EulerProof2} into \cref{eq:EulerProof1} concludes the proof.
\end{proof}

The representation \cref{eq:EulerRepresentation} will simplify whenever there is a row $i$ with $B_{ij}=1$ for all $j=1,\ldots,r$, i.e.\ there is already a factor of the form $\Gamma(\gamma+|k|)$. However, when $B$ is (a part of) a Gale dual of an acyclic vector configuration this will only happen in the case $r=1$ (see \cref{ssec:GaleDuality}). Moreover, one can extend \cref{thm:EulerHorn} also to general values for $\alpha\in\mathbb C^n$. In this case the integration region $\Delta^n$ has to be regularized by twisted cycles as described in \cite[sec. 3.2.5]{AomotoTheoryHypergeometricFunctions2011}. A similar result for general $\alpha\in\mathbb C^n$ was derived in \cite[thm. 5.3]{Matsubara-HeoLaplaceResidueEuler2018} with an alternative description of integration cycles.

Note, that the Horn hypergeometric series from \cref{thm:FeynSeries} will always satisfy the condition $\sum_i B_{ij}=1$ due to \cref{lem:HomogenityRelationsForAa}. Hence, by a rewriting of the series from \cref{thm:FeynSeries} by means of \cref{thm:EulerHorn} we can give also an alternative representation of Feynman integrals in terms of Euler integrals. Since the components of $\Aas^{-1}\nuu$ may also appear with negative values, the integration region of those Euler integrals can be very intricate.\bigskip

Besides of the various integral representations, there are also different series solutions for the $\Aa$-hypergeometric systems known. In this context we want to remark the canonical series \cite{SaitoGrobnerDeformationsHypergeometric2000}, which are power series with additional logarithmic terms. By those canonical series one can even span solution spaces for more general $D$-modules than $\Aa$-hypergeometric systems. This method is a direct generalization of Frobenius' method to the multivariate case. Hence, they are series solutions around the singular locus of the considered $D$-module. For the construction of those canonical series one has to determine the roots of the indicial ideal $\operatorname{ind}_w\!\left(H_\Aa(\nuu)\right) = R \cdot \initial_{(-w,w)} \!\left(H_\Aa(\nuu)\right) \cap \mathbb C[\theta]$, where $\theta=(z_1\partial_1,\ldots,z_N\partial_N)$ is the Euler operator and the other objects were defined in \cref{sec:holonomicDmodules}. In case of generic parameters $\nuu$ one can simplify $\operatorname{ind}_w\!\left(H_\Aa(\nuu)\right)$ to the fake indicial ideal. For details we refer to \cite{SaitoGrobnerDeformationsHypergeometric2000}. The $\Gamma$-series introduced in \cref{ssec:GammaSeries} can be seen as a special case of canonical series, where no logarithms appear.

    
\section{Periods and marginal Feynman integrals} \label{sec:periodMarginal}


The mechanism developed above to create series representations is not restricted to the Feynman integral in Lee-Pomeransky representation. It is a method which can be used for all integrals of Euler-Mellin type, i.e.\ for Mellin transforms of polynomials up to certain powers \cite{BerkeschEulerMellinIntegrals2013}. Hence, there are also more applications in the Feynman integral calculus. Since the $\Aa$-hypergeometric system for the representation \cref{eq:FeynmanParSpFeynman} is equivalent to the one of \cref{eq:LeePomeranskyRepresentation} as pointed out in \cref{sec:FeynmanIntegralsAsAHyp} the series representations will be the same. However, in case where one of the Symanzik polynomials drops out from \cref{eq:FeynmanParSpFeynman} due to certain constraints on $d$ and $\nu$, we will obtain a simpler series representation. The case, where the superficial degree of divergence vanishes $\omega=0$ and only the first Symanzik polynomial $\Uu$ remains is often called the \textit{period}\footnote{Here ``period'' is meant in the sense of Kontsevich and Zagier \cite{KontsevichPeriods2001}, i.e.\ a complex number where its real and imaginary part is expressible as an absolutely convergent integral over a rational function with rational coefficients over real domains defined by polynomial inequalities with rational coefficients. Indeed, every term in the $\epsilon$-expansion of a Feynman integral is a period when restricting the kinematic invariants and masses to rational numbers \cite{BognerMathematicalAspectsFeynman2009, BognerPeriodsFeynmanIntegrals2009, BrownPeriodsFeynmanIntegrals2010}. However, one often refers with ``period'' to the special case of Feynman integrals, where only the first Symanzik polynomial is included in representation \cref{eq:FeynmanParSpFeynman}.}, whereas the case $\omega=\dhalf$ where only the second Symanzik polynomial $\Ff$ remains is known as \textit{marginal Feynman integral} \cite{BourjailyBoundedBestiaryFeynman2019}. For instance all ``banana''-graphs are marginal for $\nu_i=1$ and $d=2$. This special circumstance was used in \cite{KlemmLoopBananaAmplitude2020}.

It is important to remark, that periods and marginal Feynman integrals not only appear in special configurations, which forces \cref{eq:FeynmanParSpFeynman} to drop a polynomial. Those integrals will also contribute to terms of the $\epsilon$-expansion of the Feynman integral in dimensional regularization (see \cref{sec:DimAnaReg}). Hence, it is meaningful to consider periods and marginal Feynman integrals also for general $d$ and $\nu$.\bigskip

\begin{example}[Periods of $1$-loop graphs] 
	For example, we can calculate the periods of all $1$-loop graphs (see \cref{fig:1loopFamily}) with the previous proposed approach, i.e.\ we consider
	\begin{align}
		\gls{period} := \int_{\mathbb R^n_+} \dif x\, x^{\nu-1} \delta (1-x_n) \, \Uu^{-\nu_0} \quad\text{with}\quad \Uu = \sum_{i=1}^n x_i \point
	\end{align}
	After evaluating the $\delta$-distribution we will find
	\begin{align}
		\Aa = \begin{pmatrix}
                1 & \mathbf 1_{n-1}^\top \\
		      	\mathbf 0_{n-1} & \mathbbm 1_{n-1}
		      \end{pmatrix} \text{ , } \quad \Aa^{-1} = \begin{pmatrix}
		      	1 & - \mathbf 1_{n-1}^\top \\
		      	\mathbf 0_{n-1} & \mathbbm 1_{n-1}
		      \end{pmatrix}
	\end{align}
	and $\det (\Aa) = 1$. As before, we denote by $\mathbf 0_{n-1}$, $\mathbf 1_{n-1}$ the constant zero column vectors and the constant one column vectors of size $n-1$, respectively. Since $\Aa$ is quadratic, those period integrals for $1$-loop graphs will satisfy the conditions of \cref{cor:MellinBarnesRepresentation}. Therefore, we will have
	\begin{align}
		\mathcal P_\Gamma (\nuu) = \frac{\Gamma\!\left(\Aa^{-1}\nuu\right)}{\Gamma(\nu_0)\,  |\det (\Aa)|} =  \frac{\Gamma\!\left(\nu_0 - \sum_{i=1}^n \nu_i\right) \Gamma(\nu)}{\Gamma(\nu_0)}
	\end{align}
	for all $1$-loop graphs $\Gamma$.
\end{example}

\begin{figure}[tb]
	\centering
	\begin{subfigure}{.45\textwidth}
		\centering
        \begin{tikzpicture}[thick, dot/.style = {draw, shape=circle, fill=black, scale=.3}, scale=1.5]
            \draw (0,0) circle (1);
            
            \coordinate[dot] (1) at (40:1);
            \coordinate[dot] (2) at (65:1);
            \coordinate[dot] (3) at (90:1);
            \coordinate[dot] (4) at (115:1);
            \coordinate[dot] (5) at (140:1);
            
            \draw (1) -- node[pos=0.7, above, xshift=-4pt] {\footnotesize $p_{n-1}$} (40:2);
            \draw (2) -- node[pos=0.7, above, xshift=-3pt] {\footnotesize $p_{n}$} (65:2);
            \draw (3) -- node[pos=0.7, left, xshift=3pt, yshift=4pt] {\footnotesize $p_{1}$} (90:2);
            \draw (4) -- node[pos=0.7, below, xshift=-5pt, yshift=5pt] {\footnotesize $p_{2}$} (115:2);
            \draw (5) -- node[pos=0.7, below, xshift=-4pt, yshift=2pt] {\footnotesize $p_{3}$} (140:2);
            
            \node at (52.5:1.15) {$n$};
            \node at (77.5:1.15) {$1$};
            \node at (102.5:1.15) {$2$};
            \node at (127.5:1.15) {$3$};
            
            \node at (160:1.5) {$\cdot$};
            \node at (168:1.5) {$\cdot$};
            \node at (176:1.5) {$\cdot$};
            \node at (184:1.5) {$\cdot$};
        \end{tikzpicture}
        \caption{Family of $1$-loop graphs with $n$ edges.} \label{fig:1loopFamily}
	\end{subfigure} %
	\begin{subfigure}{.45\textwidth}
		\centering
        \begin{tikzpicture}[thick, dot/.style = {draw, shape=circle, fill=black, scale=.3}, scale=1.5]
            \draw (0,0) circle (1);

            \coordinate[dot] (A) at (-1,0);
            \coordinate[dot] (B) at (1,0);
            
            \draw (A) -- node[above,midway] {$p$} ++(-0.7,0);
            \draw (B) -- node[above,midway] {$p$} ++(0.7,0);
            
            \draw (0,-0.4) ++(158.2:1.077) arc (158.2:21.8:1.077);
            \draw (0,+0.4) ++(338.2:1.077) arc (338.2:201.8:1.077);
            
            \node at (0,1.15) {$1$};
            \node at (0,0.83) {$2$};
            \node at (0,0.2) {$\vdots$};
            \node at (0,-0.49) {$n-1$};
            \node at (0,-0.82) {$n$};
        \end{tikzpicture}
        \caption{Family of banana graphs with $n=L+1$ edges.} \label{fig:bananaFamily}
	\end{subfigure}
	\caption[Feynman graphs for $1$-loop graphs and banana graphs]{Feynman graphs for $1$-loop graphs and banana graphs.}
\end{figure}
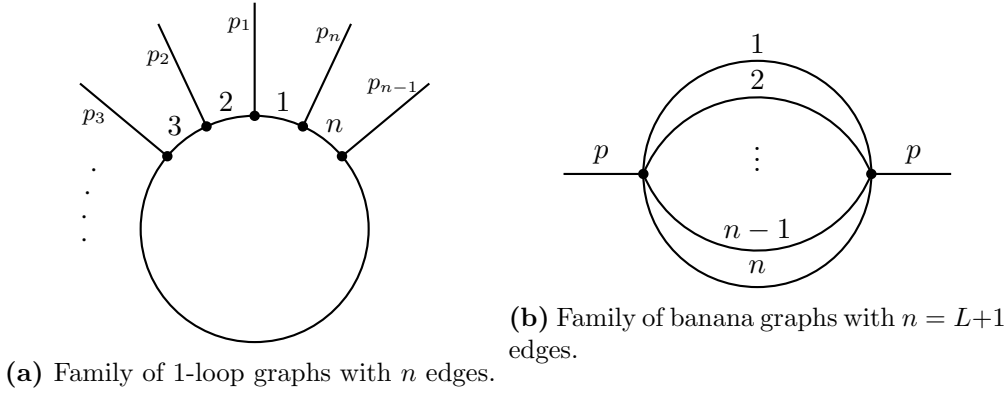

\begin{example}[Periods of $L$-loop banana graphs] \label{ex:periodsLloopBanana}
	In the same way one can give an analytic expression for the period of every $L$-loop banana graph (see \cref{fig:bananaFamily}), i.e.\ graphs which consist of $n=L+1$ edges, two legs and have the first Symanzik polynomial 
	\begin{align}
		\Uu = x_1 \cdots x_n \left(\frac{1}{x_1} + \ldots + \frac{1}{x_n} \right) \point \label{eq:BananaUSymanzik}
	\end{align}
	After evaluation of the $\delta$-distribution which sets $x_n=1$, we will obtain
	\begin{align}
		\Aa = \begin{pmatrix}
                1 & \mathbf 1_L^\top \\
		      	\mathbf 1_L & \mathbf 1_{L\times L} - \mathbbm 1_L
		      \end{pmatrix} \text{ , } \quad \Aa^{-1} = \begin{pmatrix}
		      	-L+1 & \mathbf 1_L^\top \\
		      	\mathbf 1_L & - \mathbbm 1_L
		      \end{pmatrix}
	\end{align}
	which can be verified by multiplication and $\mathbf 1_{L\times L}$ denotes the constant $L\times L$ matrix consisting only in ones. By Laplace expansion along the second column we obtain $\det\!\left(\Aa_L\right) = - \det\!\left(\Aa_{L-1}\right)$, where $\Aa_L$ and $\Aa_{L-1}$ stand for the vector configurations of a $L$-loop and $(L-1)$-loop banana graph, respectively. Therefore, we will find by induction $|\det(\Aa)|=1$. Hence, we can combine
	\begin{align}
		\mathcal P_\Gamma(\nuu) = \frac{\Gamma\!\left(\! (1-L) \nu_0 + \displaystyle\sum_{i=1}^{n-1} \nu_i \right) \displaystyle\prod_{i=1}^{n-1} \Gamma(\nu_0-\nu_i)}{\Gamma(\nu_0)} \point
	\end{align}
\end{example}

We want to remark that it seems in general not very appropriate to determine periods by hypergeometric functions, since we have to evaluate the hypergeometric functions at unity argument which is typically a very non-generic point. Moreover, there are very efficient alternatives to calculate periods, e.g.\ \cite{BrownPeriodsFeynmanIntegrals2010, PanzerFeynmanIntegralsHyperlogarithms2015, SchnetzQuantumPeriodsCensus2009, SchnetzGraphicalFunctionsSinglevalued2014}. The above examples have been only included for the purpose of showing possible alternative applications.

\begin{example}[marginal $L$-loop banana graphs with one mass]
    The counterpart of period Feynman integrals are the so-called marginal Feynman graphs, i.e.\ Feynman integrals where the first Symanzik polynomial drops from \cref{eq:FeynmanParSpFeynman}. Another example which satisfies the condition of \cref{cor:MellinBarnesRepresentation} and allows therefore an analytical expression is the class of marginal $L$-loop banana graphs with one mass, i.e.\ we consider the family of integrals
	\begin{align}
		\mathcal K_\Gamma (\nuu,z) := \int_{\mathbb R^n_+} \dif x \, x^{\nu-1} \delta(1-x_n)\, \Ff^{-\nu_0} \quad \text{with}\quad \Ff = p^2 x_1 \cdots x_n + m_n^2 x_n \Uu
	\end{align}
	where $\Uu$ was given in \cref{eq:BananaUSymanzik}. Note that $\Ff|_{x_n=1}$ has the same monomials as $\Uu|_{x_n=1}$. Therefore, we can adopt the results from \cref{ex:periodsLloopBanana} by adding the variable $z=(p^2+m_n^2,m_n^2,\ldots,m_n^2)$. Hence, we obtain
	\begin{align}
		\mathcal K_\Gamma (\nuu,z) = \frac{\Gamma\!\left(\! (1-L) \nu_0 + \displaystyle\sum_{i=1}^{n-1} \nu_i \right) \displaystyle\prod_{i=1}^{n-1} \Gamma(\nu_0-\nu_i)}{\Gamma(\nu_0)} \left(\frac{p^2+m_n^2}{m_n^2}\right)^{L\nu_0 - \sum_{i=1}^{n-1}\nu_i} (p^2+m_n^2)^{-\nu_0} \ \text{.}
	\end{align}
\end{example}
\begin{example}[marginal massive $1$-loop bubble]
	As another example of a marginal Feynman integral we will present, for the purpose of illustration once again, the $1$-loop bubble, i.e.\ the integral
	\begin{align}
		\gls{marginal} =  \int_{\mathbb R^n_+} \dif x \, x^{\nu-1} \delta(1-x_2)\, \Ff^{-\nu_0}
	\end{align}
	where $\Ff = (p^2 + m_1^2+m_2^2) x_1 x_2 + m_1^2 x_1^2 + m_2^2 x_2^2$ and which generates
	\begin{align}
		\Aa = \begin{pmatrix}
		      	1 & 1 & 1 \\
		      	0 & 1 & 2
		      \end{pmatrix}\quad\text{,} \qquad z = (m_2^2,p^2+m_1^2+m_2^2,m_1^2) \point
	\end{align}
	When choosing the triangulation $\hatT = \{\{1,2\},\{2,3\}\}$ this will result in
	\begin{align}
		\mathcal K_\Gamma (\nuu,z) &= \frac{1}{\Gamma(\nu_0)} \left[ z_1^{-\nu_0+\nu_1} z_2^{-\nu_1} \sum_{\lambda\in\mathbb N_0} \frac{\Gamma(\nu_0-\nu_1-\lambda)\Gamma(\nu_1+2\lambda)}{\lambda!} (-y)^\lambda \right. \nonumber \\
		&\qquad \left. + z_2^{-2\nu_0+\nu_1} z_3^{\nu_0-\nu_1} \sum_{\lambda\in\mathbb N_0} \frac{\Gamma(2\nu_0-\nu_1+2\lambda)\Gamma(-\nu_0+\nu_1-\lambda)}{\lambda!} (-y)^\lambda \right] \nonumber \\
		&= \frac{1}{\Gamma(\nu_0)} \left[ z_1^{-\nu_0+\nu_1} z_2^{-\nu_1} \Gamma(\nu_0-\nu_1)\Gamma(\nu_1)\ \HypF{\frac{\nu_1}{2},\frac{1+\nu_1}{2}}{1-\nu_0+\nu_1}{4y} \right. \nonumber \\
		&\qquad \left. + z_2^{-2\nu_0+\nu_1} z_3^{\nu_0-\nu_1} \Gamma(2\nu_0-\nu_1)\Gamma(-\nu_0+\nu_1)\ \HypF{\nu_0-\frac{\nu_1}{2},\frac{1}{2}+\nu_0-\frac{\nu_1}{2}}{1+\nu_0-\nu_1}{4y} \right]
    \end{align}
	where $y=\frac{z_1z_3}{z_2^2} = \frac{m_1^2m_2^2}{\left(p^2+m_1^2+m_2^2\right)^2}$.
\end{example}

Last but not least, we want to mention that the so-called \textit{stringy integrals} \cite{Arkani-HamedStringyCanonicalForms2021} also belong to the class of Euler-Mellin integrals. These stringy integrals are generalizations of (open) string amplitudes and can also be treated by the series approach presented in this chapter.



\section{Series representation for the fully massive sunset graph}

\label{sec:ExampleSunset}

We want to conclude this chapter about series representations by an extensive example to illustrate the methods stated above as well as to show the scope of this approach. For this reason we will consider the sunset Feynman integral with three different masses according to \cref{fig:sunset}. The corresponding Feynman graph consists in $n=3$ edges and the Lee-Pomeransky polynomial includes $N=10$ monomials
\begin{align}
    \Gg &= x_1 x_2 + x_1 x_3 + x_2 x_3 + \left(m_1^2+m_2^2+m_3^2+p^2\right) x_1 x_2 x_3 \nonumber \\
    &\qquad + m_1^2 x_1^2 (x_2+x_3) + m_2^2 x_2^2 (x_1+x_3) + m_3^2 x_3^2 (x_1+x_2) \point
\end{align}

\begin{figure}[ht]
    \centering
    \begin{tikzpicture}[thick, dot/.style = {draw, shape=circle, fill=black, scale=.5}, scale=1.5]
        \draw (0,0) circle (1);

        \coordinate[dot] (A) at (-1,0);
        \coordinate[dot] (B) at (1,0);
        
        \draw (A) -- node[above] {$p$} ++(-0.7,0);
        \draw (B) -- node[above] {$p$} ++(+0.7,0);
        
        \draw (A) -- (B);
        
        \node at (0,1.15) {$m_1$};
        \node at (0,0.15) {$m_2$};
        \node at (0,-0.85) {$m_3$};
    \end{tikzpicture}
    \caption[$2$-loop self-energy Feynman graph (``sunset'')]{The $2$-loop $2$-point function (sunset graph) with three different masses.}
    \label{fig:sunset}
\end{figure}
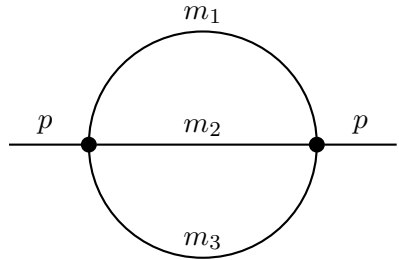

In the representation of equation \cref{eq:Gsupport} we encode this polynomial by
\begin{align}
    \Aa &= \begin{pmatrix}
        1 & 1 & 1 & 1 & 1 & 1 & 1 & 1 & 1 & 1 \\
        0 & 1 & 1 & 0 & 1 & 0 & 1 & 2 & 1 & 2 \\
        1 & 0 & 1 & 1 & 0 & 2 & 1 & 0 & 2 & 1 \\
        1 & 1 & 0 & 2 & 2 & 1 & 1 & 1 & 0 & 0
    \end{pmatrix} \\
    z &= (1,1,1,m_3^2,m_3^2,m_2^2,m_1^2+m_2^2+m_3^2+p_1^2,m_1^2,m_2^2,m_1^2) \point
\end{align}
The rank of the kernel of $\Aa$ is equal to $r=N-n-1=6$, and therefore we will expect $6$-dimensional $\Gamma$-series. Moreover, the Newton polytope $\Newt(\Gg)=\Conv (A)$ has the volume $\vol(\Newt(\Gg))=10$, which can e.g.\ be calculated with \softwareName{polymake} \cite{GawrilowPolymakeFrameworkAnalyzing2000}, see \cref{sec:SoftwareTools}. This leads to $10$ basis solutions, and there are $826$ different ways for a regular triangulation of the Newton polytope $\Newt(\Gg)$, where $466$ of those triangulations are unimodular. We choose the unimodular triangulation (calculated with \softwareName{Topcom} \cite{RambauTOPCOMTriangulationsPoint2002})
\begin{align}
    \hatT_{152} &= \big\{\{3,6,7,9\},\{3,7,9,10\},\{3,7,8,10\},\{2,5,7,8\},\{2,3,7,8\},\nonumber \\
    &\qquad \{2,4,5,7\},\{1,4,6,7\},\{1,2,4,7\},\{1,3,6,7\},\{1,2,3,7\}\big\} \label{eq:sunsetTriang}
\end{align}
in order to get series, which converge fast for highly relativistic kinematics $m_i^2 \ll m_1^2+m_2^2+m_3^2+p^2$. Further, we set $\nu_i=1$ and $d=4-2\epsilon$. 

In the limit $z\rightarrow(1,1,1,m_3^2,m_3^2,m_2^2,m_1^2+m_2^2+m_3^2+p_1^2,m_1^2,m_2^2,m_1^2)$ the series $\phi_1$, $\phi_3$, $\phi_5$, $\phi_6$, $\phi_8$ and $\phi_9$ are divergent for small values of $\epsilon>0$. We will write $\phi_i := \phi_{\sigma_i} (\nuu,z)$ for short, where the numeration of simplices $\sigma_i$ is oriented towards \cref{eq:sunsetTriang}. By the method described in \cref{sec:AnalyticContinuation} one can split all these series by the transformation formula for the ${_2}F_1$ Gauss' hypergeometric function in a convergent and a divergent part. The divergent parts of these series cancel each other. In doing so the resulting $\Gamma$-series have linear dependences and the dimension of the solution space will reduce from $10$ to $7$. 

By applying all these steps one arrives at the following series representation of the fully massive sunset integral
\begin{align}
    \mathcal I_\Aa (\nuu,z) &= \frac{s^{1-2\epsilon}}{\Gamma(3-3\epsilon)} \big[ y_2^{1-\epsilon} \phi_1^\prime + (y_1 y_2)^{1-\epsilon} \phi^\prime_2 + y_1^{1-\epsilon} \phi^\prime_3 + (y_1 y_3)^{1-\epsilon} \phi^\prime_4 + y_1^{1-\epsilon} \phi^\prime_5 \nonumber \\
    &\qquad  + y_3^{1-\epsilon} \phi^\prime_6 + (y_2 y_3)^{1-\epsilon} \phi^\prime_7 + y_3^{1-\epsilon} \phi^\prime_8 + y_2^{1-\epsilon} \phi^\prime_9 + \phi^\prime_{10} \big] \comma
\end{align}
where we adapted the notation of $\Gamma$-series slightly for convenience. Those $\Gamma$-series are given by
\begin{align}
    &\phi^\prime_1 = \sum_{k_2,k_3,k_4,k_5,k_6=0}^\infty\frac{ (-y_2)^{k_2} (-y_3)^{k_3} (-y_2 y_3)^{k_4} (-y_1 y_2)^{k_5} (-y_1)^{k_6} }{k_2!\, k_3!\, k_4!\, k_5!\, k_6!} \nonumber\\
    &\quad \Gamma (k_2-3 \epsilon +3) \Gamma (k_2+k_3+k_4-k_6-2 \epsilon +3) \Gamma (k_3-k_5-k_6-\epsilon +1) \nonumber \\
    &\quad \frac{\Gamma (k_2+k_3+2 k_4+2 k_5+k_6+\epsilon ) \Gamma (k_4+k_5+2 \epsilon -1) \Gamma (-k_2-k_3-k_4+k_6+2 \epsilon -2)}{\Gamma (k_2+k_4+k_5-\epsilon +2) \Gamma (k_3+k_4-k_6+\epsilon )}  \allowdisplaybreaks \nonumber\\[\baselineskip]
    &\phi^\prime_2 = \sum_{k_1,k_2,k_3,k_4,k_5,k_6=0}^\infty\frac{ (-y_1)^{k_1+k_5} (-y_2)^{k_2+k_6}  (-y_1 y_3)^{k_3} (-y_2 y_3)^{k_4}}{k_1!\, k_2!\, k_3!\, k_4!\, k_5!\, k_6!} \nonumber\\
    &\quad \Gamma (k_1+k_2+2 k_3+2 k_4+k_5+k_6+1) \Gamma (k_1+k_2-3 \epsilon +3) \nonumber \\
    &\quad \Gamma (-k_2-k_4+k_5-k_6+\epsilon -1) \Gamma (-k_1-k_3-k_5+k_6+\epsilon -1) \allowdisplaybreaks \nonumber\\[\baselineskip]
    &\phi^\prime_3 = \sum_{k_1,k_3,k_4,k_5,k_6=0}^\infty\frac{(-y_1)^{k_1}  (-y_1 y_3)^{k_3} (-y_3)^{k_4} (-y_1 y_2)^{k_5} (-y_2)^{k_6} }{k_1!\, k_3!\, k_4!\, k_5!\, k_6!} \nonumber\\
    &\quad \Gamma (k_1-3 \epsilon +3) \Gamma (k_1+k_3+k_4-k_6-2 \epsilon +3) \Gamma (k_4-k_5-k_6-\epsilon +1) \nonumber \\
    &\quad \frac{\Gamma (k_1+2 k_3+k_4+2 k_5+k_6+\epsilon ) \Gamma (k_3+k_5+2 \epsilon -1) \Gamma (-k_1-k_3-k_4+k_6+2 \epsilon -2)}{\Gamma (k_1+k_3+k_5-\epsilon +2) \Gamma (k_3+k_4-k_6+\epsilon )} \allowdisplaybreaks \nonumber\\[\baselineskip]
    &\phi^\prime_4 = \sum_{k_1,k_2,k_3,k_4,k_5,k_6=0}^\infty\frac{ (-y_1)^{k_1+k_3}  (-y_3)^{k_2+k_6} (-y_1 y_2)^{k_4}  (-y_2 y_3)^{k_5}}{k_1!\, k_2!\, k_3!\, k_4!\, k_5!\, k_6!} \nonumber\\
    &\quad \Gamma (k_1+k_2+k_3+2 k_4+2 k_5+k_6+1) \Gamma (k_1+k_2-3 \epsilon +3) \nonumber \\
    &\quad \Gamma (-k_2+k_3-k_5-k_6+\epsilon -1) \Gamma (-k_1-k_3-k_4+k_6+\epsilon -1) \allowdisplaybreaks \nonumber\\[\baselineskip]
    &\phi^\prime_5 = \sum_{k_1,k_2,k_3,k_4,k_5=0}^\infty\frac{ (-y_1)^{k_1} (-y_1 y_3)^{k_2} (-y_3)^{k_3}  (-y_1 y_2)^{k_4} (-y_2)^{k_5} }{k_1!\, k_2!\, k_3!\, k_4!\, k_5!} \nonumber\\
    &\quad \Gamma (k_1+k_2+k_3-k_5-2 \epsilon +2) \Gamma (-k_2-k_3+k_5-\epsilon +1) \Gamma (-k_1-k_2-k_4+\epsilon -1)  \nonumber \\
    &\quad \frac{\Gamma (k_1+2 k_2+k_3+2 k_4+k_5+\epsilon ) \Gamma (k_2+k_4+2 \epsilon -1) \Gamma (-k_1-k_2-k_3+k_5+2 \epsilon -1)}{\Gamma (-k_3+k_4+k_5+\epsilon ) \Gamma (-k_1+3 \epsilon -2)} \allowdisplaybreaks \nonumber\\[\baselineskip]
    &\phi^\prime_6 = \sum_{k_2,k_3,k_4,k_5,k_6=0}^\infty\frac{ (-y_3)^{k_2} (-y_2)^{k_3}  (-y_1)^{k_4}   (-y_2 y_3)^{k_5} (-y_1 y_3)^{k_6} }{k_2!\, k_3!\, k_4!\, k_5!\, k_6!} \nonumber\\
    &\quad \Gamma (k_2-3 \epsilon +3) \Gamma (k_2+k_3-k_4+k_5-2 \epsilon +3) \Gamma (k_3-k_4-k_6-\epsilon +1) \nonumber \\
    &\quad \frac{\Gamma (k_2+k_3+k_4+2 k_5+2 k_6+\epsilon ) \Gamma (-k_2-k_3+k_4-k_5+2 \epsilon -2) \Gamma (k_5+k_6+2 \epsilon -1)}{\Gamma (k_2+k_5+k_6-\epsilon +2) \Gamma (k_3-k_4+k_5+\epsilon )} \allowdisplaybreaks \nonumber\\[\baselineskip]
    &\phi^\prime_7 = \sum_{k_1,k_2,k_3,k_4,k_5,k_6=0}^\infty\frac{ (-y_2)^{k_1+k_3}(-y_3)^{k_2+k_5} (-y_1 y_2)^{k_4}  (-y_1 y_3)^{k_6}}{k_1!\, k_2!\, k_3!\, k_4!\, k_5!\, k_6!} \nonumber\\
    &\quad \Gamma (k_1+k_2+k_3+2 k_4+k_5+2 k_6+1) \Gamma (k_1+k_2-3 \epsilon +3) \nonumber \\
    &\quad  \Gamma (-k_1-k_3-k_4+k_5+\epsilon -1) \Gamma (-k_2+k_3-k_5-k_6+\epsilon -1) \allowdisplaybreaks \nonumber\\[\baselineskip]
    &\phi^\prime_8 = \sum_{k_1,k_3,k_4,k_5,k_6=0}^\infty\frac{(-y_3)^{k_1}  (-y_2)^{k_3} (-y_1)^{k_4} (-y_2 y_3)^{k_5}  (-y_1 y_3)^{k_6} }{k_1!\, k_3!\, k_4!\, k_5!\, k_6!} \nonumber\\
    &\quad \Gamma (k_1+k_3-k_4+k_5-2 \epsilon +2) \Gamma (-k_3+k_4-k_5-\epsilon +1) \Gamma (-k_1-k_5-k_6+\epsilon -1) \nonumber \\
    &\quad \frac{\Gamma (k_1+k_3+k_4+2 k_5+2 k_6+\epsilon ) \Gamma (-k_1-k_3+k_4-k_5+2 \epsilon -1) \Gamma (k_5+k_6+2 \epsilon -1)}{\Gamma (-k_3+k_4+k_6+\epsilon ) \Gamma (-k_1+3 \epsilon -2)} \allowdisplaybreaks \nonumber\\[\baselineskip]
    &\phi^\prime_9 = \sum_{k_1,k_2,k_3,k_4,k_6=0}^\infty\frac{  (-y_2)^{k_1} (-y_3)^{k_2} (-y_2 y_3)^{k_3} (-y_1 y_2)^{k_4}  (-y_1)^{k_6} }{k_1!\, k_2!\, k_3!\, k_4!\, k_6!} \nonumber\\
    &\quad \Gamma (k_1+k_2+k_3-k_6-2 \epsilon +2) \Gamma (-k_2-k_3+k_6-\epsilon +1) \Gamma (-k_1-k_3-k_4+\epsilon -1)  \nonumber \\
    &\quad \frac{\Gamma (k_1+k_2+2 k_3+2 k_4+k_6+\epsilon ) \Gamma (k_3+k_4+2 \epsilon -1) \Gamma (-k_1-k_2-k_3+k_6+2 \epsilon -1)}{\Gamma (-k_2+k_4+k_6+\epsilon ) \Gamma (-k_1+3 \epsilon -2)} \allowdisplaybreaks \nonumber\\[\baselineskip]
    &\phi^\prime_{10} = \sum_{k_1,k_2,k_3,k_4,k_5,k_6=0}^\infty\frac{ (-y_3)^{k_1+k_2} (-y_2)^{k_3+k_5} (-y_1)^{k_4+k_6} }{k_1!\, k_2!\, k_3!\, k_4!\, k_5!\, k_6!} \nonumber\\
    &\quad \Gamma (k_2-k_3+k_4-k_5-\epsilon +1) \Gamma (k_1+k_3-k_4-k_6-\epsilon +1) \nonumber \\
    &\quad  \Gamma (-k_1-k_2+k_5+k_6-\epsilon +1) \Gamma (k_1+k_2+k_3+k_4+k_5+k_6+2 \epsilon -1)
\end{align}
with $y_i = \frac{m_i^2}{m_1^2+m_2^2+m_3^2+p^2}$ and $s=m_1^2+m_2^2+m_3^2+p^2$. All these series converge for small values of $y_i$ and the series representation can be obtained by a very simple algorithm, which is a straightforward implementation of the steps described in the previous sections. In fact, some of these $\Gamma$-series are related to each other. One can reduce the whole system only to $\phi^\prime_1, \phi^\prime_2, \phi^\prime_5$ and $\phi^\prime_{10}$ by the relations
\begin{align}
    \phi^\prime_1 (y_1,y_2,y_3) &= \phi^\prime_3 (y_2,y_1,y_3) = \phi^\prime_6 (y_1,y_3,y_2) \\
    \phi^\prime_2 (y_1,y_2,y_3) &= \phi^\prime_4 (y_1,y_3,y_2) = \phi^\prime_7 (y_3,y_2,y_1) \\
    \phi^\prime_5 (y_1,y_2,y_3) &= \phi^\prime_8 (y_2,y_3,y_1) = \phi^\prime_9 (y_2,y_1,y_3) \point	
\end{align}
By these relations one can also verify the expected symmetry of the Feynman integral under the permutation $y_1 \leftrightarrow y_2 \leftrightarrow y_3$.

In order to expand the Feynman integral $\mathcal I_\Aa$ for small values of $\epsilon>0$ one can use the methods described in \cref{sec:epsilonExpansion}. The correctness of these results was checked numerically by \softwareName{Fiesta} \cite{SmirnovFIESTAOptimizedFeynman2016} with arbitrary kinematics and masses, satisfying $y_i< 0.5$ up to the order $\epsilon^2$. For small values of $y_i$ the resulting series converge fast, such that for a good approximation one only has to take the first summands into account.


\chapter{Kinematic singularities} 

\label{ch:singularities}
    

Feynman integrals are usually understood as functions depending on various observables and parameters. Even though the physical observables do not take complex values in measurements, these Feynman integrals can only be thought consistently in complex domains. By considering Feynman integrals as complex functions and examining their analytic properties, surprising connections were found, for example dispersion relations and Cutkosky's rules \cite{CutkoskySingularitiesDiscontinuitiesFeynman1960, BlochCutkoskyRulesOuter2015}. As conjectured for the first time by T. Regge, it seems that these connections are not just arbitrary and indicate a more fundamental relation between the monodromy group and the fundamental group for Feynman integrals in the context of Picard-Lefschetz theory (see e.g.\ \cite{SpeerGenericFeynmanAmplitudes1971, PonzanoMonodromyRingsClass1969} for Regge's perspective). Apart from these conceptual questions, the analytic structure plays also an important role in many practical approaches, 
for example sector decomposition \cite{BinothNumericalEvaluationMultiloop2004, AnastasiouEvaluatingMultiloopFeynman2007, BorowkaNumericalEvaluationMultiloop2013}, Steinman relations \cite{CahillOpticalTheoremsSteinmann1975} or certain methods in QCD \cite{LibbyMassDivergencesTwoparticle1978}.

Hence, formally speaking, a Feynman integral maps a Feynman graph $\Gamma$ containing loops to a multivalued function $\mathcal I_\Gamma(\nuu,z)$, which depends on several variables $z$ and parameters $\nuu$. As specific representations of these functions $\mathcal I_\Gamma$, we can write down different kinds of integrals (see \cref{sec:ParametricFeynmanIntegrals}), each valid only on a restricted domain. Alternatively, we can express $\mathcal I_\Gamma$ also by series representations as done in \cref{thm:FeynSeries} or by means of various other functions (see \cref{sec:EulerIntegrals}). Thus, we do not want to use the term ``Feynman integral'' to refer to individual integrals, but rather to the analytic, common continuation of these representations to a maximal domain for the parameters and the variables.

In the process of analytic continuation to complex numbers there will arise two kinds of singularities: singularities in the parameters $\nuu$ and singularities in the variables $z$. The first type were already discussed in \cref{sec:DimAnaReg}. These singularities are only poles and $\mathcal I_\Gamma(\nuu,z)$ will be a meromorphic function with respect to $\nuu\in\mathbb C^{n+1}$. The analytic behaviour with respect to $\nuu$ was fully described by \cref{thm:meromorphicContinuation}. Considerably more difficult is the situation for the variables $z$ of the Feynman integral. We will find certain combinations of momenta $p$ and masses $m$, such that the Feynman integral has lacking analyticity or differentiability. This chapter will be devoted to the study of those singularities, which we will call \textit{kinematic singularities} or \textit{Landau singularities}.

Up to now, we restricted ourselves to Euclidean kinematics $\Re (z_j) >0$. Equivalently, we assumed that norms of momenta are real numbers, e.g.\ for those $\left| \sum_{e_i\notin F} \pm \widehat q_i \right| \in\mathbb R$ appearing in the description of Symanzik polynomials (\cref{thm:SymanzikPolynomialsTreeForest}). Therefore, when taking also the Minkowskian region into account, we have to consider those norms to be complex or more generally, we have to assume $z\in\mathbb C^N$. Hence, the analytic continuation of variables $z$ to complex numbers is indispensable when considering Minkowskian momenta\footnote{At first sight, it seems sufficient to continue the variables from $z_j\in\mathbb R_{>0}$ to the real numbers $z_j\in\mathbb R$. However, for real numbers we will expect poles in the integrands of Feynman integrals, e.g.\ in \cref{eq:FeynmanParSpFeynman}. Hence, we have to elude those poles by going to the complex plane. Therefore, we have to shift the integration contour in the complex region, or equivalently we can assume the variables $z_j\in\mathbb C$ to be complex. A minimal version of introducing complex numbers to \cref{eq:FeynmanMomSp} is the so-called $i\varepsilon$ prescription. We will elaborate on this in \cref{sec:Coamoebas}.}. \bigskip

Before studying the kinematic singularities on the level of Feynman integrals, let us first encounter this subject from a different angle, which will motivate the appearance of thresholds. This perspective was developed in the '60s, e.g.\ in \cite{GunsonUnitarityMassShell1965, CosterPhysicalRegionDiscontinuity1969, CosterPhysicalRegionDiscontinuity1970}, and we will recall the very concise summary in \cite{HannesdottirWhatVarepsilonSmatrix2022}. As aforementioned in \cref{sec:FeynmanIntegralsIntro} the $S$-matrix is the central object of interest, describing the probabilities for certain events in a scattering experiment. From the conservation of probabilities $S^\dagger S=\mathbbm 1$ and the separation of the trivial scattering $\gls{Smatrix}=\mathbbm 1 + i \gls{Tmatrix}$, we will obtain
\begin{align}
	T T^\dagger = \frac{1}{i} \left(T - T^\dagger\right) = 2 \Im (T) \label{eq:OpticalTheorem} \comma
\end{align}
which is often referred as the optical theorem. From \cref{eq:OpticalTheorem} we will obtain $T^\dagger = (\mathbbm 1 +i T)^{-1} T = i \sum_{k\geq 0} (-iT)^{k+1}$ by Neumann series. Hence, we have
\begin{align}
	 \Im (T) = \frac{1}{2i} \left(T - T^\dagger\right) = - \frac{1}{2} \sum_{k\geq 2}(-iT)^k \label{eq:SMatrixImSeries} \point
\end{align}
Every term in this series stands for a sequence of $k$ separated, non-trivial scattering processes. In this manner, \cref{eq:SMatrixImSeries} is an expression of the fact that we typically cannot determine in a scattering experiment whether it is a single scattering process or a chain of such processes\footnote{This holds independently of the indistinguishability of particles in QFTs and is simply an effect of the experimental setup.}. This chain of processes will be connected by real (on-shell), intermediate particles. Hence, we will only have processes in this chain, if they are kinematically allowed, i.e.\ if there is enough center of mass energy to produce these intermediate particles. Therefore, from \cref{eq:SMatrixImSeries} we will expect ``jumps'' in the imaginary part of $T$ when exceeding certain center of mass energies, such that a new chain of processes is kinematically allowed.

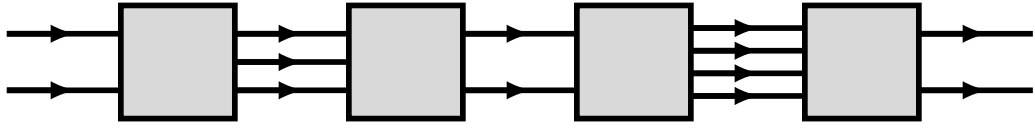
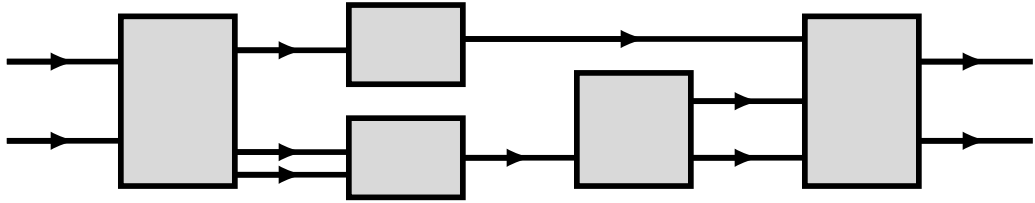
\begin{figure}[hbt]
	\centering
	\begin{subfigure}{\textwidth}
		\centering
		\begin{tikzpicture}[thick, scale=1.5]
			\filldraw[draw=black, line width=0.75mm, fill = gray!30] (1,0) rectangle ++(1,1);
			\filldraw[draw=black, line width=0.75mm, fill = gray!30] (3,0) rectangle ++(1,1);
			\filldraw[draw=black, line width=0.75mm, fill = gray!30] (5,0) rectangle ++(1,1);
			\filldraw[draw=black, line width=0.75mm, fill = gray!30] (7,0) rectangle ++(1,1);
			
			\draw [-latex, line width=0.75mm] (0,0.25) -- ++(0.6,0);
			\draw [-, line width=0.75mm]      (0,0.25) -- ++(1,0);
			\draw [-latex, line width=0.75mm] (0,0.75) -- ++(0.6,0);
			\draw [-, line width=0.75mm]      (0,0.75) -- ++(1,0);
			
			\draw [-latex, line width=0.75mm] (2,0.25) -- ++(0.6,0);
			\draw [-, line width=0.75mm]      (2,0.25) -- ++(1,0);
			\draw [-latex, line width=0.75mm] (2,0.5) -- ++(0.6,0);
			\draw [-, line width=0.75mm]      (2,0.5) -- ++(1,0);
			\draw [-latex, line width=0.75mm] (2,0.75) -- ++(0.6,0);
			\draw [-, line width=0.75mm]      (2,0.75) -- ++(1,0);
			
			\draw [-latex, line width=0.75mm] (4,0.25) -- ++(0.6,0);
			\draw [-, line width=0.75mm]      (4,0.25) -- ++(1,0);
			\draw [-latex, line width=0.75mm] (4,0.75) -- ++(0.6,0);
			\draw [-, line width=0.75mm]      (4,0.75) -- ++(1,0);
			
			\draw [-latex, line width=0.75mm] (6,0.2) -- ++(0.6,0);
			\draw [-, line width=0.75mm]      (6,0.2) -- ++(1,0);
			\draw [-latex, line width=0.75mm] (6,0.4) -- ++(0.6,0);
			\draw [-, line width=0.75mm]      (6,0.4) -- ++(1,0);
			\draw [-latex, line width=0.75mm] (6,0.6) -- ++(0.6,0);
			\draw [-, line width=0.75mm]      (6,0.6) -- ++(1,0);
			\draw [-latex, line width=0.75mm] (6,0.8) -- ++(0.6,0);
			\draw [-, line width=0.75mm]      (6,0.8) -- ++(1,0);
			
			\draw [-latex, line width=0.75mm] (8,0.25) -- ++(0.6,0);
			\draw [-, line width=0.75mm]      (8,0.25) -- ++(1,0);
			\draw [-latex, line width=0.75mm] (8,0.75) -- ++(0.6,0);
			\draw [-, line width=0.75mm]      (8,0.75) -- ++(1,0);
		\end{tikzpicture}
		\caption{Illustration of a chain of processes producing a normal threshold.} \label{fig:normalThreshold}
	\end{subfigure} \\
    \vspace{1cm}
    \begin{subfigure}{\textwidth}
		\centering
		\begin{tikzpicture}[thick, scale=1.5]
			\filldraw[draw=black, line width=0.75mm, fill = gray!30] (1,0) rectangle ++(1,1.5);
			\filldraw[draw=black, line width=0.75mm, fill = gray!30] (3,-0.1) rectangle ++(1,0.7);
			\filldraw[draw=black, line width=0.75mm, fill = gray!30] (3,0.9) rectangle ++(1,0.7);
			\filldraw[draw=black, line width=0.75mm, fill = gray!30] (5,0) rectangle ++(1,1);
			\filldraw[draw=black, line width=0.75mm, fill = gray!30] (7,0) rectangle ++(1,1.5);
			
			\draw [-latex, line width=0.75mm] (0,0.4) -- ++(0.6,0);
			\draw [-, line width=0.75mm]      (0,0.4) -- ++(1,0);
			\draw [-latex, line width=0.75mm] (0,1.1) -- ++(0.6,0);
			\draw [-, line width=0.75mm]      (0,1.1) -- ++(1,0);
			
			\draw [-latex, line width=0.75mm] (2,0.1) -- ++(0.6,0);
			\draw [-, line width=0.75mm]      (2,0.1) -- ++(1,0);
			\draw [-latex, line width=0.75mm] (2,0.3) -- ++(0.6,0);
			\draw [-, line width=0.75mm]      (2,0.3) -- ++(1,0);
			\draw [-latex, line width=0.75mm] (2,1.2) -- ++(0.6,0);
			\draw [-, line width=0.75mm]      (2,1.2) -- ++(1,0);
			
			\draw [-latex, line width=0.75mm] (4,0.25) -- ++(0.6,0);
			\draw [-, line width=0.75mm]      (4,0.25) -- ++(1,0);
			\draw [-latex, line width=0.75mm] (4,1.3) -- ++(1.6,0);
			\draw [-, line width=0.75mm]      (4,1.3) -- ++(3,0);
			
			\draw [-latex, line width=0.75mm] (6,0.25) -- ++(0.6,0);
			\draw [-, line width=0.75mm]      (6,0.25) -- ++(1,0);
			\draw [-latex, line width=0.75mm] (6,0.75) -- ++(0.6,0);
			\draw [-, line width=0.75mm]      (6,0.75) -- ++(1,0);
			
			\draw [-latex, line width=0.75mm] (8,0.4) -- ++(0.6,0);
			\draw [-, line width=0.75mm]      (8,0.4) -- ++(1,0);
			\draw [-latex, line width=0.75mm] (8,1.1) -- ++(0.6,0);
			\draw [-, line width=0.75mm]      (8,1.1) -- ++(1,0);
		\end{tikzpicture}
		\caption{Illustration of a ``chain'' of processes producing an anomalous threshold.} \label{fig:anomalousThreshold}
	\end{subfigure}
	\caption[Illustration of normal/anomalous thresholds from $S$-matrix theory]{Illustration of normal/anomalous thresholds from $S$-matrix theory. The figures are oriented towards \cite{CosterPhysicalRegionDiscontinuity1969,CosterPhysicalRegionDiscontinuity1970}. Every rectangle stands for a specific type of a process, i.e.\ a sum of Feynman graphs.} \label{fig:thresholds}
\end{figure}

We can classify those chains of scattering processes in two different types. In the first situation all intermediate, real particles outgoing of an intermediate process will join the next process (\cref{fig:normalThreshold}). Jumps in $\Im (T)$ relating to this situation are called \textit{normal thresholds}. They appear if the energy exceeds a squared sum of particle masses in the considered theory. The second possible situation we can imagine is, that the outgoing intermediate particles participate in several distinct further processes (\cref{fig:anomalousThreshold}). The analysis of those chains is much more involved and jumps in $\Im (T)$ relating to this situation are called \textit{anomalous thresholds}. 

Since the transfer matrix $T$ is built by algebraic expressions related to Feynman graphs, we will find the same analytic behaviour also on the level of Feynman integrals. Thus, Feynman integrals will also have singularities with respect to their variables $z$, whenever the energies of incoming momenta allow particles to go on-shell. Hence, those thresholds will also be apparent in Feynman integrals, and we will focus mainly on the anomalous thresholds, which are also known as Landau singularities or kinematic singularities. \bigskip

To begin with, we want to have a look at the kinematic singularities from the perspective of momentum space Feynman integrals \cref{eq:FeynmanMomSp}, where we will now allow the norm of momenta to be complex $|q_i|\in\mathbb C$. In these integrals those singularities may appear if some inverse propagators $D_i=q_i^2+m_i^2$ vanish and additionally the integration contour is trapped in such a way, that we are not able to elude the singularity by deforming the contour in the complex plane. These situations are called pinches and if they appear, the equations
\begin{align}
    x_i D_i = 0 \qquad \text{for all} \quad i=1,\ldots,n \label{eq:MomLandau1}\\
    \frac{\partial}{\partial k_j} \sum_{i=1}^n x_i D_i = 0 \qquad \text{for all} \quad j=1,\ldots,L \label{eq:MomLandau2}
\end{align}
have a solution for $x\in\mathbb C^n\setminus\{0\}$ and $k\in\mathbb C^{L\times d}$. We will call every point $z$ admitting such a solution a \textit{Landau critical point}. Landau critical points do not depend on the choice of internal momenta or their orientation.\bigskip

The equations \cref{eq:MomLandau1}, \cref{eq:MomLandau2} are called \textit{Landau equations} and were independently derived in 1959 by Bjorken \cite{BjorkenExperimentalTestsQuantum1959}, Landau \cite{LandauAnalyticPropertiesVertex1959} and Nakanishi \cite{NakanishiOrdinaryAnomalousThresholds1959}. We recommend \cite{MizeraCrossingSymmetryPlanar2021} for a comprehensive summary of the known research results in Landau's analysis of over 60 years and restrict ourselves to a very short historical overview. For a summary of the first steps of this subject from the 1960s we refer to the monograph of Eden et al. \cite{EdenAnalyticSmatrix1966}. A much more mathematically profound investigation was carried out by Pham et al. in terms of homology theory \cite{HwaHomologyFeynmanIntegrals1966, PhamSingularitiesIntegrals2011}. Pham's techniques have recently brought back into focus by Bloch and Kreimer \cite{BlochCutkoskyRulesOuter2015}. An alternative approach avoiding the introduction of homology theory was initiated by Regge, Ponzano, Speer and Westwater \cite{PonzanoMonodromyRingsClass1969, SpeerGenericFeynmanAmplitudes1971}. Their work was also the starting point for a mathematical treatment due to Kashiwara and Kawai \cite{KashiwaraConjectureReggeSato1976}, Sato \cite{SatoRecentDevolpmentHyperfunction1975} and Sato et al. \cite{SatoHolonomyStructureLandau1977}, which are all heavily based on $D$-modules. Currently, there is a renewed interested in Landau varieties, and we refer to \cite{CollinsNewCompleteProof2020, MuhlbauerMomentumSpaceLandau2020, MizeraLandauDiscriminants2021} for a selection of modern approaches as well as \cite{BonischAnalyticStructureAll2021, BonischFeynmanIntegralsDimensional2021}, where the analytic structure of specific Feynman integrals was studied in the context of differential equations by methods of topological string theory on Calabi-Yau manifolds. However, as already mentioned in \cite{MizeraCrossingSymmetryPlanar2021, CollinsNewCompleteProof2020} there is dismayingly little known about kinematic singularities. \bigskip

Strictly speaking the Landau equations \cref{eq:MomLandau1}, \cref{eq:MomLandau2} are neither necessary nor sufficient conditions to have a singularity of the analytic continuated Feynman integral. Thus, there are on the one hand singularities which do not correspond to a solution of Landau equations. Those singularities are often called second-type singularities or non-Landau singularities and were found for the first time in \cite{CutkoskySingularitiesDiscontinuitiesFeynman1960}. And on the other hand, not all solutions of Landau equations result in a singularity \cite{CollinsNewCompleteProof2020}. However, Landau equations are necessary and sufficient for the appearance of a trapped contour \cite{CollinsNewCompleteProof2020} and can be a necessary condition for certain restrictions on Feynman integrals. Apart from the distinction between normal thresholds and anomalous thresholds, certain further distinctions are common. Singularities with all $x_i\neq 0$ in \cref{eq:MomLandau1}, \cref{eq:MomLandau2} are known as \textit{leading singularities}, and we will further distinguish between solutions with real positive $x_i\geq 0$ and general complex $x_i\in\mathbb C$. Landau critical points corresponding to solutions with $x\in (\mathbb{C}\setminus\mathbb R_{>0})^n$ are also known as \textit{pseudo thresholds}.\bigskip

As pointed out above, Landau equations are understood to determine when internal (virtual) particles going on-shell. Hence, the Feynman diagram describes then an interaction between real particles with a specific lifetime \cite{ColemanSingularitiesPhysicalRegion1965}. The extraordinary meaning for Feynman integrals owing the Landau singularities also from various methods, which construct the whole Feynman integral on the basis of these singularities. All these methods root more or less in the optical theorem and the corresponding unitarity cuts, introduced by Cutkosky \cite{CutkoskySingularitiesDiscontinuitiesFeynman1960} shortly after Landau's article. However, it should be mentioned that Cutkosky's rules are unproven up today. We refer to \cite{BlochCutkoskyRulesOuter2015} for the recent progress of giving a rigorous proof of Cutkosky's rules. \bigskip

In this chapter, we want to take a look at kinematic singularities from the perspective of $\Aa$-hypergeometric systems. Especially, we can combine this with the considerations about the singular locus of $\Aa$-hypergeometric systems $\Sing(H_\Aa(\nuu))$ from \cref{ssec:SingularLocusAHyp}. This new point of view enables us to give a mathematical rigorous description of kinematic singularities. Furthermore, we will also notice certain discrepancies to the classical treatment of Landau varieties. Inter alia, it appears that in the common factorization of Landau varieties into leading Landau varieties of subgraphs certain non-trivial contributions were overlooked. Moreover, by means of the results in $\Aa$-hypergeometric theory, the most considerations can be reduced to polytopal combinatorics instead of algebraic topology as done in \cite{PhamSingularitiesIntegrals2011, HwaHomologyFeynmanIntegrals1966, MuhlbauerMomentumSpaceLandau2020, BrownPeriodsFeynmanIntegrals2010}. We also want to draw attention to the very interesting work of \cite{MizeraLandauDiscriminants2021}, published shortly after the article  \cite{KlausenKinematicSingularitiesFeynman2022} that constitutes the basis of this chapter.\bigskip

We will begin in \cref{sec:LandauVarieties} with a discussion of the Landau variety, which is the central object for the analytic structure of Feynman integrals, and we will relate this variety to principal $A$-determinants. Within this framework we will also notice overlooked contributions in Landau varieties for graphs beyond one loop or the banana family. Moreover, this relation to the principal $A$-determinant enables us to give an efficient but indirect determination of Landau varieties by means of the \HKp. However, the Landau variety will not describe all kinematic singularities, and we will consider all remaining singularities, also known as second-type singularities, in \cref{sec:2ndtypeSingularities}. By those methods, we will exemplarily determine the Landau variety of the dunce's cap graph in \cref{sec:ExampleDuncesCap}. 

Last but not least we will give a glimpse to the monodromy structure of Feynman integrals. Since kinematic singularities result in a non-trivial monodromy, Feynman integrals become multivalued functions. Unfortunately, the sheet structure of Feynman integrals is usually very sophisticated, and we will propose a related concept which is slightly simpler in \cref{sec:Coamoebas} called the coamoeba.


\section{Landau varieties} \label{sec:LandauVarieties}


Landau varieties are the central objects in the study of kinematic singularities. Unfortunately, Landau varieties come with several subtleties. In this section we want to give a definition of Landau varieties, and we will also discuss several of those subtleties. Furthermore, we want to relate Landau varieties to principal $A$-determinants. This will allow us to draw various consequences from the $\Aa$-hypergeometric theory to Landau varieties. In particular, we will see certain discrepancies in the classical approach of treating Landau varieties. But before giving a definition of Landau varieties, we will start with a reformulation of the Landau equations.\bigskip

The Landau equations stated in \cref{eq:MomLandau1}, \cref{eq:MomLandau2} involve the integration variables in momentum space as well as the integration variables of parametric space. There are also equivalent equations, which are stated in the parametric variables only. As aforementioned, the second Symanzik polynomial $\Ff (p,x)$ can be written as a discriminant of $\Lambda (k,p,x)\, \Uu(x)$ with respect to the loop momenta $k$, where $\Lambda(k,p,x) = \sum_{i=1}^n x_i D_i$ was defined in \cref{eq:LambdaxDrelation}. Hence, it is immediately clear, that Landau's equations  \cref{eq:MomLandau1}, \cref{eq:MomLandau2} will be conditions on the second Symanzik polynomial $\Ff$. Instead of eliminating $k$ from \cref{eq:MomLandau1}, \cref{eq:MomLandau2}, one can show the Landau equations in parametric space also directly by considering the parametric integral representations \cite{NakanishiOrdinaryAnomalousThresholds1959}.
      
\begin{lemma}[Parametric space Landau equations \cite{NakanishiGraphTheoryFeynman1971, EdenAnalyticSmatrix1966}]     
    \label{lem:ParLandau}
    Under the assumption $\Uu(x)\neq 0$, a point $z\in\mathbb C^N$ is a Landau critical point, if and only if the equations
    \begin{align}
    	x_i \pd{\Ff}{x_i} &= 0 \quad \text{for } i=1,\ldots,n \label{eq:LandauEqParameterSpA} \\
        \Ff &= 0 \label{eq:LandauEqParameterSpB}
    \end{align}
    have a solution in $x\in\mathbb P^{n-1}_{\mathbb C}$. Solutions with $\Uu(x)=0$ are connected with the second-type singularities, which we will examine in \cref{sec:2ndtypeSingularities}.
\end{lemma}
\begin{proof}
    ``$\Rightarrow$'':
    Consider $\Lambda=k^\top\! M k + 2 Q^\top\! k + J$ from equation \cref{eq:LambdaxDrelation}. We find
    \begin{align}
        \pd{\Lambda(k,p,x)}{k} &= 2 Mk + 2 Q^\top \point
    \end{align}
    By the assumption $\Uu\neq 0$, $M$ is invertible and thus $\pd{\Lambda}{k_j}=0$ from \cref{eq:MomLandau2} implies $k=-M^{-1}Q^\top$. Inserting this equation for $k$ in  \cref{eq:LambdaxDrelation} and comparing with \cref{eq:SymanzikPolynomialsMatrices} we find $\Lambda(-M^{-1}Q^\top\!,p,x) = \Ff/\Uu$. Hence, we have shown that $\Ff$ is the discriminant of $\Uu \Lambda$ with respect to $k$. Therefore, $\Lambda=0$ and $\pd{\Lambda}{k_j}=0$ implies $\Ff=0$ and furthermore
    \begin{align}
        x_j \pd{\Ff}{x_j} = x_j \frac{\partial}{\partial x_j} \left( \Uu \Lambda \right) = \Uu x_j D_j = 0 \point
    \end{align}
    ``$\Leftarrow$'': Since the equations \cref{eq:MomLandau1},\cref{eq:MomLandau2} contain more variables than in parameter space, we can always find a value $k^\prime$, s.t. $\Lambda(k^\prime,p,x) = \Ff/\Uu$, without restricting the possible solutions for $x$. This can also be vindicated by the fact that the Feynman integral \cref{eq:FeynmanMomSp} is invariant under linear transformations of loop momenta. We conclude
    \begin{align}
        \frac{\partial\Lambda}{\partial k'} &=  2Mk^\prime + 2Q^\top = 0\\
        x_j D_j &= x_j \frac{\partial \Lambda}{\partial x_j} = x_j \frac{\partial \Uu^{-1}}{\partial x_j} \Ff + x_j \frac{\partial \Ff}{\partial x_j} \Uu = 0 \point
    \end{align}
\end{proof}

Note, that by Euler's theorem one of the $n+1$ equations in \cref{lem:ParLandau} is redundant, which is the reason why we look for solutions in projective space. \bigskip

According to \cite{MizeraLandauDiscriminants2021, BrownPeriodsFeynmanIntegrals2010, PhamSingularitiesIntegrals2011, MuhlbauerMomentumSpaceLandau2020} we will call the variety defined by the Zariski closure of all Landau critical points the \textit{Landau variety} $\mathcal L(\mathcal I_\Gamma)$. Due to Riemann's second removable theorem \cite{KaupHolomorphicFunctionsSeveral1983}, all singularities corresponding to a part of the Landau variety with $\codim \mathcal L(\mathcal I_\Gamma) > 1$ are removable singularities. Hence, we can focus on the codimension one part of $\mathcal L(\mathcal I_\Gamma)$, which we will denote by $\gls{Landau1}$. Based on the Landau equations in parameter space (\cref{lem:ParLandau}), we can directly read off the following relation between the Landau variety and the principal $A$-determinant (see \cref{ssec:principalAdet}). The following theorem can also be understood as an alternative definition of Landau varieties.

\begin{theorem}[Landau variety] \label{thm:LandauVar}
    Let $\Ff\in\mathbb C[x_1,\ldots,x_n]$ be the second Symanzik polynomial of a Feynman graph $\Gamma$ and let $\Aa_\Ff\subset \mathbb Z^n$ be the support of $\Ff$, i.e.\ $\Ff = \sum_{a^{(j)}\in \Aa_\Ff} z_j x^{a^{(j)}}$. The Landau variety is given by the (simple) principal $A$-determinant of $\Ff$
    \begin{align}
        \mathcal L_1(\mathcal I_\Gamma) = \Var\!\left( E_{\Aa_\Ff}(\Ff)\right) =  \Var\!\left( \widehat E_{\Aa_\Ff}(\Ff)\right) \point
    \end{align}
    In particular, $\mathcal L_1(\mathcal I_\Gamma)$ is independent of the parameters $\nuu$.
\end{theorem}

Usually, one  splits the calculation of Landau singularities into a leading singularity with all $x_i\neq 0$ in \cref{eq:LandauEqParameterSpA}, \cref{eq:LandauEqParameterSpB} and non-leading singularities with $x_i=0$ for $i\in I$, where $\emptyset\neq I\subsetneq \{1,\ldots,n\}$ \cite{EdenAnalyticSmatrix1966}. Every non-leading singularity can be interpreted as a leading singularity of the subgraph where all edges corresponding to $I$ are contracted. This is due to \cref{eq:UContracted} and its corresponding identity for $\Ff_0$. Hence, setting a Schwinger parameter $x_i$ in the second Symanzik polynomial $\Ff$ to zero is equivalent to considering a subgraph, where the corresponding edge $e_i$ is contracted. Note that the second Symanzik polynomial vanishes if the edge $e_i$ corresponds to a tadpole. This procedure seems at first glance to be a natural distinction of cases appearing in \cref{eq:LandauEqParameterSpA} and \cref{eq:LandauEqParameterSpB} and works for special cases. However, when considering systems of polynomials in several variables as \cref{eq:LandauEqParameterSpA}, \cref{eq:LandauEqParameterSpB}, we should rather consider the ideal generated by these polynomials \cite{SturmfelsSolvingSystemsPolynomial2002}. Due to the more intricate situation with multivariate polynomials, we cannot expect primary decomposition \cite{GreuelSingularIntroductionCommutative2002} to work so naively, except for special cases. Hence, in general we will not expect such a simple decomposition of Landau varieties $\mathcal L_1(\mathcal I_\Gamma)$.

When comparing with the results of \cref{ssec:principalAdet}, we will find a similar but in general different splitting. According to \cref{thm:pAdet-factorization} we have a factorization
\begin{align}
	\widehat E_{\Aa_\Ff}(\Ff) = \pm \prod_{\tau\subseteq \Newt(\Ff)} \Delta_{\Aa_\Ff\cap\tau} (\Ff_\tau) \label{eq:LandauFactorizationPrincipalADet}
\end{align}
into $A$-discriminants, where $\Ff_\tau$ denotes the truncated polynomial of $\Ff$ defined by \cref{eq:truncatedPolynomialDefinition} and the product runs over all faces $\tau$ of $\Newt(\Ff)$. The decomposition into subgraphs and the one into faces of $\Newt(\Ff)$ coincide if the second Symanzik polynomial $\Ff$ consists in all monomials of a given degree. However, in general the procedure of subgraphs will miss certain contributions.

\begin{lemma} \label{lem:SubgraphPolynomialsVsTruncated}
    For an index set $I\subseteq\{1,\ldots,n\}$ we call $\gls{FIx} :=\Ff(x)|_{\{x_i=0\,|\, i\in I\}}$ the subgraph polynomial associated to $I$. Every subgraph polynomial is also a truncated polynomial $\Ff_\tau$ with a face $\tau\subseteq\Newt(\Ff)$. The converse is true if $\Ff$ consists in all monomials of degree $L+1$. However, the converse is not true in general.
\end{lemma}
\begin{proof}
    Let $\phi(p) = - \sum_{i=1}^n b_i p_i$ be a linear form with $b_i=1$ for all $i\in I$ and $b_i=0$ otherwise. This linear form takes its maximal value  $\max \phi(p) =0$ for precisely those points $p\in\mathbb R_+^n$ with $p_i=0$ for $i\in I$. Since all points of $\Newt(\Ff)$ are contained in the positive orthant $\mathbb R_+^n$, such a linear map $\phi$ defines the corresponding face $\tau$ according to \cref{eq:facedef}.

    In case, where $\Ff$ consists in all possible monomials of a given degree, the Newton polytope is an $n$-simplex. The faces of this simplex are trivially in a one to one correspondence with subsets of $\{1,\ldots,n\}$.
\end{proof}

Thus, beyond $1$-loop graphs and banana graphs, which contain all monomials of a given degree in the second Symanzik polynomial, one may get additional singularities from the truncated polynomials, which will be missed with the approach of the subgraphs. Remarkably, there are also non-trivial factors, which will be missed by the classical approach of subgraphs, i.e.\ there are $A$-discriminants of missed truncated polynomials which are neither $1$ nor contained in another $A$-discriminant. For example the Landau variety of the dunce's cap graph will contain a factor which has the shape of a Landau variety of a $1$-loop bubble graph. This contribution will be overlooked by subgraphs, see \cref{sec:ExampleDuncesCap}. Hence, this observation describes a serious and unexpected issue in the current understanding of Landau varieties beyond $1$-loop and banana graphs. Moreover, those additional, overlooked singularities can also appear on the principal sheet, i.e.\ we will find solutions $x\in\mathbb R^n_{> 0}$ (see \cref{sec:ExampleDuncesCap}).

That there is a serious issue in the decomposition of the Landau variety in the classical approach with subgraphs was also indicated in \cite{LandshoffHierarchicalPrinciplePerturbation1966} and further discussed in \cite{BoylingHomologicalApproachParametric1968}\footnote{I would like to thank Marko Berghoff who brought these two articles to my attention.}. By means of the principal $A$-determinant, this problem can now be cleared up and the correct decomposition of Landau varieties can be easily described by the truncated polynomials of $\Ff$. \bigskip

Apart from those general questions, the relation between Landau varieties and principal $A$-determinants leads also to a very efficient tool to determine Landau varieties. By means of \HKP (see \cref{ssec:Adiscriminants}) one can compute a parameterization of the Landau variety very fast. If we decompose the Landau variety into its irreducible components $\mathcal L_1 (\mathcal I_\Gamma) = \bigcup_\tau \mathcal L_1^{(\tau)} (\mathcal I_\Gamma)$ in the sense of \cref{eq:LandauFactorizationPrincipalADet}, every component corresponds to an $A$-discriminant. Hence, we can write these components as the image of a parameterization $\psi^{(\tau)}$ defined in \cref{eq:HKPpsi}
\begin{align}
	\mathcal L_1^{(\tau)} (\mathcal I_\Gamma) &= \psi^{(\tau)} \!\left(\mathbb P^{r-1}_{\mathbb C} \right) \qquad\text{with} \nonumber \\
    \psi^{(\tau)} [t_1 : \ldots : t_r ] &= \left( \prod_{i=1}^N \left(\sum_{j=1}^r b^{(\tau)}_{ij} t_j \right)^{b^{(\tau)}_{i1}}, \ldots, \prod_{i=1}^N \left(\sum_{j=1}^r  b^{(\tau)}_{ij} t_j \right)^{b^{(\tau)}_{ir}}\right) \label{eq:LandauHKPpsi}
\end{align}
where $b^{(\tau)}_{ij}$ are the components of a Gale dual of $\Aa\cap\tau\in\mathbb Z^{(n+1)\times N}$ and $r=N-n-1$. Hence, one only has to determine Gale duals of the corresponding vector configurations of $\Aa\cap\tau$. Since Gale duals can be determined very efficiently, this is particularly pleasing since the existing results for Landau varieties in the literature are limited to very few number of graphs \cite{OliveSingularitiesScatteringAmplitudes1962,IslamLeadingLandauCurves1966,RiskAnalyticityEnvelopeDiagrams1968}. Recently, there was published another new tool for an efficient determination of Landau varieties \cite{MizeraLandauDiscriminants2021} which also extends the calculable graphs significantly. We will present the scope of \HKP with an example in \cref{sec:ExampleDuncesCap}. However, we have to recall that \HKP only generates an indirect representation of Landau varieties. Depending on the purpose, a parameterization can be even more useful than the generating polynomial of the Landau variety. But turning \cref{eq:LandauHKPpsi} in a generating equation of a variety by elimination of parameters is still a very costly task.\bigskip

We want further to recall the observation of \cref{thm:NewtSec}, which relates the Newton polytope of the principal $A$-determinant $E_{\Aa_\Ff}(\Ff)$ by the secondary polytope $\Sigma(\Aa_\Ff)$. As sketched in \cref{ssec:principalAdet}, one can use this relation to determine the defining equation of the Landau variety by means of all regular triangulations of $\Conv(\Aa_\Ff)$. We refer to \cref{ex:cubicPrincipalAdetFromTriangs} for an illustration of this idea. Furthermore, this connection gives also a lower bound for the number of monomials in the defining polynomial of $\mathcal L_1(\mathcal I_\Gamma)$. Thus, this polynomial will contain at least as many monomials as there are regular triangulations of $\Newt(\Ff)$. Even though this estimation is far from being a sharp bound\footnote{The number of regular triangulations counts only the extreme monomials of the defining polynomial of $\mathcal L_1(\mathcal I_\Gamma)$, i.e.\ the vertices of $\Newt(E_{\Aa_\Ff}(\Ff))$.}, the number of these triangulations is growing very fast (see \cref{tab:characteristics}). \bigskip

Another direct consequence of the relation to $\Aa$-hypergeometric theory is, that singularities with respect to $z$ can not be worse than logarithmic singularities. This can be seen by the representation of Feynman integrals in the neighbourhood of the singular locus by means of canonical series solution \cite{SaitoGrobnerDeformationsHypergeometric2000}. Hence, $\Aa$-hypergeometric systems are ``regular'' in a generalized sense of regular singular points in linear ordinary differential equations \cite[sec. 2.3]{CattaniThreeLecturesHypergeometric2006}. A similar result was also shown in \cite[sec. 7.2]{HannesdottirWhatVarepsilonSmatrix2022}.\bigskip

However, we need to point out a further subtlety in the definition of Landau varieties as well as in \cref{thm:LandauVar}. In both places we will assume that the coefficients in the second Symanzik polynomial $z\in\mathbb C^N$ are generic. In the physically relevant case, there will be relations among these coefficients, and they are not necessarily generic. Hence, the application of these extra relations will restrict the Landau variety to a subspace. This is not only an issue in parametric representation, which involves the Symanzik polynomials. It also appears in momentum space, where the external momenta are treated as vectors in $d$-dimensional Minkowskian space. Thus, there can not be more than $d$ independent external momenta, and additionally we suppose an overall momentum conservation. If the variables are not generic, it can occur that a factor in the defining equation of the Landau variety is identical to zero. Therefore, the Landau variety would cover the whole space. On the other hand, we know that there can not be a singularity with unbounded functional value for all points $z\in\mathbb C^N$ due to the convergence considerations from \cref{sec:DimAnaReg}. Thus, in the limit to the physically relevant case, we want to exclude such that ``overall singularities'' in order to make the other singularities apparent\footnote{The problem of vanishing defining equation of the Landau variety in the presence of physical relations between variables is known for a long time and will usually be ignored as we will do with \cref{eq:LandauPrincipalADetPhysical}, see e.g.\ \cite[sec. 2.10]{EdenAnalyticSmatrix1966} and also \cref{ex:1loopSecondtype}. However, a deeper understanding of this behaviour would be desirable, and we would identify this as one of the most uncharted areas in this subject.}. Therefore, we want to define
\begin{align}
    \gls{pAdetPh} := \, \pm \hspace{-2em} \prod_{\substack{\tau\subseteq\Newt(\Ff) \\ \left.\Delta_{\Aa_\Ff\cap\tau}(\Ff_\tau)\right|_{z\rightarrow z^{(ph)}} \neq 0}} \hspace{-1em} \Delta_{\Aa_\Ff\cap\tau}(\Ff_\tau) \label{eq:LandauPrincipalADetPhysical}
\end{align}
a simple principal $A$-determinant which contains only the physically relevant parts of the principal $A$-determinant, i.e.\ we omit the factors, which vanish after inserting the physical restrictions $z^{(ph)}$ on the variables. Equivalently, we define $\gls{Landau1Ph} := \Var\!\left( \widehat E_{\Aa_\Ff}^{ph}(\Ff)\right)$ as the physically relevant Landau variety. Note that the Landau variety $\mathcal L_1(\mathcal I_\Gamma)$ is independent of the parameters $\nuu$, whereas the physical Landau variety $\mathcal L_1^{ph}(\mathcal I_\Gamma)$ may principally depend on the parameters $\nuu$ since the relations between the variables $z$ may depend on $\nuu$. For example, specific choices of spacetime dimension $d$ can change the relations between the external momenta.\bigskip

The fact that the singularities of Feynman integrals are related to principal $A$-determinants, comes with no surprise. As aforementioned, the singular locus of an $\Aa$-hypergeometric function will always be generated by a principal $A$-determinant. Comparing the Landau variety $\mathcal L_1(\mathcal I_\Gamma)$ from \cref{thm:LandauVar} with the results of \cref{ssec:SingularLocusAHyp} about $A$-hypergeometric functions (especially \cref{thm:SingularLocusPrincipalAdet}), we would rather expect $\Var(E_{A_\Gg}(\Uu+\Ff))$ instead of $\mathcal L_1(\mathcal I_\Gamma)$ to be the singular locus of $\mathcal I_\Gamma$. Directly from the factorization of the principal $A$-determinant we can see the relation of these two varieties.
\begin{lemma} \label{lem:LandauVarietyContainedInSing}
    The Landau variety is contained in the singular locus of the $A$-hyper\-ge\-o\-me\-tric function:
    \begin{align}
        \mathcal L_1 (\mathcal I_\Gamma) = \Var\big(E_{\Aa_\Ff}(\Ff)\big) \subseteq \Var \big(E_{A_\Gg}(\Uu+\Ff)\big) = \Sing(H_{\Aa_\Gg}(\nuu)) \point
    \end{align}
\end{lemma}
\begin{proof}
    $\Uu$ and $\Ff$ are homogeneous polynomials of different degrees. Therefore, $\Newt(\Uu+\Ff)$ has points on two different, parallel hyperplanes and thus $\Newt(\Uu)$ and $\Newt(\Ff)$ are two facets of $\Newt(\Uu+\Ff)$. By \cref{thm:pAdet-factorization} we see that $\Var (E_{A_\Gg}(\Uu+\Ff)) = \Var (E_{\Aa_\Ff}(\Ff)) \cup \Var (E_{\Aa_\Uu}(\Uu)) \cup \Var (\Delta_{A_\Gg}(\Uu+\Ff)) \cup \Var(R)$, where the remaining polynomial $R$, corresponds to all discriminants coming from proper, mixed faces, i.e.\ faces $\tau\subsetneq \Newt(\Gg)$ having points of $\Uu$ and $\Ff$.
\end{proof}

Thus, \cref{lem:LandauVarietyContainedInSing} shows what we already indicated above: The Landau variety covers not all\footnote{These singularities of the Feynman integral, which are not contained in $\mathcal L_1(\mathcal I_\Gamma)$, have nothing to do with the overlooked singularities in $\mathcal L_1(\mathcal I_\Gamma)$  mentioned above. The overlooked singularities discussed above come from the fact that the usually assumed decomposition in subgraphs is in general not the correct approach. The singularities discussed here come from the fact that the second Symanzik polynomial $\Ff$ only represents a part of the Feynman integral.} kinematic singularities of the Feynman integrals and in general $\Var(E_{\Aa_\Ff}(\Ff))$ will be a proper subvariety of $\Var (E_{A_\Gg}(\Gg))$. Based on the prime factorization of the principal $A$-determinant, we will divide the singular locus of the Feynman integral $\Var (E_{A_\Gg}(\Gg))$ into four parts
\begin{align}
    \Var \big(E_{A_\Gg}(\Gg)\big) = \Var \big(\widehat E_{A_\Gg}(\Gg)\big) = \Var \big(\widehat E_{\Aa_\Ff}(\Ff)\big) \cup \Var \big(\widehat E_{\Aa_\Uu}(\Uu)\big) \cup \Var \big(\Delta_{A_\Gg}(\Gg)\big) \cup \Var (R) \point
\end{align}
Namely, we factorize $\widehat E_{A_\Gg}(\Gg)$ in a polynomial $\widehat E_{\Aa_\Ff}(\Ff)$ generating the classical Landau variety according to \cref{thm:LandauVar}, a polynomial $\widehat E_{\Aa_\Uu}(\Uu)$ which is constant in the physically relevant case and a polynomial $\Delta_{A_\Gg}(\Gg)$, which we will associate to the second-type singularities. The remaining polynomial
\begin{align}
    \gls{R} := \prod_{\substack{\tau\subsetneq \Newt(\Uu+\Ff)\\ \tau \nsubseteq \Newt (\Uu), \tau \nsubseteq \Newt (\Ff)}} \Delta_{A\cap\tau} (\Gg_\tau) \label{eq:defRproperMixedFaces}
\end{align}
will correspond to second-type singularities of subgraphs, and we will call the roots of $R$ the \textit{mixed type singularities of proper faces}. In the following section, we will analyze step by step these further contributions to the singular locus.\bigskip


\section{Second-type singularities} \label{sec:2ndtypeSingularities}


As aforementioned the defining polynomial of the singular locus $\Sing(H_\Aa(\nuu))$ splits into several $A$-discriminants. With the $A$-discriminant $\Delta_{A_\Gg}(\Uu+\Ff)$ we will associate the so-called second-type singularities \cite{CutkoskySingularitiesDiscontinuitiesFeynman1960, FairlieSingularitiesSecondType1962}. We have to remark, that the notion of second-type singularities differs slightly in various literature. Moreover, there is very little known about second-type singularities. Usually, a distinction is made between pure second-type singularities and mixed second-type singularities \cite{EdenAnalyticSmatrix1966, NakanishiGraphTheoryFeynman1971}. The pure second-type singularities do not depend on masses and can be expressed by Gram determinants, whereas the latter appear in higher loops. Second-type singularities are slightly better understood in momentum space, whereas they are endpoint singularities at infinity \cite{MuhlbauerMomentumSpaceLandau2020}. In parametric space, second-type singularities are connected to the case where $\Uu=0$ \cite[sec. 16]{NakanishiGraphTheoryFeynman1971}.\bigskip

In our approach we will call the variety generated by $\Delta_{A_\Gg}(\Gg)$ the second-type singularities. By introducing a new variable $x_0$, we can change to the homogeneous setting $\Delta_{\widetilde A_\Gg}(x_0 \Uu + \Ff)$ which has the same discriminant, since there is an appropriate injective, affine map connecting $A_\Gg$ with $\widetilde A_\Gg$ according to \cref{ssec:Adiscriminants}. Writing the corresponding polynomial equation system explicitly, the $A$-discriminant $\Delta_{A_\Gg}(\Gg)$ is the defining polynomial of the closure of the set of all coefficients $z\in\mathbb C^N$, such that the equations
\begin{align}
    \Uu = 0, \quad \Ff_0 = 0, \quad \pd{\Ff_0}{x_i} + \pd{\Uu}{x_i} \left( x_0 + \sum_{j=1}^n x_j m_j^2 \right) = 0 \quad\text{for}\quad i=1,\ldots,n \label{eq:secondtype} \point
\end{align}
have a solution for $(x_0,x)\in(\mathbb C^*)^{n+1}$.  As before, we denote by $\Ff_0$ the massless part of the second Symanzik polynomial. Not all conditions herein \cref{eq:secondtype} are independent, because the polynomial $x_0\Uu+\Ff$ is homogeneous again. Thus, we can drop an equation from \cref{eq:secondtype}. These equations for second-type singularities \cref{eq:secondtype} agree with the description in \cite[sec. 16]{NakanishiGraphTheoryFeynman1971}.

\begin{example}[\nth{2} type singularities of all banana graphs]
    Consider the family of massive $L$-loop $2$-point functions, which are also called banana graphs (see \cref{fig:bananaFamily}). These graphs having $n=L+1$ edges and the Symanzik polynomials
    \begin{align}
        \Uu = x_1 \cdots x_n \left(\frac{1}{x_1} + \ldots + \frac{1}{x_n} \right)\text{ , } \qquad  \Ff_0 = p^2 x_1\cdots x_n \point
    \end{align}
    Applying the conditions of \cref{eq:secondtype} we will find the second-type singularity for all banana graphs to be
    \begin{align}
        p^2 = 0 \point
    \end{align}
\end{example}
\begin{example}[\nth{2} type singularities of all $1$-loop graphs] \label{ex:1loopSecondtype}
    A massive $1$-loop graph with $n$ edges has Symanzik polynomials (see \cref{fig:1loopFamily}
    \begin{align}
        \Uu = x_1 + \ldots + x_n \qquad \Ff_0 = \sum_{1\leq i < j \leq n} u_{ij} x_i x_j 
    \end{align}
    where $u_{ij} := \left(\sum_{k=i}^{j-1} p_k\right)^2$ for $i<j$ defines the dependence on external momenta. We will set $u_{ij}=u_{ji}$ for $i>j$ and $u_{ii}=0$. By these definitions we obtain $\pd{\Uu}{x_j} = 1$ and $\pd{\Ff_0}{x_j} = \sum_{i\neq j} u_{ij} x_i$ for the derivatives. Since we can drop one equation from \cref{eq:secondtype} due to the homogenity of $x_0 \Uu +\Ff$, we will obtain a system of linear equations in the $1$-loop case. Eliminating $x_0$ by subtracting equations, we can combine these conditions to a determinant
    \begin{align}
        \begin{vmatrix}
        \hspace{.3cm} 1 & 1 & \cdots & 1 & \\
        \multicolumn{4}{c}{\multirow{2}{*}{$\bigl(u_{ij} - u_{jn}\bigr)_{\scalebox{0.7}{$\substack{1\leq i \leq n-1 \\ 1\leq j \leq n}$}}$}} \\
        &
        \end{vmatrix} = 0 \label{eq:1loopSecondtype}
    \end{align}
    as the (not necessarily irreducible) defining polynomial of the second-type singularity for all $1$-loop graphs. By the same argument as used in \cite[sec. 16]{NakanishiGraphTheoryFeynman1971}, the condition \cref{eq:1loopSecondtype} is equivalent to the vanishing of the Gram determinant of external momenta, which is usually set for the pure second-type singularity \cite{FairlieSingularitiesSecondType1962, EdenAnalyticSmatrix1966}. Note that in the $1$-loop case, the second-type singularity does not depend on masses. Furthermore, for higher $n$, the external momenta satisfy certain relations, since there can not be more than $d$ linearly independent vectors in $d$-dimensional Minkowski space. In addition, the external momenta ensure a conservation law. Thus, for\footnote{One can even find, that the Gram determinant vanishes for $n > d$, see \cite[sec. 2.10]{EdenAnalyticSmatrix1966}.} $n > d + 1$ the condition \cref{eq:1loopSecondtype} is satisfied for all physical external momenta. Hence, we will remove this contribution to the singular locus, when we restrict us to the physically relevant case as done by \cref{eq:LandauPrincipalADetPhysical} whenever $n > d+1$. We want to emphasize again that this removing of ``the unwanted zeros'' does not only appear for our special approach. As it can be seen by this example, such a phenomenon does also appear in the ``classical way'' of treating Landau singularities.
\end{example}

As the next contribution to the singular locus $\Sing(H_\Aa(\nuu))$ we will consider the principal $A$-determinant of the first Symanzik polynomial $E_{\Aa_\Uu}(\Uu)$. Since the coefficients of the first Symanzik polynomial are all equal to one, the principal $A$-determinant $E_{\Aa_\Uu}(\Uu)$ can be either $0$ or $\pm 1$.  Note that singularities corresponding to $E_{\Aa_\Uu}(\Uu)=0$ can only be pseudo thresholds, i.e.\ all solutions of $\Uu = \pd{\Uu}{x_1} = \ldots = \pd{\Uu}{x_n} = 0$ have to satisfy $x\notin \mathbb R_{>0}^n$. Moreover, we can determine $E_{\Aa_\Uu}(\Uu)$ also for certain classes of graphs.

\begin{lemma} \label{lem:EAUfor1LoopAndBanana}
    For all $1$-loop graphs and all $L$-loop banana graphs we obtain
    \begin{align}
        E_{\Aa_\Uu}(\Uu) = 1 \point
    \end{align}
\end{lemma}
\begin{proof}
    Note, that in both cases $\Newt(\Uu)$ describes a simplex. Since $\Uu$ has no free coefficients, $E_{\Aa_\Uu}(\Uu)$ is either $0$ or $\pm 1$. Therefore, in order to prove $E_{\Aa_\Uu}(\Uu)=1$ it is sufficient to show, that $\Uu=\pd{\Uu}{x_1} = \ldots=\pd{\Uu}{x_n}=0$ contains a contradiction. For $1$-loop graphs this contradiction is obvious since $\pd{\Uu}{x_i}=1$.

    For $L$-loop banana graphs consider the relation following from \cref{eq:UContractedDeleted}
    \begin{align}
        \Uu - x_i \pd{\Uu}{x_i} = - \frac{x_1 \cdots x_n}{x_i} = 0
    \end{align}
    which has no solution for $x\in(\mathbb C^*)^n$.
\end{proof}

We want to emphasize, that \cref{lem:EAUfor1LoopAndBanana} does not hold for general Feynman graphs. The dunce's cap graph is the simplest example, which will result in a vanishing principal $A$-determinant of the first Symanzik polynomial. Note, that a vanishing of $E_{\Aa_\Uu}(\Uu)$ will not correspond to a singularity with unbounded functional value, due to the convergence considerations of \cref{sec:DimAnaReg}. Hence, singularities stemming from $E_{\Aa_\Uu}(\Uu)$ can be justifiably omitted in the physical approach, as their singular behaviour will only show up when we will add variables to the first Symanzik polynomials, i.e.\ they have an impact on the neighbourhood of $z_i=1$. Therefore, we will exclude these overall contributions to the singular locus in the spirit of \cref{eq:LandauPrincipalADetPhysical} if we consider the physically relevant case. Thus, we will set $E_{\Aa_\Uu}^{ph}(\Uu)=1$. \bigskip

The last contribution $R$ to the singular locus of Feynman integrals defined in \cref{eq:defRproperMixedFaces} comes from discriminants of proper, mixed faces, i.e.\ proper faces of $\Newt(\Uu+\Ff)$ which are neither completely contained in $\Newt(\Ff)$ nor in $\Newt(\Uu)$. These $A$-discriminants can be associated with second-type singularities of subgraphs as it can be observed in the following example.

\begin{example}[proper, mixed singularities for the triangle graph] \label{ex:triangleGraphProperMixed}

    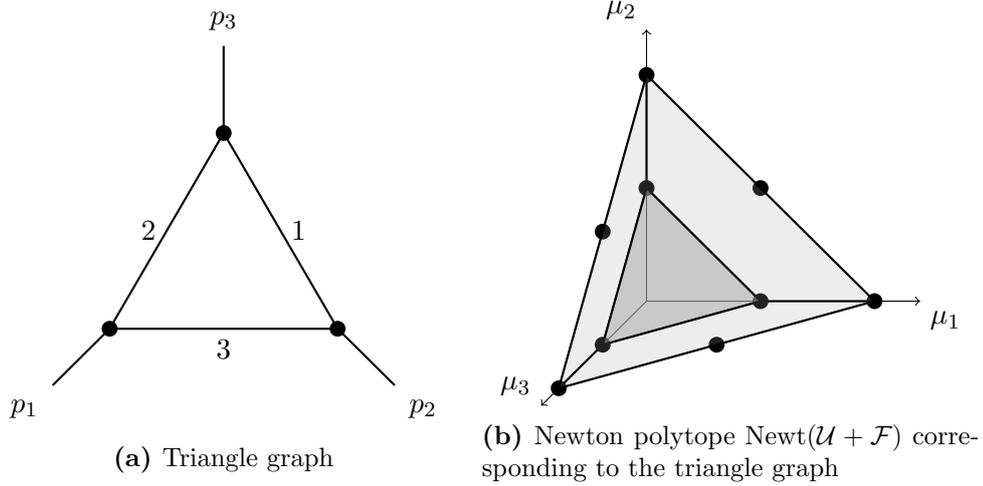
\begin{figure}
    	\centering
    	\begin{subfigure}{.45\textwidth}
            \centering
            \begin{tikzpicture}[thick, dot/.style = {draw, shape=circle, fill=black, scale=.5}, scale=1.5]
                \coordinate[dot] (A) at (-1,0);
                \coordinate (A1) at (-1.5,-0.5);
                \coordinate[dot] (B) at (1,0);
                \coordinate (B1) at (1.5,-0.5);
                \coordinate[dot] (C) at (0,1.73);
                \coordinate (C1) at (0,2.5);
                
                \draw (A) -- node[below] {$3$} (B);
                \draw (B) -- node[right] {$1$} (C);
                \draw (C) -- node[left] {$2$} (A);
                \draw (A1) -- (A);
                \draw (B1) -- (B);
                \draw (C1) -- (C);
                
                \node at (A1) [below left = 0.7mm of A1] {$p_1$};
                \node at (B1) [below right = 0.7mm of B1] {$p_2$};
                \node at (C1) [above = 0.7mm of C1] {$p_3$};
            \end{tikzpicture}
            \caption{Triangle graph}\label{fig:TriangleGraph}
    	\end{subfigure} 
    	\begin{subfigure}{.45\textwidth}
            \centering
    		\begin{tikzpicture}[thick, dot/.style = {draw, shape=circle, fill=black, scale=.5}, scale=1.5]
                \draw[thin,->] (0,0,0) -- (2.4,0,0) node[anchor=north west] {$\mu_1$};
                \draw[thin,->] (0,0,0) -- (0,2.4,0) node[anchor=south east] {$\mu_2$};
                \draw[thin,->] (0,0,0) -- (0,0,2.4) node[anchor=south east] {$\mu_3$};
                
                \coordinate[dot] (U1) at (1,0,0);
                \coordinate[dot] (U2) at (0,1,0);
                \coordinate[dot] (U3) at (0,0,1);
                
                \coordinate[dot] (F1) at (0,1,1);
                \coordinate[dot] (F2) at (1,0,1);
                \coordinate[dot] (F3) at (1,1,0);
                
                \coordinate[dot] (F4) at (2,0,0);
                \coordinate[dot] (F5) at (0,2,0);
                \coordinate[dot] (F6) at (0,0,2);
            
                \path[fill=gray!70, draw, opacity=0.5] (1,0,0) -- (0,1,0)--(0,0,1)--cycle; 
                \path[fill=gray!70, draw, opacity=0.2] (2,0,0) -- (0,2,0)--(0,0,2)--cycle; 
                \draw (U1) -- (F4); \draw (U2) -- (F5); \draw (U3) -- (F6);  
                \draw (U1) -- (U2) -- (U3) -- (U1);
                \draw (F4) -- (F5) -- (F6) -- (F4);
            \end{tikzpicture}
            \caption{Newton polytope $\Newt (\Uu+\Ff)$ corresponding to the triangle graph}\label{fig:NewtonPolytopeTriangleGraph}
    	\end{subfigure}    
    	\caption[Feynman graph and Newton polytope of the triangle graph]{Feynman graph and Newton polytope of the triangle graph as treated in \cref{ex:triangleGraphProperMixed}. The Newton polytope has one face of codimension zero, five faces of codimension one, nine faces of codimension two and six faces of codimension three. The two parallel triangles in this Newton polytope correspond to $\Uu$ and $\Ff$, respectively.}
    \end{figure}

    We will determine the discriminants of proper, mixed faces of the triangle graph (see \cref{fig:TriangleGraph}). The Symanzik polynomials for the triangle graph are given by
    \begin{align}
        \Uu &= x_1 + x_2 + x_3\\
        \Ff &= p_1^2 x_2 x_3 + p_2^2 x_1 x_3 + p_3^2 x_1 x_2 + \Uu \left(x_1 m_1^2 + x_2 m_2^2 + x_3 m_3^2\right) \point
    \end{align}
    The corresponding Newton polytope $\Newt(\Uu+\Ff)$ has in total $21$ faces, where one face is the Newton polytope $\Newt(\Uu+\Ff)$ itself, $7$ faces corresponding to faces of $\Newt(\Uu)$, another $7$ corresponding to faces of $\Newt(\Ff)$ and the remaining $6$ faces are proper, mixed faces (see \cref{fig:NewtonPolytopeTriangleGraph}). The truncated polynomials corresponding to the $6$ proper, mixed faces are up to permutations  $1\leftrightarrow 2 \leftrightarrow 3$ given by
    \begin{align}
        r_1 &= m_2^2 x_2^2 + m_3^2 x_3^2 + x_2 + x_3 + \left(p_1^2+m_2^2+m_3^2\right) x_2 x_3 \\
        h_1 &= m_1^2 x_1^2 + x_1 \point
    \end{align}
    As $h_1 = \pd{h_1}{x_1} = 0$ has no common solution we have $\Delta(h_1) = \Delta(h_2) = \Delta(h_3) = 1$ for the $A$-discriminant of $h_1$ as well as for its symmetric permutations. Considering $r_1$, this is nothing else than the sum of Symanzik polynomials of a $1$-loop bubble graph. Hence, the solutions of $r_1=\pd{r_1}{x_2}=\pd{r_1}{x_3} = 0$ describe the second-type singularity of the bubble graph, and we obtain $p_1^2=0$ and similarly for all the permutations. Thus, we get
    \begin{align}
        R = p_1^2 p_2^2 p_3^2
    \end{align}
    as the contribution to the singular locus of the triangle graph from proper, mixed faces. 
\end{example}

We want to mention, that also for the discriminants of proper, mixed faces $R$ there can be contributions which are identically zero in the physically relevant case. The simplest example where such a behaviour appears is the $2$-loop $2$-point function also known as the sunset graph. Again, we will define a physically relevant singular locus, where we have to remove these overall contributions.

Although, one does not usually consider the contribution of $R$ to the singular locus of Feynman integrals in the literature, we want to remark, that they can give a non-trivial contribution to the singular locus.


\section{Landau variety of the double-edged triangle graph} \label{sec:ExampleDuncesCap}


In order to demonstrate the methods described before, we want to calculate the Landau variety of a massive $2$-loop $3$-point function according to \cref{fig:duncescap}, which is known under various names, e.g.\ ``double-edged triangle graph'', ``dunce's cap'', ``parachute'' or ``ice cream cone''. To our knowledge this Landau variety was not published before. This graph is particularly interesting, since it is the simplest graph which does not belong to the cases discussed in \cref{lem:SubgraphPolynomialsVsTruncated} and \cref{lem:EAUfor1LoopAndBanana}. Using \HKP we will give the leading Landau variety in a parametrized form. Compared with the standard methods of eliminating variables from the Landau equations, the \HKP can be calculated very fast. We have to mention, that the reason for the effectiveness of the \HKP lies in a different representation of the result. To determine the defining polynomial of the Landau variety from the parametrized form is still a very time-consuming task. However, for many approaches the parametrized form of Landau varieties can be even more convenient, since the parameterization specifies the Landau singularities directly. \bigskip

\begin{figure}[bht]
    \centering
    \includegraphics[trim=0 0 0 1.5cm, clip]{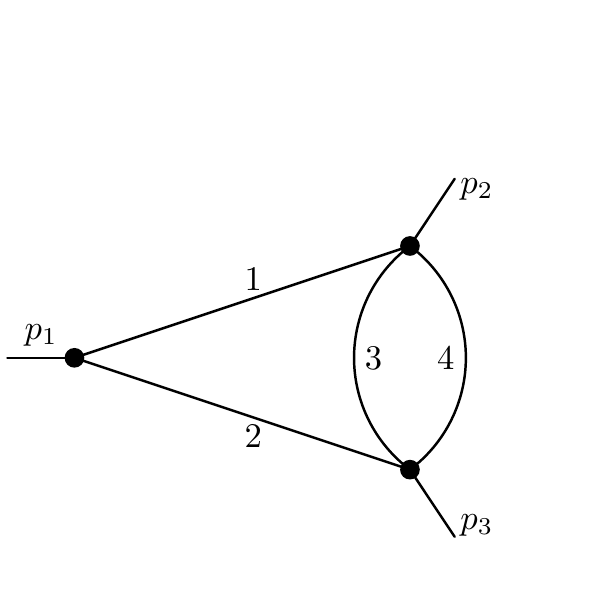}
    \caption{Double-edged triangle graph or ``dunce's cap'' graph} \label{fig:duncescap}
\end{figure}

For the Feynman graph of the dunce's cap graph in \cref{fig:duncescap} the Symanzik polynomials are given by
\begin{align}
    \Uu &= (x_1+x_2)(x_3+x_4) + x_3 x_4 \\
    \Ff &= s_1 x_1 x_2 (x_3 + x_4) + s_2 x_1 x_3 x_4 + s_3 x_2 x_3 x_4 \nonumber \\
    & + m_1^2 x_1^2 (x_3 + x_4) + m_2^2 x_2^2 (x_3 + x_4) + m_3^2 x_3^2 (x_1 + x_2 + x_4) + m_4^2 x_4^2 (x_1 + x_2 + x_3)
\end{align}
where we abbreviate $s_1 = p_1^2 + m_1^2 + m_2^2$, $s_2= p_2^2 + m_1^2 + m_3^2 + m_4^2$ and $s_3 = p_3^2 + m_2^2 + m_3^2 + m_4^2$. The Newton polytope $\Newt(\Uu+\Ff)$ has $85$ faces in total. Hence, we will expect $85$ contributions to the singular locus. Among these faces, $33$ belong to $\Newt(\Ff)$, $19$ to $\Newt(\Uu)$ and $32$ are proper, mixed faces. We will concentrate on the Landau variety and consider only the faces of $\Newt(\Ff)$. Recalling \cref{eq:LandauFactorizationPrincipalADet} we will have
\begin{align}
	\mathcal L_1 (\mathcal I_\Gamma) = \hspace{-1em} \bigcup_{\tau\subseteq\Newt(\Ff)} \hspace{-1em} \mathcal L_1^{(\tau)} (\mathcal I_\Gamma) \qquad\text{with}\quad \mathcal L_1^{(\tau)} (\mathcal I_\Gamma) = \Var \!\left( \Delta_{\Aa\cap\tau} (\Ff_\tau) \right) 
\end{align}
as the decomposition of the Landau variety $\mathcal L_1 (\mathcal I_\Gamma)$. We will treat the leading singularity $\Delta_\Aa(\Ff)$ by means of \HKp, whereas the $A$-discriminants of all the other truncated polynomials will be determinated by \softwareName{Macaulay2}. 

Starting with the leading singularity, we can read off the points generating the monomials of $\Ff=\sum_{a^{(j)}\in\Aa} z_j x^{a^{(j)}}$
\begin{align}
    \Aa &= \left(
            \begin{array}{cccccccccccccc}
            1 & 1 & 1 & 0 & 2 & 2 & 0 & 0 & 1 & 0 & 0 & 1 & 0 & 0 \\
            1 & 1 & 0 & 1 & 0 & 0 & 2 & 2 & 0 & 1 & 0 & 0 & 1 & 0 \\
            1 & 0 & 1 & 1 & 1 & 0 & 1 & 0 & 2 & 2 & 2 & 0 & 0 & 1 \\
            0 & 1 & 1 & 1 & 0 & 1 & 0 & 1 & 0 & 0 & 1 & 2 & 2 & 2 \\
            \end{array} \right) \\
    z &= (s_1,s_1,s_2,s_3,m_1^2,m_1^2,m_2^2,m_2^2,m_3^2,m_3^2,m_3^2,m_4^2,m_4^2,m_4^2) \label{eq:DuncesCapZvars} \point
\end{align}
The leading Landau variety is nothing else than the hypersurface $\{\Delta_{\Aa}(\Ff) = 0\}$, and we will determine this hypersurface by a convenient parameterization. For this \HKP we need a Gale dual, and we have several possibilities to choose a Gale dual of $\Aa$, e.g.\
\begin{align}
    \mathcal B = \resizebox{0.5\hsize}{!}{$ \left(
            \begin{array}{cccccccccc}
            1 & 1 & 1 & 0 & -1 & -1 & 0 & -1 & 0 & -1 \\
            0 & -1 & -1 & 1 & 1 & 1 & -1 & 0 & -1 & 0 \\
            -1 & 0 & -1 & -1 & 0 & -1 & 1 & 1 & -1 & -1 \\
            -1 & -1 & 0 & -1 & -1 & 0 & -1 & -1 & 1 & 1 \\
            0 & 0 & 0 & 0 & 0 & 0 & 0 & 0 & 0 & 1 \\
            0 & 0 & 0 & 0 & 0 & 0 & 0 & 0 & 1 & 0 \\
            0 & 0 & 0 & 0 & 0 & 0 & 0 & 1 & 0 & 0 \\
            0 & 0 & 0 & 0 & 0 & 0 & 1 & 0 & 0 & 0 \\
            0 & 0 & 0 & 0 & 0 & 1 & 0 & 0 & 0 & 0 \\
            0 & 0 & 0 & 0 & 1 & 0 & 0 & 0 & 0 & 0 \\
            0 & 0 & 0 & 1 & 0 & 0 & 0 & 0 & 0 & 0 \\
            0 & 0 & 1 & 0 & 0 & 0 & 0 & 0 & 0 & 0 \\
            0 & 1 & 0 & 0 & 0 & 0 & 0 & 0 & 0 & 0 \\
            1 & 0 & 0 & 0 & 0 & 0 & 0 & 0 & 0 & 0 \\
            \end{array}\right) $} \point
\end{align}
    
According to \cref{ssec:Adiscriminants} a Gale dual directly gives rise to a parameterization of generalized Feynman integrals. Let $z_1,\ldots,z_{14}$ be the coefficients of a generalized Feynman integral, i.e.\ coefficients of the second Symanzik polynomial $\Ff$. Then, the generic discriminant hypersurface $\{\Delta_{\Aa_\Ff}(\Ff)=0\}$ is parametrized by $t\in\mathbb P_{\mathbb C}^{10-1}$
\begin{align}
    y_1 &= \frac{z_1 z_{14}}{z_3 z_4} = \frac{R_1 t_1}{R_3 R_4} \quad,\qquad    
    y_2 = \frac{z_1 z_{13}}{z_2 z_4} = - \frac{R_1 t_2}{R_2 R_4} \quad,\qquad  
    y_3 = \frac{z_1 z_{12}}{z_2 z_3} = - \frac{R_1 t_3}{R_2 R_3} \nonumber\\
    y_4 &= \frac{z_2 z_{11}}{z_3 z_4} = \frac{R_2 t_4}{R_3 R_4} \quad,\qquad 
    y_5 = \frac{z_2 z_{10}}{z_1 z_4} = - \frac{R_2 t_5}{R_1 R_4} \quad,\qquad 
    y_6 = \frac{z_2 z_9}{z_1 z_3} = - \frac{R_2 t_6}{R_1 R_3} \nonumber\\
    y_7 &= \frac{z_3 z_8}{z_2 z_4} = \frac{R_3 t_7}{R_2 R_4} \quad,\qquad 
    y_8 = \frac{z_3 z_7}{z_1 z_4} = \frac{R_3 t_8}{R_1 R_4} \quad,\qquad 
    y_9 = \frac{z_4 z_6}{z_2 z_3} = \frac{R_4 t_9}{R_2 R_3} \nonumber\\
    y_{10} &= \frac{z_4 z_5}{z_1 z_3} = \frac{R_4 t_{10}}{R_1 R_3} \label{eq:duncescapHKPgeneric}
\end{align}
with
\begin{align}
    R_1 &= t_1+t_2+t_3-t_5-t_6-t_8-t_{10} \ \text{,}\qquad   
    R_2 = -t_2-t_3+t_4+t_5+t_6-t_7-t_9 \\
    R_3 &= t_1+t_3+t_4+t_6-t_7-t_8+t_9+t_{10} \ \text{,}\quad   
    R_4 = t_1+t_2+t_4+t_5+t_7+t_8-t_9-t_{10} \nonumber \ \text{.}
\end{align} 

However, up to now this \HKP gives the zero locus of discriminants for polynomials with generic coefficients. Thus, in order to adjust this parameterization to the physically relevant case we have to include the constraints given by the relations between coefficients of the Symanzik polynomials \cref{eq:DuncesCapZvars}. This can be accomplished e.g.\ by \softwareName{Mathematica} \cite{WolframResearchIncMathematicaVersion12} using the command ``Reduce'' or by \softwareName{Macaulay2} \cite{GraysonMacaulay2SoftwareSystem}. For algorithmical reasons it can be more efficient to reduce the effective variables $y_1,\ldots,y_{10}$, the linear forms $R_1,\ldots,R_4$ and the parameters $t_1,\ldots,t_{10}$ step by step. In doing so, we can eliminate $4$ parameters and the Landau variety splits into two components
\begin{align}
    \frac{m_2^2}{m_1^2} &= \frac{t_6^2 t_8}{t_5^2 t_{10}} \quad,\qquad 
    \frac{m_3^2}{m_1^2} = -\frac{t_6^2 t_9}{\left(t_5+t_6\right) \left(t_9+t_{10}\right)t_{10} } \quad,\qquad 
    \frac{m_4^2}{m_1^2} = -\frac{t_3 t_6 t_{10}}{\left(t_5+t_6\right) \left(t_9+t_{10}\right)t_9 } \nonumber\\
    \frac{s_1}{m_1^2} &= -\frac{t_5 \left(t_6 t_9+t_3 t_{10}\right)+t_6 \left(t_6 t_9+t_8 t_9+t_3 t_{10}+t_9t_{10}\right)}{t_5 t_9 t_{10}} \nonumber \\
    \frac{s_2}{m_1^2} &= -\frac{t_5 \left(t_9+t_{10}\right) \left(t_6 t_9+t_3 t_{10}\right) + t_6 \left[t_6 t_9^2+t_8 t_9 \left(t_9+t_{10}\right) -t_{10} \left(t_9^2+t_{10} t_9-t_3 t_{10}\right)\right]}{\left(t_5+t_6\right)  \left(t_9+t_{10}\right)t_9 t_{10}} \nonumber \\
    \frac{s_3}{m_1^2} &= -\frac{t_6 \left[t_6 \left(t_9+t_{10}\right) \left(t_6 t_9-t_8 t_9+t_3 t_{10} + t_9 t_{10}\right)+t_5 \left(t_6 t_9^2+t_3 t_{10}^2\right)\right]}{\left(t_5+t_6\right) \left(t_9+t_{10}\right) t_5 t_9 t_{10} } \label{eq:duncesCapResult1}
\end{align}
and 
\begin{align}
    \frac{m_2^2}{m_1^2} &=  \frac{t_6^2 t_8}{t_5^2 t_{10}} \quad,\qquad 
    \frac{m_3^2}{m_1^2} = \frac{t_6^2 t_9}{t_4 t_{10}^2} \quad,\qquad 
    \frac{m_4^2}{m_1^2} = \frac{t_6^2}{t_4 t_9} \quad,\qquad
    \frac{s_1}{m_1^2} = \frac{t_6 \left(\left(t_4-t_9\right) t_{10} -t_8 t_9\right)}{t_5 t_9 t_{10}} \nonumber \\
    \frac{s_2}{m_1^2} &= -\frac{t_6 \left[t_{10} \left(t_4 \left(t_9+t_{10}\right)+t_9 \left(2 t_6+t_9+t_{10}\right)\right)-t_8 t_9 \left(t_9+t_{10}\right)\right]}{t_4 t_9 t_{10}^2} \nonumber\\
    \frac{s_3}{m_1^2} &= -\frac{t_6^2 \left[t_8 t_9 \left(t_9+t_{10}\right)+t_{10} \left(t_4 \left(t_9+t_{10}\right)-t_9 \left(-2 t_5+t_9+t_{10}\right)\right)\right]}{t_4 t_5 t_9 t_{10}^2} \point \label{eq:duncesCapResult2}
\end{align}
Hence, the leading Landau variety of the dunce's cap graph will be given by the values of \cref{eq:duncesCapResult1} and \cref{eq:duncesCapResult2} for all values $[t_3 : t_5 : t_6 : t_8 : t_9 : t_{10}]\in\mathbb P^{6-1}_{\mathbb C}$ and $[t_4 : t_5 : t_6 : t_8 : t_9 : t_{10}]\in\mathbb P^{6-1}_{\mathbb C}$, respectively. We dispense with renaming of the parameters $t$ in order to ensure the reproducibility of the results from \cref{eq:duncescapHKPgeneric}. \bigskip


The remaining $32$ proper faces of $\Newt(\Ff)$ result in the contributions of the non-leading part of the Landau variety. We will determine them by means of a \softwareName{Macaulay2} \cite{GraysonMacaulay2SoftwareSystem,} routine, which can be found in the \cref{ssec:Macaulay2} and which is based on \cite{StaglianoPackageComputationsSparse2020, StaglianoPackageComputationsClassical2018}.

For the $7$ faces of $\Newt(\Ff)$ with codimension $1$ we obtain the following truncated polynomials $B_1,\ldots,B_7$ with their $A$-discriminants
\begin{align}
	B_1 &= s_3 x_2 x_3 x_4 + m_2^2 x_2^2 (x_3+x_4) + m_3^2 x_3^2 (x_2+x_4) + m_4^2 x_4^2 (x_2+x_3) = \Ff|_{x_1=0} \nonumber \\
	&\Rightarrow \Delta(B_1) = m_2^4 m_3^4 m_4^4 \left(s_3-m_2^2-m_3^2-m_4^2\right)^2 \nonumber \\
	&\qquad \left(s_3^{4} - 8 s_3^{2} m_2^2 m_3^2 + 16m_2^4m_3^4-8s_3^{2}m_2^2m_4^2-8s_3^{2}m_3^2m_4^2 + 64s_3m_2^2m_3^2m_4^2 \right. \nonumber \\
	&\qquad \left. - 32m_2^4m_3^2m_4^2-32m_2^2m_3^4m_4^2+16m_2^4m_4^4-32m_2^2m_3^2m_4^4+16m_3^4m_4^4\right)  \allowdisplaybreaks\\
	B_2 &= s_2 x_1 x_3 x_4 + m_1^2 x_1^2 (x_3+x_4) + m_3^2 x_3^2 (x_1+x_4) + m_4^2 x_4^2 (x_1+x_3) = \Ff|_{x_2=0} \nonumber \\
	&\Rightarrow \Delta(B_2) = m_1^4 m_3^4 m_4^4 \left( s_2-m_1^2-m_3^2-m_4^2 \right)^2 \nonumber \\
	&\qquad \left( s_2^{4}-8s_2^{2}m_1^2m_3^2+16m_1^4m_3^4-8s_2^{2}m_1^2m_4^2-8s_2^{2}m_3^2m_4^2 + 64s_2m_1^2m_3^2m_4^2 \right. \nonumber \\
	&\qquad \left. - 32m_1^4m_3^2m_4^2-32m_1^2m_3^4m_4^2+16m_1^4m_4^4-32m_1^2m_3^2m_4^4+16m_3^4m_4^4 \right)  \allowdisplaybreaks\\
	B_3 &= x_4 [ s_1 x_1 x_2 + m_1^2 x_1^2 + m_2^2 x_2^2 + m_4^2 x_4 (x_1+x_2)] = \Ff|_{x_3=0} \nonumber \\
	&\Rightarrow \Delta(B_3) = m_4^4(s_1-m_1^2-m_2^2)  \allowdisplaybreaks\\
	B_4 &= (x_3+x_4) ( s_1 x_1 x_2 + m_1^2 x_1^2 + m_2^2 x_2^2)  \Rightarrow \Delta(B_4) = s_1^2 - 4m_1^2m_2^2  \allowdisplaybreaks\\
	B_5 &= m_3^2 x_3^2 (x_1 + x_2 + x_4)  \Rightarrow \Delta(B_5) = 1  \allowdisplaybreaks\\
	B_6 &= x_3 ( s_1 x_1 x_2 + m_1^2 x_1^2 + m_2^2 x_2^2 + m_3^2 x_3 (x_1+x_2)) = \Ff|_{x_4=0} \nonumber \\
	&\Rightarrow \Delta(B_6) = m_3^4(s_1-m_1^2-m_2^2)  \allowdisplaybreaks \\
	B_7 &= m_4^2 x_4^2 (x_1 + x_2 + x_2)  \Rightarrow \Delta(B_7) = 1 \point
\end{align}
The truncated polynomials $B_1,B_2,B_3$ and $B_6$ would be also present in the approach of subgraphs, i.e.\ we can generate those polynomials also by setting certain variables $x_i$ to zero. However, for the other polynomials which can be not generated by subgraphs we will also obtain non-trivial contributions. Note, that the $A$-discriminant of $B_4$ will vanish when considering general solutions in $x\in\Csn$. The discriminant for $B_4$ given above is for the restriction of solutions with positive real part.

From the faces of $\Newt(\Ff)$ with codimension $2$ we will obtain $15$ truncated polynomials, which are listed below together with their $A$-discriminants
\begin{align}
	C_1 &= x_3 x_4 (m_3^2 x_3 + m_4^2 x_4) = \Ff|_{x_1=x_2=0} \Rightarrow \Delta(C_1) = 1 \allowdisplaybreaks\\
	C_2 &= x_2 x_4 (m_2^2 x_2 + m_4^2 x_4) = \Ff|_{x_1=x_3=0} \Rightarrow \Delta(C_2) = 1 \allowdisplaybreaks\\
	C_3 &= m_2^2 x_2^2 (x_3 + x_4) \Rightarrow \Delta(C_3) = 1 \allowdisplaybreaks\\
	C_4 &= m_3^2 x_3^2 (x_2 + x_4) \Rightarrow \Delta(C_4) = 1 \allowdisplaybreaks\\
	C_5 &= x_2 x_3 (m_2^2 x_2 + m_3^2 x_3) = \Ff|_{x_1=x_4=0} \Rightarrow \Delta(C_5) = 1 \allowdisplaybreaks\\
	C_6 &= m_4^2 x_4^2 (x_2 + x_3) \Rightarrow \Delta(C_6) = 1 \allowdisplaybreaks\\
	C_7 &= x_1 x_4 (m_1^2 x_1 + m_4^2 x_4) = \Ff|_{x_2=x_3=0} \Rightarrow \Delta(C_7) = 1 \allowdisplaybreaks\\
	C_8 &= m_1^2 x_1^2 (x_3 + x_4) \Rightarrow \Delta(C_8) = 1 \allowdisplaybreaks\\
	C_9 &= m_3^2 x_3^2 (x_1 + x_4) \Rightarrow \Delta(C_9) = 1 \allowdisplaybreaks\\
	C_{10} &= x_1 x_3 (m_1^2 x_1 + m_3^2 x_3) = \Ff|_{x_2=x_4=0} \Rightarrow \Delta(C_{10}) = 1 \allowdisplaybreaks\\
	C_{11} &= m_4^2 x_4^2 (x_1 + x_3) \Rightarrow \Delta(C_{11}) = 1 \allowdisplaybreaks\\
	C_{12} &= x_4 (s_1 x_1 x_2 + m_1^2 x_1^2 + m_2^2 x_2^2) \Rightarrow \Delta(C_{12}) = s_1^{2}-4 m_1^2 m_2^2 \allowdisplaybreaks\\
	C_{13} &= m_4^2 x_4^2 (x_1 + x_2) \Rightarrow \Delta(C_{13}) = 1 \allowdisplaybreaks\\
	C_{14} &= x_3 (s_1 x_1 x_2 + m_1^2 x_1^2 + m_2^2 x_2^2) \Rightarrow \Delta(C_{14}) = s_1^{2}-4 m_1^2 m_2^2 \allowdisplaybreaks\\ 
	C_{15} &= m_3^2 x_3^2 (x_1 + x_2) \Rightarrow \Delta(C_{15}) = 1 \point
\end{align}
The polynomials $C_1, C_2, C_5, C_7$ and $C_{10}$ also appear as Symanzik polynomials of subgraphs. Remarkably, the truncated polynomials $C_{12}$ and $C_{14}$, which do not correspond to subgraph polynomials, will result in non-trivial factors. These truncated polynomials have the shape of graph polynomials of a $1$-loop bubble graph. Their contribution to the singular locus is also not contained in any other discriminant coming from subgraph polynomials. Hence, those contributions were overlooked by the classical approach. Moreover, for $\Re(s_1)<0$ those additional singularities are on the principal sheet, i.e.\ they allow a solution in $x\in\mathbb R^n_{>0}$.

The remaining faces of codimension $3$ are the vertices of $\Newt(\Ff)$. They result in
\begin{align}
	D_1 &= m_3^2 x_3^2 x_4 \Rightarrow \Delta(D_1) = m_3^2 \qquad\quad  D_6 = m_3^2 x_3^2 x_2 \Rightarrow \Delta(D_6) = m_3^2 \allowdisplaybreaks\\  
	D_2 &= m_4^2 x_4^2 x_3 \Rightarrow \Delta(D_2) = m_4^2 \qquad\quad  D_7 = m_1^2 x_1^2 x_4 \Rightarrow \Delta(D_7) = m_1^2 \allowdisplaybreaks\\
	D_3 &= m_2^2 x_2^2 x_4 \Rightarrow \Delta(D_3) = m_2^2 \qquad\quad  D_8 = m_4^2 x_4^2 x_1 \Rightarrow \Delta(D_8) = m_4^2 \allowdisplaybreaks\\
    D_4 &= m_4^2 x_4^2 x_2 \Rightarrow \Delta(D_4) = m_4^2 \qquad\quad  D_9 = m_1^2 x_1^2 x_3 \Rightarrow \Delta(D_9) = m_1^2 \allowdisplaybreaks\\
	D_5 &= m_2^2 x_2^2 x_3 \Rightarrow \Delta(D_5) = m_2^2 \qquad\quad  D_{10} = m_3^2 x_3^2 x_1 \Rightarrow \Delta(D_{10}) = m_3^2 \point
\end{align}

Therefore, expressed by graphs, we will obtain the following contributions to the Landau variety of the dunce's cap graph

\begin{align}
    \widehat E_A (\Ff) &= \Delta \!\left( %
        \begin{tikzpicture}[baseline={([yshift=-.5ex]current bounding box.center)}, thick, dot/.style = {draw, shape=circle, fill=black, scale=.3}, every node/.style={scale=0.8}, scale=0.7]
            \coordinate[dot] (A) at (0,0);
            \coordinate[dot] (B) at (2,-1);
            \coordinate[dot] (C) at (2,1);
            \draw (A) -- node[above] {$1$} (C);  
            \draw (A) -- node[below] {$2$} (B);  
            \draw (C) arc[start angle=135, end angle=225, radius=1.4142];
            \draw (B) arc[start angle=-45, end angle=45, radius=1.4142];
            \node at (1.3,0) {$3$};
            \node at (2.7,0) {$4$};
            \draw (A) -- node[below] {$p_1$} ++(-1,0);
            \draw (B) -- node[below right,shift={(0,.1)}] {$p_3$} ++(0.2,-.7); 
            \draw (C) -- node[above right,shift={(0,.1)}] {$p_2$} ++(0.2,.7); 
        \end{tikzpicture} %
    \right) \cdot \Delta \!\left( %
        \begin{tikzpicture}[baseline={([yshift=-.5ex]current bounding box.center)}, thick, dot/.style = {draw, shape=circle, fill=black, scale=.3}, every node/.style={scale=0.8}, scale=0.7]
            \coordinate[dot] (A) at (0,0);
            \coordinate[dot] (B) at (2,0);
            \coordinate (A1) at (1,1);
            \coordinate (A2) at (1,-1);
            \draw (A) -- node[above] {$3$} (B);  
            \draw (1,0) circle (1); \node at (A1) [above = 0.5mm of A1] {$2$}; \node at (A2) [below = 0.5mm of A2] {$4$};
            \draw (A) -- node[below] {$p_3$} ++(-1,0);
            \draw (B) -- node[above,shift={(0,.1)}] {$p_1$} ++(1,1); 
            \draw (B) -- node[below,shift={(0,-.1)}] {$p_2$} ++(1,-1);
        \end{tikzpicture} %
    \right) \cdot \Delta \!\left( %
        \begin{tikzpicture}[baseline={([yshift=-.5ex]current bounding box.center)}, thick, dot/.style = {draw, shape=circle, fill=black, scale=.3}, every node/.style={scale=0.8}, scale=0.7]
            \coordinate[dot] (A) at (0,0);
            \coordinate[dot] (B) at (2,0);
            \coordinate (A1) at (1,1);
            \coordinate (A2) at (1,-1);
            \draw (A) -- node[above] {$3$} (B);  
            \draw (1,0) circle (1); \node at (A1) [above = 0.5mm of A1] {$1$}; \node at (A2) [below = 0.5mm of A2] {$4$};
            \draw (A) -- node[below] {$p_2$} ++(-1,0);
            \draw (B) -- node[above,shift={(0,.1)}] {$p_1$} ++(1,1); 
            \draw (B) -- node[below,shift={(0,-.1)}] {$p_3$} ++(1,-1);
        \end{tikzpicture} %
    \right) \cdot \nonumber \\
    & \Delta \!\left( %
        \begin{tikzpicture}[baseline={([yshift=-.5ex]current bounding box.center)}, thick, dot/.style = {draw, shape=circle, fill=black, scale=.3}, every node/.style={scale=0.8}, scale=0.7]
            \coordinate[dot] (A) at (0,0);
            \coordinate[dot] (B) at (2,0);
            \coordinate (A1) at (1,1);
            \coordinate (A2) at (1,-1);
            \coordinate (A3) at (3.2,0);
            \draw (1,0) circle (1); \node at (A1) [above = 0.5mm of A1] {$1$}; \node at (A2) [below = 0.5mm of A2] {$2$};
            \draw (A) -- node[below] {$p_1$} ++(-1,0);
            \draw (B) -- node[above,shift={(0,.6)}] {$p_2$} ++(0,1.3); 
            \draw (B) -- node[below,shift={(0,-.6)}] {$p_3$} ++(0,-1.3);
            \draw (2.6,0) circle (0.6); \node at (A3) [left = 0.5mm of A3] {$4$};
        \end{tikzpicture} %
    \right) \cdot \Delta \!\left( %
        \begin{tikzpicture}[baseline={([yshift=-.5ex]current bounding box.center)}, thick, dot/.style = {draw, shape=circle, fill=black, scale=.3}, every node/.style={scale=0.8}, scale=0.7]
            \coordinate[dot] (A) at (0,0);
            \coordinate[dot] (B) at (2,0);
            \coordinate (A1) at (1,1);
            \coordinate (A2) at (1,-1);
            \coordinate (A3) at (3.2,0);
            \draw (1,0) circle (1); \node at (A1) [above = 0.5mm of A1] {$1$}; \node at (A2) [below = 0.5mm of A2] {$2$};
            \draw (A) -- node[below] {$p_1$} ++(-1,0);
            \draw (B) -- node[above,shift={(0,.6)}] {$p_2$} ++(0,1.3); 
            \draw (B) -- node[below,shift={(0,-.6)}] {$p_3$} ++(0,-1.3);
            \draw (2.6,0) circle (0.6); \node at (A3) [left = 0.5mm of A3] {$3$};
        \end{tikzpicture} %
    \right) \cdot \Delta \!\left( %
        \begin{tikzpicture}[baseline={([yshift=-.5ex]current bounding box.center)}, thick, dot/.style = {draw, shape=circle, fill=black, scale=.3}, every node/.style={scale=0.8}, scale=0.7]
            \coordinate[dot] (A) at (0,0);
            \coordinate[dot] (B) at (2,0);
            \coordinate (A1) at (1,1);
            \coordinate (A2) at (1,-1);
            \coordinate (A3) at (4,0);
            \draw (1,0) circle (1); \node at (A1) [above = 0.5mm of A1] {$1$}; \node at (A2) [below = 0.5mm of A2] {$2$};
            \draw (A) -- node[below] {$p_1$} ++(-1,0);
            \draw (B) -- node[above,shift={(0,.1)}] {$p_2$} ++(1,1); 
            \draw (B) -- node[below,shift={(0,-.1)}] {$p_3$} ++(1,-1);
        \end{tikzpicture} %
    \right) \cdot\nonumber \\
    &  \Delta \!\left( %
        \begin{tikzpicture}[baseline={([yshift=-.5ex]current bounding box.center)}, thick, dot/.style = {draw, shape=circle, fill=black, scale=.3}, every node/.style={scale=0.8}, scale=0.7]
            \coordinate[dot] (A) at (0,0);
            \coordinate (A1) at (1.6,0);
            \draw (0.8,0) circle (.8); \node at (A1) [left = 0.5mm of A1] {$1$}; 
            \draw (A) -- node[left,shift={(-.3,0)}] {$p_2$} ++(-.7,0);
            \draw (A) -- node[left,shift={(-.3,.3)}] {$p_1$} ++(-.7,.7);
            \draw (A) -- node[left,shift={(-.3,-.3)}] {$p_3$} ++(-.7,-.7);
        \end{tikzpicture} %
    \right) \cdot \Delta \!\left( %
        \begin{tikzpicture}[baseline={([yshift=-.5ex]current bounding box.center)}, thick, dot/.style = {draw, shape=circle, fill=black, scale=.3}, every node/.style={scale=0.8}, scale=0.7]
            \coordinate[dot] (A) at (0,0);
            \coordinate (A1) at (1.6,0);
            \draw (0.8,0) circle (.8); \node at (A1) [left = 0.5mm of A1] {$2$}; 
            \draw (A) -- node[left,shift={(-.3,0)}] {$p_2$} ++(-.7,0);
            \draw (A) -- node[left,shift={(-.3,.3)}] {$p_1$} ++(-.7,.7);
            \draw (A) -- node[left,shift={(-.3,-.3)}] {$p_3$} ++(-.7,-.7);
        \end{tikzpicture} %
    \right) \cdot \Delta \!\left( %
        \begin{tikzpicture}[baseline={([yshift=-.5ex]current bounding box.center)}, thick, dot/.style = {draw, shape=circle, fill=black, scale=.3}, every node/.style={scale=0.8}, scale=0.7]
            \coordinate[dot] (A) at (0,0);
            \coordinate (A1) at (1.6,0);
            \draw (0.8,0) circle (.8); \node at (A1) [left = 0.5mm of A1] {$3$}; 
            \draw (A) -- node[left,shift={(-.3,0)}] {$p_2$} ++(-.7,0);
            \draw (A) -- node[left,shift={(-.3,.3)}] {$p_1$} ++(-.7,.7);
            \draw (A) -- node[left,shift={(-.3,-.3)}] {$p_3$} ++(-.7,-.7);
        \end{tikzpicture} %
    \right) \cdot \Delta \!\left( %
        \begin{tikzpicture}[baseline={([yshift=-.5ex]current bounding box.center)}, thick, dot/.style = {draw, shape=circle, fill=black, scale=.3}, every node/.style={scale=0.8}, scale=0.7]
            \coordinate[dot] (A) at (0,0);
            \coordinate (A1) at (1.6,0);
            \draw (0.8,0) circle (.8); \node at (A1) [left = 0.5mm of A1] {$4$}; 
            \draw (A) -- node[left,shift={(-.3,0)}] {$p_2$} ++(-.7,0);
            \draw (A) -- node[left,shift={(-.3,.3)}] {$p_1$} ++(-.7,.7);
            \draw (A) -- node[left,shift={(-.3,-.3)}] {$p_3$} ++(-.7,-.7);
        \end{tikzpicture} %
    \right) \label{eq:DuncesCapDecompositionGraphical}
\end{align}
where we suppressed multiplicities of those $A$-discriminants, and we dropped all the $A$-discriminants which are trivial. The first $5$ factors would be also expected by considering subgraphs. In contrast to the approach of subgraphs, we will additionally obtain contributions in terms of bubble graphs and tadpole graphs. Those parts will be missed within the classical approach. That the singularities corresponding to the truncated polynomials $C_{12}$ and $C_{14}$ result indeed to a threshold of the Feynman integral was found in \cite{AnastasiouTwoloopAmplitudesMaster2007}. \bigskip

We want to mention that it is by no means clear or trivial, that the overlooked contributions have the shape of discriminants of graph polynomials. In this particular example we will find that they behave like discriminants of graph polynomials. Furthermore, we can summarize the involved diagrams in \cref{eq:DuncesCapDecompositionGraphical} as follows: We take all subdiagrams which arise from the dunce's cap graph by shrinking edges and deleting tadpoles. Whether such a behaviour also applies in general can only be speculated at this point. \bigskip

%
%

%
%


\section{Coamoebas and Feynman's \texorpdfstring{$i\varepsilon$}{i\unichar{"03B5}} prescription} \label{sec:Coamoebas}

%

In the previous sections we only were asking if there are values of $p$ and $m$ such that there exists a multiple root of the Symanzik polynomials anywhere in the complex plane. However, when applied to Feynman integral representations, e.g.\ the Lee-Pomeransky representation \cref{eq:LeePomeranskyRepresentation} (or more generally to Euler-Mellin integrals) we only should worry if these multiple roots lie inside the integration region $\mathbb R^n_+$. Thus, in this section we want to study the position of singularities in the space of Schwinger parameters $x$. The other main aspect of this section will concern the multivalued structure of Feynman integrals, which is closely related to the previous question. 

In order to better investigate the nature of the multiple roots, we will introduce the concept of coamoebas. The coamoeba was invented by Mikael Passare as a related object to the amoeba, which goes back to Gelfand, Kapranov and Zelevinsky \cite{GelfandDiscriminantsResultantsMultidimensional1994}. Amoebas as well as coamoebas provide a link between algebraic geometry and tropical geometry \cite{NisseAmoebasCoamoebasLinear2016}. We suggest \cite{KazarnovskiiNewtonPolytopesTropical2021} for a recent survey about tropical geometry and its connection to toric varieties. The relations between coamoebas and Euler-Mellin integrals were investigated in \cite{NilssonMellinTransformsMultivariate2010, BerkeschEulerMellinIntegrals2013}, which state the key references for this section. For further reading about coamoebas we refer to \cite{NisseGeometricCombinatorialStructure2009, JohanssonCoamoebas2010, ForsgardHypersurfaceCoamoebasIntegral2012, JohanssonArgumentCycleCoamoeba2013, PassareDiscriminantCoamoebasHomology2012, NisseHigherConvexityCoamoeba2015, ForsgardOrderMapHypersurface2015, ForsgardTropicalAspectsReal2015}.\bigskip

One of the main reasons that only very little is known about Landau varieties is the very sophisticated nature of multivalued sheet structure of Feynman integrals. Hence, the purpose of this section is to suggest the study of a slightly simpler object. For many aspects as e.g.\ the monodromy of Feynman integrals, it turns out that the coamoeba provides deep insights in the sheet structure of Feynman integrals. Even though coamoebas are simpler objects than the sheet structure itself, they are still challenging and difficult to compute. However, coamoebas seem somehow easier accessible, as we will sketch in the end of this section.

This section does not intend to give a complete answer to those questions and should rather be seen as the first step of a bigger task. We will focus here on a very basic overview about the mathematical fundament of those objects and clarify some first questions.\bigskip

But before introducing coamoebas, we will concern with a very classical idea about handling the multivaluedness of Feynman integrals. As aforementioned, one often introduces a small imaginary part in the denominators of the Feynman integral in momentum space \cref{eq:FeynmanMomSp}. This so-called \textit{Feynman's $i\varepsilon$ prescription} is a kind of a minimal version to complexify the momenta and masses in the Feynman integral. When considering squared momenta and masses to be real numbers, we can elude poles in the integrand of \cref{eq:FeynmanMomSp} by introducing a small imaginary part. On the other hand, one can understand this $i\varepsilon$ also as an implementation of the time ordering \cite[sec. 6.2]{SchwartzQuantumFieldTheory2014}, whence the $i\varepsilon$ prescription is often related to causality \cite{HannesdottirWhatVarepsilonSmatrix2022}. Thus, we consider
\begin{align}
    \gls{FeynEps} = \int_{\mathbb R^{d\times L}} \prod_{j=1}^L \frac{\ddif k_j}{\pi^{d/2}} \prod_{i=1}^n \frac{1}{(q_i^2+m_i^2-i\varepsilon)^{\nu_i}} \comma \label{eq:FeynmanMomSpEpsilon}
\end{align}
where $m_i^2\in\mathbb R_{\geq 0}$ and $q_i^2 \in \mathbb R$, and replace \cref{eq:FeynmanMomSp} by $\lim_{\varepsilon\rightarrow 0^+} \mathcal I_\Gamma^\varepsilon$. Thus, the advantage of this procedure is, that the Feynman integral defined now by $\lim_{\varepsilon\rightarrow 0^+} \mathcal I_\Gamma^\varepsilon$ becomes a single-valued function. Hence, those definitions are nothing else than a choice of a principal sheet by determining the direction from which we want to enter the branch cut. Therefore, we will denote the discontinuity at those branch cuts by
\begin{align}
	\gls{Disc} = \lim_{\varepsilon\rightarrow 0^+} \left(\mathcal I_\Gamma^\varepsilon (d,\nu,p,m) - \mathcal I_\Gamma^{-\varepsilon} (d,\nu,p,m) \right) \label{eq:discontinuity1}
\end{align}
according to \cite{CutkoskySingularitiesDiscontinuitiesFeynman1960}. Note, that on the level of the transfer matrix $T$ the discontinuity is related to the imaginary part as seen in \cref{eq:OpticalTheorem}. By means of Schwarz' reflection principle, this behaviour can also be extended to individual Feynman integrals in many cases. We refer to \cite[sec. 4.4]{HannesdottirWhatVarepsilonSmatrix2022} for a discussion of this procedure. 

In equation \cref{eq:FeynmanMomSpEpsilon} we can absorb the introduction of $i\varepsilon$ by a transformation of $m_i^2 \mapsto m_i^2 - i\varepsilon$. Hence, we can simply determine the $i\varepsilon$-prescripted parametric Feynman integral representations, and we will obtain
\begin{align}
	\gls{FeynGenEps} &= \frac{\Gamma(\nu_0)}{\Gamma(\nu_0-\omega)\Gamma(\nu)} \int_{\mathbb R^n_+} \dif x \, x^{\nu-1} \left(\Uu+\Ff - i \varepsilon \Uu \sum_{i=1}^n x_i\right)^{-\nu_0} \nonumber \\
	&\stackrel{|\epsilon|\ll 1}{=} \frac{\Gamma(\nu_0)}{\Gamma(\nu_0-\omega)\Gamma(\nu)} \int_{\mathbb R^n_+} \dif x \, x^{\nu-1} \left(\Uu+\Ff - i \varepsilon\right)^{-\nu_0} \label{eq:LeePomeranskyRepresentationEpsilon}
\end{align}
as the equivalent for the generalized Feynman integral in representation \cref{eq:LeePomeranskyRepresentation}. When $|\varepsilon|\ll 1$ is small (as we usually want to assume), we can drop the factor $\Uu\sum_{i=1}^n x_i$ which is positive inside the integration region $(0,\infty)^n$. Hence, the $i\varepsilon$ prescription describes a possibility to circumvent poles in the integrand. However, this happens at the price of a limit that must subsequently be carried out.\bigskip

We want to suggest an alternative to Feynman's $i\varepsilon$ prescription, which occurs in the study of Euler-Mellin integrals \cite{NilssonMellinTransformsMultivariate2010,BerkeschEulerMellinIntegrals2013}. This variant coincides with the $i\varepsilon$ prescription in the limit $\varepsilon\rightarrow 0^+$. However, it avoids the need to consider a limit value. Since we want to make the argument slightly more general, let us introduce Euler-Mellin integrals for that purpose
\begin{align}
    \gls{EM} := \Gamma(t) \int_{\mathbb R^n_+} \dif x \, x^{s-1} f(x)^{-t} = \Gamma(t) \int_{\mathbb R^n} \dif w \, e^{s\cdot w} f\left(e^w\right)^{-t}  \label{eq:EulerMellin}
\end{align}
where $s\in \mathbb C^n$ and $t\in \mathbb C^k$ are complex parameters. Further, we denote $f^{-t}:= f_1^{-t_1}\cdots f_k^{-t_k}$ for powers of polynomials $f_i\in\mathbb C[x_1,\ldots,x_n]$, and we write $f:=f_1 \cdots f_k$. Parametric Feynman integrals can be treated as a special case of Euler-Mellin integrals with $k=1$ in the Lee-Pomeransky representation \cref{eq:LeePomeranskyRepresentation} or with $k=2$ in the Feynman representation \cref{eq:FeynmanParSpFeynman} for a convenient hyperplane $H(x)$.\bigskip

Following \cite{BerkeschEulerMellinIntegrals2013}, a polynomial $f\in\mathbb C[x_1,\ldots,x_n]$ with fixed coefficients is called \textit{\CNV on} $X\subseteq \mathbb C^n$ if for all faces $\tau\subseteq\Newt(f)$ of the Newton polytope, the truncated polynomials $f_\tau$ do not vanish on $X$. We can consider the vanishing of some truncated polynomial as a necessary condition in the approach of the principal $A$-determinant. Thus, if a polynomial $f$ is \CNV on a set $X$, roots of the principal $A$-determinant $E_A(f)$ will correspond to a solution $x\notin X$ outside of this set $X$. \bigskip

\begin{figure}[tb]
    \centering
    \begin{tikzpicture}[thick, dot/.style = {draw, cross out, scale=.5}, scale=1]
        \draw[thin,->] (-2,0) -- (4.3,0) node[anchor=north west] {$\Re (x)$};
        \draw[thin,->] (0,-1) -- (0,3) node[anchor=south east] {$\Im (x)$};
        \draw (0,0) -- (4.3,0);
        \draw (0,0) -- (4,1.5);
        \draw (1,0) arc[start angle=0, end angle=20.56, radius=1];
        \node at (0.84,0.14) {\footnotesize $\theta$};
        \node [anchor=west] (note1) at (5,0.5) {$\Arg^{-1}(0)=\mathbb R^n_+$};
        \draw [-stealth] (note1) to[out=180, in=70] (4.1,0.1);
        \node [anchor=west] (note2) at (4.5,2) {$\Arg^{-1}(\theta)$};
        \draw [-stealth] (note2) to[out=180, in=70] (3.8,1.525);
        \coordinate[dot] (a1) at (1.5,2);
        \coordinate[dot] (a2) at (0,2.2);
        \coordinate[dot] (a3) at (-1.4,0);
        \coordinate[dot] (a4) at (0.8,-0.8);
        \coordinate[dot] (a5) at (2,0);
        \node [anchor=north] (note3) at (3,4) {roots of $f(x)$};
        \draw [-stealth] (note3) to[out=270, in=60] (a1);
    \end{tikzpicture}
    \caption[Sketch of the idea behind the $\theta$-analogue Euler-Mellin integrals]{Sketch of the idea behind the $\theta$-analogue Euler-Mellin integrals. Thus, we allow the integration contour to rotate in the complex space $\mathbb C^n$. The value of $\mathscr M_f^\theta(s,t)$ changes only if we exceed zeros of $f(x)$ with the integration contour $\Arg^{-1}(\theta)$. Hence, $\mathscr M_f^\theta(s,t)$ is locally constant in $\theta$. Especially, when roots hit the original integration region $\mathbb R^n_+$ one has to consider a  $\theta$-analogue Euler-Mellin integral instead. For this picture, we would expect $5$ connected components in the complement of the closure of the coamoeba.} \label{fig:SketchThetaEM}	
\end{figure}
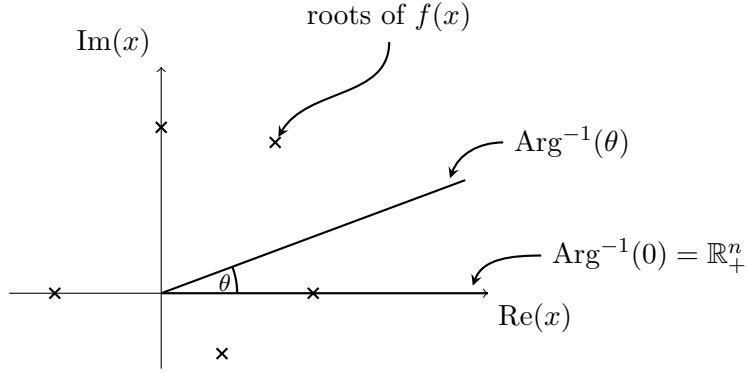

Analogue to the Feynman integral, Euler-Mellin integrals \cref{eq:EulerMellin} become ill-defined, when $f$ is not \CNV on the positive orthant $(0,\infty)^n$. Hence, in \cite{BerkeschEulerMellinIntegrals2013, NilssonMellinTransformsMultivariate2010} a slightly more general version of Euler-Mellin integrals was introduced, where we rotate the original integration contour by an angle $\theta=(\theta_1,\ldots,\theta_n)$ in the complex plane
\begin{align}
    \gls{EMtheta} := \Gamma(t) \int_{\Arg^{-1} (\theta)} \dif x\, x^{s-1} f(x)^{-t} = e^{is\theta}\, \Gamma(t)  \int_{\mathbb R_+^n} \dif x\, x^{s-1} f\!\left(x e^{i\theta}\right)^{-t} \label{eq:EulerMellinTheta}
\end{align}
with the component-wise argument map $\gls{Arg} := (\arg x_1,\ldots,\arg x_n)$. For short we write $f\!\left(x e^{i\theta}\right) := f\!\left(x_1 e^{i\theta_1},\ldots,x_n e^{i\theta_n}\right)$, and we will call \cref{eq:EulerMellinTheta} the \textit{$\theta$-analogue Euler-Mellin integral}. Deforming the integration contour slightly in cases where poles of the integrand hit the integration contour, is the same procedure as Feynman's $i\varepsilon$ prescription. To illustrate the idea behind $\theta$-analogue Euler-Mellin integrals we sketched those integration regions in \cref{fig:SketchThetaEM} in a one-dimensional case. By choosing $\theta=(\varepsilon,\ldots,\varepsilon)$ the $\theta$-analogue Euler-Mellin integral and the $i\varepsilon$-prescripted Feynman integral coincide for small $\varepsilon$
\begin{align}
    \mathcal I_\Aa^\varepsilon (\nuu,z) &= \frac{\Gamma(\nu_0)}{\Gamma(\nu)\Gamma\left(\nu_0-\omega\right)}\int_{\Arg^{-1}(\varepsilon,\ldots,\varepsilon)} \dif x\, x^{\nu-1} (\Uu+\Ff)^{-\nu_0} \nonumber\\
    &=\frac{\Gamma(\nu_0)}{\Gamma(\nu)\Gamma\left(\nu_0-\omega\right)}\int_{\mathbb R^n_+} \dif x\, x^{\nu-1} e^{i\varepsilon(\omega-\nu_0)}(\Uu e^{-i\varepsilon} +\Ff)^{-\nu_0} \nonumber\\
    &\stackrel{|\varepsilon|\ll 1}{=}  \frac{\Gamma(\nu_0)}{\Gamma(\nu)\Gamma\left(\nu_0-\omega\right)}\int_{\mathbb R^n_+} \dif x\, x^{\nu-1} (\Uu + \Ff - i\varepsilon)^{-\nu_0} \label{eq:iepsilonThetaEM}
\end{align}
where we used the homogenity of Symanzik polynomials. As we are only interested in the limit $\varepsilon\rightarrow 0^+$, we will assume that $\varepsilon$ is small enough\footnote{Obviously, \cref{eq:iepsilonThetaEM} holds only if the $\theta$-analogue Euler-Mellin integral $\mathscr M_\Gg^\varepsilon (\nuu,z)$ exists for small values of $\varepsilon$. Furthermore, we will expect in general that the different parametric Feynman representations have distinct $\varepsilon$-domains in their $\theta$-analogue extension.}.\bigskip

This alternative description by means of a rotated integration contour has several advantages. Note that the $\theta$-analogue Euler-Mellin integral is locally constant in $\theta$. Hence, the value of $\mathcal I^\varepsilon_\Aa$ and $\mathcal I^{\varepsilon^\prime}_\Aa$ can only differ if there are roots of $\Gg(x)$ in the segment of $\mathbb C^n$ spanned by $\Arg^{-1}(\varepsilon,\ldots,\varepsilon)$ and $\Arg^{-1}(\varepsilon^\prime,\ldots,\varepsilon^\prime)$. This follows directly from residue theorem by closing the integration region at infinity, since the integrand of Euler-Mellin integrals vanishes sufficiently fast when $|x|$ tends to infinity. Therefore, we will find regions in which a variation of $\varepsilon$ leaves the Feynman integral $\mathcal I_\Aa^\varepsilon$ constant. Especially, we do not have to take limits to determine the discontinuity, i.e.\ we will now have
\begin{align}
	\Disc \mathcal I_\Aa (\nuu,z) = \mathcal I_\Aa^\varepsilon (\nuu,z) - \mathcal I_\Aa^{-\varepsilon} (\nuu,z)
\end{align}
instead of \cref{eq:discontinuity1}, where $\varepsilon$ has to be sufficiently small (but not infinitesimally small). This can reduce the complexity of determining discontinuities significantly. \bigskip

As seen above, in order to get well-defined integrals in \cref{eq:EulerMellinTheta} we have to track the poles of the integrand. Denote by $\gls{Zf} := \left\{ x\in(\mathbb C^*)^n \,\rvert\, f(x) = 0 \right\}$ the set of (non-zero) roots of a polynomial $f$. To analyze when these poles meet the integration contour it is natural to consider
\begin{align}
    \gls{Coamoeba} := \operatorname{Arg}(\mathcal Z_f) \subseteq \mathbb T^n \label{eq:defCoamoeba}
\end{align}
the argument of the zero locus. We will call the set $\mathcal C_f$ the \textit{coamoeba} of $f$, which is motivated by the fact that the coamoeba presents the imaginary counterpart of the amoeba (see \cref{fig:ComplexLogarithmMaps}). Since the argument map is a periodic function we can restrict the discussion without loss of generality to the $n$-dimensional real torus $\gls{RTorus} := \left(\normalslant{\mathbb R}{2\pi\mathbb Z}\right)^n$. Moreover, the closure of the coamoeba is closely related to the \CNV of $f$, which we will see in the following lemma.

\begin{lemma}[\cite{BerkeschEulerMellinIntegrals2013,JohanssonCoamoebas2010}] \label{lem:CoamoebasAndCNV}
    For $\theta\in\mathbb T^n$, the polynomial $f(x)$ is \CNV on $\Arg^{-1} (\theta)$ if and only if $\theta \notin \overline{\mathcal C_f}$. Equivalently, we have 
    \begin{align}
    	\overline{\mathcal C_f} = \bigcup_{\tau\subseteq\Newt(f)} \mathcal C_{f_\tau} \point
    \end{align}
\end{lemma}

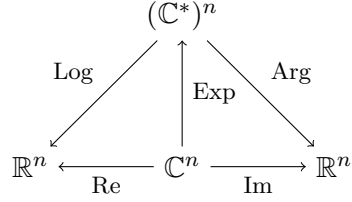
\begin{figure}[t]
    \centering
	\begin{tikzpicture}
        \node (C) at (0,0) {$\mathbb C^n$};
        \node (R1) at (2,0) {$\mathbb R^n$};
        \node (R2) at (-2,0) {$\mathbb R^n$};
        \node (CS) at (0,2) {$\Csn$};
        \draw[->] (C) to node[below] {\footnotesize $\Im$} (R1);
        \draw[->] (C) to node[below] {\footnotesize $\Re$} (R2);
        \draw[->] (C) to node[right] {\footnotesize $\operatorname{Exp}$} (CS);
        \draw[->] (CS) to node[above right] {\footnotesize $\Arg$} (R1);
        \draw[->] (CS) to node[above left] {\footnotesize $\operatorname{Log}$} (R2);
	\end{tikzpicture}
	\caption[Relations between the multivariate complex logarithm]{Relations between the multivariate complex logarithm \cite{NilssonMellinTransformsMultivariate2010}. We write $\operatorname{Log} (x) := (\ln |x_1|,\ldots,\ln |x_n|)$ for the logarithm of absolute values. For the set of roots $\mathcal Z_f$ we will call $\operatorname{Log}(\mathcal Z_f)$ the amoeba of $f$ and $\Arg(\mathcal Z_f)$ the coamoeba of $f$.} \label{fig:ComplexLogarithmMaps}
\end{figure}

Applied to Feynman integrals, this lemma gives a criterion, whether the limit of the $i\varepsilon$ prescripted Feynman integral \cref{eq:LeePomeranskyRepresentationEpsilon} $\lim_{\varepsilon\rightarrow 0^+} \mathcal I^\varepsilon_\Aa(\nuu,z)$ differs from the original Feynman integral $\mathcal I_\Aa(\nuu,z)$.

\begin{lemma} \label{lem:nonphysical}
    Assume that $z\in\Var(E_A(\Gg))$ is a singular point. If $0\notin \overline{\mathcal C_\Gg}$, the common solution of $\Gg=x_1 \pd{\Gg}{x_1} = \ldots = x_n \pd{\Gg}{x_n} = 0$ generating the singular point $z$ according to \cref{sec:ADiscriminantsReultantsPrincipalADets} does not lie on the integration contour $x \notin\mathbb R^n_+$. The same is true, when $\Gg$ is replaced by $\Ff$, which applies to the Landau variety $\mathcal L_1 (\mathcal I_\Gamma)$.
\end{lemma}
\begin{proof}
	The lemma follows obviously from the previous \cref{lem:CoamoebasAndCNV}, which states in particular that $0\notin\overline{\mathcal C_\Gg}$ is equivalent to $\Gg(x)$ being \CNV on $\mathbb R^n_+$. The latter means nothing more than that $\Gg_\tau(x)=0$ has no solutions with $x\in\mathbb R^n_+$ for all faces $\tau\subseteq\Newt(\Gg)$. The application of \cref{thm:pAdet-factorization} concludes the proof.
\end{proof}

Singular points corresponding to multiple roots outside the original integration contour $\mathbb R^n_+$ are also known as \textit{pseudo thresholds}. Thus, \cref{lem:nonphysical} provides a criterion for those pseudo thresholds. We have to remark, that \cref{lem:nonphysical} is not necessarily an effective criterion to determine pseudo thresholds. However, certain approximations of coamoebas are known, which may result in more effective criteria. We will present an overview of those approximations below.\bigskip

Hence, we are mainly interested in the complement of the closure of coamoebas $\gls{Cc}:=\mathbb T^n\setminus\overline{\mathcal C_f}$ because only for points $\theta\in\Cc$ the $\theta$-analogue Euler-Mellin integrals are well-defined. It is well known \cite{NisseGeometricCombinatorialStructure2009}, that this complement $\Cc$ is structured into a finite number of connected, convex components, and we will denote such a connected component by $\Theta$. Moreover, the number of connected components of $\Cc$ will be bounded by $\vol (\Newt (f))$. As pointed out above, the $\theta$-analogue Euler-Mellin integral \cref{eq:EulerMellinTheta} only depends on a choice of a connected component $\Theta$, as one can see simply by a homotopic deformation of the integration contour. Thus, for all $\theta\in\mathbb T^n$ inside the same connected component $\Theta\subset \Cc$, the value of $\mathscr M^\theta_f(s,t)$ stays the same, whereas the value may change for another component $\Theta^\prime$. Therefore, we can also write $\mathscr M^\Theta_f(s,t)$ instead of $\mathscr M^\theta_f(s,t)$ from \cref{eq:EulerMellinTheta}.

Connected components of $\Cc$ will also relate to different branches of the Feynman integral. Therefore, the connected components $\Theta$ will generate the multivalued sheet structure in a certain sense. For the connection between coamoeba and homology groups we refer to \cite{PassareDiscriminantCoamoebasHomology2012, NisseHigherConvexityCoamoeba2015}. However, we have to remark that the coamoeba may not generate the full fundamental group, as the coamoeba can not distinguish between poles of the integrand having the same argument, i.e.\ points which lie on the same ray.

\begin{example}[Coamoeba for the $1$-loop self-energy graph with one mass] \label{ex:CoamoebaBubble1Mass}
    To illustrate the concept of coamoebas we will consider the $1$-loop self-energy graph with one massive edge for different kinematic regions. The Lee-Pomeransky polynomial and the simple principal $A$-determinant are given by
    \begin{align}
    	\Gg &= x_1 + x_2 + (p^2+m_1^2) x_1 x_2 + m_1^2 x_1^2 \\
    	\widehat E_A(\Gg) &= p^2 m_1^2 (p^2+m_1^2)
    \end{align}
    for this particular graph. Therefore, with $m_1^2\neq 0$ we will expect two singular points $p^2=0$ and $p^2=-m_1^2$. Recall, the conventions for momenta from \cref{sec:FeynmanIntegralsIntro}, i.e.\ we will allow $p^2$ to be negative, when not restricting ourselves to Euclidean kinematics. Thus, when exceeding those thresholds the structure of the coamoeba changes, which we can be seen explicitly in \cref{fig:CoamoebaShellBubble1Mass}. \bigskip
    
    Thereby, every point $(\theta_1,\theta_2)\in\mathbb T^2$ which lies not in the closure of the coamoeba defines an integration contour of the Feynman integral $\mathscr M_\Gg^\theta (\nuu,z)$. As aforementioned, the value of $\mathscr M_\Gg^\theta(\nuu,z)$ does not change for two angles $\theta$ and $\theta^\prime$ which are in the same connected component of $\Cc[\Gg]$. The coamoebas in picture \cref{fig:coamoebaBubble1MassA} to \cref{fig:coamoebaBubble1MassE} have only one connected component of $\Cc[\Gg]$. However, when exceeding the threshold $p^2=-m_1^2$ the complement of the closure of the coamoeba consists in two components (\cref{fig:coamoebaBubble1MassF}). Thus, in the first cases we would expect a single analytic continuation, whereas we will expect two different analytic continuations in the region $p^2<-m_1^2$. This will agree with the actual behaviour of the Feynman integral, which is shown in \cref{fig:Bubble1MassBranches}.
    
    Furthermore, we can also see the application of \cref{lem:nonphysical}. Thus, for the threshold $p^2 = 0$ the origin $\theta_1=\theta_2=0$ does not lie in the closure of the coamoeba $0\notin\overline{\mathcal C_\Gg}$. Therefore, $p^2=0$ is only a pseudo threshold, because we have not necessarily to change the original integration contour $\mathbb R^2_+ = \Arg^{-1}(0,0)$.
    
    \begin{figure}
        \centering
        \begin{subfigure}{.49\textwidth}
            \centering\includegraphics[width=.75\textwidth]{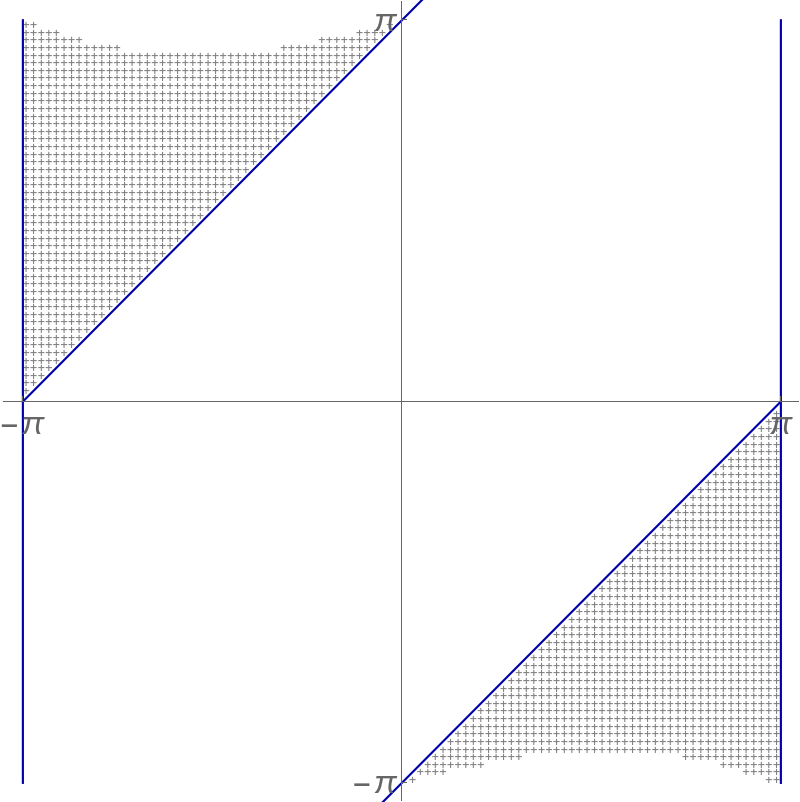}
            \caption{$-p^2<0$} \label{fig:coamoebaBubble1MassA}
        \end{subfigure}
        \begin{subfigure}{.49\textwidth}
            \centering\includegraphics[width=.75\textwidth]{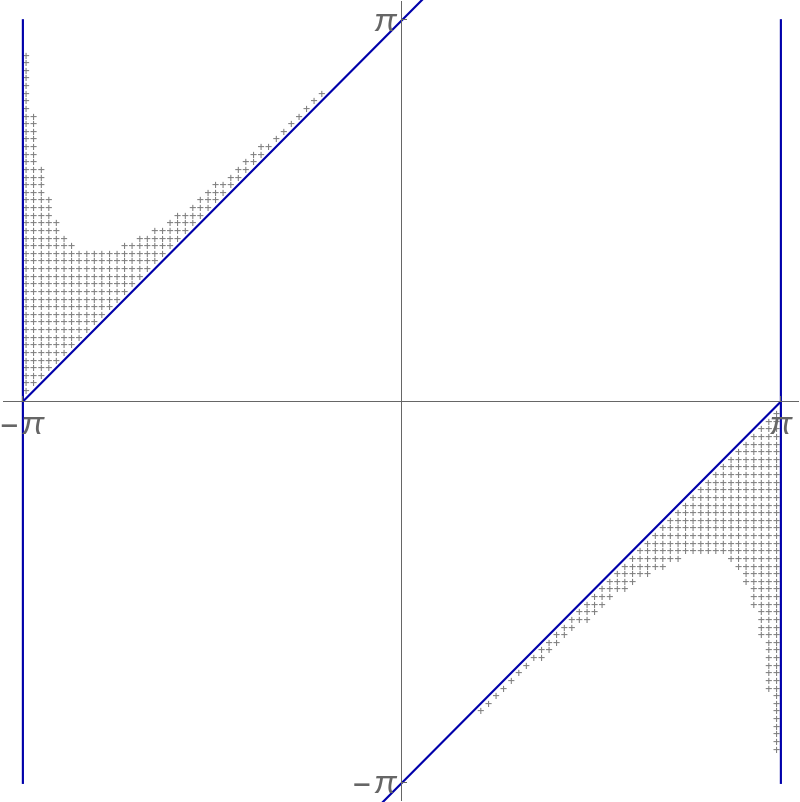}
            \caption{$-p^2<0$}
        \end{subfigure} \\[.5\baselineskip]
        \begin{subfigure}{.49\textwidth}
            \centering\includegraphics[width=.75\textwidth]{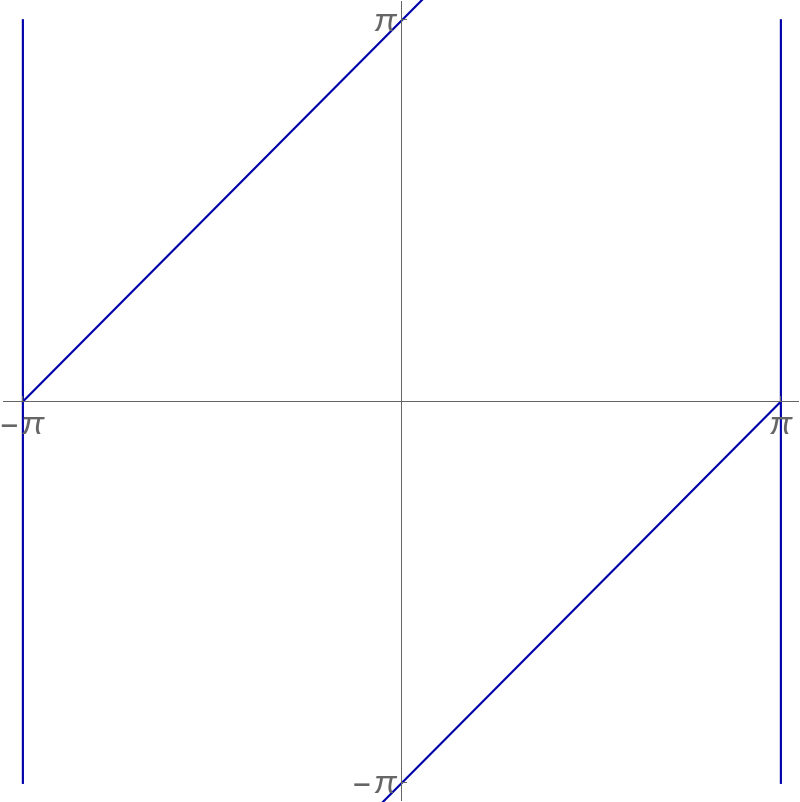} 
            \caption{$p^2=0$}
        \end{subfigure}
        \begin{subfigure}{.49\textwidth}
            \centering\includegraphics[width=.75\textwidth]{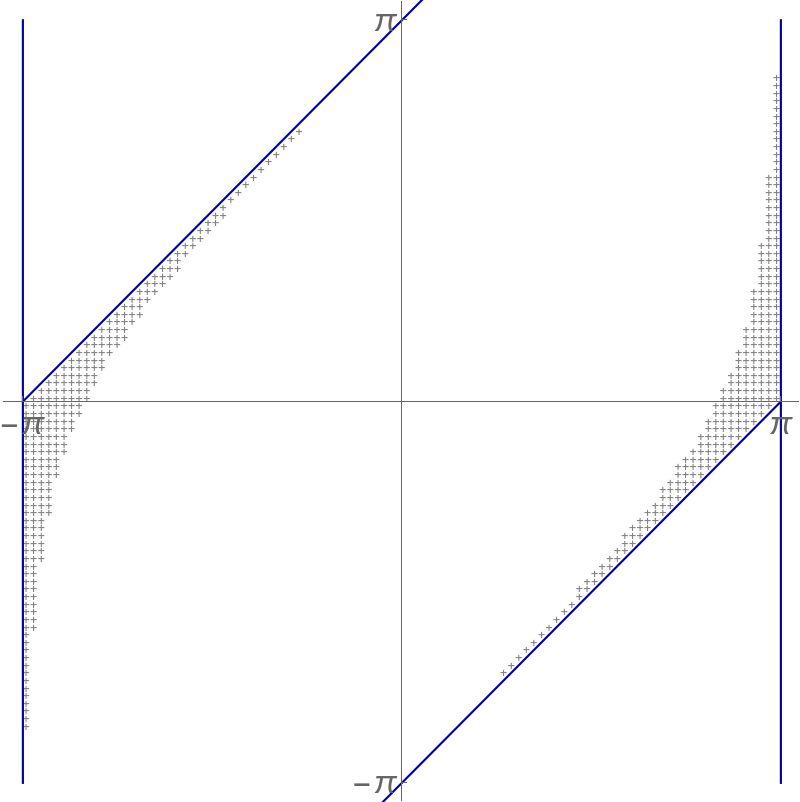}
            \caption{$0<-p^2<m_1^2$}
        \end{subfigure}  \\[.5\baselineskip]
        \begin{subfigure}{.49\textwidth}
            \centering\includegraphics[width=.75\textwidth]{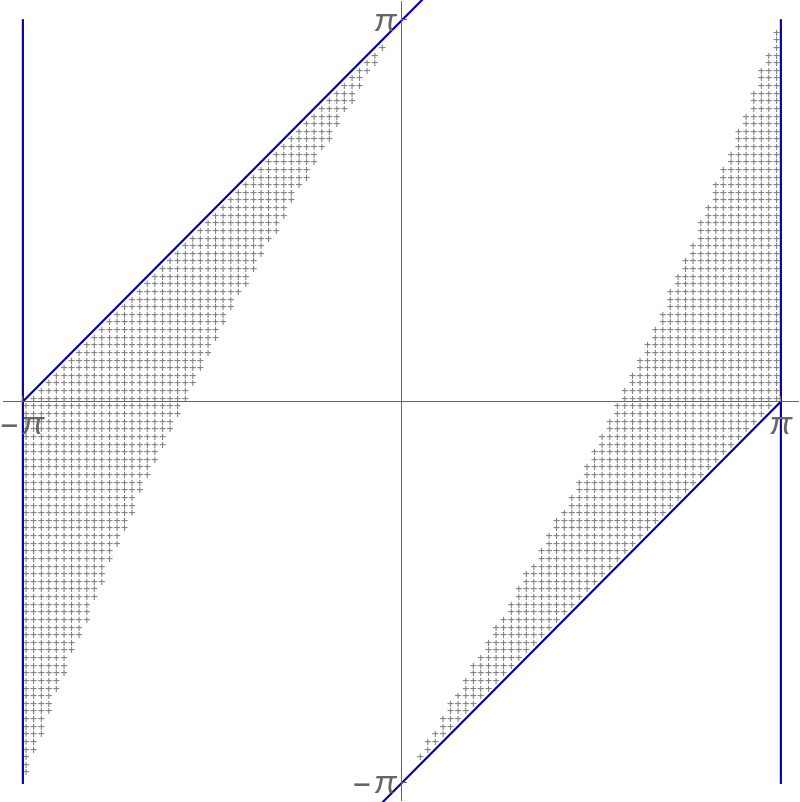}
            \caption{$0<-p^2<m_1^2$} \label{fig:coamoebaBubble1MassE}
        \end{subfigure}
        \begin{subfigure}{.49\textwidth}
            \centering\includegraphics[width=.75\textwidth]{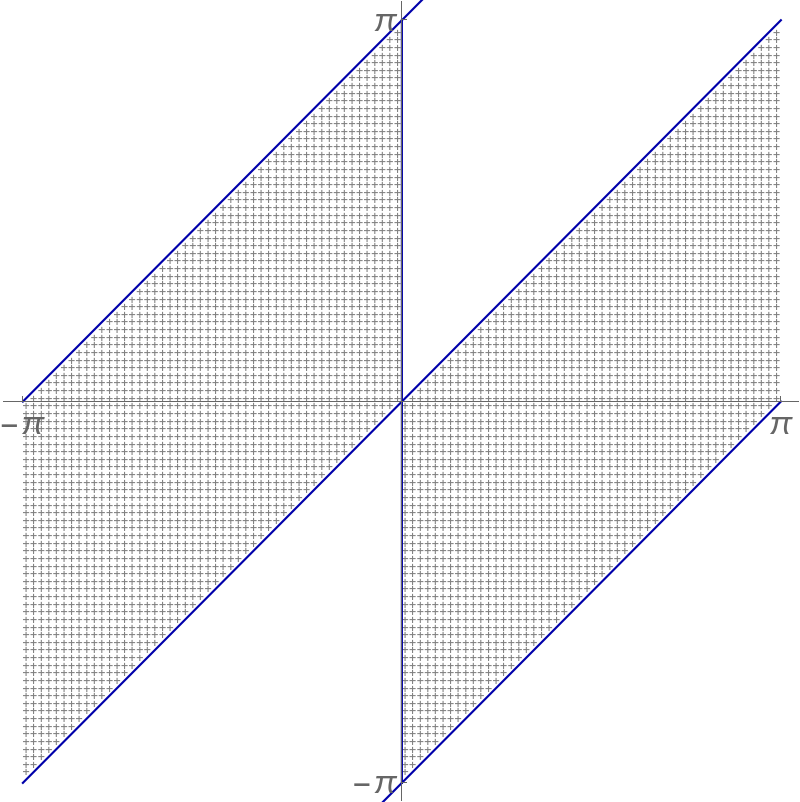}
            \caption{$0<m_1^2<-p^2$} \label{fig:coamoebaBubble1MassF}
        \end{subfigure}
        \caption[Coamoeba $\mathcal C_\Gg$ of the $1$-loop self-energy graph with one mass]{Coamoeba $\mathcal C_\Gg$ of the $1$-loop self-energy graph with one mass discussed in \cref{ex:CoamoebaBubble1Mass}. From \cref{fig:coamoebaBubble1MassA} to \cref{fig:coamoebaBubble1MassF} the momentum $-p^2$ is increased. Whenever a threshold is exceeded the structure of the coamoeba changes. All those coamoebas are drawn only on the principle domain $\mathbb T^2$. Exemplarily, we printed the coamoeba in the case $p^2>0$ on a larger domain in \cref{fig:CoamoebaShellBubble1MassBigPicture}. The blue lines depict the shell $\mathcal H_\Gg$ of the coamoeba.} \label{fig:CoamoebaShellBubble1Mass}
    \end{figure}
        
    \begin{figure}
        \centering\includegraphics[width=.5\textwidth]{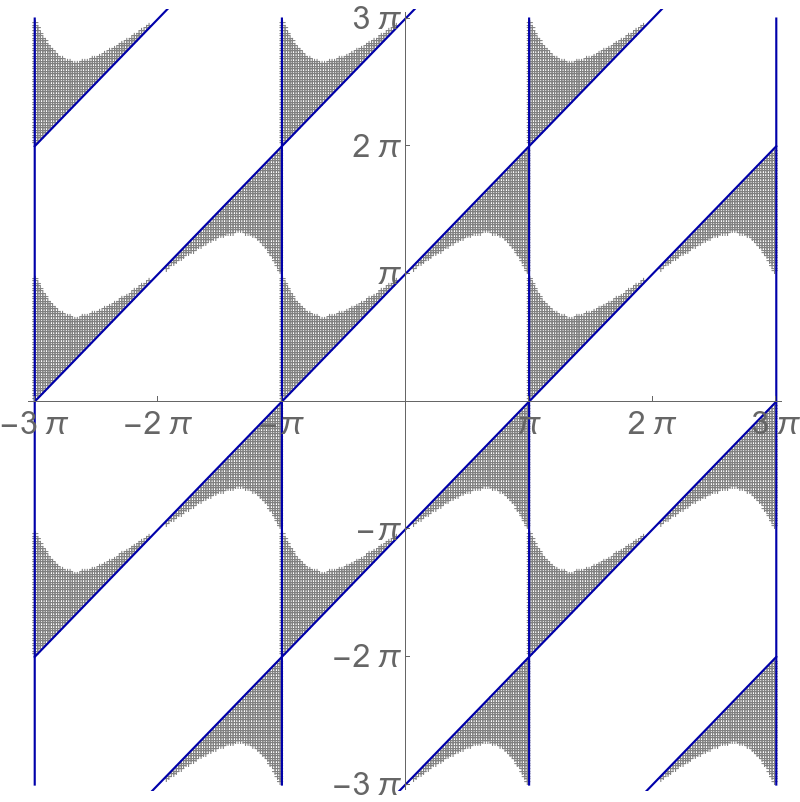}
        \caption[Coamoeba $\mathcal C_\Gg$ of the $1$-loop self-energy graph with one mass (larger region)]{Coamoeba $\mathcal C_\Gg$ of the $1$-loop self-energy graph with one mass for the domain $(-3\pi,3\pi)^2$. Depicted is a coamoeba for the kinematical region $p^2 > 0$. The blue lines show the shell $\mathcal H_\Gg$ of the coamoeba.} \label{fig:CoamoebaShellBubble1MassBigPicture}
    \end{figure}

    \begin{figure}
        \centering\includegraphics[width=.5\textwidth]{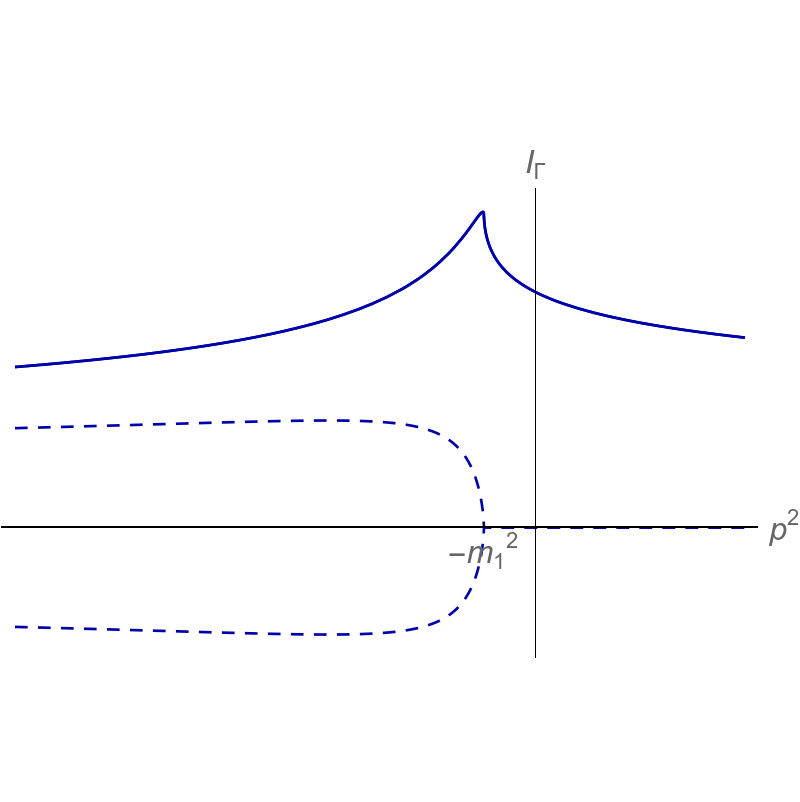}
        \caption[Real and imaginary part of $\mathcal I_\Gamma$ for the $1$-loop self-energy graph with one mass]{Real (solid line) and imaginary part (dashed line) of $\mathcal I_\Gamma$ for the $1$-loop self-energy graph with one mass. The imaginary part of $\mathcal I_\Gamma$ will split into two branches at the threshold $p^2+m_1^2=0$. That observation is consistent with \cref{fig:CoamoebaShellBubble1Mass}. When crossing this particular threshold, the coamoeba has two connected components in $\Cc[\Gg]$. Therefore, we also expect two analytic continuations.} \label{fig:Bubble1MassBranches}
    \end{figure}
\end{example}


As aforementioned, the coamoeba is not very easily accessible in its original definition \cref{eq:defCoamoeba}. However, there are certain ways to approximate coamoebas more efficiently. As seen in \cref{lem:CoamoebasAndCNV}, the closure of the coamoeba $\overline{\mathcal C_f}$ is the union of all coamoebas of truncated polynomials of $f$. Restricting us to the faces of dimension $1$ we obtain the \textit{shell} of $f$ \cite{JohanssonCoamoebas2010, NisseGeometricCombinatorialStructure2009, ForsgardOrderMapHypersurface2015} 
\begin{align}
    \gls{shell} = \bigcup_{\substack{\tau\subseteq\Newt(f) \\ \dim (\tau) = 1}} \mathcal C_{f_\tau} \point
\end{align}
It is known, that each cell of the hyperplane arrangement of $\mathcal H_f$ contains at most one connected component of $\Cc$ \cite{ForsgardHypersurfaceCoamoebasIntegral2012}. Moreover, every line segment with endpoints in $\Cc$ intersecting $\overline{\mathcal C}_f$ has also to intersect a coamoeba $\mathcal C_{f_\tau}$ of a truncated polynomial $f_\tau$ where $\dim(\tau)=1$ \cite[lem. 2.3]{ForsgardOrderMapHypersurface2015}. This behavior will also explain the name ``shell''. Thus, $\mathcal H_f$ carries the rough structure of $\mathcal C_f$. \bigskip

\begin{example}[Shell for the $1$-loop self-energy graph with one mass] \label{ex:ShellBubble1Mass}
    We want to continue \cref{ex:CoamoebaBubble1Mass} and determine the shell $\mathcal H_\Gg$ for the $1$-loop self-energy graph with one mass. There are four truncated polynomials corresponding to a face $\tau\subseteq\Newt(\Gg)$ having dimension one
    \begin{gather}
    	\Gg_{\tau_1} (x) = x_1 + x_2\ \text{,}\quad \Gg_{\tau_2} (x) = x_1 + m_1^2 x_1^2 \\
    	\Gg_{\tau_3} (x) = (p^2+m_1^2) x_1 x_2 + m_1^2 x_1^2\ \text{,}\quad \Gg_{\tau_4} (x) = x_2 + (p^2+m_1^2) x_1 x_2 \point
    \end{gather}
    Replacing the variables $x_i \mapsto r_i e^{i\theta_i}$ with $r_i\in\mathbb R_{>0}$ and $\theta_i\in\mathbb T$ we can directly conclude to the coamoebas
    \begin{align}
    	\mathcal C_{\Gg_{\tau_1}} &= \big\{ \theta\in\mathbb T^2 \,\big\rvert\, \theta_1-\theta_2 \in (2\mathbb Z +1)\pi \big\} \\
    	\mathcal C_{\Gg_{\tau_2}} &= \big\{ \theta\in\mathbb T^2 \,\big\rvert\, \theta_1 \in (2\mathbb Z +1)\pi \big\} \\
    	\mathcal C_{\Gg_{\tau_3}} &= \begin{cases}
                                        \big\{ \theta\in\mathbb T^2 \,\big\rvert\, \theta_1-\theta_2 \in (2\mathbb Z +1)\pi \big\} & \text{if}\  \frac{p^2+m_1^2}{m_1^2} > 0 \\
                                        \big\{ \theta\in\mathbb T^2 \,\big\rvert\, \theta_1-\theta_2 \in (2\mathbb Z) \pi  \big\} & \text{if}\ \frac{p^2+m_1^2}{m_1^2} < 0
    	                             \end{cases} \\
        \mathcal C_{\Gg_{\tau_4}} &= \begin{cases}
                                        \big\{ \theta\in\mathbb T^2 \,\big\rvert\, \theta_1 \in (2\mathbb Z +1)\pi \big\} & \text{if}\ p^2+m_1^2 > 0 \\
                                        \big\{ \theta\in\mathbb T^2 \,\big\rvert\, \theta_1 \in (2\mathbb Z)\pi \big\} & \text{if}\ p^2+m_1^2 < 0
    	                             \end{cases}                     
    \end{align}
    where we assume $m_1^2 \in \mathbb R_{>0}$ and $p^2\in\mathbb R$. We have depicted the shell together with its coamoeba in \cref{fig:CoamoebaShellBubble1Mass}. Thus, we will also see a structural change in the shell, when exceeding the threshold $p^2=-m_1^2$. In particular, we can read off from the shell, that below $p^2=m_1^2$ there can be at most one connected component in $\Cc[\Gg]$, whereas above this threshold there are at most two connected components of $\Cc[\Gg]$ possible. Hence, we get without much effort (the shell is always an arrangement of hyperplanes) the rough structure of the coamoeba.
\end{example}

\vspace{\baselineskip}
Another way, for approximating the coamoeba, is the so-called \textit{lopsided coamoeba} $\mathcal L \mathcal C_f$. According to \cite{ForsgardOrderMapHypersurface2015} we can characterize the lopsided coamoeba as
\begin{align}
	\gls{lop} := \left\{ \theta\in\mathbb T^n \, \big\rvert \, \exists t\in\mathbb R^N_{>0}\ \text{ s.t.}\ t_1 e^{i [\arg(z_1) + a^{(1)} \cdot \theta ]} + \ldots + t_N e^{i [\arg(z_N) + a^{(N)} \cdot \theta ]} = 0 \right\}
\end{align}
where $z_j\in\mathbb C$ and $a^{(j)}\in\mathbb Z^{n}$ refer to the representation of the polynomial $f(x) = \sum_{j=1}^N z_j x^{a^{(j)}}$. Therefore, we have to consider only a system of linear equations to determine if a point $\theta\in\mathbb T^n$ belongs to the lopsided coamoeba. This is considerably simpler than the situation for coamoebas where we have to consider in general non-linear systems $\mathcal C_f = \{\theta\in\mathbb T^n \,\rvert\, \exists r\in\mathbb R^n_{>0}\ \text{s.t.}\ f(r_1 e^{i\theta_1},\ldots,r_n e^{i\theta_n}) = 0 \}$. The lopsided coamoeba is a coarser version of the coamoeba and will contain the coamoeba $\mathcal C_f \subseteq \mathcal L \mathcal C_f$. Furthermore, each connected component of $\Cc$ contains at most one connected component of $\overline{\mathcal L \mathcal C}_f^c$ \cite[prop. 2.2.10]{ForsgardHypersurfaceCoamoebasIntegral2012}. But we will not necessarily find all connected components of $\Cc$ by considering $\overline{\mathcal L \mathcal C}_f^c$. We get immediately the following extension of \cref{lem:nonphysical}.

\begin{lemma} \label{lem:nonphysicalLopsided}
    Assume that $z\in\Var(E_A(\Gg))$ is a singular point. If $0\notin \overline{\mathcal L \mathcal C_\Gg}$, the common solution of $\Gg=x_1 \pd{\Gg}{x_1} = \ldots = x_n \pd{\Gg}{x_n} = 0$ generating the singular point $z$ according to \cref{sec:ADiscriminantsReultantsPrincipalADets} does not lie on the integration contour $x \notin\mathbb R^n_+$. The same is true, when $\Gg$ is replaced by $\Ff$, which applies to the Landau variety $\mathcal L_1 (\mathcal I_\Gamma)$.
\end{lemma}
\begin{proof}
	This follows trivially from \cref{lem:nonphysical} and $\overline{\mathcal C_\Gg} \subseteq \overline{\mathcal L \mathcal C_\Gg}$.
\end{proof}

There is also an alternative description of the lopsided coamoeba. Let $\operatorname{Tri}(f)$ be the set of trinomials, i.e.\ all polynomials which can be formed by removing all but three monomials from $f$. Then we have the following identity \cite[prop. 2.2.5]{ForsgardHypersurfaceCoamoebasIntegral2012}
\begin{align}
    \overline{\mathcal L \mathcal C_f} = \bigcup_{g\in\operatorname{Tri}(f)} \overline{\mathcal C_g} \point
\end{align}
where we only have to determine the roots of trinomials. Furthermore, lopsided coamoebas can also be seen as the union of all coamoebas $\mathcal C_f$, where we scale the coefficients $z_j$ of $f$ by positive reals numbers \cite[lem. 3.8]{ForsgardOrderMapHypersurface2015}. Hence, the lopsided coamoeba is in particular sensible to a change of signs in the coefficients of $f$.

\begin{figure}
	\centering
	\begin{subfigure}{.45\textwidth}
		\centering\includegraphics[width=.75\textwidth]{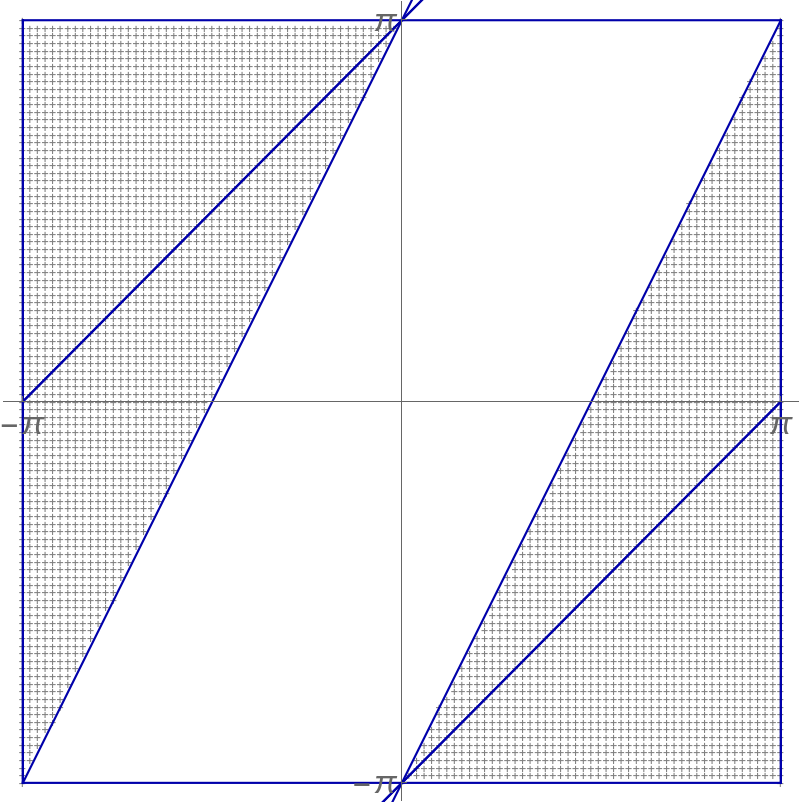}
        \caption{$p^2+m_1^2>0$} 
	\end{subfigure}
	\begin{subfigure}{.45\textwidth}
		\centering\includegraphics[width=.75\textwidth]{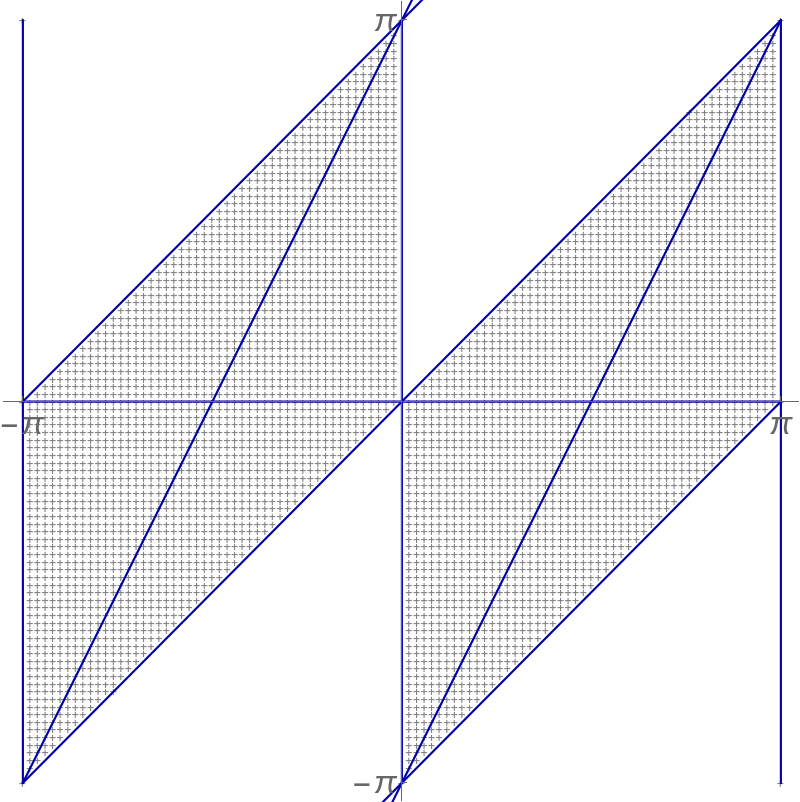}
        \caption{$p^2+m_1^2<0$} 
	\end{subfigure}
	\caption[Lopsided coamoeba $\mathcal C_\Gg$ of the $1$-loop self-energy graph with one mass]{The lopsided coamoeba $\mathcal L \mathcal C_\Gg$ of the $1$-loop self-energy graph with one mass. Unlike the coamoeba (\cref{fig:CoamoebaShellBubble1Mass}) we can only detect two different shapes when varying $p^2\in\mathbb R$ and $m_1^2\in\mathbb R_{>0}$. However, we can see that the threshold $p^2=0$ corresponds only to a pseudo threshold, since the origin $\theta_1=\theta_2=0$ is not contained in the closure of the lopsided coamoeba. The blue lines show the shell $\mathcal L\mathcal H_\Gg$ of the lopsided coamoeba.} \label{fig:LopsidedcoamoebaBubble1Mass}
\end{figure}

Similar to the shell of coamoebas, one can also define a shell of lopsided coamoebas \cite{ForsgardHypersurfaceCoamoebasIntegral2012}, which reads as
\begin{align}
	\gls{lopshell} = \bigcup_{g\in\operatorname{Bin}(f)} \mathcal C_g
\end{align}
where $\operatorname{Bin}(f)$ denotes the set of all binomials of $f$. As roots of binomials are trivial, one can always give a simple algebraic expression for the lopsided shell $\mathcal L \mathcal H_f$. Note, that the boundary $\overline{\mathcal L \mathcal C}_f$ of lopsided coamoebas will be contained in the shell of lopsided coamoebas $\mathcal L \mathcal H_f$ \cite[prop. 3.6]{ForsgardOrderMapHypersurface2015}. Hence, one can determine lopsided coamoebas effectively by their shells. In \cref{fig:LopsidedcoamoebaBubble1Mass} we continued the \cref{ex:CoamoebaBubble1Mass,ex:ShellBubble1Mass} related to \cref{fig:CoamoebaShellBubble1Mass} for the lopsided coamoeba. One can observe, that the lopsided coamoeba changes only for the second threshold $p^2+m_1^2=0$. However, we can detect by means of the much simpler lopsided coamoeba, that $p^2=0$ belongs to a pseudo threshold. Because the origin is not contained in the closure of the lopsided coamoeba $\overline{\mathcal L\mathcal C}_\Gg$ for $p^2=0$ and $m_1^2>0$, the origin will also not contained in the closure of coamoeba due to $\overline{\mathcal L \mathcal C}^c_\Gg \subseteq \overline{\mathcal C}^c_\Gg$. \bigskip

We want to give a second example to illustrate the behaviour of coamoebas and lopsided coamoebas further.

\begin{example}[Coamoebas for the $1$-loop self-energy graph with two masses] \label{eq:CoamoebaBubble2Masses}
	Let us add a second mass to the previous example, and consider
	\begin{align}
		\Uu = x_1 + x_2 \qquad \Ff = (p^2+m_1^2+m_2^2) x_1 x_2 + m_1^2 x_1^2 + m_2^2 x_2^2 \point
	\end{align}
	We can take two different perspectives and either consider the Euler-Mellin integral $\mathscr M_{\Uu+\Ff}^\theta (\nuu,z)$ according to \cref{eq:LeePomeranskyRepresentation} or the Euler-Mellin integral $\mathscr M_{\tilde \Uu\tilde \Ff}^\theta (\beta,z)$ where $\tilde \Uu=\Uu|_{x_n=1}$, $\tilde \Ff = \Ff|_{x_n=1}$ and $\beta = (\nu_0-\omega,\omega,\nu_1,\ldots,\nu_{n-1})$ according to \cref{eq:FeynmanParSpFeynman}. The advantage of $\mathscr M_{\tilde \Uu\tilde \Ff}^\theta (\beta,z)$ is that we have one Schwinger parameter less than in $\mathscr M_{\Uu+\Ff}^\theta (\nuu,z)$. Thus, the coamoeba of the first is an object in $\mathbb T$, whereas the coamoeba for the latter is in $\mathbb T^2$. However, they will both have the same singular locus, since there $\Aa$-hypergeometric systems are equivalent according to \cref{sec:FeynmanIntegralsAsAHyp}
	\begin{align}
        \widehat E_A (\Uu+\Ff) &= \pm p^2 m_1^2 m_2^2 \left[\!\left( p^2+m_1^2+m_2^2\right)^2 - 4m_1^2m_2^2\right] \\
		\widehat E_{\tilde A} (\tilde \Uu \tilde \Ff) &= \pm p^4 m_1^2 m_2^2 \left[\!\left( p^2+m_1^2+m_2^2\right)^2 - 4m_1^2m_2^2\right] \point
	\end{align}
	Hence, with $m_1^2\neq 0$, $m_2^2\neq 0$ we will expect thresholds for $p^2=0$ and $p^2 = - (m_1 \pm m_2)^2$. We printed the corresponding coamoebas $\mathcal C_\Gg$ in \cref{fig:CoamoebaShellBubble2MassG} as well as $\mathcal C_{\tilde \Uu\tilde\Ff}$ in \cref{fig:CoamoebaShellBubble2MassUF}. The main difference in those coamoebas lies in the behaviour for $p^2 < - (m_1+m_2)^2$, where $\mathcal C_\Gg$ does not allow to choose small values $\theta=(\varepsilon,\varepsilon)$ for $\mathscr M_{\Uu+\Ff}^\theta (\nuu,z)$, whereas the neighbourhood of the origin  is not contained in $\mathcal C_{\tilde \Uu\tilde\Ff}$. In both coamoebas we can locate the thresholds $p^2$ and $p^2=-(m_1-m_2)^2$ as pseudo thresholds by means of \cref{lem:nonphysical}.

    \begin{figure}
        \centering
        \begin{subfigure}{.49\textwidth}
            \centering\includegraphics[width=.75\textwidth]{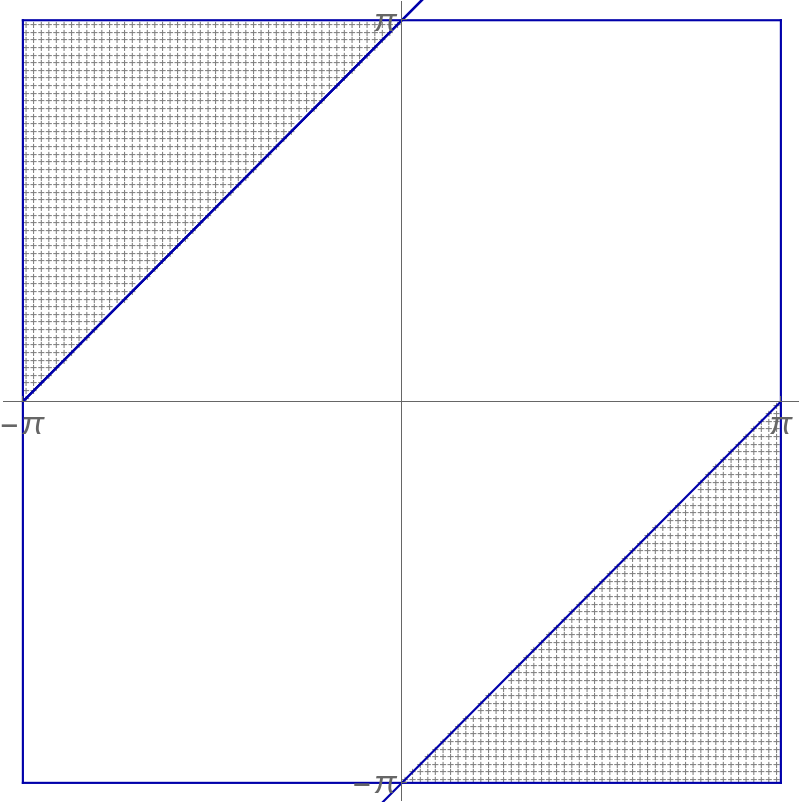}
            \caption{$p^2>0$} \label{fig:coamoebaBubble2MassA}
        \end{subfigure}
        \begin{subfigure}{.49\textwidth}
            \centering\includegraphics[width=.75\textwidth]{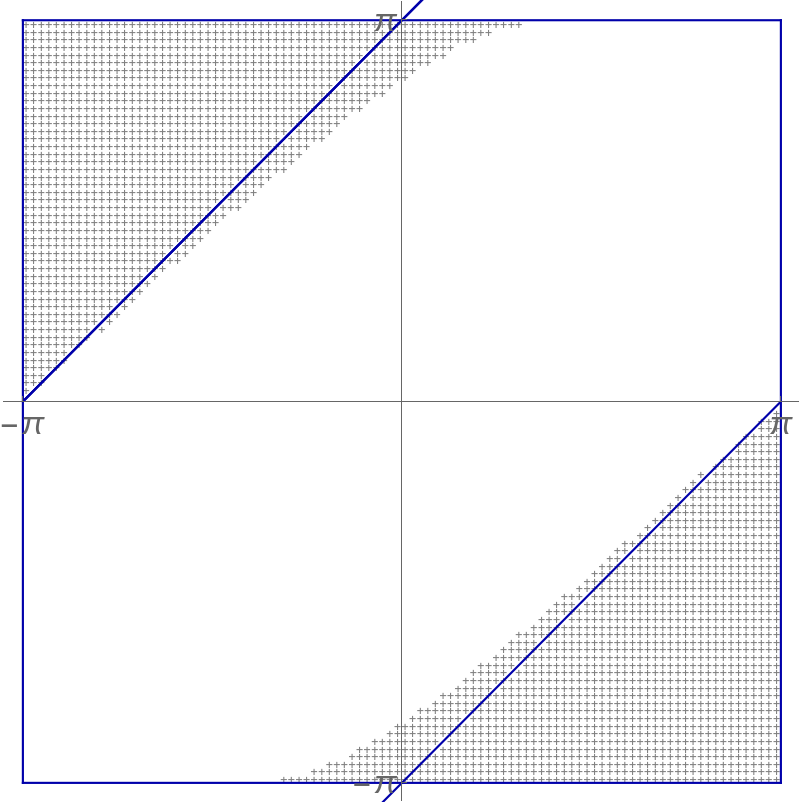}
            \caption{$0<-p^2<(m_1-m_2)^2$} \label{fig:coamoebaBubble2MassB}
        \end{subfigure} \\[.5\baselineskip]
        \begin{subfigure}{.49\textwidth}
            \centering\includegraphics[width=.75\textwidth]{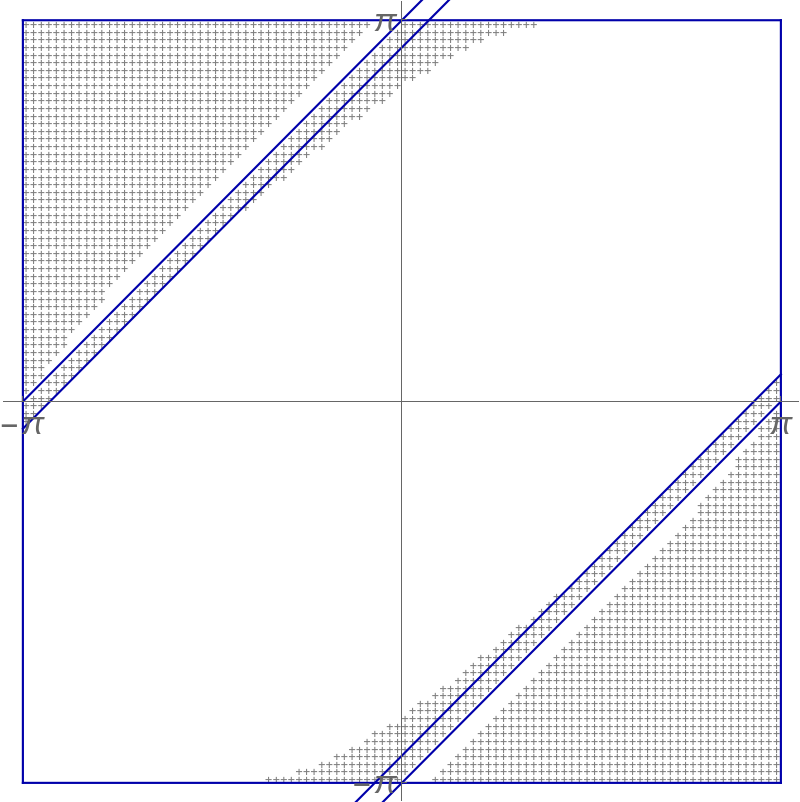}
            \caption{$(m_1-m_2)^2<-p^2<(m_1+m_2)^2$} \label{fig:coamoebaBubble2MassC}
        \end{subfigure}
        \begin{subfigure}{.49\textwidth}
            \centering\includegraphics[width=.75\textwidth]{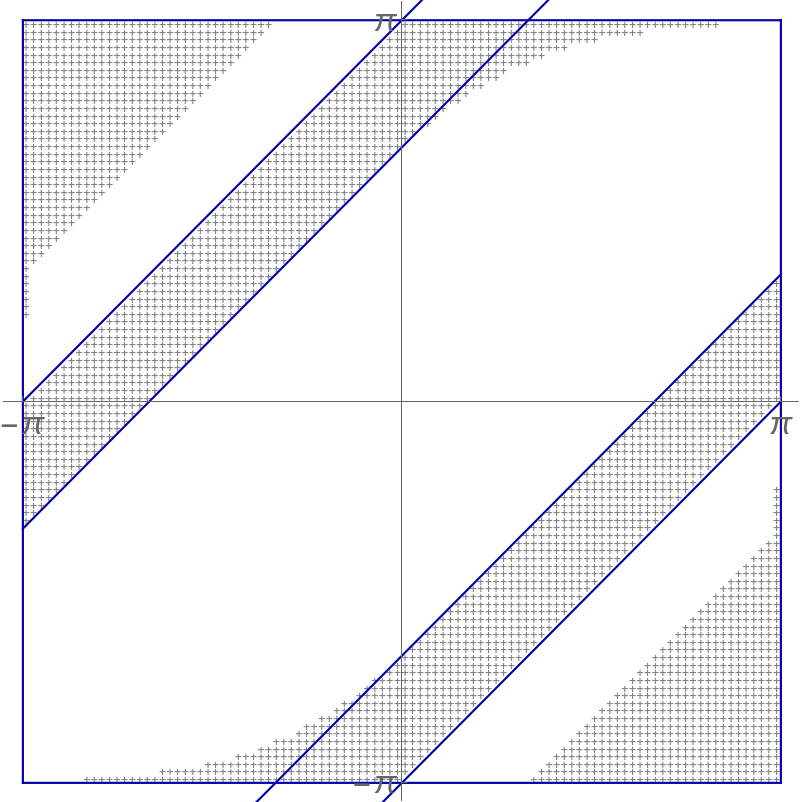}
            \caption{$(m_1-m_2)^2<-p^2<(m_1+m_2)^2$} \label{fig:coamoebaBubble2MassD}
        \end{subfigure}  \\[.5\baselineskip]
        \begin{subfigure}{.49\textwidth}
            \centering\includegraphics[width=.75\textwidth]{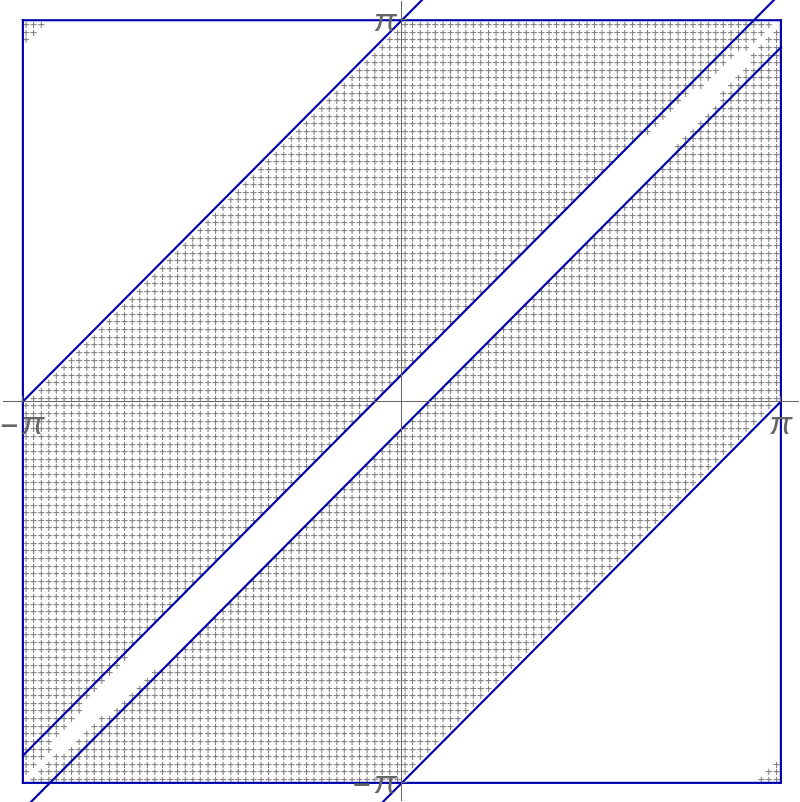}
            \caption{$(m_1-m_2)^2<-p^2<(m_1+m_2)^2$} \label{fig:coamoebaBubble2MassE}
        \end{subfigure}
        \begin{subfigure}{.49\textwidth}
            \centering\includegraphics[width=.75\textwidth]{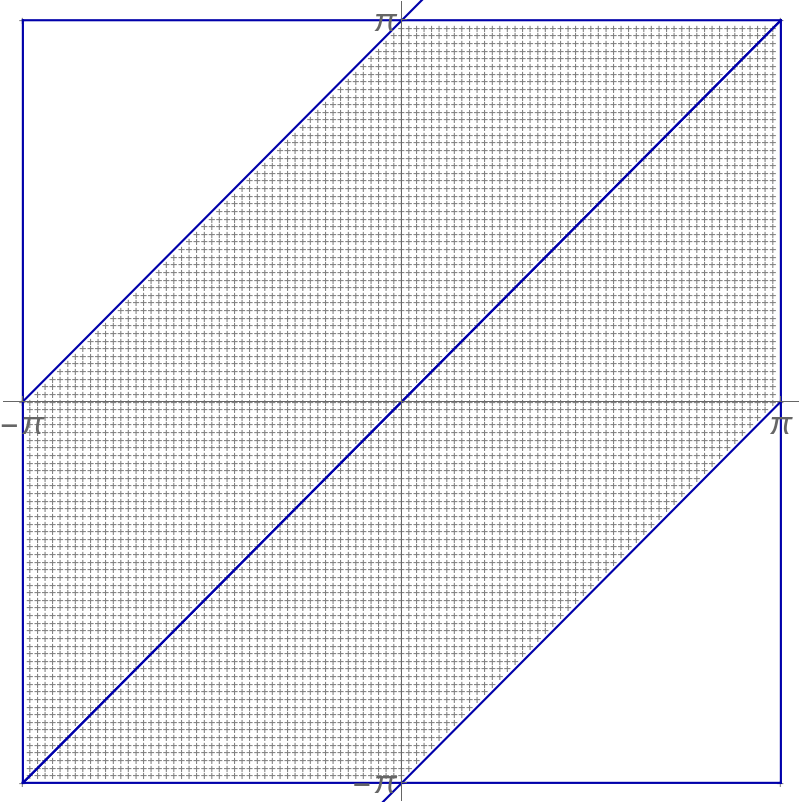}
            \caption{$(m_1+m_2)^2<-p^2$} \label{fig:coamoebaBubble2MassF}
        \end{subfigure}
        \caption[Coamoeba $\mathcal C_\Gg$ of the $1$-loop self-energy graph with two masses]{Coamoeba $\mathcal C_\Gg$ of the $1$-loop self-energy graph with two masses discussed in \cref{eq:CoamoebaBubble2Masses}. From \cref{fig:coamoebaBubble2MassA} to \cref{fig:coamoebaBubble2MassF} the momentum $-p^2$ is increased. Whenever a threshold is exceeded the structure of the coamoeba changes. In \cref{fig:coamoebaBubble2MassA,fig:coamoebaBubble2MassB} $\Cc[\Gg]$ consists in one component, in \cref{fig:coamoebaBubble2MassC,fig:coamoebaBubble2MassD,fig:coamoebaBubble2MassE} $\Cc[\Gg]$ contains three components, whereas for \cref{fig:coamoebaBubble2MassF} there are two components in $\Cc[\Gg]$. The blue lines depict the shell $\mathcal H_\Gg$ of the coamoeba. Remarkably, in \cref{fig:coamoebaBubble2MassF} the closure of the coamoeba $\overline{\mathcal C_\Gg}$ contains the full neighbourhood of the origin.} \label{fig:CoamoebaShellBubble2MassG}
    \end{figure}

   \begin{figure}
        \centering
        \begin{subfigure}{.45\textwidth}
            \centering\includegraphics[width=5cm,trim=0cm 2cm 0cm 2cm, clip]{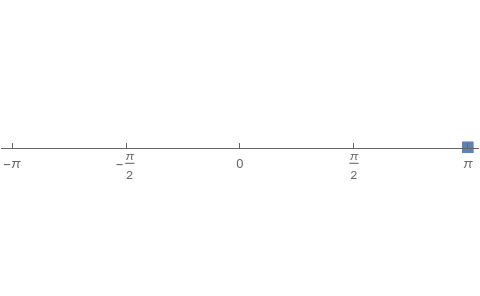}
            \caption{$-p^2<(m_1-m_2)^2$}
        \end{subfigure}
        \begin{subfigure}{.45\textwidth}
            \centering\includegraphics[width=5cm,trim=0cm 2cm 0cm 2cm, clip]{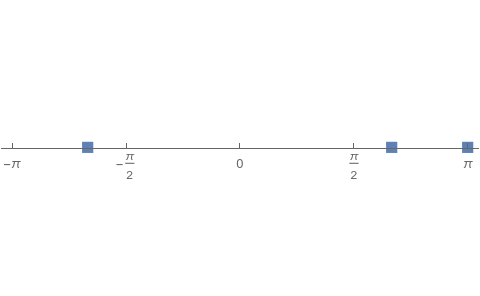}
            \caption{$(m_1-m_2)^2 < -p^2 < (m_1+m_2)^2$}
        \end{subfigure} \\
        \begin{subfigure}{.45\textwidth}
            \centering\includegraphics[width=5cm,trim=0cm 2cm 0cm 2cm, clip]{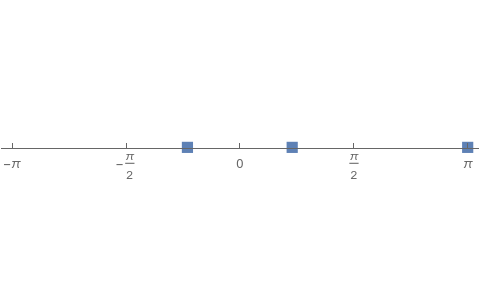}
            \caption{ $(m_1-m_2)^2 < -p^2 < (m_1+m_2)^2$}
        \end{subfigure}
        \begin{subfigure}{.45\textwidth}
            \centering\includegraphics[width=5cm,trim=0cm 2cm 0cm 2cm, clip]{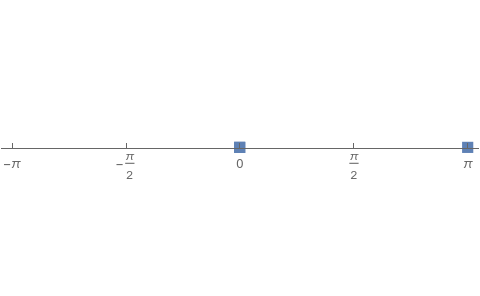}
            \caption{$-p^2>(m_1+m_2)^2$}
        \end{subfigure}
        \caption[Coamoeba $\mathcal{C}_{\tilde{\Uu} \tilde{\Ff}}$ of the $1$-loop self-energy graph with two masses]{Coamoeba $\mathcal C_{\tilde\Uu\tilde\Ff}$ of the $1$-loop self-energy graph with two masses discussed in \cref{eq:CoamoebaBubble2Masses}. In $\mathcal C_{\tilde\Uu\tilde\Ff}$ we will only observe the two thresholds $p^2=-(m_1\pm m_2)^2$, where we can determine $p^2=-(m_1-m_2)^2$ to be a pseudo threshold. Note, that after exceeding the second threshold $p^2=-(m_1+m_2)^2$, the neighbourhood of the origin belongs not to $\overline{\mathcal C_{\tilde\Uu\tilde\Ff}}$. However, the number of connected components in $\Cc[\tilde\Uu\tilde\Ff]$ coincides with the number of connected components in $\Cc[\Gg]$ from \cref{fig:CoamoebaShellBubble2MassG}.} \label{fig:CoamoebaShellBubble2MassUF}
    \end{figure}

   \begin{figure}
        \centering
        \begin{subfigure}{.45\textwidth}
            \centering\includegraphics[width=.75\textwidth]{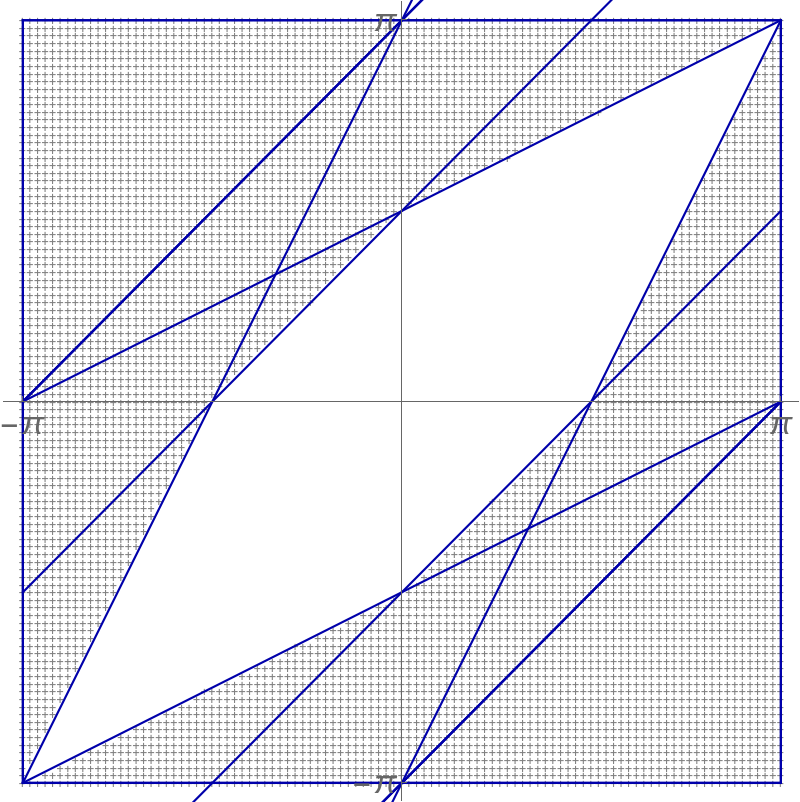}
            \caption{$p^2+m_1^2+m_2^2>0$}
        \end{subfigure}
        \begin{subfigure}{.45\textwidth}
            \centering\includegraphics[width=.75\textwidth]{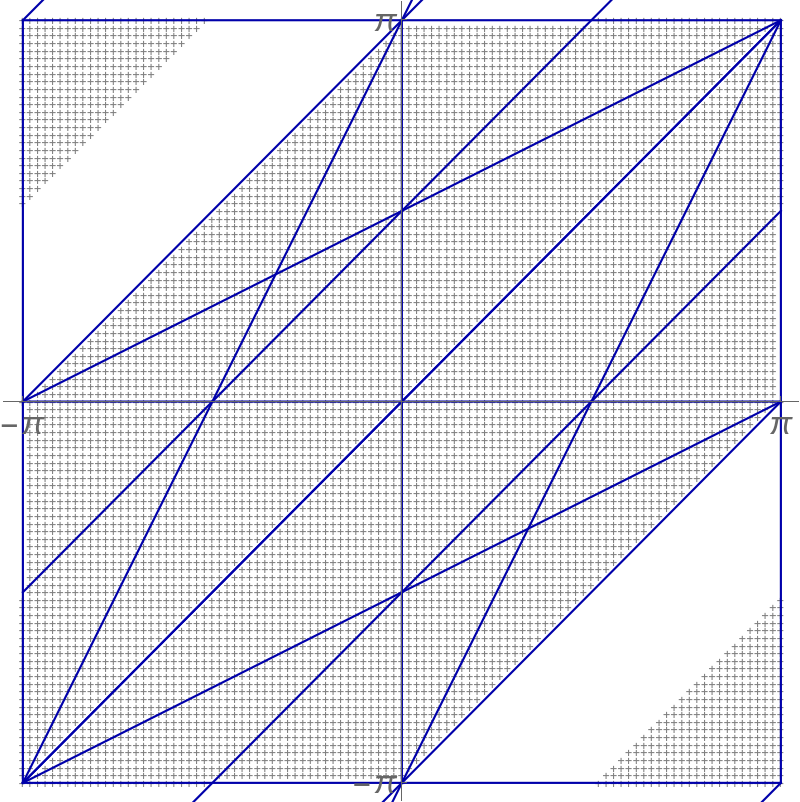}
            \caption{$p^2+m_1^2+m_2^2<0$}
        \end{subfigure}
        \caption[Lopsided coamoeba $\mathcal C_{\Gg}$ of the $1$-loop self-energy graph with two masses]{Lopsided coamoeba $\mathcal L \mathcal C_{\Gg}$ of the $1$-loop self-energy graph with two masses discussed in \cref{eq:CoamoebaBubble2Masses}. The lopsided coamoeba changes only when passing $p^2+m_1^2+m_2^2 = 0$, i.e.\ when the coefficients of $\Gg$ change their signs. However, one can conclude from the lopsided coamoeba, that $p^2=-(m_1-m_2)^2$ corresponds to a pseudo threshold. The blue lines depict the lopsided shell $\mathcal L\mathcal H_\Gg$.}
    \end{figure}
\end{example}

At the end of this chapter we want to give a comprehensive answer to the question of analytic continuations of Feynman integrals. We already discussed the analytic continuation with respect to parameters $\nuu$ in \cref{sec:DimAnaReg}. We also discussed the singular locus $\Sing(H_\Aa(\nuu))$ with respect to the variables $z$ in \cref{sec:LandauVarieties,sec:2ndtypeSingularities} and introduced the $\theta$-analogue, which applies whenever $\Gg(x)$ is not \CNV on $(0,\infty)^n$. Hence, we want to complete the picture by a result of \cite{BerkeschEulerMellinIntegrals2013}, that the $\theta$-analogue of a generalized Feynman integral has a multivalued analytic continuation to all points except the singular locus.

\begin{theorem}[{analytic continuation of $\theta$-analogue, generalized Feynman integrals \cite[thm. 4.2]{BerkeschEulerMellinIntegrals2013}}]
    Let $\Sing(H_\Aa(\nuu)) = \Var( E_A(\Gg)) \subset \mathbb C^N$ be the singular locus of $\mathscr M_\Gg^\Theta$, $z\in\mathbb C^N\setminus\Sing(H_\Aa(\nuu))$ a point outside the singular locus and $\Theta$ a connected component of $\mathbb T^n\setminus\overline{\mathcal C_\Gg}$. Then $\mathscr M_\Gg^\Theta(\nuu,z)$ has a multivalued analytic continuation to $\mathbb C^{n+1}\times (\mathbb C^N\setminus\Sing(H_\Aa(\nuu)))$, which is everywhere $\Aa$-hypergeometric.
\end{theorem} 

Hence, we can analytically continue the generalized, $\theta$-analogue Feynman integral for its parameters $\nuu$ as well as for its variables $z$ for any connected component $\Theta$ of $\Cc[\Gg]$. \bigskip

We want to conclude this section by a short recapitulation of the main results. We have shown, that the $\theta$-analogue Euler-Mellin integrals (when they exist for small values of $\theta$) states an alternative to Feynman's $i\varepsilon$ prescription. Those $\theta$-analogue Euler-Mellin integrals will help us to investigate the kind of solutions of the Landau equations. Thereby, the behaviour of $\theta$-analogue Euler-Mellin integrals is determined by the coamoeba, i.e.\ the coamoeba shows which contours we can take in the Feynman integral. Hence, we can find the properties of Landau singularities by considering the coamoeba, e.g.\ we can give conditions for pseudo thresholds as well as we can determine the discontinuity of Feynman integrals by means of components of $\Cc[\Gg]$. Moreover, certain aspects of coamoebas can be approximated by the shells and the lopsided coamoeba.\bigskip

\pagestyle{withoutSections}

\chapter{Conclusion and outlook}


The aim of this thesis was to characterize Feynman integrals in the context of $\Aa$-hypergeometric theory. This point of view is quite recently and was initiated inter alia by one of the articles on which this thesis is based on. However, the connection between Feynman integrals and general hypergeometric functions is a long-standing discussion in research going back to Tullio Regge. As we have illustrated above, $\Aa$-hypergeometric functions are the appropriate framework to address and answer these questions. In particular, we showed in \cref{thm:FeynmanAHypergeometric} that every generalized, scalar Feynman integral is an $\Aa$-hypergeometric function. Thereby, $\Aa\in\mathbb Z^{(n+1)\times N}$ is given as the homogenized vector configuration of the support $A$ of the Lee-Pomeransky polynomial $\Gg=\Uu+\Ff = \sum_{a^{(j)}\in A} z_j x^{a^{(j)}}$, and the GKZ-parameter $\nuu=\left(\dhalf,\nu_1,\ldots,\nu_n\right)\in\mathbb C^{n+1}$ is determined by the spacetime dimension $d$ and the indices $\nu_i$ of the Feynman integral. Hence, the Feynman integral $\mathcal I_\Aa(\nuu,z)$ depends only on three objects: The vector configuration $\Aa\in\mathbb Z^{(n+1)\times N}$, which is determined by the graph topology, the parameters $\nuu\in\mathbb C^{n+1}$, which are especially important in dimensional and analytical regularization and on the variables $z\in\mathbb C^N$, which encode the dependence on masses and external momenta. The characterization of Feynman integrals to be $\Aa$-hypergeometric means, that every generalized, scalar Feynman integral $\mathcal I_\Aa(\nuu,z)$ satisfies the following two types of partial differential equations
\begin{align}
    \left(\sum_{j=1}^N a^{(j)} z_j \partial_j + \nuu\right) \bullet \mathcal I_\Aa(\nuu,z) &= 0 \\
    \left(\prod_{l_j>0} \partial_j^{l_j} - \prod_{l_j<0} \partial_j^{-l_j}\right) \bullet \mathcal I_\Aa(\nuu,z) &= 0 
\end{align}
for all $l\in\mathbb L = \Dep(\Aa)\cap\mathbb Z^N$, which generate the holonomic $D$-module $H_\Aa(\nuu)$. From the theory of holonomic $D$-modules, one can determine the dimension of the solution space $\Sol(H_\Aa(\nuu))$. Thus, the dimension of the vector space of holomorphic functions satisfying this system of partial differential equations is given by the volume of the Newton polytope $\Newt(\Gg)$, whenever $\nuu$ is generic (\cref{thm:CKKT,thm:HolRankAHyp}). For many classes of Feynman integrals, this holds even for non-generic $\nuu$ (\cref{thm:CohenMacaulayFeynman}). Furthermore, it was shown that also every coefficient in a Laurent expansion as appearing in dimensional or analytical regularization can be expressed by $\Aa$-hypergeometric functions (\cref{sec:epsilonExpansion}).\bigskip

From the property to be $\Aa$-hypergeometric, we also derived series representations for every Feynman integral. Hence, for each regular triangulation $\Tt$ of the Newton polytope $\Newt(\Gg)$ one obtains a series representation in the following form (\cref{thm:FeynSeries})
\begin{align}
    \mathcal I_\Aa(\nuu,z) = \frac{1}{\Gamma(\nu_0-\omega)\Gamma(\nu)} \sum_{\sigma\in \hatT} \frac{z_\sigma^{-\Aas^{-1}\nuu}}{|\det (\Aas)|} \sum_{\lambda\in\mathbb N_0^{N-n-1}} \!\!\!\! \frac{(-1)^{|\lambda|}}{\lambda!} \Gamma\!\left(\Aas^{-1}\nuu+\Aas^{-1}\Aabs\lambda\right) \frac{z_{\bar\sigma}^\lambda}{z_\sigma^{\Aas^{-1}\Aabs\lambda}} \ \text{.}
\end{align}
Power series of this type are also known as Horn hypergeometric functions. As there are in general many different ways to triangulate a polytope, there also exist many different series representations for a single Feynman integral, which can also be connected by hypergeometric transformations. Those series representations are especially made for an efficient numerical evaluation, since one can choose those series representations which converge fast for a given kinematic region. For the purpose of a practical usage of this concept, we discussed possible obstacles which can appear in the concrete evaluation and gave certain strategies to solve them. Especially, we have offered a simple method to determine the analytic continuation in case of non-generic values $z$ (\cref{sec:AnalyticContinuation}), as they appear for more complex Feynman graphs. An implementation of this method sketched in \cref{ch:seriesRepresentations} is already in progress and is planned for publication. 

Besides of Horn hypergeometric series, the GKZ approach allows also the representation of Feynman integrals in terms of other hypergeometric functions, as e.g.\ Euler integrals, which we discussed shortly in \cref{sec:EulerIntegrals}.\bigskip

Hence, we have characterized the Feynman integral in the three different notions in which hypergeometric functions can appear: as various types of integrals, as solutions of specific partial differential equations and by a certain type of series. From the perspective of general hypergeometric functions the common ground of these representations is the vector configuration $\Aa\in\mathbb Z^{(n+1)\times N}$, which is also the only information from the topology of the Feynman graph that has an influence to the Feynman integral. The possibility to characterize Feynman integrals in those three different ways is one of the key features of $\Aa$-hypergeometric functions. \bigskip

Besides numerical applications, there are also structurally very interesting implications for the Feynman integral from $\Aa$-hypergeometric theory. Thus, we have investigated the kinematic singularities of scalar Feynman integrals from the $\Aa$-hypergeometric perspective. This point of view provides a mathematically rigorous description of those singularities by means of principal $A$-determinants. More precisely, it turns out that the singular locus of Feynman integrals is the variety defined by the principal $A$-determinant of the sum of the Symanzik polynomials $\Gg=\Uu+\Ff$ (\cref{thm:SingularLocusPrincipalAdet}) 
\begin{align}
    \Sing(H_\Aa(\nuu)) = \Var (E_A(\Gg)) = \Var (\widehat E_A(\Gg))
\end{align}
or by the simple principal $A$-determinant $\widehat E_A(\Gg)$, respectively. Thereby, principal $A$-determinants $E_A(\Gg)\in\mathbb Z [z_1,\ldots,z_N]$ are polynomials in the coefficients of the polynomial $\Gg$, uniquely determined up to a sign. 
Every principal $A$-determinant $E_A(\Gg)$ can be factorized into $A$-discriminants, where each one corresponds to a face of $\Newt(\Gg)$  (\cref{thm:pAdet-factorization}). Further, we can divide the $A$-discriminants for $\widehat E_A(\Gg)$ into four different parts
\begin{align}
	\widehat E_A (\Gg) = \pm \vspace{-1em} \prod_{\tau\subseteq \Newt(\Gg)} \Delta_{A\cap\tau} (\Gg_\tau) = \pm \, \Delta_A(\Gg)  \cdot \widehat E_{\Aa_\Ff}(\Ff) \cdot \widehat E_{\Aa_\Uu}(\Uu) \cdot R \point
\end{align}
We have shown, that the $A$-discriminant of the full polytope $\Delta_A(\Gg)$ can be identified with the so-called second-type singularities (\cref{sec:2ndtypeSingularities}). The principal $A$-determinant of the second Symanzik polynomial $\Ff$ can be associated with the Landau variety, i.e.\ $\mathcal L_1 (\mathcal I_\Gamma) = \Var \Big(\widehat E_{\Aa_\Ff}(\Ff) \Big)$ (\cref{thm:LandauVar}). The polynomial $R$ contains all remaining $A$-discriminants coming from proper, mixed faces of $\Newt(\Gg)$ and corresponds to second-type singularities of subgraphs, whereas $\widehat E_{\Aa_\Uu} (\Uu)$ has no influence to the singular locus when considering the restriction to physically relevant values of $z$. Remarkably, we have found that the singular locus of Feynman integrals (especially the Landau variety $\mathcal L_1(\mathcal I_\Gamma)$) includes parts, which were overlooked in previous approaches. This is due to the fact that not all truncated polynomials will have an equivalent polynomial coming from a subgraph. However, this difference has an impact only for Feynman graphs beyond $1$-loop or banana graphs (\cref{lem:SubgraphPolynomialsVsTruncated}). This is may the reason why these forgotten singularities were not detected earlier. Exemplarily, we presented the Landau variety of the dunce's cap graph in \cref{sec:ExampleDuncesCap}, where one can observe also those additional contributions.\bigskip

From the perspective of $\Aa$-hypergeometric theory, the $\Aa$-hypergeometric functions are often said to be ``quantizations'' of $A$-discriminants \cite{GelfandDiscriminantsResultantsMultidimensional1994}. Here, ``quantization'' should not be confused with the quantization in physics, which is why the use of this term may be misleading in this context. However, the relation between $A$-discriminants and $\Aa$-hypergeometric functions has certain formal similarities with the quantization procedure in physics. Thus, if one maps the $A$-discriminants from the commutative polynomial ring in a specific manner to a non-commutative Weyl algebra, one obtains a partial differential equation system which precisely characterizes the $\Aa$-hypergeometric function. In this sense, Feynman integrals are the ``quantization'' of Landau varieties. \bigskip

By means of \cref{thm:NewtSec} we revealed an unexpected connection between Landau varieties and the set of all triangulations of $\Newt(\Ff)$ going back to \cite{GelfandDiscriminantsResultantsMultidimensional1994}. It states, that the Newton polytope of $E_{\Aa_\Ff}(\Ff)$ coincides with the secondary polytope
\begin{align}
	\Newt(E_{\Aa_\Ff}(\Ff)) = \Sigma(A_\Ff) \point
\end{align}
The same holds also for the principal $A$-determinant $E_A(\Gg)$ and the secondary polytope $\Sigma(A_\Gg)$. This relation can be understood against the background of series representations, where the structure of variables -- and thus also their singularities -- depends on triangulations. \bigskip

Apart from the description of the singular locus, we also introduced a powerful tool to determine the singular locus: the \HKp. This parameterization relies in the same vein on the above-mentioned relation between the singular locus and series representations. Thus, the calculation of a Gale dual is sufficient to obtain a parameterization of the hypersurface defined by an $A$-discriminant. Clearly, such a parametric representation of a variety differs from a representation via defining polynomials. However, such a representation can be even more convenient for many approaches, as we describe the singularities directly. Also having in mind, that a representation of Landau varieties by a defining polynomial will be an incommensurable effort for almost all Feynman graphs, we want to advertise the usage of \HKp.\bigskip

Finally, in order to study the monodromy of Feynman integrals, we introduced an Euler-Mellin integral with a rotated integration contour. We showed that this \mbox{$\theta$-analogue} Euler-Mellin integral coincides with Feynman's $i\varepsilon$ prescription in the limit $\varepsilon\rightarrow 0^+$, whenever this limit exists. However, these $\theta$-analogue Euler-Mellin integrals have several advantages compared to the $i\varepsilon$ prescription. In particular, in the descriptions of those $\theta$-analogue Euler-Mellin integrals one does not have to take a limit to discuss the discontinuity at a branch cut.

Moreover, the behaviour of $\theta$-analogue Euler-Mellin integrals and the monodromy of Feynman integrals is substantially determined by the coamoeba of the polynomial $\Gg$. From the shape of the coamoeba of $\Gg$, we can also conclude to the nature of the singularities of Feynman integrals, i.e.\ we can identify pseudo thresholds. We sketched also several ways to approximate the coamoeba in order to derive efficient algorithms. However, the application of coamoebas to Feynman integrals leaves many questions open and will surely be a worthwhile focus for future research.\bigskip

We would like to conclude this thesis with an incomplete list of questions and ideas that we think are of interest and that might be answered with the help of $\Aa$-hypergeometric theory. First of all, it would be desirable to extend the knowledge about kinematic singularities and monodromy of Feynman integrals. Especially, we would like to find a clear description what happens with the singular locus when variables take non-generic values, since this case is usually left out in classical treatments about this subject. In particular, we would like to find the dimension of the solution space when considering non-generic values. Further, it would be interesting to know if the truncated polynomials which contribute to the singular locus $\Sing(H_\Aa(\nuu))$ always have a graphical equivalent, as appeared in the example of \cref{sec:ExampleDuncesCap}. To study the monodromy of Feynman integrals, we would also like to investigate more in the coamoeba.

Moreover, it would be interesting to apply the $\Aa$-hypergeometric approach also on other stages of perturbative QFT. Hence, we can consider the GKZ systems for each Horn hypergeometric series appearing in the Laurent expansion of a Feynman integral. This would give an alternative way for the analytic continuation of those series, as well as a method to find simpler representations. Further, it would be promising to try the $\Aa$-hypergeometric approach also to the $S$-matrix itself, which goes back to a suggestion in \cite{KawaiMicrolocalStudySmatrix1975, SatoRecentDevolpmentHyperfunction1975}. 

As Feynman integrals form only a subclass of $\Aa$-hypergeometric functions, it would be interesting, if one can characterize this subclass further. For example, there is the legitimate hypothesis that Feynman integrals are always Cohen-Macaulay (see \cref{thm:CohenMacaulayFeynman}), when external momenta are considered to be generic. 

Furthermore, we would like to get an understanding of linear relations of Feynman integrals from the $\Aa$-hypergeometric perspective. Very recently, there was a method proposed to generate linear relations by means of Pfaffian systems coming from $\Aa$-hypergeometric system \cite{ChestnovMacaulayMatrixFeynman2022}. Usually, one constructs partial differential equations in the kinematic variables from those linear relations. Hence, it could be inspiring to see this relation on the more formal level of GKZ. Thus, one has to consider the relation\footnote{The connection between the $D$-module $H_\Aa(\nuu)$ and a shift algebra, shifting the values of $\nuu$ by integer amounts, is already indicated by \cref{eq:FeynmanJDerivative}. In this context we also refer to \cite[ch. 5]{SaitoGrobnerDeformationsHypergeometric2000}.} between the $D$-module $H_\Aa(\nuu)$ and the $s$-parametric annihilator $\Ann(\Gg^s)$ generating all linear relations among Feynman integrals \cite{BitounFeynmanIntegralRelations2019}.

To connect the hypergeometric perspective stronger to the graph perspective it would be fruitful to consider the graph operations deletion and contraction on the level of Feynman integrals. \bigskip

By the series representations and the singular locus of Feynman integrals we studied only two particular aspects which arise from the connection between $\Aa$-hypergeometric theory and Feynman integrals. Since Feynman integrals and $\Aa$-hypergeometric functions were mainly developed independently of each other, we strongly believe, that physicists and mathematicians can learn much from each other in this area.


\pagestyle{fancyStandard}
\appendix


\thispagestyle{plain}

\printglossary[type=symbols,title={List of Symbols}]

%
%
%
%
%
%
%
%
%

\chapter{Appendix}


\section{Stirling numbers} \label{sec:StirlingNumbers}


In \cref{sec:epsilonExpansion} we introduced Stirling numbers of the first kind $\StirlingFirstSmall{n}{k}$ in order to expand $\Gamma$-functions around integer arguments. We would like to recall here certain facts about Stirling numbers from previous sections, as well as include further useful relations. Those numbers were originally introduced by Stirling in 1730 \cite{TweddleJamesStirlingAnnotated2003} in the study of Pochhammer symbols $(x)_n$, which are also known as rising factorials\footnote{Note that there are many different, contradictory notational conventions in literature for rising and falling factorials.} \glsadd{Pochhammer}
\begin{gather}
	(x)_n := \frac{\Gamma(x+n)}{\Gamma(x)} \qquad\text{for } x,x+n\notin\mathbb Z_{\leq 0} \label{eq:PochhammerAppendixDef}\\
	(x)_n = \prod_{k=0}^{n-1} (x+k) \quad\text{and}\quad (x)_{-n} = \prod_{k=1}^n \frac{1}{x-k} \qquad\text{for } n\in\mathbb Z_{\geq 0} \point \label{eq:PochhammerAppendixIntegerDef}
\end{gather}
There are many useful relations for Pochhammer symbols known, which can be deduced directly from their definitions \cref{eq:PochhammerAppendixDef}, \cref{eq:PochhammerAppendixIntegerDef} by taking recourse to properties of $\Gamma$-functions
\begin{align}
	(x)_{-n} &= (-1)^n \frac{1}{(1-x)_n} \qquad \text{for } n\in\mathbb Z  \label{eq:PochhammerAppendixA} \\ 
	(x)_{m+n} &= (x)_m (x+m)_n  \\
	(x)_{mn} &= m^{mn} \prod_{j=0}^{m-1} \left(\frac{x+j}{m}\right)_n \qquad \text{for } m\in\mathbb Z_{>0} \\
	(x+y)_n &= \sum_{k=0}^n \Binomial{n}{k} (x)_k (y)_{n-k} \qquad \text{for } n\in\mathbb Z_{\geq 0} \point
\end{align}
We refer to \cite{WolframResearchIncMathematicalFunctionsSite} for a comprehensive collection of further identities. Stirling numbers of the first kind can be defined as the coefficients in the expansion of Pochhammer symbols $(x)_n$ around $x=0$ \glsadd{StirlingFirst}
\begin{align}
	(x)_n = \frac{\Gamma(x+n)}{\Gamma(x)} = \sum_{k\geq 0} \StirlingFirst{n}{k} x^k \point
\end{align}
Originally defined for $n\in\mathbb N$ only, this expansion extends naturally to $n\in\mathbb Z$ by means of \cref{eq:PochhammerAppendixA}. However, when $n\in\mathbb Z_{\geq 0}$, the Pochhammer symbol $(x)_n$ is only a polynomial in $x$. Hence, many Stirling numbers have to be zero
\begin{align}
	\StirlingFirst{n}{0} = \StirlingFirst{0}{k} = 0 \quad\text{for all } n,k\in\mathbb Z_{>0} \quad\text{and}\quad \StirlingFirst{n}{k}=0\quad\text{for all } n,k\in\mathbb Z_{>0}\quad\text{with } k>n\point
\end{align}
Furthermore, we have $\StirlingFirstSmall{n}{n} = 1$ for all $n\in\mathbb Z_{\geq 0}$. From the perspective of combinatorics, Stirling numbers of the first kind give the number of permutations of $\{1,\ldots,n\}$ into $k$ non-empty cycles \cite{ComtetAdvancedCombinatoricsTheory1974}.\bigskip

As aforementioned, one can relate Stirling numbers of the first kind with $S$-sums and $Z$-sums from \cite{MochNestedSumsExpansion2002}, and we find \glsadd{Zsum1} \glsadd{Ssum}\glsunset{Ssum} \glsadd{Ssum1}\glsunset{Ssum1}
\begin{align}
	\StirlingFirst{n}{k} &= (n-1)! \sum_{n>i_{k-1}>\ldots>i_1>0} \frac{1}{i_1\cdots i_{k-1}} = (n-1)! \, \Zsum{n-1}{\underbrace{\scriptstyle 1,\ldots,1}_{k-1}} \quad\text{for } n,k\in\mathbb Z_{>0} \\
	\StirlingFirst{-n}{k} &= \frac{(-1)^n}{n!} \sum_{n\geq i_k \geq \ldots \geq i_1 \geq 1} \frac{1}{i_1\cdots i_k} = \frac{(-1)^n}{n!} \, \Ssum{n}{\underbrace{\scriptstyle 1,\ldots,1}_{k}} \quad\text{for } n,k\in\mathbb Z_{> 0} \point
\end{align}
As explained in \cite{MochNestedSumsExpansion2002} one can transform $S$-sums and $Z$-sums into each other by a recursive approach of \cref{eq:splittingNested}. Thus, one can also find connections between Stirling numbers of the first kind with positive $n$ and those with negative $n$. For a more general consideration of relations between $S$-sums and multiple $\zeta$-values we suggest \cite{HoffmanMultipleHarmonicSeries1992}. Especially, for Stirling numbers of the first kind with negative first argument there is the following simpler expression \cite{BransonExtensionStirlingNumbers1996}
\begin{align}
	\StirlingFirst{-n}{k} = \frac{(-1)^n}{n!} \sum_{i=1}^n \frac{(-1)^{i+1}}{i^k} \Binomial{n}{i} \quad\text{for } n\in\mathbb Z_{\geq 0}, k\in\mathbb Z_{>0} \point \label{eq:StirlingFirstNegativeDefAppendix}
\end{align}
Stirling numbers $\StirlingFirstSmall{n}{k}$ with positive first argument $n\in\mathbb Z_{>0}$, have the alternative representation, which relates directly to their combinatorial meaning \cite{ComtetAdvancedCombinatoricsTheory1974, OlverNISTHandbookMathematical2010}
\begin{align}
	\StirlingFirst{n}{k} = \sum_{n > i_{n-k} > \ldots > i_1 > 0} i_1 \cdots i_{n-k} \quad\text{for } n,k\in\mathbb Z_{> 0}, n > k \comma
\end{align}
which is especially useful, when $n-k$ is small.\bigskip

For practical purpose, one can relate the Stirling numbers of the first kind $\StirlingFirstSmall{n}{k}$ with fixed $k$ to harmonic numbers $H_n^{(k)} := \sum_{i=1}^k \frac{1}{i^k}$ and $H_{n} := H_n^{(1)}$. We will extend the list of relations given in \cref{eq:StirlingPartPos} and \cref{eq:StirlingPartNeg}. For positive first argument $n\in\mathbb Z_{>0}$ we find by use of \cref{eq:StirlingNumbersRecursionPos}
\begin{align}
	\StirlingFirst{n}{0} &= 0 \allowdisplaybreaks \\
	\StirlingFirst{n}{1} &= (n-1)! \allowdisplaybreaks \\
	\StirlingFirst{n}{2} &= (n-1)!\ H_{n-1} \allowdisplaybreaks \\
	\StirlingFirst{n}{3} &= \frac{(n-1)!}{2!}  \left(H_{n-1}^2 - H_{n-1}^{(2)}\right) \allowdisplaybreaks \\
	\StirlingFirst{n}{4} &= \frac{(n-1)!}{3!}  \left( H_{n-1}^3 - 3 H_{n-1} H_{n-1}^{(2)} + 2 H_{n-1}^{(3)}\right) \allowdisplaybreaks \\
	\StirlingFirst{n}{5} &= \frac{(n-1)!}{4!} \left[ H_{n-1}^4 - 6 H_{n-1}^2 H_{n-1}^{(2)} + 3 \left(H_{n-1}^{(2)}\right)^2 + 8 H_{n-1} H_{n-1}^{(3)} - 6 H_{n-1}^{(4)}\right] \allowdisplaybreaks \\
	\StirlingFirst{n}{6} &= \frac{(n-1)!}{5!} \left[ H_{n-1}^5 - 10 H_{n-1}^3 H_{n-1}^{(2)} + 20 H_{n-1}^2 H_{n-1}^{(3)} \right. \nonumber\\
	& \left. + 15 H_{n-1} \left(H_{n-1}^{(2)}\right)^2  - 30 H_{n-1} H_{n-1}^{(4)} - 20 H_{n-1}^{(2)} H_{n-1}^{(3)} + 24 H_{n-1}^{(5)} \right] \allowdisplaybreaks \\
	\StirlingFirst{n}{7} &= \frac{(n-1)!}{6!} \left[ H_{n-1}^6 - 15 H_{n-1}^4 H_{n-1}^{(2)} + 40 H_{n-1}^3 H_{n-1}^{(3)} + 45 H_{n-1}^2 \left(H_{n-1}^{(2)}\right)^2 \right. \nonumber \\
	& - 90 H_{n-1}^2 H_{n-1}^{(4)} - 120  H_{n-1} H_{n-1}^{(2)} H_{n-1}^{(3)} + 144 H_{n-1} H_{n-1}^{(5)} - 15 \left(H_{n-1}^{(2)}\right)^3 \nonumber \\
	& \left. + 40 \left( H_{n-1}^{(3)} \right)^2 + 90 H_{n-1}^{(2)} H_{n-1}^{(4)} - 120 H_{n-1}^{(6)} \right] \allowdisplaybreaks \\
	\StirlingFirst{n}{8} &= \frac{(n-1)!}{7!} \left[ H_{n-1}^7 - 21 H_{n-1}^{(2)} H_{n-1}^5 + 70 H_{n-1}^{(3)} H_{n-1}^4 + 105 \left( H_{n-1}^{(2)} \right)^2 H_{n-1}^3  \right. \nonumber \\
	& \left. - 210 H_{n-1}^{(4)} H_{n-1}^3 - 420 H_{n-1}^{(2)} H_{n-1}^{(3)} H_{n-1}^2 + 504 H_{n-1}^{(5)} H_{n-1}^2 - 105 \left( H_{n-1}^{(2)} \right)^3 H_{n-1}  \right. \nonumber \\
	& \left. + 280 \left( H_{n-1}^{(3)} \right)^2 H_{n-1} + 630 H_{n-1}^{(2)} H_{n-1}^{(4)} H_{n-1} - 840 H_{n-1}^{(6)} H_{n-1} + 210 \left( H_{n-1}^{(2)} \right)^2 H_{n-1}^{(3)}  \right. \nonumber \\
	& \left. - 420 H_{n-1}^{(3)} H_{n-1}^{(4)} - 504 H_{n-1}^{(2)} H_{n-1}^{(5)} + 720 H_{n-1}^{(7)} \right] \allowdisplaybreaks \\
	\StirlingFirst{n}{9} &= \frac{(n-1)!}{8!} \left[ H_{n-1}^8 - 28 H_{n-1}^{(2)} H_{n-1}^6 + 112 H_{n-1}^{(3)} H_{n-1}^5 + 210 \left( H_{n-1}^{(2)} \right)^2 H_{n-1}^4  \right. \nonumber \\
	& \left. - 420 H_{n-1}^{(4)} H_{n-1}^4 - 1120 H_{n-1}^{(2)} H_{n-1}^{(3)} H_{n-1}^3 + 1344 H_{n-1}^{(5)} H_{n-1}^3 - 420 \left( H_{n-1}^{(2)} \right)^3 H_{n-1}^2  \right. \nonumber \\
	& \left. + 1120 \left( H_{n-1}^{(3)} \right)^2 H_{n-1}^2 + 2520 H_{n-1}^{(2)} H_{n-1}^{(4)} H_{n-1}^2 - 3360 H_{n-1}^{(6)} H_{n-1}^2  \right. \nonumber \\
	& \left. + 1680 \left( H_{n-1}^{(2)} \right)^2 H_{n-1}^{(3)} H_{n-1} - 3360 H_{n-1}^{(3)} H_{n-1}^{(4)} H_{n-1} - 4032 H_{n-1}^{(2)} H_{n-1}^{(5)} H_{n-1}  + \right. \nonumber \allowdisplaybreaks\\
	& \left. + 5760 H_{n-1}^{(7)} H_{n-1} + 105 \left( H_{n-1}^{(2)} \right)^4 - 1120 H_{n-1}^{(2)} \left( H_{n-1}^{(3)} \right)^2 + 1260 \left( H_{n-1}^{(4)} \right)^2  \right. \nonumber \\
	& \left. - 1260 \left( H_{n-1}^{(2)} \right)^2 H_{n-1}^{(4)} + 2688 H_{n-1}^{(3)} H_{n-1}^{(5)} + 3360 H_{n-1}^{(2)} H_{n-1}^{(6)} - 5040 H_{n-1}^{(8)} \right] \point
\end{align}

For Stirling numbers $\StirlingFirstSmall{-n}{k}$ where the first argument attends negative values, we can derive very similar relations by means of \cref{eq:StirlingFirstNegativeDefAppendix}
\begin{align}
	\StirlingFirst{-n}{0} &= \frac{(-1)^n}{n!} \allowdisplaybreaks \\
	\StirlingFirst{-n}{1} &= \frac{(-1)^n}{n!} H_n \allowdisplaybreaks \\
	\StirlingFirst{-n}{2} &= \frac{(-1)^n}{2!\, n!}  \left(H_n^2 + H_n^{(2)}\right) \allowdisplaybreaks \\
	\StirlingFirst{-n}{3} &= \frac{(-1)^n}{3!\, n!}  \left( H_n^3 + 3 H_n H_n^{(2)} + 2 H_n^{(3)}\right) \allowdisplaybreaks \\
	\StirlingFirst{-n}{4} &= \frac{(-1)^n}{4!\, n!} \left[ H_n^4 + 6 H_n^2 H_n^{(2)} + 3 (H_n^{(2)})^2 + 8 H_n H_n^{(3)} + 6 H_n^{(4)}\right] \allowdisplaybreaks \\
	\StirlingFirst{-n}{5} &= \frac{(-1)^n}{5!\, n!} \left[ H_n^5 + 10 H_n^3 H_n^{(2)} + 20 H_n^2 H_n^{(3)} \right. \nonumber\\
	& \left. + 15 H_n \left(H_n^{(2)}\right)^2  + 30 H_n H_n^{(4)} + 20 H_n^{(2)} H_n^{(3)} + 24 H_n^{(5)} \right] \allowdisplaybreaks \\
	\StirlingFirst{-n}{6} &= \frac{(-1)^n}{6!\, n!} \left[ H_n^6 + 15 H_n^4 H_n^{(2)} + 40 H_n^3 H_n^{(3)} + 45 H_n^2 \left(H_n^{(2)}\right)^2 + 90 H_n^2 H_n^{(4)} \right. \nonumber \\
	&\left. + 120  H_n H_n^{(2)} H_n^{(3)}\! + \! 144 H_n H_n^{(5)} \! + \! 15 \left(H_n^{(2)}\right)^3 \! + \! 40 \left( H_n^{(3)} \right)^2 \! + \! 90 H_n^{(2)} H_n^{(4)} \! + \! 120 H_n^{(6)} \right] \allowdisplaybreaks \\
	\StirlingFirst{-n}{7} &= \frac{(-1)^n}{7!\, n!} \left[ H_n^7 + 21 H_n^{(2)} H_n^5 + 70 H_n^{(3)} H_n^4 + 105 \left( H_n^{(2)} \right)^2 H_n^3  \right. \nonumber \\
	& \left. + 210 H_n^{(4)} H_n^3 + 420 H_n^{(2)} H_n^{(3)} H_n^2 + 504 H_n^{(5)} H_n^2 + 105 \left( H_n^{(2)} \right)^3 H_n  \right. \nonumber \\
	& \left. + 280 \left( H_n^{(3)} \right)^2 H_n + 630 H_n^{(2)} H_n^{(4)} H_n + 840 H_n^{(6)} H_n + 210 \left( H_n^{(2)} \right)^2 H_n^{(3)}  \right. \nonumber \\
	& \left. + 420 H_n^{(3)} H_n^{(4)} + 504 H_n^{(2)} H_n^{(5)} + 720 H_n^{(7)} \right] \allowdisplaybreaks \\
	\StirlingFirst{-n}{8} &= \frac{(-1)^n}{8!\, n!} \left[ H_n^8 + 28 H_n^{(2)} H_n^6 + 112 H_n^{(3)} H_n^5 + 210 \left( H_n^{(2)} \right)^2 H_n^4  \right. \nonumber \\
	& \left. + 420 H_n^{(4)} H_n^4 + 1120 H_n^{(2)} H_n^{(3)} H_n^3 + 1344 H_n^{(5)} H_n^3 + 420 \left( H_n^{(2)} \right)^3 H_n^2  \right. \nonumber \\
	& \left. + 1120 \left( H_n^{(3)} \right)^2 H_n^2 + 2520 H_n^{(2)} H_n^{(4)} H_n^2 + 3360 H_n^{(6)} H_n^2 + \right. \nonumber \allowdisplaybreaks \\
	& \left. + 1680 \left( H_n^{(2)} \right)^2 H_n^{(3)} H_n + 3360 H_n^{(3)} H_n^{(4)} H_n + 4032 H_n^{(2)} H_n^{(5)} H_n  \right. \nonumber \\
	& \left. + 5760 H_n^{(7)} H_n + 105 \left( H_n^{(2)} \right)^4 + 1120 H_n^{(2)} \left( H_n^{(3)} \right)^2 + 1260 \left( H_n^{(4)} \right)^2  \right. \nonumber \\
	& \left. + 1260 \left( H_n^{(2)} \right)^2 H_n^{(4)} + 2688 H_n^{(3)} H_n^{(5)} + 3360 H_n^{(2)} H_n^{(6)} + 5040 H_n^{(8)} \right]
\end{align}
where $n\in\mathbb Z_{\geq 0}$.\bigskip

Stirling numbers of the first kind $\StirlingFirstSmall{n}{k}$ will grow very fast for $n\rightarrow\infty $ and are close to zero when $n\rightarrow -\infty$. Hence, for numerical calculations it is much more convenient to work with a slightly adapted definition in order to avoid the handling of very large or very small numbers. For this reason, we introduced the following variant in \cref{sec:numerics}
\begin{align}
	\gls{sigmank} &= \frac{1}{(n-1)!} \StirlingFirst{n}{k+1} = \sum_{ n > i_k > \ldots > i_1 > 0} \frac{1}{i_1 \cdots i_k} \label{eq:AppendixSigmaPos} \\
	\sigma(-n,k) &= (-1)^n n! \StirlingFirst{-n}{k} = \sum_{ n \geq i_k \geq \ldots \geq i_1 \geq 1} \frac{1}{i_1 \cdots i_k} = \sum_{i=1}^n \frac{(-1)^{i+1}}{i^k} \Binomial{n}{i} \label{eq:AppendixSigmaNeg}
\end{align}
for $n,k\in\mathbb Z_{>0}$ and $\sigma(n,0)=1$ for $n\in\mathbb Z$. In comparison to the Stirling numbers, those $\sigma(n,k)$ are numerical stable even for high values of $|n|$. Immediately from \cref{eq:AppendixSigmaPos} and \cref{eq:AppendixSigmaNeg} we can read off the very rough estimations\footnote{These bounds are far from being strict and are only given for the purpose to see, that $\sigma(n,k)$ takes moderate values for high values of $|n|$. To derive those inequalities, note that they are nothing else than the restriction to the first summand for lower bounds and an extension of the summation region to $i\in \{1,\ldots,n-1\}^k$ and $i\in \{1,\ldots,n\}^k$, respectively, to obtain the upper bounds. However, it is not very hard to establish more sharp bounds for $\sigma(n,k)$ which show that $\sigma(n,k)$ is also stable for large $k$.} $\frac{1}{k!} \leq \sigma(n,k) \leq (H_{n-1})^k$ as well as $\frac{1}{k!} (H_n)^k \leq \sigma(-n,k) \leq (H_n)^k$ with $n,k\in\mathbb Z_{>0}$, which show that $\sigma(n,k)$ is stable even when $|n|$ attends high values. \bigskip

Closely related to the Stirling numbers of the first kind are the \textit{Stirling numbers of the second kind} $\StirlingSecondSmall{n}{k}$. Those numbers can be defined by
\begin{gather}
    \StirlingSecond{0}{0} = 1 \quad\text{,}\qquad \StirlingSecond{n}{0} = \StirlingSecond{0}{n} = 0 \quad\text{for } n\in\mathbb Z_{>0} \\
	\gls{StirlingSecond} = \frac{1}{k!} \sum_{i=0}^k (-1)^{k-i} \Binomial{k}{i} i^n \quad\text{for } n\in\mathbb Z_{>0},k\in\mathbb Z_{\geq 0}
\end{gather}
and give the number of partitions of $\{1,\ldots,n\}$ into $k$ non-empty subsets. Alternatively, we can write Stirling numbers of the second kind as the $k$-th order finite differences of the function $x^n$ at $x=0$, i.e.\ we have $\StirlingSecondSmall{n}{k} = \frac{1}{k!} \Delta^k x^n(0)$ (see \cref{eq:kthFiniteDifference} for the definition of finite differences). 

Stirling numbers of first and second kind can be understood as inverse to each other in the following sense. Let us consider two sequences $\{a_0,\ldots\}$ and $\{b_0,\ldots\}$, which are connected by the so-called \textit{Stirling transform} $\mathscr S:\{a_i\}_{i\in\mathbb N} \mapsto \{b_i\}_{i\in\mathbb N}$ defined by
\begin{align}
	b_n = \sum_{k=0}^n \StirlingSecond{n}{k} a_k \point
\end{align}
Then, the inverse transformation $\mathscr S^{-1}$ can be formulated by means of Stirling numbers of the first kind \cite{BernsteinCanonicalSequencesIntegers2002}
\begin{align}
	a_n = \sum_{k=0}^n (-1)^{n-k} \StirlingFirst{n}{k} b_k \point
\end{align}
Therefore, Stirling numbers obey a certain orthogonality relation \cite[sec. 26.8]{OlverNISTHandbookMathematical2010}
\begin{align}
	\sum_{j\geq 0} (-1)^{n-j} \StirlingFirst{n}{j} \StirlingSecond{j}{k} = \sum_{j\geq 0} (-1)^{j-k} \StirlingSecond{n}{j} \StirlingFirst{j}{k} = \delta_{nk} \label{eq:StirlingOrthogonality}
\end{align}
for $n,k\in\mathbb Z_{\geq 0}$. Note, that the summation in \cref{eq:StirlingOrthogonality} effectively runs only over $j=k,\ldots,n$ as all the other summands will vanish. 

On the level of the exponential generating functions $A(x) = \sum_{n\geq 0} a_n \frac{x^n}{n!}$ and $B(x) = \sum_{n\geq 0} b_n \frac{x^n}{n!}$, the Stirling transform $\mathscr S: A(x) \mapsto B(x)$ can be written as \cite{BernsteinCanonicalSequencesIntegers2002}
\begin{align}
	B(x) = A(e^x -1) \point \label{eq:StirlingTransformGeneratingFunctions}
\end{align}
Hence, we can efficiently calculate nested series of the form
\begin{align}
	A(x) = \sum_{n\geq k \geq 0} \frac{(-1)^{n-k} x^n}{n!} \StirlingFirst{n}{k} b_k \label{eq:nestedSumStirlingTransform}
\end{align}
by means of the one-fold series $B(x) = \sum_{n\geq 0} \frac{x^n}{n!} b_n$, which we want to illustrate with a short example.
\begin{example}[Series evaluation by Stirling transformation]
	By the Stirling transformation we could simply evaluate the nested series
	\begin{align}
		\sum_{n=0}^\infty \sum_{k=0}^n k \frac{(-x)^n (-y)^k}{n!} \StirlingFirst{n}{k} = \sum_{n=0}^\infty n \frac{\ln(1+x)^n y^n}{n!} = (1+x)^y y \ln(1+x)
	\end{align}
	where we used $b_k = k \, y^k$ in \cref{eq:nestedSumStirlingTransform} and inverted equation \cref{eq:StirlingTransformGeneratingFunctions}.
\end{example}
Moreover, the Stirling transformation has the remarkable property to connect finite differences with differentiations \cite[sec. 26.8]{OlverNISTHandbookMathematical2010}
\begin{align}
	\frac{1}{k!} \od[k]{}{x} f(x) &= \sum_{n\geq 1} \frac{(-1)^{n-k}}{n!} \StirlingFirst{n}{k} \Delta^n f(x) \label{eq:StirlingDifferences1} \\
	\frac{1}{k!} \Delta^k f(x) &= \sum_{n\geq 1} \frac{1}{n!} \StirlingSecond{n}{k} \od[n]{}{x} f(x) \point	
\end{align}
We want to demonstrate the advantage of these relations in a further small example.
\begin{example}[Series evaluation by finite differences]
	The $n$-th forward difference can be written as $\Delta^n f(x) = \sum_{k=0}^n \binom{n}{k} (-1)^{n-k} f(x+k)$. Hence, we can evaluate the following series
	\begin{align}
		\sum_{n=1}^\infty \sum_{k=0}^n \frac{(-1)^{k-p}}{k!\, (n-k)!\, (q+k)!}  \StirlingFirst{n}{p} = \sum_{n=1}^\infty \frac{(-1)^{n-p}}{n!} \StirlingFirst{n}{p} \Delta^n \frac{1}{\Gamma(q+1)} = \frac{1}{p!} \od[p]{}{q} \frac{1}{\Gamma(q+1)}
	\end{align}
	by the aid of the identity \cref{eq:StirlingDifferences1}.
\end{example}

Another useful transformation between sequences we want to mention here, is the \textit{binomial transformation}, which is closely related to the Stirling transformation. This transformation is defined by \cite{BernsteinCanonicalSequencesIntegers2002}
\begin{align}
	b_n = \sum_{k=0}^n \Binomial{n}{k} a_k \quad\text{,}\qquad a_n = \sum_{k=0}^n (-1)^{n-k} \Binomial{n}{k} b_k \point
\end{align}
The exponential generating functions of these two sequences $\{a_i\}_{i\in\mathbb N}$ and $\{b_i\}_{i\in\mathbb N}$ are connected by $B(x) = e^x A(x)$. Note, that Stirling numbers of the first $\StirlingFirstSmall{-n}{k}$ kind with negative first argument, can be expressed by a binomial transform by means of \cref{eq:StirlingFirstNegativeDefAppendix}. Using those transformations in a clever way can dramatically simplify the evaluation of nested series, as they appear in the approach of series representations from \cref{ch:seriesRepresentations}. However, an algorithmic procedure to use those transformations systematically in the evaluation and simplification of those series has not yet been developed.\bigskip

We will conclude this small overview about Stirling numbers with certain useful series including Stirling numbers. A very elementary type of Stirling series was considered in \cite{AdamchikStirlingNumbersEuler1997} and can be written directly by the approach of $Z$-sums into a multiple $\zeta$-value 
\begin{align}
	\sum_{n\geq 1} \frac{1}{n!\, n^q} \StirlingFirst{n}{k} = \Zsum{\infty}{\underbrace{\scriptstyle 1,\ldots,1}_{k-1},q+1} = \zeta(\underbrace{1,\ldots,1}_{k-1}, q+1) \point \label{eq:StirlingSeries1}
\end{align}
The study of those series is closely related to the Nielsen polylogarithm, which can be written as \cite{AdamchikStirlingNumbersEuler1997}
\begin{align}
	S_{n,p}(z) = \sum_{k\geq 1} \frac{z^k}{k!\, k^n} \StirlingFirst{k}{p} \point
\end{align}
A generalization of the series \cref{eq:StirlingSeries1} was derived in \cite[thm. 1]{KubaNoteStirlingSeries2010}
\begin{align}
	\sum_{n\geq 1} \frac{p!}{(p+n)!} \frac{1}{n^q} \StirlingFirst{n}{k} = (-1)^q \Ssum{p}{k+1,\underbrace{\scriptstyle 1,\ldots,1}_{q-1}} + \sum_{i=2}^{q+1} (-1)^{q+1-i} \Zsum{\infty}{\underbrace{\scriptstyle 1,\ldots,1}_{k-1}, i} \Ssum{p}{\underbrace{\scriptstyle 1,\ldots,1}_{q+1-i}} \label{eq:StirlingSeries2}
\end{align}
which can be written as a finite sum of multiple $\zeta$-values and harmonic sums, and we will assume $k\in\mathbb Z_{>0}$ and $p,q\in\mathbb Z_{\geq 0}$. By means of \cite{BorweinEvaluationsKfoldEuler1996} those multiple $\zeta$-values can also be written in terms of single $\zeta$-values. This result \cref{eq:StirlingSeries2} can be generalized further \cite[thm. 2]{KubaNoteStirlingSeries2010}. We also refer to \cite[table 265]{GrahamConcreteMathematicsFoundation1994} and \cite{SrivastavaZetaQZetaFunctions2012}, which contain a wide range of useful relations for sums containing Stirling numbers.


\section{Feynman's trick}

In \cref{sec:ParametricFeynmanIntegrals} we derived the parametric Feynman integral \cref{eq:FeynmanParSpFeynman} from the momentum space Feynman integral \cref{eq:FeynmanMomSp} via the diversion of Schwinger's representation (\ref{thm:SchwingerRepresentation}). Alternatively, one can show the parametric Feynman integral \cref{eq:FeynmanParSpFeynman} also directly by means of the so-called Feynman trick. This integral relation goes back to an identity in \cite{FeynmanSpaceTimeApproachQuantum1949} and can be found in most textbooks e.g.\ \cite{PeskinIntroductionQuantumField1995, SchwartzQuantumFieldTheory2014}. In order to provide an advantage over the classical derivations of these textbooks, we want to include the freedom to choose a hyperplane $H(x)$, which is sometimes referred as Cheng-Wu theorem.

\begin{lemma}[Feynman's trick] \label{lem:FeynmanTrick}
	Let $D_1,\ldots,D_n$ be positive real numbers $D_i>0$ and $\nu\in\mathbb C^n$ with $\Re(\nu_i)>0$. Then we have the following identity
	\begin{align}
		\frac{1}{\prod_{i=1}^n D_i^{\nu_i}} = \frac{\Gamma(|\nu|)}{\Gamma(\nu)} \int_{\mathbb R^n_+} \dif x\, x^{\nu-1} \frac{\delta(1-H(x))}{\left(\sum_{i=1}^n x_i D_i \right)^{|\nu|}} \comma
	\end{align}
	where $|\nu|:=\sum_{i=1}^n$, $H(x) = \sum_{i=1}^n h_i x_i$ defines a hyperplane with $h_i\geq 0$ not all zero and $\delta(x)$ denotes the $\delta$-distribution. As before we use a multi-index notation which shortens $\Gamma(\nu) := \prod_{i=1}^n \Gamma(\nu_i)$ and $\dif x \, x^{\nu-1} := \prod_{i=1}^n \dif x_i \, x_i^{\nu_i-1}$.
\end{lemma}
\begin{proof}
	We will show this integral relation by an induction over $n$. For $n=1$ the lemma is trivial. For $n=2$ we will assume without loss of generality that $h_2\neq 0$. Integrating the $\delta$-distribution, we therefore obtain
	\begin{align}
		\int_{\mathbb R^2_+} \dif x \, x^{\nu-1}  \frac{\delta(1-h_1 x_1 - h_2 x_2)}{(x_1 D_1+x_2 D_2)^{\nu_1+\nu_2}} = \int_0^\infty \dif x_1 \, \theta(1-h_1 x_1) \frac{ x_1^{\nu_1-1} (1-h_1 x_1)^{\nu_2-1} h_2^{\nu_1} }{[D_2+ (h_2 D_1-h_1 D_2)x_1]^{\nu_1+\nu_2}} \ \text{.}
	\end{align}
	In case of $h_1=0$ the integral can be evaluated after a substitution easily by beta function, and we obtain the asserted result. In case of $h_1\neq 0$ we will substitute $t = h_1 x_1$ and obtain
	\begin{align}
		& \left(\frac{h_2}{h_1}\right)^{\nu_1} D_2^{-\nu_1-\nu_2} \int_0^1 \dif t \,  t^{\nu_1-1} (1-t)^{\nu_2-1} \left(1 + \frac{h_2 D_1 -  h_1 D_2}{D_2 h_1} t \right)^{-\nu_1-\nu_2} \nonumber \\
		& = \frac{\Gamma(\nu_1+\nu_2)}{\Gamma(\nu_1)\Gamma(\nu_2)} \left(\frac{h_2}{h_1}\right)^{\nu_1} D_2^{-\nu_1-\nu_2} {_2}F_1 \left(\nu_1+\nu_2,\nu_1,\nu_1+\nu_2 \Big\rvert 1 - \frac{D_1 h_2}{D_2 h_1} \right) \nonumber \\
		&= \frac{\Gamma(\nu_1+\nu_2)}{\Gamma(\nu_1)\Gamma(\nu_2)} D_1^{-\nu_1} D_2^{-\nu_2}
	\end{align}
	where we use the integral representation of Gauss' hypergeometric function \cref{eq:GaussEuler} and simplify the present case by the series representation \cref{eq:GaussSeries} which reduces to the binomial series.
	
	For $n>2$ we show the lemma by induction. We will write $\nu = (\nu_1,\ldots,\nu_n)$, $x=(x_1,\ldots,x_n)$ and  $\tilde\nu = (\nu_1,\ldots,\nu_{n+1})$, $\tilde x=(x_1,\ldots,x_{n+1})$. By using the induction hypothesis and the $n=2$ case we arrive at
	\begin{align}
		\prod_{i=1}^{n+1} D_i^{-\nu_i} &= \frac{\Gamma(|\nu|)}{\Gamma(\nu)} \int_{\mathbb R^n_+} \dif x \, x^{\nu-1} \frac{\delta(1-H(x))}{\left(\sum_{i=1}^n x_i D_i \right)^{|\nu|} D_{n+1}^{\nu_n+1}} \nonumber \\
		&= \frac{\Gamma(|\tilde\nu|)}{\Gamma(\tilde\nu)} \int_{\mathbb R^{n+1}_+} \dif \tilde x \, \tilde x^{\tilde \nu-1} \int_0^\infty \dif y \, y^{|\nu|-1} \frac{\delta(1-H(x)) \delta(1-y-h_{n+1}x_{n+1})}{\left(y \sum_{i=1}^n x_i D_i + x_{n+1} D_{n+1} \right)^{|\tilde\nu|}} \point 
	\end{align}
	Performing the $y$-integration by means of the second $\delta$-distribution and substituting $t_i = (1-h_{n+1} x_{n+1} ) x_i$ for $i=1,\ldots,n$ and $t_{n+1} = x_{n+1}$ this results in
	\begin{align}
		\prod_{i=1}^{n+1} D_i^{-\nu_i} &= \frac{\Gamma(|\tilde\nu|)}{\Gamma(\tilde\nu)} \int_{\mathbb R^{n+1}_+} \dif t \, t^{\tilde \nu-1} \frac{\delta\!\left(1- \sum_{i=1}^n \frac{t_i h_i}{1-h_{n+1} t_{n+1}}\right)\theta(1-h_{n+1}t_{n+1})}{(1-h_{n+1}t_{n+1}) \left(\sum_{i=1}^{n+1} t_i D_i\right)^{|\tilde \nu|}} \nonumber\\
		& = \frac{\Gamma(|\tilde\nu|)}{\Gamma(\tilde\nu)} \int_{\mathbb R^{n+1}_+} \dif t \, t^{\tilde \nu-1} \frac{\delta\!\left(1 - \sum_{i=1}^{n+1} h_i t_i\right)}{\left(\sum_{i=1}^{n+1} t_i D_i\right)^{|\tilde \nu|}} \point 
	\end{align}
\end{proof}

Thus, when applying \cref{lem:FeynmanTrick} to the momentum space representation \cref{eq:FeynmanMomSp}, we will obtain \cref{eq:FeynmanParSpFeynman} directly by similar steps as in the proof of \cref{thm:SchwingerRepresentation}. We omit the explicit proof here to avoid redundancy.


\section{Software tools} \label{sec:SoftwareTools}
\VerbatimFootnotes 


In this section we would like to illustrate the usage of various software tools related to polytopes and $D$-modules. In the following we will present computations with the ANSI C package \softwareName{lrslib} \cite{AvisLrslib}, the C++ package \softwareName{Topcom} \cite{RambauTOPCOMTriangulationsPoint2002}, the software \softwareName{Polymake} \cite{GawrilowPolymakeFrameworkAnalyzing2000, AssarfComputingConvexHulls2015}, the software system \softwareName{Macaulay2} \cite{GraysonMacaulay2SoftwareSystem} as well as with the computer algebra system \softwareName{Singular} \cite{DeckerSingular421Computer2021, GreuelSINGULARComputerAlgebra2009}. However, there are many further programs allowing computations with polyhedra and $D$-modules, e.g.\ \softwareName{Gfan} \cite{JensenGfanSoftwareSystem} or the \softwareName{Parma Polyhedra Library} \cite{BagnaraParmaPolyhedraLibrary2006}. \bigskip

For illustration reasons we go through all those software packages by treating the $1$-loop bubble graph as a running example, i.e.\ we consider
\begin{gather}
	\Uu = x_1 + x_2 \qquad \Ff = \left(p^2 + m_1^2 + m_2^2\right) x_1 x_2 + m_1^2 x_1^2 + m_2^2 x_2^2 \\
	\Aa =   \begin{pmatrix} 
                1 & 1 & 1 & 1 & 1\\
	      	    1 & 0 & 1 & 2 & 0\\
	      	    0 & 1 & 1 & 0 & 2
	        \end{pmatrix} \point \label{eq:SoftwareAppendixExampleBubble}
\end{gather}
In the following subsections we want to show how to determine the representation of $\Conv(\Aa)$ by means of intersections of halfspaces in the sense of \cref{eq:HPolytope}, how to calculate its volume $\vol(\Conv(\Aa))$ and generate its regular triangulations as well as to perform basic calculations with $D$-modules. Further, we will also see how to determine the truncated polynomials and calculate the corresponding $A$-discriminants.

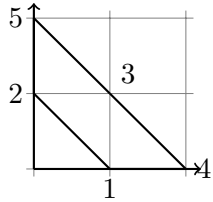
\begin{figure}
    \centering
    \begin{tikzpicture}[thick, dot/.style = {draw, shape=circle, fill=black, scale=.5}, scale=1]
        \draw[step=1cm,gray,very thin] (-0.1,-0.1) grid (2.1,2.1);  
        \draw[thick,->] (0,0) -- (2.2,0) node[anchor=north west] {};
        \draw[thick,->] (0,0) -- (0,2.2) node[anchor=south east] {};
        \coordinate[label=below:$1$] (A) at (1,0);
        \coordinate[label=left:$2$] (B) at (0,1);
        \coordinate[label=above right:$3$] (C) at (1,1);
        \coordinate[label=right:$4$] (D) at (2,0);
        \coordinate[label=left:$5$] (E) at (0,2);
        \draw (A) -- (B);
        \draw (B) -- (E);
        \draw (E) -- (C);
        \draw (C) -- (D);
        \draw (D) -- (A);
    \end{tikzpicture}
    \caption[Newton polytope for the fully massive $1$-loop self-energy graph]{Newton polytope $\Newt(\Gg) = \Conv(A)$ for the fully massive $1$-loop self-energy graph according to \cref{eq:SoftwareAppendixExampleBubble}. The numbering of vertices corresponds to the order of columns of $\Aa$.} \label{fig:BubbleNewtonPolytopeAppendix}
\end{figure}

\subsection{\softwareName{lrslib}} \label{ssec:lrslib}

The lean C library \softwareName{lrslib} \cite{AvisLrslib} provides algorithms to transform efficiently between the two representations of polyhedra, i.e.\ the vertex-based representation \cref{eq:VPolytope} and the halfspace-based representation \cref{eq:HPolytope}. Further we can also calculate the volume by means of \softwareName{lrslib}. This library is also contained in Debian/Ubuntu distributions as \texttt{lrslib}, which allows an easy installation. The \softwareName{lrslib} package works with input and output files in a specific format. In our example \cref{eq:SoftwareAppendixExampleBubble} the input file has to look like:
\begin{lstlisting}[title={\raggedright the file \texttt{bubble.ine}}, frame=single]
bubble.ine
V-representation
begin
5 3 rational
1 1 0
1 0 1
1 1 1
1 2 0
1 0 2
end
\end{lstlisting}
Thus, the format type is very elementary and contains (in that order) the filename, the type of representation (\texttt{V-representation} or \texttt{H-representation}), the size of $\Aa$ as well as the points in homogeneous coordinates defined by $\Aa$. Additional format specifications can be found in the manual \cite{AvisLrslib}. The conversion to the H-representation can be accomplished by running \texttt{lrs bubble.ine}, which generates in our example:
\begin{lstlisting}
$ lrs bubble.ine 
*Input taken from  bubble.ine
bubble.ine
H-representation
begin
***** 3 rational
 0  1  0 
 2 -1 -1 
-1  1  1 
 0  0  1 
end
*Totals: facets=4 bases=2
\end{lstlisting}
or can be alternatively stored in an output file by running \texttt{lrs [input file] [output file]}. Hence, the polyhedra $\Conv(\Aa)$ will be described by the inequalities
\begin{align}
	\nu_1 &\geq 0 \label{eq:lrslibFacets1}  \\ 
	2\nu_0 - \nu_1 - \nu_2 &\geq 0 \\
	-\nu_0 + \nu_1 + \nu_2 &\geq 0 \\
	\nu_2 &\geq 0 \point \label{eq:lrslibFacets4}
\end{align}
The polytope $\Conv(A)$ can be obtained by setting $\nu_0=1$ in \cref{eq:lrslibFacets1} -- \cref{eq:lrslibFacets4}. By adding the keyword \texttt{volume} in the end of an input file \texttt{bubble.ine} with a polytope in V-representation, we get also the Euclidean volume of the polytope. Thus, the input file
\begin{lstlisting}[title={\raggedright the file \texttt{bubble.ine}}, frame=single]
bubble.ine
V-representation
begin
5 3 rational
1 1 0
1 0 1
1 1 1
1 2 0
1 0 2
end
volume
\end{lstlisting}
will give additionally the Euclidean volume $\frac{3}{2}$. Note, that the volume from \cref{ssec:PolytopesPointConfigs} was a normalized Euclidean volume counting the number of standard simplices in $\Conv(A)$, i.e.\ we have to multiply by $n!$. Hence, the volume in this particular example equals $3$, which agrees with \cref{fig:BubbleNewtonPolytopeAppendix}.

\subsection{\softwareName{Topcom}} \label{ssec:topcom}

A powerful tool for the determination of triangulations is the C++ package \softwareName{Topcom} \cite{RambauTOPCOMTriangulationsPoint2002}. In general, it is a very simple task to construct a single triangulation (e.g.\ from a height vector or as a placing triangulation, see \cref{ssec:TriangulationsPolyhedra}). However, it is a hard problem to find all triangulations. \softwareName{Topcom} is a software tool which is made to manage this task. To use \softwareName{Topcom} we have to store the vector configuration $\Aa$ in a separate file, which for our example \cref{eq:SoftwareAppendixExampleBubble} looks as follows
\begin{lstlisting}[title={\raggedright the file \texttt{bubble.dat}}, frame=single]
[[1,1,0],[1,0,1],[1,1,1],[1,2,0],[1,0,2]]
\end{lstlisting}
As before, we have to use homogeneous coordinates. To find the triangulations of $\Aa$ one has to run \texttt{points2triangs < bubble.dat}, which results in
\begin{lstlisting}
$ points2triangs < bubble.dat
T[0]:=[0->5,3:{{0,1,2},{0,2,3},{1,2,4}}];
T[1]:=[1->5,3:{{1,2,4},{1,2,3},{0,1,3}}];
T[2]:=[2->5,3:{{0,2,3},{0,2,4},{0,1,4}}];
T[3]:=[3->5,3:{{0,1,3},{1,3,4}}];
T[4]:=[4->5,3:{{0,1,4},{0,3,4}}];
\end{lstlisting}
Hence, we have $5$ triangulations for our running example, where the triangulations are identified by the vertices generating the maximal simplices. Similar to most of the software packages in this area, \softwareName{Topcom} starts the numbering of vertices with $0$. Note, that the command \texttt{points2triangs} produces all triangulations, which are connected by flips starting with an arbitrary regular triangulation. Thus, this command will construct all regular triangulations, but potentially also non-regular triangulations. One can specify the behaviour of \softwareName{Topcom} by the options \texttt{--regular} and \texttt{--nonregular} to consider only regular and only non-regular triangulations\footnote{Note that the graph of flips between regular triangulations is connected, whereas the graph of flips between non-regular triangulations is not necessarily connected \cite[ch. 5]{DeLoeraTriangulations2010}. Hence, we may find not all non-regular triangulations by \texttt{points2triangs --nonregular}. However, one can use the alternative command \texttt{points2alltriangs --nonregular}.}, respectively. 

To find unimodular triangulations one can restrict to a certain number $k$ of simplices generating the triangulation by the option \texttt{--cardinality [k]}. The option \texttt{--heights} will print also heights $\omega$ which generate the triangulations. For further commands and options we refer to the manual of \softwareName{Topcom}. There exists also a parallelized version of \softwareName{Topcom} \cite{JordanParallelEnumerationTriangulations2018}, which is convenient for more elaborate examples.

\subsection{\softwareName{Polymake}} \label{ssec:Polymake}

A very comprehensive software for calculations with polyhedra is \softwareName{Polymake} \cite{GawrilowPolymakeFrameworkAnalyzing2000, AssarfComputingConvexHulls2015}. As with \softwareName{lrslib} we can transform between V- and H-polyhedra. Furthermore, \softwareName{Polymake} also allows to determine regular triangulations from given heights. \softwareName{Polymake} consists in different parts. Thus, when starting \softwareName{Polymake} one has to make sure to use the polytope part by typing \texttt{application 'polytope';}. For our running example \cref{eq:SoftwareAppendixExampleBubble} we can determine the H-representation by
\begin{lstlisting}
> $P = new Polytope(POINTS=>[[1,1,0],[1,0,1],[1,1,1],[1,2,0],[1,0,2]]); 
> print $P->FACETS; 
0 1 0
-1 1 1
0 0 1
2 -1 -1
\end{lstlisting}
The output should be understood in the same way as for \softwareName{lrslib} (\cref{ssec:lrslib}), i.e.\ the polytope is also characterized by the inequalities \cref{eq:lrslibFacets1} -- \cref{eq:lrslibFacets4}. Also the conversion in the other direction works fine:
\begin{lstlisting}
> $P1 = new Polytope(INEQUALITIES=>[[0,1,0],[-1,1,1],[0,0,1],[2,-1,-1]]);
> print $P1->VERTICES;
1 2 0
1 0 2
1 0 1
1 1 0
\end{lstlisting}
Further, we can determine the Euclidean volume of polytopes by \softwareName{Polymake}:
\begin{lstlisting}
> print $P->VOLUME;
3/2
\end{lstlisting}
As before, we have to multiply the Euclidean volume by $n!$ to obtain the normalized volume defined in \cref{ssec:PolytopesPointConfigs}. Moreover, we can construct regular triangulations from a given height, e.g.\ $\omega=(0,1,2,1,2)$
\begin{lstlisting}
> $M = new Matrix<Rational>([[1,1,0],[1,0,1],[1,1,1],[1,2,0],[1,0,2]]);
> $w = new Vector<Rational>([0,1,2,1,2]);
> $S = new fan::SubdivisionOfPoints(POINTS=>$M, WEIGHTS=>$w);
> print $S->MAXIMAL_CELLS;
{0 1 4}
{0 3 4}
\end{lstlisting}
where we generated the non-unimodular triangulation \texttt{T[4]} from \cref{ssec:topcom}. Note, that a regular subdivision is not for every choice of $\omega$ a triangulation. However, it is a triangulation for generic $\omega$, see also \cref{ssec:TriangulationsPolyhedra}.

\subsection{\softwareName{Macaulay2}} \label{ssec:Macaulay2}

\softwareName{Macaulay2} \cite{GraysonMacaulay2SoftwareSystem} is a very powerful computer algebra system which is devoted to calculations in algebraic geometry. Thus, there is also a package for \softwareName{Macaulay2} to work with polyhedra and as before we can transform between H- and V-representations and calculate volumes\footnote{In principle \softwareName{Macaulay2} can also calculate regular triangulations. However, it seems to subdivide rather the polytope instead of the point configuration. This will cause problems, if we are taking also points into account, which are not vertices of the polytope. This can be seen by extending the following \softwareName{Macaulay2}-session by \verb|i11 : w = matrix {{0,1,2,1,2}}; regularSubdivision(A,w)| resulting in \verb|o12 = {{0, 1, 3}, {0, 2, 3}}}| and comparing the result with the previous section (\cref{ssec:Polymake}).}. 
We will illustrate these operations with our example \cref{eq:SoftwareAppendixExampleBubble}. Below we will print a \softwareName{Macaulay2}-session determine the basic properties for our example \cref{eq:SoftwareAppendixExampleBubble}. Comments in \softwareName{Macaulay2} start with a double dash \texttt{--}. 


\begin{lstlisting}[escapeinside={(*}{*)}]
$ M2 --no-preload
Macaulay2, version 1.18
i1 : loadPackage "Polyhedra";
i2 : QQ[s,b1,b2][x1,x2]; -- (*defines the base ring $\mathbb Q[x_1,x_2]$ with parameters $s, b_1, b_2$*)
i3 : U = x1+x2; -- (* 1. Symanzik polynomial*)
i4 : F = s*x1*x2 + b1*x1^2 + b2*x2^2; -- (* 2. Symanzik polynomial*)
i5 : P = newtonPolytope(U+F); -- (*definition of $P$ as Newton polytope*)
i6 : A = matrix {{1,0,1,2,0},{0,1,1,0,2}}; 
i7 : P1 = convexHull A; -- (* alternative definition by vertices*)                                                
i8 : volume P -- (*Euclidean volume*)
     3
o8 = -
     2
i9 : isFullDimensional P
o9 = true
i10 : facets P -- (*gives the H-representation of $P$*)
o10 = (| -1 0  |, | 0  |)
      | 0  -1 |  | 0  |
      | -1 -1 |  | -1 |
      | 1  1  |  | 2  |
\end{lstlisting}

Besides of polyhedral computations, \softwareName{Macaulay2} is also very well suited for the handling of multivariate polynomials. Therefore, we also want to use it for the calculation of $A$-discriminants and related objects. This is implemented in the two additional libraries for classical (i.e.\ dense) discriminants and resultants \cite{StaglianoPackageComputationsClassical2018} and $A$-discriminants and $A$-resultants \cite{StaglianoPackageComputationsSparse2020}. For convenience reasons we will present below an elementary package \texttt{Landau.m2}, which is adjusted to the approach of Feynman integrals. To use this additional package it should be stored in the \softwareName{Macaulay2} path.

\begin{lstlisting}[title={\raggedright the package \texttt{Landau.m2}}, frame=single, basicstyle=\tiny\ttfamily, basewidth=0.5em, breaklines=true,breakatwhitespace=true,columns=flexible]
-- Landau - a small package to calculate Landau varieties, 
-- i.e. the singular locus of a GKZ system or the principal A-determinant, mostly fitted to Feynman integrals 
--
-- general instructions:
-- the packages Polyhedra, Resultants and "SparseResultants" have to be installed. If they are not yet installed, run e.g. installPackage "Polyhedra"
-- this file has to be stored in a Macaulay2 path. The Macaulay2 paths can be displayed by the command: path
-- this package can be used by the command: loadPackage "Landau"

newPackage(
        "Landau",
        Version => "2.1", 
        Date => "August 24, 2021",
        Authors => {{Name => "Rene P. Klausen", 
                  Email => "klausen@physik.hu-berlin.de"}},
        Headline => "Calculating Landau varieties by means of principal A-determinants",
        DebuggingMode => true
        )
        
export {"principalAdet","generalDiscriminant","allTruncs"}

needsPackage "Polyhedra"
needsPackage "Resultants"
needsPackage "SparseResultants"

-- general stuff
ListTimes = (L1,L2) -> apply(L1,L2, (i,j) -> i*j)

-- truncation of polynomials
poly2A = f -> transpose(matrix(exponents(f))) -- giving the exponents of a monomial list for a polynomial
ptsOfFace = (A,face) -> apply(splice({0..numgens(source(A))-1}),i -> if contains(face,convexHull(submatrix(A,{i})))==true then 1 else 0 )    
faceTruncation = (f,A,face) -> sum(ListTimes(terms(f),ptsOfFace(A,face)))
truncatedPolynomial = (f,k) -> (A := poly2A(f); P := convexHull(A); face := facesAsPolyhedra(k,P); for i from 0 to #face-1 list faceTruncation(f,A,face_i)) 
allTruncs = f -> (n := numgens(ring(f)); mingle delete({},for i from 0 to n list truncatedPolynomial(f,i)))

-- fit rings
usedVars = (f,R) -> (gR := gens(R); delete("del",for i from 0 to #gR-1 list if diff(gR_i,f)==0 then "del" else i ))
fitRing = f -> (R := ring(f); substitute(f,first(selectVariables(usedVars(f,R),R))))
factorOut = (f,var) -> if pseudoRemainder(f,var)==0 then f//var else f
completeFactorOut = (f,varList) -> ((for i from 0 to #varList-1 do f = factorOut(factorOut(f,varList_i),varList_i)); f) -- at most for quadratic expressions (as in Symanzik polynomials)
dehomogenize = f -> (if isHomogeneous(f) then (sub(f,last(gens(ring f))=>1) ) else f);

--principal A determinant
generalDiscriminant = f -> (try ( -- distinguish cases by the number of monomials
  m:= # terms f;
  if m == 1 then ( -- contains only one monomial
    print("vertex type"); 
    coeff := ((coefficients(f))_1)_(0,0);
    sub(coeff, coefficientRing ring f) )
  else (
    f = fitRing f;
    n:= numgens ring f;
    if m-1 <= n then ( 
      print("dense discriminant"); 
      f = fitRing dehomogenize f; 
      f = completeFactorOut(f, gens ring f); --try to factor out trivial parts
      f = fitRing dehomogenize f;
      denseDiscriminant f )
    else (
      print("sparse discriminant"); 
      sparseDiscriminant f) )
  ) else (print("NN"); "NN")
);

principalAdet = f -> apply(allTruncs(f),generalDiscriminant);


beginDocumentation()
document { 
        Key => Landau,
        Headline => "Calculating Landau varieties by means of principal A-determinants",
        EM "Landau", " is a basic package to calculate Landau singularities."
        }
document {
        Key => {allTruncs},
        Headline => "all truncated polynomials",
        Usage => "allTruncs(f)",
        Inputs => {"a polynomial f"},
        Outputs => {"a list of all truncated polynomials, ordered by their codimension"},
        EXAMPLE lines ///
           QQ[m1,m2,s][x1,x2]; F = m1*x1^2 + m2*x2^2 + s*x1*x2;
           allTruncs F
        ///
        }
document {
        Key => {generalDiscriminant},
        Headline => "calculate the A-discriminant",
        Usage => "generalDiscriminant f",
        Inputs => {"f a polynomial"},
        Outputs => {"A-discriminant"},
        Caveat => {"Is mostly adjusted to polynomials which are at most quadratic in every variable. This behaviour can simply be generalized"},
        EXAMPLE lines ///
           QQ[m1,m2,s][x1,x2]; F = m1*x1^2 + m2*x2^2 + s*x1*x2;
           generalDiscriminant F
        ///
        }     
document {
        Key => {principalAdet},
        Headline => "calculate the simple principal A-determinant",
        Usage => "principalAdet f",
        Inputs => {"f a polynomial"},
        Outputs => {"A list of all A-discriminants of all truncated polynomials, the (simple) principal A-determinant is the product of all list elements. Additionally the used method for every A-discriminant is printed on screen."},
        Caveat => {"Is mostly adjusted to polynomials which are at most quadratically in every variable. This behaviour can simply generalized"},
        EXAMPLE lines ///
           QQ[m1,m2,s][x1,x2]; F = m1*x1^2 + m2*x2^2 + s*x1*x2;
           principalAdet F
        ///
        }             
end--
\end{lstlisting}

This package \texttt{Landau.m2} is written for convenience reasons and will preprocess the polynomials for the application of the algorithms from the packages \texttt{Resultants} and \texttt{SparseResultants}. It mainly contains the command \texttt{principalAdet}, which returns a list of $A$-discriminants, whose product is the simple principal $A$-determinant and the command \texttt{allTruncs}, which gives a list of all truncated polynomials. Further information is provided in the included documentation, which can be printed by means of the \texttt{help} command.

In our particular example \cref{eq:SoftwareAppendixExampleBubble} this would look like the following:
\begin{lstlisting}[basewidth=0.5em,basicstyle=\footnotesize\ttfamily,breaklines=false]
$M2 --no-preload
Macaulay2, version 1.18
i1 : loadPackage "Landau";
i2 : QQ[s,b1,b2][x1,x2]; 
i3 : U = x1+x2; F = s*x1*x2 + b1*x1^2 + b2*x2^2; G = U+F; E = principalAdet G
                                            2
o6 = {- s + b1 + b2, 1, 1, 1, b2, 1, 1, - s  + 4b1*b2, b1}
o6 : List
i7 : allTruncs(U+F)
           2                  2                 2                2            2     
o7 = {b1*x1  + s*x1*x2 + b2*x2  + x1 + x2, b2*x2  + x2, x2, b1*x1  + x1, b2*x2 , 
      --------------------------------------------------------------------------
                        2                  2       2
      x1 + x2, x1, b1*x1  + s*x1*x2 + b2*x2 , b1*x1 }
o7 : List
\end{lstlisting}
One can also use the original variables, which looks like
\begin{lstlisting}[basewidth=0.5em,basicstyle=\footnotesize\ttfamily,breaklines=false,escapeinside={(*}{*)}]
i8 : QQ[p,m1,m2][x1,x2]; 
i9 :  U = x1+x2; F = -p^2*x1*x2 + U*(x1*m1^2 + x2*m2^2); G = U+F; E = principalAdet G
        2             2           4     2  2     4     2  2      2  2     4    2
o12 = {p , 1, 1, 1, m2 , 1, 1, - p  + 2p m1  - m1  + 2p m2  + 2m1 m2  - m2 , m1 }
o12 : List
i13 : factor product E
          2    2   2
o13 = (m2) (m1) (p) (p - m1 - m2)(p - m1 + m2)(p + m1 - m2)(p + m1 + m2)(-1)
o13 : Expression of class Product
\end{lstlisting}
%

Furthermore, we will use \softwareName{Macaulay2} also for computations with $D$-modules. Below we will determine a Gale dual $\Bb$ of $\Aa$ (i.e.\ a basis of the kernel $\ker_{\mathbb Z}(\Aa)$) and generate the GKZ ideal $H_\Aa(\nuu)$. Note, that \softwareName{Macaulay2} uses the opposite sign convention for the GKZ parameter $\nuu$ as we did. Further, we will check holonomicity of the GKZ ideal $H_\Aa(\nuu)$ and determine its holonomic rank. We also check the toric ideal $I_\Aa$ for Cohen-Macaulayness, which is true in this example. Hence, we will have $\vol(\Conv(\Aa)) = \rank(H_\Aa(\nuu))$ for all values of $\nuu$. Moreover, we calculated the singular locus of the GKZ system and compared it with our previous result using \texttt{Landau.m2}.

\begin{lstlisting}[basewidth=0.5em,basicstyle=\footnotesize\ttfamily,breaklines=false,escapeinside={(*}{*)}]
i14 : loadPackage "Dmodules";
i15 : A = matrix {{1,1,1,1,1},{1,0,1,2,0},{0,1,1,0,2}};
i16 : kernel A -- calculate the kernel of A, i.e. a Gale dual (*$\Bb$*)
o16 = image | -1 2  |
            | 1  -2 |
            | 1  0  |
            | 0  -1 |
            | -1 1  |
                                5
o16 : ZZ-module, submodule of ZZ
i17 : I = gkz(A,{-2,-1,-1}) -- the GKZ ideal (*$H_\Aa(\nuu)$ with $\nuu =(2,1,1)$*)                                                                              
o17 = ideal (x D  + x D  + x D  + x D  + x D  + 2, x D  + x D  + 2x D  + 1, x D 
              1 1    2 2    3 3    4 4    5 5       1 1    3 3     4 4       2 2
     ---------------------------------------------------------------------------
                          2
     + x D  + 2x D  + 1, D  - D D , - D D  + D D , D D  - D D )
        3 3     5 5       3    4 5     2 3    1 5   1 3    2 4
o17 : Ideal of QQ[x ..x , D ..D ]
                   1   5   1   5

i18 : isHolonomic I  -- check holonomicity
o18 = true
i19 : holonomicRank I  -- calculate (*$\rank H_\Aa(\nuu)$*)
o19 = 3

i20 : loadPackage "Depth";  -- package to check Cohen-Macaulayness
i21 : R = QQ[d1,d2,d3,d4,d5]; -- commutative ring for (*$I_\Aa \subset \mathbb Q [\partial_1,\ldots,\partial_5]$*)
i22 : T = ideal(d1*d5-d2*d3, d1^2*d5 - d2^2*d4);  -- toric ideal (*$I_\Aa$*)
i23 : isCM(R/T)  -- ask for Cohen-Macaulayness of quotient ring
o23 = true

i24 : S = singLocus I  -- calculate (*$\Sing(H_\Aa(\nuu))$*)
             2 2 3          3 2 2      3   2   2     2 2   2 2       3 3 2     3   2 3
o24 = ideal(x x x x x  - x x x x x  - x x x x x  - 4x x x x x  + 4x x x x  + 4x x x x )
             1 2 3 4 5    1 2 3 4 5    1 2 3 4 5     1 2 3 4 5     1 2 4 5     1 2 4 5
o24 : Ideal of QQ[x ..x , D ..D ]
                   1   5   1   5
i25 : f = (first entries gens S)_0; v = gens ring f; factor sub(f,{v_0=>1,v_1=>1}) 
           -- compare with the previous result via (*$E_A(\Gg)$*) by specifying the coefficients
                              2
o27 = (x )(x )(x  - x  - x )(x  - 4x x )
        5   4   3    4    5   3     4 5
o27 : Expression of class Product
\end{lstlisting}

\subsection{\softwareName{Singular}}

The last computer algebra system we would like to present here is \softwareName{Singular} \cite{DeckerSingular421Computer2021, GreuelSINGULARComputerAlgebra2009}. Similar to \softwareName{Macaulay2}, \softwareName{Singular} provides also tools to determine the holonomic rank and the singular locus of the GKZ ideal $H_\Aa(\nuu)$, which we will demonstrate below.

\begin{lstlisting}[basewidth=0.5em,basicstyle=\footnotesize\ttfamily,breaklines=true,escapeinside={(*}{*)}]
> LIB "ncalg.lib";
> LIB "dmodloc.lib";
> intmat A[3][5]=
. 1,1,1,1,1,
. 1,0,1,2,0,
. 0,1,1,0,2;
> def D1 = GKZsystem(A,"lp","ect");
> setring D1;
> D1;
> print(GKZid);
x(1)*d(1)+x(2)*d(2)+x(3)*d(3)+x(4)*d(4)+x(5)*d(5)+(-b(1)),
x(1)*d(1)+x(3)*d(3)+2*x(4)*d(4)+(-b(2)),
x(2)*d(2)+x(3)*d(3)+2*x(5)*d(5)+(-b(3)),
d(3)^2-d(4)*d(5),
d(1)*d(5)-d(2)*d(3),
d(1)*d(3)-d(2)*d(4)
> holonomicRank(GKZid);
3
> DsingularLocus(GKZid);
_[1]=-x(1)^3*x(2)*x(3)^2*x(4)*x(5)^2+4*x(1)^3*x(2)*x(4)^2*x(5)^3+x(1)^2*x(2)^2*x(3)^3*x(4)*x(5)-4*x(1)^2*x(2)^2*x(3)*x(4)^2*x(5)^2-x(1)*x(2)^3*x(3)^2*x(4)^2*x(5)+4*x(1)*x(2)^3*x(4)^3*x(5)^2
\end{lstlisting}

Furthermore, we can also determine the secondary polytope $\Sigma(A)$ by \softwareName{Singular}. For this reason also \softwareName{Topcom} should be installed on the system as it is internally used for the determination of triangulations.
\begin{lstlisting}[basewidth=0.5em,basicstyle=\footnotesize\ttfamily,breaklines=true,escapeinside={(*}{*)}]
> LIB "polymake.lib";
> list A1 = intvec(1,0),intvec(0,1),intvec(1,1),intvec(2,0),intvec(0,2);
> list secpoly = secondaryPolytope(A1);
> print(secpoly[1]);
     2     2     3     1     1
     1     3     2     2     1
     3     1     2     1     2
     1     3     0     3     2
     3     1     0     2     3
> print(secpoly[2]);
[1]:
   [1]:
      1,2,3
   [2]:
      1,3,4
   [3]:
      2,3,5
[2]:
   [1]:
      2,3,5
   [2]:
      2,3,4
   [3]:
      1,2,4
[3]:
   [1]:
      1,3,4
   [2]:
      1,3,5
   [3]:
      1,2,5
[4]:
   [1]:
      1,2,4
   [2]:
      2,4,5
[5]:
   [1]:
      1,2,5
   [2]:
      1,4,5
\end{lstlisting}
These lines give the vertices of $\Sigma(A)$, i.e.\ the weights of triangulations $\varphi_\Tt(A)$ by \texttt{print(secpoly[1]);}. Thus, the secondary polytope in this example \cref{eq:SoftwareAppendixExampleBubble} is the $2$-dimensional polytope generated by
\begin{align}
	\Sigma(A) = \Conv \left\{ %
	   \begin{pmatrix}  2 \\ 2 \\ 3 \\ 1 \\ 1 \end{pmatrix}, %
	   \begin{pmatrix}  1 \\ 3 \\ 2 \\ 2 \\ 1 \end{pmatrix}, %
	   \begin{pmatrix}  3 \\ 1 \\ 2 \\ 1 \\ 2 \end{pmatrix}, %
	   \begin{pmatrix}  1 \\ 3 \\ 0 \\ 3 \\ 2 \end{pmatrix}, %
	   \begin{pmatrix}  3 \\ 1 \\ 0 \\ 2 \\ 3 \end{pmatrix}\right\} \subset\mathbb R^5 \point
\end{align}
The particular triangulations can be displayed by \texttt{print(secpoly[2]);}.


\section{Characteristics of specific Feynman graphs} \label{sec:AppendixCharacteristics}

In order to classify different Feynman graphs with regard to their complexity from a hypergeometric perspective, we collate certain characterizing numbers for several standard Feynman graphs in \cref{tab:characteristics}. According to \cref{thm:FeynSeries} a series representation of a Feynman integral consists in a linear combination of $\vol(\Newt(\Gg))$ multivariate series in $r:=N-n-1$ variables, each one having depth $r$. Every summand of those series contains a product of $(n+1)$ $\Gamma$-functions. Furthermore, every regular triangulation $\Tt$ provides a way to write the Feynman integral in terms of Horn hypergeometric functions. Hence, we will have $\mathfrak T_\Gg$ ways to write series representations for a given generalized Feynman integral, where $\mathfrak T_\Gg$ denotes the number of regular triangulations of $\Newt(\Gg)$. As aforementioned, unimodular regular triangulations behave slightly simpler. Thus, we listed also the number of unimodular regular triangulations $\mathfrak T_\Gg^{\text{unimod}}$. Non-regular triangulations $\mathfrak T_\Gg^{\text{non-reg}}$ are a relatively rare phenomenon. \bigskip

The geometry of the Newton polytopes of Symanzik polynomials will also have an influence to the kinematic singularities. In particular, the Landau variety $\mathcal L_1 (\mathcal I_\Gamma)$ consists in $|\Newt(\Ff)|$ irreducible components, where $|\Newt(\Ff)|$ stands for the number of faces of $\Newt(\Ff)$. The full singular locus (including second-type singularities and singularities of proper mixed faces) $\Sing(H_\Aa(\nuu)) = \Var (E_A(\Gg))$ decomposes into $|\Newt(\Gg)|$ components. Moreover, $\mathfrak T_\Ff$ will give a lower bound of the number of monomials in the defining polynomial of the Landau variety $\mathcal L_1 (\mathcal I_\Gamma)$ and $\mathfrak T_\Gg$ is a lower bound of the number of monomials in the defining polynomial of the full singular locus. As one can observe, this lower bound of the number of monomials grows very fast. This shows that the determination of the Landau variety by means of a defining polynomial is an almost hopeless endeavour for more complex graphs. \bigskip

However, we can also notice, that certain Feynman graphs results in the same Newton polytopes. Hence, also their Feynman integrals are equal when replacing their variables in a specific way. For example, the Symanzik polynomials of the ``dunce's cap'' graph with massive edges $e_1$, $e_3$ and of the ``flying saucer'' graph with massive edges $e_1$, $e_2$ have the same support. Also, relations between different polynomials exist, e.g.\ the Newton polytope of the second Symanzik polynomial $\Ff$ for the ``kite'' graph with massive edge $e_1$ coincides with the Newton polytope of the Lee-Pomeransky polynomial $\Gg$ of the ``dunce's cap'' graph with massive edge $e_3$ as well as with the Newton polytope of the Lee-Pomeransky polynomial $\Gg$ of the ``flying saucer'' graph with massive edge $e_2$.\bigskip

Therefore, by the following numbers one can estimate the complexity for calculations from the hypergeometric perspective.

\clearpage
\newgeometry{inner=17.4mm, top=27mm, outer=35.0mm, bottom=50mm, marginparsep=0mm, bindingoffset=10mm, heightrounded} 
\begin{landscape}
    {\small
    \begin{longtable}[l]{c|>{\centering\arraybackslash}p{2cm}|cccc| >{$}c<{$} >{$}c<{$} >{\centering\arraybackslash $}p{3cm}<{$} >{$}c<{$} >{$}c<{$}}
        \topologyDescr{2}{topology}{L}{n}{m} & \tmultirow{masses} & %
        \SymanzikSpec{N_\Ff}{N}{N^{ph}}{r} 
        \vol (\Newt(\Ff)) & |\Newt(\Ff)| & \mathfrak T_\Ff & \mathfrak T_\Ff^{\text{unimod}}  & \mathfrak T_\Ff^{\text{non-reg}} \\
        & & & & & & \vol(\Newt(\Gg)) & |\Newt(\Gg)| &  \mathfrak T_\Gg & \mathfrak T_\Gg^{\text{unimod}}  & \mathfrak T_\Gg^{\text{non-reg}} \\
        \hhline{=|=|====|=====}
        \endhead
        \topologyDescrTikz{6}{bubble}{1}{2}{2}{
                \coordinate[dot] (A) at (0,0);
                \coordinate[dot] (B) at (2,0);
                \draw (1,0) circle (1);
                \draw (A) -- ++(-0.7,0);
                \draw (B) -- ++(0.7,0);
                \node at (1,1.3) {$1$};
                \node at (1,-1.3) {$2$};}{0.5} & %
        \tmultirow{--} & %
            \SymanzikSpec{1}{3}{1}{0}
            \triangsF{1}{1}{1}{1}{0}{}
            \triangsG{1}{7}{1}{1}{0}
        \hhline{~|-|----|-----}
        & \mmultirow{1} & %
            \SymanzikSpec{2}{4}{2}{1}
            \triangsF{1}{3}{1}{1}{0}{}
            \triangsG{2}{9}{2}{2}{0}
        \hhline{~|-|----|-----}
        & \mmultirow{1,2} & %
            \SymanzikSpec{3}{5}{3}{2}
            \triangsF{2}{3}{2}{1}{0}{}
            \triangsG{3}{9}{5}{3}{0}
        \hhline{-|-|----|-----}
        \topologyDescrTikz{8}{sunset}{2}{3}{2}{
                \coordinate[dot] (A) at (0,0);
                \coordinate[dot] (B) at (2,0);
                \draw (A) -- node[above] {$2$} (B);
                \draw (1,0) circle (1);
                \draw (A) -- ++(-0.7,0);
                \draw (B) -- ++(0.7,0); 
                \node at (1,1.3) {$1$};
                \node at (1,-1.3) {$3$};}{0.7} & %
        \tmultirow{--} & %
            \SymanzikSpec{1}{4}{1}{0}
            \triangsF{1}{1}{1}{1}{0}{}
            \triangsG{1}{15}{1}{1}{0}
        \hhline{~|-|----|-----}
        & \mmultirow{1} & %
            \SymanzikSpec{3}{6}{2}{2}
            \triangsF{1}{7}{1}{1}{0}{}
            \triangsG{3}{21}{6}{6}{0}
        \hhline{~|-|----|-----}
        & \mmultirow{1,2} & %
            \SymanzikSpec{5}{8}{3}{4}
            \triangsF{3}{9}{5}{3}{0}{}
            \triangsG{6}{25}{68}{44}{0}
        \hhline{~|-|----|-----}
        & \mmultirow{1,2,3} & %
            \SymanzikSpec{7}{10}{4}{6}
            \triangsF{6}{13}{32}{18}{0}{}
            \triangsG{10}{33}{826}{466}{0}
        \hhline{-|-|----|-----}
        \topologyDescrTikz{10}{banana}{3}{4}{2}{
                 \coordinate[dot] (A) at (0,0);
                \coordinate[dot] (B) at (2,0);
                \draw (1,0) circle (1);
                \draw (B) arc[start angle=35, end angle=145, radius=1.220775];
                \draw (A) arc[start angle=215, end angle=325, radius=1.220775];
                \draw (A) -- ++(-0.7,0);
                \draw (B) -- ++(0.7,0); 
                \node at (1,1.3) {$1$};
                \node at (1,.3) {$2$};
                \node at (1,-.3) {$3$};
                \node at (1,-1.3) {$4$};}{0.7} & %
        \tmultirow{--} & %
            \SymanzikSpec{1}{5}{1}{0}
            \triangsF{1}{1}{1}{1}{0}{}
            \triangsG{1}{31}{1}{1}{0}
        \hhline{~|-|----|-----}
        & \mmultirow{1} & %
            \SymanzikSpec{4}{8}{2}{3}
            \triangsF{1}{15}{1}{1}{0}{}
            \triangsG{4}{45}{24}{24}{0}
        \hhline{~|-|----|-----}
        & \mmultirow{1,2} & %
            \SymanzikSpec{7}{11}{3}{6}
            \triangsF{4}{21}{18}{12}{0}{}
            \triangsG{10}{57}{2\,486}{1\,618}{0}
        \hhline{~|-|----|-----}
        & \mmultirow{1,2,3} & %
            \SymanzikSpec{10}{14}{4}{9}
            \triangsF{10}{33}{826}{466}{0}{}
            \triangsG{20}{81}{522\,206}{248\,420}{3\,952\mNA}
        \hhline{~|-|----|-----}
        & \mmultirow{1,2,3,4} & %
            \SymanzikSpec{13}{17}{5}{12}
            \triangsF{20}{51}{78\,764}{34\,184}{1\,120}{}
            \triangsG{35}{117}{> 24\,729\,630\mNR}{}{} 
        \hhline{-|-|----|-----}
        \topologyDescrTikz{8}{vertex}{1}{3}{3}{
                \coordinate[dot] (A) at (0,0);
                \coordinate[dot] (B) at (2,0);
                \coordinate[dot] (C) at (1,1.73205);
                \draw (A) -- node[below] {$1$} (B);
                \draw (B) -- node[above right] {$2$} (C);
                \draw (C) -- node[above left] {$3$} (A);
                \draw (A) -- ++(-0.7,-0.2);
                \draw (B) -- ++(0.7,-0.2); 
                \draw (C) -- ++(0,0.7);}{0.7} & %
        \tmultirow{--} & %
            \SymanzikSpec{3}{6}{3}{2}
            \triangsF{1}{7}{1}{1}{0}{}
            \triangsG{4}{27}{3}{3}{0}
        \hhline{~|-|----|-----}
        & \mmultirow{1} & %
            \SymanzikSpec{4}{7}{4}{3}
            \triangsF{2}{9}{2}{2}{0}{}
            \triangsG{5}{27}{16}{15}{0}
        \hhline{~|-|----|-----}
        & \mmultirow{1,2} & %
            \SymanzikSpec{5}{8}{5}{4}
            \triangsF{3}{9}{5}{3}{0}{}
            \triangsG{6}{25}{68}{44}{0}
        \hhline{~|-|----|-----}
        & \mmultirow{1,2,3} & %
            \SymanzikSpec{6}{9}{6}{5}
            \triangsF{4}{7}{14}{4}{0}{}
            \triangsG{7}{21}{261}{99}{0}
        \hhline{-|-|----|-----}
        \topologyDescrTikz{10}{box}{1}{4}{4}{
                \coordinate[dot] (A) at (0,0);
                \coordinate[dot] (B) at (2,0);
                \coordinate[dot] (C) at (2,2);
                \coordinate[dot] (D) at (0,2);
                \draw (A) -- node[below] {$1$} (B);
                \draw (B) -- node[right] {$2$} (C);
                \draw (C) -- node[above] {$3$} (D);
                \draw (D) -- node[left] {$4$} (A);
                \draw (A) -- ++(-0.4,-0.4);
                \draw (B) -- ++(0.4,-0.4); 
                \draw (C) -- ++(0.4,0.4);
                \draw (D) -- ++(-0.4,0.4);}{0.7} & %
        \tmultirow{--} & %
            \SymanzikSpec{6}{10}{4}{5}
            \triangsF{4}{27}{3}{3}{0}{}
            \triangsG{11}{81}{102}{102}{0}
        \hhline{~|-|----|-----}
        & \mmultirow{1} & %
            \SymanzikSpec{7}{11}{5}{6}
            \triangsF{5}{27}{16}{15}{0}{}
            \triangsG{12}{75}{1\,689}{1\,260}{0}
        \hhline{~|-|----|-----}
        & \mmultirow{1,2} & %
            \SymanzikSpec{8}{12}{6}{7}
            \triangsF{6}{25}{68}{44}{0}{}
            \triangsG{13}{67}{14\,003}{8\,004}{0}
        \hhline{~|-|----|-----}
        & \mmultirow{1,2,3} & %
            \SymanzikSpec{9}{13}{7}{8}
            \triangsF{7}{21}{261}{99}{0}{}
            \triangsG{14}{57}{87\,657}{34\,143}{0}
        \hhline{~|-|----|-----}
        & \mmultirow{1,2,3,4} & %
            \SymanzikSpec{10}{14}{8}{9}
            \triangsF{8}{15}{948}{192}{0}{}
            \triangsG{15}{45}{469\,722}{114\,276}{192\mNA}
        \hhline{-|-|----|-----}
        \topologyDescrTikz{22}{flying saucer}{2}{4}{2}{
                \coordinate[dot] (A) at (0,0);
                \coordinate[dot] (B) at (2,0);
                \coordinate[dot] (C) at (1,-1);
                \draw (1,0) circle (1);
                \draw (C) arc[start angle=0, end angle=90, radius=1];
                \draw (A) -- ++(-0.7,0);
                \draw (B) -- ++(0.7,0); 
                \node at (1,1.3) {$1$};
                \node at (.85,-.15) {$2$};
                \node at (.15,-.95) {$3$};
                \node at (1.85,-.85) {$4$};}{0.7} & %
        \tmultirow{--} & %
            \SymanzikSpec{3}{8}{1}{3}
            \triangsF{3}{7}{1}{1}{0}{}
            \triangsG{5}{51}{20}{20}{0}
        \hhline{~|-|----|-----}
        & \mmultirow{1} & %
            \SymanzikSpec{5}{10}{2}{5}
            \triangsF{2}{19}{2}{2}{0}{}
            \triangsG{8}{57}{448}{432}{0}
        \hhline{~|-|----|-----}
        & \mmultirow{2} & %
            \SymanzikSpec{7}{12}{2}{7}
            \triangsF{5}{27}{16}{15}{0}{}
            \triangsG{16}{79}{15\,040}{11\,122}{106}
        \hhline{~|-|----|-----}
        & \mmultirow{4} & %
            \SymanzikSpec{6}{11}{2}{6}
            \triangsF{3}{21}{6}{6}{0}{}
            \triangsG{11}{67}{2\,388}{1\,968}{0}
        \hhline{~|-|----|-----}
        & \mmultirow{1,2} & %
            \SymanzikSpec{9}{14}{3}{9}
            \triangsF{8}{29}{274}{169}{0}{}
            \triangsG{20}{77}{992\,603}{474\,855}{17\,978\mNA}
        \hhline{~|-|----|-----}
        & \mmultirow{1,4} & %
            \SymanzikSpec{8}{13}{3}{8}
            \triangsF{5}{21}{60}{38}{0}{}
            \triangsG{14}{61}{91\,052}{43\,864}{0}
        \hhline{~|-|----|-----}
        & \mmultirow{2,3} & %
            \SymanzikSpec{10}{15}{3}{10}
            \triangsF{10}{31}{963}{425}{0}{}
            \triangsG{25}{87}{6\,297\,182}{1\,929\,202}{279\,796\mNA}
        \hhline{~|-|----|-----}
        & \mmultirow{2,4} & %
            \SymanzikSpec{9}{14}{3}{9}
            \triangsF{8}{29}{274}{169}{0}{}
            \triangsG{20}{77}{992\,603}{474\,885}{17\,978\mNA}
        \hhline{~|-|----|-----}
        & \mmultirow{1,2,3} & %
            \SymanzikSpec{12}{17}{4}{12}
            \triangsF{14}{35}{27\,286}{10\,580}{0\mNA}{}
            \triangsG{30}{89}{> 22\,768\,460}{}{}
        \hhline{~|-|----|-----}
        & \mmultirow{1,2,4} & %
            \SymanzikSpec{11}{16}{4}{11}
            \triangsF{11}{25}{5\,718}{1\,664}{0}{}
            \triangsG{24}{69}{53\,236\,122}{> 6\,184\,000}{1\,285\,610\mNA}
        \hhline{~|-|----|-----}
        & \mmultirow{2,3,4} & %
            \SymanzikSpec{12}{17}{4}{12}
            \triangsF{14}{35}{27\,286}{10\,580}{0\mNA}{}
            \triangsG{30}{89}{> 23\,973\,000\mNR}{}{}
        \hhline{~|-|----|-----}
        & \mmultirow{1,2,3,4} & %
            \SymanzikSpec{14}{19}{5}{14}
            \triangsF{18}{33}{885\,524}{196\,214}{3\,520\mNA}{}
            \triangsG{35}{85}{> 15\,067\,000\mNR}{}{}
        \hhline{-|-|----|-----}
        \topologyDescrTikz{24}{double bubble}{2}{5}{2}{
                \coordinate[dot] (A) at (0,0);
                \coordinate[dot] (B) at (2,0);
                \coordinate[dot] (C) at ($(1,0) + (135:1)$);
                \coordinate[dot] (D) at ($(1,0) + (45:1)$);
                \draw (1,0) circle (1);
                \draw (C) arc (225:315:1);
                \draw (A) -- ++(-0.7,0);
                \draw (B) -- ++(0.7,0); 
                \node at (1,-.8) {$1$};
                \node at (-.2,0.4) {$2$};
                \node at (1,1.25) {$3$};
                \node at (1,.15) {$4$};
                \node at (2.2,0.4) {$5$};
                }{0.7} & %
        \tmultirow{--} & %
            \SymanzikSpec{5}{12}{1}{6}
            \triangsF{8}{19}{1}{1}{0}{}
            \triangsG{11}{123}{3\,164}{3\,004}{0}
        \hhline{~|-|----|-----}
        & \mmultirow{1} & %
            \SymanzikSpec{10}{17}{2}{11}
            \triangsF{8}{57}{448}{432}{0}{}
            \triangsG{31}{175}{> 25\,468\,000\mNR}{}{}
        \hhline{~|-|----|-----}
        & \mmultirow{2} & %
            \SymanzikSpec{7}{14}{2}{8}
            \triangsF{3}{43}{6}{6}{0}{}
            \triangsG{15}{129}{190\,680}{164\,892}{0\mNA}
        \hhline{~|-|----|-----}
        & \mmultirow{3} & %
            \SymanzikSpec{11}{18}{2}{12}
            \triangsF{11}{67}{2\,388}{1\,968}{0}{}
            \triangsG{40}{195}{> 18\,217\,000\mNR}{}{}
        \hhline{~|-|----|-----}
        & \mmultirow{1,2} & %
            \SymanzikSpec{12}{19}{3}{13}
            \triangsF{11}{57}{10\,372}{5\,508}{0\mNA}{}
            \triangsG{35}{157}{> 16\,960\,000\mNR}{}{}
        \hhline{~|-|----|-----}
        & \mmultirow{1,3} & %
            \SymanzikSpec{15}{22}{3}{16}
            \triangsF{23}{83}{8\,345\,888}{3\,297\,834}{259\,286\mNA}{}
            \triangsG{60}{215}{> 12\,008\,000\mNR}{}{}
        \hhline{~|-|----|-----}
        & \mmultirow{2,3} & %
            \SymanzikSpec{13}{20}{3}{14}
            \triangsF{15}{69}{101\,540}{55\,028}{0\mNA}{}
            \triangsG{45}{181}{> 15\,441\,000\mNR}{}{}
        \hhline{~|-|----|-----}
        & \mmultirow{2,5} & %
            \SymanzikSpec{12}{19}{3}{13}
            \triangsF{11}{57}{10\,372}{5\,508}{0\mNA}{}
            \triangsG{35}{157}{> 16\,960\,000\mNR}{}{}
        \hhline{~|-|----|-----}        
        & \mmultirow{3,4} & %
            \SymanzikSpec{15}{22}{3}{16}
            \triangsF{20}{75}{3\,507\,528}{1\,072\,166}{1\,248\mNA}{}
            \triangsG{59}{211}{> 12\,323\,000\mNR}{}{}
        \hhline{~|-|----|-----}
        & \mmultirow{1,2,3} & %
            \SymanzikSpec{17}{24}{4}{18}
            \triangsF{27}{73}{> 20\,208\,000\mNR}{}{}{}
            \triangsG{65}{189}{> 10\,753\,000\mNR}{}{}
        \hhline{~|-|----|-----}
        & \mmultirow{1,2,3,4} & %
            \SymanzikSpec{21}{28}{5}{22}
            \triangsF{41}{87}{> 13\,648\,000\mNR}{}{}{}
            \triangsG{90}{127}{> 8\,939\,000\mNR}{}{}
        \hhline{~|-|----|-----}
        & \mmultirow{1,2,3,4,5} & %
            \SymanzikSpec{23}{30}{6}{24}
            \triangsF{46}{73}{> 11\,750\,000\mNR}{}{}{}
            \triangsG{96}{189}{> 8\,698\,000\mNR}{}{}            
        \hhline{-|-|----|-----}
        \topologyDescrTikz{6}{kite}{2}{5}{2}{
                \coordinate[dot] (A) at (0,0);
                \coordinate[dot] (B) at (2,0);
                \coordinate[dot] (C) at (1,1);
                \coordinate[dot] (D) at (1,-1);
                \draw (1,0) circle (1);
                \draw (C) -- node[left] {$5$} (D);
                \draw (A) -- ++(-0.7,0);
                \draw (B) -- ++(0.7,0); 
                \node at (.15,.85) {$1$};
                \node at (1.85,.85) {$2$};
                \node at (1.85,-.85) {$3$};
                \node at (.15,-.85) {$4$};}{0.7} & %
        \tmultirow{--} & %
            \SymanzikSpec{8}{16}{1}{10}
            \triangsF{5}{51}{20}{20}{0}{}
            \triangsG{42}{219}{> 37\,916\,000\mNR}{}{}
        \hhline{~|-|----|-----}
        & \mmultirow{1} & %
            \SymanzikSpec{12}{20}{2}{14}
            \triangsF{16}{79}{15\,040}{11\,122}{106}{}
            \triangsG{61}{235}{> 12\,816\,000\mNR}{}{}
        \hhline{~|-|----|-----}
        & \mmultirow{5} & %
            \SymanzikSpec{14}{22}{2}{16}
            \triangsF{22}{93}{901\,622}{418\,170}{59\,302\mNA}{}
            \triangsG{76}{265}{> 11\,558\,000\mNA}{}{}
        \hhline{~|-|----|-----}
        & \mmultirow{1,2} & %
            \SymanzikSpec{15}{23}{3}{17}
            \triangsF{25}{89}{> 29\,254\,000\mNR}{}{}{}
            \triangsG{75}{229}{> 10\,339\,000\mNR}{}{}
        \hhline{~|-|----|-----}
        & \mmultirow{1,2,3} & %
            \SymanzikSpec{19}{27}{4}{21}
            \triangsF{37}{81}{> 13\,541\,000\mNR}{}{}{}
            \triangsG{95}{213}{> 8\,728\,000\mNR}{}{}
        \hhline{~|-|----|-----}
        & \mmultirow{1,2,3,4} & %
            \SymanzikSpec{22}{30}{5}{24}
            \triangsF{47}{69}{> 12\,047\,000\mNA}{}{}{}
            \triangsG{110}{189}{> 8\,534\,000\mNA}{}{}
        \hhline{~|-|----|-----}
        & \mmultirow{1,2,3,4,5} & %
            \SymanzikSpec{26}{34}{6}{28}
            \triangsF{62}{81}{> 10\,143\,000\mNA}{}{}{}
            \triangsG{136}{213}{> 7\,357\,000\mNA}{}{}
        \hhline{-|-|----|-----}
        \topologyDescrTikz{18}{dunce's cap}{2}{4}{3}{
                \coordinate[dot] (A) at (0,0);
                \coordinate[dot] (B) at (2,-1);
                \coordinate[dot] (C) at (2,1);
                \draw (A) -- node[above] {$1$} (C);  
                \draw (A) -- node[below] {$2$} (B);  
                \draw (C) arc[start angle=135, end angle=225, radius=1.4142];
                \draw (B) arc[start angle=-45, end angle=45, radius=1.4142];
                \draw (A) -- ++(-1,0);
                \draw (B) -- ++(0.2,-.7); 
                \draw (C) -- ++(0.2,.7); 
                \node at (1.3,0) {$3$};
                \node at (2.7,0) {$4$};}{0.7} & %
        \tmultirow{--} & %
            \SymanzikSpec{4}{9}{3}{4}
            \triangsF{1}{15}{1}{1}{0}{}
            \triangsG{8}{67}{42}{42}{0}
        \hhline{~|-|----|-----}
        & \mmultirow{1} & %
            \SymanzikSpec{6}{11}{4}{6}
            \triangsF{3}{21}{6}{6}{0}{}
            \triangsG{11}{67}{2\,388}{1\,968}{0}
        \hhline{~|-|----|-----}
        & \mmultirow{3} & %
            \SymanzikSpec{7}{12}{4}{7}
            \triangsF{5}{27}{16}{15}{0}{}
            \triangsG{16}{79}{15\,040}{11\,122}{106}
        \hhline{~|-|----|-----}
        & \mmultirow{1,2} & %
            \SymanzikSpec{8}{13}{5}{8}
            \triangsF{5}{21}{60}{30}{0}{}
            \triangsG{14}{61}{91\,052}{43\,864}{0\mNA}
        \hhline{~|-|----|-----}
        & \mmultirow{1,3} & %
            \SymanzikSpec{9}{14}{5}{9}
            \triangsF{8}{29}{274}{169}{0}{}
            \triangsG{20}{77}{992\,603}{474\,855}{17\,978\mNA}
        \hhline{~|-|----|-----}
        & \mmultirow{3,4} & %
            \SymanzikSpec{10}{15}{5}{10}
            \triangsF{10}{31}{963}{425}{0}{}
            \triangsG{25}{87}{6\,297\,182}{1\,929\,202}{279\,796\mNA}
        \hhline{~|-|----|-----}
        & \mmultirow{1,2,3} & %
            \SymanzikSpec{11}{16}{6}{11}
            \triangsF{11}{25}{5\,718}{1\,664}{0}{}
            \triangsG{24}{69}{53\,236\,122}{> 6\,184\,000}{1\,285\,610\mNA}
        \hhline{~|-|----|-----}
        & \mmultirow{2,3,4} & %
            \SymanzikSpec{12}{17}{6}{12}
            \triangsF{14}{35}{27\,286}{10\,580}{0\mNA}{}
            \triangsG{30}{89}{> 23\,973\,000\mNR}{}{}
        \hhline{~|-|----|-----}
        & \mmultirow{1,2,3,4} & %
            \SymanzikSpec{14}{19}{7}{14}
            \triangsF{18}{33}{885\,524}{196\,214}{3\,520\mNA}{}
            \triangsG{35}{85}{>15\,067\,000\mNR}{}{}
        \hhline{-|-|----|-----}
        \caption{\normalsize Classification of certain basic Feynman graphs with respect to their behaviour from the perspective of $\Aa$-hypergeometric theory. We list the loop number $L$, the number of edges $n$, the number of external edges $m$ (thus, there are $m-1$ independent external momenta) and the set of massive edges. Thereby, masses assumed to be different, and we will not always list all possible mass configurations. $N_\Ff$ denotes the number of monomials in $\Ff$, $N$ is the number of monomials in $\Gg=\Uu+\Ff$, $N^{ph}$ is the number of physically relevant variables and $r=N-n-1$ is the corank of $\Aa$. Furthermore, we will give the volume of the Newton polytopes $\vol(\Newt(\Ff))$ and $\vol(\Newt(\Gg))$, the number of faces of those Newton polytopes $|\Newt(\Ff)|$ and $|\Newt(\Gg)|$, the number of regular triangulations $\mathfrak T_\Ff$ and $\mathfrak T_\Gg$, the number of unimodular, regular triangulations $\mathfrak T_\Ff^{\text{unimod}}$ and $\mathfrak T_\Gg^{\text{unimod}}$ as well as the number of non-regular triangulations $\mathfrak T_\Ff^{\text{non-reg}}$ and $\mathfrak T_\Gg^{\text{non-reg}}$ for the second Symanzik polynomial $\Ff$ and for $\Gg=\Uu+\Ff$, respectively. Those numbers were calculated by \softwareName{Macaulay2} and \softwareName{Topcom}, according to \cref{ssec:Macaulay2} and \cref{ssec:topcom}. For more complex graphs it is not possible to count all of the triangulations, because the memory is quickly exhausted on a usual computer. In these cases, we have given lower limits for the number of all triangulations. For polytopes with larger volume and more vertices, the counting of triangulations stops earlier because the memory consumption is higher. However, one should expect that a polytope with a larger volume admits usually more triangulations. The large number of triangulations shows also how difficult it is to determine Landau varieties by means of a defining polynomial. We recall that the number of triangulations is a lower bound for the number of monomials in the defining polynomial of the Landau variety.} \label{tab:characteristics}
    \end{longtable}
    \vfill 
    \hrule 
    \noindent$\dagger$: may not contain all non-regular triangulations (the flip graph of non-regular triangulations is not necessarily connected) \\
    $\ddagger$: may include also non-regular triangulations
    }

\end{landscape}

\clearpage

\restoregeometry



\chapter{Bibliography}

\printbibliography[heading=none]

\glsaddallunused

\cleardoublepage

\end{document}